\def\versionno{ CS dissertation      --   version  4.6    --  05.10.12 }

                   \newcommand\void[1] {}
\catcode`\@=11
\newif\if@fewtab\@fewtabtrue
{\count255=\time\divide\count255 by 60
\xdef\hourmin{\number\count255}
\multiply\count255 by-60\advance\count255 by\time
\xdef\hourmin{\hourmin:\ifnum\count255<10 0\fi\the\count255}}
\def\ps@draft{\let\@mkboth\@gobbletwo
    \def\@oddfoot{\hbox to 7 cm{\tiny \versionno
       \hfil}\hskip -7cm\hfil\rm\thepage \hfil {\tiny\draftdate}}
    \def\@oddhead{}
    \def\@evenhead{}\let\@evenfoot\@oddfoot}
\def\draftdate{\number\month/\number\day/\number\year\ \ \ \hourmin }

\def\citen#1{\if@filesw \immediate\write \@auxout {\string\citation{#1}}\fi%
\@tempcntb\m@ne \let\@h@ld\relax \def\@citea{}%
\@for \@citeb:=#1\do {\@ifundefined {b@\@citeb}%
    {\@h@ld\@citea\@tempcntb\m@ne{\bf ?}%
    \@warning {Citation `\@citeb ' on page \thepage \space undefined}}%
    {\@tempcnta\@tempcntb \advance\@tempcnta\@ne
    \setbox\z@\hbox\bgroup\ifcat0\csname b@\@citeb \endcsname \relax
    \egroup \@tempcntb\number\csname b@\@citeb \endcsname \relax
    \else \egroup \@tempcntb\m@ne \fi \ifnum\@tempcnta=\@tempcntb
    \ifx\@h@ld\relax \edef \@h@ld{\@citea\csname b@\@citeb\endcsname}%
    \else \edef\@h@ld{\hbox{--}\penalty\@highpenalty
    \csname b@\@citeb\endcsname}\fi
    \else \@h@ld\@citea\csname b@\@citeb \endcsname \let\@h@ld\relax \fi}%
\def\@citea{,\penalty\@highpenalty\hskip.13em plus.13em minus.13em}}\@h@ld}
\def\@citex[#1]#2{\@cite{\citen{#2}}{#1}}%
\def\@cite#1#2{\leavevmode\unskip\ifnum\lastpenalty=\z@\penalty\@highpenalty\fi%
  \ [{\multiply\@highpenalty 3 #1%
  \if@tempswa,\penalty\@highpenalty\ #2\fi}]}   %
\makeatother 
\catcode`\@=12

\documentclass[10pt,a4paper,twoside,openany,makeidx]{report}
\usepackage{palatino,latexsym, amsmath, amsthm, amsfonts, enumerate, amssymb, bbm, xspace, xypic}
\usepackage{fancybox}
\usepackage[all]{xy}
\usepackage[mathscr]{eucal}
\usepackage{graphicx,color}
\usepackage{rotating}
\usepackage{epstopdf,hyperref}
\usepackage{makeidx}

 \usepackage[T1]{fontenc}
 \usepackage[english]{babel}
 \usepackage[all]{xy}
 \usepackage{verbatim}
 \usepackage{mathrsfs}
 \usepackage{fancyhdr}
 \usepackage{stmaryrd}

\usepackage{titlesec}

\titleformat{\chapter}
  {\normalfont\LARGE\bfseries}{ \thechapter }{15pt}{\LARGE}
\titlespacing*{\chapter}
  {0pt}{30pt}{25pt}
\newtheorem{thm}{Theorem}

\newtheorem{lemma}[thm]{Lemma}
\newtheorem{cor}[thm]{Corollary}
\newtheorem{prop}[thm]{Proposition}
\newtheorem{proposition}[thm]{Proposition}
\newtheorem{theorem}[thm]{Theorem}

\theoremstyle{definition}
\newtheorem{definition}[thm]{Definition}

\newtheorem{remark}[thm]{Remark}
\newtheorem{conv}[thm]{Convention}

\def\proof {\noindent{{\it Proof.}}\hspace{7pt}}
\def\endofproof {\hfill{$\Box$}\\}

\def\be{\begin{equation}}
\newcommand\bee[5] {\begin{eqnarray} #5 \nonumber\\[-#1.#2em]~\\[#3.#4em]~\nonumber\end{eqnarray}}
\def\bearl         {\begin{array}{l}}
\def\bearll        {\begin{array}{ll}}

\def\eear          {\end{array}}
\def\ee{\end{equation}}

\newcommand\eqpic[4]{\begin{eqnarray}
                   \begin{picture}(#2,#3){}\end{picture}\nonumber\\
                   \raisebox{-#3pt}{ \begin{picture}(#2,#3) #4 \end{picture} }
                   \label{#1} \\~\nonumber \end{eqnarray} }
\newcommand\Eqpic[4]{\begin{eqnarray}
                   \begin{picture}(#2,#3){}\end{picture}\nonumber\\
                   \raisebox{-#3pt}{ \begin{picture}(#2,#3) #4 \end{picture} }
                   \nonumber \\[3pt]~\label{#1} \end{eqnarray} }
\newcommand\includepic[2]   {{\begin{picture}(0,0)(0,0)
                            \scalebox{.#1}{\includegraphics{#2.eps}}\end{picture}}}
\newcommand\Includepic[1]   {{\begin{picture}(0,0)(0,0)
                            \scalebox{.38}{\includegraphics{#1.eps}}\end{picture}}}
\newcommand\INcludepic[1] {{\begin{picture}(0,0)(0,0)
                   \scalebox{.456}{\includegraphics{def_fact_#1.eps}}\end{picture}}}

\newcommand\includepicclax[3] {{\begin{picture}(0,0)(0,0) \scalebox{.#1#2}
                   {\includegraphics{pic_clal_#3.eps}}\end{picture}}}

\newcommand\INcludepicclal[2] {{\begin{picture}(0,0)(0,0) \scalebox{.#2}
                   {\includegraphics{pic_clal_#1.eps}}\end{picture}}}
\newcommand\Includepicfj[2]   {{\begin{picture}(0,0)(0,0)
                            \scalebox{.#1}{\includegraphics{def_fact_#2.eps}}\end{picture}}}
\newcommand\includepichtft[1] {{\begin{picture}(0,0)(0,0)
                   \scalebox{.304}{\includegraphics{pic_htft_#1.eps}}\end{picture}}}
\newcommand\Includepichtft[1] {{\begin{picture}(0,0)(0,0)
                   \scalebox{.38}{\includegraphics{pic_htft_#1.eps}}\end{picture}}}
\newcommand\INcludepichtft[2] {{\begin{picture}(0,0)(0,0)
                   \scalebox{.#2}{\includegraphics{pic_htft_#1.eps}}\end{picture}}}
\newcommand\includepichopf[1] {{\begin{picture}(0,0)(0,0)
                   \scalebox{.304}{\includegraphics{pic_hopf_#1.eps}}\end{picture}}}

\newcommand\Includepichopf[1] {{\begin{picture}(0,0)(0,0)
                   \scalebox{.38}{\includegraphics{pic_hopf_#1.eps}}\end{picture}}}

\def\blue          {\color{mydarkblue}}
\def\boxA          {{\color{Acol}\framebox(8,8){$\sss \!A$}}}
\def\boxB          {{\color{Bcol}\framebox(8,8){$\sss \!B$}}}
\def\btcs          {\small \begin{tabular}c}

\newcommand\Fbox[1]{\Ovalbox{#1}}
\newcommand\fbY[3] {{\color{Brcl}\fbox{$\scriptstyle\mathrm Y_{\!#1;#2}^{#3}$}}}
\newcommand\hsp[1] {\mbox{\hspace{#1 em}}}

\newcommand\nxl[1] {\\[#1mm]}
\newcommand\Nxl[1] {\\[-1.3em]\\[#1mm]}
\def\nxx           {\nxl{-2.3}}
\def\nxy           {\nxl{-1.6}}

\def\papI   {I}
\def\papII   {II}
\def\papIII   {III}
\def\papIV   {IV}
\def\papV   {V}
\def\papVI   {VI}
\def\papVII   {VII}
\def\papVIII   {VIII}

\def\quand  {\quad\text{and}\quad}
\def\qquand  {\quad\text{and}\quad\quad}

\newcommand\setulen[2]{\setlength\unitlength{.#1#2pt}}
\newcommand\subject[1]  {\subsubsection{#1}}
\newcommand\ssubject[1]    {\paragraph{#1}}

\newcommand\vleq[2]{\ensuremath{\stackrel{\mbox{\rule{#1.#2em}{.03em}}}
                                         {\mbox{\rule{#1.#2em}{.03em}}}}}
\newcommand\vlleftarrow[2]{\ensuremath{\longleftarrow\!\!\!\raisebox{.28em}
                   {\rule{#1.#2em}{.03em}}}}
\newcommand\vlrightarrow[2]{\ensuremath{\raisebox{.28em}
                   {\rule{#1.#2em}{.03em}}\!\!\!\longrightarrow}}
\newcommand\Corr[2]     {\ensuremath{\mathfrak C(\Surf {#1}{#2})}}

\newcommand\Corrw[2]     {\ensuremath{\mathfrak C^{\ra}(\Surf {#1}{#2})}}
\def\Corrgn     {\ensuremath{\Corr gn}}
\def\Corrgnw     {\ensuremath{\Corrw gn}}

\newcommand\HCorr[2]     {\ensuremath{\mathfrak C_H(\Surf {#1}{#2}{})}}

\def\HCorrgn     {\ensuremath{\mathfrak C_H(\Surf gn)}}
\def\HK            {{{\mathcal K}}}
\def\HKH            {{{\ensuremath K}}}

\def\LAak   {\ensuremath{Z_{a_k}}}
\def\LAbk   {\ensuremath{Z_{b_k}}}
\def\LAdk   {\ensuremath{Z_{d_k}}}
\def\LAek   {\ensuremath{Z_{e_k}}}
\def\LAg   {\ensuremath{Z_{\gamma}}}
\def\LASk   {\ensuremath{Z_{S_k}}}
\def\LAtjk  {\ensuremath{Z_{t_{j,k}}}}
\def\LAwi   {\ensuremath{Z_{\omega_i}}}
\def\LARi   {\ensuremath{Z_{R_i}}}

\def\ak   {\ensuremath{a_k}}
\def\bk   {\ensuremath{b_k}}
\def\dk   {\ensuremath{d_k}}
\def\ek   {\ensuremath{e_k}}
\def\sk   {\ensuremath{S_k}}
\def\tjk  {\ensuremath{t_{j,k}}}
\def\wi   {\ensuremath{\omega_i}}
\def\Ri   {\ensuremath{R_i}}

\def\Qq            {\mathcal Q^{\rm l}}
\def\QQ            {\mathcal Q}
\def\pb            {{\bar p}}
\def\qb{{\bar q}}

\def\rad  {\ohr_\diamond}

\newcommand\Sk[3] {\ensuremath{\mathfrak {F}^{\,#1}_{#2,#3}}}

\newcommand\SkH[3] {\ensuremath{\mathfrak F^{H,#1}_{#2,#3}}}
\newcommand\Skw[3] {\ensuremath{\mathfrak F^{H,\ra;#1}_{#2,#3}}}
\def\SK {S_{\HK}}
\def\SKH {S_{\HKH}}
\def\SL {S_{\L}}
\newcommand\Surf[2] {\ensuremath{\Sigma_{#1,#2}}}

\def\T             {{\mathscr T}}
\def\TL {T_{\L}}
\def\TK {T_{\HK}}
\def\TKH {T_{\HKH}}

\newcommand\Aampk[3]{\mathrm A_{#3#1}^{\;\;\;\;#2}}
\def\AA {\ensuremath{\mathcal{A}}}
\def\ad {\text{ad}}
\def\alghf  {\ensuremath{\mathcal F}}
\def\alphd         {{\alpha_{\!2}^{}}}
\def\alphe         {{\alpha_{\!1}^{}}}
\def\alphv         {{\alpha_{\!4}^{}}}
\def\alphz         {{\alpha_{\!3}^{}}}
\def\apo           {\mbox{\sc s}}
\def\apoi          {\mbox{\sc s}^{-1}}
\def\ASF        {_{_{\BFH}}}
\def\ASK        {_{_{\HKH}}}
\def\Bimod  {\text{-Bimod}}

\def\BA  {\ensuremath{B}}
\def\BF  {\ensuremath{\mathrm F}}
\def\BFw  {\ensuremath{F_{\ra}}}
\def\BFww  {\ensuremath{F_{\ra'^{-1}\cir\ra}}}
\def\BFss  {\ensuremath{\mathrm F_{\!\circ}}}
\def\BFH  {\ensuremath{F}}

\def\bhHopf {bulk handle Hopf algebra}
\def\blS{\mathcal{B}}
\def\blV{{B}}
\newcommand\BOF[6]{(#1,#2,#3,#4,#5,#6)}
\def\boundA {boundary Frobenius algebra}

\def\boxti  {\,{\boxtimes}\,}
\def\btimes {\boxti}
\def\budef         {defect-crossing}

\newcommand\BUF[6]{(#1,#2,#3,#4,#5,#6)}
\def\bulkA {bulk Frobenius algebra}
\def\BulkA {Bulk Frobenius algebra}
\def\bulkAC {Cardy bulk Frobenius algebra}

\def\CA {\ensuremath{\mathcal A}}
\def\catpic        {{\scriptsize \shadowbox{\C}}}
\def\cbl    {\mathcal Bl}
\def\cbulk         {c^{\text{bulk}}}
\def\cbulki        {c^{\text{bulk}^{\,\scriptstyle-1}}}
\def\CC {\ensuremath{E^c}}
\def\ccA           {{\sss (A)}}
\def\ccB           {{\sss (B)}}
\newcommand\cdef[5]{(c^{\mathrm{def}}_{#1,#2#3})_{#4#5}^{\phantom|}}
\newcommand\cdefinv[5]{{(c^{\mathrm{def}\;-1}_{#1,#2#3})}_{\!#4#5}^{}}
\def\CD            {\ensuremath{\mathcal D_{\!A|B}}}
\def\CDone            {\ensuremath{\mathcal D_{\!A|\one}}}
\def\CF    {(FIN)}
\newcommand\Cfact[5]    {C^{\rm{fact;#1}}_{#2#3,#4#5}}
\newcommand\Cfactnorm[5]    {C^{\rm{FFRS;#1}}_{#2#3,#4#5}}
\def\CFN    {(MOD)}
\def\chii          {\raisebox{.15em}{$\chi$}}
\def\cir    {\,{\circ}\,}
\newcommand\clc[9] {C^{#1,#2#3;#4,#5#6}_{#7,#8#9}}

\def\coa           {_\triangleright}
\def\coar          {_\triangleleft}
\def\cob    {\ensuremath{\mathcal M}}
\newcommand\coen[1]{\int^{#1}\hspace*{-.23em}}
\newcommand\Coend[2]{\ensuremath{\int^X\!#1(#2,#2)}}
\def\coop          {^{\mathrm{coop}}}
\def\Cor    {\text{Cor}}
\def\cXo           {\ensuremath{c_{\!X;0}^{}}}

\newcommand\dcoef[5]{d^{#1#2,#3#4}_{#5}}

\def\defspace  {space of defect fields}
\def\Delun  {G_{\square}}
\newcommand\DF[8]{(#1,#2,#3,#4,#5,#6,#7,#8)}
\def\isod {^{\surd}}
\def\deff   {\,{:=}\,}
\def\dsty          {\displaystyle }

\def\e{\text{e}}
\newcommand\Eev[1] {{{}^{\vee\!}}\!{#1}}
\newcommand\emb[3] {e^{#1}_{#2,#3}}
\newcommand\Emb[2] {e^{#1}_{#2}}
\def\End{\text{End}}
\def\EndH{\text{End}_H}
\newcommand\eota[2]{e^{}_{#1\otimes_{\!A}^{}#2\prec}}
\def\eps{\varepsilon}
\def\eq            {\,{=}\,}

\newcommand\erf[1] {(\ref{#1})}
\def\ES {\ensuremath{E}}
\newcommand\ExtTdef[1]    {\ensuremath{T_{#1,#1}}}

\def\Fbb           {{\ensuremath{G_{\!\Box}^H}}}
\def\Fbx           {{\ensuremath{G_{\!\otimes_\k}^H}}}
\def\Fbxo          {{\ensuremath{G_{\!\otimes_{\k}}^{H;\ra}}}}
\newcommand\flip[2]   {\ensuremath{\tau_{#1,#2}}}
\def\fmap          {{{\ensuremath{\Psi}}}}
\def\FP {\star}
\def\fundws     {\ensuremath{\mathcal S}}
\def\fundwsred     {\ensuremath{\tilde{\mathcal S}}}

\def\g  {\ensuremath{\mathfrak g}}

\newcommand\gcoef[5]{\Psi^{#1#2,#3#4}_{#5}}
\newcommand\GLL[2] {G\!\ell_{#1#2}}

\def\H{\ensuremath{\mathcal H}}
\def\HA {\mathrm H_A}
\newcommand\Hatwsf[2]{\widehat{\ws~}_{\!\!\!#1#2}}
\def\HBimod        {{\ensuremath{H\mathrm{-Bimod}}}}
\def\HBImod        {{\ensuremath{H\text{-}\mathrm{Bimod}}}}
\def\Hbimodpic        {{\scriptsize \shadowbox{\HBimod}}}

\def\HK            {\ensuremath {\mathrm K}}
\def\HKH            {{{\ensuremath K}}}
\def\HMod          {{\ensuremath{H}\text{-Mod}}}
\def\Hom{\ensuremath{\mathrm{Hom}}}
\newcommand\Homa[2]{\ensuremath{\text{Hom}_{A}(#1,#2)}}
\newcommand\Homaa[2]{\ensuremath{\text{Hom}_{A|A}(#1,#2)}}
\newcommand\Homab[2]{\ensuremath{\text{Hom}_{A|B}(#1,#2)}}
\newcommand\Homac[2]{\ensuremath{\text{Hom}_{A|C}(#1,#2)}}
\newcommand\Hombb[2]{\ensuremath{\text{Hom}_{B|B}(#1,#2)}}
\newcommand\Hombc[2]{\ensuremath{\text{Hom}_{B|C}(#1,#2)}}
\newcommand\Homcc[2]{\ensuremath{\text{Hom}_{C|C}(#1,#2)}}
\def\HomH          {{\ensuremath{\mathrm{Hom}_H}}}
\def\Homk          {{\ensuremath{\mathrm{Hom}_\k}}}
\def\HomP{\ensuremath{\mathrm{Hom}^P}}
\def\Hs {\ensuremath{H^*}}
\def\Hss           {{H^*_{}}}

\def\i{\text{i}}
\def\I{\ensuremath{\mathcal I}}
\def\IJ{\ensuremath{\mathcal J}}
\def\ia                     {{\ensuremath{\imath}}}
\def\ib                     {\ensuremath{{\bar\imath}}}
\def\id		                {\mathrm{id}}
\def\Id		                {\mathrm{Id}}

\def\idHs          {\ensuremath{\id_{{H^{\phantom:}}^{\!\!*}}}}
\def\iHb           {\imath^{\BFH}}
\def\iHbo          {\imath^{\BFw}}
\def\idsm          {\mbox{\footnotesize\sl id}}
\def\iHKH{\iota^{\HKH}}
\def\im {\mathrm{Im}}
\def\image  {\text{Im}}
\def\In     {\,{\in}\,}
\newcommand\instord[5]  {\ensuremath{\mathcal{\overline L}^{#1}_{#2,#3,#4,#5}}}
\def\inv                    {^{-1}}
\def\iso    {\,{\cong}\,}

\def\ja                     {{\ensuremath{\jmath}}}
\def\jb{\ensuremath         {{\bar\jmath}}}

\newcommand\KC[8]   {\mathcal K_{#1#2#3;#4#5#6}^{#7;#8}}
\def\kb                     {\ensuremath {{\bar k}}}

\def\L  {\ensuremath{\mathrm L}}
\def\LH  {\ensuremath{L}}
\newcommand\Lact[1] {\rho^\L_{#1}}
\def\lad {\rho_\diamond}
\def\lb            {\ensuremath{{\bar l}}}
\def\ld {\,{^\vee}\!}
\def\lsqarrow      {\scalebox{.38}{\includegraphics{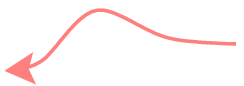}}}
\def\Lfun   {\ensuremath{F_0}}

\newcommand\lyubspace[2]  {\ensuremath{#1^{#2}}}
\def\Lyubspace  {\ensuremath{U^{n}}}

\def\ra {\ensuremath{\omega}}
\newcommand\refc[3]{b^{#1,#2}_{#3}}
\def\rev    {^{\text{rev}}}
\def\rhov          {\rho_{\scriptscriptstyle\!\vee}^{}}
\def\rhoV          {{}_{\scriptscriptstyle\vee\!}^{}\rho}
\def\Lss  {\ensuremath{{\mathcal L_{\text{ssi}}}}}

\def\mapdef {\,{:}\,}
\def\M             {{\mathcal M}}
\def\Mf             {\check{\mathcal M}}
\def\Map    {\text{Map}}

\def\Mapgn    {\ensuremath{\text{Map}_{g,n}}}
\newcommand\Mapio[2]    {\ensuremath{\text{Map}_{#1,#2}}}

\def\MGL           {\M_{G\!\ell}}
\def\MN            {{\ensuremath{\mathscr N}}}
\def\Mod    {\text{-mod}}
\def\rMod    {mod-}
\def\MST           {\M_{\varSigma,\torus}}
\newcommand\Mt[1] {\M_{T;#1}}
\def\Mws    {\ensuremath{M_{\ws}}}
\def\Mwscut {\Mws^\text{cut}}
\newcommand\MwsSD[5]{\M^{S^{2};#1}_{#2#3,#4#5}}

\newcommand\N[3]            {\mathcal{N}_{#1#2}^{\;\;#3}}
\newcommand\nbfinv[4]{\ensuremath{(\nbfm\inv)_{#1#2,#3#4}}}
\def\nbfm{\omega}
\def\nl{\cdot}

\def\ohrv          {\ohr_{\scriptscriptstyle\!\vee}^{}}
\def\ohrV          {{}_{\scriptscriptstyle\vee\!}^{}\ohr}
\def\ohr           {\reflectbox{$\rho$}}
\def\roh           {\raisebox{.4em}{\rotatebox{180}{$\rho$}}}
\def\op {^{\mathrm{op}}}
\def\Or{\text{or}}
\def\ort{\text{or}_2}

\def\oti    {\,{\otimes}\,}
\def\otiA    {\oti_{\!\!A}\,}
\def\OtiA    {{\otimes_{\!A}}}
\def\otiB    {\oti_{\!B}\,}
\def\otiC    {\oti_{\!\CN}\,}
\def\otipm  {\,{\otimes^\pm}}
\def\otip  {\,{\otimes^+}}
\def\Otip  {\,{\otimes^+}\,}
\def\otim  {\,{\otimes^-}\,}
\def\Oti           {{\otimes}}
\def\otiBA{\,{\otimes_\Delta}\,}
\def\otii{{\otimes}}
\def\otik{\,{\otimes_\k}\,}
\def\Otik          {{\otimes_\k}}

\newcommand\pA[1]  {\sse$\color{Acol} #1$}
\def\pb{{\bar p}}
\newcommand\pB[1]  {\sse$\color{Bcol} #1$}
\newcommand\PHi[1] {\Phi_{\!#1}}
\def\Phisod   {\Phi_\surd}
\def\piL    {\pi_\L^{g,n}}
\newcommand\piLio[2]  {\pi_\L^{g,#1+#2}}
\newcommand\PiLio[2]  {\pi_\L^{g,#1,#2}}

\newcommand\PIKio[2]  {\tilde\pi_{\HKH}^{g,#1,#2}}
\newcommand\PIKione[2]  {\tilde\pi_{\HKH}^{1,#1,#2}}
\def\piK    {\pi_\HKH}
\def\piKC    {\pi_\HK}
\def\piv           {\pi}
\newcommand\pg[1]{\scriptsize$#1$}
\newcommand\pl[1]{\scriptsize$#1$}
\newcommand\pota[2]{p^{}_{#1\otimes_{\!A}^{}#2}}
\def\proj{\mathrm{P}}
\def\projMF{\mathcal{P}}
\newcommand\pX[1]  {\sse$\color{defcol} #1$}

\def\qb{{\bar q}}

\def\rb     {\bar r}
\def\rd {{^\vee}\!}
\newcommand\rota[2]{r^{}_{\!\succ #1\otimes_{\!A}^{}#2}}
\newcommand\res[3] {r^{#1}_{#2,#3}}
\newcommand\Res[2] {r^{#1}_{#2}}
\newcommand\rhoHKH[1]   {\rho^{\HKH}_{#1}}
\newcommand\RR[5]   {{\sf R}^{(#1\,#2)#3}_{#4\,#5}}

\newcommand\Rm[5]  {{\sf R}^{-\,(#1\,#2)#3}_{\;#4\,#5}}

\def\one           {{\bf1}}

\def\sb {\bar s}
\def\Scut          {S_{\text{cut}}}
\newcommand\sob[4]{\lambda_{(#1,#2),#3}^{#4}}
\newcommand\sobd[4]{\Lambda^{(#1,#2),#3}_{#4}}
\def\sse           {\scriptsize }
\def\sss           {\scriptscriptstyle}

\def\tft    {\texttt{tft}}
\def\tftA   {\ensuremath{A}}
\def\tftfun {$\tft$-functor}
\def\tftm   {Z}
\def\tfts   {\H}
\newcommand\Tfact[6] {\M^{\torus,#1#2}_{#3#5,#4#6}}
\newcommand\Thet[4]{\Theta^{#1}_{#2#3,#4}}
\newcommand\Thets[4]{\Theta^{#1}_{#2#3,#4}}
\newcommand\Thete[4]{\tilde\Theta^{#1}_{#2#3,#4}}
\def\To            {\,{\to}\,}
\def\tovacuum      {\btcs projection\nxx \vlrightarrow46 \nxx to vacuum
                    \nxy channel \end{tabular}}
\def\tovacuuw      {\btcs projection\nxx \vlleftarrow46  \nxx to vacuum
                    \nxy channel \end{tabular}}

\def\TOR           {\TT_{\!X}}
\def\TORS          {\mathcal T_{\!X}}

\def\TOT           {\TT_{X;pq,\alpha\beta}}
\def\torus         {{\ensuremath{{\mathrm T}^2}}}
\def\tr     {\text{trace}}
\def\Tr     {\text{tr}}
\def\Times         {\,{\times}\,}
\def\TT            {{\mathcal {\tilde T}}}
\newcommand\twistmod[3]   {^{#2}\!\!#1^{#3}}
\def\twoact {_\rho\mathfrak{F}_{\scriptsize\ohr}}
\def\tws           {{\mathring{\ws}}}
\def\TXY           {_{\torus}^{X|Y}}
\def\TXYij         {_{\torus,\,ij}^{X|Y}}

\def\V  {\ensuremath{\mathcal V}}
\def\VA    {\ensuremath{\mathcal V}}
\def\Vectpic        {{\scriptsize \shadowbox{\Vectk}}}
\def\Vee           {^{\vee}}
\def\Vir    {\ensuremath{\mathcal Vir}}
\def\ws    {\ensuremath{\Sigma}}
\def\wsD    {\ensuremath{\widehat{\Sigma}}}
\def\wsC    {\ensuremath{\Sigma^C}}
\def\wsCD    {\ensuremath{\widehat{\Sigma^C}}}
\def\wsemb{\iota}
\newcommand\wsf[4]{\ws_{#1#2,#3#4}}
\newcommand\wsSD[5]{S^{2}_{#1;#2#3,#4#5}}

\def\Xs            {X^{*_{}}_{\phantom:}}

\def\zb {\bar z}

\def\k {\ensuremath{\Bbbk}}
\def\Z {\ensuremath{\mathbb{Z}}}

\def\R {\ensuremath{\mathbb{R}}}
\def\CN {\ensuremath{\mathbb{C}}}
\def\Vect{\ensuremath{\mathcal Vect_{\CN}}}
\def\Vectk{\ensuremath{\mathcal Vect_{\k}}}

\def\C{\ensuremath{\mathcal C}}
\def\Ce {\ensuremath{\C \,{\boxtimes}\,\C\rev}}
\def\cD{\ensuremath{\mathcal D}}
\def\cE{\ensuremath{\mathcal E}}
\def\CobC{\ensuremath{3\mbox{-}\mathcal Cob(\C)}}
\def\eC {\ensuremath{\C\rev {\boxtimes}\,\C}}
\def\Obj {\ensuremath{\mathrm{Obj}}}

\def\Rep    {\mathcal Rep}

\definecolor{Acol}       {rgb}  {0.494118,0.015686,0.015686}
\definecolor{Bcol}       {rgb}  {0.098039,0.070588,0.552941}
\definecolor{Brcl}       {rgb}  {0.600000,0.200000,0.200000}
\definecolor{DarkGreen}  {rgb}  {0.000000,0.392156,0.000000}
\definecolor{defcol}     {rgb}  {0.623529,0.062745,1.000000}
\definecolor{DarkViolet} {rgb}  {0.580392,0.000000,0.827450}
\definecolor{ForestGreen}{rgb}  {0.100000,0.408823,0.100000}
\definecolor{green3}     {rgb}  {0.000000,0.803921,0.000000}
\definecolor{mydarkblue} {rgb}  {0.282352,0.239215,0.803921}
\definecolor{OrangeRed}  {rgb}  {1.000000,0.270588,0.000000}
\definecolor{red3}       {rgb}  {0.803921,0.000000,0.000000}
\definecolor{SP1}        {rgb}  {0.4,0.16,0.4}
\definecolor{SP2}        {rgb}  {0.39,0.23,0.47}

\makeindex

\begin{document}
\setcounter{tocdepth}{1}
\numberwithin{equation}{chapter}
\numberwithin{thm}{chapter}

\pagenumbering{roman}
\selectlanguage{english}
\begin{center}
    \Huge
    Hopf and Frobenius algebras in conformal field theory\\
    \vspace{40pt}
    \Large
    PhD thesis\\
    Carl Stigner\\
    Karlstad University 2012
\end{center}
\begin{picture}(240,250)
    \put(0,0){\setulen80
  \put(104,-5)  {\includepicclax3{04}{85}}
  \put(104,-16) {
  \put(186,43)  {\pl {s}}
  \put(194,194) {\pl {\bar s}}
  \put(203,106) {\pl {r}}
  \put(213,140) {\pl {\bar r}}
  \put(165.8,117.5) {\pl {\phi_{\!\beta_3}^{}}}
  \put(179,156) {\pl {\phi_{\!\beta_4}^{}}}
  \put(95,191)  {\pA A}
  \put(148,104) {\pB B}
  \put(171,167) {\pB B}
  \put(157,191) {\pB B}
  \put(119,115) {\pA A}
  \put(207,88)  {\pX Y}
  \put(60.2,43) {\pl {q}}
  \put(50,171)  {\pl {\pb}}
  \put(42,155)  {\pg {\pb}}
  \put(62.8,115.5){\pl {\phi_{\!\beta_2}^{}}}
  \put(74,203.2){\pl {\phi_{\!\beta_1}^{}}}
  \put(89,175)  {\pX X}
  } }
\end{picture}
\vspace{70pt}

\noindent\scshape{Author's note: This is an updated version, with some minor adjustments and typos corrected, of my thesis published at:}\\
\href{http://urn.kb.se/resolve?urn=urn:nbn:se:kau:diva-14456}{
\tt\small http://urn.kb.se/resolve?urn=urn:nbn:se:kau:diva-14456}.
\normalfont
\pagestyle{empty}
\pagebreak
\begin{abstract}
There are several reasons to be interested in conformal field theories in two dimensions. Apart from arising in various physical applications, ranging from statistical mechanics to string theory, conformal field theory is a class of quantum field theories that is interesting on its own. First of all there is a large amount of symmetries. In addition, many of the interesting theories satisfy a finiteness condition, that together with the symmetries
allows for a fully non-perturbative treatment, and even for a complete solution in a mathematically rigorous manner. One of the crucial tools which make such a treatment possible is provided by category theory.

This thesis contains results relevant for two different classes of conformal field theory. We partly treat rational conformal field theory, but also derive results that aim at a better understanding of logarithmic conformal field theory. For rational conformal field theory, we generalize the proof that the construction of correlators, via three-dimensional topological field theory, satisfies the consistency conditions to oriented world sheets with defect lines. We also derive a classifying algebra for defects. This is a semisimple commutative associative algebra over the complex numbers whose one-dimensional representations are in bijection with the topological defect lines of the theory.

Then we relax the semisimplicity condition of rational conformal field theory and consider a larger class of categories, containing non-semisimple ones, that is relevant for logarithmic conformal field theory. We obtain, for any finite-dimensional factorizable ribbon Hopf algebra $H$, a family of symmetric commutative Frobenius algebras in the category of bimodules over $H$. For any such Frobenius algebra, which can be constructed as a coend, we associate to any Riemann surface a morphism in the bimodule category. We prove that this morphism is invariant under a projective action of the mapping class group of the Riemann surface. This suggests to regard these morphisms as candidates for correlators of bulk fields of a full conformal field theory whose chiral data are described by the category of left-modules over $H$.
\end{abstract}
\pagestyle{empty}
\begin{titlepage}
\fancyhf{}
\newpage
\end{titlepage}
\pagenumbering{roman}
\section*{Basis and outline of the thesis}
This thesis is based primarily on the following papers:
\ssubject{Paper \papI.} J. Fjelstad, J. Fuchs, C. Stigner,\\
           {\em {RCFT with defects: Factorization and fundamental world sheets}},\\
           {Nucl. Phys. B} {863} (2012) {213-259},\\
           \href{http://dx.doi.org/10.1016/j.nuclphysb.2012.05.011}
           {\small{\tt{doi: 10.1016/j.nuclphysb.2012.05.011}}},
           {arxiv: \small{\tt1202.3929 [hep-th]}}
\vspace{-6pt}
\ssubject{Paper \papII.} J. Fuchs, C. Schweigert C. Stigner,\\
           {\em {The classifying algebra for defects}},
           {Nucl. Phys. B} {843} {(2011)} {673-723},\\
            \href{http://dx.doi.org/10.1016/j.nuclphysb.2010.10.008}
           {\small\tt{doi: 10.1016/j.nuclphysb.2010.10.008}},
           {arxiv: \small\tt{1007.0401   [hep-th]}}
\vspace{-6pt}
\ssubject{Paper \papIII.} J. Fuchs, C. Schweigert C. Stigner,\\
           {\em {Modular invariant Frobenius algebras from ribbon Hopf algebra automorphisms }},\\
           J. Algebra {363} {(2012)} {29-72},\\
           \href{http://dx.doi.org/10.1016/j.jalgebra.2012.04.008}
           {\small\tt{doi: 10.1016/j.jalgebra.2012.04.008}},
           {arxiv: \small\tt{1106.0210  [math.QA]}}
\vspace{-6pt}
\ssubject{Paper \papIV.} J. Fuchs, C. Schweigert C. Stigner,\\
           {\em {Higher genus mapping class group invariants
            from factorizable Hopf algebras}}, \\
           Preprint:
           {arxiv: \small\tt{1207.6863 [math.QA]}}

\bigskip\noindent
In addition, results from the following papers will appear:
\ssubject{Paper \papV.} J. Fuchs, C. Stigner,
           {\em {On Frobenius algebras in rigid monoidal categories}}, \\
           {Arabian Journal for Science and Engineering} {63} (2008) {175-191}\\
           \href{http://ajse.kfupm.edu.sa/title_new.asp?issue=200812&display}{\small\tt{ajse.kfupm.edu.sa: 33-2C}},
           {arxiv: \small\tt{0901.4886[math.CT]}}
\vspace{-6pt}
\ssubject{Paper \papVI.} J. Fuchs, ,C. Schweigert, C. Stigner,\\
           {\em {The three-dimensional origin of the classifying algebra}},\\ {Nucl. Phys. B} {824} (2010) {333-364} \\
           \href{http://dx.doi.org/10.1016/j.nuclphysb.2009.07.017}{\small\tt{doi: 10.1016/j.nuclphysb.2009.07.017}},
           {arxiv: \small\tt{0907.0685[hep-th]}}
\vspace{-6pt}
\ssubject{Paper \papVII.} J. Fuchs ,C. Schweigert, C. Stigner,\\
           {\em {The Cardy-Cartan modular invariant}},\\ in: {\em{Strings, Gauge Fields, and the Geometry Behind.\ The
            legacy of Maximilian Kreuzer},}
            {A.\ Rebhan, L.\ Katzarkov, J.\ Knapp, R.\ Rashkov, and E.\
            Scheidegger, eds.}
            {(World Scientific, Singapore, 2012)} {p. 289-304},
           {arxiv: \small\tt[1201.4267[hep-th]]}
\vspace{-6pt}
\ssubject{Paper \papVIII.} C. Stigner,\\
           {\em {Factorization constraints and boundary conditions in rational CFT}}, \\
           {Banach Center Publ.} {93} (2011) {211-223},\\
           \href{http://dx.doi.org/10.4064/bc93-0-17}{\small\tt{doi: 10.4064/bc93-0-17}},
           {arxiv: \small\tt{1006.5923 [hep-th]}}
\subsection*{My contribution to the papers}
My contribution to the results of the four papers that form the basis of this thesis is as follows:

\ssubject{Paper \papI} Large parts of the paper were developed in collaboration with the coauthors.
The generalization of the proof of bulk factorization was to a large extent elaborated by me.
\ssubject{Paper \papII} Large parts were developed in discussion with all coauthors. I contributed essential parts to the strategy and most of the results.
\ssubject{Paper \papIII} I was involved in the proof that the bulk Frobenius algebra is symmetric and commutative.  The proof of modular invariance and the generalization to algebras associated to ribbon Hopf algebra automorphisms, as well as the analysis of the ribbon structure of $H$\Bimod, was developed in collaboration with all coauthors.
\ssubject{Paper \papIV} Large parts were developed in discussion with the coauthors. I contributed essential parts to most of the proofs.
\pagebreak
\subsection*{Outline of the thesis}
The thesis is organized as follows: In chapter \ref{chap:intro} we give an introduction to conformal field theory. Many of the notions that will be relevant later are introduced. The discussion in this chapter is partly heuristic.

Chapter \ref{chap:cat} is a summary of the notions from category theory that we need in the rest of the text. Here we also describe relevant aspects of finite-dimensional factorizable ribbon Hopf algebras.

In chapter \ref{chap:alg} we review several classes of algebras that appear in connection with conformal field theory. In particular we introduce the algebras that are the center of attention in this thesis.

Chapter \ref{chap:RCFT} first gives on overview over the construction of correlators of rational conformal field theory with the help of 3d TFT. The rest of chapter \ref{chap:RCFT} summarizes the results of paper \papI, which extends the proof of this construction to world sheets with defect lines.

Chapter \ref{sec:class_alg} describes the classifying algebra for defects, obtained in paper \papII. The classifying algebra is a semisimple associative algebra over the complex numbers that classifies the defect lines of a rational conformal field theory.

Chapter \ref{sec:beyond} is devoted to the results of paper \papIII\ and \papIV\ related to non-semi\-simple conformal field theories. For a so-called factorizable finite ribbon category \C, we associate to each Riemann surface $\Surf gn$ a morphism $\Corrgn$. This morphism takes as an input a symmetric commutative Frobenius algebra in \C. We restrict to a subclass of factorizable finite ribbon categories and consider the category $H$\Mod\ of representations of a finite-dimensional factorizable ribbon Hopf algebra. We obtain a family of symmetric commutative Frobenius algebras in $H$\Mod, and for each such algebra, we prove that the morphism $\Corrgn$ is invariant under an action of the mapping class group of $\Surf gn$. We interpret these morphisms as candidates for correlators of oriented closed world sheets.

\begin{titlepage}
\fancyhf{}
\newpage
\end{titlepage}
\mbox{}
\vspace{30pt}
\section*{Acknowledgements}
\thispagestyle{empty}
First of all I would like to thank my supervisor Jürgen Fuchs for many interesting, helpful and inspiring discussions, for pleasant collaboration, and for generally making my time as a PhD student good in every respect.
\nxl1
Further, many thanks to Christoph Schweigert for enjoyable collaboration, inspiration, the hospitality during my stay in Hamburg, and comments on parts of early versions of this text.
\nxl1
I also want to thank Jens Fjelstad for many helpful discussions, the pleasant collaboration, comments on parts of this text, and the good time in Nanjing.
\nxl1
In addition I want to thank present and former staff of the department of physics and electrical engineering for making the everyday life at the university enjoyable.
\nxl1
Finally, thanks to friends and family for making my life great. In particular I am deeply grateful to Maria and Lova for love and support and for putting up with me during the time I finished this thesis.
\nxl3 Thanks also to Ingo Runkel for comments implemented in this updated version.
\begin{titlepage}
\fancyhf{}
\newpage
\end{titlepage}

\chapter*{Introduction}
\thispagestyle{empty}
A crucial property of physical systems are their symmetries.
By a symmetry, or a symmetry transformation, of a physical system we mean a non-trivial transformation of the system such that some interesting quantity is left unaffected. By making use of the symmetries one can often simplify the analysis of a problem in physics considerably. In fact, it may even be necessary to explore the symmetries in order to be able to solve the problem at all.
Symmetries are encoded mathematically in an algebraic structure, whose precise form depends on the situation one wants to describe.

Quantum field theories defined on two-dimensional surfaces naturally appear in connection with various physical models.
Among these two-dimension\-al theories are the \emph{conformal field theories}, which have conformal transformations as part of their symmetries. One characterization of a conformal transformation is that it leaves the angle between any two vectors invariant.
Two-dimensional conformal field theory appears in rather different areas in physics, ranging from various applications in statistical mechanics and condensed matters systems to string theory.

Conformal field theories form a very special class of quantum field theories in several respects. They are considered on arbitrary \emph{world sheets}. A world sheet is a two-dimensional surfaces which may have a non-empty boundary.
Thus, as opposed to quantum field theory as used in e.g. particle physics, there is no single preferred background space-time, and in particular the classification of boundary conditions is a non-trivial issue.
Second, in the situations we consider in this thesis, the theory satisfies a finiteness condition which, combined with the huge amount of symmetries, allows one to solve the theory completely. This is in sharp contrast with e.g.\ the standard model of particle physics which can largely only be treated in a perturbative manner. In addition, the solution of a (local) conformal field theory on a world sheet can be obtained from a non-local theory on a related surface. The latter type of theory is called a \emph{chiral} conformal field theory. This relationship, known as \emph{holomorphic factorization}, naturally splits the process of solving the physical theory into two parts, the first one dealing only with chiral conformal field theory.

This thesis is concerned with issues related to the second part of the solution of conformal field theory, i.e.\ constructing the full (local) theory from the underlying chiral one. This is a purely algebraic and combinatorial problem that can be
addressed by rigorous mathematical reasoning and by making use of the properties of the \emph{representation category} for the symmetries. One great advantage of approaching the problem in terms of category theory is that we can treat an entire class of conformal field theories simultaneously. It is also worth pointing out that already when considering the algebraic part of the problem alone, one obtains certain physical quantities even without using an explicit solution of the first part of the problem, i.e.\ the underlying chiral theory.

The categories of interest have properties that allow us to perform calculations rigorously in terms of \emph{graphical calculus}. This means that instead of writing mathematical expressions
with the help of symbols, using letters from some alphabet, we use instead elements of a graphical code.
This is not only a convenient, but even a crucial tool for many of the calculations in this thesis.

In addition to being part of solutions to various physical problems, conformal field theory is interesting on its own. It may serve as a "theoretical laboratory" for quantum field theory, which could help us to get a deeper understanding of more general theories.
To justify this idea it is crucial that conformal field theory is not only a mathematically beautiful theory, but that it also appears in physical applications.
Even though the methods and concepts used to obtain a complete solution are related to the very special nature of the conformal field theories, these methods and concepts are potentially of interest in more general situations as well.
Conformal field theory has also given rise to many new developments in mathematics, such as vertex algebras and braided tensor categories.

\pagebreak


\fancypagestyle{plain}{
\fancyhf{} 
\fancyfoot[C]{} 
\renewcommand{\headrulewidth}{0pt}
\renewcommand{\footrulewidth}{0pt}}
\tableofcontents
\pagebreak

\pagenumbering{arabic}
\pagestyle{myheadings}
\fancypagestyle{plain}{
\fancyhf{} 
\fancyfoot[C]{\thepage} 
\renewcommand{\headrulewidth}{0pt}
\renewcommand{\footrulewidth}{0pt}}
\pagenumbering{arabic}
\chapter{A flavor of conformal field theory}\label{chap:intro}
We start with a brief introduction to conformal field theory. The purpose of this chapter is to introduce, in a heuristic manner, a number of concepts which are relevant in conformal field theory, in order to put the questions we address in this thesis into a context. Since the focus in this chapter is on concepts the discussion will from time to time be somewhat over-simplified. Those concept that are central for the results presented in this thesis will be described in more detail in later chapters.
\section{Conformal invariance in $d$ and $2$ dimensions}

A conformal transformation of $\R^d$ is characterized by the property that it preserves the metric up to a scale-factor that may depend on position. Equivalently we can say that a conformal transformation preserves angles between vectors.  For general dimension $d$, any conformal transformation can be written as a combination of translations, scalings, rotations and special conformal transformations. The latter is a translation preceded and followed by an inversion. In $d$ dimensions, these transformations are well defined everywhere, and are accordingly referred to as \emph{global} conformal transformations.
The global conformal transformations form a Lie group with underlying Lie algebra $\mathfrak{so}(d+1,1)$.

In two dimensions the situation is rather special.
Considering the Riemann sphere $\CN\cup\{\infty\}$, the transformation $z\mapsto1/z$ is well defined everywhere. In particular, the group of global conformal transformations is the Möbius group -- the group of automorphisms of the Riemann sphere.
However, the conformal transformations are not exhausted by Möbius group. This is most easily realized by thinking of the Riemann sphere as endowed with a global complex coordinate and recalling that any analytic map constitutes a conformal transformation: The condition for a transformation to be conformal is equivalent to the Cauchy-Riemann equations (see e.g. \cite[Chapter 5.1]{DIms}), and consequently a conformal transformation in two dimensions is equivalent to a complex analytic map.  The global conformal transformations described above are certainly analytic. However there are, in addition to the global transformations, analytic maps with poles of higher order. It follows that the group of conformal transformations that may not we well defined on the entire Riemann sphere is larger than the Möbius group.

Consider now a function $\phi(z)$ of a complex variable. Under an infinitesimal conformal transformation $z\mapsto z'\eq z+\eps(z)$, the function $\phi$ transforms as $\delta\phi=\eps(z)\partial\phi(z)=\sum_n c_n\ell_n\phi(z)$, where the last equality comes from a Laurent expansion of $\eps(z)$ and we have defined $\ell_n:=-z^{n+1}\partial_z$ for $n\In\Z$. Thus the infinitesimal conformal transformations are generated by  $\ell_n$ with $n\in\Z$. These differential operators satisfy the commutations relations
\be\label{Witt_br}
    [\ell_m,\ell_n]=(m-n)\ell_{m+n}\,.
\ee
The complex Lie algebra spanned by elements $\ell_n$, with Lie bracket \eqref{Witt_br}, is called the Witt algebra.

When considering quantum theories, we want
the basis of the Witt algebra to become modes of a current that generates conformal transformations.
The relevant algebra is then a central extension of the Witt algebra -- the \emph{Virasoro} algebra \Vir, see e.g.\ \cite[Section 3]{Sche2}. A basis of \Vir\ is given by $\{L_n|n\In\Z\}\cup C$, where $C$ is a central element, and the non-vanishing Lie brackets are
\be\label{Vir_br}
    [L_m,L_n]=(m-n)L_{m+n}+\frac C{12}m(m^2-1)\delta_{m+n,0}\,.
\ee
Note that from \eqref{Vir_br} it follows directly that the Virasoro algebra has a subalgebra with basis $\{L_{-1},L_{0},L_{1}\}$ and this subalgebra is $\mathfrak{sl}(2)$.
Upon a detailed analysis of the action of this subalgebra in conformal field theory, one recovers \cite[Section 2.8]{Sche2} the real form $\mathfrak{so}(3,1)$, i.e.\ the Lie algebra of the global conformal transformations.
Since we will exclusively work in two dimensions we will drop the epithet "two-dimensional" and refer to a \emph{two}-dimensional quantum field theory whose symmetry algebra contains the Virasoro algebra simply as conformal field theory, abbreviated as CFT\index{CFT}.

\section{Chiral and full CFT}\label{sec:chi_full}
Let us now set the stage for the considerations in this thesis. We will be interested in conformal field theories defined on world sheets\index{world sheet}. A \emph{world sheet} is a two-dimensional compact smooth manifold, possibly with non-empty boundary, equipped with a conformal structure with Euclidean signature, i.e. an equivalence class of metrics with respect to local rescalings. Note that we allow the world sheet to be non-orientable. In addition we require that world sheets can be equipped with  the following structures:
\pagebreak
\begin{itemize}\addtolength{\itemsep}{-6pt}
    \item Boundary conditions
    \item Defect lines
    \item Field insertions
\end{itemize}
In order to give a somewhat more detailed description of these structures it is instructive to describe a couple of applications of conformal field theory.
\subject{Conformal field theory in physics}
An important setting for CFT is the description of critical phenomena in statistical mechanics. A standard example is the two-dimensional Ising model. This is a model that describes the interaction  of spin variables positioned on the sites of a square lattice in some domain. The system is characterized by two length scales, the lattice spacing $a$ and the correlation length $\xi$, which describes the distance at which the interaction is effectively non-zero. Typical observables in the theory are local magnetization (the value $s(r)=\pm1$ of the spin variable at site $r$) and the energy density. The \emph{scaling limit} is obtained by taking $a\rightarrow0$ while keeping the domain and $\xi$ fixed. In the scaling limit one obtains a continuous two -dimensional CFT that describes the long-range behavior of the system, see e.g.\ \cite{card33} for more details.

A similar phenomenon, for which CFT plays a role is \emph{critical percolation}. Percolation describes the statistical process of a liquid that seeps through a porous medium. One of the simplest mathematical models for percolation is obtained by taking a square lattice, with lattice spacing $a$, such that each site can be either open, with some probability $p$, or closed, with probability $1{-}p$.
The \emph{crossing probability} of such a lattice of size $n\times n$ in some domain $\mathcal D$ is the probability $\pi_n(p)$ that there is at least one path, say from left to right, such that all sites along this path are open.

Just like for the Ising model, CFT is relevant for the  understanding the critical behavior of such systems in the scaling limit, i.e.\ when we take the limit $a\rightarrow0$ while keeping $\mathcal D$ fixed. For percolation, the critical point is a distinguished value of $p$ denoted by $p_c$. In the scaling limit, one finds $0<\pi_n(p_c)<1$ at $p\eq p_c$ while $\pi_n(p)$ approaches $0$ for $p<p_c$ and $1$ for $p>p_c$. The critical behavior can be described geometrically as follows: For small $p$ the average size of clusters of open sites is small. However, at $p=p_c$ the average size of the clusters diverges and for $p>p_c$ there is a finite probability that a given site belongs to an infinitely large cluster. See e.g.\ \cite{card19,HEnk,laps} for more details on percolation and CFT.

CFT is also an important part of string theory, which is a candidate for a quantum theory unifying the four fundamental forces. The fundamental objects in string theory are one-dimensional strings which move in space-time. While moving, the string sweeps out a two-dimensional surface -- the world sheet. One important ingredient is a CFT defined on this world sheet, see e.g. \cite{POlc}.

In addition there are condensed matter systems that are effectively $2$- or $1\,{+}\,1$-dimensional, giving rise to theories with approximate conformal symmetry. This happens e.g.\ when studying the Kondo effect \cite{affl7} or quantum Hall fluids.

It is also worth pointing out that a generic CFT does not come with a Lagrangian formulation. However, there is an important class of models, the \emph{WZW-models}, that come from the quantization of a Lagrangian. In that case, the fields take values in a group manifold and symmetries of this theory are encoded in what is known as a current algebra, which is based on an affine Lie algebra.
\subject{The world sheet}
We are now in position to describe the structures of the world sheet in some more detail
\ssubject{Boundary conditions:} In the case of non-empty boundary, we need to assign a boundary condition to each boundary component, describing the behavior of bulk fields close to the boundary. The classification of consistent boundary conditions is far from trivial. A consistent boundary condition always respects the conformal symmetry. However, we will usually require that the boundary conditions respect all the symmetries of the theory. There are examples of situation where the boundary conditions have a simple formulation:
\begin{itemize}
\item In lattice spin models, such as the Ising model, some boundary conditions
    can be obtained by imposing specific values for the spin variables on the sites at the boundary.
\item In theories with a Lagrangian formulation, such as WZW-models, some boundary conditions may be described in terms of Dirichlet or Neumann conditions, i.e.\ by prescribing the value of the fields or their derivatives that appear in the Lagrangian on the boundary.
\end{itemize}
However, in general boundary conditions cannot be described this way. Even if some of them can, they typically do not exhaust all allowed boundary conditions. Consider e.g.\ the 3-states Potts model, which is a lattice model with three allowed values for each spin variable. In this model there are totally 8 conformal boundary conditions \cite{afos,fuSc9}. One them has an interpretation of a free boundary condition, three of them as fixed directions and three of them as mixed spin directions which means that one direction is forbidden. The eighth one, constructed in \cite{afos} does not have such a nice interpretation.
\pagebreak
\ssubject{Defect lines:}
The world sheet may be divided into several regions such that the CFT on the two sides are related in a sense that will be described later, but which may differ e.g. with respect to the field content and allowed boundary conditions. We say that the CFTs on the two sides of the line may be in different \emph{phases}.
The two regions are separated by a \emph{defect line}\index{defect!line}, which may be thought of as a two-sided boundary. The behavior of fields in the vicinity of the defect line is described by assigning a defect condition to each defect line.

We will be interested in \emph{topological\index{defect!line!topological}} defect lines, which in particular means that the defect lines can be deformed without affecting the physics of the theory.
Whenever two regions are in the same phase there is a particular defect condition -- the trivial defect condition. This means that the defect line can be omitted and the two regions can be considered as one. Accordingly, we refer to such defect lines as \emph{trivial} or \emph{invisible}. Trivial defects can be added or removed freely from the world sheet. It is convenient to  always require a minimal number of defects, which however may be trivial.
Topological defects arise naturally e.g.\ in statistical systems
modeling condensed matter systems of physical interest, see e.g.\ \cite{savi,osaf2}. The general description of defect lines will be continued in section \ref{top_def}.

\ssubject{Field insertions:} There are field insertions in the bulk as well as on the boundary of the world sheet. We will refer to fields inserted in the bulk, on the boundary or on a defect line as bulk, boundary and defect fields, respectively.
Boundary components and defect lines may be partitioned into several segments. Boundary and defect fields are inserted on the junction between to such segments and may change the boundary or defect condition.
We will sometimes consider a bulk field as a defect field inserted on an invisible defect. A defect field that changes the defect condition to or from the trivial one is referred to as a \emph{disorder field}. In other words, by disorder field we mean a field at which a non-trivial defect starts or ends.

The interpretation of the field insertions depend on the application.
In statistical models the field insertions may e.g.\ be spin variables or disorder fields at which frustration lines may end, and in string theory the field insertions are interpreted as asymptotic string states.

\subject{Full CFT}
We regard a conformal field theory as solved if we are able to obtain all correlation functions. That is, for any world sheet, give a linear map from the space of fields to the complex numbers that satisfies all consistency conditions of the theory. Of course, explicit calculation of all correlation functions would require an infinite amount of calculations. A more reasonable goal is to establish the existence of all correlation functions and give algorithms from which they can be computed.
At least for an important class of theories, known as \emph{rational\index{CFT!rational}} CFT, this has been achieved. In particular this approach turns part of the problem into a finite one. There is an algorithm describing how any correlation function can be calculated from a finite number of quantities. The purpose of the rest of this section is to explain this approach in some more detail.

We will be interested in \emph{full\index{CFT!full}, local} conformal field theory. By this we mean a conformal field theory defined on world sheets, as defined above, with single-valued correlation functions. It is crucial to distinguish full CFT from \emph{chiral\index{CFT!chiral}} CFT. The latter, which will be discussed in somewhat more detail below, is instead defined on complex curves and the correlation "functions" are not necessarily single-valued. We will soon return to the discussion of chiral and full CFT, but first we introduce some concepts that will be useful in the discussion.

Consider two fields close to each other. Intuitively, one may argue that from a point very far away it is not possible to distinguish the two fields from a single one.
However, for the manifolds of our interest, the notions of "close to each other" and "far away" are not well defined. A conformal transformation may separate points to an arbitrary distance and move a point far from away close to one of the insertion points.
Nevertheless, this suggests the existence of a product of fields giving rise to a superposition of fields inserted at a single point.
This is the \emph{operator product expansion} (OPE for short), controlling the short distance behavior of fields. The OPE relates $n\,{+}\,1$-point functions to $n$-point functions. That is, for two field insertions $\phi_{a}(z)$ and $\phi_{b}(w)$, such that $|z-w|$ is much smaller than the distance between any other pair of fields in a correlator, we have
\be\label{OPE_corr}
    \langle\phi_{a}(z)\,\phi_{b}(w)\cdots\rangle=\sum_{c}C_{abc}(z,w)\,\langle\phi_c(w)\cdots\rangle\,.
\ee
This is the concrete meaning of the OPE. It is customary to write the OPE \eqref{OPE_corr} in terms if the following shorthand notation:
\be\label{OPE}
    \phi_{a}(z)\,\phi_{b}(w)\sim\sum_{c}C_{abc}(z,w)\,\phi_c(w)\,.
\ee

There is a separation (which will be discussed in more detail below) of the fields into \emph{primary} fields and their \emph{descendants}.
Conformal symmetry strongly restricts the dependance of $C_{abc}$ on $z$ and $w$. E.g.\ translation invariance implies that $C_{abc}$ only depends on $z\,{-}\,w$. A more detailed analysis shows that, for rational CFTs, the OPE of two primary fields has the form
\be\label{OPE_prim}
\phi_{i}(z)\phi_{j}(w)\sim\sum_{k}c_{ijk}(z-w)^{(\Delta_k-\Delta_i-\Delta_j)}(\zb-\bar w)^{(\bar\Delta_k-\bar\Delta_i-\bar\Delta_j)}\phi_k(w)+\text{desc},
\ee
where the summation over $k$ is over the primary fields, the coefficients $c_{ijk}$ are constants, and "desc" stands for the terms involving the descendants. The numbers $\Delta_p$ and $\bar\Delta_p$ are the \emph{conformal weights} of the fields and describe how the fields behave under scaling transformations.

The fields with conformal weight $\bar\Delta_p\eq0$ constitute the holomorphic fields, and the fields for which $\Delta_p\eq0$ constitute the anti-holomorphic fields. The holomorphic fields close under the OPE and thus furnish a closed "algebra" of fields -- the \emph{chiral algebra\index{chiral algebra}}. Analogously the anti-holomorphic fields generate the anti-chiral algebra. In non-rational CFTs, it is no longer guaranteed by the chiral symmetries that all coefficients  in the expansion \eqref{OPE_prim} are powers of $z-w$. However, the holomorphic fields still give rise to a chiral algebra.

Note that the expansion coefficients in the chiral algebra are not numbers but rather functions of the insertion points. Thus, the chiral algebra is not an algebra in the conventional sense. However, this property (and a lot more) can be taken care of in the structure of a conformal  \emph{vertex algebra}\index{vertex algebra}, also called \emph{chiral algebra}, which will be discussed in more detail in section \ref{subsec_VOA}. One may consider theories in which the anti-chiral algebra is different from the chiral algebra. However, we will always assume that the chiral algebra is the same as the anti-chiral.

A conformal vertex algebra furnishes in particular a representation of the Virasoro Lie algebra. As  a consequence, there will be two copies of \Vir\ acting on the fields, one from the chiral and one from the anti-chiral algebra. The two conformal weights $\Delta_p$ and $\bar\Delta_p$ mentioned above describe the transformation with respect to the actions of these two copies.
The central element $C$  of \Vir\ acts as $c\,\id$ for some $c\In\CN$ in all  representations of the chiral algebra. The number $c$ is called the \emph{central charge} of the theory.

The OPE coefficients and correlation functions involving descendants can be obtained from the ones involving only primary fields via the chiral symmetry.
In the case that the chiral algebra is built from the Virasoro algebra alone, the primary fields can be labeled by highest weight representations of \Vir.
However, in many situations we are interested in theories for which the chiral algebra is larger than the one associated with \Vir.
In general, there need no longer be a concept of highest weight state for a representation of the chiral algebra. However, there is still a subset of the fields with the property that all OPE coefficients and correlation functions can be obtained from the ones involving only this subset via the chiral symmetry. In the rational case, these fields, to which we below refer to as \emph{primary fields}, are in bijection with the irreducible representations (counting multiplicities) that appear in the state space.

\pagebreak
\subsection{Chiral conformal field theory}\label{sec:chiral}
Studying chiral CFT\index{CFT!chiral|textbf} involves in particular studying the representation theory of conformal vertex algebras.
Here we mention a few of the important properties of conformal vertex algebras. A full definition will be given in section \ref{subsec_VOA}. The data of a vertex algebra contain a graded vector space $\VA$ and a \emph{vertex operation}, $Y(\cdot,z):\VA\rightarrow\End(\VA)\llbracket z^{\pm1}\rrbracket$, that assigns to each vector $a\In\VA$ a formal power series $Y(a,z)$ with coefficients in $\End(\VA)$. These power series are used to describe the fields in the chiral theory mathematically. The Laurent coefficients in the expansions constitute the \emph{modes} of the fields. In a conformal vertex algebra there is a distinguished vector $T\In\VA$, the \emph{conformal vector}, whose modes furnish a representation of the Virasoro algebra. The minimal chiral algebra is generated by $Y(T,z)$ alone. The full CFTs with chiral algebras of this form are the \emph{minimal models}. However, typically we deal with \emph{extensions}, in which case the chiral algebra is generated by additional fields.

Chiral CFT is defined on a complex curve \CC, i.e.\
a closed surface with a complex structure. A complex structure is equivalent to having an orientation together with a conformal structure, i.e.\ an equivalence class of metrics with respect to local rescalings.
The correlation "functions" of a chiral CFT are called \emph{conformal blocks}. We will see below why here the term "function" is put in quotation marks. For the results in this thesis the existence and some properties of the conformal blocks are crucial, but we will not need a general definition of the conformal blocks.

\subject{Conformal blocks}
\index{conformal blocks|textbf}Consider a complex curve \CC\ of genus $g$ and with $n$  marked points $p_1,...,p_n$ such that $p_i$ is labeled by the representation $R_i$ of \VA. For now, it is enough to think of a marked point as a distinguished point on \CC. It is common to think of the marked points as small open discs cut out from the world sheet. For this reason we sometimes refer to marked points as \emph{holes}.
The space of conformal blocks on \CC\ is a subspace
\be\label{CBL_subspace}
    \cbl(\CC)\subset (R_1\otiC \cdots\otiC R_n)^*
\ee
of the space dual to $R_1\otiC\cdots\otiC R_n$. The dependence of the conformal blocks on the positions of the marked points enters in a coinvariance condition that determines the subspace $\cbl(\CC)$ in \eqref{CBL_subspace}, see e.g.\ \cite[Chapter 9]{BF}.

Instead of giving the general definition of the space of conformal blocks we illustrate how the construction, via the coinvariance condition mentioned above, works in a WZW-model. A WZW-model
is built from an affine Lie algebra $\hat\g$, at a fixed integer value of the level $k$.
Let \CC\ be a complex curve of genus $g$ and with $n$ marked points $(p_1,...,p_n)\equiv\vec p$, labeled by (irreducible) representations $V_1,...,V_n$ of $\hat\g$. Consider the algebra $\alghf\oti\g$ of $\g$-valued functions that are holomorphic on $\CC\setminus\{p_1,...,p_n\}$ and meromorphic on \CC.
There is a Lie algebra homomorphism \cite{Bea,BAki}
\be
    \alghf\oti\g\rightarrow U(\hat g)_k\oti\cdots\oti U(\hat g)_k\,,
\ee
where $U(\hat g)_k$ is the universal enveloping algebra of $\hat g$ at level $k$. This defines an action of $\alghf\oti\g$ on  the tensor product $V_1\oti...\oti V_n$ of $\hat\g$-modules. Next, define $V_{\vec{n}}(\CC,\vec p)$ to be the quotient of $V_1\oti...\oti V_n$ that is coinvariant under the action of $\alghf\oti\g$:
\be
    V_{\vec{n}}(\CC,\vec p):=V_1\oti...\oti V_n\Big/\alghf\oti\g\,(V_1\oti...\oti V_n)\,.
\ee
The space of conformal blocks, $\cbl(\CC)$, is defined to be the dual space
\be
    \cbl(\CC):=V_{\vec{n}}(\CC,\vec p)^*\cong\Hom_{\alghf\oti\g}(V_1\oti...\oti V_n,\one)\,,
\ee
where $\one$ is the trivial $\alghf\oti\g$-module.

So far we restricted our attention to a fixed conformal structure and fixed locations of insertion points. When varying these moduli, the conformal blocks fit into a vector bundle, the bundle of conformal blocks\index{conformal blocks!bundle of}, over the moduli space $\hat{\mathcal M}_{g,n}$ of curves \CC. The fiber over each point in $\hat{\mathcal M}_{g,n}$ is the corresponding \emph{space} of conformal blocks. The term conformal blocks is not only used for the vector bundle, but also for locally defined sections in the vector bundle. The bundle of conformal blocks can be equipped with a projectively flat connection, the Knizhnik-Zamolodchikov connection \cite{frsh2}. By studying the holonomies of this connection the space of conformal blocks can be  equipped with a projective action of the fundamental group of $\hat{\mathcal M}_{g,n}$, i.e. of the mapping class group\index{mapping class group} of $\CC$. This fact will be crucial when solving a full CFT, and we will return to it later.

A bundle of conformal blocks is generically not trivial, i.e.\ it is not globally a product of the base and the fiber. In particular, the monodromies are generically non-trivial, and as a consequence there are no global sections, i.e.\ the conformal blocks do not give single-valued functions of the moduli of \CC.

\subsection{Rational conformal field theory}
\index{CFT!rational}The class of full CFTs that is best understood mathematically is the important (and particularly well-behaved) class  known as \emph{rational conformal field theories}.
For a \emph{rational} CFT,  the chiral algebra, \VA, is a \emph{rational} vertex algebra, see e.g.\ \cite[definition 1.2.4]{zhu3} for a definition of this notion. The key feature is that a rational vertex algebra has a \emph{finite} number of (isomorphism classes of) simple representations and any representation is a finite direct sum of the simple ones. The primary bulk fields are in the rational case labeled by pairs $(i,j)$ of representatives of isomorphism classes of simple representations of $\VA$. A consequence of the nice structure of the representation theory of \VA\ is that, for each pair $(i,j)$, the space of primary bulk fields is \emph{finite}. In addition, any boundary or defect condition can be written as a superposition of a finite set of irreducible boundary or defect conditions. Thus there is a finite number of OPE's and correlation functions \footnote{Even after restricting to primary fields there is in principle still an infinite number of correlation functions to calculate, considering e.g. different topologies and an arbitrary number of field insertions. However, due to the factorization constraints, which will be described in section \ref{sec:Cons_cond}, all of them can be expressed through a finite number of fundamental correlation functions.} to calculate.

\subsection{Full conformal field theory and holomorphic factorization}\label{fullCFT}
\index{CFT!full|textbf}\index{holomorphic factorization|textbf}Despite their differences chiral and full CFT are not unrelated: Full CFT can be obtained from an underlying chiral CFT, via the principle of \emph{holomorphic factorization}. This is sometimes phrased by saying that the correlation functions of the full CFT are obtained by combining two "chiral halves". This is done in such way that the correlation functions are single-valued and satisfy all consistency conditions. Holomorphic factorization has only been worked out and proven in full generality for rational conformal field theory. However, it is generally believed that holomorphic factorization works also beyond the rational case.

In the rational case, the notion of holomorphic factorization can be made precise: Given a world sheet \wsC, the \emph{complex double}\index{double of a world sheet|textbf}, \wsCD, is obtained from \wsC\ as the orientation bundle over \wsC, modulo an identification of points over the boundary of \wsC:
\be\label{double}
		\wsCD:=\text{Or}(\wsC)\;/\,\sim\,,\quad(x,\Or)\sim(x,-\Or)\quad\forall x\In\partial \wsC\,.
\ee
Thus, e.g.\ the double of a disc is a sphere and the double of a torus consists of two disjoint tori with opposite orientation.
The double is a complex curve and consequently one can define a chiral CFT on \wsCD. In this setting, holomorphic factorization amounts to the statement that:
\be\nonumber
    \text{\emph{The correlation function of \wsC\ is an element in the space of conformal blocks on \wsCD.}}
\ee

When constructing the double, a point $p$ in the interior of $\wsC$ gets mapped to two points $p'$ and $p''$ on the double. By construction, the double admits an orientation reversing involution $\sigma$, that interchanges the two points on $\wsCD$ over each point on $\wsC$. Note that, in the chiral CFT on the double, the points $p'$ and $p''$ over $p$ can be varied independently. However, when constructing the full CFT from the chiral one, we require that the element in the space of conformal blocks that appears in the correlation is such that
\be\label{Insertion_double}
    p''=\sigma(p')\,.
\ee
This way local coordinates around $p'$ and $p''$ are related.

Holomorphic factorization splits the solution of full CFT into two parts:
\begin{enumerate}
    \item Study the chiral conformal field theory on the double \wsCD\ to obtain the space $\cbl(\wsCD)$ of conformal blocks.
        That involves in particular solving "Ward identities". These implement the local chiral symmetries globally on the double and thus also on the world sheet.
    \item For a given world sheet \wsC, determine the actual correlation function of the full CFT on \wsC\ as an element of $\cbl(\wsCD)$ .
\end{enumerate}
The strategy to solve the second part of the problem is to find a vector in the relevant space of conformal blocks that satisfies all consistency conditions of the theory. In the case of rational conformal field theory this problem has been worked out completely \cite{fuRs4,fuRs8,fuRs10,fjfrs,fjfrs2,fjfs} in the \emph{TFT-construction of rational CFT}. The TFT-constructions is used in this thesis to study rational CFT; it is described in section \ref{sec:TFT}.

The second part of the solution is a purely algebraic consideration. For this reason it is convenient to work not directly with the vertex algebra \VA\ itself, but rather with the \emph{category of representations}, $\Rep(\VA)$, of \VA. In fact we will not work directly in $\Rep(\VA)$ but rather in an abstract category that shares the properties of $\Rep(\VA)$. Relevant notions from category theory are described in detail in the next chapter. Working in an abstract category instead of with the vertex algebra itself has two major advantages. First of all it allows us to forget about a lot of the intricate structure of the vertex algebra. The representation category remembers exactly the information relevant to solve the problem. Second, and not less important, we are able to treat an entire class of CFTs simultaneously.

Below we will only be interested in questions related to the second part of the solution as described above, i.e.\ the problem to select the correlator out of a space of conformal blocks, in rational as well as non-rational theories. This means that we are able to suppress the conformal structure and consider \emph{topological world sheets\index{world sheet!topological}}, i.e.\ we think of the world sheet as a topological manifold.
As a consequence we consider also the double as a topological manifold and suppress the conformal structure.

It is worth pointing out that even when considering the second problem alone, without having solved the chiral theory, one still obtains non-trivial relevant physical information, such as OPE coefficients.

\subsection{Consistency conditions}\label{sec:Cons_cond}
In this subsection we describe the two types of non-chiral consistency conditions, \emph{mapping class group invariance} and \emph{factorization constraints}, that we require the correlators to satisfy.
\subject{Mapping class groups and modular transformations}
The mapping class group\index{mapping class group|textbf} $\Map(\ws)$ of a world sheet $\ws$ consists of the homotopy classes $[f]$ of homeomorphisms $f:\ws\rightarrow\ws$ that preserves the decorations (i.e. labels of field insertions and defect and boundary conditions), and for oriented world sheets, the orientation.
The correlators are required to be invariant under an action of $\Map(\ws)$. This guarantees that the correlators are indeed single-valued functions. Let us explain how this works in a bit more detail.

Recall from section \ref{sec:chiral} that conformal blocks fit into a (generically non-trivial) vector bundle over the moduli space, and as a consequence a conformal block is in general multi-valued with respect to the moduli of $\wsD$\index{double of a world sheet}. This is reflected in a non-trivial action of the fundamental group $\pi_0(\hat{\mathcal M}_{g,n})$ of the moduli space of $\wsD$ on the space $\cbl(\wsCD)$. The fundamental group $\pi_0(\hat{\mathcal M}_{g,n})$ is the same as the mapping class group $\Map(\wsD)$ of the double, see e.g.\ \cite[Theorem 6.1.13]{BAki}.

For correlators, the relevant group is the subgroup of $\Map(\wsD)$ that commutes with the orientation reversing involution $\sigma$, which relates points pairwise on the double, c.f \eqref{Insertion_double}. This subgroup, sometimes called the relative modular group, is isomorphic to the mapping class group $\Map(\ws)$ of the world sheet. Invariance under the action of $\Map(\ws)$ therefore assures that the correlators are single-valued.

Let us describe a concrete example. Consider the torus without any field insertions. The mapping class group of this world sheet is the modular group $PSL(2,\Z)$.
Any torus can be described as a quotient of the upper half plane by identifying points that differ by integer combinations of two vectors. Due to conformal invariance the CFT is invariant under scalings and rotations of this torus.
Thus we can take the vectors to be $\tau$ with $\text{Im}(\tau)>0$ and $1$.
The group $PSL(2,\Z)$ has a presentation  by 2 generators $S$ and $T$ modulo the relations $S^2=1$ and $(ST)^3=1$. In terms of the modular parameter $\tau$, $S$ and $T$ acts as $T:\tau\mapsto\tau+1$ and $S:\tau\mapsto-\frac1\tau$, see e.g.\ \cite{DIms} or \cite{Sche2} for more details. One can show that these transformations leave the conformal structure invariant.

\subject{Factorizations constraints}
\index{factorization!constraints|textbf}
The sewing constraints, formulated in \cite{sono2,sono3,lewe3}, relate world sheets of different topology. They require that when sewing world sheets together, the correlator of the so obtained world sheet can be expressed as a sum over the correlators of the various types of disjoint world sheets that are sewn together. In particular, one can obtain one and the same world sheet by sewing together world sheets in several distinct ways. One can, e.g.\ obtain a four-point function on the sphere by sewing together two three-point functions in three distinct ways, see \cite[Fig. 9(a)]{lewe3}.
The sewing constraints give relations between these different ways of sewing.

Equivalently to sewing we can consider the notion of \emph{factorization}, which means that instead of sewing we give a prescription on how to cut a world sheet up into smaller pieces, such that the correlator of the original world sheet can be written as a sum over the correlators of the so obtained world sheets. Factorization is described in detail in \cite{fjfrs}.

There are two types of factorization: \emph{bulk} and \emph{boundary} factorization.
Bulk factorization takes place in a cylindrical region of the world sheet and in short it works as follows in a rational CFT: A world sheet is cut along a circle embedded in the cylindrical region. This results in two holes  in the world sheet, which are then closed by gluing a disc with one bulk field to each hole, see picture \eqref{dis_sphere} below. The bulk fields that we place on the two discs are related by a non-degenerate pairing on the space of bulk fields. This pairing comes from the two-point function on the sphere.
Thus, factorization gives rise to a collection of world sheets that are labeled by bulk fields.
We will refer to world sheets obtained this way as \emph{factorized world sheets}.
It is sometimes useful to have the following picture of factorization in mind: Picture the cylindrical region, in which we cut, as a \emph{tube} (as in picture \eqref{dis_sphere}). Factorization can then be thought of as "squeezing" the tube to smaller and smaller diameter until it breaks.
Boundary factorization works similarly, but the factorization is performed along a line segment embedded between two boundary components.

The factorization constraints should be thought of as a concrete realization of the notion of inserting a complete set of states. The factorization identity involves a sum over all primary bulk fields. This may be thought of as summing over the entire state space.

In \cite{fjfrs} it is proven that
the correlators of world sheets without defects, obtained via the TFT-construction, satisfy the factorization constraints. In \cite{fjfs} the proof is extended to orientable world sheets with an arbitrary network of defects. The factorization identity, Theorem \ref{thm:bulkfac}, states that summing over the correlators of the collection of factorized world sheets, with suitable coefficients, we obtain the correlator of the world sheet we started with. Bulk factorization is described in more detail in section \ref{sec:bulkfact}

The factorization identities drastically reduce the number of correlators that have to be calculated in order to know any correlator. In  a rational CFT, the correlator of any oriented world sheet can be written as a sum involving correlators of world sheets of only three types \cite{fjfs}:
\begin{itemize}\addtolength{\itemsep}{-6pt}
    \item Three defect fields on the sphere
    \item Three boundary fields on the disc
    \item One boundary and one disorder field on the disc
\end{itemize}
Any correlator is determined, via the factorization identity, by calculating all possible correlators of these types involving only primary fields. In addition, due to semisimplicity it is enough to consider the irreducible boundary and defect conditions. Thus, as explained in \cite{fjfs}, there is a \emph{finite} number of fundamental world sheets\index{world sheet!fundamental} from which any correlator can be obtained algorithmically.

Historically, there have been attempts to solve full CFT by picking one of the constraints and search for all solutions that satisfy that constraint.
For example searching for modular invariant torus partition functions, in order to find bulk state spaces of consistent conformal field theories, has led to what is known as the ADE-classification of torus partition functions for the $\mathfrak{su}(2)$-WZW-models \cite{caiz2}.
However, considering modular invariance alone is far from enough. One can write down modular invariant partition functions which are not part of any consistent conformal field theory. The appearance of such spurious solutions should not come as a surprise. There could very well be modular invariant partitions functions that satisfy one consistency condition but not all. The set of constraints is highly over-determined and it is a priori not clear that a solution exists at all. One advantage of the TFT-construction is that it allows one to treat the whole set of constraints simultaneously.

\section{Topological defect lines}\label{top_def}
\index{defect!line!topological|textbf}Recall from section \ref{sec:chi_full} that topological defect lines are characterized by the fact that they can be deformed without affecting the correlator, as long as the defect is not taken across a field insertion or another defect line. Technically that means that the defects we are interested in are \emph{tensionless} in the sense that the holomorphic and the anti-holomorphic  components, $T$ and $\overline T$, of the energy-momentum tensor are continuous across the defect line. In fact we  require more.
We require that both the holomorphic  and the antiholomorphic components of all generating currents of the chiral algebra are continuous across the defect line. Below "defect line" or simply "defect" will always refer to a defect line that is topological in this strong sense.

Even though some aspects of defect lines mentioned below are valid in a more general setting we will describe defects in rational CFT.
In particular we will use that in a rational CFT there is a finite set of simple defect conditions and any defect condition can be written as a superposition of the simple ones.
An $A$-$B$-defect separates two regions such that the region to the left (right) of the defect is in phase $A$ ($B$):
 \eqpic{def_line}{118}{40}{\setulen 75
  \put(0,7) { \includepic{30}{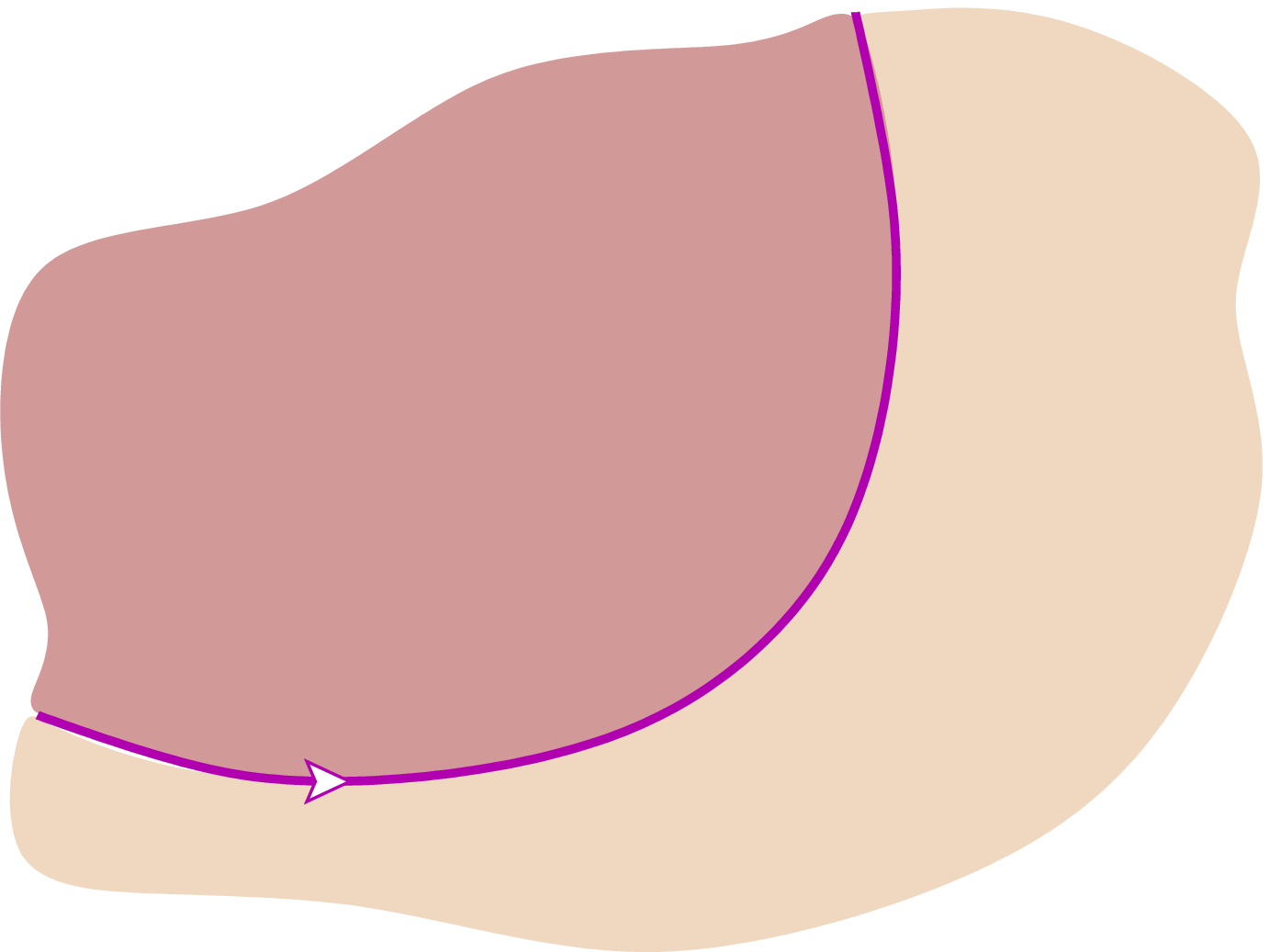}
  \put(115,80) {\pX X }
  \put(72,90)  {\pA {\fbox{$ A $}}}
  \put(124,60) {\pB {\fbox{$ B $}}}
  } }
We refer to such a defect line, with defect condition $X$, simply as the \emph{defect} $X$.
There is a notion of the dual $X^\vee$ of an $A$-$B$-defect  $X$. $X^\vee$ is a $B$-$A$-defect obtained from $X$ by reversing its orientation.

We will be interested in world sheets with arbitrary networks of defects. Thus, several defects may be joined at a single point or one or more defects may end on a single point on the boundary. We refer to such points as a \emph{network vertices\index{network vertex}}\label{def_net-vert}.
\subject{Defect operators}
\index{defect!operator}
A defect line gives rise an operator between spaces of bulk fields.
Consider a bulk field $\phi$ in phase $A$ encircled by an $A$-$B$-defect $X$.
Since the defect line is topological it can be shrunk to an arbitrarily small radius $\eps$. In the limit $\eps\rightarrow 0$, we obtain a bulk field $D_X(\phi)$ in phase $B$. Thus, the defect $X$ gives rise to a linear map  $D_X$ between the spaces of bulk fields in phase $A$ and $B$ separated $X$. This is illustrated in the following picture:
 \Eqpic{defop}{320}{30}{\setulen 90
  \put(0,7) { \includepic{24}{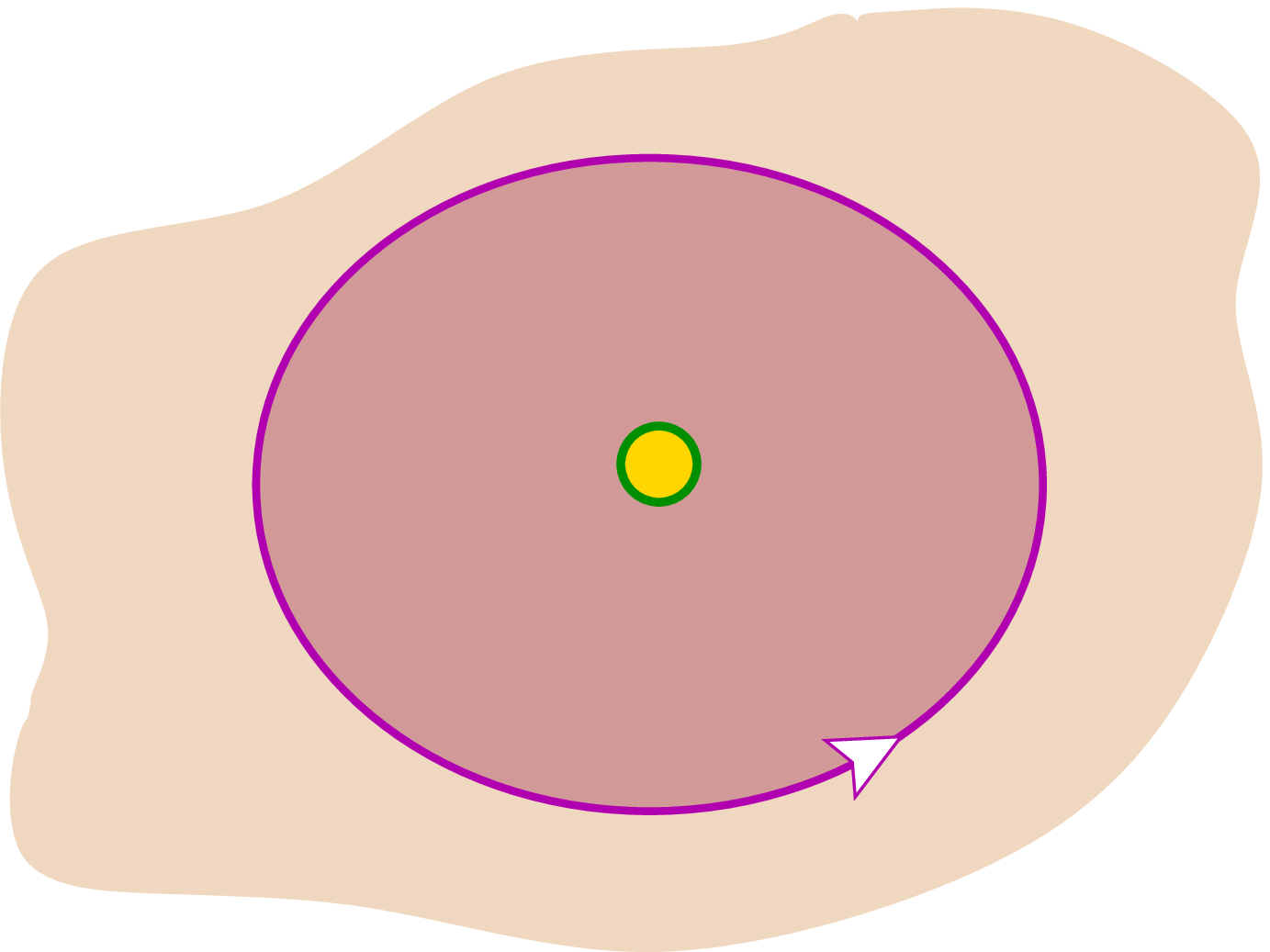}
  \put(83,17) {\pX X }
  \put(55,22)  {\pA {\fbox{$ A $}}}
  \put(61,39) {\pl { \phi}}
  \put(85,62) {\pB {\fbox{$ B $}}}
  }
  \put(115,40)  {$\mapsto$}
  \put(132,7) { \includepic{24}{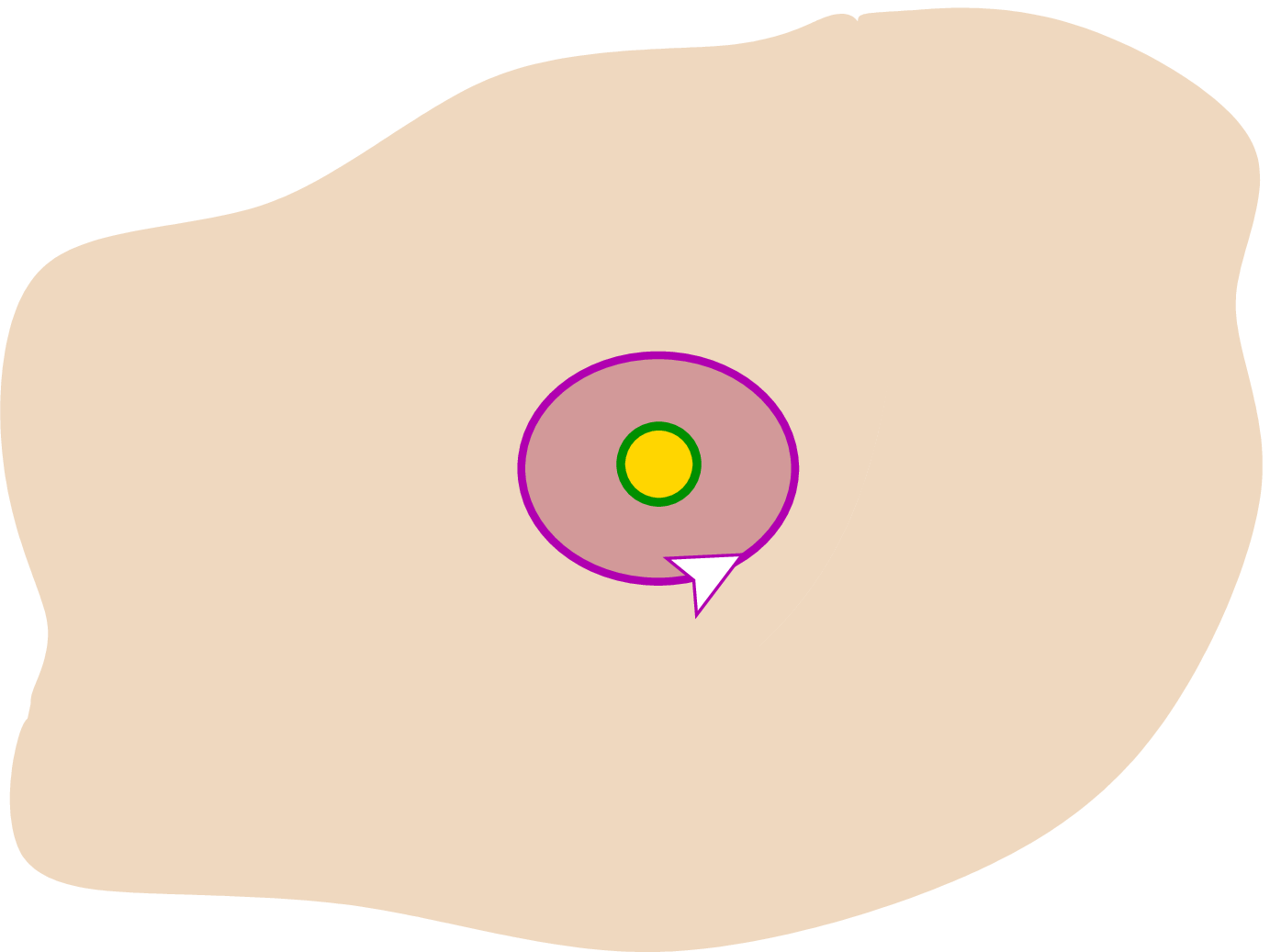}
  \put(60,39) {\pl { \phi}}
  \put(64,27) {\pX X }
  \put(85,62) {\pB {\fbox{$ B $}}}
  }
  \put(253,40)  {$\mapsto$}
  \put(270,7) { \includepic{24}{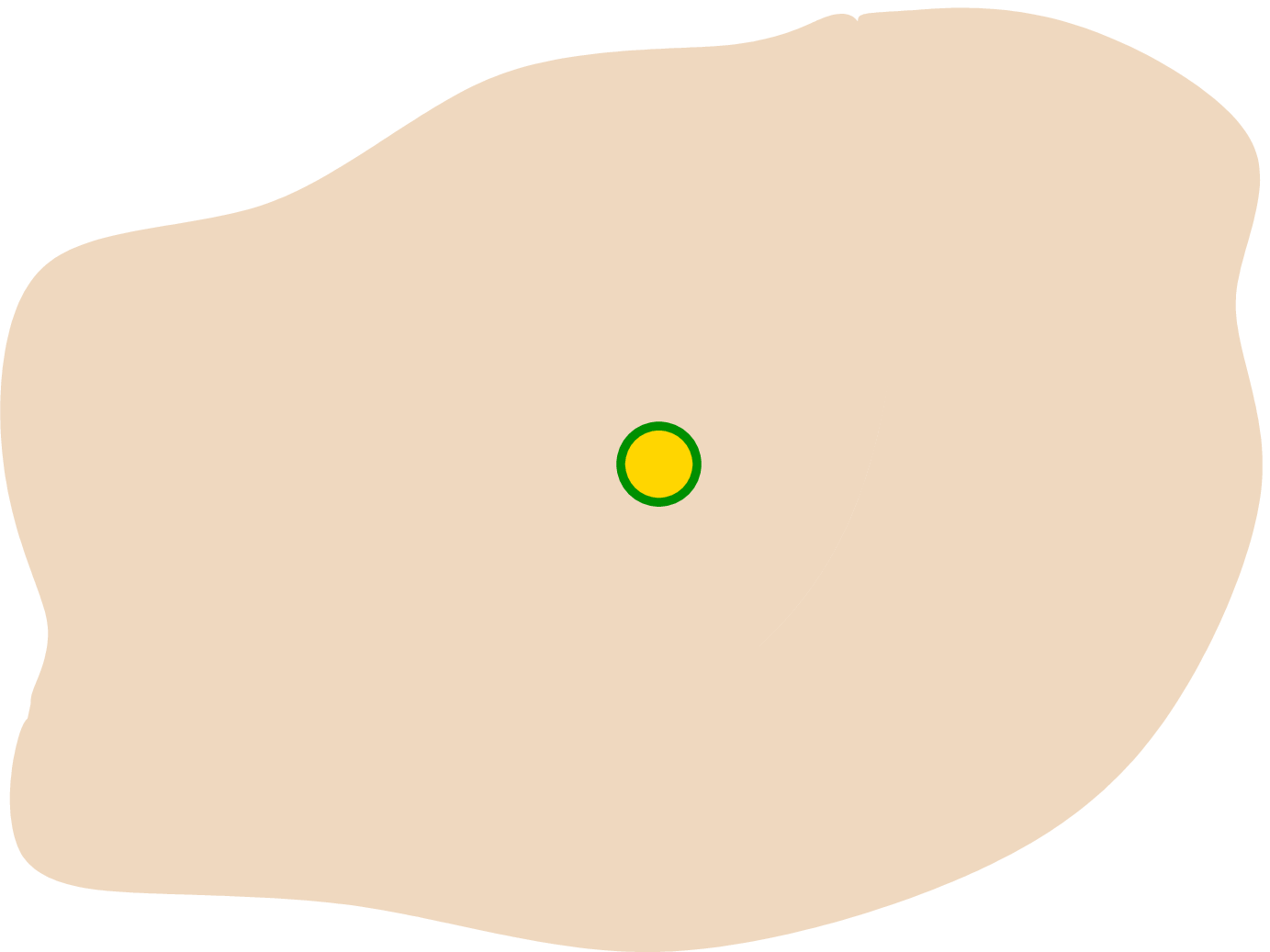}
  \put(60,39) {\pl {D_X(\phi)}}
  \put(85,62) {\pB {\fbox{$ B $}}}
  }
  }

\subject{Fusion of defects}
Consider two parallel defects $X$ and $Y$ in a region where there are no field insertions or other defect lines between them. Since we can move the two defects arbitrarily close to each other  they can be considered as a single defect: their \emph{fusion product} $X\FP Y$. In rational CFT, where there is a complete characterization of defect lines, there is a precise representation theoretic description of the defect condition that is assigned to $X\FP Y$. This description is given in section \ref{Dec_ws}.

Fusion of defects allow us to take a defect across a bulk field at the cost of transforming the bulk field into a disorder field.
In rational CFT we can, due to semisimplicity, map a bulk field in phase $A$ to a sum over defect conditions and disorder fields in phase $B$ as illustrated in the following picture:
  \Eqpic{def_trans}{321}{71}{ \setulen 67
  \put(0,-4){
  \put(0,125) { \includepicclax2{55}{99a}
  \put(56,44)  {$\color{DarkGreen} \PHi\alpha $}
  \put(109,80) {\pX X }
  \put(72,90)  {\pA {\fbox{$ A $}}}
  \put(124,60) {\pB {\fbox{$ B $}}}
  }
  \put(189,177) {$ = $}
  \put(230,125) { \includepicclax2{55}{99b}
  \put(109,80) {\pX X }
  \put(72,90)  {\pA {\fbox{$ A $}}}
  \put(124,60) {\pB {\fbox{$ B $}}}
  }
  \put(-10,52) {$=~ \dsty \sum_Y\sum_\tau$}
  \put(58,0) { \includepicclax2{55}{99c}
  \put(109,80) {\pX X }
  \put(76,39)  {\pX X }
  \put(56,53) {\begin{rotate}{51}{\pX Y}\end{rotate}}
  \put(72,90)  {\pA {\fbox{$ A $}}}
  \put(124,60) {\pB {\fbox{$ B $}}}
  }
  \put(218,52) {$=~ \dsty \sum_{Y,\tau}\sum_{\gamma}~
                  d^{\alpha\gamma}_{A,X,B;X_{\mu},\tau} $}
  \put(358,0) { \includepicclax2{55}{99d}
  \put(109,80) {\pX X }
  \put(56,53){\begin{rotate}{51}{\pX Y}\end{rotate}}
  \put(75,43)  {$\blue \varTheta_{\gamma} $}
  \put(72,90)  {\pA {\fbox{$ A $}}}
  \put(124,60) {\pB {\fbox{$ B $}}}
  } } }
Here, the summation over $Y$ is a sum over the simple defects, the summation over $\tau$ is over the multiplicity of $Y$ in the fusion of $X$ with $X^\vee$ and $\varTheta_{\gamma}$ label the elementary disorder fields at which $Y$ can end.

Considering the coefficients $d^{\alpha\gamma}_{A,X,B;B,\tau}$ on the right hand side of \eqref{def_trans}, i.e.\ the contribution  of the trivial defect condition $Y\eq B$, we obtain the \emph{defect transmission coefficients}\index{defect!transmission coefficient}, which contain a lot of the physical information about defects.
They describe the bulk field contribution in the sum over disorder fields in \eqref{def_trans} and thus to what extent the defect line is transmissive with respect to bulk fields.
As a consequence, the defect operator $D_X$, c.f. \eqref{defop}, is completely characterized by the defect transmission coefficients. In addition, as will be described in section \ref{sec:def_PF}, the \emph{defect partition function} has an expansion in terms of the defect transmission coefficients.
The defect transmission coefficients are in one-to-one correspondence with the irreducible representation of a finite-dimensional commutative unital associative algebra over \CN\ \cite{fuSs2}, the \emph{classifying algebra for defects}. Thus, the simple defects are classified by this algebra. This is described in chapter \ref{sec:class_alg}.

A defect can also be fused to a boundary. When an $A$-$B$-defect $X$ is running parallel to a boundary condition $M$ in phase $B$ it can be moved arbitrarily close to the boundary. This gives rise to a fusion between the defect line and the boundary. The mathematical description of this fusion is analogous to the fusion of two defects and is described in section \ref{Dec_ws}.

\subject{Defects in statistical models}
Defect lines appear naturally in lattice models. Consider e.g.\ the Ising model on a square lattice with ferromagnetic coupling at low temperature. In that situation most bonds are locked into place by the ferromagnetic coupling. However, it may happen that there are plaquettes, i.e.\ squares in the lattice bounded by four bonds (edges in the lattice), such that the bonds around this plaquette are not locked. This phenomenon, known as frustration, may occur if an odd number of bonds that form the square are anti-ferromagnetic. Frustration is encoded in a disorder parameter placed in the middle of the plaquette (i.e.\ on a site of the dual lattice). The calculation of correlators of such disorder parameter involves a line, along some path in the lattice, connecting the disorder parameters. The choice of such path is in fact a gauge choice, i.e.\ the correlator is independent of the particular choice of path. In the scaling limit we obtain a CFT-correlator that involves disorder fields connected by defect lines. See e.g.\ \cite{savi} for more details.

\subject{Group-like and duality defects}
Since defect lines connect different phases of a CFT, it should not come as a surprise that they can be used to derive non-chiral symmetries and dualities between CFTs. A \emph{group-like} defect is an $A$-$A$-defect with the property that its fusion product with the dual $X^\vee$ is the invisible defect.
Group-like defects describe non-chiral symmetries of CFTs. Since non-isomorphic defect conditions describe distinct defect operators \cite{ffrs5} non-isomorphic defect conditions give rise to distinct symmetries.

A generalization of group-like defects are \emph{duality} defects. A duality defect $X$ has the property that there is another defect $X'$ such that first taking $X$ across a bulk field as in \eqref{def_trans} and then taking $X'$ across the resulting disorder field results in a sum over bulk fields rather than genuine disorder fields. Duality defects can be used to derive duality relations between CFTs \cite{ffrs5}. It has e.g.\ been shown in \cite{ffrs3} that the order disorder duality of the critical Ising model can be implemented by a duality defect.

\section{Logarithmic CFT}
Above we have  largely restricted our attention to rational CFT, whose mathematical formulation by now is fairly well understood.
In recent years there has been a growing interest in a class of non-rational CFTs, called \emph{logarithmic} CFT, first discussed in \cite{gura}.

For logarithmic CFT we keep the requirement that the chiral algebra has a finite number of simple modules.
The novelty  as compared to rational CFT is that not every module is fully reducible, i.e.\ the relevant representation category is no longer semisimple.
A consequence of the non-semisimplicity is that there are conformal blocks with logarithmic branch cuts. This is explained in e.g. section 5.2 of \cite{gabe8}. One might fear that chiral theories of this form cannot be used to construct a full (local) CFT. However, this has been achieved for e.g.\ the triplet models, the best understood type of logarithmic models. In \cite{gaKa3} a modular invariant partition function is obtained, and a corresponding boundary theory has been studied in \cite{garu}.

It is worth noting that semisimplicity arises in quantum physics as a consequence of unitarity. Nevertheless the logarithmic theories appear in various physical applications, despite the fact that they are not unitary.
Logarithmic CFT can be used to understand e.g.\ the critical behavior of percolation, which was described in section \ref{sec:chi_full}, and critical polymers in two dimensions, see e.g.\ \cite{maRi3,PeRa}.

So far it is not fully known how to formalize logarithmic CFT mathematically.
In this context it is worth mentioning that for rational CFT it took two decades to obtain a model-independent mathematical description.
The work in \cite{fuSs3,fuSs4,fuSs5} presented in chapter \ref{sec:beyond} of this thesis aims at a better understanding of logarithmic CFT from a categorical perspective.

Holomorphic factorization\index{holomorphic factorization} is not expected to be valid only for rational CFT. Thus, one could explore whether the separation of the solution of full CFT, that was described in section \ref{fullCFT}, still works in the non-rational case.
In this thesis we give a prescription of a morphism $\Corrgn$ associated to a closed oriented surface $\Surf gn$ of genus $g$ with $n$ marked points. We show \cite{fuSs3,fuSs5} that for a particular class of non-semisimple categories, the morphism $\Corrgn$ is invariant under a natural  projective action \cite{lyub6} of the mapping class group of $\Surf gn$. This motivates us to think of the morphisms $\Corrgn$ as candidates for correlators in non-rational theories.
Considering $\Corr 10$ we also show that the partition function provided by this construction can be "chirally decomposed" \cite{fuSs4}. In particular, in the semisimple case, the well known charge conjugation partition function, which is present in any rational CFT (see section \ref{sec_bulkA}), is reproduced. All of this is described in chapter \ref{sec:beyond}.
\chapter{Category theory}\label{chap:cat}
In this chapter we introduce various notions from category theory that we need in the rest of the thesis.
A category\index{category} $\C$ consists of two types of data: A collection of \emph{objects} $\Obj(\C)$ and for any ordered pair $U,V\In\Obj(\C)$ a set of morphisms from $U$ to $V$ denoted by $\Hom(U,V)$. The morphisms are subject to the following conditions:
\begin{itemize}\addtolength{\itemsep}{-6pt}
\item
Any two morphisms $f\In\Hom(U,V)$ and $g\In\Hom(V,W)$ can be composed to a morphism $g\cir f\In\Hom(U,W)$ and the composition is associative
\item
For any object $V\In\Obj(\C)$ there is an identity morphism, $\id_V$, such that
\be
    \id_V\cir f=f\quand g\cir\id_V=g\,,
\ee
for any $f\In\Hom(U,V)$ and $g\In\Hom(V,W)$
\end{itemize}

As a first elementary example we mention the category \Vect\ of finite dimensional complex vector space. The objects of the category $\Vect$ are the finite dimensional complex vectors spaces and the morphisms are the linear maps between such vector spaces. The composition of morphisms is the ordinary composition of linear maps. It is easily checked  that the conditions above are indeed satisfied.
So far this statement is just a formalization of a well known fact. However, we will see that the category \Vect\ fits into a class of categories, relevant for CFT,  with a lot more interesting structure. In particular we will see that there are natural generalizations of structures such as associative algebras and their representation theory.

In order to describe aspects of CFT we will be interested in categories with quite a lot additional structure.
It is convenient to describe some of these structures in terms of functors. A \emph{functor} $F:\C\rightarrow\cD$ from a category $\C$ to a category $\cD$ associates to each object $U\In\Obj(\C)$ an object $F(U)\In\Obj(\cD)$ and to each morphism $f\In\Hom_{\C}(U,V)$ a morphism $F(f)\In\Hom_{\cD}(F(U),F(V))$ in such a way that the morphisms in $\cD$ satisfy
\be\label{functor_def}
    F(\id_U)=\id_{F(U)}\quand F(g)\cir F(f)=F(g\cir f)\,,
\ee
for all $f\In\Hom(X,Y)$ and $g\In\Hom(Y,Z)$. Analogously a \emph{bifunctor} is a functor $F:\C\times\cD\rightarrow\cE$ (the category $\C\times\cD$ is the category whose objects and morphisms are just pairs of objects and morphisms from $\C$ and $\cD$) associating to a pair of objects (morphisms) in $\C\times\cD$ an object (morphism) in $\cE$, such that $F$ is functorial (i.e. \eqref{functor_def} is satisfied) in each argument.

Below we introduce step by step all additional structure that impose on our categories. We will first introduce the special class of (semisimple) \emph{modular tensor categories}. After that, we describe the category $H$\Mod\ of representations of a finite-dimensional factorizable ribbon Hopf algebras, which may or may not be modular. In the end of this chapter we define a larger class of categories, in which we relax the semisimplicity condition, containing the semisimple modular ones, and in particular $H$\Mod.

\section{Monoidal categories}
A \emph{monoidal category} is a category $\C$ equipped with a bifunctor $\oti:\C\times\C\rightarrow\C$ associating to any pair of objects $U,V\In\Obj(\C)$ an object $U\oti V$ and to any pair of morphisms $f\In\Hom(U,V)$ and $g\In\Hom(U',V')$ a morphism $f\oti g\In\Hom(U{\oti}\linebreak U',V\oti V')$. In addition there is a \emph{tensor unit} $\one\In\Obj(\C)$ and families of isomorphisms such that
\be\label{iso_tensor}
    U\oti(V\oti W)\cong (U\oti V)\oti W\quand 1\oti U\cong U\cong U\oti\one\,,
\ee
for all $U,V,W\In\Obj(\C)$. The existence of such isomorphisms is not enough. They also have to satisfy the pentagon and triangle identities. The pentagon identity assures that the families of isomorphisms in \eqref{iso_tensor} provide a unique isomorphism between any two bracketings of multiple tensor products and the triangle identity is a compatibility condition between the associativity and unit constraint. See e.g. \cite[Chapter XI.2]{KAss}
for the explicit form of these identities.
\begin{remark}
Note that sometimes the word \emph{tensor category}\index{category!tensor} is used for a monoidal category. In mathematics literature, the notion of a tensor category typically involves additional structure. In particular a tensor category is required to be \emph{rigid}. This means that left and right duals, that are introduced in the next subsection, exist for each object. We will sometimes use the word tensor category, but only in situations when the category is a tensor category in this stronger sense.
\end{remark}
A monoidal category is called \emph{strict} if all the isomorphisms in \eqref{iso_tensor} are equalities.
In many interesting situations this is not the case. However, according to the coherence theorems \cite[Chapter VII.2]{MacL}, any monoidal category is equivalent to a strict one, see also \cite[Chapter XI.5]{KAss}. (Two categories are equivalent\index{equivalence! of categories} if there exists a functor between them that is invertible up to isomorphism, see eq. \cite[Chapter IV.4]{MacL}.) Thus we can (and will) always replace any non-strict category by an equivalent strict one and accordingly take the isomorphisms in \eqref{iso_tensor} to be equalities. Typically, this equivalence preserves also additional structures on the categories, see e.g. \cite[Theorem 2.2]{ngsc}.

In strict monoidal categories we will make extensive use of \emph{graphical calculus}. In graphical calculus we write morphisms as pictures in the following manner: A morphism is drawn as a box (or some other shape) connecting lines labeled by the source and target object, and the identity morphism is drawn as a single line labeled by the corresponding object. The convention in this text is to read pictures from bottom to top. Composition of morphisms is depicted as concatenation of morphisms and the tensor product of morphisms as juxtaposition. This is illustrated in the following picture:
 \eqpic{graph_calc} {300} {28} {
 \scalebox{0.95}{
  \put(-10,0){
  \put(40,0)     {\Includepic{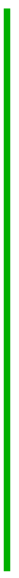}}
  \put(0,30)   {$ \id_U^{}\,= $}
  \put(38.0,-8.8){\pl{ U }}
  \put(38.5,65.5){\pl{ U }}
  }
\put(65,0){
  \put(20,0)     {\Includepic{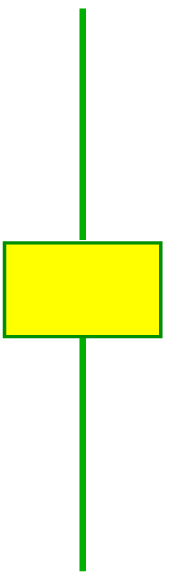}}
  \put(-10,30)   {$ f~= $}
  \put(25.3,-8.8){\pl{ U }}
  \put(26.2,65.5){\pl V }
  \put(26.6,30.2){\pl f }
}
\put(154,0){
  \put(20,0)     {\Includepic{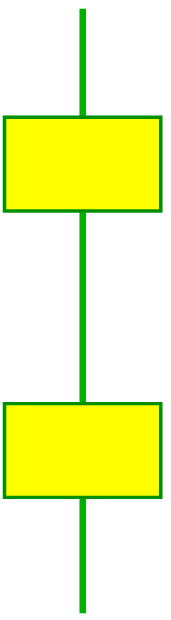}}
  \put(-22,30)   {$ g\cir f~= $}
  \put(25.1,69.5){\pl W }
  \put(25.3,-8.8){\pl U }
  \put(26.6,17.6){\pl f }
  \put(26.6,49.5){\pl g }
  \put(30.5,32.9){\pl V }
}
\put(246,0){
  \put(20,0)     {\Includepic{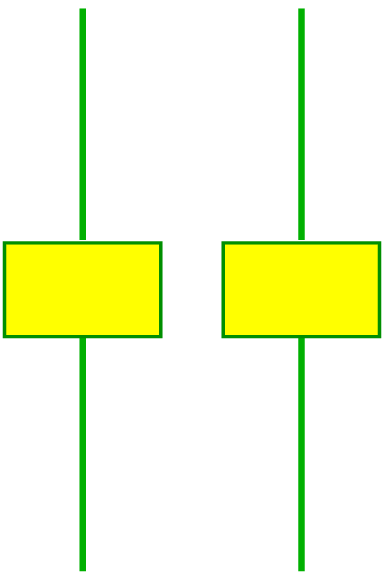}}
  \put(-25,30)   {$ f\oti f'~= $}
  \put(25.3,-8.8){\pl U }
  \put(26.2,65.5){\pl V }
  \put(26.6,29.4){\pl f }
  \put(49.1,-8.8){\pl {U'} }
  \put(50.0,65.5){\pl{ V' }}
  \put(49.6,29.4){\pl{ f'}}
} } }

\section{Modular tensor categories}\label{sec_modC}
In this section we introduce (semisimple) \emph{modular} tensor categories.

\subsection{Ribbon categories} A \emph{ribbon category\index{category!ribbon|textbf}} is a monoidal category with three additional compatible structures: \emph{Duality, braiding} and \emph{twist}:
\begin{itemize}
\item A (right) \emph{duality} assigns to every object $U$ a dual object $U^\vee$ and a pair of morphisms, the evaluation and coevaluation morphisms:
\be
    d_U\In\Hom(U^\vee\oti U,\one)\quand b_U\In\Hom(\one,U\oti U^\vee)\,.
\ee
Via the evaluation and coevaluation morphisms the duality also acts on morphism according to
\be
    f\mapsto f^\vee:=(d_V\oti\id_{U^\vee})\cir(\id_{V^\vee}\oti f\oti\id_{U^\vee})\cir(\id_{V^\vee}\oti b_U)\,.
\ee
In fact, as is easily checked, this way the duality is a functor $^\vee\mapdef\C\rightarrow\C\op$, where $\C\op$ is the opposite category defined below in section \ref{sec:add_cat}.
The morphisms $b_U$, $d_U$ and $f^\vee$ are depicted graphically as:
\eqpic{right_dual} {280} {25} {
\put(0,0){
  \put(40,0)     {\Includepic{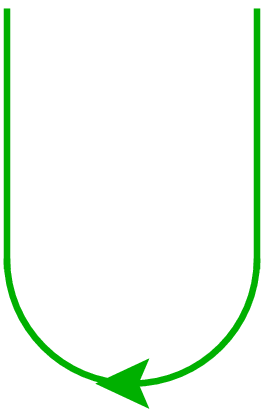}}
  \put(0,22)   {$ b_U~= $}
  \put(37.3,47.3) {\pl U }
  \put(63.9,47.3){\pl{ U^\vee }}
}
\put(120,6){
  \put(10,0)     {\Includepic{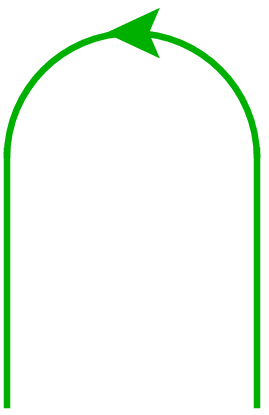}}
  \put(-32,16)   {$ d_U~= $}
  \put(6.2,-8.2) {\pl {U^\vee }}
  \put(34.8,-8.2){\pl U }
}
\put(215,-5){
  \put(10,0)     {\Includepic{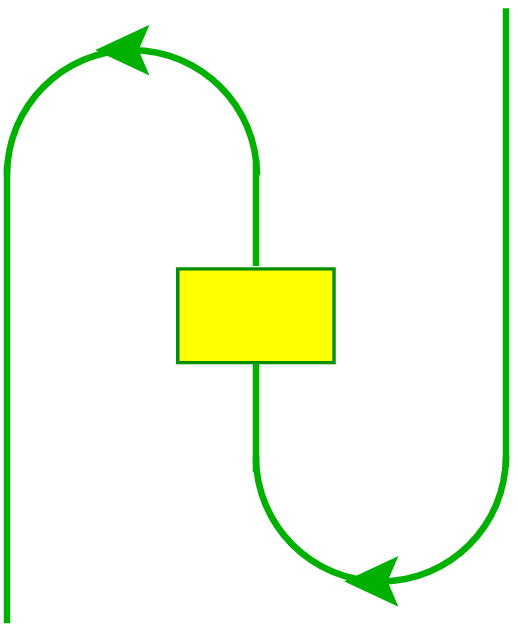}}
  \put(-22,27)   {$ f^\vee~= $}
  \put(6.2,-8.2) {\pl {V^\vee }}
  \put(62,74){\pl {U^\vee} }
  \put(34,32.5){\pl f}
}
}
The duality is unique up to unique isomorphism, see appendix \ref{app:dual_un}.
\item A \emph{braiding} is a family of natural isomorphisms $c_{U,V}\In\Hom(U\oti V,V\oti U)$, one for each pair of objects $U.V\In\Obj(\C)$. We depict the braiding and its inverse as
\eqpic{braiding} {280} {25} {
\put(80,5){
  \put(10,0)     {\Includepic{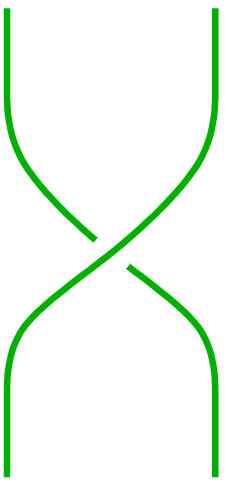}}
  \put(-32,22)   {$ c_{U,V}~= $}
  \put(8.3,-8) {\pl U }
  \put(31,-8){\pl{ V}}
  \put(8.3,55) {\pl V }
  \put(31,55){\pl{ U}}
}
\put(190,5){
  \put(10,0)     {\Includepic{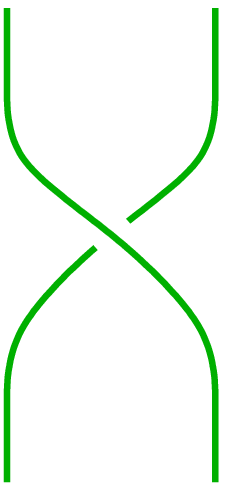}}
  \put(-32,16)   {$ c_{U,V}^{-1}~= $}
  \put(8.3,-8) {\pl V }
  \put(31,-8){\pl{ U}}
  \put(8.3,55) {\pl U }
  \put(31,55){\pl{ V}}
}
}
\item A \emph{twist} is a family of automorphisms $\theta_U\In\End(U)$, one for each object $U\In\Obj(\C)$. Graphically the twist and its inverse are denoted by
    \eqpic{twist} {280} {23} {
\put(80,0){
  \put(10,0)     {\Includepic{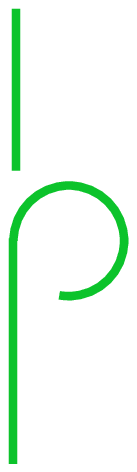}}
  \put(-32,22)   {$ \theta_U~= $}
  \put(9,-8) {\pl U }
  \put(9,55) {\pl U }
}
\put(190,0){
  \put(10,0)     {\Includepic{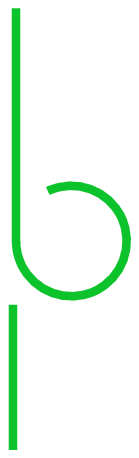}}
  \put(-32,22)   {$ \theta_U^{-1}~= $}
  \put(9,-8) {\pl U}
  \put(9,55) {\pl U }
}
}
\end{itemize}
The duality, braiding and twist satisfy a number of consistency conditions.
Among the them are the \emph{braid relations} which state that the isomorphisms $c_{U,V}$ (and analogously their inverses) have to satisfy
\be\label{braid_rel}
    \bearl
    c_{U,V\otimes W}=(\id_V\oti c_{U,W})\circ(c_{U,V}\oti\id_W)\quand
    \nxl{2}\dsty \hsp{0}
    c_{U\otimes V,W}=(c_{U,W}\oti\id_V)\circ(\id_U\oti c_{V,W})\,.
    \eear
\ee
The second of these equality is displayed in graphically in \eqref{tens_braid}.
The rest of the compatibility conditions is given in appendix \ref{app:comp_ribbon}, see also e.g.\ \cite[Chapter XIV.3]{KAss}.

\subject{Sovereignty}
There is also a notion of \emph{left} duality, i.e.\ a functor assigning to each object $U\In\Obj(\C)$ a left dual object $\rd U$ and to each morphism $f$ the left dual morphism $\rd f$. The left evaluation and coevaluation morphisms will be denoted by $\tilde d_U\In\Hom(U\oti \rd U,\one)$ and $\tilde b_U\In\Hom(\one,\rd U\oti U)$ respectively. The pictures for the left duality are the mirrored versions of \eqref{right_dual}, see \eqref{left_dual}.

In a ribbon category we can easily define a left duality by defining $\rd U:=U^\vee$ and taking the left (co)evaluation morphisms to be given by combinations of braidings, twists and right (co)evaluation morphisms, see \eqref{left_dual}. This defines a \emph{strictly sovereign structure}.

For some purposes strict sovereignty is too strong. We wish to relax the equality of left and right duality to $\rd U \cong U^\vee$. Explicitly: If it exists (in ribbon categories it does), a \emph{sovereign structure} is a choice of a monoidal natural isomorphism $\piv$ between the right and left duality functors. Naturality means that the family $\{\pi_U\}$ of isomorphisms is a natural transformation. Explicitly:
\be
    \rd f\circ\piv_V=\piv_U\circ f^\vee
\ee
for any $f\In\Hom(U,V)$. That $\{\pi_U\}$ is monoidal means that $\piv_{U\oti V}=\piv_U\oti\piv_V$.
The equivalent notion of \emph{balanced}, or sometimes also called \emph{pivotal}, structure requires instead the existence of an isomorphism between the (left or right) double dual functors and the identity functor.

In a sovereign category \C\ we can define right and left traces of an endomorphism $f\in\End(X)$ according to
\eqpic{trace_C}{225}{25}{
  \put(0,28.5){\pl {\tr_{\text r}(f):=}}
  \put(125,28.5){\pl {\tr_{\text l}(f):=}}
  \put(50,0){
  \put(0,0){\Includepic{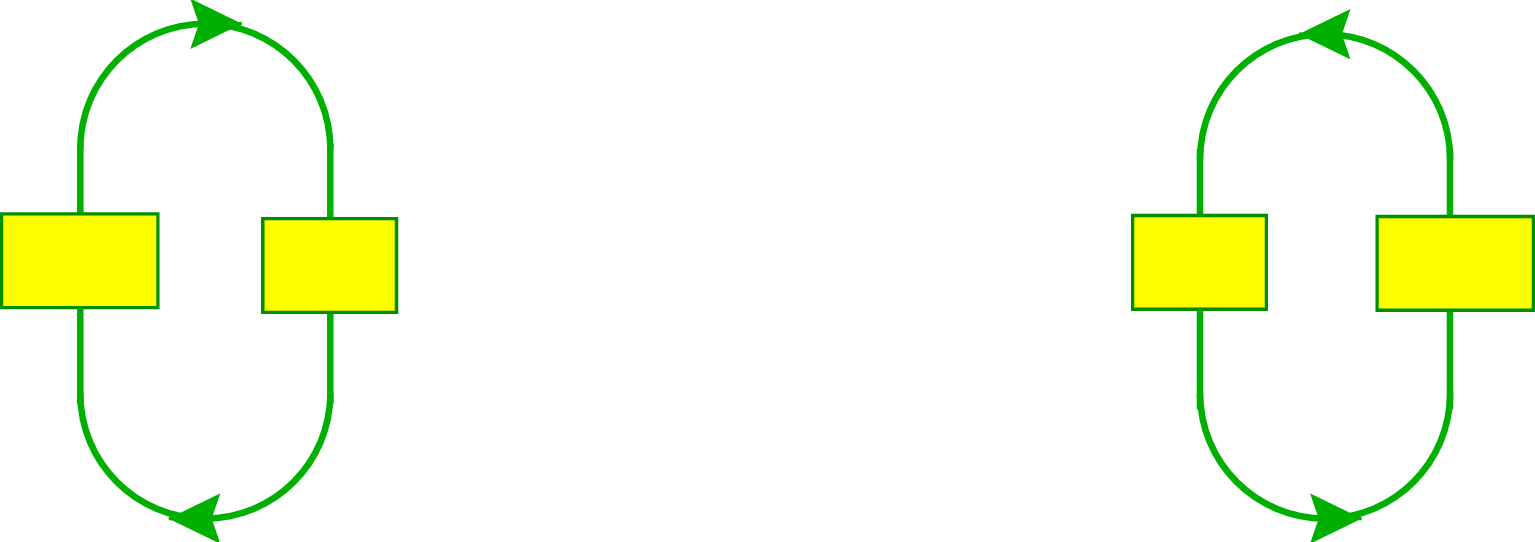}}
  \put(5,28.5){\pl f}
  \put(32,28.5){\pl {\piv_X}}
  \put(156,28.5){\pl f}
  \put(125,28.5){\pl {\piv_X^{-1}}}
  }
  }
If \C\ is in addition ribbon, it is \emph{spherical}, i.e.\ the left and right traces coincide
\be
    \tr_{\text l}(f)=\tr_{\text r}(f):=\tr(f)\,.
\ee
Using the properties of the duality functors it is also easily checked that the trace is cyclic.
The trace of the identity morphism is the \emph{quantum dimension} of the object
\be
    \dim(U):=\tr(\id_U)\,.
\ee
\subject{Graphical calculus for sovereign ribbon categories}
A nice way of thinking of the compatibility conditions for a strictly sovereign ribbon category is to think of the lines in graphical notation as two-dimensional ribbons. We then picture the twist endomorphism as a $2\pi$ rotation of a ribbon around its core:
\eqpic{twist_ribbon}{280}{22}{
  \put(190,0){\setulen80
  \put(0,-5){\includepic{304}{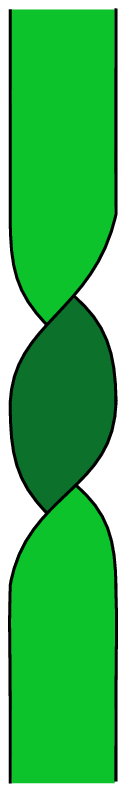}}
  \put(3,5){\pl U}
  \put(3,71){\pl U}
  }
  \put(145,25) {$=$}
  \put(90,0){
  \put(10,0){\Includepic{twist_C}}
  \put(9,-8) {\pl U }
  \put(9,55) {\pl U }
  }
  }
and dualities and braidings as ribbons with the same shapes. The compatibility conditions of the duality, braiding and twist are then exactly such that the allowed deformations of a morphism are the ones that are possible to perform with the corresponding two-dimensional ribbons.
\subsection{Modular tensor categories}
A modular tensor category is  a finitely semisimple, abelian, \k-linear ribbon category \C,
where $\k$ is an algebraically closed field of characteristic $0$, with simple tensor unit and a non-degenerate $s$-matrix. Let us describe what this means.

That \C\ is \emph{abelian} \k-linear  means in particular that every morphism set is a vector space over \k. What is relevant for CFT is the case $\k=\CN$. In an abelian category there is also a notion of direct sum, see e.g. \cite[Chapter VIII.3]{MacL} for more details. When discussing modular categories we will assume that $\C$ is strictly sovereign.

A \emph{simple object} is characterized by the property that it does not have a proper subobject (the notion of subobject is defined in section \ref{sec:add_cat}). In the present situation this implies that the space of endomorphisms of a simple object is the ground field. Such an object is called absolutely simple.

That \C\ is \emph{finitely semisimple} means that there is a finite number of simple object and every object is a finite direct sum of simple objects. We will denote a set of representatives of the simple objects by $\{U_i|i\In\I\}$, where $i=0,...,|\I|-1$ and we choose $U_0=\one$.

The \emph{$s$-matrix} is the $|\I|\Times|\I|$ matrix with entries
\be\label{s-matrix}
    s_{i,j}:=\tr(c_{j,i}\circ c_{i,j})\in\k\,.
\ee
Thus, a non-degenerate $s$-matrix constitutes  a non-degeneracy condition on the braiding.
Using that left and right traces coincide, the $s$-matrix can be written as
\eqpic{Sij}{103}{20}{
  \put(0,0){
  \put(42,0){\Includepic{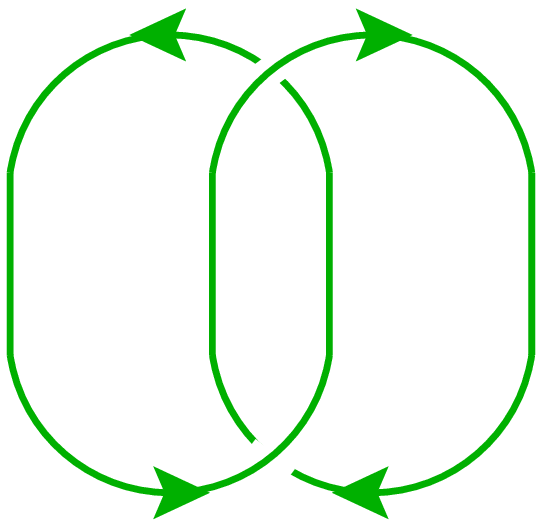}}
  \put(0,25.3)  {$s_{i,j}\;=$}
  \put(59.4,22.8){\scriptsize$i$}
  \put(80.3,22.8){\scriptsize$j$}
  }
  }
Note that
the first row (and column since the trace is symmetric) of the $s$-matrix consists of the quantum dimensions of the simple objects
\be
    \dim(U_i)=s_{i,0}\,.
\ee
The $s$-matrix together with the matrix $T$, with entries $T_{ij}:=\delta_{ij}\theta_i$, furnish a projective representation of the modular group. There is also a genuine representation of the modular group in which the $S$-transformation is generated by a matrix whose elements are related to the elements of the matrix  \eqref{s-matrix} by
\be
    S_{i,j}:=S_{0,0}\,s_{i,j}\,.
\ee
\ssubject{The category $\boldsymbol{\Vectk}$} The category \Vectk, of finite-dimensional vector spaces over an algebraically closed field \k, is an elementary example of a ribbon category: The duality assigns to any vector space $V$ the dual vector space $V^*\,{=}\linebreak\Hom_\k(V,\k)$. The braiding is the flip map $\flip VW$ which just interchange the two factors $V$ and $W$ in a tensor product; in particular over- and under-braiding coincide. When the braiding and inverse braiding coincide the category is referred to as \emph{symmetric}. Finally the twist is trivial: For any $V\In\Obj(\Vectk)$, $\theta_V\eq\id_V$. In fact $\Vectk$ is even modular. There is a single simple object, the ground field \k, and since the braiding is symmetric, the single $s$-matrix element equals $1$.

On the other hand, the category of super vector spaces is not modular. The category of super vector spaces has two simple objects (both of them 1-dimensional), $(\k,0)$ and $(0,\k)$, and again due to the symmetric braiding all four $s$-matrix elements are $1$. Thus the $s$-matrix is degenerate.
\ssubject{Basis choices}
It is convenient to describe some aspects of modular tensor categories in terms of bases of morphisms spaces. For every object $X$ and every simple object $U_i$ we fix embedding and restriction morphisms $\emb Xi\alpha\In\Hom(U_i,X)$ and $\res Xi\alpha\in\Hom(X,U_i)$. Here $i\In\I$ and $\alpha$ labels the multiplicity of $U_i$ in $X$. The embedding and restriction morphisms have the properties
\be\label{emb_res_prop}
    \res Xj\alpha\circ\emb Xi\beta=\delta_{ij}\,\delta_{\alpha\beta}\,\id_{U_i}\quand\sum_{i\In\I}\sum_{\alpha}\emb Xi\alpha\circ\res Xi\alpha=\id_X\,.
\ee

Denote by
\be
    \N ijk=\dim(\Hom(U_i\oti U_j,U_k))\,,
\ee
i.e.\ the dimension of the space $\Hom(U_i\oti U_j,U_k)$ for all $i,j,k$ in $\I$ and choose a basis $\{\sob ijk\alpha|\alpha=i,2,...,\N ijk\}$ for this space. The dual basis of \mbox{$\Hom(U_k,U_i\oti U_j)$} is denoted by $\{\sobd ijk\alpha|\alpha=1,2,...,\N ijk\}$. We depict the basis morphisms as
\eqpic{basis_morphspace}{300}{25}{
  \put(20,25)   {$\sob ijk\alpha~=$}
  \put(80,0){
  \put(0,0)		{\Includepic{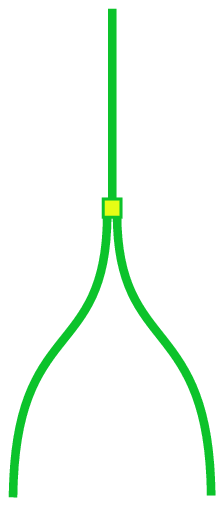}}
  \put(-1,-7)	{\pl i}
  \put(20,-7)	{\pl j}
  \put(10,57)	{\pl k}
  \put(15,31)	{\pl\alpha}
  }
  \put(170,25)   {$\sobd ijk\alpha~=$}
  \put(230,0){
  \put(0,0){\Includepic{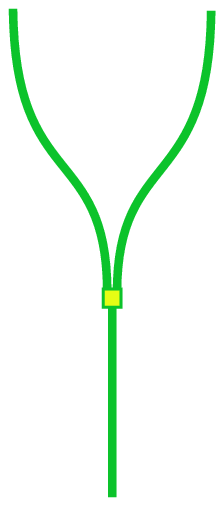}}
  \put(10,-7)	{\pl k}
  \put(-1,57)	{\pl i}
  \put(20,57)	{\pl j}
  \put(15,22)	{\pl{\overline\alpha}}
  }
  }
That the two bases are dual to each other means that they satisfy
\be\label{hombasdual}
    \sob ijk\alpha\circ\sobd ijl\beta~=~\delta_{kl}\,\delta_{\alpha\beta}\,\id_k\,.
\ee
In addition they satisfy
\be
		\sum_{k\In\I}\,\sum_{\alpha=1}^{\N ijk}\sobd ijk\alpha\circ\sob ijk\alpha~=~\id_{U_i}\oti\id_{U_j}\,.
\ee
This is due to semisimplicity of $\C$.

\section{Algebra objects in monoidal categories}
Recall the concept of an associative unital algebra $\AA$ over a field $\k$: \AA\ is a vector space over \k\ endowed with an associative \emph{product} $m\mapdef\AA\times\AA\rightarrow\AA$ and a unit, e.g.\ an element $e\In\AA$ such that $m(e,a)=m(a,e)=a$ for all $a\In\AA$. We will make extensive use of a categorical version of this concept.
\subsection{Algebras and coalgebras}
\begin{definition}\label{def_algebra}
   An \emph{algebra} $(A,m,\eta)$ in a monoidal category \C\ is an object $A$ in $\C$, together with a \emph{product} $m\In\Hom(A\oti A,A)$, and a \emph{unit} $\eta\In\Hom(\one,A)$, such that $m$ is associative and $\eta$ satisfies the unit constraint. Explicitly,
   \be
   \bearl
        m\circ(m\oti\id_A)=m\circ(\id_A\oti m)\quad\quand
   \nxl{0.5}
   m\circ(\id_A\oti\eta)= \id_A=m\circ(\eta\oti\id_A)\,.
   \eear
   \ee
\end{definition}
Indeed, an associative algebra \AA\ in the classical sense described above is nothing but an algebra in the category \Vectk.
We will also need the dual notion of a coalgebra:
\begin{definition}\label{def_coalgebra}
   A \emph{coalgebra} $(C,\Delta,\eps)$ in a monoidal category \C\ is an object $C$ in \C, together with a \emph{coproduct} $\Delta\In\Hom(C,C\oti C)$, and a \emph{counit} $\eps\In\Hom(C,\one)$, such that $\Delta$ is associative and $\eps$ satisfies the unit constraint. Explicitly,
   \be
   \bearl
        (\Delta\oti\id_C)\circ\Delta=(\id_C\oti \Delta)\circ\Delta\quad\quand
        \nxl{0.5}
        (\id_C\oti\eps)\circ\Delta= \id_C=(\eps\oti\id_C)\circ\Delta\,.
   \eear
   \ee
\end{definition}
In pictures, depicting the product and the unit as
  \eqpic{alg_def} {150} {15} {
\put(10,-2){
  \put(10,0)     {\Includepic{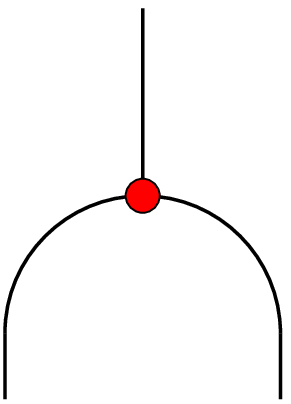}}
  \put(-29,20)   {$ m~= $}
  \put(6.0,-8.8) {\pl A }
  \put(22.5,46.5){\pl A }
  \put(36.0,-8.8){\pl A }
}
\put(140,4){
  \put(10,0)     {\Includepic{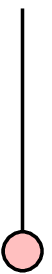}}
  \put(-25,14)   {$ \eta~= $}
  \put(9.5,32.2) {\pl A }
} }
the associativity and unit properties are drawn as:
\eqpic{jf29-07} {300} {23} {
\put(-25,3){
  \put(25,0)     {\Includepic{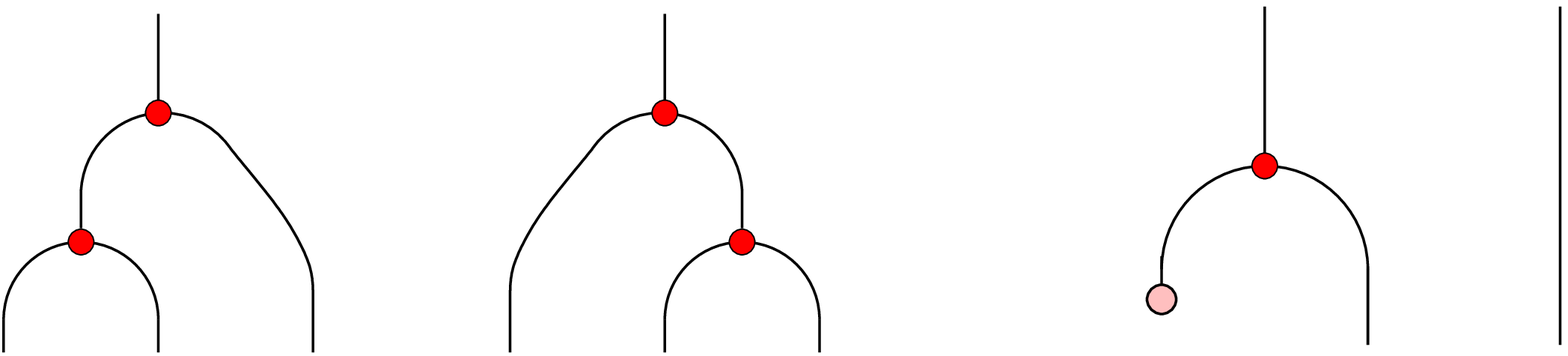}}
  \put(20.8,-9.7){\pl A }
  \put(43.3,-9.7){\pl A }
  \put(45.0,53.8){\pl A }
  \put(66.2,-9.7){\pl A }
  \put(80.2,23)  {$=$}
  \put(95.3,-9.7){\pl A }
  \put(118,-9.7){\pl A }
  \put(119.5,53.8){\pl A }
  \put(140.7,-9.7){\pl A }
  \put(207.2,54.8){\pl A }
  \put(221.2,-8.7){\pl A }
  \put(233.2,23)  {$=$}
  \put(249.6,-8.7){\pl A }
  \put(250.8,54.8){\pl A }
  \put(265.4,23)  {$=$}
  \put(278.9,-8.7){\pl A }
  \put(295.1,54.8){\pl A }
} }
Analogously, depicting the coproduct and counit as
\eqpic{coalgebra_def} {140} {15} {
\put(20,-3){
  \put(10,0)     {\Includepic{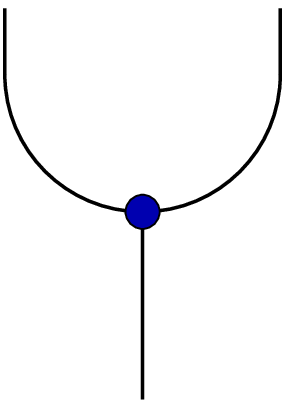}}
  \put(-29,20)   {$ \Delta~= $}
  \put(7.5,46.5) {\pl C }
  \put(21.3,-8.8){\pl C }
  \put(37.6,46.5){\pl C}
}
\put(134,4){
  \put(10,5)     {\Includepic{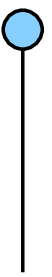}}
  \put(-25,13.2)   {$ \eps~= $}
  \put(8.1,-3.6) {\pl C}
} }
the coassociativity and counit constraints are drawn as:
  \eqpic{app63,app64} {300} {26} {
\put(-10,2){
  \put(10,0)     {\Includepic{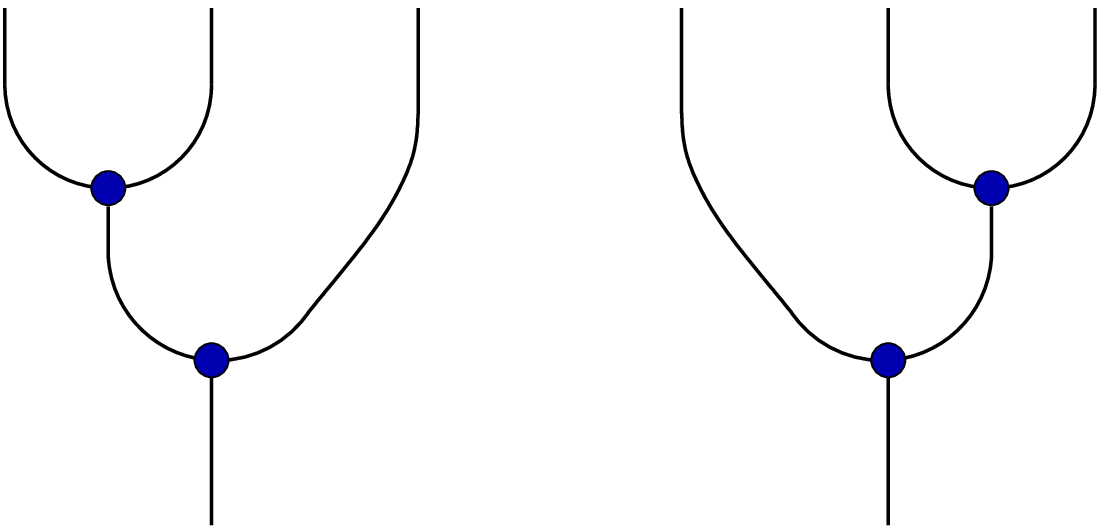}}
\put(7.8,60.4)   {\pl C}
\put(28.7,-9.2)  {\pl C}
\put(30.3,60.4)  {\pl C}
\put(53.0,60.4)  {\pl C}
\put(64.5,28)    {\small$=$}
\put(82.0,60.4)  {\pl C}
\put(103.1,-9.2) {\pl C}
\put(104.6,60.4) {\pl C}
\put(127.3,60.4) {\pl C}
}
\put(164,0){
  \put(10,0)     {\Includepic{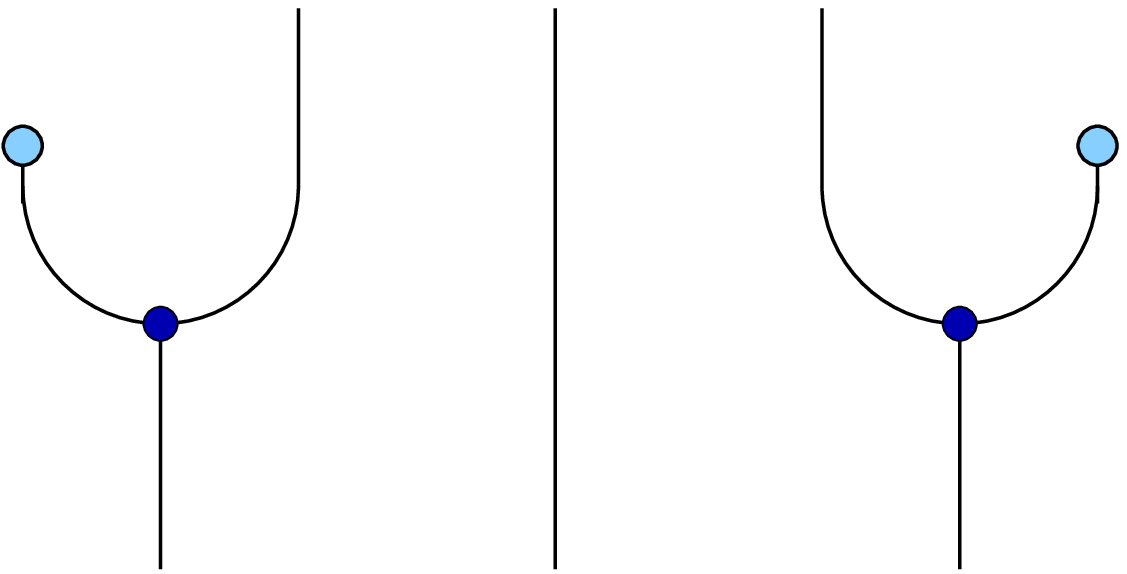}}
\put(23.2,-9.2)  {\pl C}
\put(39.8,65.4)  {\pl C}
\put(50.5,28)    {\small$=$}
\put(66.5,-9.2)  {\pl C}
\put(68.2,65.4)  {\pl C}
\put(81.5,28)    {\small$=$}
\put(97.5,65.4)  {\pl C}
\put(110.9,-9.2) {\pl C}
} }

A crucial ingredient in our description of conformal field theory will be objects which are both algebras and coalgebras and  in addition posses some compatibility condition between these two structures. We will consider two different compatibility conditions. The first one is the notion of a Frobenius algebra:
\begin{definition}\label{def_Frobalgebra}
   A \emph{Frobenius algebra}\index{Frobenius!algebra} $(A,m,\eta,\Delta,\eps)$ in a monoidal category \C\ is an object $A$ such that $(A,m,\eta)$ is an algebra, $(A,\Delta,\eps)$ is a coalgebra and the algebra and coalgebra structures are compatible in the sense that
     \eqpic{frob} {280} {31} {
\put(20,-1){
  \put(10,0)     {\Includepic{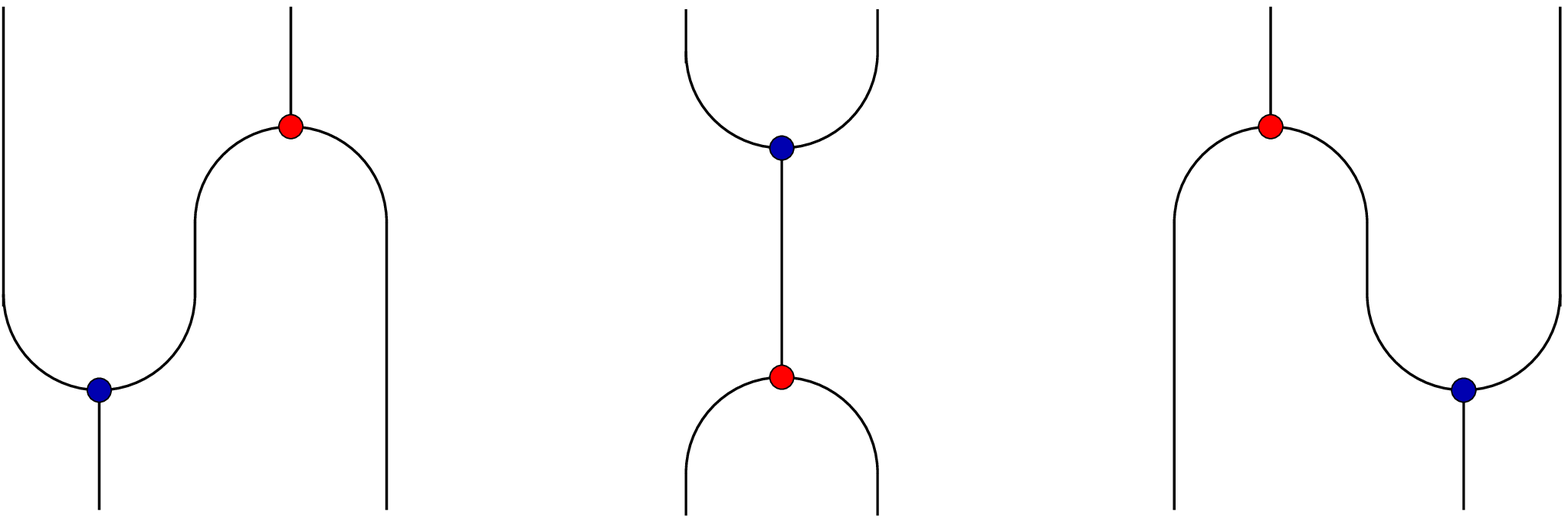}}
\put( 7.3,76.9)  {\pl A}
\put(20.5,-9.2)  {\pl A}
\put(48.5,76.9)  {\pl A}
\put(61.6,-9.2)  {\pl A}
\put(83.5,32)    {\small$=$}
\put(104.6,-9.2) {\pl A}
\put(105.4,76.9) {\pl A}
\put(132.0,-9.2) {\pl A}
\put(132.8,76.9) {\pl A}
\put(152.5,32)   {\small$=$}
\put(174.5,-9.2) {\pl A}
\put(189.0,76.9) {\pl A}
\put(215.8,-9.2) {\pl A}
\put(230.4,76.9) {\pl A}
} }
\end{definition}
The Frobenius algebras of our interest typically have additional properties:
\begin{definition}\label{def_sym}
~\nxl1
(i) An algebra in a $\k$-linear monoidal category \C\ which is also a coalgebra is called \emph{special} iff
\be
    m\circ\Delta=\beta_A\,\id_A\quand\eps\circ\eta=\beta_\one\,\id_\one
\ee
for non-zero numbers $\beta_\one$ and $\beta_A$.
~\nxl1
(ii) An \emph{invariant pairing} on an algebra $A \eq (A,m,\eta)$ in a monoidal
category $\C$ is a morphism $\kappa\In\Hom_\C(A\oti A,\one)$ satisfying
$\kappa \cir (m \oti \id_A) \eq \kappa \cir (\id_A \oti m)$.
\nxl1
(iii) A \emph{symmetric} algebra $(A,\kappa)$ in a sovereign category \C\ is an algebra
$A$ in \C\ together with an invariant pairing $\kappa$ that is symmetric, i.e.\ satisfies
   \eqpic{pic_csp_14} {230} {22} {
   \put(0,0)   {\Includepichtft{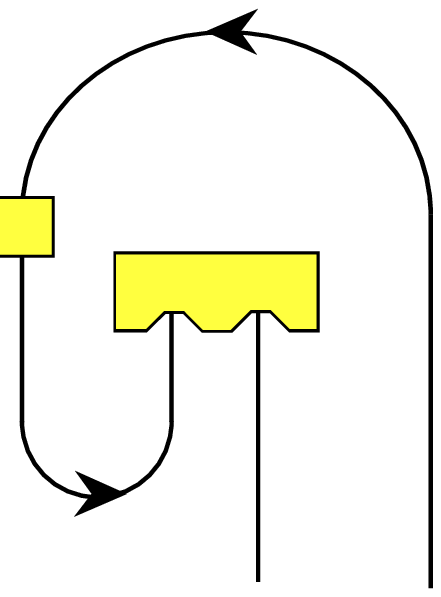}}
   \put(74,23)   {$ = $}
   \put(25.4,-8.5) {\sse$A$}
   \put(30.7,38.5) {\sse$\kappa$}
   \put(43.8,-8.5) {\sse$A$}
   \put(-13.5,40.9){\sse$\piv_{\!A}^{-1}$}
   \put(104,10) {\Includepichtft{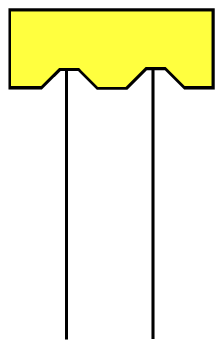}
   \put(2.4,-8.5)  {\sse$A$}
   \put(12.7,-8.5) {\sse$A$}
   \put(24.2,31.3) {\sse$\kappa$}
   }
   \put(148,23)  {$ = $}
   \put(179,0) {\Includepichtft{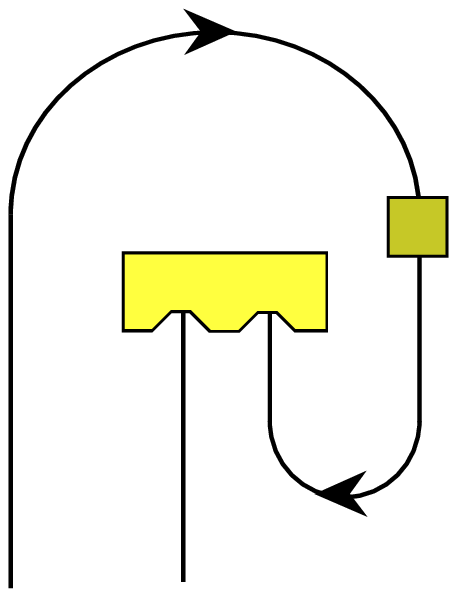}
   \put(-3.3,-8.5) {\sse$A$}
   \put(12.7,38.5) {\sse$\kappa$}
   \put(16,-8.5)   {\sse$A$}
   \put(50,39)     {\sse$\piv_{\!A}^{}$}
   }
   }
\end{definition}
\begin{remark}\label{rem_symm}
(i) In a sovereign category \C, the Frobenius property \eqref{frob} is equivalent to the statement that the morphisms $\Phi_{\text r}$ and $\Phi_{\text l}$ defined as
\be\label{def_Phi}
    \bearl
    \Phi_{\text r}:=((\eps\oti m)\otimes\id_{A^\vee})\circ(\id_A\oti b_A)\in\Hom(A,A^\vee)\quand
    \nxl{0.5}
    \Phi_{\text l}:=(\id_{\rd A}\oti (m\oti\eps))\circ(\tilde b_A\oti \id_A)\in\Hom(A,\rd\! A)
    \eear
\ee
are isomorphisms of right- and left-modules, respectively \cite{fuSt}. In addition, if \C\ is strictly sovereign, $A$ is symmetric iff $\Phi_{\text r}=\Phi_{\text l}$.
\nxl1
(ii) The two equalities in \eqref{pic_csp_14} imply each other.
\nxl1
(iii) An algebra with an invariant pairing $\kappa$, i.e.\ that satisfies $\kappa\cir(m\oti\id_A)=\kappa\cir(\id_A\oti m)$, is Frobenius iff $\kappa$ is non-degenerate \cite{fuSt}. In case $A$ is Frobenius as in Definition \ref{def_Frobalgebra}, $\kappa_\eps:=\eps\cir m$ is a non-degenerate invariant pairing.
\nxl1
(iv) Note that not only algebras and coalgebras, but even Frobenius algebras can be defined in any monoidal category. We do not even need the category to be abelian. A natural setting for symmetric algebras is a sovereign monoidal category.
\nxl1
(v) For a special Frobenius algebra $A$, we will choose the normalization such that $\beta_A=1$ and $\beta_\one=\dim(A)$.
\end{remark}

In a braided monoidal category we can define:
\begin{definition}\label{def_bialgebra}
   A \emph{bialgebra} $(A,m,\eta,\Delta,\eps)$ in a braided monoidal category \C\ is an object $A$ such that $(A,m,\eta)$ is an algebra, $(A,\Delta,\eps)$ is a coalgebra and the algebra and coalgebra structures are compatible in the sense
\eqpic{con_axiom} {130} {31} {
\put(0,0){
  \put(0,0)     {\Includepichopf{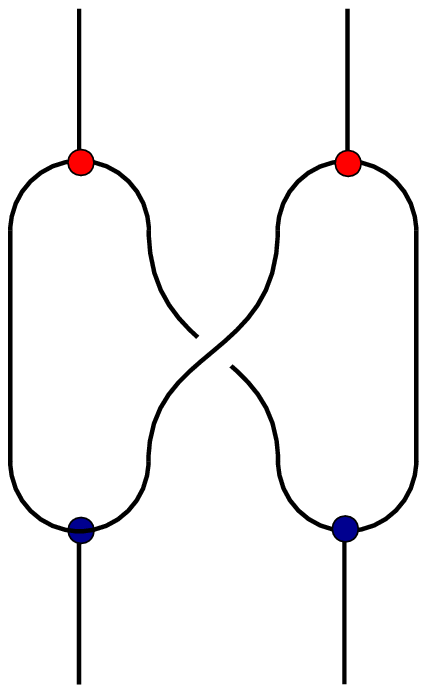}}
  \put( 4.5,76.9)  {\pl A}
  \put(2.5,-9.2)  {\pl A}
  \put(34,76.9)  {\pl A}
  \put(32,-9.2)  {\pl A}
}
\put(70,31) {$=$}
\put(110,0){
  \put(0,0)     {\Includepichopf{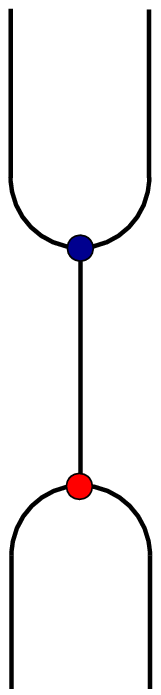}}
  \put( -3,76.9)  {\pl A}
  \put(-3,-9.2)  {\pl A}
  \put(12.5,76.9)  {\pl A}
  \put(12.5,-9.2)  {\pl A}
}
}
and in addition
\be\label{eta_m}
    \eps\circ m=\eps\oti\eps\quand \Delta\circ\eta=\eta\oti\eta\,.
\ee
\end{definition}

We will in particular be interested in bialgebras with an additional property:
\begin{definition}\label{def_Hopfalgebra}
   A \emph{Hopf algebra}\index{Hopf algebra|textbf} $(H,m,\eta,\Delta,\eps,\apo)$ in a monoidal category \C\ is an object $H$ such that $(H,m,\eta,\Delta,\eps)$ is a bialgebra, and the \emph{antipode} $\apo\In\End(H)$ satisfies
\eqpic{apo_prop} {180} {21} {
\put(0,0){
  \put(0,0)     {\Includepichopf{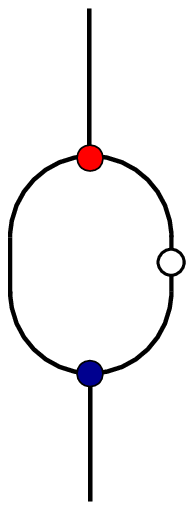}}
  \put( 7.5,58.9)  {\pl H}
  \put(6.5,-9.2)  {\pl H}
  \put(21,25)   {\pl\apo}
}
\put(40,21) {$=$}
\put(80,0){
  \put(0,0)     {\Includepichopf{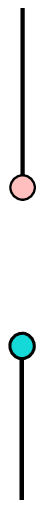}}
  \put( -3,58.9)  {\pl H}
  \put(-3,-9.2)  {\pl H}
}
\put(120,21) {$=$}
\put(160,0){
  \put(0,0)     {\Includepichopf{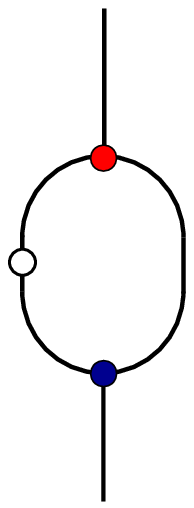}}
  \put( 7.5,58.9)  {\pl H}
  \put(6.5,-9.2)  {\pl H}
  \put(-5,25)   {\pl\apo}
}
}
\end{definition}
Graphically, we write the antipode as a circle (like in \eqref{apo_prop}), and if the antipode is invertible we write it as a filled circle:
\eqpic{graphics} {120} {9} { \put(-350,-9){
   \put(349,18)     {$ \apo~= $}
   \put(382,9) {\Includepichtft{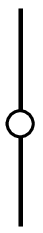}
   \put(-1.9,-8.8)  {\sse$ H $}
   \put(-1.7,27.2)  {\sse$ H $}
   }
   \put(420,18)     {$ \apoi~= $}
   \put(463,9) {\Includepichtft{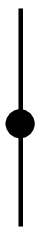}
   \put(-1.9,-8.8)  {\sse$ H $}
   \put(-1.7,27.2)  {\sse$ H $}
   } } }
It follows (see e.g. \cite[Section 3.1]{lyub3}) from \eqref{con_axiom} and \eqref{apo_prop} that $\apo$ is an anti-endomorphism of $H$ as a (co)-algebra:
\eqpic{anti_alg_apo} {300} {21} {
\put(0,0){
  \put(0,0)     {\Includepichopf{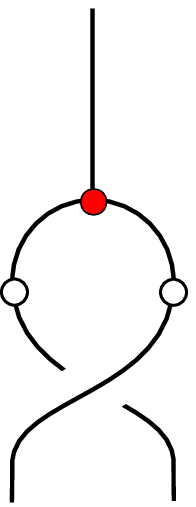}}
  \put( 7.5,58.9)  {\pl H}
  \put(-3,-9.2)  {\pl H}
  \put(15,-9.2)  {\pl H}
}
\put(40,21) {$=$}
\put(80,0){
  \put(0,0)     {\Includepichopf{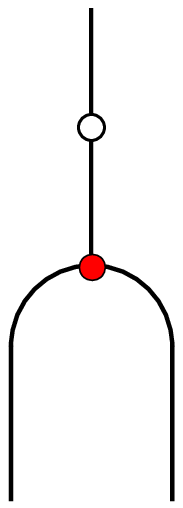}}
  \put( 5.5,58.9)  {\pl H}
  \put(-3,-9.2)  {\pl H}
  \put(15,-9.2)  {\pl H}
}
\put(138,21)   {and}
\put(200,0){
  \put(0,0)     {\Includepichopf{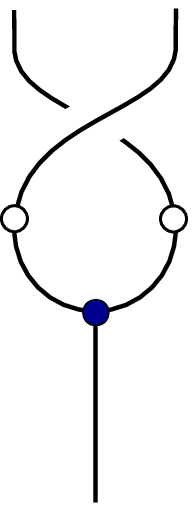}}
  \put( 5.5,-9.2)  {\pl H}
  \put(-3,58.9)  {\pl H}
  \put(15,58.9)  {\pl H}
}
\put(240,21) {$=$}
\put(280,0){
  \put(0,0)     {\Includepichopf{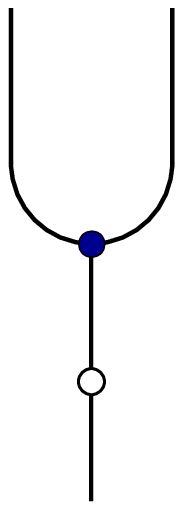}}
  \put( 5.5,-9.2)  {\pl H}
  \put(-3,58.9)  {\pl H}
  \put(15,58.9)  {\pl H}
}
}
If in addition the antipode is invertible, it follows that the inverse antipode satisfies analogous relations with braiding replaced by inverse braiding.
In other words, antipodes can be pushed through products and coproducts at the cost of introducing extra braidings.

From \eqref{eta_m} and the unit property it follows that
\be
    \eps\circ\eta=1\,.
\ee
Furthermore, combining $\eps\cir\eta\eq 1$ with the defining property \eqref{apo_prop} of $\apo$ and \eqref{eta_m} it follows that
\be\label{eta_inv_s}
    \eps\circ\apo=\eps \quand\apo\circ\eta=\eta\,.
\ee

\subject{Tensor products of (Frobenius) algebras} Consider two objects $A$ and $B$ in a braided category, such that $A$ and $B$ are algebras as well as coalgebras. We can define tensor products of algebras as well as coalgebras by using the braiding. Consider the following morphisms:
\be
\begin{split}\label{frob_tens}
    &m^\pm:=(m_A\oti m_B)\circ(\id_A\oti c^{\mp}_{A,B}\oti\id_B),\quad \eta^\pm:=\eta_A\oti\eta_B,\\
    &\Delta^\pm:=(\id_A\oti c_{A,B}^\pm\oti\id_B)\circ(\Delta_A\oti\Delta_B),\quad\;\eps^\pm:=\eps_A\oti\eps_B\,.
    \end{split}
\ee
If we don't require any compatibility condition between product and coproduct on $A\oti B$ it is easily checked that any choice of $m^\pm$ and $\Delta^\pm$ as defined in \eqref{frob_tens} gives an algebra and a coalgebra structure on $A\oti B$.
If $A$ and $B$ are Frobenius algebras in a braided category, then
\be\label{ABpm}
    A\otipm B\equiv(A\otipm B,m^\pm,\eta^\pm,\Delta^\pm,\eps^\pm)
\ee
is a Frobenius algebra for either choice of $\pm$.
Note that, if $A$ and $B$ are bialgebras in a symmetric braided category, the structure morphisms in \eqref{frob_tens} give a bialgebra structure on $A\otip B=A\otim B\equiv A\oti B$.

\subsection{Representations of algebras}\label{reps_alg}
The notion of representation of associative algebras have a straightforward categorical generalization.
\begin{definition}
A (left) \emph{module} $M(\rho)$ over an algebra $A$ in a monoidal category \C\ is an object $M\In\Obj(\C)$ together with a morphism $\rho\In\Hom(A\oti M,M)$ that satisfies
    \eqpic{pic_csp_13} {300} {35} {
\put(0,0){
  \put(0,0)     {\Includepic{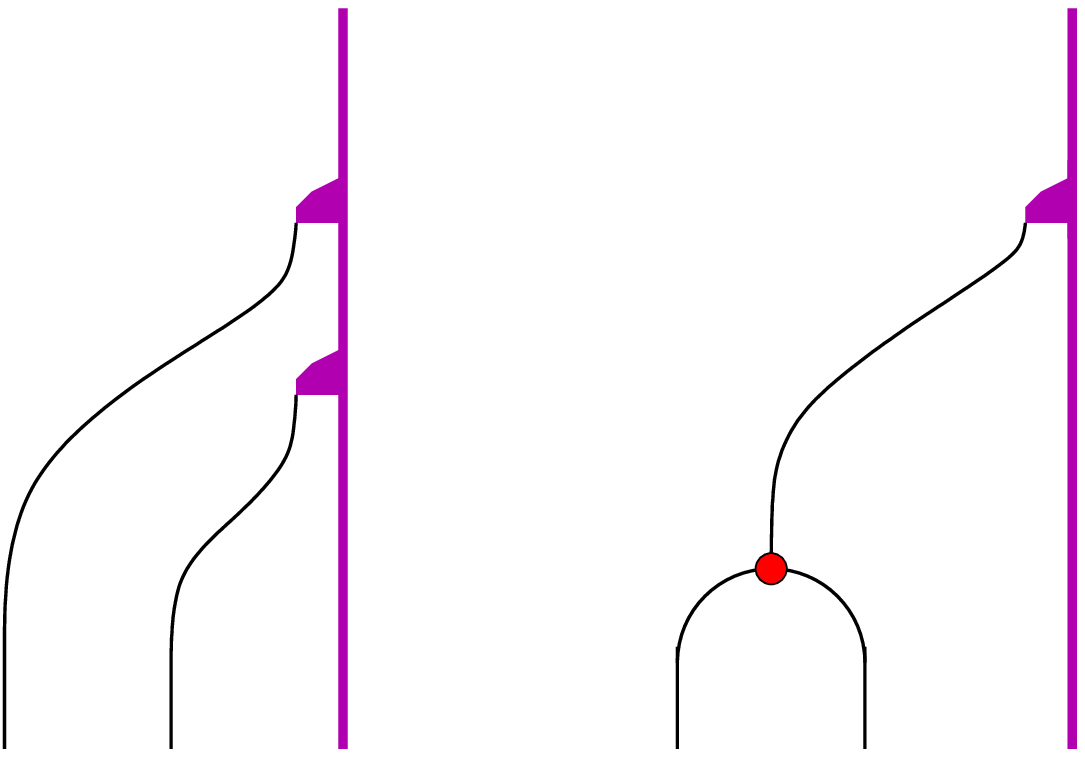}}
  \put(-4,-9){\pl A}
  \put(14,-9){\pl A}
  \put(34,-9){\pl { M}}
  \put(69,-9){\pl A}
  \put(91,-9){\pl A}
  \put(112,-9){\pl { M}}
  \put(40,40){\pl \rho}
  \put(40,60){\pl \rho}
  \put(121,60){\pl \rho}
  \put(58,35)   {$ = $}
  \put(34,85){\pl { M}}
  \put(112,85){\pl { M}}

  }
  \put(165,35)    {and}
  \put(220,0){
  \put(0,0)     {\Includepic{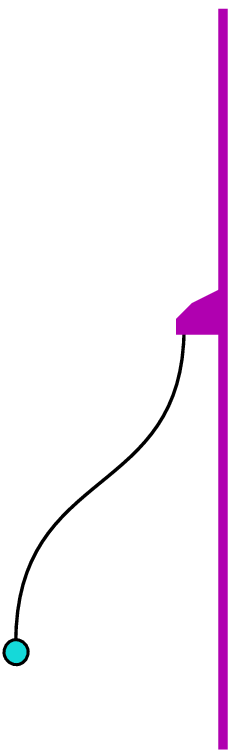}}
  \put(19,-9){\pl M}
  \put(28,46){\pl \rho}
  \put(50,35)   {$ = \id_M$ .}
  \put(19,85){\pl M}
  }
  }
\end{definition}
Analogously a \emph{right module} over an algebra $A$ in a category \C\ is an object $M\In\Obj(\C)$ together with a morphism $\ohr\In\Hom( M\oti A,M)$ satisfying the analogous relations. Finally  a \emph{bimodule} $(X;\rho,\ohr)$ over $A$ is an object $X$ such that $(X,\rho)$ is a left module, $(X,\ohr)$ is a right module and $\rho$ and $\ohr$ commute.
In order to lighten notation we will omit labels of representation morphisms when the morphism is clear from the context.

For any two objects $U$ and $V$ and any $A$-$B$-bimodule $(X,\rho,\ohr)$ of algebras $A$ and $B$ in a braided category we can define an induced bimodule $U\oti^+X\oti^-V$. As an object $U\oti^+X\oti^-V$ is just $U\oti X\oti V$, and the action is obtained by combining the inverse braiding with the representation morphisms of $X$ as follows:
\eqpic{ind_mod} {235} {44} {
\put(100,4){
  \put(0,0)     {\Includepic{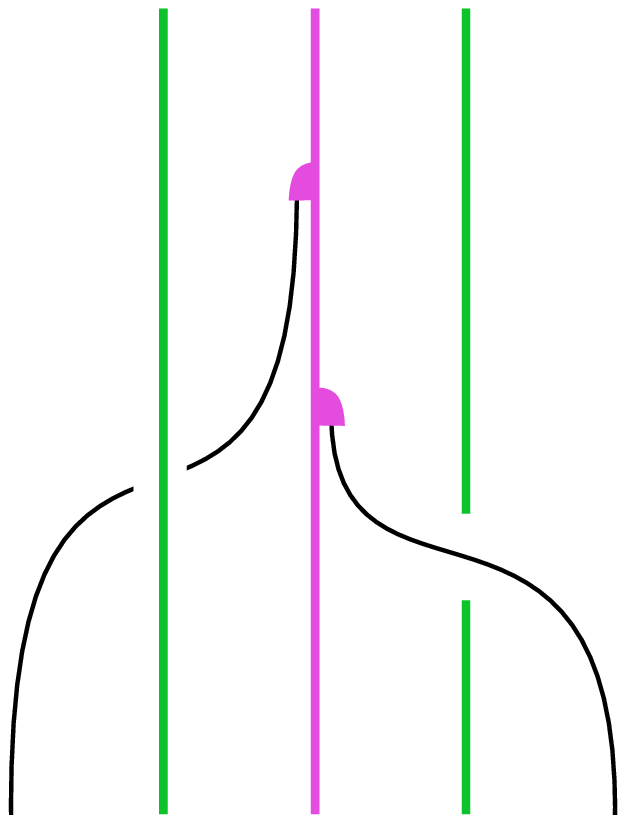}}
  \put(-4,-9){\pl A}
  \put(14,-9){\pl U}
  \put(30,-9){\pl {X}}
  \put(47,-9){\pl V}
  \put(63,-9){\pl B}
  \put(38,44){\pl {\ohr}}
  \put(24,70){\pl {\rho}}
  \put(14,92){\pl U}
  \put(30,92){\pl {X}}
  \put(47,92){\pl V}
  }}

Note that modules can be defined in any monoidal category \C. If \C\ is in addition sovereign we can define:
\begin{definition}\label{def:char}
The \emph{character} $\chi_M^A$ of a module $M$ over an algebra $A$ in a sovereign monoidal category is the morphism
\eqpic{char_def} {185} {26} {
    \put(0,26)  {$\chi_M^A:=$}
    \put(50,-2) {\Includepic{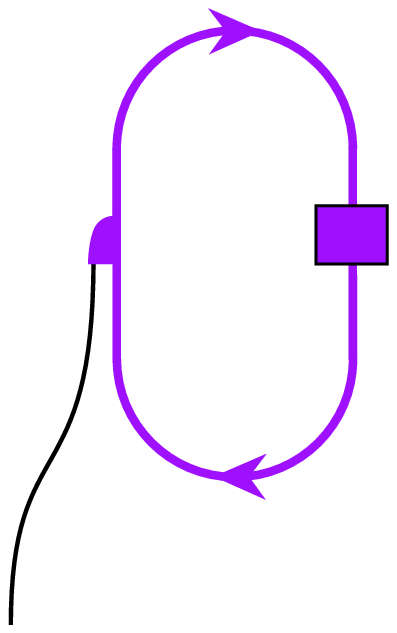}
    \put(-4,-6)   {\sse$A$}
    \put(4,56)   {\sse$M$}
    \put(-4,42)   {\sse$\rho_M^A$}
    \put(44,42)   {\sse$\piv_M$}
    }
    \put(120,27)  {$\in\Hom(A,1)$.}
}
\end{definition}
For any algebra $A$ in $\C$, the category $\C_A$ of left $A$-modules has as objects all $A$-modules in $\C$, and as morphisms the morphisms in $\C$ that intertwine the $A$-action. For any two $A$-modules $X$ and $Y$, the set of such morphisms is denoted by $\Homa XY$.
Analogously, there is a category $\C_{A|B}$ of $A$-$B$-bimodules. The objects of the category $\C_{A|B}$ consist of all $A$-$B$-bimodules in $\C$ and the morphisms in $\C_{A|B}$ consist of all morphisms in $\C$ that are $A$-$B$-bimodule morphisms.
We denote the set of morphisms from $X$ to $Y$ in $\C_{A|B}$ by $\Homab XY$. Note that if $\C$ is a modular category the categories $\C_A$ and $\C_{A|B}$ do not inherit this structure. In fact, they are in general not even monoidal categories. However, $\C_{A|A}$ may be naturally equipped with a monoidal structure. To see how this works, let us introduce tensor products of bimodules.

\subject{Tensor products of bimodules} If $X$ is an $A$-$B$-bimodule and $Y$ is a $B$-$C$ bimodule for any algebras $A$, $B$ and $C$ in some monoidal category $\C$, we can define the \emph{tensor product over $B$}, $X\otiB Y$, of $X$ and $Y$ as the coequalizer in $\C_{A|C}$ of the left and right actions of $B$ on $X$ and $Y$. That means first of all that $X\otiB Y$ is an object in $\C_{A|C}$ and there is a morphism $r$ in $\Hom_{A|C}(X\oti Y,X\otiB Y)$ such that $r\cir(\ohr^B_X\oti\id_Y)\eq r\cir(\id_X\oti\rho^B_Y)$. In addition $r$ satisfies a universal property, see e.g. \cite[Section III.3]{MacL}.

In the case we will be interested in, in which $\C$ is modular and $B$ is a special symmetric Frobenius algebra, $X\otiB Y$ can be written  as the image (see \cite[Lemma A.3]{fjfs})
\be\label{fus_def}
    X\otiB Y=\image(P_{X,Y})
\ee
of the projector
\eqpic{TP_proj} {100} {25} {
  \put(0,25)    {$P_{X,Y}:~=$}
  \put(70,0){
  \put(0,0)     {\Includepic{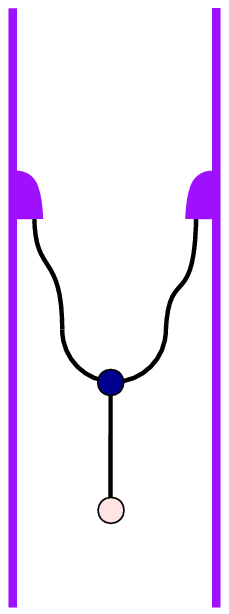}}
  \put( -4.5,-9.2)  {\pl X}
  \put( 19.5,-9.2)  {\pl Y}
  \put(-12,44.9)  {\pl{\ohr^{B}_X}}
  \put(25,44.9)  {\pl{\rho^{B}_Y}}
}
}
In addition, a modular tensor category $\C$ is idempotent complete which implies that there are embedding and restriction morphisms $r\In\Hom_{A|C}(X\oti Y,\linebreak X\otiB Y)$ and $e\In\Hom_{A|C}(X\otiB Y, X\oti Y)$ such that \be\label{TP_dec}
P_{X,Y}=e\circ r \quand r\circ e=\id_{X\otiB Y}\,.
\ee
This tensor product of bimodules turns the category $\C_{A|A}$ into a monoidal category. Note that by taking $Y$ to be a $B$-$\one$-bimodule $M$, this prescription defines $X\otiB M\In\C_A$ for any $A$-$B$-bimodule $X$ and any left $B$-module $M$. Analogously, one can define the right $C$-module $N\otiB Y$, for any $B$-$C$-bimodule $Y$ and any right $B$-module $N$

\subject{Tensor products of bialgebra (bi-)modules}
For bialgebras we have the notion of tensor product of bimodules described above. However, the connecting axiom \eqref{con_axiom} allows also a different notion of tensor product, which can be defined not only for bimodules, but already for left or right modules. For two $H$-$H'$-bimodules $X$ and $Y$, over two bialgebras $H$ and $H'$, there is a tensor product $X\otiBA Y\equiv(X\oti Y,\rho_{X\otiBA Y},\ohr_{X\otiBA Y})$. As an object, $X\otiBA Y$ is just the underlying object $X\oti Y$ in $\C$, and the actions are constructed from the coproduct and the braiding as in the following picture:
\eqpic{TP_Hopf} {170} {30} {
\put(0,0){
  \put(0,0)     {\Includepichtft{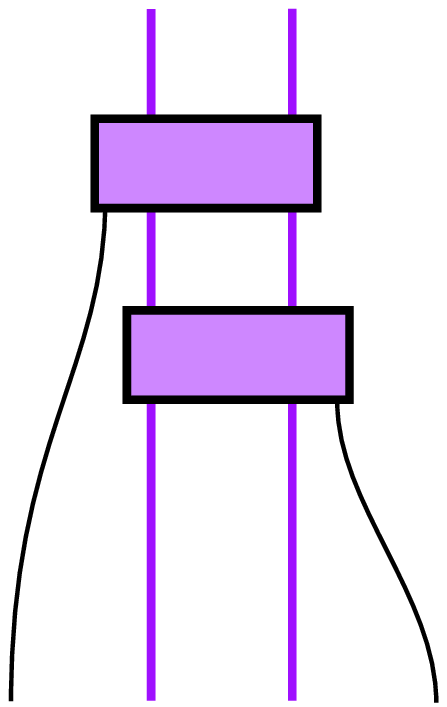}}
  \put( -4.5,-9.2)  {\pl H}
  \put( 11.5,-9.2)  {\pl X}
  \put( 27.5,-9.2)  {\pl Y}
  \put(43,-9.2)  {\pl H}
  \put(15,58.9)  {\pl{\rho_{\otiBA }}}
  \put(18,38.9)  {\pl{\ohr_{\otiBA }}}
}
\put(70,30) {$:=$}
\put(100,0){
  \put(0,0)     {\Includepichtft{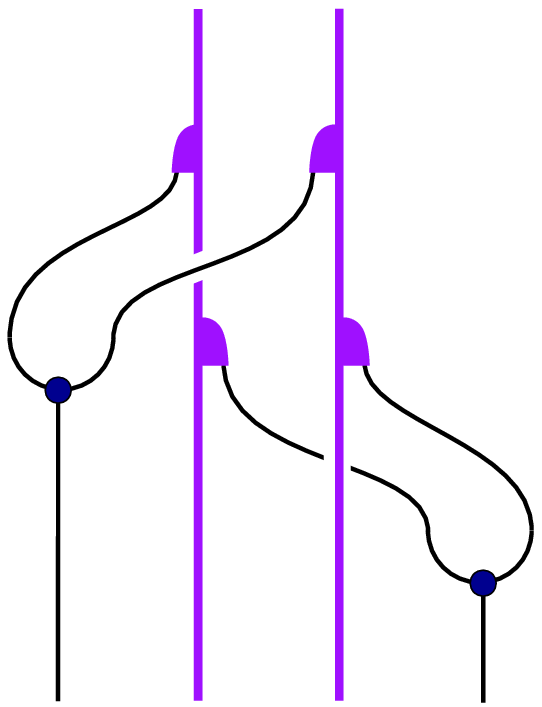}}
  \put( 1.5,-9.2)  {\pl H}
  \put( 17.5,-9.2)  {\pl X}
  \put( 33.5,-9.2)  {\pl Y}
  \put(49,-9.2)  {\pl H}
}
}
It is straightforward to check that these actions indeed satisfy \eqref{pic_csp_13} and the corresponding condition for right actions.
The tensor products of left- or right-modules are obtained from this description by simply forgetting about the right- or left-action.
\begin{remark}
    (i) Note that the connecting axiom \eqref{con_axiom} for a bialgebra means that the coproduct of a bialgebra is an algebra morphism from $A$ to $A\otim A$ with $A\otim A$ as defined in \eqref{ABpm}. \nxl1
    (ii)
    Endowing the object $A\oti A$ with a bimodule structure with left and right actions $m\oti\id_A$ and $\id_A\oti m$, the Frobenius property \eqref{frob} means that the coproduct of a Frobenius algebra is a bimodule morphism from $A$ to $A\oti A$.
\end{remark}
\section{Representation categories of factorizable ribbon Hopf algebras}

In this section we consider a finite-dimensional Hopf algebra $H$, over an algebraically closed field $\k$ of characteristic zero, or in other words, a Hopf algebra object in the category \Vectk, c.f Definition \ref{def_Hopfalgebra}. We describe the relevant structures and properties of such Hopf algebras below. For more information see e.g. \cite{KAss}.
We denote by $H$\Mod\ the category of left-modules over $H$. The objects of $H$\Mod\ are the left-modules, i.e. vector spaces over $\k$ endowed with a left $H$-action, and the morphisms of $H$\Mod\ are all left-module morphisms, i.e.\ linear maps that intertwine the action of $H$.

In our analysis we often work directly in the category \Vectk, i.e.\ morphisms are linear maps.
In particular we will draw pictures in \Vectk\ rather than in $H$\Mod. This has the advantage that various relations can be expressed in a more explicit manner. However, we have to be careful and make sure that the linear maps we deal with indeed are morphisms in $H$\Mod.
Recall that the braiding of \Vectk\ is symmetric, so we do not need to distinguish between over- and under-braiding.

Below we will identify $H$ with $\Homk(\k,H)$ and accordingly think of elements of $H$ as linear maps in $\Homk(\k,H)$. Analogously we will think of elements of $H^*$ as linear maps in $\Homk(H,\k)$. Thus in pictures, elements $h\In H$ and $f\In H^*$ will be drawn as
\eqpic{Hopf_el} {100} {10} {
  \put(0,0){
  \put(0,0)     {\Includepic{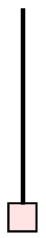}}
  \put( 5,0)  {\pl h}
}
  \put(40,10)   {and}
  \put(90,0){
  \put(0,0)     {\Includepic{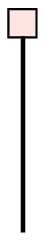}}
  \put( 5,20)  {\pl f}
}
}
respectively.

A \emph{left integral} of a Hopf algebra H is an element $\Lambda\In H$ satisfying
\be\label{int_prop}
    m\circ(\id_H\oti \Lambda)=\Lambda\circ\eps\,,
\ee
and a right integral is an element $\tilde\Lambda\In H$ satisfying $m\circ(\tilde\Lambda\oti\id_H)=\tilde\Lambda\circ\eps$.
A \emph{right cointegral} is an element $\lambda\In\H^*$ satisfying
\be\label{prop_lambda}
    (\lambda\oti\id_H)\circ\Delta=\eta\circ\lambda\,,
\ee
and a \emph{left cointegral} is defined analogously.

For a finite dimensional Hopf algebra $H$ over \k, the antipode is invertible and $H$ has a (up to normalization) unique non-zero left integral $\Lambda$ and a non-zero right cointegral $\lambda$, see \cite{laSw}. In addition $\lambda\circ\Lambda\In\k$ is invertible.
Recall from \eqref{anti_alg_apo} that the antipode $\apo$, and consequently also $\apoi$, is an anti-algebra as well as an anti-coalgebra morphism.

\subsection{Some Hopf algebra (bi-)modules}
In the present setting, $H$ can act on $H$ as well as $\Hss$ in more than one way.
Let us list the left and right actions that will be important to us
\ssubject{The regular bimodule}
Any algebra is a bimodule over itself with the actions given by the multiplication. We will refer to $H$ endowed with this action as the regular bimodule.
\ssubject{The coregular bimodule}
Combining the antipode and dualities we obtain the coregular actions of $H$ on \Hs.
The \emph{coregular bimodule}, $\BFH := (\Hs,\rho_{\BFH},\ohr_{\BFH} )$, is the vector space \Hs\ endowed
with the following actions of $H$:
\eqpic{rhoHb,ohrHb} {130} {41} {
   \put(0,0)   {\Includepichtft{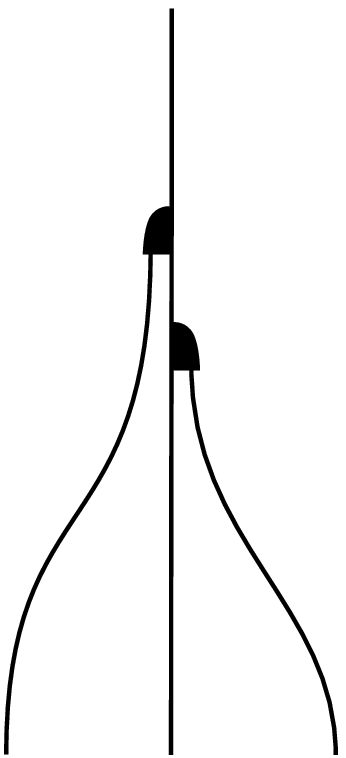}
   \put(-4,-8.8)    {\sse$ H $}
   \put(2,56)       {\sse$ \rho_{\BFH}$}
   \put(13,-8.8)    {\sse$ \Hss $}
   \put(14.3,84.9)  {\sse$ \Hss $}
   \put(23,43)      {\sse$ \ohr_{\BFH}$}
   \put(32,-8.8)    {\sse$ H $}
   }
   \put(58,38)      {$ :=$}
   \put(88,0) { \Includepichtft{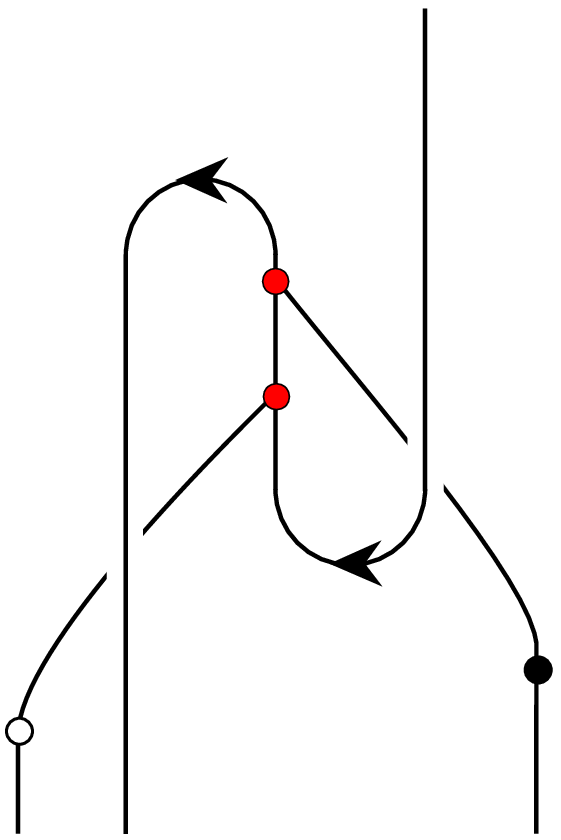}
   \put(-5,9)       {\sse$ \apo $}
   \put(-3,-8.8)    {\sse$ H $}
   \put(10,-8.8)    {\sse$ \Hss $}
   \put(42.3,93.3)  {\sse$ \Hss $}
   \put(54,-8.8)    {\sse$ H $}
   \put(61.6,16.2)  {\sse$ \apoi $}
   } }

\ssubject{The adjoint actions}
$H$ can also act on itself via the \emph{adjoint} left and right actions, defined by
\eqpic{adj_act} {250} {38} {
  \put(0,-5){
  \put(0,40)    {$\lad:=$}
  \put(43,0) { {\includepichopf{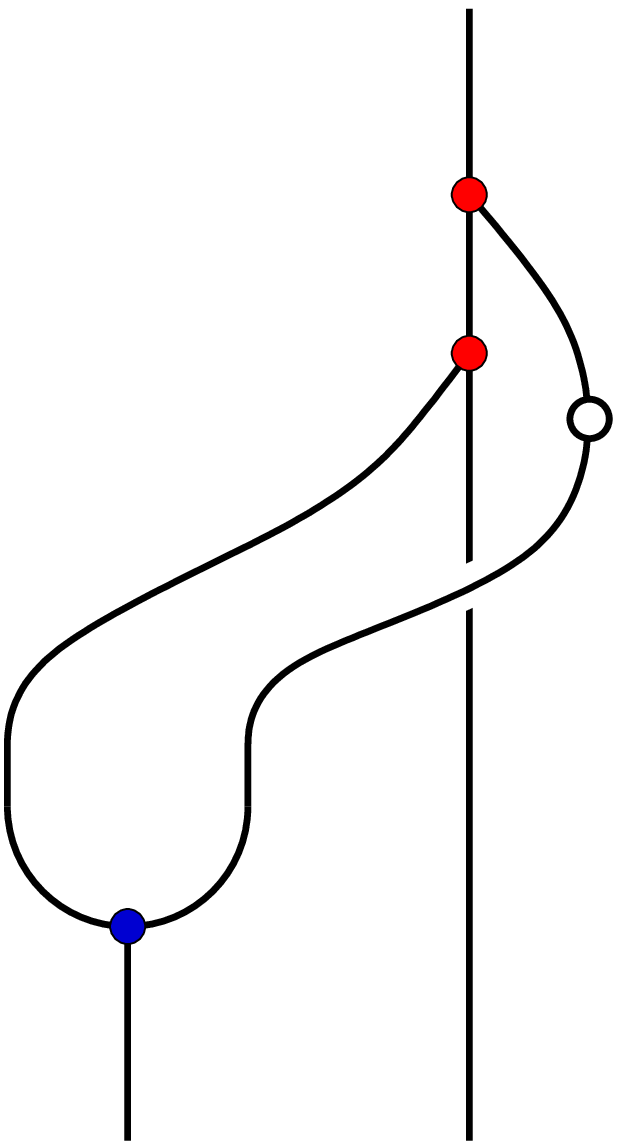}}}
  }
  \put(150,-5){
  \put(0,40)    {$\rad:=$}
  \put(43,0) { {\includepic{304}{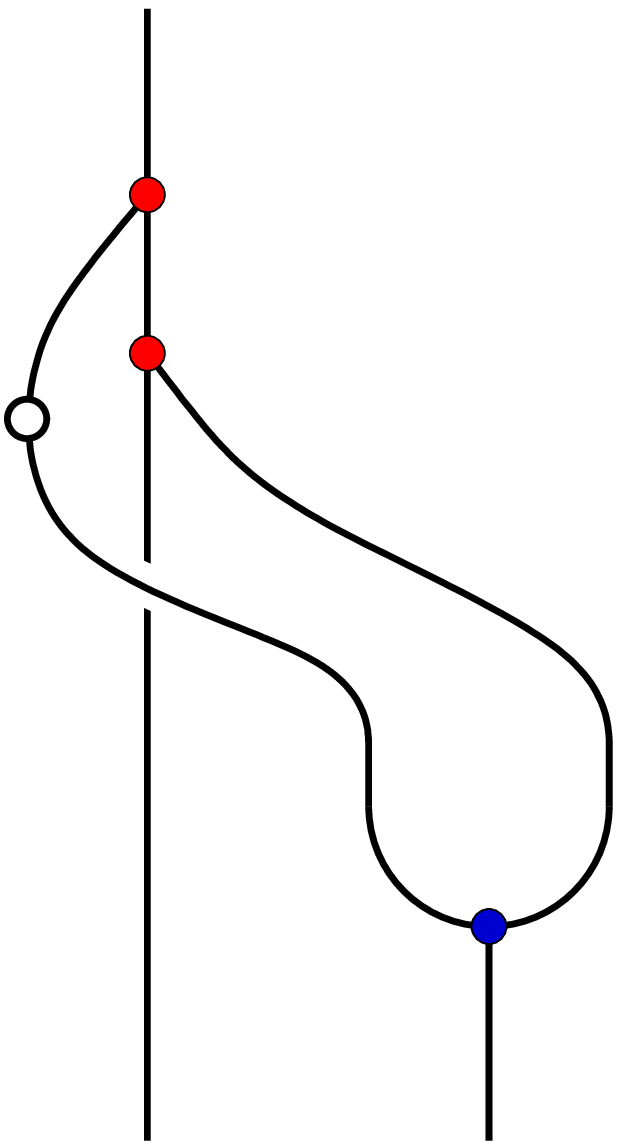}}}
  }
  }
\ssubject{The coadjoint left and right actions}
Analogously we can define the left and right coadjoint actions $\rho\coa$ and $\ohr\coar$ of $H$ on the vector space $\Hs$
\eqpic{def_lads} {285} {44} {\setulen85
   \put(0,44)      {$ \rho\coa ~:= $}
   \put(42,0) { \includepic{2584}{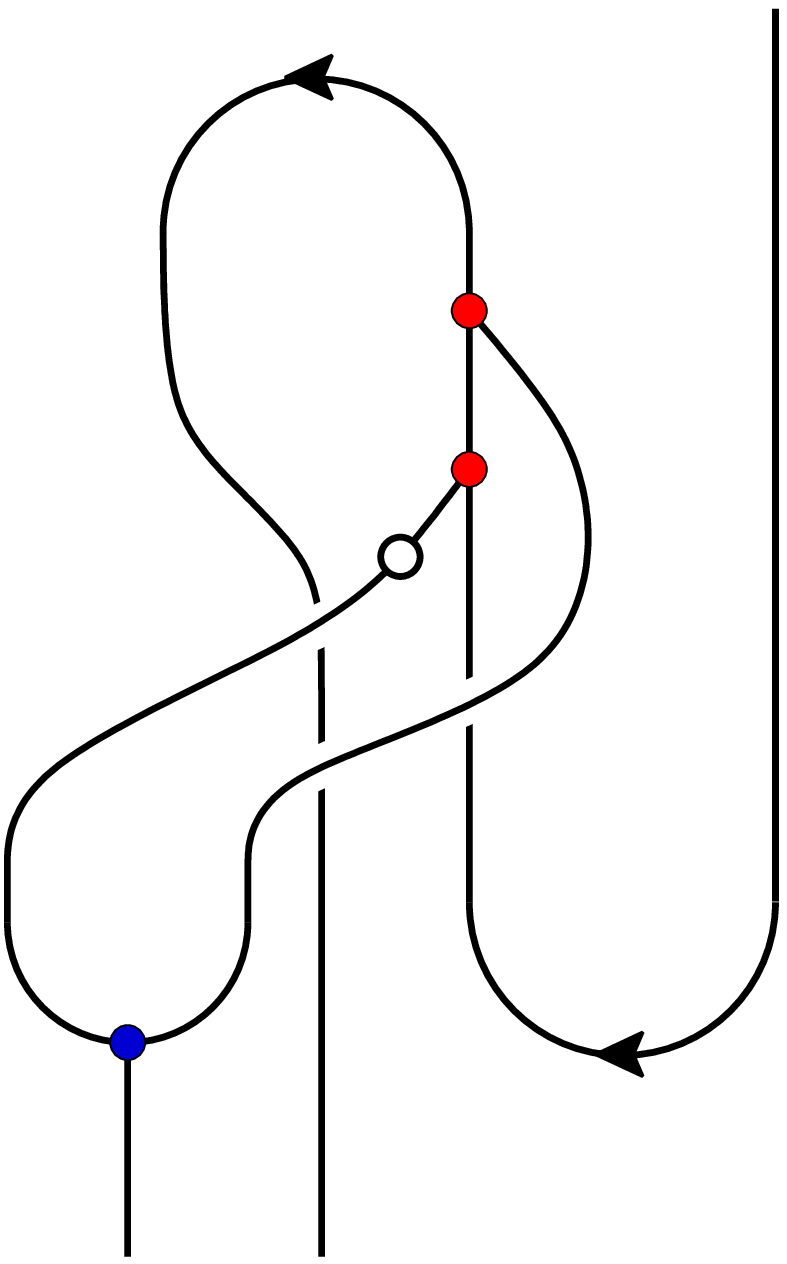}
   \put(6.5,-8.5){\sse$ H $}
   \put(23.4,-8.5) {\sse$ \Hss $}
   \put(32.3,65.6) {\sse$ \apo $}
   \put(64.5,114)  {\sse$ \Hss $}
   }
   \put(157,44)  {and}
   \put(208,44)    {$ \ohr\coar ~:= $}
   \put(258,0) { \includepic{2584}{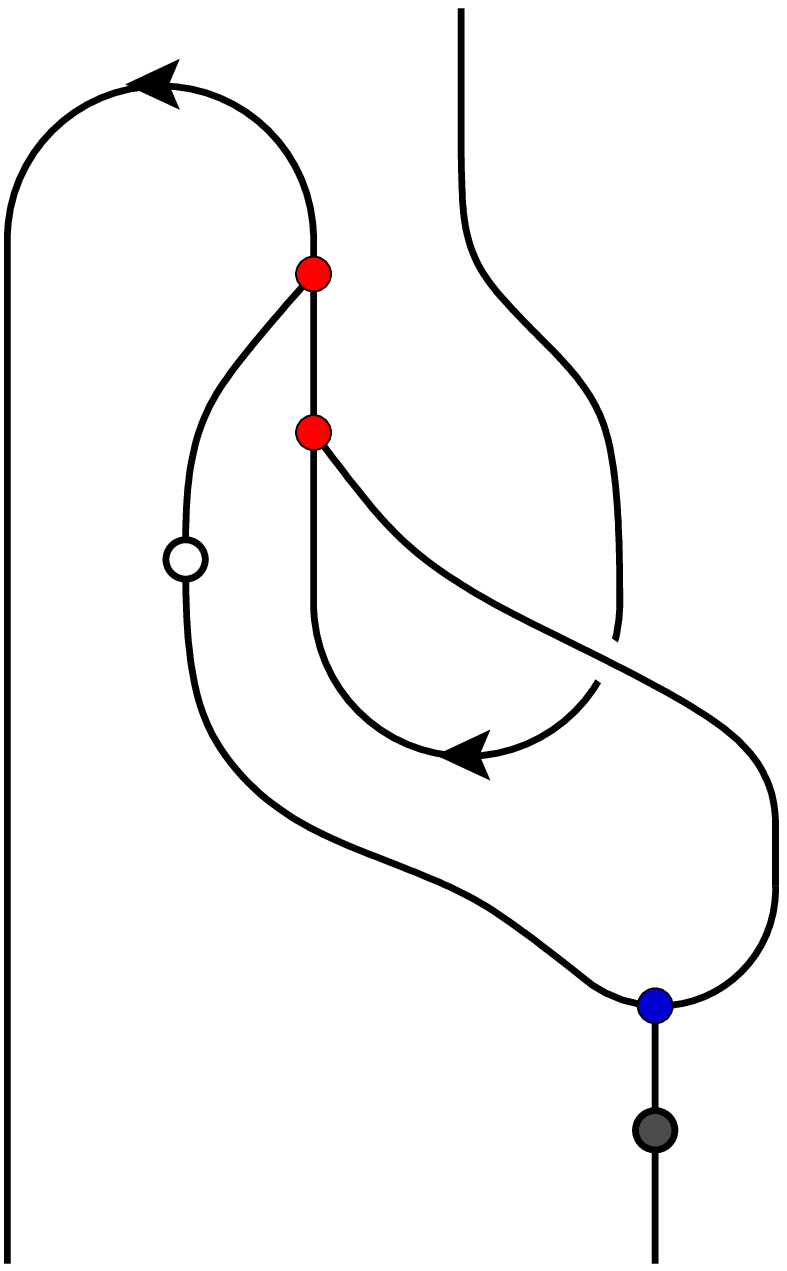}
   \put(-4,-8.5)   {\sse$ \Hss $}
   \put(8.8,60.5)  {\sse$ \apo $}
   \put(37,114)    {\sse$ \Hss $}
   \put(53.4,-8.5) {\sse$ H $}
   \put(61.4,10.5) {\sse$ \apoi $}
   } }
We will denote by $H_\diamond\eq(H,\lad)$ and $\Hs_{\triangleright}\eq(\Hs,\rho\coa)$ the adjoint and coadjoint left modules
\begin{remark}
Note that a left (right) integral is a left (right) module morphism from the trivial $H$-module $(\k,\eps)$ to the regular left (right) $H$-module. Analogously a left (right) cointegral is left (right) comodule morphism from H with the regular left (right) coaction to the trivial comodule $(\k,\eta)$.
\end{remark}

\subject{The Frobenius map}

The \emph{Frobenius map}\index{Frobenius!map} $\fmap:H\rightarrow \Hss$ and its inverse $\fmap^{-1}:\Hss\rightarrow H$ defined by
\eqpic{def_fmap} {290} {37} {\setulen80
   \put(1,37)     {$ \fmap ~= $}
   \put(39,0) {  \includepichopf{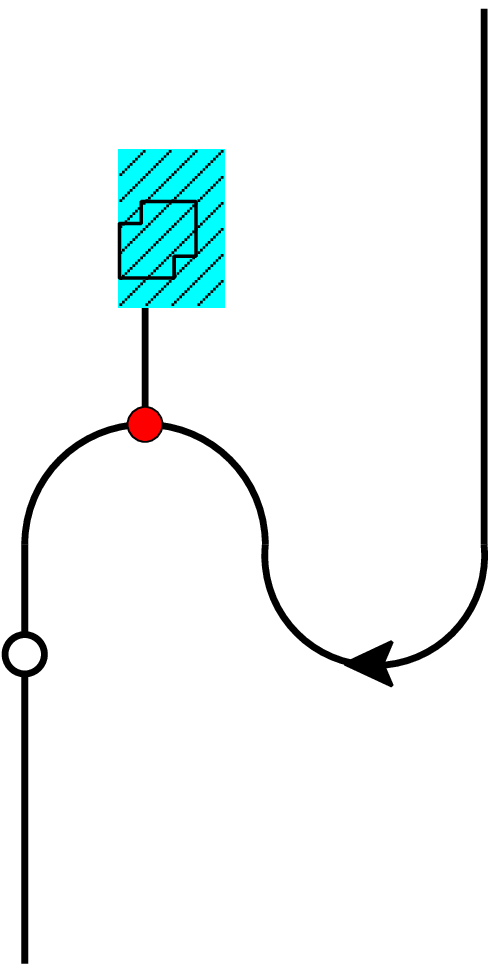}
   \put(-1.8,-8.5){\sse$ H $}
   \put(5.9,32.4) {\sse$ \apo $}
   \put(23.6,77)  {\sse$ \lambda $}
   \put(48.7,109) {\sse$ \Hss $}
   \put(97,37)   {and}
   }
  \put(196,0){
   \put(0,37)     {$ \fmap^{-1} ~= (\lambda\cir\Lambda)^{-1}$}
   \put(109,0) {  \includepichopf{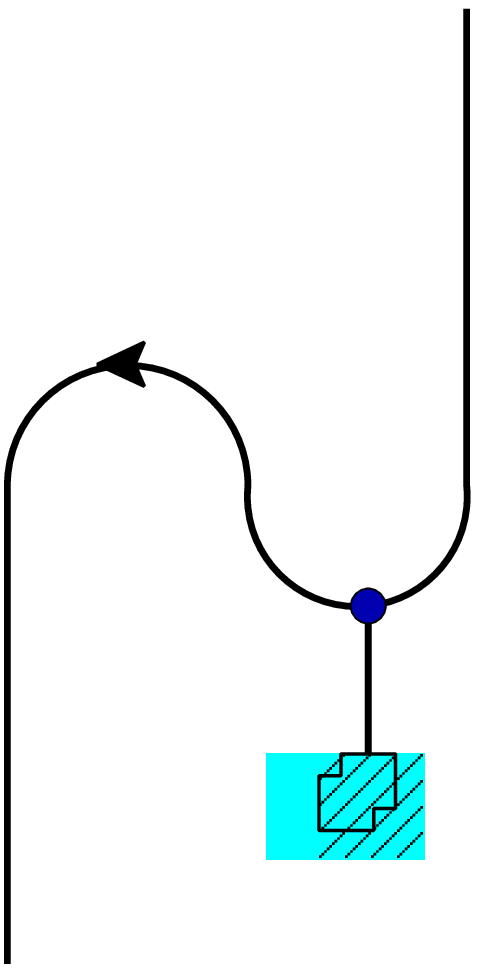}
   \put(-4.1,-8.5){\sse$ \Hss $}
   \put(41.5,10.2){\sse$ \Lambda $}
   \put(47.5,109) {\sse$ H $}
   } } }
That \fmap\ and $\fmap^{-1}$ are inverse to each other means that\footnote{In all considerations below we will have $\lambda\cir\Lambda=1$, c.f.\ \eqref{norm_int_1}.}
\eqpic{fmap_inv} {210} {40} { \setlength\unitlength{.8pt}
   \put(0,0)   {\includepichopf{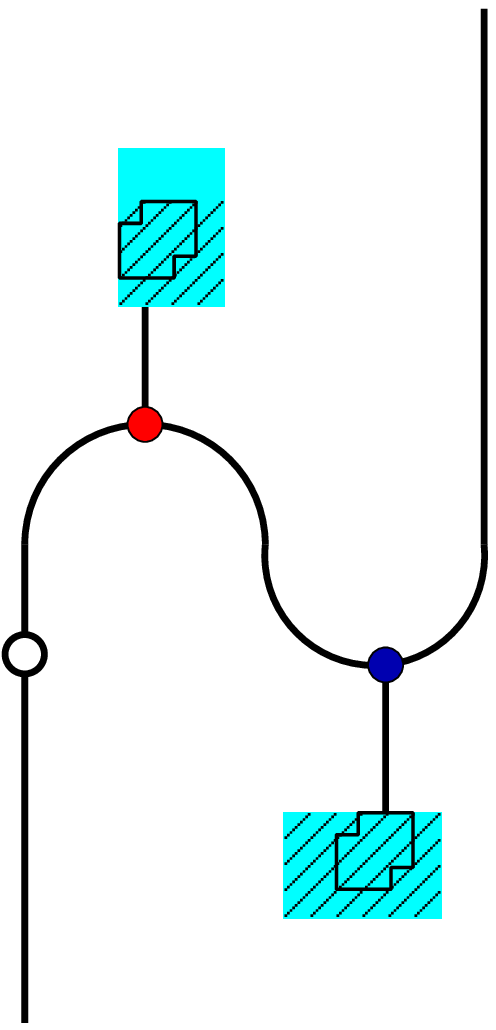}
   \put(-5.2,39) {\sse$\apo$}
   \put(-3,-10.5) {\sse$H$}
   \put(23.1,84) {\sse$\lambda$}
   \put(47.2,17) {\sse$\Lambda$}
   \put(49,116)  {\sse$H$}
   }
   \put(83,50)   {$=\quad\lambda\cir\Lambda$}
   \put(145,00) {\includepichopf{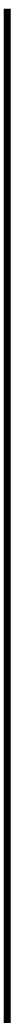}
   \put(-4.5,-10.5){\sse$H$}
   \put(-3.6,116){\sse$H$}
   }
   \put(170,50)  {$=$}
   \put(210,0) {\includepichopf{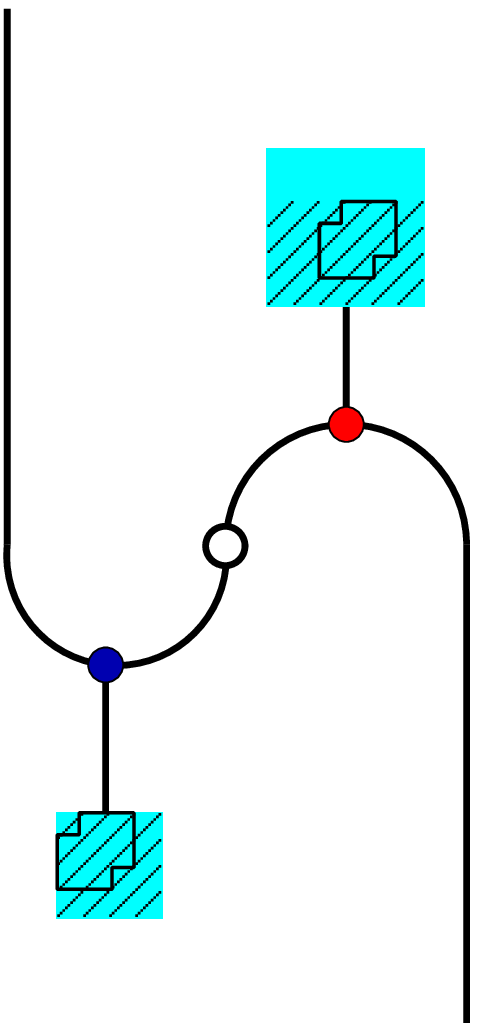}
   \put(-3.8,116){\sse$H$}
   \put(46,-10.5) {\sse$H$}
   } }
A graphical proof of this equality can be found in \cite[Appendix A.2]{fuSc17}
\begin{remark}
    It is easily checked, using that $\apo$ is an anti-algebra morphism, that the Frobenius map $\Psi$ intertwines the regular and the coregular actions.
\end{remark}
\subsection{Monoidal structure on $H$\Mod}\label{sec:Hmod_mon}
Considering any left-module $M$ as an $H$-bimodule, with right action given by $\id_M\oti\eps$, the tensor product defined in \eqref{TP_Hopf} turns $H$\Mod\ into a monoidal category.
Thus, the tensor product $M\oti N$ of two $H$-modules is the vector space $M\otik N$, equipped with the left action obtained from \eqref{TP_Hopf} by forgetting the right action. The monoidal unit is given by the \emph{trivial left-module} ($\k,\eps$), i.e. the vector space $\k$ with the $H$-action given by the counit of $H$.
From now on we will always think of $H$\Mod\ as monoidal equipped with this tensor functor.

Analogously there is a monoidal category \rMod$H$ of right representations of $H$. The category $H$\Bimod\ is the corresponding monoidal category of bimodules, i.e. the category whose objects are simultaneously left and right modules, with commuting left and right actions, whose morphisms are bimodule morphisms and the tensor product given by \eqref{TP_Hopf}. The monoidal unit is given by the trivial bimodule $(\k,\eps,\eps)$.
\subsection{Factorizable ribbon Hopf algebras}\label{fact_ribbon_hopf}
We now introduce, step by step, the additional properties of $H$, that we will require in chapter \ref{sec:beyond}.
\subject{Quasitriangular Hopf algebras}
A Hopf algebra $H$ is called \emph{quasitriangular}\index{Hopf algebra!quasitriangular} if there is an invertible element $R$ of $H\oti H$, called the \emph{$R$-matrix}, that intertwines the coproduct and the opposite coproduct, i.e
\be\label{Rmat1}
    \ad_R(\Delta):=R\cdot\Delta\cdot R^{-1}=\Delta\op\equiv\flip HH\circ\Delta\,,
\ee
where $\flip HH$ is the flip map, i.e.\ the braiding of $\Vectk$, and
\be\label{R_YB}
    (\id_H\oti \Delta)\circ R=R_{13}\cdot R_{12},\quad(\Delta\oti\id_H)\circ R=R_{13}\cdot R_{23}\,.
\ee
The $R$-matrix satisfies in addition (recall that $H$ is finite-dimensional and accordingly the antipode $\apo$ is invertible)
\be\label{Rmat2}
    (\apo\oti\id_H)\circ R=R^{-1}=(\id_H\oti\apoi)\circ R=R\quand(\apo\oti\apo)\circ R=R.
\ee
We also define the monodromy matrix $Q$ as
\be\label{Qmat}
    Q:=R_{21}\cdot R\equiv(m\oti m)\circ(\flip H{H\oti H}\oti\id_H)\circ(R\oti R)\,.
\ee
The monodromy matrix is invertible with inverse given by
\be
    Q^{-1}=R^{-1}\cdot R^{-1}_{21}\,.
\ee
From the definition of the $Q$-matrix and \eqref{Rmat2} it follows that
\be
    (\apo\oti\apo)\circ Q=\flip HH \circ Q\,,
\ee
and analogous relations with $Q$ replaced by $Q^{-1}$ or $\apo$ replaced by $\apoi$.
\subject{Ribbon Hopf algebras}
A ribbon Hopf algebra is a quasitriangular Hopf algebra $(H,R)$ endowed with a central invertible element $v$, called the \emph{ribbon element}, that satisfies
\be\label{prop_v}
    \eps\circ v=1,\quad\apo\circ v=v\quand\Delta\circ v=(v\otimes v)\cdot Q^{-1}\,.
\ee

When $H$ is a ribbon Hopf algebra, the category $H$\Mod\ is a ribbon category. In particular, the $R$-matrix endows $H$\Mod\ with a braiding as follows:
\eqpic{braid_Hmod} {120} {40} {
  \put(0,40)    {$ c_{X,Y}^{H\text{\Mod}} ~= $}
  \put(60,0)  {\Includepichtft{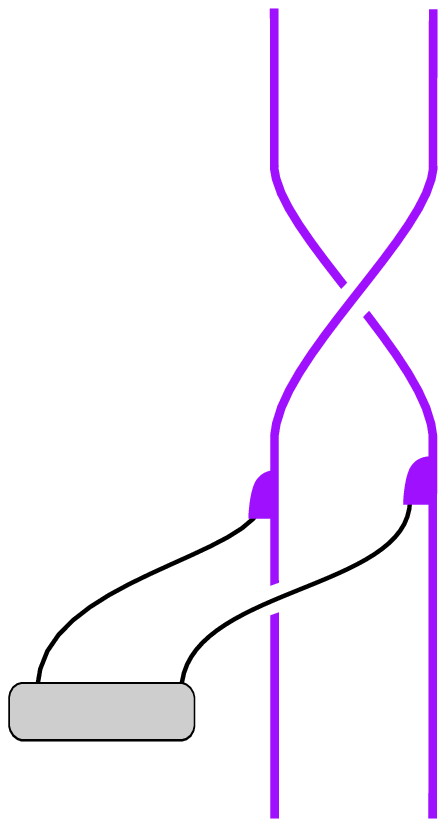}
  \put(24,-8.5) {\sse$ X $}
  \put(24.8,95){\sse$ X $}
  \put(42.3,-8.5) {\sse$ Y $}
  \put(43.5,95){\sse$ Y $}
  \put(-6,3)   {\sse$ R $}
  \put(42,56)   {\sse$ \tau_{X,Y}^{} $}
  } }
The braid relations, \eqref{braid_rel}, follows directly from \eqref{braid_Hmod} and \eqref{R_YB}.
The inverse braiding is obtained by replacing $R$ in \eqref{braid_Hmod} by $R^{-1}_{21}=\flip HH\circ R^{-1}$.

\subject{Factorizable ribbon Hopf algebras} A quasitriangular Hopf algebra $(H,R)$ is called \emph{factorizable}\index{Hopf algebra!factorizable} if the monodromy matrix $Q$ can be written as $Q=\sum_\ell h_\ell\oti k_\ell$ where $\{h_\ell\}$ and $\{k_\ell\}$ are two vector space bases of H. We will in particular be interested in the situation that $H$ is a factorizable ribbon Hopf algebra. In this case $H$ is unimodular \cite[Proposition 3(c)]{radf13}, which means that the left integral $\Lambda$ is also a right integral. This, in turn, implies $\apo\circ\Lambda=\Lambda$. Recall that the integral and cointegral of a finite dimensional Hopf algebra are unique up to normalization. We choose the normalization such that
\be\label{norm_int_1}
    \lambda\circ\Lambda=1\,.
\ee
In addition, when $H$ is unimodular the right cointegral $\lambda$ satisfies
\be\label{OS2}
    \lambda\cir m \eq \lambda\cir\flip HH\cir(\id_H\oti\apo^2)\,,
\ee
which implies \cite{coWe6}
\be\label{OS2_mod}
   \lambda\circ m\circ ((\apo\cir m)\oti\id_H)
   = \lambda\circ m\circ [\apo\oti(m\cir\flip HH\cir(\apoi\oti\id_H))] \,.
\ee
\begin{conv}
    Below we will, unless explicitly stated otherwise, always take $H$ to be a finite-dimensional factorizable ribbon Hopf algebra over an algebraically closed field $\k$ of characteristic zero.
\end{conv}
\ssubject{The Drinfeld map} The \emph{Drinfeld}\index{Drinfeld map} map is the linear map
\eqpic{Drin_def} {120} {28} {
  \put(0,28)    {$ f_Q ~:= $}
  \put(50,0){\Includepic{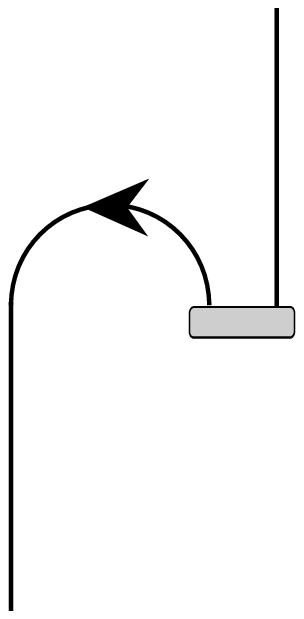}
  \put(-4,-8.5) {\sse$\Hs$}
  \put(27,69) {\sse$H$}
  \put(23,23) {\sse$Q$}
  }
  \put(105,28)    {$ \in\Hom_\k(\Hs,H) $\,.}}
For a quasitriangular Hopf algebra $H$, the Drinfeld map intertwines the coadjoint and adjoint left actions of $H$ \cite[Proposition 2.5 (5)]{coWe4}
\be\label{drin_int}
    f_Q\In\Hom(\Hs_{\triangleright},H_\diamond)\,.
\ee
In addition, for a finite-dimensional quasitriangular Hopf algebra, $f_Q$ is invertible iff $H$ is factorizable. If $H$ is factorizable $f_Q$ maps any non-zero cointegral $\lambda$ to a non-zero integral $\Lambda$, see \cite[Remark 2.4]{coWe5} and \cite[Theorem 2.3.2]{geWe3}. We choose the normalization of $\lambda$ and $\Lambda$ such that in addition to \eqref{norm_int_1} we have
\be\label{norm_int_2}
    f_Q(\lambda)=\Lambda\,.
\ee
Together with \eqref{norm_int_1} this fixes the normalization of $\Lambda$ and $\lambda$ up to a common factor $\pm1$.

By \cite[Lemma 2.5]{coWe5} the inverse of the Drinfeld map is obtained from
\be\label{Drin_inv}
    f_Q\circ\Psi\circ f_{Q^{-1}}\circ\Psi=1\,,
\ee
where $f_{Q^{-1}}$ is given by the expression for the Drinfeld map with $Q$ replaced by $Q^{-1}$
Composing \eqref{Drin_inv} with $\Psi^{-1}$ from the right we obtain the first of the following two equalities:
\Eqpic{fQS_Psi} {320} {35} { \put(0,-2){\setulen80
   \put(0,0) { \includepichtft{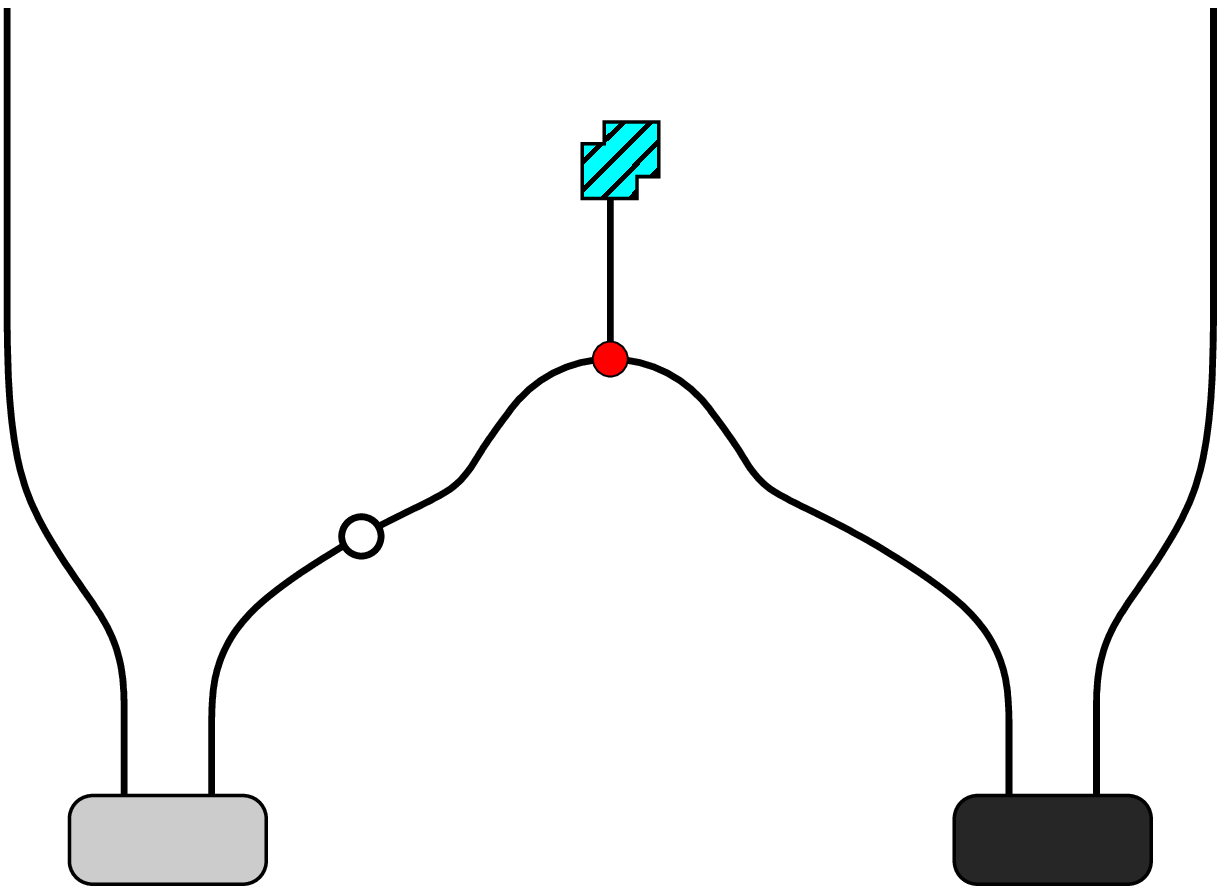}
   \put(-2.5,100){\sse$ H $}
   \put(14.5,3.4){\sse$ Q $}
   \put(43.6,35) {\sse$ \apo $}
   \put(73.1,75.5) {\sse$ \lambda $}
   \put(128.3,3.2) {\sse$ Q^{-1} $}
   \put(130,100) {\sse$ H $}
   }
   \put(162,46)  {$ = $}
   \put(183,13) { \includepichtft{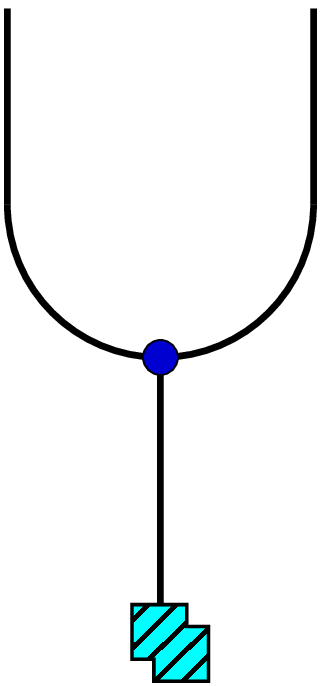}
   \put(-2.2,79) {\sse$ H $}
   \put(7.7,-1)  {\sse$ \Lambda $}
   \put(30.8,79) {\sse$ H $}
   }
   \put(241,46)  {$ = $}
   \put(274,0) { \includepichtft{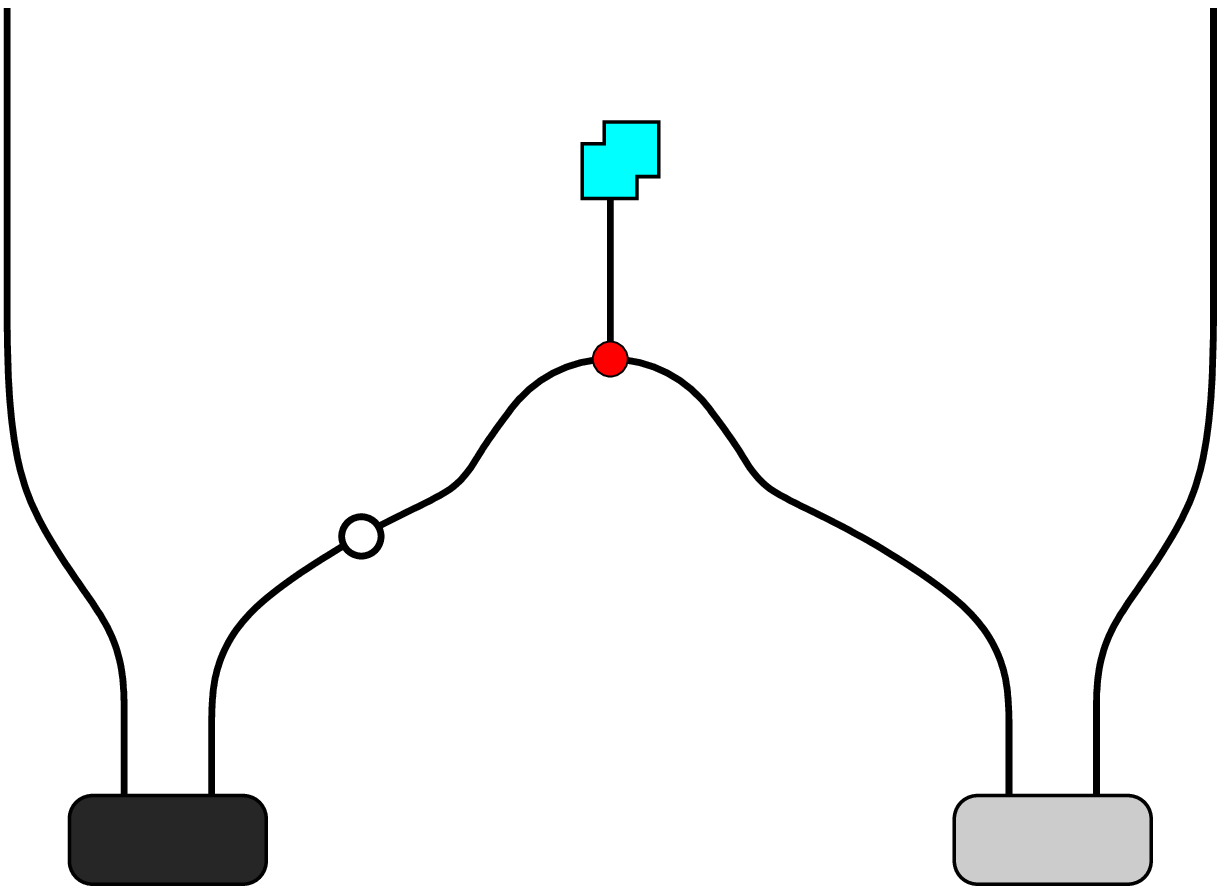}
   \put(-2.5,100){\sse$ H $}
   \put(31.1,3.2){\sse$ Q^{-1} $}
   \put(39.5,32.6) {\sse$ \apo $}
   \put(73.7,85.4) {\sse$ \lambda $}
   \put(112.5,3.4) {\sse$ Q $}
   \put(130,100) {\sse$ H $}
   } } }
The second equality is obtained the same way as the first one by first by rewriting \eqref{Drin_inv} as $f_{Q^{-1}}\circ\Psi\circ f_Q\circ\Psi=1$.
By using that $\apo$ is an anti-algebra morphism and \eqref{OS2} it is straightforward to check that
\be
    \Psi\In\Hom(H_\diamond,\Hs_{\triangleright})\,.
\ee
Combining this result with \eqref{Drin_inv} and the intertwining property \eqref{drin_int} of $f_Q$, it follows that also $f_{Q-1}$ intertwines the adjoint and coadjoint left-modules:
\be\label{norm_int_inv}
    f_{Q^{-1}}\In\Hom(\Hs_{\triangleright},H_\diamond)\,.
\ee
In addition, composing the first equality in \eqref{fQS_Psi} with a counit on the right, it follows that $f_{Q^{-1}}$ maps the cointegral to the integral
\be
    f_{Q^{-1}}(\lambda)=\Lambda\,.
\ee
\begin{remark}
    Here we have discussed Hopf algebras in \Vectk. However many of the notions described above have generalizations. For example, if \H\ is a Hopf algebra in any symmetric category, such that the category of \H-modules is braided, the braiding of the latter can be written in terms of an $R$-matrix. The notion of integral can be generalized to any rigid braided category \cite{lyub8}. In the more general setting one starts from the collection of morphisms in $\Hom(X,H)$ for any object $X$ satisfying relations analogous to the ones for integrals. One then defines the \emph{object of integrals}, $\text{Int}\H$, in terms of a universal property with respect to these morphisms, see \cite[definition 3.1 \& Proposition 3.1]{bklt}. In the case of Hopf algebras in \Vectk\ the object of integrals is simply the trivial $H$-module.
\end{remark}
\subsection{The ribbon category $H$\Bimod}
For any  finite-dimensional ribbon Hopf algebra $H$, we can equip $H$\Mod\ as well as $H$\Bimod\ with the structure of a ribbon category.
We will describe the ribbon structure of $H$\Bimod\ in detail. The ribbon structure of $H$\Mod\ can then be obtained by forgetting about the right action of $H$ in the relevant formulas.

Recall from section \ref{sec:Hmod_mon} that the tensor product of $H$-bimodules, defined in \eqref{TP_Hopf}, with the monoidal unit given by the trivial bimodule $(\k,\eps,\eps)$, equips $H$\Bimod\ with the structure of a monoidal category.
The ribbon structure of $H$\Bimod\ is the following:
\begin{itemize}
    \item \textbf{Duality:} The right (left) dual object of the $H$-bimodule $X$ is the dual vector space $X^*$ endowed with the left and right actions $\rhov$ and $\ohrv$ ($\rhoV$ and $\ohrV$) given by:
\eqpic{Hbim_rightdualactions} {180} {30} {\setulen80
  \put(0,0){
  \put(0,35)    {$\rhov ~:= $}
  \put(34,0)  {\includepichtft{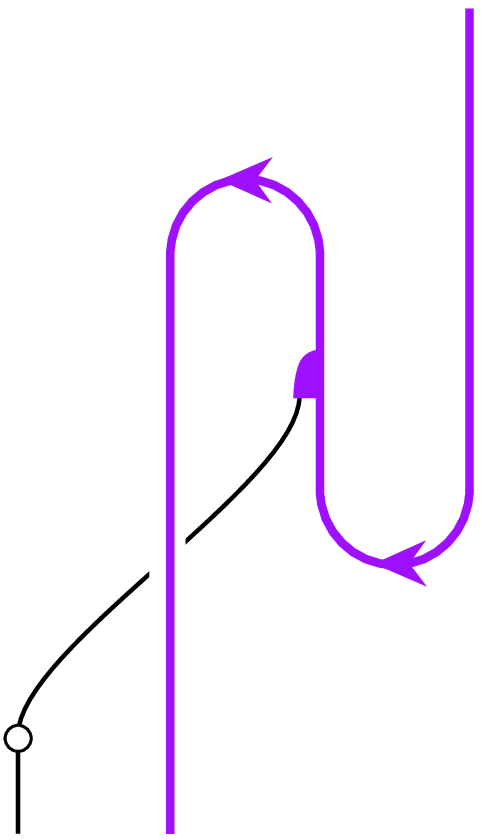}
  \put(-3,-8.5) {\sse$ H $}
  \put(12,-8.5) {\sse$ X^*_{} $}
  \put(25.3,49.8) {\sse$ \rho $}
  \put(45.2,94) {\sse$ X^*_{\phantom|} $}
  }
  \put(126,35)  {$\ohrv ~:= $}
  \put(175,0) {\includepichtft{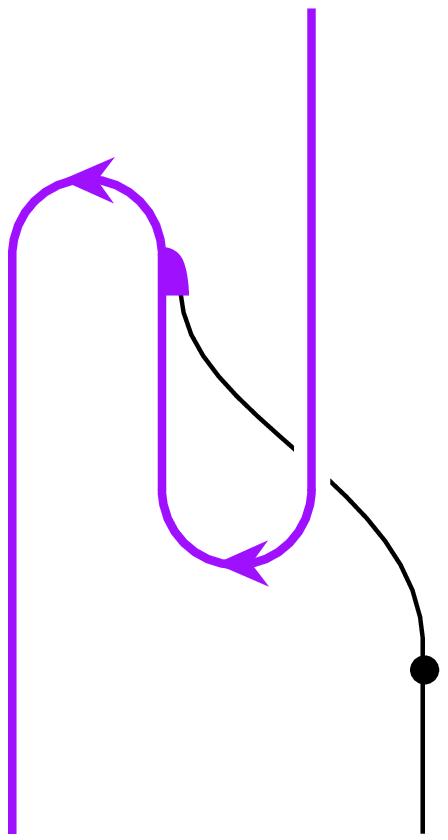}
  \put(-5,-8.5) {\sse$ X^*_{} $}
  \put(11,60.7) {\sse$ \ohr $}
  \put(29.8,94) {\sse$ X^*_{\phantom|} $}
  \put(41,-8.5) {\sse$ H $}
  }} }
  and
\eqpic{Hbim_leftdualactions} {180} {35} {\setulen80
  \put(0,10){
  \put(0,35)  {$\rhoV ~:= $}
  \put(34,0) {\includepichtft{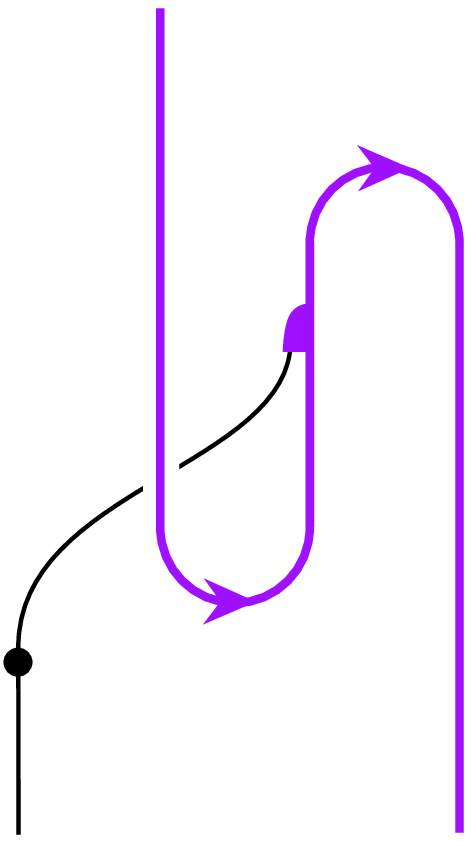}
  \put(-4,-8.5) {\sse$ H $}
  \put(13.6,94) {\sse$ X^*_{\phantom|} $}
  \put(35,54.7) {\sse$ \rho $}
  \put(44,-8.5) {\sse$ X^*_{} $}
  }
  \put(126,35)  {$\ohrV ~:= $}
  \put(175,0) {\includepichtft{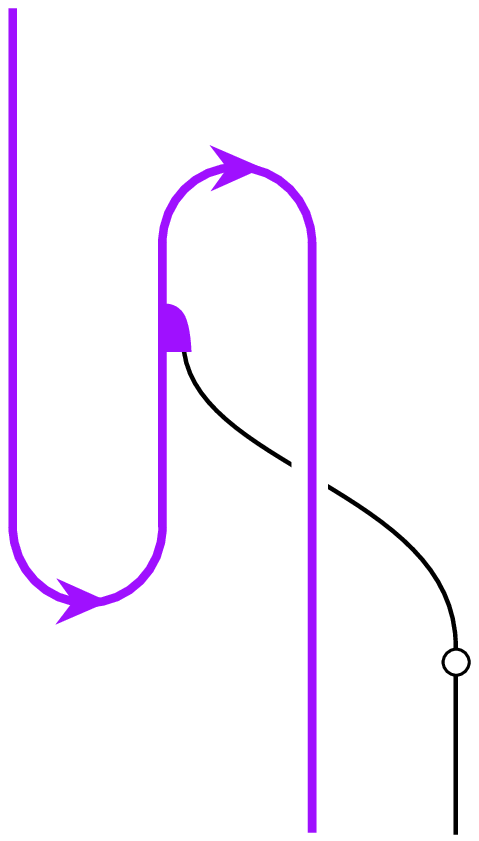}
  \put(11,54.1) {\sse$ \ohr $}
  \put(27,-8.5) {\sse$ X^*_{\phantom|} $}
  \put(-2.9,94) {\sse$ X^*_{\phantom|} $}
  \put(44,-8.5) {\sse$ H $}
  } } }
    The (co)evaluation morphisms appearing here are the ones of $\Vectk$. It is straightforward to check that they are indeed bimodule morphisms. The left and right duality functors map a morphism $f$ to the dual linear map $f^*$.\\
    Note that the coregular bimodule is the right dual of the regular bimodule, c.f. \eqref{rhoHb,ohrHb}
    \item \textbf{Braiding:} We have already seen that the $R$-matrix endows $H$\Mod\ with a braiding. Analogously the $R$-matrix endows \rMod$H$ with a braiding. Combining them we can have the following braiding in $H$\Bimod:
    \eqpic{bibraid} {120} {47} {
  \put(0,47)    {$ c_{X,Y}^{} ~= $}
  \put(60,0)  {\Includepichtft{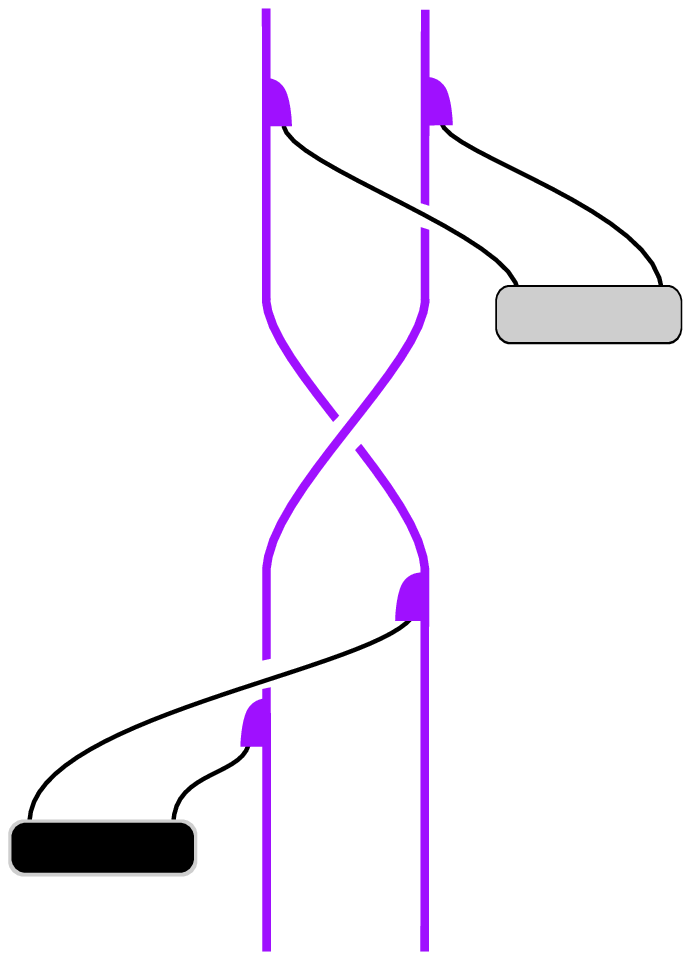}
  \put(24,-8.5) {\sse$ X $}
  \put(24.8,107){\sse$ Y $}
  \put(42.3,-8.5) {\sse$ Y $}
  \put(43.5,107){\sse$ X $}
  \put(-17,9)   {\sse$ R^{-1} $}
  \put(75,68)   {\sse$ R $}
  \put(42,56)   {\sse$ \tau_{X,Y}^{} $}
  } }
  Just like in the $H$\Mod\ case it follows from \eqref{R_YB} (and the analogous relations for $R^{-1}$) that \eqref{bibraid} indeed satisfies the braid relations \eqref{braid_rel}.
    \item\textbf{Twist:} The twist of $H$\Bimod\ is obtained by acting with the ribbon element $v$ from the left and $v^{-1}$ from the right \cite[Lemma 4.8]{fuSs3}
   \eqpic{bimod-twist} {90} {31} {
    \put(0,35)    {$ \theta_X ~= $}
  \put(50,0) {\Includepichtft{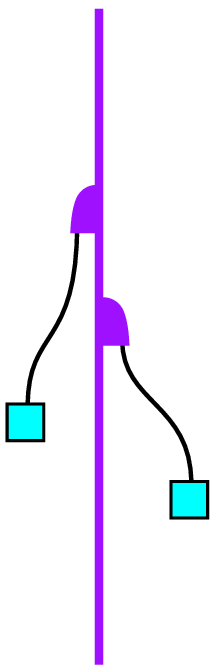}
  \put(6.1,-8.5) {\sse$X$}
  \put(6.9,76)   {\sse$X$}
  \put(-5.4,25.7){\sse$v$}
  \put(24.3,17)  {\sse$v^{-1}$}
  } }
The equalities in \eqref{prop_v} assure the compatibility with the duality and the braiding and that the trivial module has trivial twist.
\end{itemize}
\ssubject{Sovereign structure} A quasitriangular Hopf algebra with invertible antipode possess a distinguished invertible element $u$, known as the \emph{canonical element} or  the  \emph{Drinfeld element}, defined by
\be
    u:=m\circ(\apo\oti\id_H)\circ\flip HH\circ R\,.
\ee
The element $u$ satisfies $\ad_u=\apo^2$
\cite[prop VIII.4.1]{KAss} and consequently $(\apoi)^2=\ad_{u^{-1}}$. The product of $u$ and the inverse of the ribbon element
\be\label{t_def}
    t:=uv^{-1}\equiv m\circ(u\oti v^{-1})
\ee
is the \emph{special group-like element}. A group-like element $g$ is an element that satisfies $\Delta\circ g=g\oti g$. Since $v^{-1}$ is central, the relation $\ad_u=\apo^2$ with the antipode is inherited  by $t$:
\be\label{apo_t}
    \apo^2=\ad_t\,.
\ee
$H$\Bimod\ is sovereign \cite[Lemma 4.1]{fuSs3} with the isomorphism $\piv_X$ of the sovereign structure and its inverse for the bimodule $X$ given by the linear maps
\eqpic{pivX} {265} {37} {
   \put(0,39)    {$ \piv_X ~:= $}
   \put(50,0)  {\Includepichtft{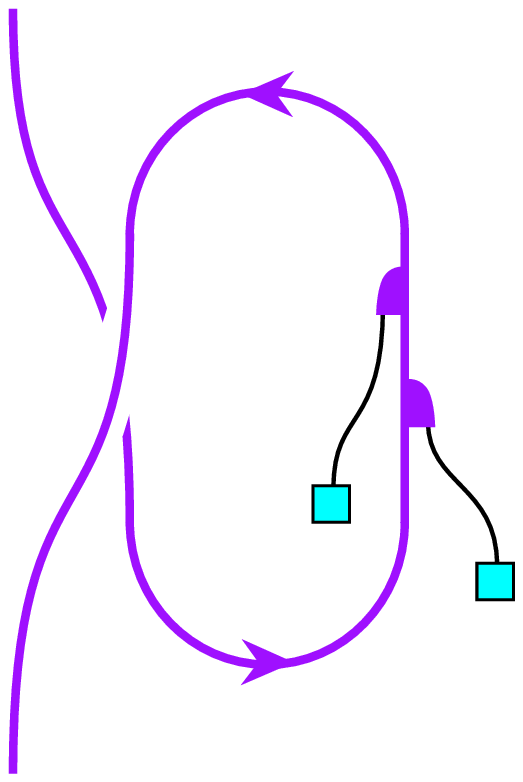}
   \put(-4.5,-8) {\sse$ \Xs $}
   \put(-3,88)   {\sse$ \Xs $}
   \put(28,28)   {\sse$ t $}
   \put(58,20)   {\sse$ t $}
   }
      \put(150,39)    {$ \piv_X^{-1} ~:= $}
   \put(200,0)  {\Includepichtft{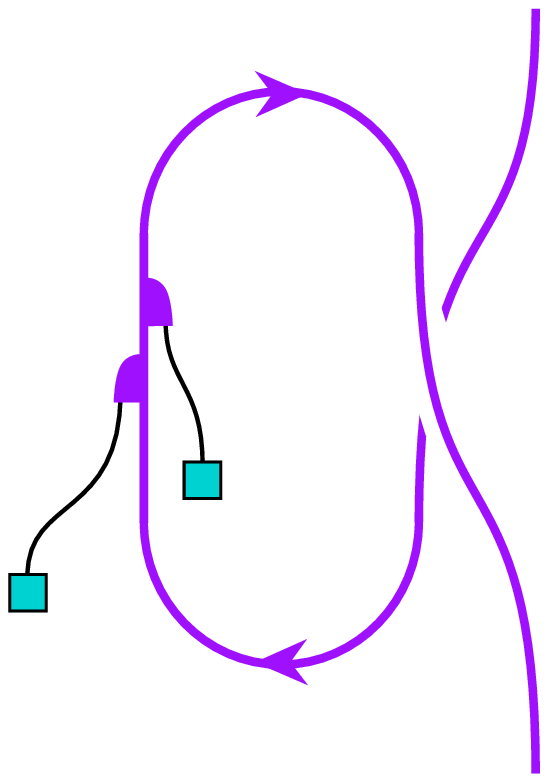}
   \put(53,-8.5) {\sse$ X^*_{\phantom|} $}
   \put(54,88)   {\sse$ X^*_{\phantom|} $}
   \put(24.9,30.7) {\sse$t^{-1}$}
   \put(-14,18.5){\sse$t^{-1}$}
   }
   }
It follows from \eqref{apo_t}, and the fact that $t$ is group-like, that these linear maps are indeed monoidal bimodule morphisms. Naturality then follows from properties of the duality, see \cite[Lemma 4.1]{fuSs3}.
\begin{remark}
(i)
The formula \eqref{bimod-twist} for the twist is in fact obtained from the more general formula
\eqpic{twist_NS} {120} {37} {
   \put(0,40)      {$ \theta_X ~= $}
   \put(52,0)   {\Includepichtft{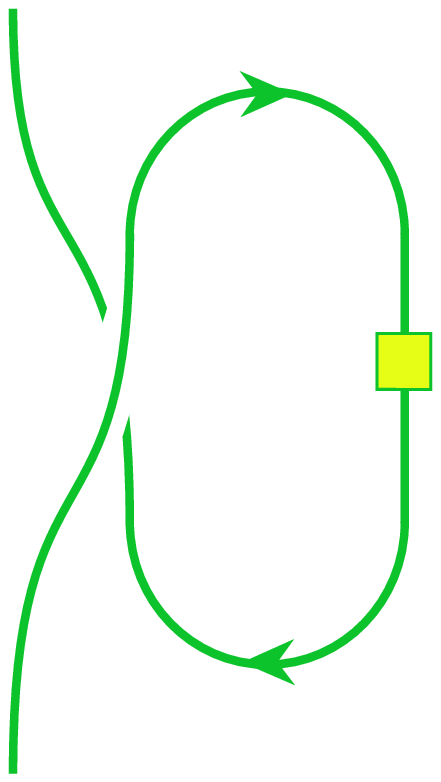}
   \put(-4.1,-8.5) {\sse$X$}
   \put(-2.5,88)   {\sse$X$}
   \put(15.6,34)   {\begin{turn}{90}\sse$c_{X,X}^{}$\end{turn}}
   \put(47.9,44.1) {\sse$\piv_{\!X}$}
   }
   }
for the twist of a sovereign braided monoidal category, see \cite[Lemma 4.8]{fuSs3}.
\nxl1
(ii)
Note that upon forgetting the right action on bimodules we obtain $H$\Mod\ with the inverse braiding and twist, c.f. \eqref{braid_Hmod} and comment afterwards.
This is a deliberate choice we make for $H$\Bimod\ that will be motivated in section \ref{sec:equiv_cat}, see Remark \ref{rem:bibraid_choice}. We denote the so-obtained ribbon category by $(\HMod)\rev$, c.f. \eqref{braid_twist_rev}.
\nxl1
(iii) To obtain the ribbon category $H$\Mod\ from the ribbon category $H$\Bimod, we use instead the braiding defined in \eqref{braid_Hmod} and the define the twist as in \eqref{bimod-twist}, but with the inverse ribbon element $v^{-1}$ acting from the left.
\nxl1
(iv)
We have not assumed that $H$ is semisimple. $H$\Mod\, is semisimple iff $H$ is semisimple. If $H$ is semisimple
$H$\Mod\ is in fact also modular.
\end{remark}
\begin{conv}
    Below $H$\Bimod\ denotes the category of $H$-bimodules with the ribbon structure just introduced. Analogously, $H$\Mod\ denotes the category of left $H$-modules with the analogous ribbon structure.
\end{conv}

\section{Topological field theory}\label{TFT}
In this section we describe the notion of \emph{3-dimensional topological field theory} (3d TFT for short) in the sense developed in \cite{retu,retu2,TUra,tuVi}. See also \cite{BAki,walk}.
\subsection{The cobordism category \CobC}\label{sec_Cob}
Given a strictly sovereign modular tensor category \C, there is another \emph{geometric} category \CobC\ associated to \C. By geometric we mean that objects as well as morphisms are manifolds. The underlying modular category \C\ enters in terms of decoration of the objects and morphisms. Explicitly:
\begin{itemize}
    \item The \textbf{objects} of \CobC\ are \emph{extended surfaces}\index{extended surface}. An extended surface $\ES$ is a compact oriented closed two-dimensional surface with a finite number of \emph{marked points} and a choice of \emph{Lagrangian subspace} $\lambda\In H_1(\ES,\R)$.
    \item The \textbf{morphisms} of \CobC\ are \emph{cobordisms}\index{cobordism}. A cobordism $\cob(\ES,\ES')$ is a compact oriented three-manifold $\cob$, bounded by $-\ES\sqcup\ES'$, with an embedded \emph{ribbon graph} such that there is a ribbon ending at each marked point.
\end{itemize}
Let us describe the notions appearing here. A \emph{marked point} $(p_i[\gamma_i],V_i,\eps_i)$ is defined by a quadruple of data: $p_i$ is a point on $E$, $[\gamma_i]$ is a germ of arcs such that $\gamma_i(0)=p_i$, $V_i$ is an object in $\C$ and $\eps_i\In\{\pm1\}$. A germ of arcs is an equivalence class of embeddings $\gamma$ of an interval $[-\delta,\delta]$ into $E$. Two embeddings $\gamma\mapdef[-\delta,\delta]\,{\rightarrow}\, E$ and $\gamma'\mapdef[-\delta',\delta']\,{\rightarrow}\, E$ are equivalent iff there is an $\eps<\delta,\delta'$ such that the two embeddings coincide when restricted to $[-\eps,\eps]$. We will refer to a germ of arcs such that $\gamma(0)=p$ as an arc germ at $p$. $\eps_i$ is related to orientation; we will soon describe its precise meaning.

For a symplectic vector space $H$, with a symplectic form $\omega$, a \emph{Lagrangian subspace} $\lambda\subset H$ is defined as a maximal subspace $\lambda$ such that $\omega(\lambda,\lambda)=0$.

A \emph{ribbon}\index{ribbon} is an embedding of a rectangle in a three-manifold together with a choice of 2-orientation and orientation of its core.
By this we mean that we choose a natural parametrization $\phi:[0,1]\Times[-\frac1{10},\frac1{10}]\rightarrow \cob$ of the ribbon. This defines an orientation of the ribbon. The interval $[0,1]\times\{0\}$, two which we refer to as the core of the ribbon, has a natural 1-orientation as a subset of the $X$-axis.
Each ribbon is labeled by an object of \C. A ribbon ending at a marked point is labeled by the object labeling that point. The label $\eps_i$ at this marked point is $+1$ if the core of this ribbon is pointing away from the surface and $-1$ otherwise. Ribbons can be joined giving rise to a \emph{ribbon graph}\index{ribbon!graph}. Ribbons are joined at a \emph{coupon}\index{coupon}, that is an embedding of an oriented rectangle with two preferred sides at which ribbons are attached. The coupon is labeled by a morphism in $\Hom(W_{\text{in}},W_{\text{out}})$, where $W_{\text{in}}\In\Obj(\C)$ is a tensor product of the objects decorating the ribbons on the bottom side of the coupon and $W_{\text{out}}\In\Obj(\C)$ is a tensor product of the objects coloring the ribbons on the top side. Here, a ribbon labeled by $U$ is interpreted as $U$ if the core is pointing towards the bottom side or away from the top side of the coupon and as $U^\vee$ otherwise,
As an illustration, the coupon in
\eqpic{coupon_ex}{50}{25}{
\put(0,0)   {\Includepic{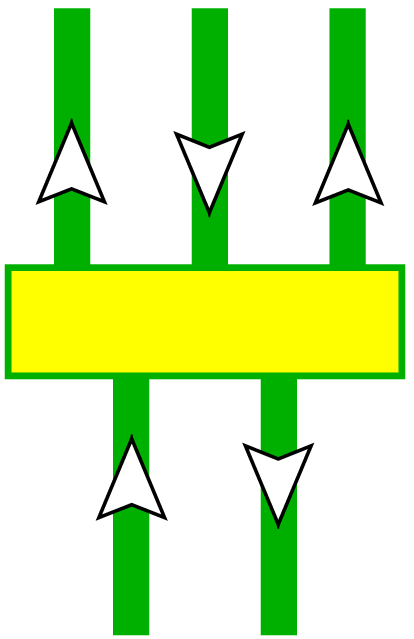}
\put(-3.2,60) {\scriptsize$X$}
  \put(14,60)   {\scriptsize$Y$}
  \put(42,60)   {\scriptsize$Z$}
  \put(20,33.3) {\scriptsize$\varphi$}
  \put(4.3,3)   {\scriptsize$U$}
  \put(34,3)    {\scriptsize$V$}
}
}
is labeled by a morphism in $\Hom(U\oti V^\vee,X\oti Y^\vee\oti Z)$.

We display an example of an extended surface $E$ and  a cobordism \linebreak$(M,E,E')\colon E\rightarrow E'$ in the following picture:
\Eqpic{cob}{300}{70}{
  \put(0,20){ \put(0,0){\Includepic{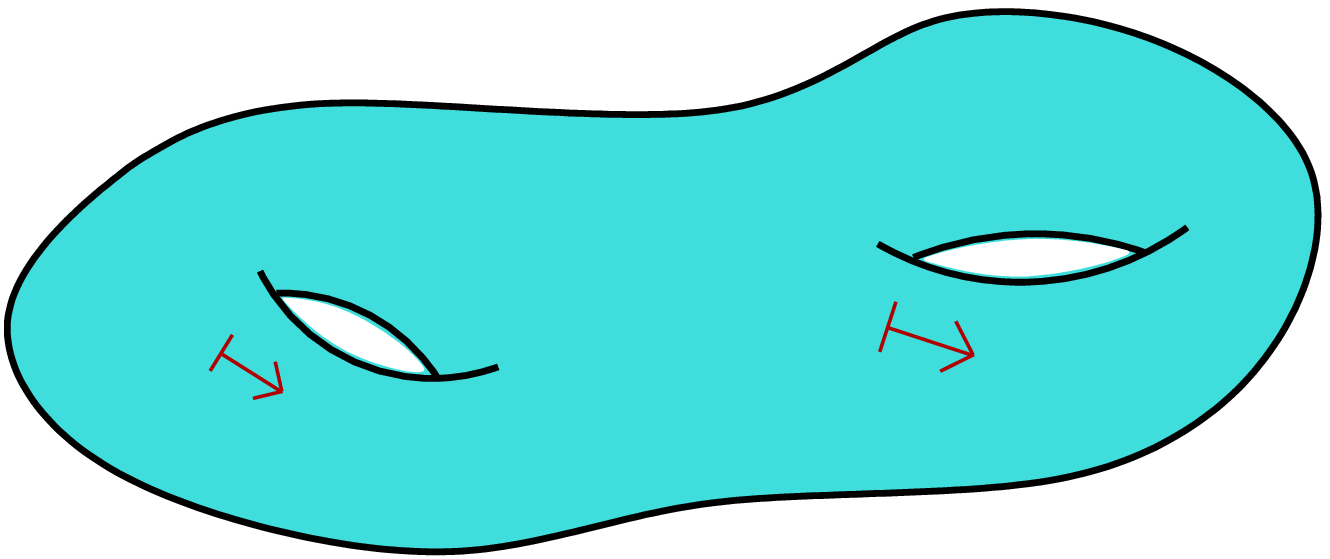}}
  \put(8,18)     {\begin{turn}{-32}\pl {(U,+)}\end{turn}}
  \put(85,16)     {\begin{turn}{-17}\pl {(V,-)}\end{turn}}
  \put(-5,5)   {\pl E}
  }
  \put(170,-5){ \put(0,0){\Includepic{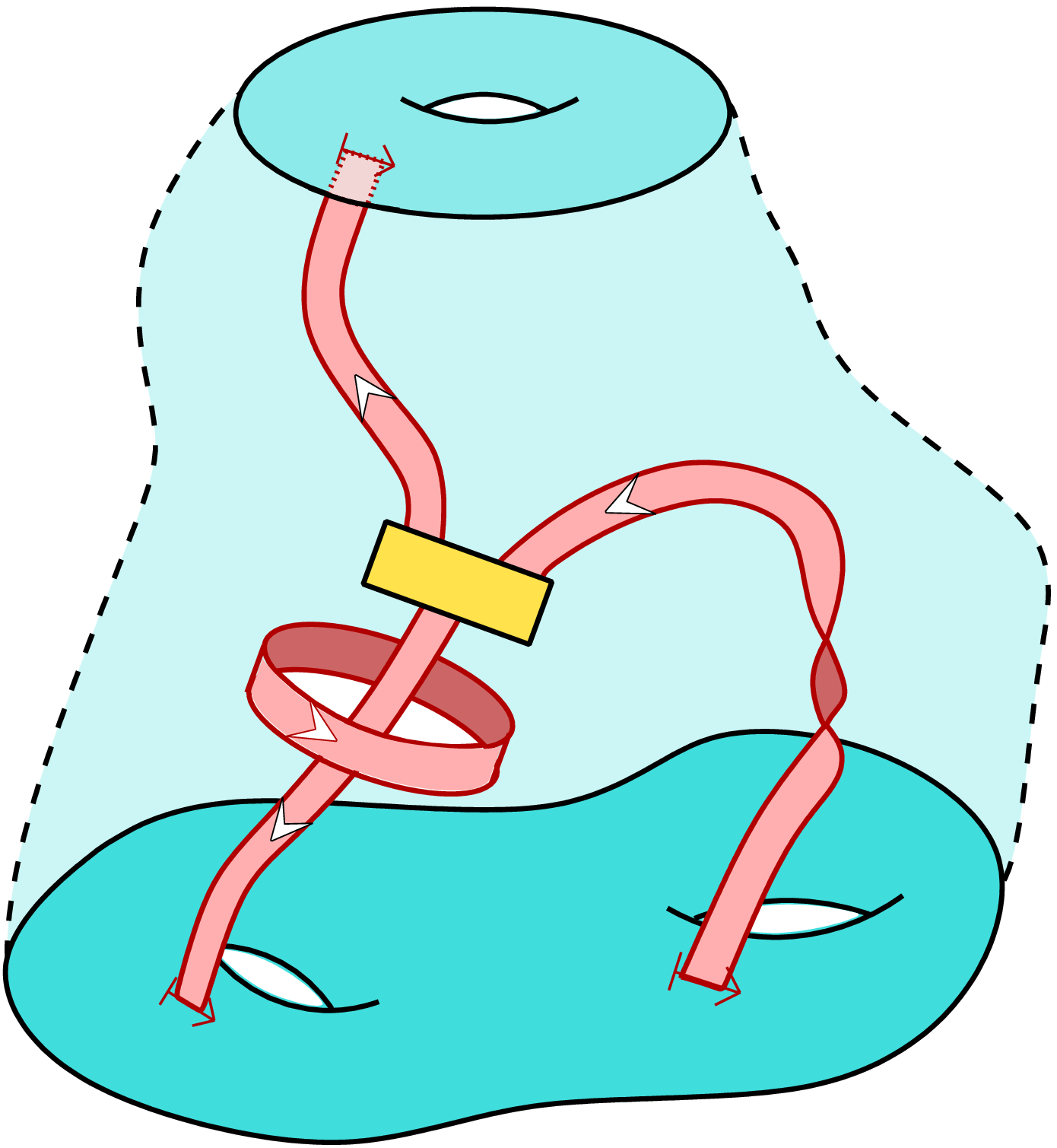}}
  \put(8,18)     {\begin{turn}{-32}\pl {(U,-)}\end{turn}}
  \put(85,16)     {\begin{turn}{-17}\pl {(V,+)}\end{turn}}
  \put(-8,2)   {\pl {-E}}
  \put(42,35) {\pl{U}}
    \put(121,90){\pl{ V}}
    \put(54,114){\pl{ W}}
    \put(31,73) {\pl{ V}}
    \put(63,79.5) {\pl f}
    \put(58,140){\pl{(W,-)}}
    \put(98,160){\pl{E'}}
  }
  }
\begin{remark}\label{Rem_cob}
    (i) We will think of the two possible orientations of the ribbons as if the ribbon is showing one of its two sides. Thus we will say that a ribbon with its preferred 2-orientation is showing its white side, whereas a ribbon with the opposite 2-orientation is showing its black side.\nxl1
    (ii) In case an extended surface $E=\partial M$ appears as the boundary of some oriented three-manifold $M$, a natural choice of Lagrangian subspace is the kernel of the inclusion map $H_1(E,\mathbb R)\rightarrow H_1(M,\R)$, see \cite[Section IV:4.1]{TUra}. This is indeed the situation that will be most relevant for us. For instance if $E$ is a torus, $H_1(E,\R)$ is spanned by the two standard non-contractible cycles. With this convention we choose, as Lagrangian subspace, the span of the cycle that becomes contractible when we "fill out" the torus to a solid torus. \nxl1
    (iii) We will display ribbon graphs in \emph{blackboard framing}. This means that a ribbon that shows its white side is drawn as a solid line and a ribbon that shows its black side is drawn as a dashed line. A ribbon rotated by an angle $\pi$, is depicted as in the following picture
    \eqpic{BBF}{300}{30}{
  \put(120,-10){ \put(0,0){\Includepic{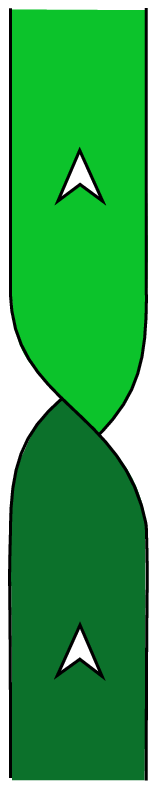}}
  }
  \put(170,30) {$=$}
  \put(220,-10){ \put(0,0){\Includepic{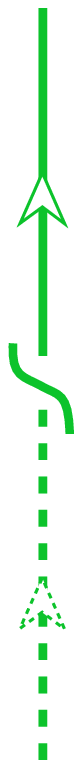}}
  }
  }
\end{remark}

Composition of two cobordisms $(\cob_1,\ES_1,\ES'_1)$ and $(\cob_2,\ES_2,\ES'_2)$ is achieved via a homeomorphism $f:\ES'_1\rightarrow\ES_2$. It is clear that this composition is associative.
The category \CobC\ has a natural structure of a monoidal category with the tensor product being the disjoint union of extended surfaces and cobordisms. The monoidal unit is given by the empty set.

\subsection{3d TFT and the \tftfun}\label{sec:tftfun}
A 3d TFT is a symmetric monoidal functor
\be
    \tft:\CobC\rightarrow\Vect\,.
\ee
Thus, to any extended surface $\ES$, the \tftfun\ assigns a vector space
\be
    \tft(\ES)=:\tfts(\ES)\in\Vect,
\ee
and to any cobordism $(\cob,\ES,\ES')$ it assigns a linear map
\be
    \tft(\cob,\ES,\ES')=:\tftm(\cob,\ES,\ES')\equiv\tftm(\cob)\in\Hom_{\Vect}(\tfts(\ES),\tfts(\ES'))\,.
\ee
To describe a \tft-functor, these mappings must satisfy a number of axioms; these are listed in e.g.\ \cite[Definition 5.1.13]{BAki}. In particular the \tft-functor assigns to any homeomorphism $f:\ES\rightarrow\ES'$ of extended surfaces an isomorphism
\be\label{homeomap}
    f_\sharp:\tfts(\ES)\rightarrow\tfts(\ES')
\ee
of vector spaces which depends only on the isotopy class of $f$. As a consequence $\tfts(\ES)$ carries a (projective) action of the mapping class group\index{mapping class group!representation of} of $\ES$, such that any mapping class $[f]$ is represented as $f_\sharp$.

By projecting a piece of ribbon graph locally to $\R^2$ in a non-singular manner\footnote{This means that we first
project ribbons that show their white sides to $\R^2$, and then shrink them to lines.
We do that in such a way that not more than two points on  the lines are mapped to the same point in the plane. Two points are allowed to be mapped to the same point only at isolated points, in which case we interpret the resulting morphism as a braiding. In addition, whenever a ribbon is twisted around its core by an angle $2\pi$ we interpret this as a twist morphism.} we can interpret the result as a morphism in \C\ and perform manipulations of this morphism following the rules of graphical calculus. Manipulations of this kind leaves $\tftm(\cob)$ invariant. In addition, $\tftm(\cob)$ is invariant under any homeomorphism that restricts to the identity on $\partial\cob$. We will refer to $\tftm(\cob)$ as the \emph{invariant} of $\cob$.

In fact, the \tft-functor is in general only a projective functor. That means that if two cobordisms $(\cob_1,\ES_1,\ES'_1)$ and $(\cob_2,\ES_2,\ES'_2)$ are composed via a homeomorphism $f:\ES'_1\rightarrow\ES_2$, the \tft-functor
    satisfies
\be
    \tftm(M,E_1,E_2')=\kappa^m\tftm(M_2,E_2,E_2')\circ f_\sharp\circ\tftm(M_1,E_1,E_1')\,,
\ee
where $f_\sharp:\tfts(E_1')\rightarrow\tfts(E_2)$ is the linear map assigned to $f$ by the \tft-functor, $m$ is an integer and $\kappa$ is the charge of the theory.
However, when using the \tft-functor to construct CFT, the extended surfaces of interest will often be doubles of world sheets. Remember that the double comes with an orientation reversing involution, which implies \cite[Lemma 2.2]{fffs3} that $m=0$. Thus whenever we glue at such a surface with an orientation reversing involution the \emph{gluing anomaly $\kappa^m$} vanishes. In addition, we will always glue by the identity. Thus in the applications we will consider we indeed have
\be
		\tftm(M,E_1,E_2')=\tftm(M_2,E_2,E_2')\circ\tftm(M_1,E_1,E_1').
\ee
As a consequence, in the construction of CFT via the \tftfun, mapping class groups act genuinely.
\section{Further categorical notions}\label{sec:add_cat}
In this section we collect additional notions from category theory that we need in the rest of the thesis. We will introduce the categories $\C\op$ and $\C\rev$, describe a finiteness condition, which we will assume later in the text, and finally introduce the notions of Deligne product and of coend.

\ssubject{The opposite category} For any category \C\ the  \emph{opposite category} $\C\op$ is the category whose objects are the objects of \C\ and for any morphism $f\in\Hom(U,V)$ in \C\ there is a morphism $f\op\In\Hom(V,U)$ in $\C\op$. Note that the duality furnishes a functor from $\C$ to $\C\op$.

\ssubject{The category $\C\rev$} Given a braided category \C\ the category $\C\rev$ is obtained from \C\ by exchanging the braiding by its inverse. As a consequence, for a ribbon category, also the twist is replaced by the inverse twist:
\be\label{braid_twist_rev}
    c^{\,\C\rev}_{U,V}= (c^\C_{V,U})^{-1}\quand\theta^{\,\C\rev}_U=\theta_U^{-1}\,,
\ee
for all $U,V\in\Obj(\C)$.
\ssubject{Subobjects and quotient objects} A morphism $e\In \Hom(U,V)$ is an \emph{epimorphism} iff for any two morphisms $f_1$ and $f_2$ in $\Hom(V,W)$, $f_1\cir e=f_2\cir e$ implies $f_1\eq f_2$. The dual notion is the one of a monomorphism: A morphism $m\In\Hom(V,W)$ is a \emph{monomorphism} iff for any two $g_1$ and $g_2$ in $\Hom(U,V)$, $m\cir g_1\linebreak\eq m\cir g_2$ implies $g_1\eq g_2$.

We will need the categorical version of a submodules and quotient modules. We say that $f\in\Hom(U,V)$ factors through $g\in\Hom(W,V)$ iff there exists a morphism $f'\In\Hom(U,W)$ such that $f\eq g\circ f'$. There is an equivalence relation among the monomorphisms with common codomain $V$: We say that $f\,{\equiv}\ g$ if $f$ factors through $g$ and $g$ factors through $f$. The equivalence classes with respect to this equivalence relation are the \emph{subobjects}\index{subobject} of $V$. We will refer to any representative $f\In\Hom(U,V)$, and also to the domain $U$, as a subobject of $V$.

There is an obvious equivalence relation among epimorphisms: Two epimorphisms $r\In\Hom(U,V)$ and $s\In\Hom(U,W)$ are equivalent if $r\eq\theta\,s$ for some invertible $\theta\In\Hom(W,V)$. The equivalence classes of such epimorphisms are the \emph{quotient objects} of $U$. See e.g.\ \cite[Chapter V. 7]{MacL} for more details on subobjects and quotient objects.

\ssubject{Projective objects}

We are now in a position to define the notion of a projective object:
\begin{definition}
(i) An object $P$ is \emph{projective} iff every morphism $f\in\Hom( P,U)$ factors through every epimorphism $g\in\Hom(V,U)$. That is, there exists a morphism $f'\In\Hom(P,V)$ such that $f\eq g\cir f'$.
\nxl1
(ii) The \emph{projective cover} of an object $U$ is a pair $(P,\pi)$ where $P$ is a projective object and $\pi$ is an epimorphism in $\Hom(P,U)$ such that for any subobject $(N,m)$ of $P$ with $N\neq P$, $\im(\pi\circ m)\neq U$.
\end{definition}
\ssubject{Finiteness} For an object $U\in\C$ the Jordan-Hölder series is a sequence
\be
    U\equiv U_n\supset U_{n-1}\supset U_{n-2}\supset\cdots\supset U_2\supset U_1\supset U_0\equiv0
\ee
such that each $U_i/U_{i-1}$ is non-zero and simple. If there is such a finite series, the integer $n$ only depends on the object $U$, and we say that the object is of finite \emph{length} $n$.
We will use the following notion of local finiteness (see e.g.\ \cite[(2.12.1)]{deli}).
\begin{definition}\label{def_loc_fin}
    A  \emph{locally finite} category is a \k-linear abelian category such that all morphism spaces are finite-dimensional and every object has finite length.
\end{definition}
We will apply the stronger condition, c.f.\ \cite{etos}:
\begin{definition}\label{def_finite}
(i) A \emph{finite}\index{category!finite} category is a locally finite category such that there is a finite number of isomorphism classes of simple objects, each of them having a projective cover.\nxl1
(ii) A finite tensor category is a rigid monoidal category that is finite in the sense (i) and has simple tensor unit.
\end{definition}

In a finite category we can define the \emph{Cartan matrix} with elements given by
\be\label{Cartan_def}
    c_{i,j}:=[P_i\,{:}\, S_j]\,,
\ee
where $[P_i\,{:}\, S_j]$ is the multiplicity of the simple objects $S_j$ in the Jordan-Hölder series of the projective cover $P_i$ of the simple object $S_i$, see e.g. \cite{loren} for more details on Cartan matrices.

\subsection{Deligne product}\label{sec:Deligne}

We define the Deligne product of two abelian \k-linear categories \C\ and \cD\ in the sense of \cite[Definition 1.46.1]{egno}.
\begin{definition}
    Let \C\ and \cD\ be two \k-linear abelian categories. The \emph{Deligne product} of \C\ and \cD\ is a \k-linear abelian category $\C\boxti\cD$ together with a bifunctor $\boxti:\C\times\cD\rightarrow\C\boxti\cD$ that is right exact and \k-linear in both variables and satisfies the following universal property: for any bifunctor $G$ from $\C\times\cD$ to a \k-linear abelian category $\cE$ that is right exact and \k-linear in both variables there exists a unique right exact \k-linear functor $\Delun:\C\boxtimes\cD\rightarrow\cE$ such that $G\eq\Delun\circ\boxti$, i.e.\ such that the following diagram commutes:
\be\label{uni_deligne}
    \xymatrixcolsep{3pc}\xymatrix@R+8pt{
    \C\times\cD \ar[r]^{\boxtimes}\ar[dr]_{G}
    &\C\boxti\cD \ar@{.>}[d]^{\exists \,!\,\Delun}\\
    &\cE}
\ee
\end{definition}
Given two abelian \k-linear locally finite braided categories \C\ and \cD, their Deligne product is again a braided category, see e.g. \cite[Section 3.3]{rgaW}.

Among the objects of $\C\boxti\cD$ there are in particular the $\boxtimes$-\emph{factorizable} objects, which are of the form $U\btimes V\In\C\boxti\cD$ with $U\In\C$ and $V\In\cD$. Not any object in $\C\boxti\cD$ is of this form. However, any object in $\C\boxti\cD$ admits a resolution by $\boxtimes$-factorizable injective objects \cite[Section 1.46]{egno}. This allows one to make statements about all objects of $\C\boxti\cD$ by considering only the $\boxtimes$-factorizable ones.

\ssubject{Monoidal structure on $\C\boxti\cD$} If \C\ and \cD\ are in addition monoidal and locally finite, then $\C\boxti\cD$ is a monoidal category \cite[Proposition 5.17]{deli}. The tensor product of two $\boxti$-factorizable objects  is
\be\label{oti_Del}
    (U\boxti V)\oti_{\!\C\boxti\cD}\,(U'\boxti V')=(U\oti_{\!\C}\, U')\boxti(V\oti_{\!\cD}\, V')\,,
\ee
and analogously for morphisms.

\ssubject{$\C\boxti\cD$ as a braided category} If \C\ and \cD\ are locally finite abelian \k-linear \linebreak braided categories, then $\C\boxti\cD$  can be equipped with a braiding, see e.g.\ \cite[Section 3.3]{rgaW}. Denote by $c^\C$, the braiding of \C, and by $c^\cD$ the braiding of \cD. The braiding of  $\C\boxti\cD$ satisfies
\be\label{c_Deli}
    c_{U\btimes V,U'\btimes V'}=c_{U,V}\boxti c_{U',V'}\,.
\ee

\ssubject{The category $\C\boxti\C\rev$} Note that the defining property of a locally finite category (see Definition \ref{def_loc_fin}) is not only satisfied for any modular category, but the condition is far more general. E.g.\ the category $H$\Mod\ described above is in this class even when $H$ is not semisimple, and accordingly $H$\Mod\ is not semisimple.

We will later consider an abstract category \C\ that is in particular an abelian \k-linear locally finite braided category. In this case the category $\C\boxti\C\rev$ exists and is an abelian \k-linear locally finite braided category as well. This follows from \cite[Proposition 5.13]{deli} and the braiding described above.
For a more detailed discussion concerning Deligne products see e.g.\ \cite[Section 3.2-3.3]{rgaW}.
\subsection{Coends}\label{coends}
Consider a functor $F:\C\op\times\C\rightarrow \cD$. A \emph{dinatural transformation} from $F$ to an object $B\In\Obj(\cD)$ is a family of morphisms $\varphi=\{\varphi_X\mapdef F(X,X)\rightarrow B\,|\,X\In\C\}$ such that the diagram
\bee4010{\label{def_dina}
   \xymatrix @R+8pt{
   F(Y,X)\ \ar^{F(\idsm_Y,f)}[rr]\ar_{F(f,\idsm_X)}[d]&&\ F(Y,Y)\ar^{\varphi_Y^{}}[d] \\
   F(X,X)\ \ar^{\varphi_X^{}}[rr] && \, B
   } }
commutes for all morphisms $f\In\Hom(X,Y)$.

A \emph{coend}\index{coend} $(A,\iota)$ of a functor $F\mapdef C\op\Times\C\rightarrow \cD$ is an object $A$ together with a dinatural transformation $\iota:F\Rightarrow A$ satisfying a universal property: For any dinatural transformation $\varphi: F\Rightarrow B$ there is a unique morphism $h\mapdef A\rightarrow B$ that makes the  diagram
\bee6232{\label{coend_uni}
   \xymatrix @R+6pt{
   F(Y,X)\ \ar^{F(\idsm_Y,f)}[rr]\ar_{F(f,\idsm_X)}[d]
   && F(Y,Y)\ar^{\iota_Y^{}}[d]  \ar @/^2pc/^{\varphi_Y^{}}[ddr]& \\
   F(X,X)\ar @/_2pc/_{\varphi_X^{}}[drrr] \ar^{\iota_X^{}}[rr] && \, A \ar^{h}@{-->}[dr] & \\
   &&& \, B
   } }
commute for all $f\In\Hom(X,Y)$. If the coend exists it is unique up to unique isomorphism. The notation for the coend is $\Coend FX$. Even though the a coend is a pair $(A,\iota)$, we will, in what follows, for brevity refer to the object $A$ as the coend.
\begin{remark}\label{rem:morph_coend}
An important property when working with coends is that a morphism $f:\Coend FX\rightarrow Y$ is equivalent to a dinatural family $\{f_X|X\In\Obj(\C)\}$ from $F(X,X)$ to $Y$. To see that this holds, assume that $f\In\Hom(\Coend FX,Y)$; then it is clear that $\{f\circ\iota_X|X\In\Obj(\C)\}$ defines a dinatural family. Conversely, assume that we are given a dinatural family $\{f_X|X\In\Obj(\C)\}$ from $F(X,X)$ to $Y$. Then, due to the universal property of the coend (c.f. \eqref{coend_uni}), there exists a unique $f\In\Hom(\Coend FX,Y)$ such that $f_X=f\circ\iota_X$. In what follows we often define morphisms $f:\Coend FX\rightarrow Y$ by giving explicitly the composition $f\circ\iota_X$.
\end{remark}
\subject{The Lyubashenko coend}
Let now $\C$ be a $\k$-linear finite braided category with a duality. It is worth pointing out that the modular tensor categories are among these categories but the class is much larger.
Consider the functor
\be\label{Lfundef}
    \Lfun:=\oti\circ(?^\vee\oti\Id):\C\op\times\C\rightarrow \C
\ee
acting as
\be\label{Lfunact}
    X\times Y\mapsto X^\vee\oti Y\quand f\times g\mapsto f^\vee\oti g\,.
\ee
As shown in \cite{lyub8,maji25} (see also \cite{vire4}) the coend
\be\label{Lyub_coend}
    \L:=\int^X X^\vee\oti X
\ee
of $\Lfun$ exists and carries the structure of a Hopf algebra in $\C$.
In this case the defining property \eqref{def_dina} of the dinatural family is written graphically as
\eqpic{dinat_trafo} {115} {29} {
  \put(0,-3)   {\Includepichtft{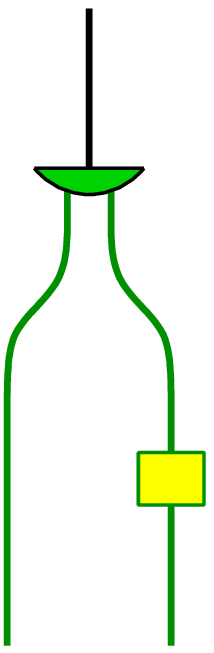}}
  \put(-5,-8)   {\sse$ Y^{\!\vee} $}
  \put(6.5,74)  {\sse$ B $}
  \put(15,-8)   {\sse$ X $}
  \put(16,38)   {\sse$ Y $}
  \put(23,16.8) {\sse$ f $}
  \put(52,31)   {$ = $}
  \put(84,0) {
  \put(0,-2){\Includepichtft{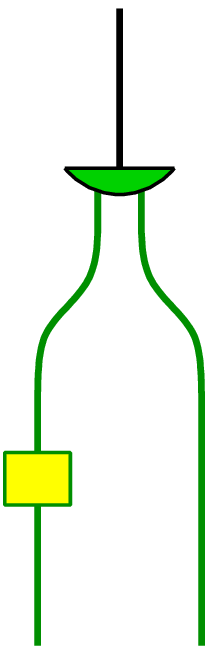}}
  \put(-2,-8)   {\sse$ Y^{\!\vee} $}
  \put(9.5,74)  {\sse$ B $}
  \put(18,-8)   {\sse$ X $}
  \put(-6,35)   {\sse$ X^{\!\vee} $}
  \put(-9,16.5){\sse$ f^{\!\vee} $}
  } }
The Hopf algebra structure is the following, c.f.\ Remark \ref{rem:morph_coend}:
  \Eqpic{Lyub_Hopf} {320} {90} {
  \put(10,10){
  \setulen80
   \put(0,113) {
  \put(0,0)   {\includepichtft{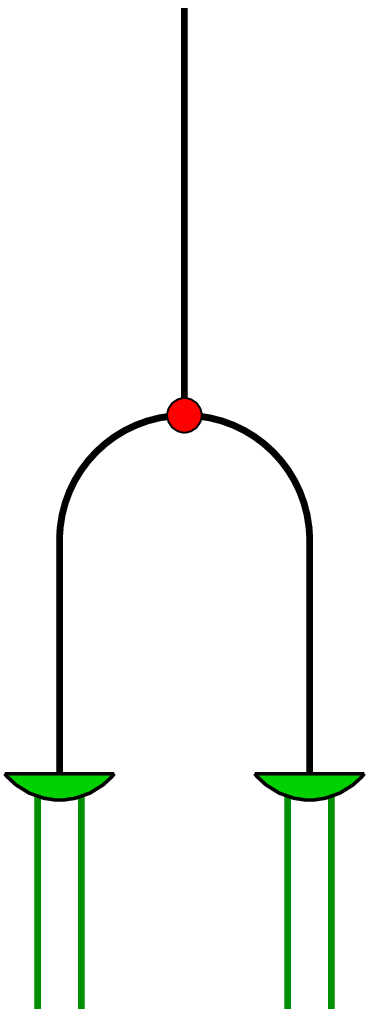}}
  \put(-4,-8)   {\scriptsize$ U^{\!\vee} $}
  \put(7,-8)    {\scriptsize$ U $}
  \put(-6.8,25) {\scriptsize$ \iota_{\!U}^{} $}
  \put(40.8,25) {\scriptsize$ \iota_{\!V}^{} $}
  \put(17.5,115){\scriptsize$ \L $}
  \put(22.6,68.3) {\scriptsize$ m_\L^{} $}
  \put(25,-8)   {\scriptsize$ V^{\!\vee} $}
  \put(36,-8)   {\scriptsize$ V $}
  \put(60,50)   {$ := $}
  \put(95,0) { {\includepichtft{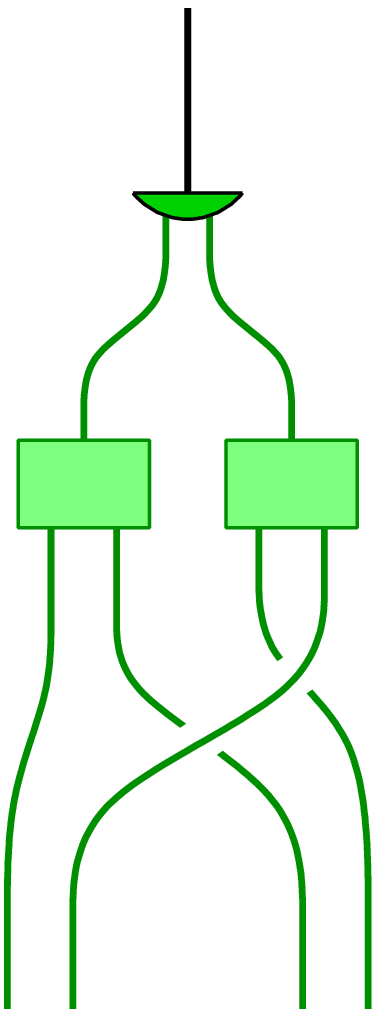}}
  \put(-6,-8)   {\scriptsize$ U^{\!\vee} $}
  \put(6,-8)    {\scriptsize$ U $}
  \put(-11,68)  {\scriptsize$ \gamma_{U,V}^{} $}
  \put(27.8,89) {\scriptsize$ \iota_{\!V\otimes U}^{} $}
  \put(33.8,66) {\scriptsize$ \id_{\!V\otimes U} $}
  \put(17.5,115){\scriptsize$ \L $}
  \put(27,-8)   {\scriptsize$ V^{\!\vee} $}
  \put(39,-8)   {\scriptsize$ V $}
   }
    \put(235,0) {
  \put(0,4) { {\includepichtft{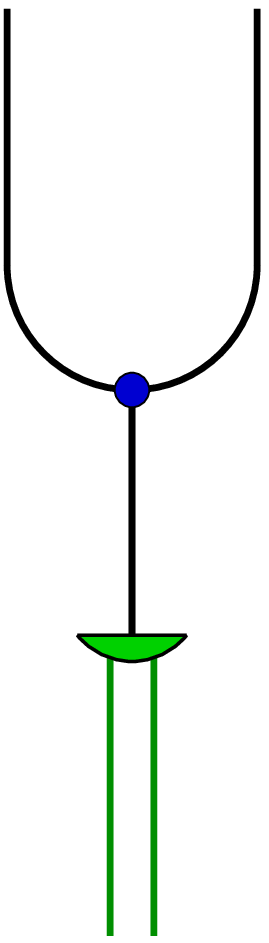}}
  \put(-2,106)  {\scriptsize$ \L $}
  \put(25.3,106){\scriptsize$ \L $}
  \put(6,-8)    {\scriptsize$ U^{\!\vee} $}
  \put(18,-8)   {\scriptsize$ U $}
  \put(15.4,54.2) {\scriptsize$ \Delta_\L $}
  \put(21,32) {\scriptsize$ \iota_{\!U}^{} $}
   }
  \put(55,50)   {$ := $}
  \put(85,0) { {\includepichtft{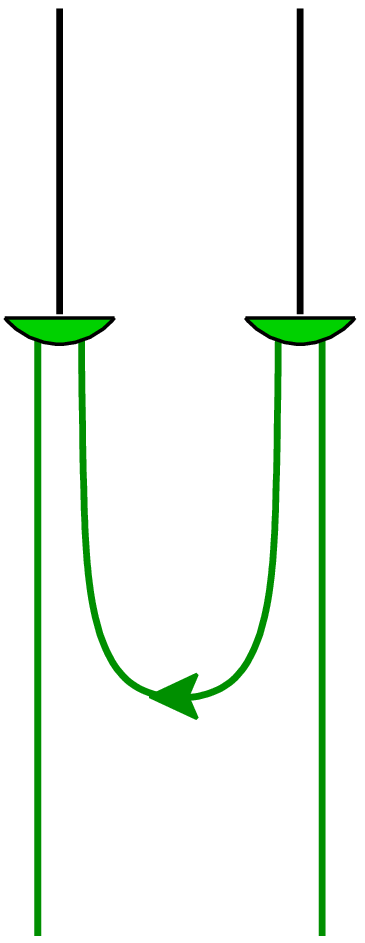}
  \put(3,106)  {\scriptsize$ \L $}
  \put(30,106)  {\scriptsize$ \L $}
  \put(0,-8)    {\scriptsize$ U^{\!\vee} $}
  \put(32,-8)   {\scriptsize$ U $}
  \put(-7.8,67) {\scriptsize$ \iota_{\!U}^{} $}
  \put(40.3,67) {\scriptsize$ \iota_{\!U}^{} $}}
   } }
   }
    \put(0,0) {
  \put(0,17) { {\includepichtft{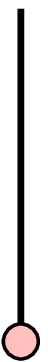}}
  \put(-.7,42) {\scriptsize$ \L $}
  \put(6,1)   {\scriptsize$ \eta_\L^{} $}
   }
  \put(25,30) {$ := $}
  \put(55,0) {\includepichtft{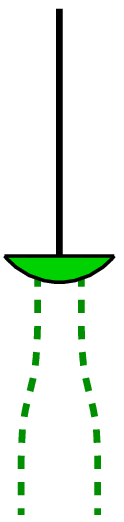}
  \put(3.3,59){\scriptsize$ \L $}
  \put(13.3,27) {\scriptsize$ \iota_{\!\one}^{} $}}
  \put(135,10) { {\includepichtft{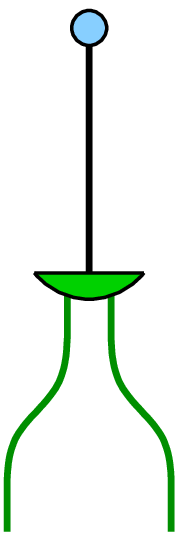}}
  \put(-3,-8) {\scriptsize$ U^{\!\vee} $}
  \put(16,-8) {\scriptsize$ U $}
  \put(13.4,54.5) {\scriptsize$ \eps_\L^{} $}
  \put(16.9,27.5) {\scriptsize$ \iota_{\!U}^{} $}
   }
  \put(174,30){$ := $}
  \put(203,0) {\includepichtft{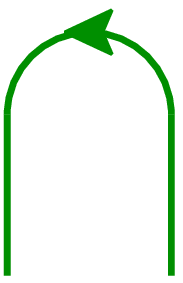}
  \put(-3,-8) {\scriptsize$ U^{\!\vee} $}
  \put(16,-8) {\scriptsize$ U $}}
   }
  \put(290,0) { {\includepichtft{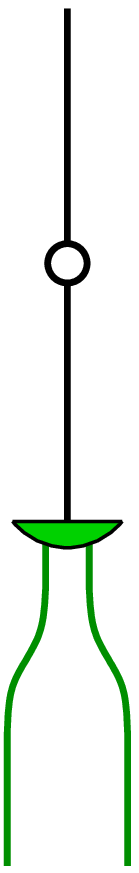}}
  \put(-3,-8) {\scriptsize$ U^{\!\vee} $}
  \put(11,-8) {\scriptsize$ U $}
  \put(11.4,65) {\scriptsize$ \apo_\L $}
  \put(4.6,99)  {\scriptsize$ \L $}
  \put(14.4,37) {\scriptsize$ \iota_{\!U}^{} $}
   }
  \put(327,44){$ := $}
  \put(358,0) {\includepichtft{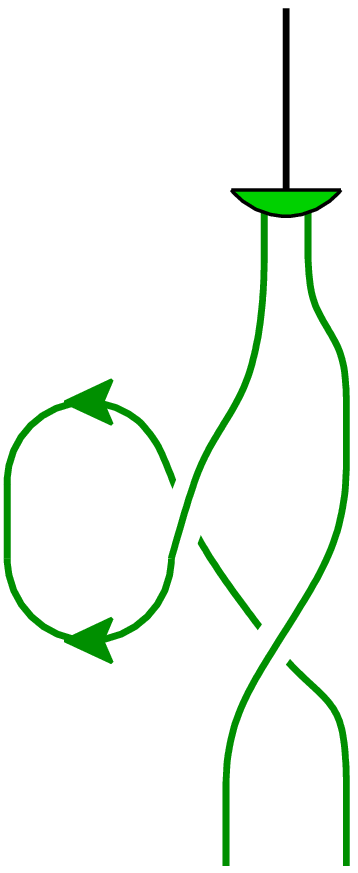}
  \put(21,-8) {\scriptsize$ U^{\!\vee} $}
  \put(35,-8) {\scriptsize$ U $}
  \put(28.4,99) {\scriptsize$ \L $}
  \put(38.4,73) {\scriptsize$ \iota_{\!U^{\!\vee}}^{} $}
  }
  } }
Here, $\gamma_{U,V}$ is the identification of $U^\vee\oti V^\vee$ with $(V\oti U)^\vee$.
The proof that \eqref{Lyub_Hopf} equips $\L$ with the structure of a bialgebra essentially amounts to some simple braid moves. The antipode property requires a little bit more work and involves in particular the dinaturality property of $\iota$. The calculation is given in appendix \ref{app:ribbon}, see eq. \eqref{coend_apopr}.

The Hopf algebra $\L$ can be equipped with a \emph{Hopf pairing} $\omega_\L$, defined via the dinatural family by \cite{lyub6}
\eqpic{coend_pair} {120} {31} {
  \put(0,0)   {\Includepichtft{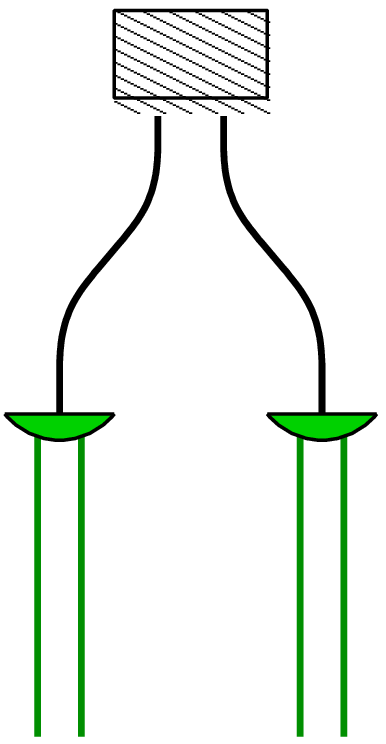}}
  \put(-4,-8)   {\sse$ U^{\!\vee} $}
  \put(7,-8)    {\sse$ U $}
  \put(26,-8)   {\sse$ V^{\!\vee} $}
  \put(37,-8)   {\sse$ V $}
  \put(16,74)   {\sse{\boldmath$ \omega_\L $}}
  \put(61,38)   {$ := $}
  \put(95,0) { {\Includepichtft{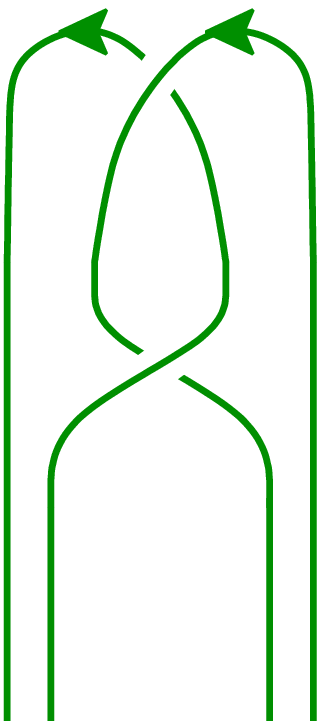}}
   } }

Below we will be interested in categories satisfying the following conditions:
\begin{enumerate}
\item[]  \label{cond:CF} \textbf{\CF:}\ \C\ is a an abelian $\k$-linear finite ribbon category, with $\k$ an algebraically closed field of characteristic zero.
\item[] \label{cond:CFN} \textbf{\CFN:}\ \C\ satisfies \CF\ and in addition the Hopf pairing \eqref{coend_pair} is non-degenerate. We will refer to a category satisfying this condition as a \emph{factorizable finite ribbon category\index{category!ribbon!factorizable finite}}.
\end{enumerate}
Note that \cite[Corollary 5.1.8]{KEly} if \C\ satisfies \CF, then the Hopf algebra $\L$ indeed exists as an object of \C\ and consequently condition \CFN\ makes sense.

For a factorizable finite ribbon category \C, the coend \L\ provides actions of mapping class groups. Denote by $\Surf gn$ a closed oriented surface of genus $g$ with $n$ marked points marked by objects $U_1,...U_n$ in \C\ and denote by \Mapgn\ the mapping class group of $\Surf gn$.
If \C\ satisfies condition \CFN\ there is \cite{lyub6,lyub11} a projective action $\piL$ of \Mapgn\ on the space $\Hom(\L^{\oti g},\Lyubspace)$, where the object $\Lyubspace$ is defined below in \eqref{Lyubspacedef}.
This action is described in detail in section \ref{sec:map_act}, see Proposition \ref{Lyubact_prop}.
\ssubject{Modules over $\L$}
Any object in \C\ can be endowed with the structure of an \L-module. The action of $\L$ on any object $V$ in \C\ is given by
\be\label{Lact}
    \Lact V\circ\iota_U:=(d_U\oti\id_V)\circ[\id_{U^\vee}\oti(c_{V,U}\circ c_{U,V})]\,.
\ee
That this defines an action follows directly from \eqref{Lyub_Hopf} and a sequence of braid moves:
\Eqpic{Qrep} {320} {50} {
  \put(10,10){\setulen80
  \put(0,0)  {\includepichtft{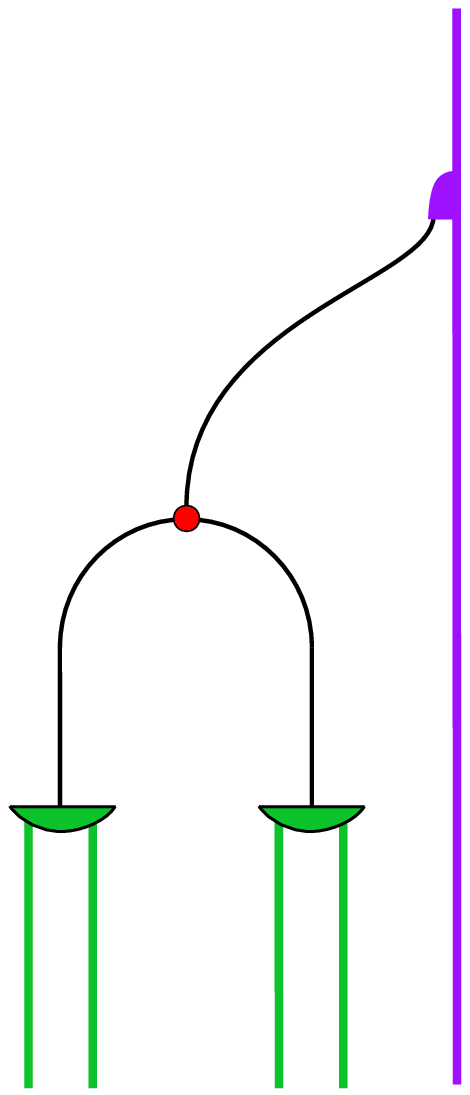}
  \put(-5,-8.5) {\scriptsize$X^{\!\vee}$}
  \put(7,-8.5)  {\scriptsize$X$}
  \put(24,-8.5) {\scriptsize$Y^{\!\vee}$}
  \put(35,-8.5) {\scriptsize$Y$}
  \put(46,-8.5) {\scriptsize$V$}
  \put(46.4,123)  {\scriptsize$V$}
  \put(22,64.5) {\scriptsize$m_\L$}
  \put(34,97)   {\scriptsize$\Lact V$}
  \put(-7.5,28.6) {\scriptsize$\iota_{\!X}^{}$}
  \put(21.5,28.6) {\scriptsize$\iota_{\!Y}^{}$}
  }
  \put(80,55)   {$ = $}
  \put(110,0)  {\includepichtft{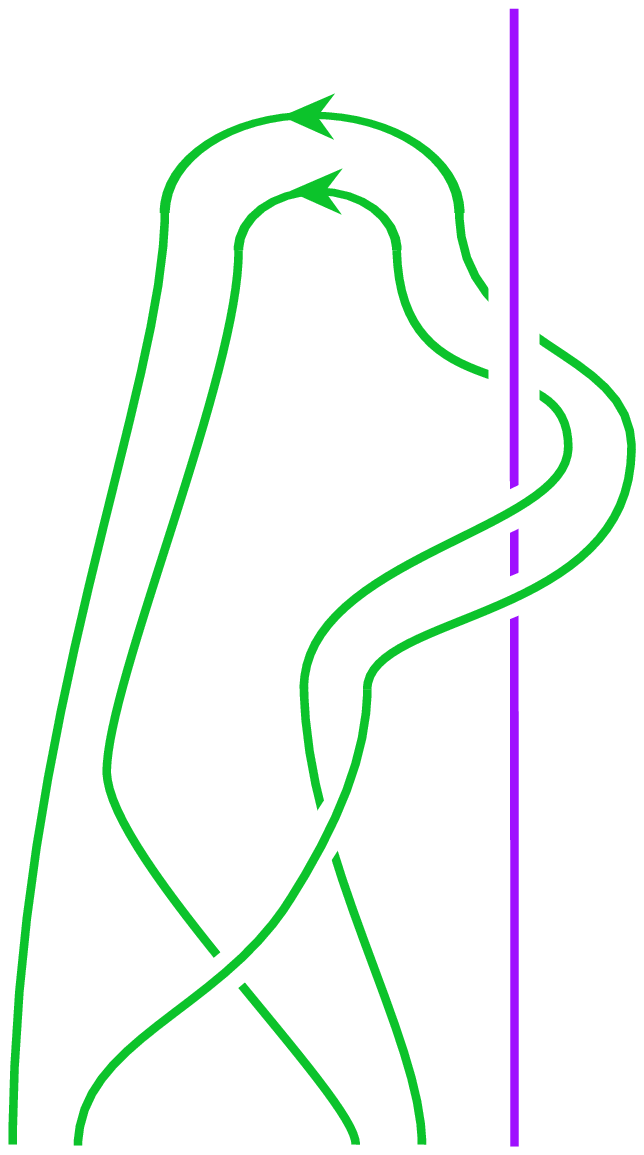}
  \put(-6,-8.5) {\scriptsize$X^{\!\vee}$}
  \put(6,-8.5)  {\scriptsize$X$}
  \put(33,-8.5) {\scriptsize$Y^{\!\vee}$}
  \put(44,-8.5) {\scriptsize$Y$}
  \put(53,-8.5) {\scriptsize$V$}
  \put(53,128)  {\scriptsize$V$}
  }
  \put(195,55)   {$ = $}
  \put(230,0)  {\includepichtft{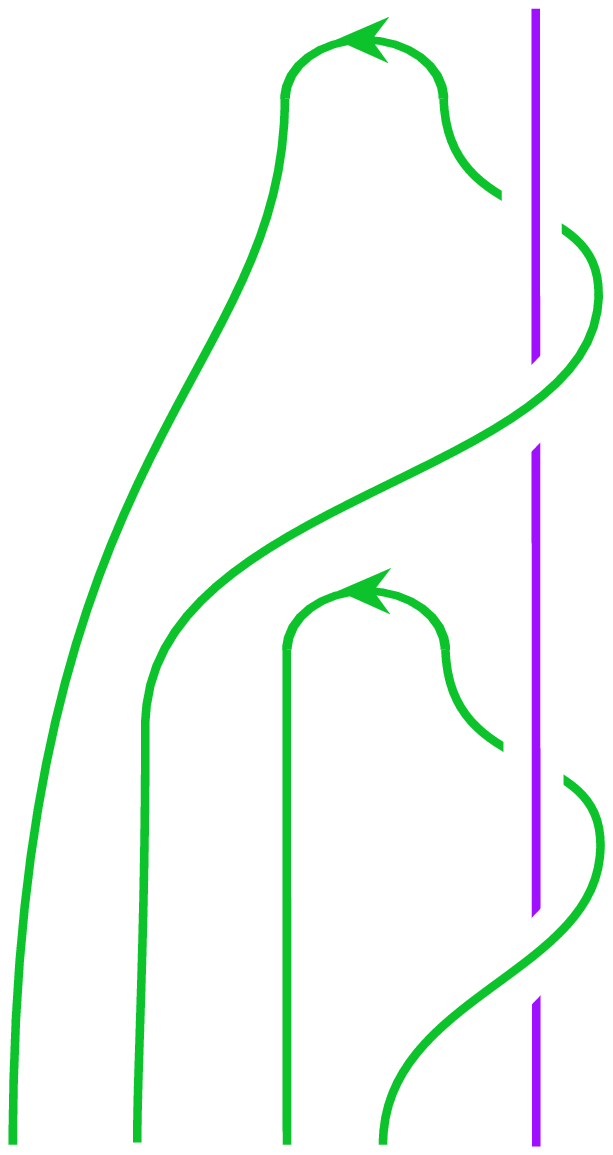}
  \put(-6,-8.5) {\scriptsize$X^{\!\vee}$}
  \put(10,-8.5) {\scriptsize$X$}
  \put(26,-8.5) {\scriptsize$Y^{\!\vee}$}
  \put(39,-8.5) {\scriptsize$Y$}
  \put(54.9,-8.5) {\scriptsize$V$}
  \put(55.5,128){\scriptsize$V$}
  }
  \put(315,55)   {$ = $}
  \put(350,0)  {\includepichtft{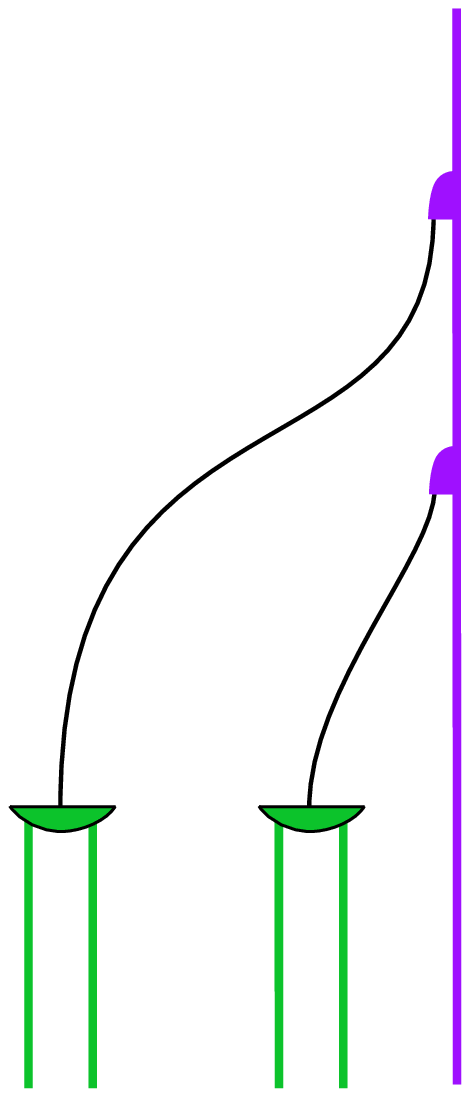}
  \put(-5,-8.5) {\scriptsize$X^{\!\vee}$}
  \put(7,-8.5)  {\scriptsize$X$}
  \put(24,-8.5) {\scriptsize$Y^{\!\vee}$}
  \put(35,-8.5) {\scriptsize$Y$}
  \put(46,-8.5) {\scriptsize$V$}
  \put(46.3,123){\scriptsize$V$}
  \put(34,68)   {\scriptsize$\Lact V$}
  \put(34,97)   {\scriptsize$\Lact V$}
  \put(12,30)   {\scriptsize$\iota_{\!X}^{}$}
  \put(40,30)   {\scriptsize$\iota_{\!Y}^{}$}
  } } }
\begin{remark}
    Any object $V$ in \C\ is also equipped with a right coaction $\roh^\L_V$ of \L\ defined by
    \be
        \roh=(\id_{V}\oti\iota_V)\circ (b_V\oti\id_V)\in\Hom(V,V\oti\L)\,.
    \ee
That this defines a coaction follows directly from the definition \eqref{Lyub_Hopf} of the coproduct of \L. Furthermore, it follows from the dinatural property of $\iota_V$ that the left module structure and right comodule structure on $V$ fit together to a Yetter-Drinfeld module over \L.
\end{remark}
\ssubject{$\L$ for modular \C} When \C\ is modular, then $\L$ is as an object
\be\label{L_ssi}
    \Lss=\bigoplus_{i\In\I}U_i^\vee\oti U_i\,,
\ee
see \cite[Section 3.2]{vire4} or \cite[Lemma 2]{kerl5}.
Denote by $\Emb \Lss i$ and $\Res\Lss i$ the embedding and restriction morphism of $U_i^\vee\oti U_i$ into \Lss.
The dinatural family, corresponding to \eqref{L_ssi}, is given by
\be
    \iota_U=\sum_{i,\alpha}\Emb\Lss i\circ\big[(\emb Ui\alpha)^\vee\oti\res Ui\alpha\big]\,,
\ee
where $\emb Ui\alpha$ and $\res Ui\alpha$ are the embedding and restriction morphisms in \eqref{emb_res_prop}.
That this defines  a dinatural family can be seen by decomposing any morphism into simple summands on both sides of \eqref{dinat_trafo}. In addition, given any dinatural transformation $(\psi_U,B)$ the unique morphism in $\Hom(\Lss,B)$ in \eqref{coend_uni}, required to turn \Lss\ into a coend is given by $\sum_{i\In\I}\psi_{U_i}\circ\res\L i\nl$.
\begin{remark}
As any algebra in a modular tensor category, the algebra $\L$ turns the space $\Hom(\one,\Lss)$ into an associative algebra over $\k$, see appendix \ref{app:alg_Hom}. Taking a basis $\{\tilde b_{U_i}|i\in\I\}$ for $\Hom(\one,\Lss)$, consisting of the left coevaluation morphisms of the simple objects, the structure constants of the algebra $\Hom(\one,\Lss)$ are nothing but the fusion rules, i.e.\ $\Hom(\one,\Lss)$ is the Verlinde algebra.
\end{remark}

\chapter{Algebras in conformal field theory}\label{chap:alg}
The purpose of this chapter is to give an overview over several classes of algebras that arise naturally in the study of CFT. For details we will refer to the literature or to later chapters in this text. Recall Definition \ref{def_algebra} of an algebra in a monoidal category. Considering e.g.\ an algebra in \Vect, we obtain an ordinary algebra over $\CN$, i.e.\ a complex vector space equipped with an associative product. The first "algebra" we consider is not an algebra in this ordinary sense.
\section{Vertex operator algebras}\label{subsec_VOA}
In this section we give a short description of vertex operators algebras, often referred to simply as vertex algebras. Since the rest of this thesis will not involve vertex algebras explicitly the discussion in this section will be very brief, for more details see e.g. \cite{BF,KAc4,LEli}.
Vertex algebras are formulated in the language of \emph{formal distributions}. A formal distribution is an expression of the form $\sum_{n\In\Z} a_n\,z^n$, where the sum over $n$ extends to infinitely large positive and negative powers and $z$ is a formal variable, not to be confused with a complex coordinate. We will be interested in the case that the coefficients $a_n\In U$ for some algebra $U$ over $\CN$.
The space of such formal distributions is denoted by $U\llbracket z^\pm\rrbracket$.
Calculus of formal distributions requires some caution, e.g.\ when multiplying formal distributions.
Two formal distributions in some formal variable cannot in general be multiplied to form a single formal distribution in the same formal variable.
However, two formal distributions in different formal variables, say $z$ and $w$ can be multiplied to form a formal distribution in $U\llbracket z^\pm,w^\pm\rrbracket$, see e.g.\ \cite[Chapter 1.1]{BF} for details.

\begin{definition}\label{def_VA}
A vertex algebra\index{vertex algebra|textbf} consists of the following data:
\begin{itemize}\addtolength{\itemsep}{-6pt}
    \item The \emph{space of states}: a $\Z_{\geq0}$-graded vector space $\VA=\bigoplus_{n\in\Z_{\geq0}}\VA_n$ with \\$\dim(\VA_n<\infty)$
    \item A \emph{vacuum vector} $|0\rangle\In\VA_0$
    \item A \emph{translation operator}: $T:\VA\rightarrow\VA$
    \item A \emph{vertex operation}: a linear map
    \be
        Y(\cdot,z):\VA\rightarrow \End(\VA)\llbracket z^\pm\rrbracket\,,
    \ee
    taking each $a\In \VA$ to a field $Y(a,z)=\sum_{n\In\Z}a_{(n)}z^{-n-1}$. Here $Y(a,z)$ is a formal distribution with coefficients $a_{(n)}\In\End(\VA)$.
\end{itemize}
The data satisfy the following axioms:
\begin{itemize}\addtolength{\itemsep}{-6pt}
    \item \emph{Translation axiom:} $[T,Y(a,z)]=\partial Y(a,z)$ for any $a\In\V$ and $T|0\rangle=|0\rangle$.
    \item \emph{Vacuum axiom:} $Y(|0\rangle,z)=\id_{\VA}$ and, for any $a\In\VA$, $Y(a,z)|0\rangle|_{z=0}=a$.
    \item \emph{Locality:} $(z-w)^N\,Y(a,z)\,Y(b,w)\eq (z-w)^N\,Y(b,w)\,Y(a,z)$ for some \linebreak$N\gg0$.
\end{itemize}
\end{definition}
A field $Y(\nu,z)$ is called a \emph{Virasoro field} if the modes of $Y(\nu,z)$ furnish a representation of the Virasoro algebra. In CFT we are typically interested in conformal vertex algebras defined as follows:
\begin{definition}
A \emph{conformal} vertex algebra is a vertex algebra with a distinguished vector $\nu\In\VA_2$ such that $Y(\nu,z)=\sum_{n\In\Z}L_n\,z^{-n-2}$ is a Virasoro field, $L_{-1}=T$ equals the infinitesimal translation operator and $L_0$ is diagonalizable on $\VA$.
\end{definition}
The notion of a vertex algebra can be regarded as a formalization of the Wightman axioms and the chiral operator product expansion. There is e.g.\ a unique vacuum, invariant under an action of the symmetry group of space-time (in the CFT-case this involves the conformal transformations)
and
a notion of completeness in the sense that acting with all fields on the vacuum generates the entire state space. In addition, using multiplication of formal distributions we can write down the OPE of two chiral fields $Y(a,z)$ and $Y(b,w)$.
Note that the variable $z$ above is a formal variable and should not be thought of as a complex number. However, in CFT, the formal variable $z$ is identified with the complex local coordinate of the world sheet in the end of the analysis.

We will not at all go into representation theory of vertex algebras, see e.g.\ \cite[Definition 5.1.1.]{BF} for the definition of a module over a vertex algebra. The reason is that we will not need any detailed information about the representation theory, only the properties of the representation category, for the considerations in this thesis. Let us just mention that, as many algebraic structures, $\VA$ is a module over itself. In the representation category, $\Rep(\VA)$, \VA\ is in fact the tensor unit. This nicely illustrates the power of working in categories. The involved structure of \VA\ is something very simple in the category.

In addition we mention that there is a notion of a character $\chii_M$ of a module $M$.
What is relevant for us is that for rational CFT the characters of the simple modules provide a basis of the space of conformal blocks on a torus with modular parameter $\tau$ and without holes. In simple situations, the character of a simple module $M$ can be written as $\Tr_M q^{L_0^M-c/24}$ where $q$ is a formal parameter.
A basis for the space of conformal blocks on a torus is in such situations obtained by taking the characters of the simple modules and identifying $q\,{\equiv}\,\e^{2\pi\i\tau}$.

Remember that when studying full CFT, our starting point will be an abstract category with the properties of the representation category of the vertex algebra of the corresponding class of CFTs. In particular, working in an abstract category allow us to treat an entire class of categories simultaneously.

Sufficient conditions for \C\ to be modular have been given in \cite[Theorem 3.1]{huan21}. Since $\Rep(\VA)$ is in general not strict, our starting point will be the equivalent strict category $\C$ -- the strictification of $\Rep(\VA)$.
In fact we define the notion of rational CFT by saying that a rational CFT is a CFT such that the strictification of $\Rep(\V)$ is modular together with an additional requirement concerning the relation to 3-d TFT, see section \ref{Cbl_TFTspace}.

\section{Frobenius algebras}
Recall from Remark \ref{rem_symm} (iii) that an algebra $A$ is a Frobenius algebra iff it can be equipped with a non-degenerate invariant pairing, or equivalently, iff $A$ is also a coalgebra that satisfies the compatibility condition \eqref{frob}.
There are several types of Frobenius algebras which are relevant when studying full CFT.
We will discuss Frobenius algebras in representation categories of vertex algebras. These algebras arise as state spaces of full CFT. We will also discuss interesting examples of Frobenius algebras in \Vect.
\subsection{The \boundA}\label{boundFA}
Starting from the pair $(\C,\tftA)$, where $\C$ is a modular tensor category and $\tftA$ is a symmetric special Frobenius algebra in \C, one can construct a full (rational) CFT with the help of the TFT-construction, which is described in section \ref{sec:TFT}.
In fact, a full rational CFT is uniquely determined by a Morita class of Frobenius algebras in a modular tensor category $\C$ \cite{fjfrs2}. As a consequence, the label of a phase $A$ of a rational CFT can be taken to be a special symmetric Frobenius algebra in \C. With the pair $(\C,\tftA)$ as input, the TFT-construction gives a prescription for how to calculate any correlator of a rational CFT. More generally, if there are several different phases on a world sheet we need a tuple $(\C,A_1,...,A_n)$ where $A_i$, $i=1,...,n$, label the phases in various regions of the world sheet. Here we will consider a single phase and describe how the algebra $\tftA$, which we call the \emph{\boundA\index{\boundA}}, can be constructed from one of the boundary conditions. The description summarizes the one given in \cite[Section 3.2]{fuSs4}.

Consider a rational CFT with some maximally symmetric boundary condition\index{boundary condition} $M_\circ$. As an object in \C, the \boundA\ $\tftA$ can be decomposed as
\be\label{boundA_obj}
    \tftA\cong\bigoplus_{a\In\I}U_a^{\oplus{n_a}}\,,
\ee
with $n_a\In\Z_{\geq0}$ and $n_0\eq 1$.
In field theoretic terms, i.e.\ when considering $\C$ concretely as (the strictification of) a representation category of a vertex algebra, the \V-module $\tftA$ is the space of states living on the boundary segment $M_\circ$, see \cite[eqs.\ (3.9) \& (3.13)]{fuRs4}. In fact, $M_\circ$ is isomorphic to $A$ as an $A$-module. Let us also mention the  boundary partition function, i.e.\ the correlator $ A_M^{\;\;N}$ of the annulus with no field insertions.
The double of the annulus is a torus and accordingly the boundary partition function, $\mathrm A_M^{\;\;N}$, is an element in the space $\cbl(T)$ of conformal blocks on the torus. As mentioned in the previous section a basis for the space $\cbl(T)$ is given by the characters of the irreducible \V-modules.
Thus, the boundary partition function $ \mathrm A_M^{\;\;N}$ can be expanded in terms of the characters:
\be\label{Ann_exp}
    \mathrm A_M^{\;\;N}=\sum_{a\in\I}\, \mathrm A_{kM}^{\;\;\;\;N}\,\chii_k\,.
\ee
In particular, the multiplicities in \eqref{boundA_obj} are given by $n^a\eq \mathrm A_{aA}^{\;\;\;\;A}$. A basis of the space of conformal blocks in the TFT-framework is given in \eqref{char_inv}.

As mentioned in the introduction, there is a bijection between the simple summands in \eqref{boundA_obj} and the primary boundary fields. Now take, for each space $\Hom(U_a,\tftA)$, a basis of the type in \eqref{emb_res_prop}.
In terms of this basis, the OPE coefficients involving only primary boundary fields can be used to equip the object $\tftA$ in \eqref{boundA_obj} with the structure of an algebra in $\C$, see \cite[eq. 3.14]{fuRs4}. It is a direct consequence of the sewing constraint \cite{lewe3} for four boundary fields on a disc that the product is associative, see \cite[Section 3.2]{fuRs4}. The unit of $A$ is the embedding of the tensor unit $U_0$ into $A$.

The correlator of two boundary fields on a disc is required to be non-\linebreak degenerate in the sense that for every boundary field $\Psi$ there is at least one field $\Psi'$ such that the correlator of these two fields is non-zero. This is so because if a field would have zero correlator with any other field it would decouple completely from the theory: \emph{every} correlator involving that field would vanish. Non-degeneracy of the boundary two-point functions gives rise to a non-degenerate invariant bilinear form on $A$ \cite[eq. 3.16]{fuRs4}. This is equivalent to $A$ being Frobenius in the sense of Definition \ref{def_Frobalgebra} \cite[Proposition 3.1]{fuSt}, see also Remark \ref{rem_symm} (iii).
Thus, the Frobenius property is a consequence of the non-degeneracy of two-point boundary correlators.

That $\tftA$ is also special and symmetric can be derived from the sewing constraint for one boundary field on an annulus. This is described in section 3.2 and Theorem 3.6 of \cite{fuRs4}.
In addition, it is shown in \cite[Section 4.4]{fuRs4} that the maximally symmetric boundary conditions can be labeled by isomorphism classes of left-modules over $A$. In particular, the elementary boundary conditions\index{boundary condition} are in bijection with simple $A$-modules. In addition elementary topological defect lines, separating two regions labeled by $A$ and $B$, are labeled by isomorphism classes of simple $A$-$B$-bimodules

The derivation of \tftA\ took as an input just one of the maximally symmetric boundary conditions. It is natural to ask what would have happened if we would have started from another boundary condition. Doing so gives rise to a symmetric special Frobenius algebra $A'$ that is \emph{Morita equivalent} to $A$, see \cite[Section 4.1]{fuRs4}.
$A$ and $A'$ are Morita equivalent\index{Morita equivalence} iff the categories $\C_A$ and $\C_{A'}$ are equivalent as abelian categories. For a precise description of the equivalence one uses a pair of $A$-$A'$- and $A'$-$A$-bimodules, see e.g.\ \cite[Theorem 5.1]{pare13}.

It follows that if $A$ and $A'$ are Morita equivalent, then there is a bijection between isomorphism classes of left-modules over $A$ and $A'$.
As a consequence, Morita equivalent \boundA s give rise to physically equivalent rational CFTs. Conversely, it has been shown \cite{fjfrs2} that any solution to the sewing constraints arises as a Morita class of symmetric special Frobenius algebras.
We will also explain below that the set of maximally symmetric boundary conditions for a given bulk theory is exhausted by those associated with a single Morita class of special symmetric Frobenius algebras in \C.
Thus, indeed equivalence classes of full rational CFTs with given chiral algebra \VA\ are classified by Morita classes of symmetric special Frobenius algebras in the strictification \C\ of $Rep(\VA)$.

\subject{Some classes of \boundA s}
Finally, let us list a couple of examples of \boundA s:
\begin{itemize}\addtolength{\itemsep}{-6pt}
\item The \emph{tensor unit} in any rigid monoidal category. This is clearly a symmetric special Frobenius algebra. The full CFT built from this \boundA\ is known as the "Cardy case".
\item A \emph{Schellekens algebra} is an algebra that is a direct sum of invertible objects. An invertible object $J$ has an inverse $J'$ such that $J\oti J'\cong\one$. The isomorphism classes of invertible objects form a group, known as the Picard group of \C. There is a Schellekens algebra, with the properties of a \boundA, for every subgroup of the Picard group satisfying a certain simple cohomological property, see e.g.\  \cite[Section 5]{ffrs4}.
\item Considering modular tensor categories based on $\mathfrak{sl}(2)$ at integral level, the
\emph{ADE-classification} of modular invariant partition functions can be formulated in terms of \boundA s in \C, see \cite[table 1]{kios}.
\end{itemize}
\subsection{The \bulkA}\label{sec_bulkA}
\index{\bulkA|textbf}
Just like the space of states on the boundary in a rational CFT carries the structure of an algebra, the space of bulk states in a rational CFT can be equipped with the structure of a symmetric special commutative Frobenius algebra. Bulk states are elements in some $\VA\otiC\VA$-module. Thus, in categorical terms, the bulk Frobenius algebra \BA\ is an element in $\Ce$:
\be\label{BA_ssi}
    \BA=\bigoplus_{p,q\In\I}(U_p\btimes U_q)^{\oplus Z_{p,q}}\,.
\ee
The bulk partition function is expanded in terms of the $\V$-characters according to
\be\label{PF_mat}
    Z=\sum_{p,q\In\I}\,Z_{p,q}\,\chii_p\oti\chii_q^*\,.
\ee

We will first consider the case $A\eq\one$, i.e.\ the Cardy case. This result may serve as a building block for the \bulkA s of a general rational CFT. In the Cardy case the coefficients of the partition function are
\be\label{PF_Cardy}
    Z_{p,q}=\delta_{p,\qb}\,,
\ee
i.e.\ the two representations in each summand of \eqref{BA_ssi} are related via charge conjugation and there is exactly one term for each simple object of \C. Accordingly, the partition function \eqref{PF_mat}, with coefficients given by \eqref{PF_Cardy}, is known as the \emph{charge conjugation} partition function. Thus, the \bulkA, which we in the Cardy case denote by $\BFss$, is
\be\label{CC_BA}
    \BFss~=~\bigoplus_{p\In\I} U_p\btimes U_\pb\quad\In\Ce\,.
\ee
It can be shown \cite[Section 6.3]{ffrs} that the object $\BFss$ can be endowed with a natural structure of a special symmetric commutative Frobenius algebra. In \cite{ffrs}, \BFss\ is called the trivializing algebra. Here we will refer to \BFss\ as the \emph{Cardy case bulk algebra}.

\subject{The algebra structure on $\BA$}
Just like in the case of the \boundA, the fact that the bulk state space can be equipped with the structure of a Frobenius algebra can be deduced from physical arguments:
Using the OPE-coefficients one defines an algebra structure on the space of bulk states, with a unit obtained from the identity field. Requiring that the two-point function on the sphere is non-degenerate then implies that this algebra is Frobenius.

The statement that the \bulkA\ $\BA$, defined in \eqref{BA_ssi}, is a symmetric commutative Frobenius algebra has also been proved rigorously in the TFT-framework:
The \bulkA\ is isomorphic to the \emph{full center} $Z(A)$ of the \boundA\ $A$. We refer the reader to e.g. \cite[Definition 3.14]{rgaW} for a definition of the full center. For a modular category, the full center $Z(A)$ can be written in terms of the \boundA\ and \BFss\ (see e.g. \cite[Remark 8.2]{davy20}):
\be
     Z(A)~=~C_l((A\btimes\one)\Otip\BFss)\quad\In\,\Ce\,.
\ee
Let us explain this equation:
First, $A$ is the \boundA\ and
$(A\btimes\one)\Otip\BFss$ is the tensor product of Frobenius algebras as defined in \eqref{frob_tens}. For any algebra $B$, $C_l(B)$ denotes its \emph{left center}. The left center  $C_l(B)$ of an algebra $B$ in a sovereign category can be defined as the kernel
\be
    C_l(B):=\ker\big[(((m\circ{c_{B,B}})\oti \id_{B^\vee})\circ(\id_B\oti b_B))-(m\oti \id_{B^\vee})\circ(\id_B\oti b_B)\big]\,.
\ee
Thus, in a  modular category, the left center is the maximal subobject of $B$ such that
\eqpic{Cl_ssi}{105}{34}{
\put(0,0)   {\setulen80
\put(-8,0) {\includepic{302}{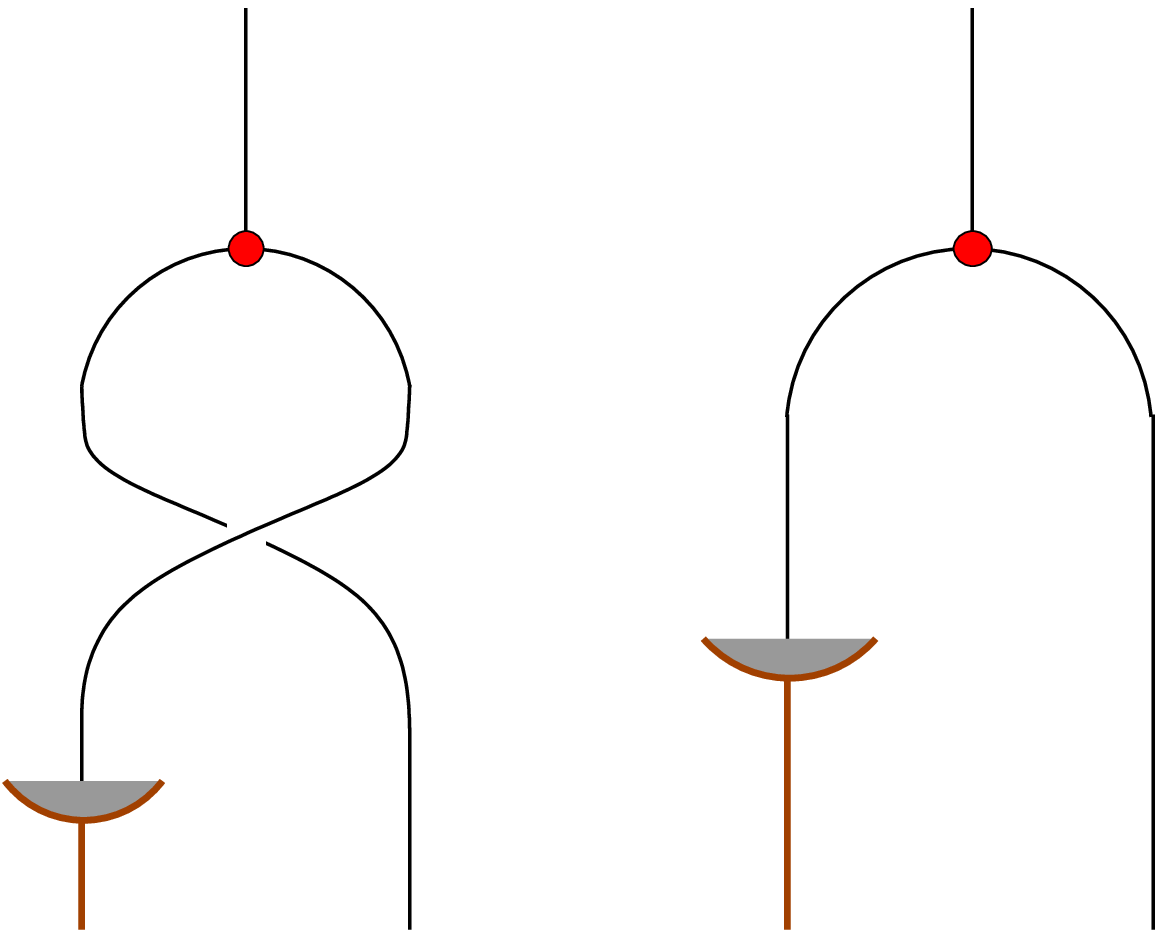}}
\put(-9.5,-8.7)  {\sse$C_l(B)$}
  \put(15.5,105.5) {\sse$B$}
  \put(33.2,-8.7)  {\sse$B$}
  \put(53,37.7)    {$=$}
  \put(69.5,-8.7)  {\sse$C_l(B)$}
  \put(95.5,105.5) {\sse$B$}
  \put(114.8,-8.7) {\sse$B$}
}
}
Here, the half-discs represent the embedding morphism of $C_l(B)$ into $B$. Note that indeed $Z(\one)=\BFss$ since $\BFss$ is commutative.
The left center of an algebra $B$ can be described as the image
\be
C^l(B)=\im(P^l_{B})
\ee
of the  projector
\eqpic{Cl_idem}{106}{42}{
\put(30,2)   {\setulen80
\includepic{302}{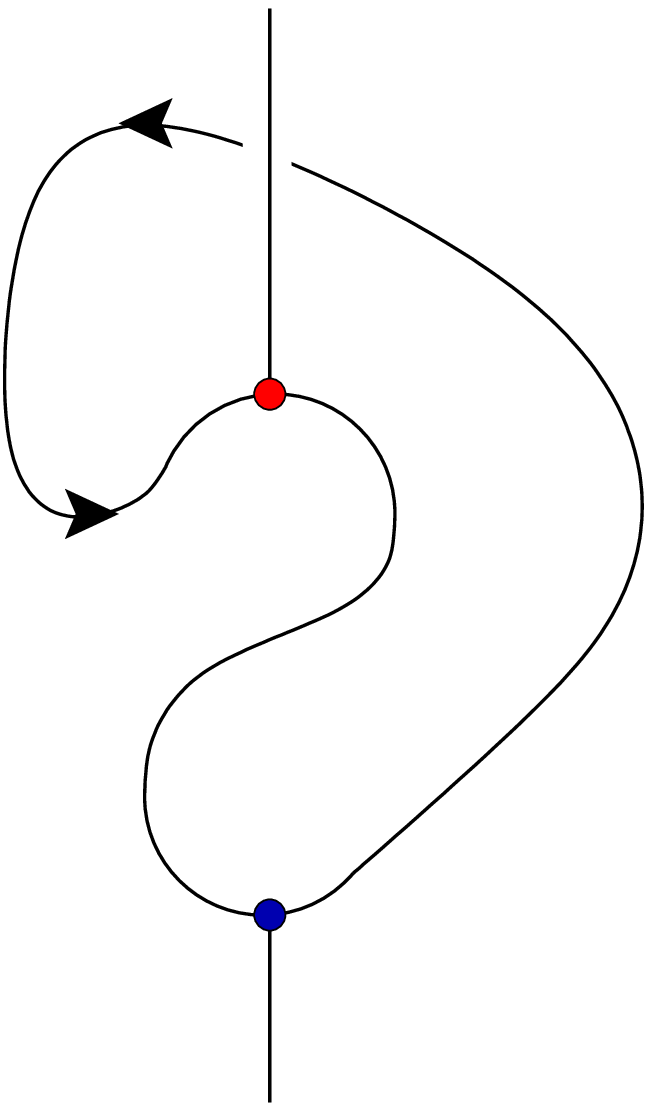}
\put(-48.6,51) {$P^l_{B}\;:=$}
  \put(8.2,30)     {\sse$B$}
  \put(25.5,-8.8)  {\sse$B$}
  \put(26.1,125)   {\sse$B$}
  \put(73.2,67)    {\sse$B$}
}
}

It can be shown \cite[Appendix A]{rffs} that the \emph{\bulkA} \eqref{BA_ssi} is isomorphic as an object to the \emph{full center}:
\be\label{BS_FC}
    \BA\cong Z(A)\,.
\ee

That the \boundA\ $A$ is a symmetric special Frobenius algebra in \C\ implies that $A\btimes\one$ is a symmetric special Frobenius algebra in \Ce, just like the Cardy case bulk algebra $\BFss$. Thus, the object $(A\btimes\one)\oti\BFss\linebreak \In\Ce$ can be equipped with the structure of a symmetric special Frobenius algebra in \Ce, via the structure morphisms $m^+,\eta^+,\delta^+$ and $\eps^+$ in \eqref{frob_tens}.
By \cite[Proposition 2.37]{ffrs}, taking the left center of such an algebra we obtain a commutative symmetric Frobenius algebra. Thus, $Z(A)$, and accordingly also \BA, is indeed a commutative symmetric Frobenius algebra.

Given a modular tensor category \C, a compatible set of boundary conditions, which we describe as a Morita class of \boundA s, determines a consistent bulk state space by the TFT-construction.
A natural question to ask is whether there could be several different sets of boundary conditions consistent with one and the same bulk state space. The answer is no. For two \boundA s $A$ and $B$, the \bulkA s $Z(A)$ and $Z(B)$ are isomorphic iff $A$ and $B$ are Morita equivalent \cite{koRu}. Thus, in particular a single Morita class of symmetric special Frobenius algebras exhausts all possible boundary conditions for a given bulk theory.
\subject{The Cardy case bulk algebra as a coend}
The description above of the Cardy case bulk algebra \BFss\ is tailored for a semisimple category. In order to be able to describe more general situations we describe \BFss\ in terms of a coend. To this end, consider the functor\footnote{Note that here we have switched from $\Ce$ to $\eC$. This is in order to follow some standard conventions for coends, c.f.\ section \ref{coends}. In order to use \Ce\ we would have to replace $\C\op\times\C$ by $\C\times\C\op$ when discussing coends.}
\be\label{BA_fun}
    ?^\vee\boxtimes \Id:\,\C\op\times\C\rightarrow\eC\,,
\ee
acting as $X\times Y\mapsto X^\vee\boxti Y$ on objects and as $f\times g\mapsto f^\vee\boxti g$ on morphisms. The coend
\be\label{BulkA_coend}
    \BF:=\int^X X^\vee\boxtimes X\quad\In\,\eC
\ee
of the functor \eqref{BA_fun} has been considered in \cite{KEly}. It follows from \cite[Corollary 5.1.8]{KEly} that if \C\ is finite $\k$-linear, the coend \BF\ exists as an object in \eC.

If $\C$ is in addition braided, then $\BF$ can be equipped with the structure of a unital algebra, in a way similar to the way \L\ is equipped with an algebra structure in \eqref{Lyub_Hopf}.  Concretely, the product of \BF\ is defined by
\be\label{prod_BF}
    m_\BF\circ(\iota_U\oti\iota_V) ~:=~ \iota_{V\!\oti\! U}\circ(\id_{U^\vee\!\oti\! V^\vee}\boxti c_{U,V})\,,
\ee
and the unit is given by $\eta_\BF=\iota_\one$. Thus, in particular for a factorizable finite ribbon category $\C$, i.e.\ a category that satisfies condition \CFN\ described on page \pageref{cond:CFN}, the coend \BF\ can be equipped with the structure of an algebra in \eC.
A pictorial description for the product of \BF\ looks the same as the picture for the product of \L\ in \eqref{Lyub_Hopf}. However one has to be careful to take into account that in such a picture there is now a "two-layered" structure corresponding to the two factors in $\eC$. This means that even though the pictorial description of the product looks exactly like the picture for the product in \eqref{Lyub_Hopf}, the crossing between the $U$- and $V^\vee$-lines in the following picture is not a braiding:
 \eqpic{m_F_pic} {130} {40} {
  \setulen80
   \put(0,0) {
  \put(0,0)   {\includepichtft{01a}}
  \put(-4,-8)   {\scriptsize$ U^{\!\vee} $}
  \put(7,-8)    {\scriptsize$ U $}
  \put(-6.8,25) {\scriptsize$ \iota_{\!U}^{} $}
  \put(40.8,25) {\scriptsize$ \iota_{\!V}^{} $}
  \put(17.5,115){\scriptsize$ \BF $}
  \put(22.6,68.3) {\scriptsize$ m_\BF^{} $}
  \put(25,-8)   {\scriptsize$ V^{\!\vee} $}
  \put(36,-8)   {\scriptsize$ V $}
  \put(60,50)   {$ := $}
  \put(95,0) { {\includepichtft{01b}}
  \put(-6,-8)   {\scriptsize$ U^{\!\vee} $}
  \put(6,-8)    {\scriptsize$ U $}
  \put(-11,68)  {\scriptsize$ \gamma_{U,V}^{} $}
  \put(27.8,89) {\scriptsize$ \iota_{\!V\otimes U}^{} $}
  \put(33.8,66) {\scriptsize$ \id_{\!V\otimes U} $}
  \put(17.5,115){\scriptsize$ \BF $}
  \put(27,-8)   {\scriptsize$ V^{\!\vee} $}
  \put(39,-8)   {\scriptsize$ V $}
  \put(30,31) {\scriptsize$ c $}
   }
  } }

The coend \BF, defined in, \eqref{BulkA_coend} indeed reproduces the algebra $\BFss$ in \eqref{CC_BA} in the case that \C\ is modular. In particular, the definition \eqref{prod_BF} of the product reproduces the one in \cite[eq. (6.48)]{ffrs}.
Whenever \C\ is such that
\BF\ can be equipped with the structure of a symmetric commutative Frobenius algebra we will refer to \BF\ as the \emph{\bulkAC}\index{\bulkA! Cardy}.
%

So far it is not known how to obtain a Frobenius coproduct for \BF\ in the case of a general finite $\k$-linear ribbon category. However, we will see in section \ref{Frob_Hmod} that taking \C\ to be the category $H$\Mod\ of representations of a finite-dimensional factorizable ribbon Hopf algebra over an algebraically closed field $\k$ of characteristic zero, we can use the integral and cointegral of $H$ to construct a coalgebra structure on \BF\ that turns it into a Frobenius algebra.
To this end we work in the category $H$\Bimod, which is braided equivalent to $(\HMod)\rev\btimes\HMod$.
In fact, the so-obtained Frobenius algebra $\BFH\,{\equiv}\,\BF$ is in addition symmetric and commutative.
Note that we do not assume $H$ to be semisimple, and accordingly $H$\Mod\ is not necessarily semisimple. We will show that $\BFH$ is special iff $H$ is semisimple.
More generally, we obtain a family $\{\BFw\}$ of symmetric commutative Frobenius algebras in $H$\Bimod. They provide candidates for \bulkA s\index{\bulkA} of full conformal field theories whose chiral data are described by $H$\Mod.
We show in section \ref{sec:twist_bulk} that we can associate such a \bulkA\ \BFw\ to any ribbon Hopf algebra automorphism \ra\ of $H$.

\subsection{Classifying algebras}\label{CD_intro}
We have argued that in rational CFT the boundary conditions and defect lines have a simple description in terms of the representation theory of \boundA s. Boundary conditions of a CFT in a phase labeled by the \boundA\ $A$ are characterized by left $A$-modules, and defect lines separating two phases $A$ and $B$ are characterized by $A$-$B$-bimodules.
However, there also exists a different characterization of the irreducible boundary conditions\index{boundary condition} and defect lines in terms of representation theory of an associative algebra over \CN.

Consider two phases, labeled by \emph{simple} \boundA s $A$ and $B$, of a CFT. There is a \emph{classifying algebra for defects\index{classifying algebra!for defects}}, \CD, which classifies the defect lines separating the two phases \cite{fuSs2}. As a vector space, \CD\ is given by
\be\label{CD_space}
    \CD=\bigoplus_{p,q\In\I} \Homaa{U_p\otip A\otim U_q} A\otimes_\CN \Hombb{U_\pb\otip A\otim U_\qb} B\,.
\ee
Using the factorization constraints\index{factorization!constraints} we will derive a product on this space that turns it into a semisimple, commutative, associative unital algebra over \CN \cite{fuSs2}.  In fact, \CD\ comes with a non-degenerate invariant bilinear form, i.e.\ \CD\ is a Frobenius algebra. The bilinear form is (up to normalization) given by a product of two two-point functions on the sphere, see \eqref{bil_form_CD}. Thus, just like in the case of boundary and bulk Frobenius algebras, the Frobenius property follows from the non-degeneracy of a two-point function. The important property of \CD\ is that the irreducible representations of  \CD\ are in bijection with the elementary defect conditions, and the representation matrices (which are all one-dimensional) are exactly the \emph{defect transmission coefficients} discussed in section \ref{top_def}, see \eqref{def_trans}.
In this sense, \CD\ classifies the $A$-$B$-defect lines of the theory. We will return to the algebra $\CD$ in chapter \ref{sec:class_alg}.

There is also a classifying algebra for boundary conditions\index{classifying algebra!for boundary conditions} \cite{fuSs}, denoted by \CA. As a vector space, the classifying algebra for boundary conditions of a CFT in phase $A$ is
\be\label{CA_space}
    \CA=\bigoplus_{p\In\I}\Homaa{U_p\otimes^+ A\otimes^- U_\pb}A\,.
\ee
Since any left $A$-module is trivially also a right $\one$-module, the categories of $A$-modules and $A$-$\one$-bimodules are equivalent. As a consequence, \CA\ is isomorphic to the classifying algebra \CDone\ of $A$-\one-defects. Indeed, after an appropriate choice of basis, the structure constants of \CDone\ coincide with the ones of \CA, see \cite[Section 2]{fuSs2}. In particular, all the properties of \CD\ are inherited by \CA. Thus, \CA\ is a semisimple commutative Frobenius algebra over \CN.

The irreducible representations of \CA\ are in bijection with the elementary boundary conditions. The representation matrices are the \emph{reflection coefficients}, which can be collected in what one calls boundary states.
A reflection coefficient $\refc q\gamma N$ is (up to a factor $\dim(N)$) the single structure constant of the expansion, in the one-dimensional space of conformal blocks, of the correlator a disc with boundary condition $N$ and a single bulk field $\phi^\gamma_{q\qb}$ inserted.
The boundary states contain essential physical information regarding boundary conditions such as ground state degeneracies \cite{aflu} and Ramond-Ramond charges of string compactifications \cite{bDlr}.
In addition, the annulus coefficients, i.e.\ the coefficients $ \mathrm A_{kM}^{\;\;\;\;N}$ appearing in the decomposition \eqref{Ann_exp} of the annulus partition function, can be expanded in terms of the invariant bilinear form of $\CA$ and the reflection coefficients \cite{stig5}:
\be\label{ann_amp_BF}
    \Aampk MNk ~=~\frac{\dim(M)\dim(N)S_{0,0}^2}{\dim(A)}\sum_{q\in\I}\frac{S_{k,q}}{S_{0,q}}\sum_{\gamma,\delta=1}^{Z_{q\qb}}
\nbfinv {\qb}\delta q\gamma\,\refc q\gamma N\,\refc {\qb}\delta M\,,
\ee
where the numbers $\nbfinv {\qb}\delta q\gamma$ are the elements of the inverse of the $\dim(\CA)\linebreak\Times\dim(\CA)$-matrix $\omega$ that implements the non-degenerate bilinear form on $\CA$.
\begin{remark}
    We conclude the discussion of classifying algebras  by mentioning that the classifying algebra for boundary conditions can be obtained from a generalization of the coend \L.
    Consider the functor
    \be
        F_A=\,\,?^\vee\otip A\otim?:\,\,\C\op\times\C\rightarrow\C_{A|A}\,,
    \ee
    acting on objects as
    \be
  X\Times Y \,\mapsto\, X^\vee\otip A\otim Y\,,
  \ee
and on morphisms as
  \be
   f\times g \,\mapsto\, f^\vee\oti \id_A\oti g \,.
  \ee
Denote by $\HA$, the coend of $F_A$:
  \be
    \HA := \Coend {F_A} X  = \coen X X^\vee\otip A\otim X \,.
  \ee
Assume that the coequalizer of the right and left actions on $X$ and $Y$ exist for all bimodules $X$ and $Y$ in $\C_{A|A}$. In this case $\C_{A|A}$ is monoidal. Assume further that  $X\otiA Y$ is a retract of $(X\oti Y,\rho_X\oti\id_Y,\id_X\oti\ohr_Y)$, with embedding and restriction morphisms $\eota XY$ and $\rota XY$.
Finally, we assume that the projectors $\pota XY \,{:=}\, \eota XY\cir\rota XY$ commute with bimodule morphisms and that the left and right unit constraints of the monoidal category $\C_{A|A}$ are given by the representation morphisms. In fact we  need the last two assumptions only for $X$ and $Y$ that are tensor powers of $\HA$.

Under these assumptions, the object $\HA$ can be equipped with the structure of an algebra in $\C_{A|A}$. The algebra structure on $\HA$ is generalization\footnote{To be precise, we generalize a version of $\L$, in which we have interchanged over- and under-braiding in the definition of the product.} of the one on $\L$:
  \eqpic{def_m_HA} {285} {45}
  {\put(7,0){\setulen90
  \put(0,5)   {\INcludepichtft{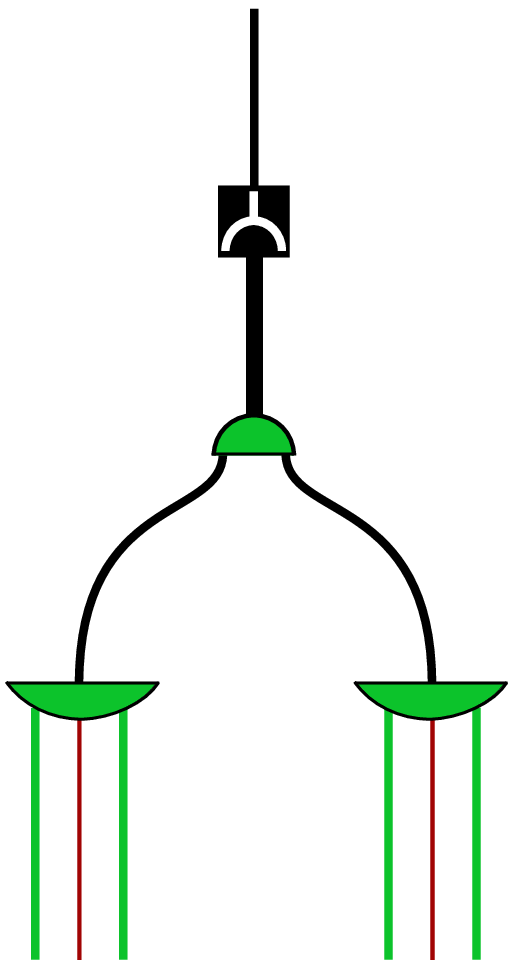}{342}
  \put(-6,-9)  {\sse$ X^{\!\vee} $}
  \put(5,-9)   {\sse$ A $}
  \put(13,-9)  {\sse$ X $}
  \put(32,-9)  {\sse$ Y^{\!\vee} $}
  \put(43,-9)   {\sse$ A $}
  \put(51,-9)  {\sse$ Y $}
  \put(19,29)  {\sse$ \iota_X$}
  \put(57,29)  {\sse$ \iota_Y$}
  \put(-3,44)  {\sse$ \HA$}
  \put(47,44)  {\sse$ \HA$}
  \put(29,67)  {\sse$ \HA\OtiA\HA$}
  \put(33,80)  {\sse$m_{\HA}$}
  \put(23,109)  {\sse$ \HA$}
  \put(-22,60)     {\sse$\rota \HA\HA$}
  }
  \put(77,45) {$ := $}
  \put(117,5) {\INcludepichtft{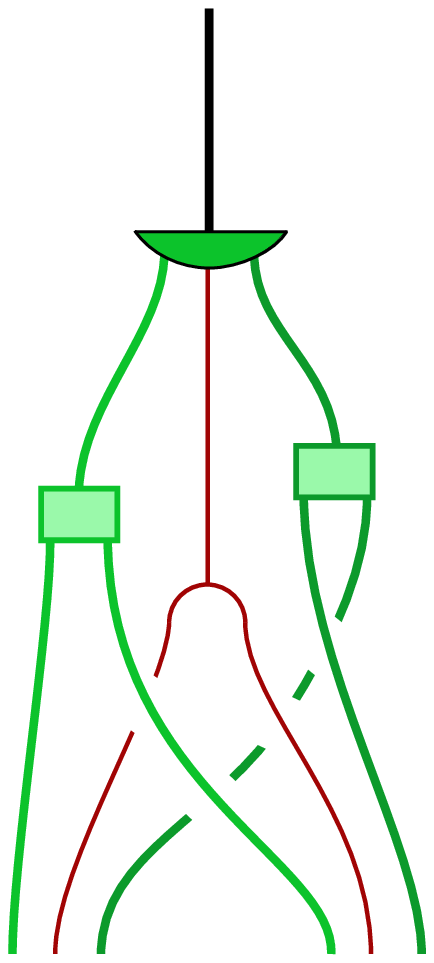}{342}
  \put(-11,-9)  {\sse$ X^{\!\vee} $}
  \put(0,-9)   {\sse$ A $}
  \put(8,-9)  {\sse$ X $}
  \put(28,-9)  {\sse$ Y^{\!\vee} $}
  \put(38,-9)   {\sse$ A $}
  \put(46,-9)  {\sse$ Y $}
  \put(-18,48)  {\sse$\gamma_{X,Y}$}
  \put(43,52)  {\sse$\id_{Y\otimes X}$}
  \put(17,109)  {\sse$ \HA$}
  \put(33,79)  {\sse$ \iota_{Y\otimes X}$}
  }
  \put(220,45)  {and$\quad\quad\eta_{\HA}=\iota_{\one}$.}
} }

Taking \C\ to be a modular tensor category, $\HA$ is as an object, given by
\be
    \HA \cong\, \bigoplus_{i\in\I} U_i^\vee \otip A\otim U_i^{}\,.
\ee
All the properties of the category that we assumed above are indeed satisfied for a modular tensor category. Accordingly, for a modular category, $\HA$ is an algebra in $\C_{A|A}$.
As described in appendix \ref{app:alg_Hom}, the product on $\HA$ induces a product on the vector space
\be
    \Homaa{A}{\HA}\cong\sum_{i\in\I}\Homaa A{U_i\otip A\otim U_{\ib}}\,.
\ee
Note that $\Homaa{A}{\HA}$ is isomorphic, as a vector space, to the classifying algebra for boundary condition $\CA$, given in \eqref{CA_space}. In fact, this isomorphism extends to an isomorphism of algebras: One can check that with a suitable choice of basis for $\Homaa{A}{\HA}$, the structure constants that endow $\Homaa{A}{\HA}$ with an algebra structure coincide with the ones of the classifying algebra.
\end{remark}

\section{Hopf algebras}
Hopf algebras play a fundamental role in physics. To mention a couple of standard examples, Hopf algebras appear in connection with representations of finite groups and of Lie algebras. When dealing with representations of a finite group $G$ in physics, one is typically dealing with representations of the group algebra $A[G]$. Similarly, the representations of a Lie algebra $\g$ are really to be considered as representations of the enveloping algebra $U(\g)$. The group algebra $A[G]$, as well as the enveloping algebra $U(\g)$, is a Hopf algebras.

In conformal field theory, Hopf algebras have been associated, via what has been termed Kazh\-dan-Lusztig correspondence, to the logarithmic $(1,p)$-models \cite{fgst2,naTs2}. The category of representations of the vertex algebra of a $(1,p)$-model is equivalent, as an abelian category, to the category of finite dimensional representations of the restricted quantum group $\bar U_q(\mathfrak{sl}_2)$ at $q=\e^{\pi\i/p}$.
However, the equivalence does not extend to an equivalence of braided categories \cite{koSai}, but there has been attempts to improve the situation, see e.g.\ \cite{SeT}.

\subsection{The \bhHopf}\label{BhHopf}
In section \ref{TFT} we introduced the notion of a topological field theory in terms of a modular functor, the \tftfun. This is a central ingredient in the description of full rational CFT via the TFT-construction. In particular it provides representations of mapping class groups.

Many aspects of modular functors have been generalized by Lyubashenko and others to a much larger class of braided categories, which are not required to be semisimple (see \cite{KEly} and references therein).
In particular one can construct an action of the mapping class group \Mapgn\ of any closed oriented surface $\Surf gn$ of genus $g$ with $n$ holes. For any factorizable finite ribbon category $\C$, i.e.\ a category that satisfies condition \CFN, described in on page \pageref{cond:CFN}, there is an action $\piL$ of \Mapgn\ on a suitable morphism space.
This construction uses as a central ingredient the coend \L, defined in \eqref{Lyub_coend}.

Let \C\ be a factorizable finite ribbon category.
Since \C\ is a $\k$-linear finite (in the sense of Definition \ref{def_finite}) category, the category \eC\ is $\k$-linear finite as well \cite[Corally 3.7]{rgaW}. In addition, the tensor product, the braiding and the sovereign structure on \C\ can be used to equip \eC\ with the structure of a ribbon category.
Thus the coend $\HK\deff\L_{\eC}$ indeed exists in \eC\ and carries the structure of a Hopf algebra in \eC.
This is the \emph{\bhHopf}:
\be\label{def_bhHopf}
    \HK:=\L_{\eC}=\int^X X^\vee\oti X\quad\in\eC\,.
\ee
If in addition, the Hopf pairing of \HK\ is non-degenerate, \eC\
is factorizable ribbon as well.
In particular, as described in section \ref{coends}, we obtain an action of the mapping class group $\Mapgn$ on the space  $\Hom(\HK^{\otii g},\BF^{\otii n})$, where \BF\ is the \bulkAC, for any integers $g,n\geq0$. In chapter \ref{sec:beyond} we give a prescription of how to associate a morphism \Corrgn\ in $\Hom(\HK^{\otii g},\BF^{\otii n})$ to any surface \Surf gn, see \eqref{CorrgnC}. This prescription uses in addition to the structure morphisms of \BF\ the action \eqref{Lact} of \HK\ on \BF.

In the case $\C\eq H$\Mod\ we show that the morphism \Corrgn\ is invariant under the projective action $\piK\equiv\piL$ of \Mapgn on $\Hom(\HKH^{\otii g},\BFH^{\otii n})$. In particular, considering the case $g\eq1$, $n\eq0$, we obtain a candidate for the torus partition.
In the case that \C\ is modular, the partition function reduces to the charge conjugation modular invariant partition function.
Due to the properties just mentioned the morphisms \Corrgn\ qualify as candidates for correlators of bulk fields also in non-semisimple CFT.

In section \ref{sec:twist_bulk} we generalize the considerations by replacing $\BFH$ by the symmetric commutative Frobenius algebra $\BFw$ associated to a ribbon Hopf algebra automorphism $\ra$. We construct a morphism $\Corrgnw\In\Hom(\HKH^{\otii g},\BFw^{\otii n})$ from each symmetric commutative Frobenius algebra $\BFw$.
We will show that the morphism $\Corrgnw\In\Hom(\HKH^{\otii g},\BFw^{\otii n})$ is also invariant under the action of the  mapping class group of $\Surf gn$ for any $g,n\geq0$.
\chapter{Rational CFT}\label{chap:RCFT}
In this chapter we describe the relevant details of the \emph{TFT-construction} of rational CFT and show that the proof in \cite{fjfrs} that the correlators of the TFT-construction satisfy the factorization constraints extends also to the situation that the world sheet contains an arbitrary network of topological defect lines. Throughout this chapter, \C\ denotes a modular tensor category, with ground field $\CN$, as defined in section \ref{sec_modC}.
\section{The TFT-construction}\label{sec:TFT}
As discussed in the introduction, the calculation of correlation functions of a world sheet $\wsC$ can be split into two parts, see section \ref{fullCFT}. Recall that the second part of the problem amounts to assigning to a surface decorated by data of a modular tensor category a vector in the space of conformal blocks on the double $\wsCD$, that  satisfies the factorization constraints and mapping class group invariance. For a rational CFT a model independent solution of this problem is provided by the TFT-construction  \cite{fuRs4,fuRs8,fuRs10,fjfrs,fjfrs2,fjfs}.

The two key ingredients of the TFT-construction are the \tftfun, relevant to the chiral theory, and the concept of the double $\wsCD$ of a world sheet $\wsC$, that links the full theory to the chiral.
To any world sheet $\ws$, the TFT-construction assigns  a cobordism $\Mws:\emptyset\rightarrow\wsD$, the \emph{connecting manifold}.
Applying the invariant $\tftm(\Mws):\CN\rightarrow\tfts(\wsD)$ to $1\In\CN$, specifies a vector in $\tfts(\wsD)$.
This vector is the correlator of the topological world sheet\index{world sheet!topological} $\ws$.
By identifying the TFT-space $\tfts(\wsD)$ with the space $\cbl(\wsCD)$ of conformal blocks on $\wsCD$, this mechanism also selects a vector in  $\cbl(\wsCD)$, which is the correlator of \wsC.

\subsection{Conformal blocks and the TFT-state space}\label{Cbl_TFTspace}
The TFT-state space $\tfts(\wsD)$, associated to the extended surface $\wsD$ by the \tftfun, and the space $\cbl(\wsCD)$ of conformal blocks of the associated vertex algebra both take a modular tensor category $\C$ as an input. Apart from that, the two spaces are a priori unrelated.

Note that we do not just want to identify $\tfts(\wsD)$ and $\cbl(\wsCD)$ as vector spaces, this would merely be a statement about dimensions. Remember that both $\cbl(\wsCD)$ and $\tfts(\wsD)$ are equipped with actions of the mapping class group\index{mapping class group!representation of} of $\wsCD$.
The first one comes from a projectively flat connection on the bundle of conformal blocks, whereas the latter comes from the isomorphisms \eqref{homeomap} which are part of the \tftfun.
By identifying $\cbl(\wsCD)$ and $\tfts(\wsD)$ we mean first that these two representations are isomorphic. Second, we require that the isomorphism is consistent with cutting and gluing of extended surfaces, i.e.\ with factorization.
It is therefore a non-trivial statement that the two spaces can be identified. In the case of three-dimensional Chern-Simons theory the TFT-state space on a component of the boundary can be thought of as the space of conformal blocks of the corresponding WZW-model, \cite{witt50,frki2}.

For general rational theories there are some results at genus zero and genus one: At genus zero, an explicit isomorphism, consistent with the mapping class group representations, between $\tfts(\wsD)$ and $\cbl(\wsCD)$ is constructed in \cite[Section 5.3.3]{fuRs4}.
Furthermore, as a consequence of the Verlinde conjecture, which is proved in \cite[Theroem 5.2]{huan20}, the two spaces $\tfts(\wsD)$ and $\cbl(\wsCD)$ can be identified also in the case $g\eq1$ and without insertions. Even though there is no general treatment for $g>1$, it is generally believed that such an identification is possible at any genus. To be on the safe side we here impose the condition (which is possibly not a restriction) that the vertex operator algebra \VA\ is such that we can identify the two spaces, and refer to $\tfts(\wsD)$ as the space of conformal blocks. Thus, we define full rational CFT\index{CFT!rational|textbf} to be a CFT with vertex algebra \VA\ such that:
    \begin{itemize}\addtolength{\itemsep}{-6pt}
    \item The strictification \C\ of the representation category of the chiral algebra \VA\ is modular.
    \item The space of conformal blocks $\cbl(\wsCD)$ of \VA\ can be identified (in the sense described above) with the TFT-state space $\tfts(\wsD)$ of the associated topological manifold, at arbitrary genus.
    \end{itemize}
The TFT-construction solves full CFT which is rational in this sense.
In particular, in such a theory we are able to restrict to topological world sheets when discussing the algebraic side of full CFT.

\subsection{Decoration of the world sheet}\label{Dec_ws}
We now describe how to obtain the correlators of such a rational CFT via the \tftfun\ described in section \ref{TFT}. The construction involves many subtleties and technical aspects, which are omitted in the description below. See e.g. \cite[Section 3-4]{fuRs10} for a more complete description or \cite[Appendix B]{fjfrs} for a shorter summary of the TFT-construction. The following description should be considered as a recipe for how to cook up the connecting manifold of a world sheet $\ws$. For the considerations in this thesis we only need oriented world sheets. For this reason, we restrict to orientable world sheets in the description below. In fact, we even implicitly assume that $\ws$ is \emph{oriented}. However, the constructions for the non-orientable case is similar, see \cite{fuRs10} for details.

Consider a topological oriented world sheet $\ws$, with field insertions and possibly with non-empty boundary and a network of finitely many defect lines. We will only consider trivalent network vertices of defect lines. Since defect lines can be fused this is not a restriction, see Remark \ref{rem_fus_vert} (i) below.

Remember that phases of a full CFT are separated by defect lines. As was explained in section \ref{boundFA}, each phase is labeled by a \boundA, i.e.\ a symmetric special Frobenius algebra $\tftA$ in \C. Each boundary segment is labeled by a module over \tftA, and each defect line by an $A$-$B$-bimodule, where $A$ and $B$ are the labels of the adjacent phases. The fusion product of two defects, labeled by an $A$-$B$-defect $X$ and a $B$-$C$-defect $Y$, is labeled by $X\otiB Y$, c.f.\ \eqref{fus_def}. Analogously, a fusion product between an $A$-$B$ defect $X$ and a boundary component labeled by the $B$-module $M$ is labeled by $X\otiB M$.

In addition there are fields living in the bulk, on defect lines or on boundary components. The fields are decorated as follows:
\begin{itemize}
    \item A \textbf{bulk field} in phase $A$ is a tuple $\Phi\eq\BUF UV\phi p{[\gamma]}{\ort(p)}$. Here $U,V$ are objects in $C$, $\phi\In\Homaa{U\otip A\otim V}{A}$, $p\In\ws\setminus\partial\ws$ is the point where the field is inserted, $[\gamma]$ is an arc germ at $p$ such that $\gamma(0)=p$, and $\ort(p)$ is the orientation of $\ws$ in a neighborhood around $p$.
    \item A \textbf{defect field} changing an $A$-$B$ defect $Y$ to the $A$-$B$ defect $X$ is a tuple $\Phi\eq\DF XYUV\phi p{[\gamma]}{\ort(p)}$. Here  $U$ and $V$  are objects in $\C$, $X$ and $Y$ are $A$-$B$-bimodules, $\phi\In\Homab{U\otip X\otim V}{Y}$, $p\In\ws$ is the point on the defect line where the field is inserted, $[\gamma]$ is an arc germ at $p$, and $\ort(p)$ is the orientation of $\ws$ in a neighborhood around $p$.

        Recall that a bulk field can be considered as a defect field on the invisible defect, and that a defect field changing the defect condition to or from the trivial defect is called a \emph{disorder field}.
    \item A \textbf{boundary field} in phase $A$ is a tuple $\Psi\eq\BOF MNV\psi p{[\gamma]}$. Here $M$ and $N$  are left $A$-modules, $V$ is an object in $\C$, $\psi\In\Homa{M\oti V}{N}$, $p\In\partial\ws$ is the insertion point of the boundary field and $[\gamma]$ is an arc germ at $p$ such that $\gamma$ is aligned to $\partial\ws$ at $p$. A boundary field of this type changes the boundary condition from $N$ to $M$.
\end{itemize}
Here, $U\otip X\otim V$ is the bimodule with actions given in \eqref{ind_mod}.

For an orientable world sheet we can fix the orientation of the entire world sheet once and for all. Thus the datum of orientation at insertion points of bulk- and defect fields is somewhat superfluous when dealing with orientable world sheets. The arc germ at $p\In\ws$ can be thought of as a remnant of the local coordinates around $p$ in the world sheet $\wsC$ from which we obtain the topological world sheet $\ws$, c.f.\ \cite[Section 5.3.1]{fuRs10}.

Due to semisimplicity, it is enough to take the objects in \C\ that label field insertions to be \emph{simple}.
For $A$-$A$-defects $X$ and $Y$ we fix once and for all a basis
\be\label{basis_DF}
  \{\, \phi^\alpha_{pq} \,|\, p,q\In\I,\,\alpha=1,2,...,\dim(H^{X,Y}_{p,q}) \,\}
\ee
for the space
\be\label{space_DF}
    H^{X,Y}_{p,q} \,{:=}\, \Homaa{U_p\,\oti^+\!X\oti^-U_q}Y\,.
\ee
The space $H^{X,Y}_{p,q}$ is the \defspace, with chiral labels $p$ and $q$, changing the defect condition from $Y$ to $X$.
We will also need a basis of the space dual to \eqref{space_DF}. We denote it by
\be\label{basis_DFd}
  \{\, \bar\phi^\alpha_{pq} \,|\, p,q\In\I,\,\alpha=1,2,...,\dim(H^{X,Y}_{p,q}) \,\}.
\ee
We will abbreviate
\be
    \phi^\alpha_{pq}\equiv\alpha\quand\bar\phi^\alpha_{pq}\equiv\bar\alpha
\ee
whenever the labels $p,q$ can be read of from the context.
The two bases are dual in the sense that
\be\label{DF_dual}
  \mathrm{Tr}\big(\phi^\alpha_{pq}\circ\bar\phi^\beta_{pq}\big)
  = \delta_{\alpha,\beta}\;\dim(Y)\,.
\ee
Such dual bases exist due to a non-degenerate pairing of the spaces
$\Homaa{U_p\linebreak{\otip}  X\otim U_q}Y$ and $\Homaa Y{U_p\,\otip \otim U_q}$. The existence of such
a pairing, in turn, follows from the following arguments: Due to semisimplicity of \C\ there exists a non-degenerate pairing of the spaces $\Hom(U_p\oti X\oti U_q,Y)$ and $\Hom(Y,U_p\oti X\oti U_q)$. We may chose bases of those spaces in such a way that the pairing is mapping a pair of basis elements $\psi^\alpha_{pq}\In\Hom(U_p\oti X\oti U_q,Y)$ and $\bar\psi^\beta_{pq}\In\Hom(Y, U_p\oti X\oti U_q)$ to
\be
    \mathrm{Tr}\big(\psi^\alpha_{pq}\,{\circ}\,\bar\psi^\beta_{pq}\big)\,{=}\,\delta_{\alpha,\beta}\dim(Y)\,.
\ee
Taking $\psi^\alpha_{pq}\In\Homaa{U_p\otip X\otim U_q}Y$ one shows, by using that $A$ is symmetric Frobenius together with the representation properties, that this pairing restricts to a pairing between \Homaa{U_p\otip X\otim U_q}Y and $\Homaa Y{\linebreak  U_p\otip X\otim U_q}$. (The calculation is an obvious generalization of eq.\ (C.16) in \cite{fjfrs}.) Analogously, taking $\bar\psi^\beta_{pq}\In\Homaa Y{U_p\otip X\otim U_q}$, the non-degene\-rate pairing $\mathrm{Tr}\big(\psi^\alpha_{pq}\,{\circ}\,\bar\psi^\beta_{pq}\big)$ again restricts to a pairing between \Homaa{U_p\linebreak\otip   X\otim U_q}Y and $\Homaa Y{U_p\otip X\otim U_q}$.
\pagebreak
\subsection{The connecting manifold}\label{sec:con_man}
We are now in a position to describe the connecting manifold $\Mws$ of an oriented world sheet. As a three-manifold, $\Mws$ is constructed from an interval bundle $\ws\times[-1,1]$ over the world sheet by pairwise identifying points over the boundary:
\be\label{Mws_def}
		\Mws=\ws\times[-1,1]\;\Big/\sim\;,
\ee
where
\be
(x,t)\sim(x,-t)\;\forall \,x\in\partial \ws\quand\forall \,t\In[-1,1].
\ee
Thus, the connecting manifold is obtained by "filling out" the double $\wsD$ of $\ws$ to a three-manifold.
By construction there is an embedding $\wsemb\mapdef\ws\hookrightarrow\Mws$ of the world sheet in $\Mws$ as the set of points $(p,0)\In\ws\times[-1,1]$. Note that the boundary of the connecting manifold is the double of the world sheet,
\be
    \partial\Mws~=~\wsD\,.
\ee

The double is orientable. We fix once and for all the orientation by the inward pointing normal. There is an interval $p\times[-1,1]$ connecting the two points on the double $p'$ and $p''$ over each point $p$ in the interior of the world sheet. We refer to this interval as the \emph{connecting interval}. A point $p$ on the boundary of the world sheet gives rise to a single point $p'$ on the double and a connecting interval $[0,1]$ stretching from $p$ to $p'$. When a point in the interior is moved to the boundary, the two halves of the connecting interval merge into an interval connecting the point $p\In\partial\ws$ with the corresponding point $p'\In\wsD$. All of this is illustrated in the following picture of the connecting manifold of a disc:
\eqpic{con_man}{175}{80}{
\put(0,0){\setulen93
\put(0,0){\includepic{316}{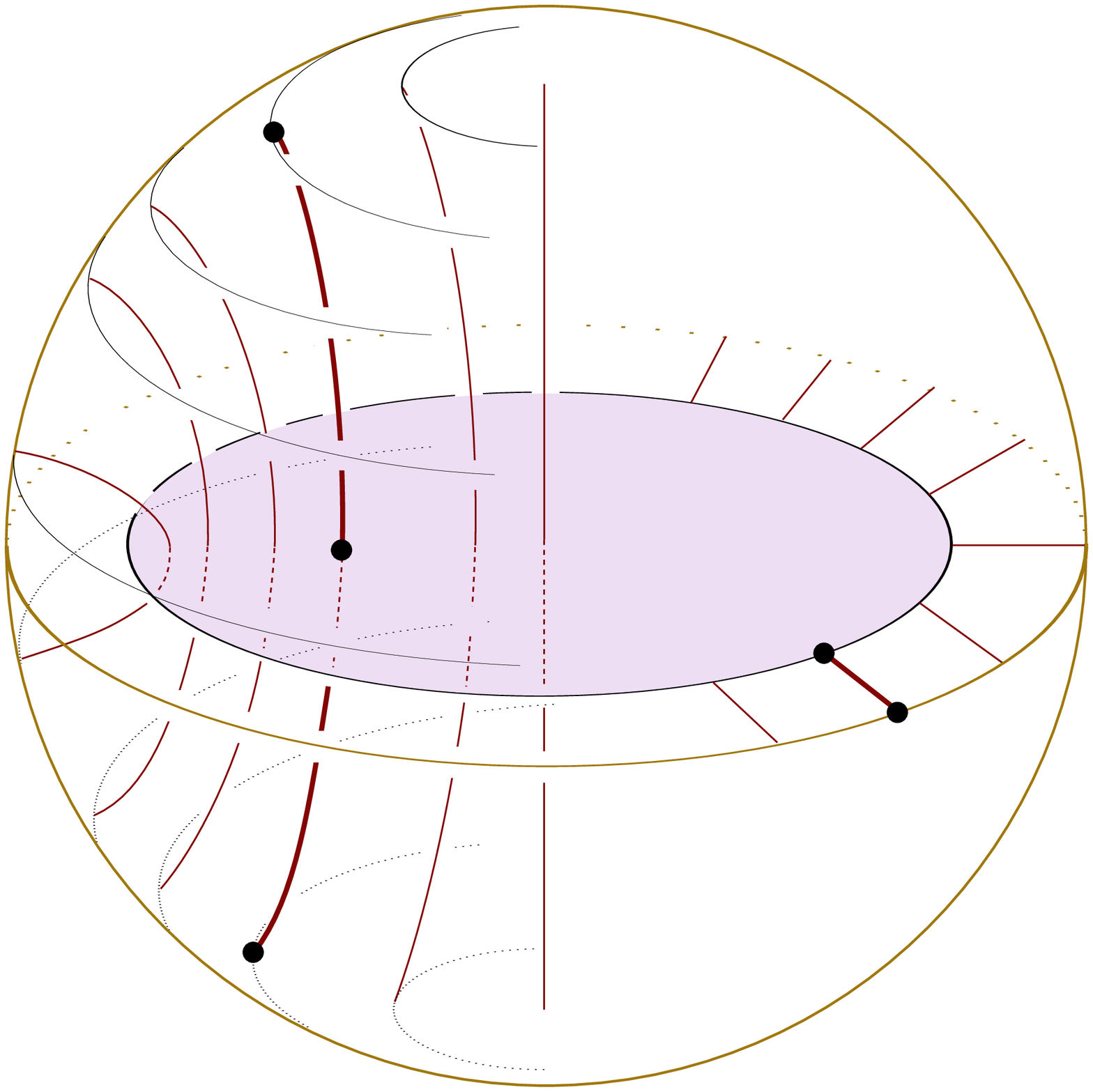}

\put(59,87)   {\sse$p_1$}
\put(48,156)   {\sse$p_1'$}
\put(45,22)   {\sse$p_1''$}
\put(130,77)   {\sse$p_2$}
\put(146,55)   {\sse$p_2'$}
}
}
}
Here, the thicker lines are examples of connecting intervals over the points $p_1$, in the interior of the world sheet $\ws$, and $p_2\In\partial\ws$.
\subject{The double as an extended surface}
Part of the data of the field insertions are lifted to the double. Each defect field $\DF XYUV\phi p{[\gamma]}{\ort(p)}$ gives rise to two marked points, $p'$ and $p''$, on the double, one of them in a neighborhood with orientation coinciding with the one around $p$, and the other one in a neighborhood with opposite orientation. The one in a neighborhood with orientation $\ort(p)$ gets labeled by the object $U$ and the other one by $V$. The arc germ $[\gamma]$ induces arc germs at the two marked points $p',p''\In\wsD$, and for each marked point, the datum $\eps$ is taken to be $\eps=-1$. Thus, a ribbon ending on the boundary always has core orientation pointing away from the boundary. Since we may think of a bulk field as a special kind of defect field the data of a bulk field are lifted in the same way. The data of boundary fields are lifted to the single point over the insertion on the double in the obvious analogous fashion. Finally, in order to turn the double into an extended surface we make the canonical choice of lagrangian subspace described in Remark \ref{Rem_cob} (ii).

\subject{The ribbon graph in \Mws}
Let us now describe how to construct the ribbon graph embedded in \Mws. First we choose a \emph{dual triangulation} of the world sheet. This involves decorating the world sheet such that there is a line running along each defect line and along each boundary component. We require that every vertex is trivalent and every field insertion lies on an edge of the triangulation and that vertices of the dual triangulation coincide with the network vertices, i.e.\ vertices at which three defect lines are joined. Note that to fulfill the last requirement, it may be necessary to include some invisible defect lines.

The dual triangulation is then covered by \emph{ribbons} and \emph{coupons}. First we place a ribbon on each edge of the triangulation of the embedded world sheet as follows:
\begin{itemize}
    \item
    On a defect line we place a ribbon labeled by the bimodule describing the corresponding defect or condition. Analogously, on a boundary segment we place a ribbon labeled by the left module describing the corresponding  boundary condition\index{boundary condition}. Each ribbon is embedded with core orientation opposite to the one of the corresponding defect line or boundary condition.
    \item On any other edge of the dual triangulation we place a ribbon labeled by a trivial defect condition, i.e.\ the \boundA\ that labels the adjacent regions.
\end{itemize}
All ribbons placed on the dual triangulation are embedded with two-orientation opposite to the one of the world sheet. I.e., using the terminology of Remark \ref{Rem_cob} (i), the ribbons are embedded with the black side facing up.

In addition we place coupons on field insertions. The coupon at a defect field changing the defect condition from $Y$ to $X$ is labeled by the morphism labeling the defect field, i.e.\ a morphism $\phi$ in $\Homab{U\oti^+X\oti^-V}Y$. The order of $X$ and $Y$ is reversed in the morphism since ribbons are embedded with core orientation opposite to the actual defect. The coupon has three incoming ribbons labeled by $U$, $X$ and $V$ and one outgoing labeled by $Y$. The ones labeled by $X$ and $Y$ are the ones placed on the dual triangulation, whereas the ones labeled by $U$ and $V$ are taken to end on the marked points over the insertion point.  We illustrate the ribbon graph in a region around a defect field in the following picture\footnote{In \eqref{DF_ribbon} the ribbon graph around the coupon labeled by $\phi$ is rotated by an angle $\pi$ in such a way that it shows it white side. This emphasizes that when
this piece of ribbon graph is projected to $\R^2$, as described in section \ref{sec:tftfun}, it is
interpreted as the morphism $\phi$.}:
\eqpic{DF_ribbon}{200}{65}{
\put(-40,0)   {\setulen6{07}
  \put(65,0){\includepic{17}{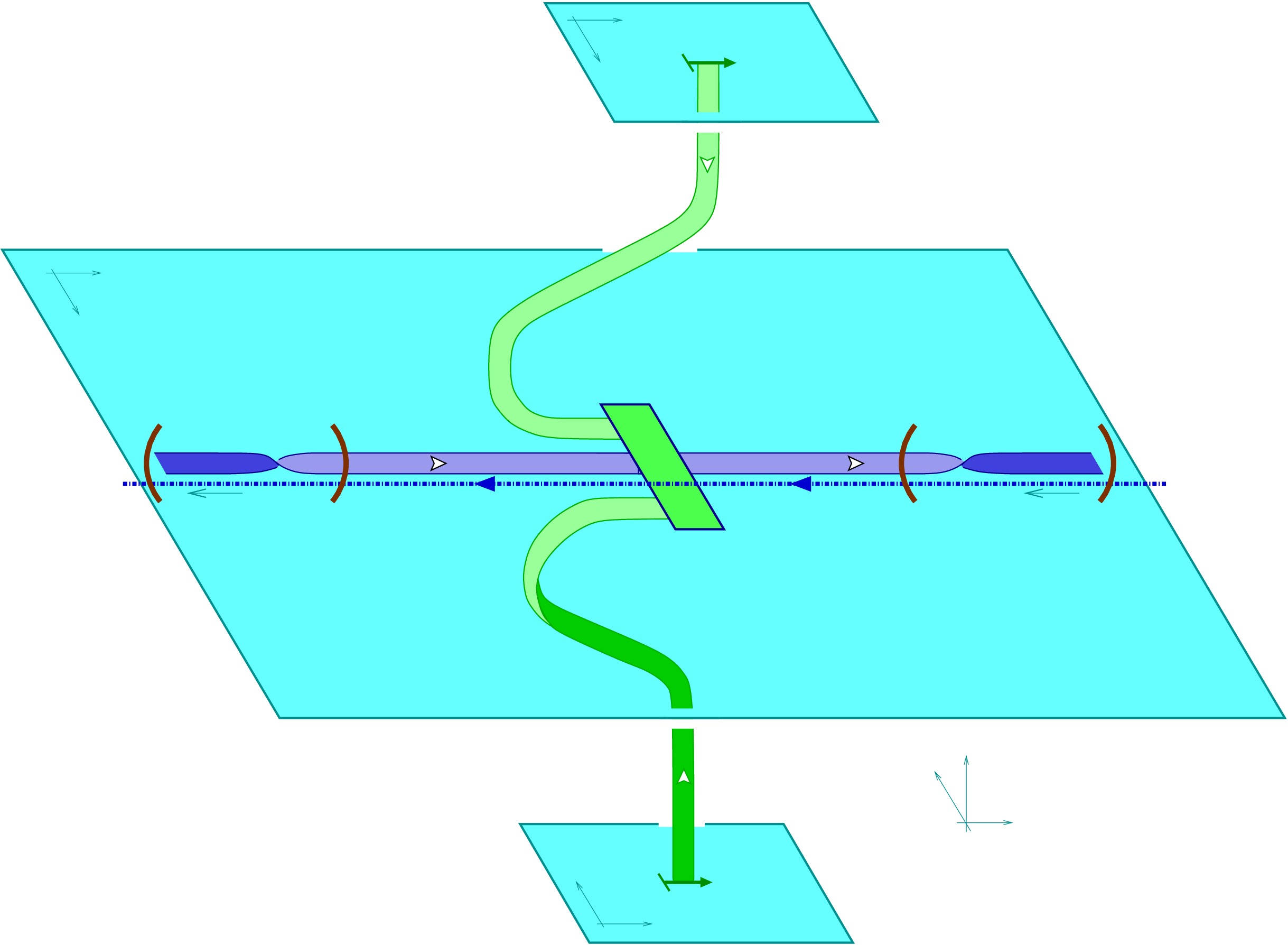}}
  \put(325,30)     {\tiny $x$}
  \put(310,50)     {\tiny $z$}
  \put(302,47)     {\tiny $y$}
  \put(234,4)      {\tiny $x$}
  \put(214,17)     {\tiny $y$}
  \put(93,170)     {\tiny $x$}
  \put(86,160)     {\tiny $y$}
  \put(226,235)    {\tiny $x$}
  \put(220,225)    {\tiny $y$}
  \put(250,192)    {\scriptsize $U$}
  \put(244,41)     {\scriptsize $V$}
  \put(172,129)    {\scriptsize $X$}
  \put(277,129)    {\scriptsize $Y$}
  \put(186,106)    {\small $\rm d $}
  \put(266,106)    {\small $\rm d $}
  \put(112,105)    {\scriptsize $\rm or(\rm d)$}
  \put(316,105)    {\scriptsize $\rm or(\rm d)$}
  \put(226,129)  {\begin{turn}{280} $\phi$\end{turn}}
}
}
Again this description covers also the bulk field case.
The coupon at a boundary field
is labeled by the morphism labeling the boundary field and the ribbons are joined with the rest of the ribbon graph in the obvious analogous fashion.

In order complete the description of the ribbon graph we need to describe how to join the ribbons at trivalent vertices. At trivalent vertices for which at least one edge is labeled by a \boundA\  we join ribbons by coupons labeled as follows:
\begin{itemize}
    \item At each vertex joining three ribbons labeled by a \boundA\ the coupon is labeled by the product or coproduct of that algebra.
    \item At each vertex such that there is a ribbon labeled by a \boundA\ ending on a defect line or boundary component the coupon is labeled by the corresponding representation morphism.
\end{itemize}
In order for this convention to make sense we have to be able to reverse the core orientation of a ribbon labeled by a \boundA. The conventions listed above may give rise to an edge of the triangulation at which two parts of ribbons labeled by some \boundA\ $A$ meet with core-orientation pointing towards or away from each other. Whenever this happens, we join them by a coupon labeled by $\eta\circ m$ or by $\Delta\circ\eps$. The corresponding pieces of ribbon graphs looks as follows:
\eqpic{rev_vertex} {280} {18} {
\put(0,0){
  \put(-5,0)     {\Includepic{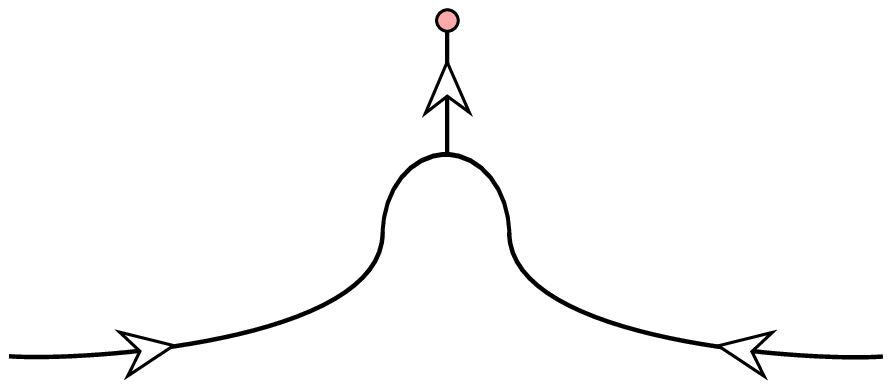}}
  \put(120,21)   {or}
  \put(175,0)     {\Includepic{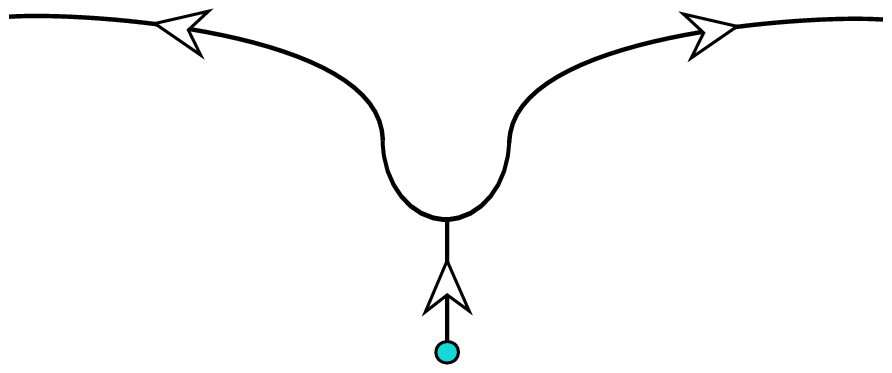}}
  }}

Finally, consider a network vertex joining three ribbons labeled by an $A$-$B$-bimodule $X$ and a $B$-$C$-bimodule $Y$ with core-orientation pointing towards the vertex, and an $A$-$C$-bimodule $Z$ with core-orientation pointing away from the vertex. On this vertex we place a coupon as in the following picture:
\eqpic{def_vertex} {23} {17} {
    \put(0,-5){\includepic{48}{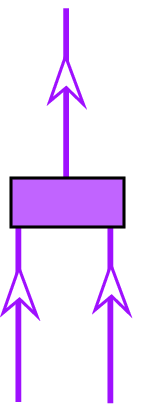}
    \put(11,44)   {\pl {Z}}
    \put(20,9)   {\pl {Y}}
    \put(-8,9)   {\pl {X}}
    \put(7,26)   {\pl {\kappa}}
    } }
The coupon is labeled by a morphism
\be
    \kappa\In\HomP_{A|C}(X\oti Y,Z)\,.
\ee
Here $\HomP_{A|C}(X\oti Y,Z)\subseteq\Hom_{A|C}(X\oti Y,Z)$ is the space of morphisms $f{\In}\linebreak\Hom_{A|C}(X\oti Y,Z)$ that satisfy
\be
    f\circ P_{X,Y}=f\,,
\ee
where $P_{X,Y}\In\End(X\oti Y)$ is the projector, defined in \eqref{TP_proj}, whose image is $X\otiB Y$.

Any vertex such that the core orientations of ribbons are not exactly as described above can be described analogously by using that reversing the orientation of a ribbon corresponds to exchanging the label of that ribbon by its dual, c.f.\ section \ref{sec_Cob} (see picture \eqref{coupon_ex}). If, e.g.\ the ribbon labeled by $Y$ have core orientation pointing away from the vertex, we label the vertex by a morphism in $\HomP_{A|C}(X\oti Y^\vee,Z)$ instead.

In fact,
it is natural to label a vertex with one ingoing defect labeled by an $A$-$C$-bimodule $X$ and two outgoing defects labeled by an $A$-$B$-bimodule $Y$ and a $B$-$C$-bimodule $Z$ by a morphism
\be
    \kappa\In\HomP_{A|C}(X,Y\oti Z)\,,
\ee
with $\HomP_{A|C}(X,Y\oti Z)$ defined, analogously to $\HomP_{A|C}(X\oti Y,Z)$, by the property that $f\In\HomP_{A|C}(X,Y\oti Z)$ iff
\be
    P_{X,Y}\circ f=f\,.
\ee
It is straightforward to show that if $\kappa\In\HomP_{A|C}(X,Y\oti Z)$, then the morphism $\kappa':=(\id_Y\oti\tilde d_Z)\circ(\kappa\oti\id_{Z^\vee})$ is an element in $\HomP_{A|C}(X\oti Z^\vee,Y)$.
\begin{remark}\label{rem_fus_vert}
    (i) As already mentioned, we will consider an $n$-valent network vertex of defects as an equivalent network vertex being a composition of several trivalent ones by fusing defects with each other as in the following picture
    \eqpic{nTo3}{270}{25}{
    \put(0,0)   {\includepic{34}{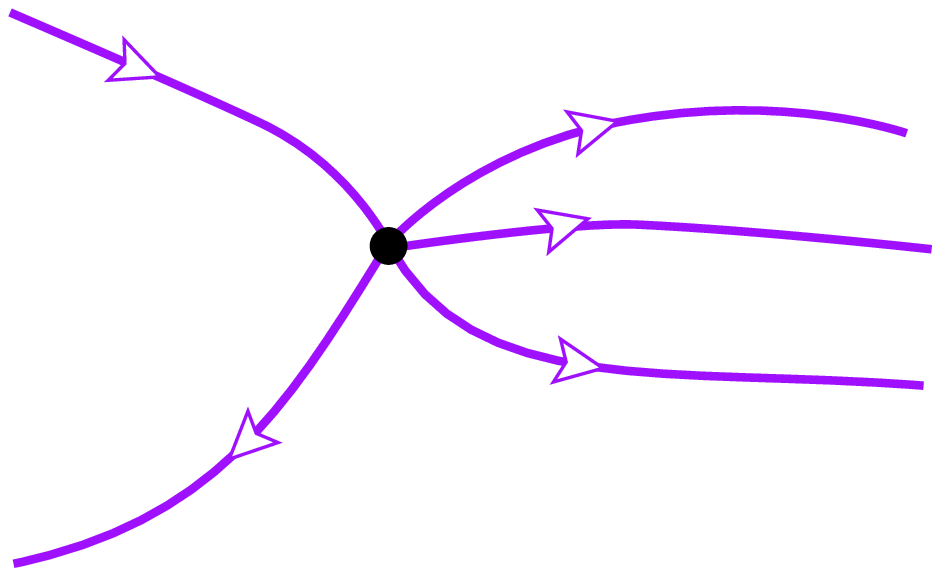}}
    \put(125,25)    {\LARGE$ \rightsquigarrow$}
    \put(160,0)   {\includepic{34}{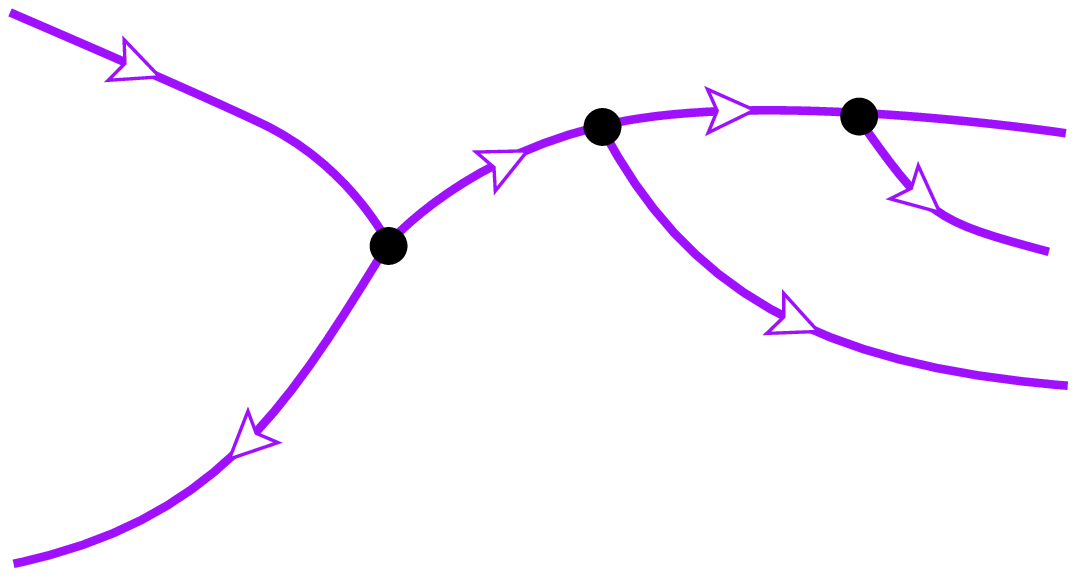}}
    }
    For brevity we will sometimes say that there is an $n$-valent vertex labeled by some morphism $f$. By this we mean that $f$ is a suitable composition of morphisms labeling the trivalent vertices that together constitute the $n$-valent network   vertex.
    \nxl1
    (ii) In \cite{fjfs} a slightly different decoration of the world sheet is used, see \cite[Definition 3.1]{fjfs}. One way in which that description differs from the present one is that instead of a dual triangulation we require only a cell-decomposition of the world sheet. This is useful for proving some of the statements in \cite{fjfs}, e.g.\ it takes care of the issues mentioned in part (i) of this remark. However, here we will not present those proofs in detail, and accordingly we omit that description in favor of the slightly more transparent one given above.
\end{remark}

This completes the description of the construction of the connecting manifold
\be
    \Mws:\emptyset\rightarrow\wsD\,.
\ee
By applying the \tftfun\ we obtain a linear map
\be
    \tft(\Mws)=\tftm(\Mws):\CN\rightarrow\tfts(\wsD)\,.
\ee
The correlator $\Cor(\ws)$ of the world sheet \ws\ is obtained by applying this map to $1\In\CN$:
\be\label{Cor_M}
    \Cor(\ws)=\tftm(\Mws)\,1\equiv\tftm(\Mws)\,.
\ee
Here and below, by a slight abuse of notation, we suppress applying $\tftm(\Mws)$ to $1\In\CN$.
This procedure defines a vector in the space of conformal blocks on the double $\wsD$.
It has been proven that the collection of vectors $\Cor(\ws)$, for all world sheets $\ws$, satisfy the factorization constraints\index{factorization!constraints} and invariance under the action of the mapping class group provided by the \tftfun. This is proved in \cite{fjfrs} for world sheets without defects, and is generalized to oriented world sheets with an arbitrary network of defect lines in \cite{fjfs}. Section \ref{sec:TFT_constr} summarizes the proof that the vectors $\Cor(\ws)$ indeed satisfy the consistency conditions just mentioned.
\subsection{Structure constants}\label{sec:TFT_SC}
Below we will be interested in calculating the structure constants, i.e.\ the expansion coefficients in a basis of the space of conformal blocks on the double:
\be\label{cor_exp_gen}
    \Cor(\ws)=\sum_ic_i\,b_i\,,
\ee
where $\{b_i\}$ is basis for $\tfts(\wsD)$. We will realize each $b_i$ as the invariant $\tftm(\M_{b_i})$ of some cobordism $\M_{b_i}\mapdef\emptyset\to\wsD$. Analogously there is a dual basis $\{b_i^*\}$ satisfying $b_j^*(b_i)\eq\delta_{i,j}$, obtained as invariants of cobordisms $\M_{b_i}^*\mapdef\wsD\to\emptyset$. Thus, the structure constants in \eqref{cor_exp_gen}, obtained in a standard manner by composing both sides with an element of the dual basis, can be written as
\be
    c_i=\tftm(\M_{b_i}^*)\tftm(\Mws)=\tftm(\M_{b_i}^*\circ\Mws)\,.
\ee
This is a ribbon graph in a closed three-manifold. Consider a closed three-manifold $\M$ with a ribbon graph contained in a contractible region.
By projecting the ribbon graph in a non-singular manner, as described in section \ref{sec:tftfun}, to $\R^2$ we obtain a morphism $f$, that is an endomorphism of the tensor unit of $\C$. Since $\End(\one)=\CN$, $f$ is given by a complex number $w_f$.
Thus, the invariant of $\M$ is given by:
\be
    \tftm(\M):=\tftm(\M_0)\,w_f\,,
\ee
where $\M_0$ is the manifold $\M$ with an empty ribbon graph. Of particular interest to us are the following two cases:
\be
    \tftm(S^2\times S^1)=1\quand \tftm(S^3)=S_{0,0}\,.
\ee

\ssubject{$n$-point blocks on the sphere} Consider the two-sphere $S^2$ with $n$ marked points, labeled by objects $U_1,...,U_n$. Denote by $B(U_1,...U_n)$ the corresponding space of conformal blocks. A basis for $B(U_1,...U_n)$ with every $U_p=U_{\ia_p}$ simple is given by (see e.g.\ \cite[Chapter 4.4]{BAki})
\be\label{npbasis_Z}
    \blV[\ia_1,...,\ia_n]_{p_1...p_{n-3},
                  \beta_1...\beta_{n-2}}:=\tftm(\widehat \blS(\ia_1,...,\ia_n)_{p_1...p_{n-3},
                  \beta_1...\beta_{n-2}}),
\ee
where $\widehat \blS(\ia_1,...,\ia_n)_{p_1...p_{n-3},\beta_1...\beta_{n-2}}$ is the following ribbon graph in a full three-ball $B^3$:
\eqpic{npbasis}{300}{74}{
  \put(0,155) {$ \widehat\blS(\ia_1,\ia_2,...,\ia_n)_{p_1p_2...p_{n-3},
                  \beta_1\beta_2...\beta_{n-2}}~:= $}
  \put(80,0){
  \scalebox{0.95}{
  \setulen 81
  \put(0,0){\includepic{308}{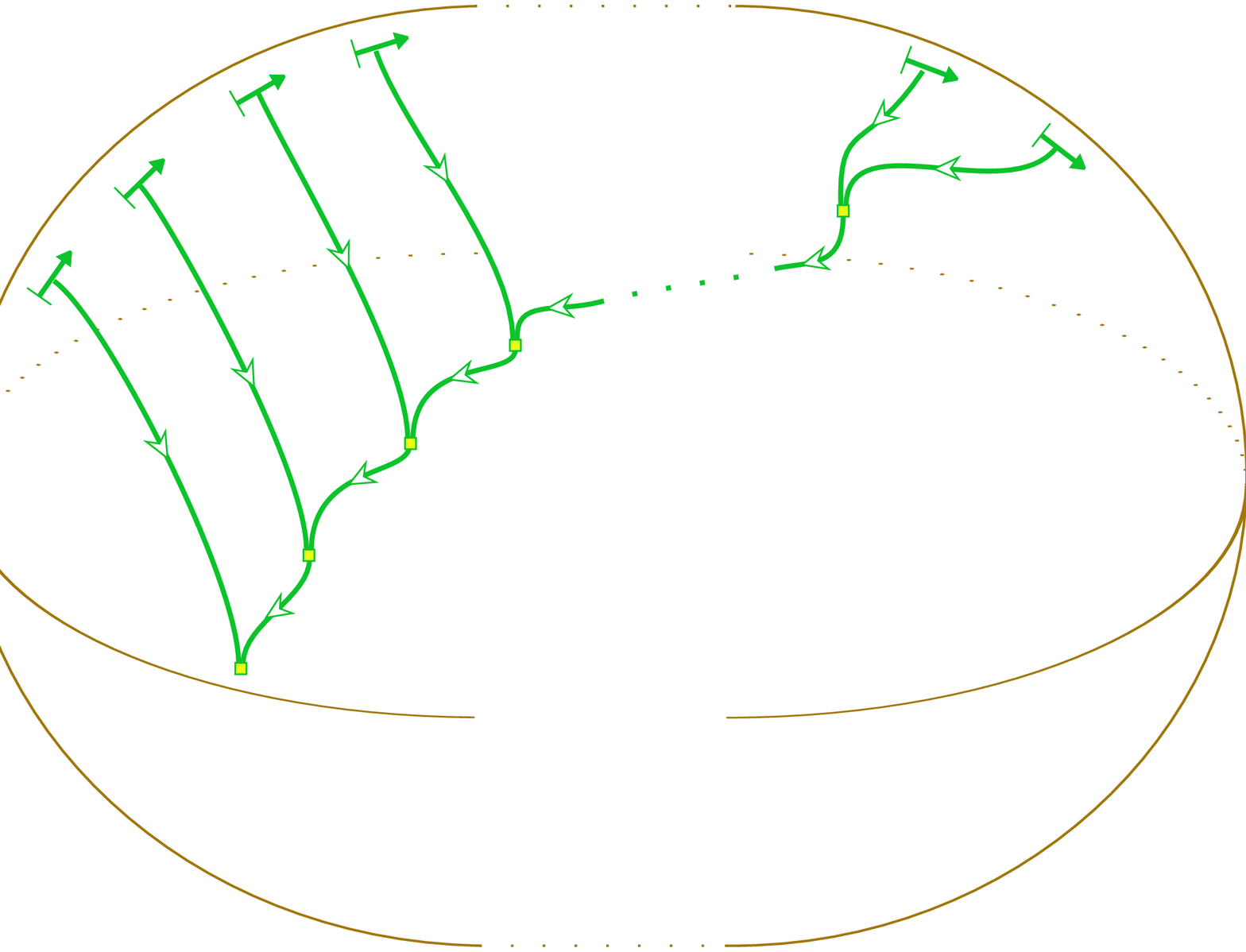}}
  \put(4,133)       {\pl {\ia_1}}
  \put(22,154.8)    {\pl {\ia_2^{}}}
  \put(48,171)      {\pl {\ia_3}}
  \put(73,179.6)    {\pl {\ia_4^{}}}
  \put(181,173)     {\pl {\ia_{n-1}}}
  \put(211,156.2)   {\pl {\ia_n^{}}}
  \put(62,59)       {\pl {\ib_1}}
  \put(76,82)       {\pl {p_1^{}}}
  \put(95,101)      {\pl {p_2^{}}}
  \put(164,126)     {\pl {p_{n-3}^{}}}
  \put(67,72)       {\pl {\beta_1^{}}}
  \put(86.5,92.5)   {\pl {\beta_2^{}}}
  \put(106,111.4)   {\pl {\beta_3^{}}}
  \put(167,137)     {\pl {\beta_{n-2}^{}}}
  } } }
Here $p_1,...p_{n-3}\in\I$ and $\beta_s$, with $s=1,...,n-2$, labels the basis depicted in \eqref{basis_morphspace} of $\Hom(U_{p_s}\oti U_{\ia_{s+1}},U_{p_{s-1}})$,  where we identify $p_0\equiv\ib_1$ and $p_{n{-}2}\equiv\ia_n$.
The dual basis is obtained as the invariants
\be\label{npdualbasis_norm}
\begin{split}
    \blV^*[\ia_1,...,\ia_n]_{p_1...p_{n-3},
                  \bar\beta_1...\bar\beta_{n-2}}
                  :=&\frac1{S_{0,0}}\;
                  \tftm(\widehat \blS^*(\ia_1,...,\ia_n)_{p_1...p_{n-3},
                  \bar\beta_1...\bar\beta_{n-2}})
\end{split}
\ee
of the ribbon graph
\eqpic{npdualbasis}{300}{74}{
  \put(5,155) {$ \widehat\blS^*(\ia_1,\ia_2,...,\ia_n)_{p_1p_2...p_{n-3},
                  \bar\beta_1\bar\beta_2...\bar\beta_{n-2}}~:= $}
  \put(80,0){
  \scalebox{0.95}{
  \setulen 81
  \put(0,0){\includepic{308}{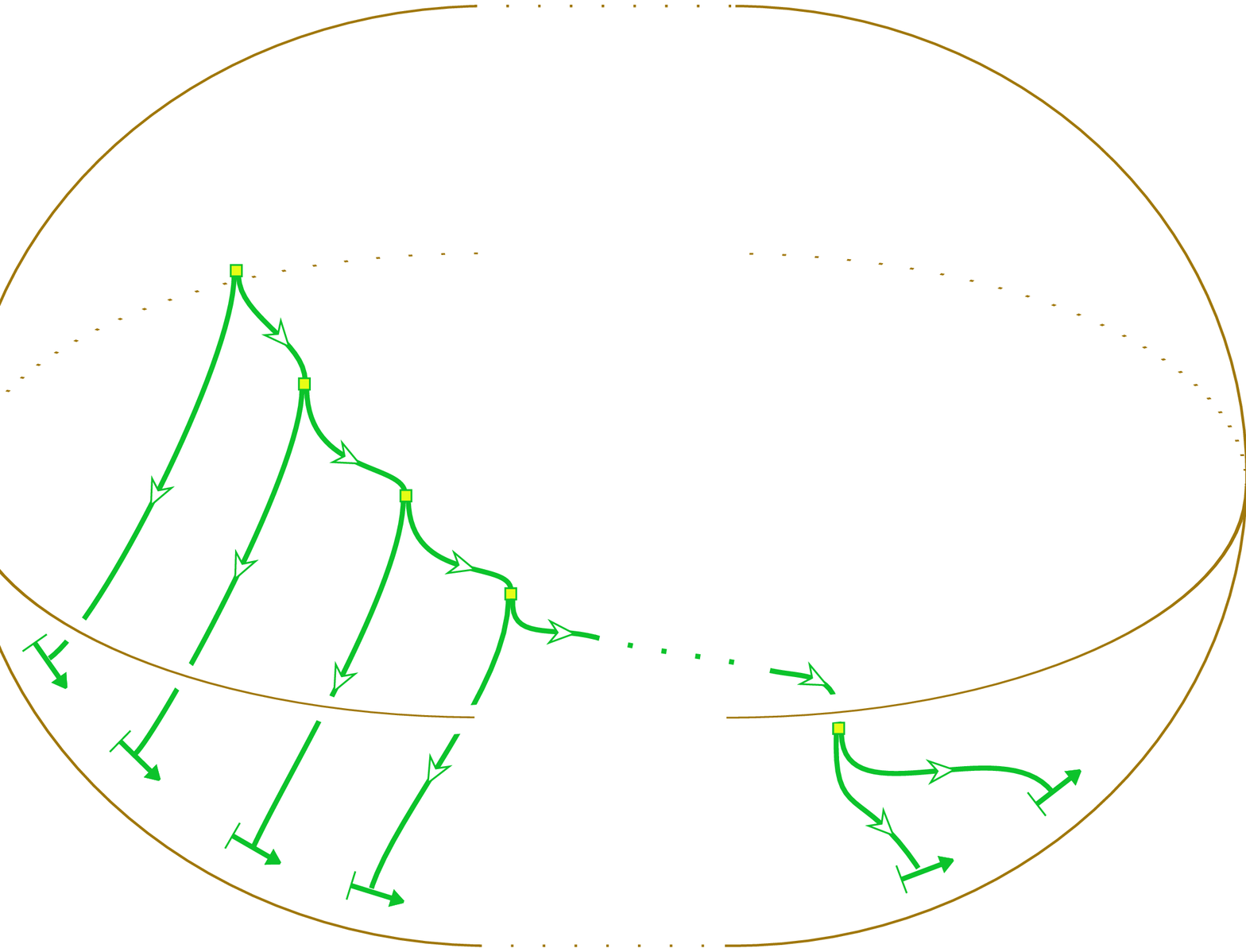}}
  \put(2,44)       {\pl {\ia_1}}
  \put(20,22.8)    {\pl {\ia_2^{}}}
  \put(45,3)      {\pl {\ia_3}}
  \put(72,-4)    {\pl {\ia_4^{}}}
  \put(180,2)     {\pl {\ia_{n-1}}}
  \put(212,19)   {\pl {\ia_n^{}}}
  \put(62,115)       {\pl {\ib_1}}
  \put(76,94)       {\pl {p_1^{}}}
  \put(95,73)      {\pl {p_2^{}}}
  \put(160,55)     {\pl {p_{n-3}^{}}}
  \put(67,104)       {\pl {\beta_1^{}}}
  \put(86.5,83.5)   {\pl {\beta_2^{}}}
  \put(106,64.4)   {\pl {\beta_3^{}}}
  \put(167,39)     {\pl {\beta_{n-2}^{}}}
  } } }
in $B^3$. Here $p_1,...p_{n-3}\in\I$ and $\bar\beta_s$ labels the basis depicted in \eqref{basis_morphspace} that is dual to the one used in \eqref{npbasis}. That \eqref{npbasis_Z} and \eqref{npdualbasis_norm} are dual follows from duality of the bases \eqref{basis_morphspace} of each space $\Hom(U_{p_s}\oti U_{\ia_{s+1}},U_{p_{s-1}})$ and its dual, and that the invariant $\tftm(S^3)$ of empty $S^3$ equals $S_{0,0}$.
\ssubject{Conformal blocks on the torus}
A basis for the space of conformal blocks on the torus is given by $\{|\chii_i;\torus\rangle|i\In\I\}$, where $|\chii_i;\torus\rangle$ is the invariant of a full torus $\Mt i$ with a single annular ribbon, labeled by $U_i$, running along the non-contractible standard cycle of the full torus. Thus
\be\label{char_inv}
   |\chii_i;\torus\rangle=\tftm(\Mt i)\,,
\ee
with $\Mt i$ given by
\eqpic{char_MF}{150}{45}{
\put(0,45)  {$\Mt i:=$}
\put (50,0){\setulen80
\includepic{304}{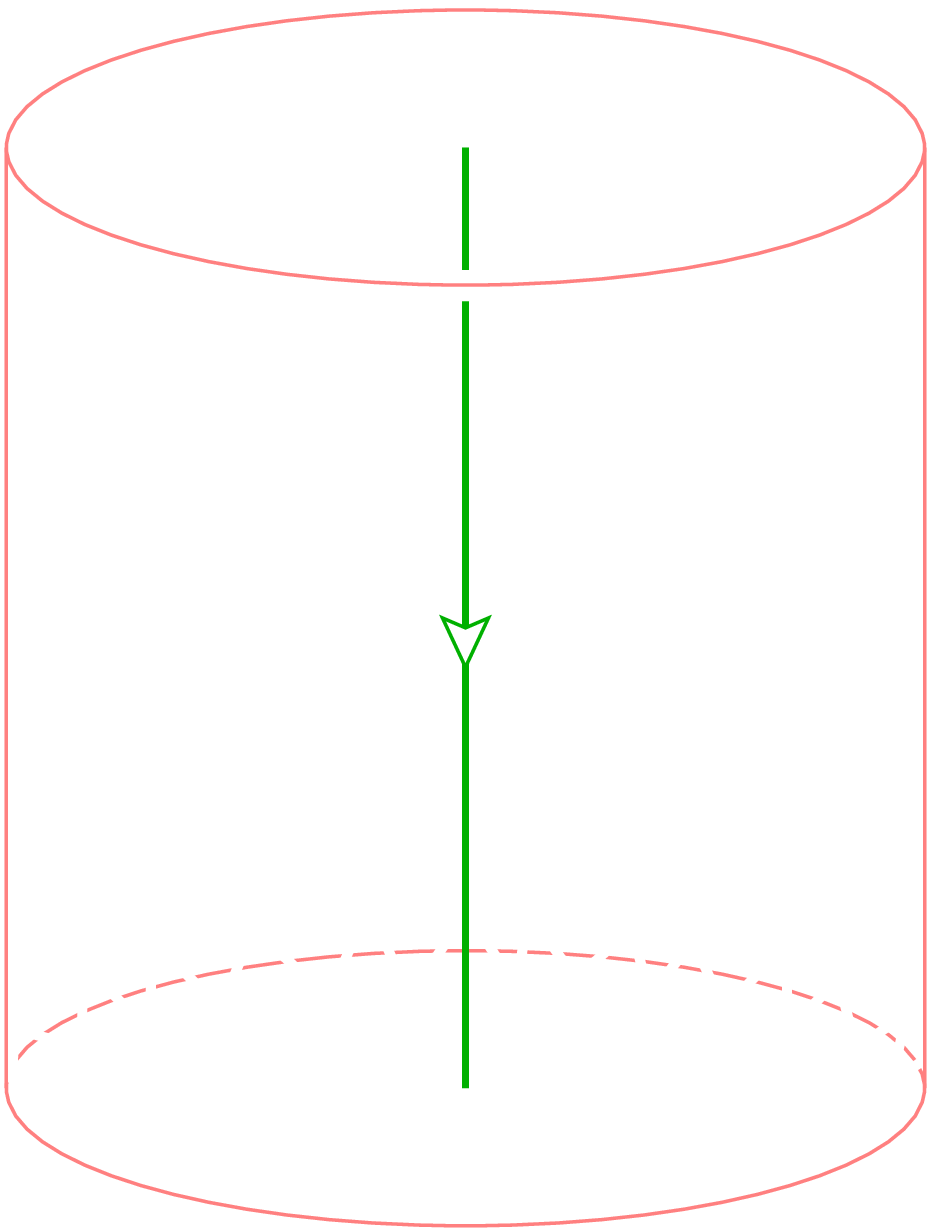}
\put(53.5,59.9) {\scriptsize$i$}
}
}
considered as a cobordism from the empty set to $\torus$, i.e.\ in \eqref{char_inv} top and bottom are identified.
The dual basis is spanned by invariants of similar full tori with opposite orientation of the boundary, see \cite[eq. (5.18)]{fuRs4}.
\subsection{The defect two-point function}
As an illustration and as preparation for the proof of bulk factorization we calculate the structure constants of the \emph{defect two-point function\index{defect!two-point function}}. Denote by  $\wsSD Xpq\alpha\beta$ the world sheet which is a two-sphere with a defect line, labeled by $X$, running from the disorder field $\Thets Xpq\alpha$ to the disorder field $\Thete X\pb\qb\beta$. The world sheet $\wsSD Xpq\alpha\beta$ is displayed in the following picture:
\eqpic{WS_twofieldsdefect}{87}{47}{
\put(0,-7){ \setlength\unitlength{1.2pt}
  \put(0,4)     {\INcludepic{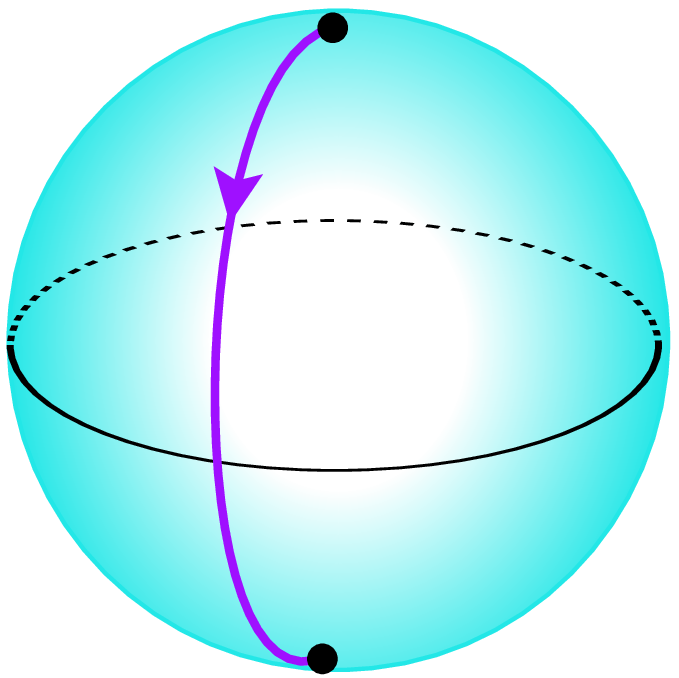}
  \put(30,-9)    {\footnotesize$\Thet X\pb\qb\beta$}
  \put(30,78)    {\footnotesize$\Thet Xpq\alpha$}
  \put(30,56)    {\footnotesize$X$}
  } } }
Following the recipe in section \ref{sec:con_man}, the connecting manifold of $\wsSD Xpq\alpha\beta$ is given by
 \Eqpic{Man_twofieldsdefect}{320}{75}{ \setulen 90
  \put(0,75) {$\MwsSD Xpq\alpha\beta ~=$}
  \put(65,0){
  \put(5,0)   {\includepic{32}{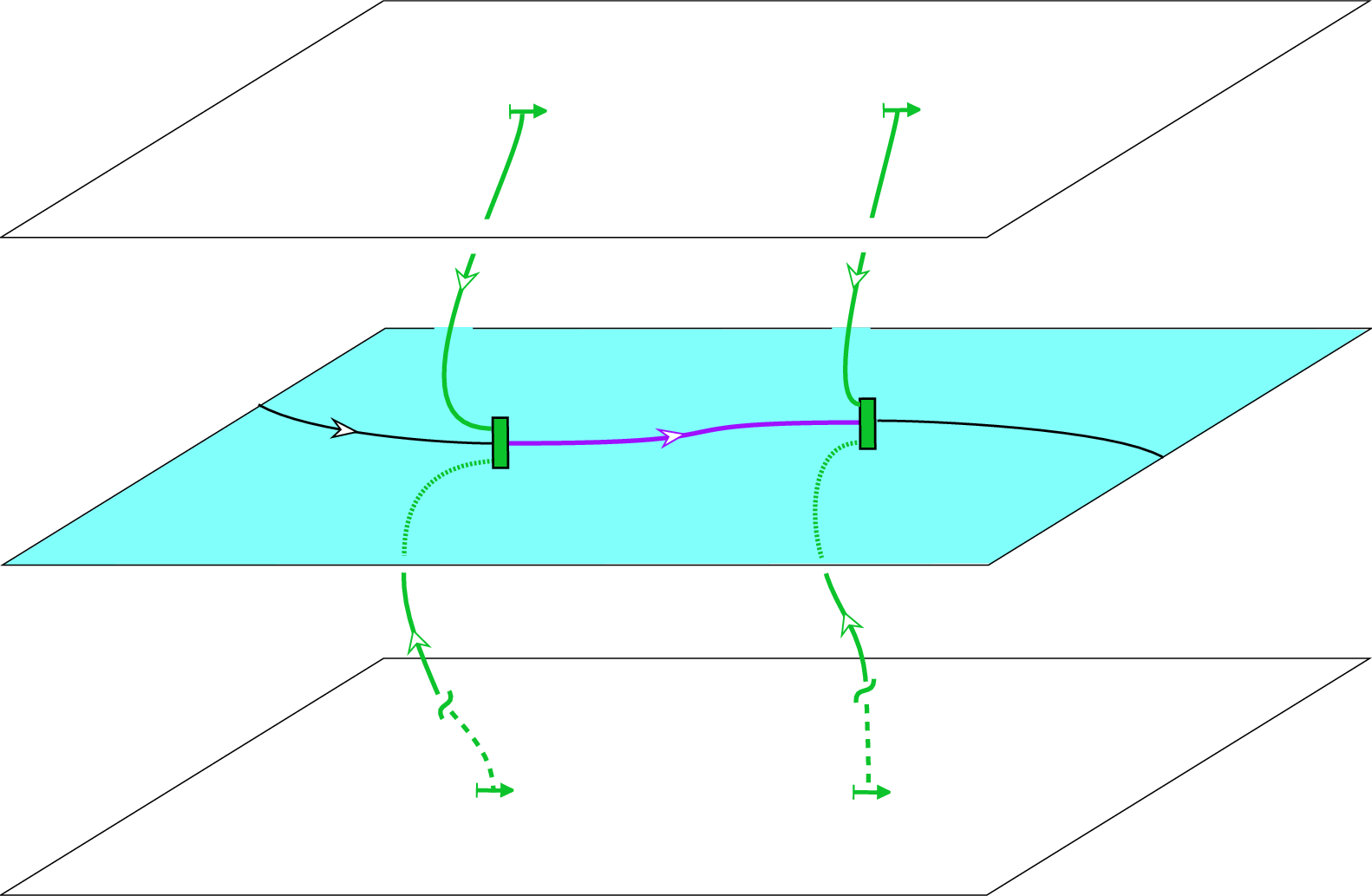}}
  \put(177.8,100){\pl{\phi_{\!\alpha}^{}}}
  \put(103.8,97){\pl{\phi_{\!\beta}^{}}}
  \put(65,83)   {\pl A}
  \put(107,161)   {\pl {\pb}}
  \put(182,161)  {\pl {p}}
  \put(100,13)    {\pl {\qb}}
  \put(174,13)   {\pl {q}}
  \put(137,82){\pl X}
  } }
Here each vertical rectangle represents a two-sphere. The top and bottom rectangle represent the two components of the boundary of the double of ${\wsSD Xpq\alpha\beta}$, while the shaded rectangle represents the world sheet $\wsSD Xpq\alpha\beta$.

The correlator, given by $\Cor(\wsSD Xpq\alpha\beta)=\tftm(\MwsSD Xpq\alpha\beta)$, is an element in the space of conformal blocks on the double of the 2-holed sphere, that is in the tensor product of two conformal blocks of the type \eqref{npbasis_Z} with $n=2$:
\be\label{basis_TPB}
    \Cor(\wsSD Xpq\alpha\beta)\in B(U_p, U_\pb)\oti B^{-}(U_q,U_\qb)\,.
\ee
Here $B(U_p, U_\pb)$ is the space of blocks on the sphere oriented by the outward  pointing normal, whereas $B^{-}(U_q,U_\qb)$ is the space of blocks on the sphere oriented by the inward pointing normal.
The space of conformal blocks on a sphere with two marked points, labeled by simple objects, is zero unless the two points are related by charge conjugation, in which case it is one-dimensional. This is the reason why we only consider pairs of fields with chiral labels related by charge conjugation in any two point function on the sphere.
Using the basis in \eqref{basis_TPB} of $2$-point blocks on the sphere, the structure constants of $\Cor(\wsSD Xpq\alpha\beta)$ are obtained by gluing two copies of \eqref{npdualbasis} with $n=2$ to the manifold \eqref{Man_twofieldsdefect}. Thus, we have
\be
    \Cor(\wsSD Xpq\alpha\beta)~=~\cdef {X}pq\alpha\beta\; B[p,\pb]\oti B^-[q,\qb]\,,
\ee
with
\eqpic{2pdef}{170}{46}{
  \put(0,51) {$\cdef {X}pq\alpha\beta ~=\dsty \frac1{S_{0,0}}$}
  \put(80,-3){ \Includepic{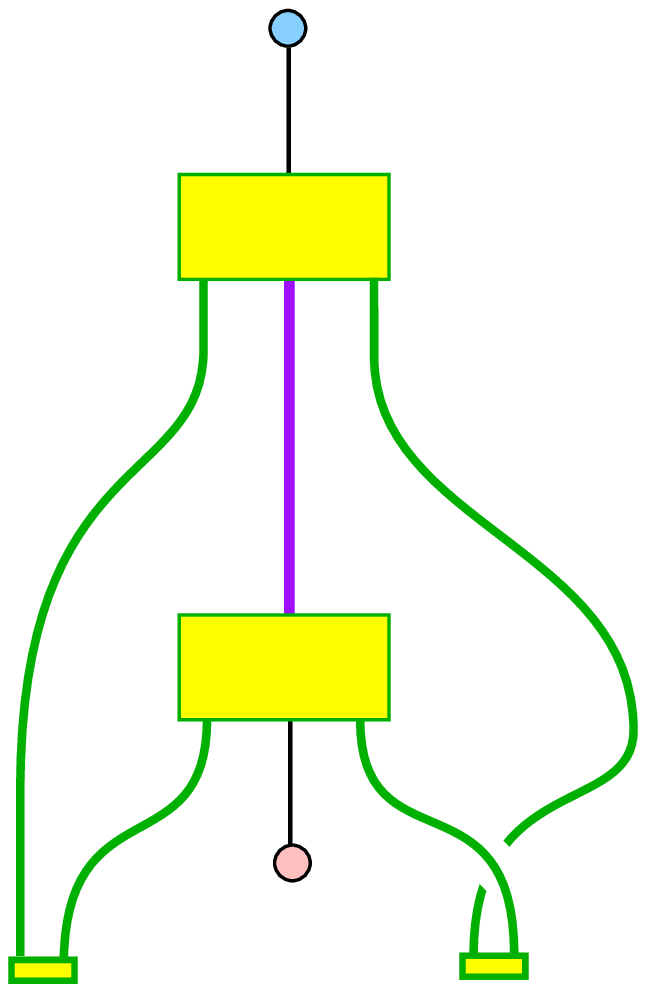}
  \put(24,81)      {\scriptsize $p $}
  \put(24,32)      {\scriptsize $\pb$}
  \put(53,81)      {\scriptsize $q $}
  \put(54,32)      {\scriptsize $\qb$}
  \put(37,42.5)    {\small $\beta$}
  \put(37,91)      {\small $\alpha$}
  \put(43.4,58)    {\small $X$}
} }
Here, the factor $1/S_{0,0}$ is a consequence of the normalization \eqref{npdualbasis_norm} and the invariant of empty $S^3$. To arrive at \eqref{2pdef} we have, after projecting the ribbon graph to $\R^2$, used that $A$ is special and symmetric and that $\phi_\alpha$ and $\phi_\beta$ are bimodule morphisms.
The defect two point function, and in particular the matrix $c^{\mathrm{def}}_{pq}$ consisting of structure constants of the defect two-point function with chiral labels $p,q$, will play a central role when we discuss factorization.
\subsection{Equivalences of world sheets}\label{sec:equiv_ws}
There are many world sheets whose decorations differ slightly, but which still have the same correlator. We will say that two world sheets $\ws$ and $\ws'$ with the same underlying surface $\tws$ are \emph{equivalent}\index{equivalence! of world sheets} iff
\be
    \Cor(\ws)=\Cor(\ws')\,.
\ee
Below we will describe some important equivalences of world sheets, see \cite[Section 3.2]{fjfs} for some more details.
\subject{Intrinsic equivalence}\label{def:intr_eq}
Recall from section \ref{sec:con_man} that, when constructing the correlator, network vertices such that at least one edge is labeled by a trivial defect are labeled by structure morphisms of \boundA s and their (bi-)modules. Two world sheets $\ws$ and $\ws'$ are \emph{intrinsically equivalent} if $\tws\eq\tws'$ and the dual triangulations differ only in the edges labeled by trivial bimodules and the network vertices involving at least one \boundA. The physical interpretation of this statement is  that any modification of a world sheet involving only trivial defects leaves the correlators invariant.
The key ingredient in the proof of intrinsic equivalence is the following lemma, see \cite[Lemma A.1 and A.2]{fjfs}:
\begin{lemma}\label{lemma:oneA_morph}
Let $A$ be a special symmetric Frobenius algebra. Denote by $A^{\eps}$, with $\eps\In\{\pm1\}$, the algebra $A$ if $\eps=+1$, and $A^\vee$ if $\eps=-1$. Consider the space
\be
    R:=\Hom(A^{\eps_1}\oti A^{\eps_2}\oti\cdots\oti A^{\eps_p},A^{\eps_1}\oti A^{\eps_2}\oti\cdots\oti A^{\eps_q})\,,
\ee
where $\eps_i\In\{\pm1\}$. Let $f\In R$ be a morphism composed of structure morphisms of $A$ such that when written graphically, $f$ is a connected graph. There is a unique such morphism and it can be written as
\be
    f=(g^{\eps_1}\oti g^{\eps_2}\oti\cdots\oti g^{\eps_p})\circ\psi\circ\varphi\circ(h^{\eps_1}\oti h^{\eps_2}\oti\cdots\oti h^{\eps_q})\,,
\ee
where $h^{+1}=\id_A=g^{+1}$, $g^{-1}=\Phi\In\Hom(A,A^\vee)$ and $h^{-1}=\Phi^{-1}$, with $\Phi$ defined in \eqref{def_Phi}, and $\varphi$ and $\psi$ the following morphisms:
 \Eqpic{frob_p5}{320}{38}{
  \put(10,38)  {$\varphi:=$}
  \put(50,-7){
  \put(10,10)	{\includepic{35}{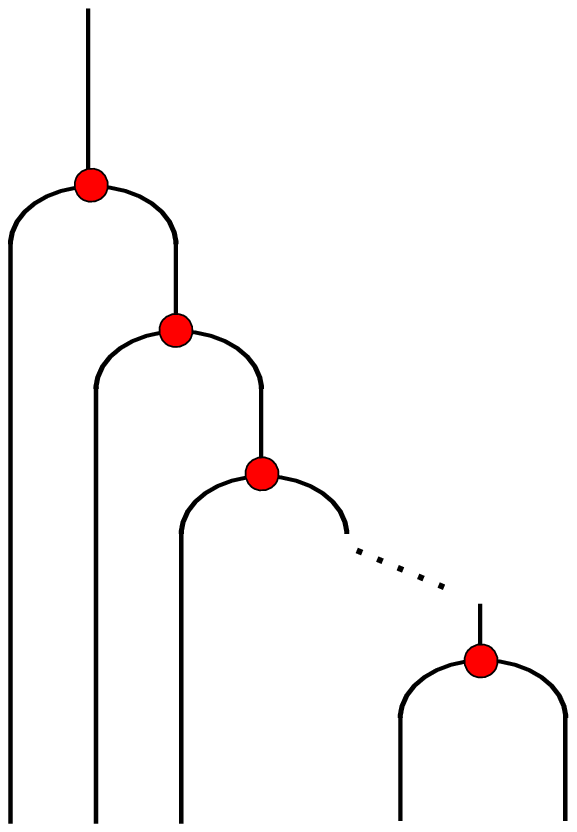}}
  \put(5,0)	{\footnotesize $A$}
  \put(14,0)	{\footnotesize $A$}
  \put(23,0)	{\footnotesize $A$}
  \put(32,0)	{\footnotesize $\cdots$}
  \put(45,0)	{\footnotesize $A$}
  \put(61,0)	{\footnotesize $A$}
  \put(14,98)	{\footnotesize $A$}
  }
  \put(160,38)    {and}
  \put(226,38)  {$\psi:=$}
  \put(262,-7){
  \put(10,0)	{\includepic{35}{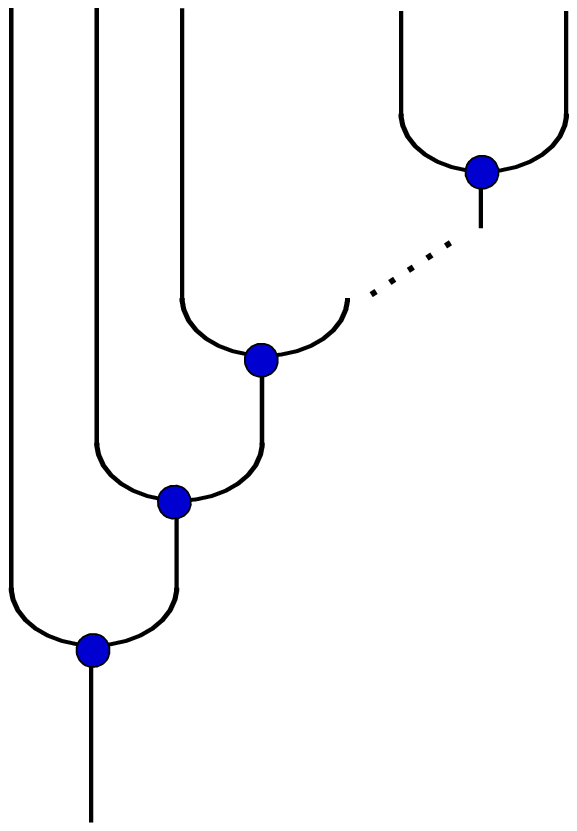}}
  \put(5,89)	{\footnotesize $A$}
  \put(14,89)	{\footnotesize $A$}
  \put(23,89)	{\footnotesize $A$}
  \put(32,89)	{\footnotesize $\cdots$}
  \put(45,89)	{\footnotesize $A$}
  \put(61,89)	{\footnotesize $A$}
  \put(13,-11)	{\footnotesize $A$}
  } }
\end{lemma}
\begin{remark}
The notion of intrinsic equivalence used here differs from the one in \cite{fjfs}, where intrinsic equivalence only refers to the auxiliary datum of a \emph{defect graph} used in that paper. Here, intrinsic equivalence is a statement about the world sheets and thus includes the additional requirement $\tws\eq\tws'$, c.f.\ \cite[Lemma 3.4]{fjfs}.
\end{remark}
\subject{Fusion of defect lines}
As already mentioned, defect lines may be fused, either with another defect line or with a boundary component. Let us make this concrete in the TFT-construction: Assume that $\ws'$ and $\ws$ coincide outside the two regions $D$ and $D'$ displayed in the following picture:
 \Eqpic{locfusion}{320}{46}{ \setulen80
  \put(20,-12){
  \put(0,0)	{\includepic{304}{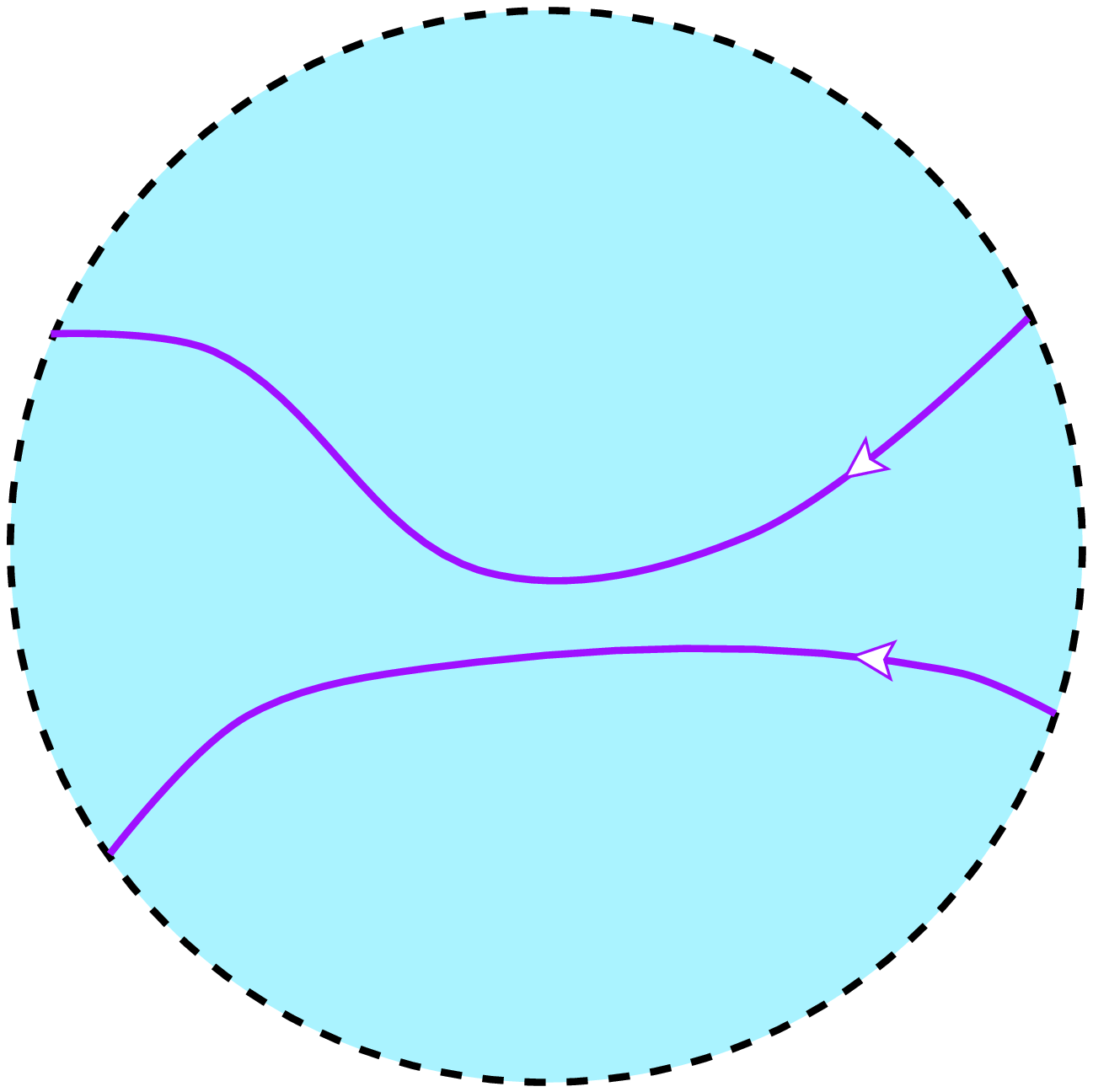}}
  \put(-11,130)	{\fbox{$D$}}
  \put(73,71)	{\footnotesize $Y$}
  \put(93,61)	{\footnotesize $X$}}
  \put(0,-12){
  \put(250,0)	{\includepic{304}{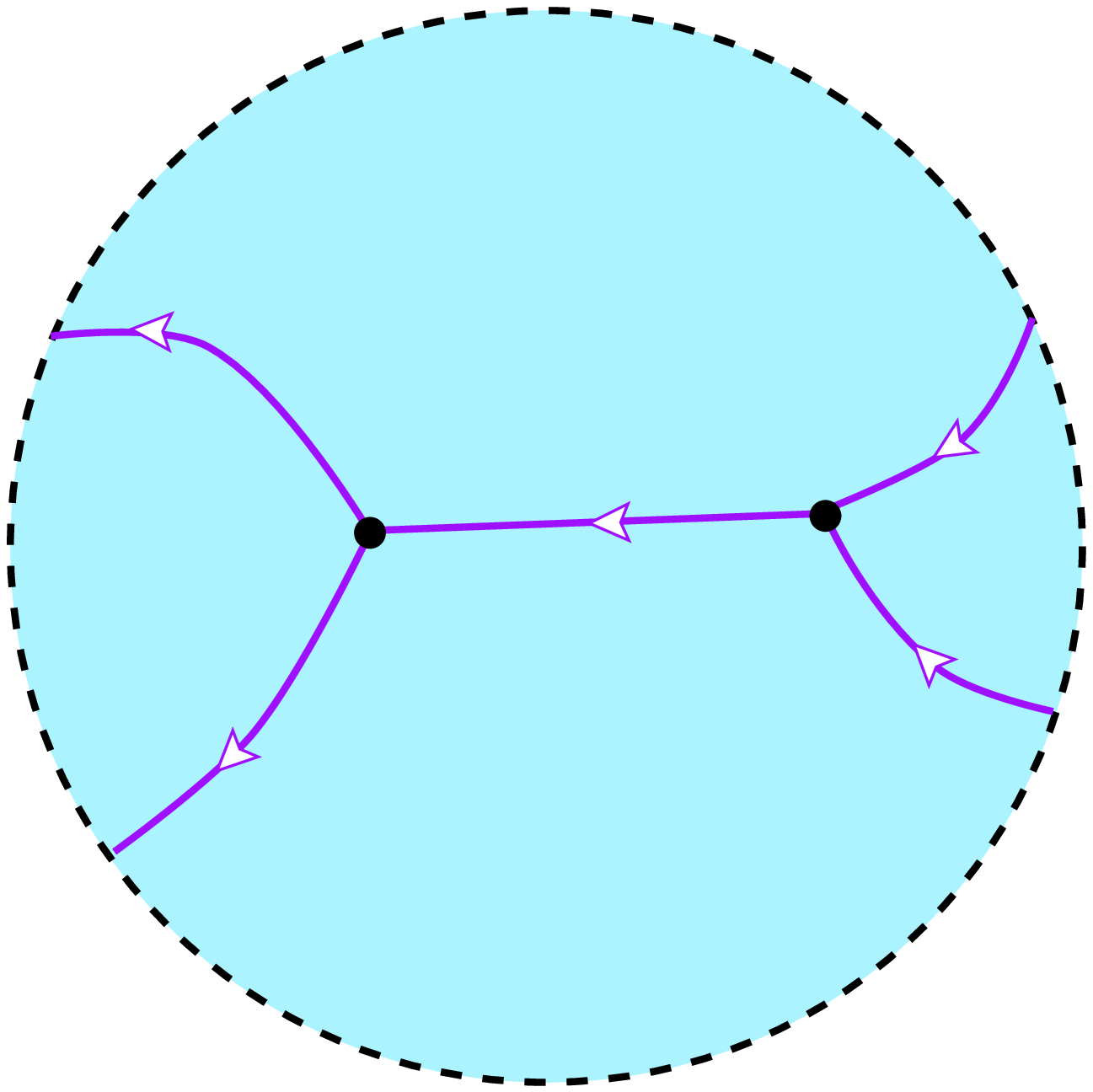}}
  \put(239,130)	{\fbox{$D'$}}
  \put(244,99)	{\footnotesize $Y$}
  \put(252,25)	{\footnotesize $X$}
  \put(386,102)	{\footnotesize $Y$}
  \put(390,45)	{\footnotesize $X$}
  \put(310,79)	{\footnotesize $X\otiB Y$}
  \put(298,64)	{\footnotesize $r$}
  \put(348,66)	{\footnotesize $e$}
  } }
Here $X$ and $Y$ are $A$-$B$- and $B$-$C$-bimodules, respectively. If the morphisms $e$ and $r$, labeling the two network vertices in $D'$, are the embedding and restriction morphisms in \eqref{TP_dec}, the two world sheets are equivalent.

Analogously a defect line labeled by an $A$-$B$-bimodule $X$ may be fused with a boundary component labeled by some left $B$-module $M$. This equivalence is in fact covered by defect fusion, by considering $M$ as an $A$-$\one$-defect.
\subject{Removal of internal vertices}
Consider a defect line connecting two network vertices $v$ and $w$ of defect lines. Such a defect line can be removed\index{removing internal vertex} provided that the remaining defect lines are joined at a single vertex $v'$. To make this statement more concrete, assume that two world sheets $\ws$ and $\ws'$ coincide outside the following two regions:
  \Eqpic{removeedgefig}{320}{53}{\setulen80
  \put(20,0){
  \put(0,0)	{\includepic{304}{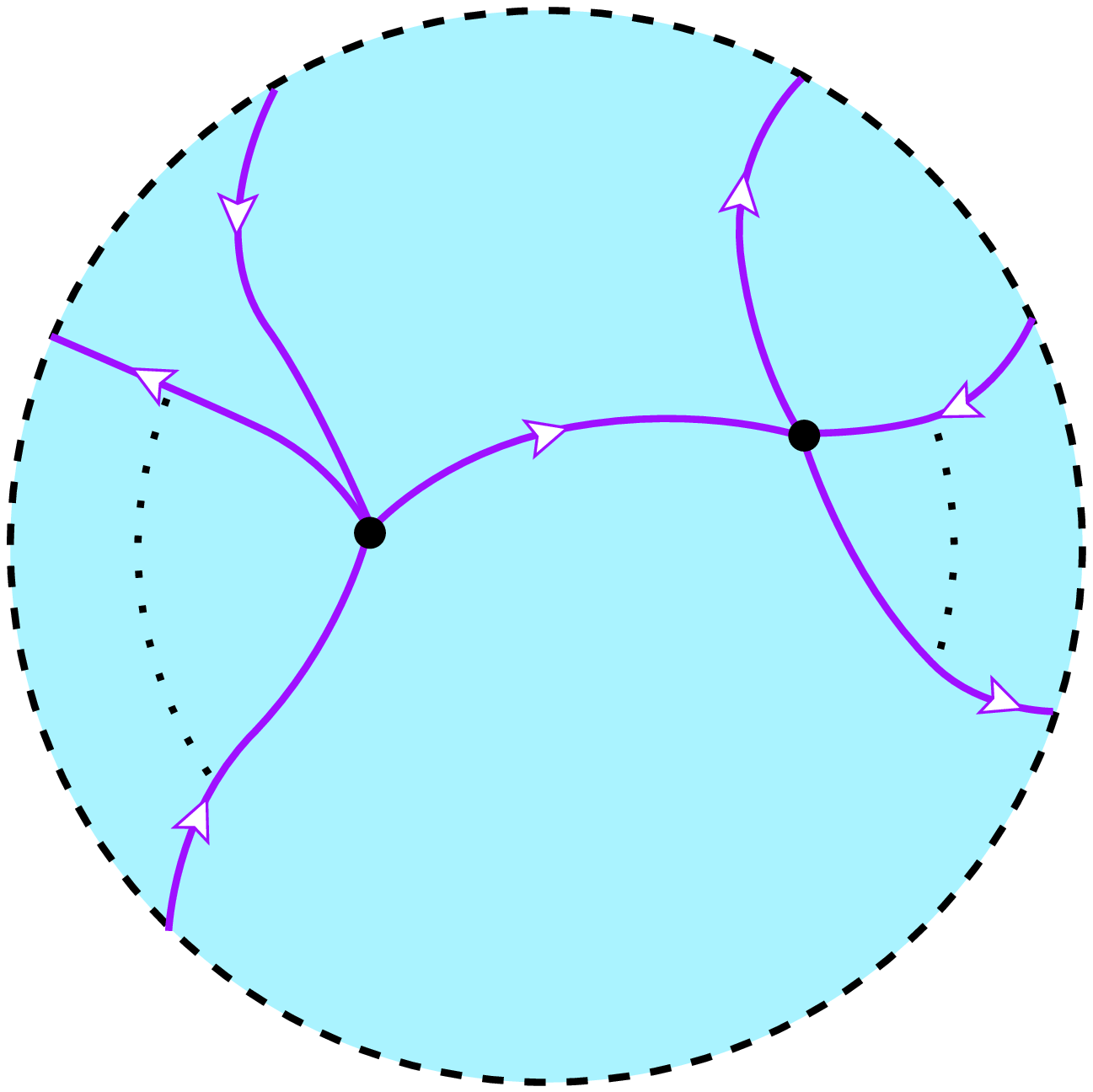}}
  \put(-11,130)	{\fbox{$D$}}
  \put(30,138)	{\footnotesize $X^v_1$}
  \put(-7,100)	{\footnotesize $X^v_2$}
  \put(7,15)	{\footnotesize $X^v_k$}
  \put(47,80)	{\footnotesize $v$}
  \put(70,93)	{\footnotesize $X_{vw}$}
  \put(103,89)	{\footnotesize $w$}
  \put(139,45)  {\footnotesize $X_1^{w}$}
  \put(135,103)	{\footnotesize $X_{m-1}^{w}$}
  \put(105,135)	{\footnotesize $X_m^{w}$}}
  \put(270,0)	{\includepic{304}{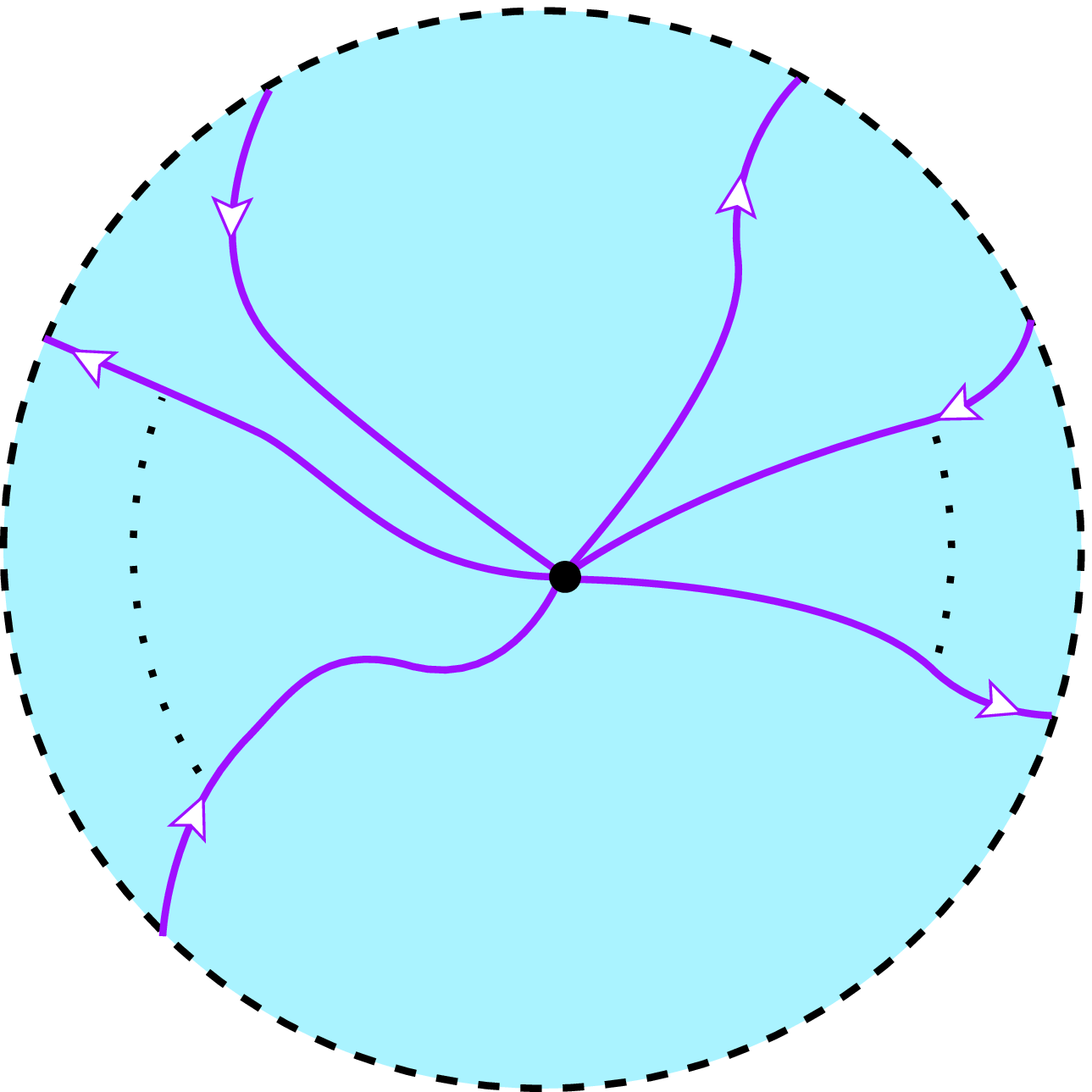}}
  \put(252,130)	{\fbox{$D'$}}
  \put(287,138)	{\footnotesize $X_{m+1}'$}
  \put(248,100)	{\footnotesize $X_{m+2}'$}
  \put(271,15)	{\footnotesize $X_{k+m}'$}
  \put(409,45)	{\footnotesize $X_1'$}
  \put(405,103)	{\footnotesize $X_{m-1}'$}
  \put(375,135)	{\footnotesize $X_{m}'$}
  \put(339,73)	{\footnotesize ${v'}$}
  }
The world sheet $\ws'$ is equivalent to $\ws$ if the vertex $v'$ is labeled by a suitable composition of the two morphisms labeling the vertices $v$ and $w$. Recalling Remark \ref{rem_fus_vert} this statement is almost a tautology.

Similarly, if a defect line $X$ is joining a defect field and a two-valent network vertex, then there is an equivalent vertex, where the defect line together with the network vertex and defect field is replaced by a single defect field. Again the labeling of the latter is a composition of the morphisms labeling the original configuration.

A boundary component\label{rem_bond} stretching between two network vertices can be removed\index{removing boundary component} by completely analogous manipulations.
\subject{Collapsing bubbles}
\index{collapsing defect bubble}A defect configuration forming a "defect bubble" involving a single defect field may be replaced by a single defect field. Assume again that $\ws$ and $\ws'$ coincide outside the following regions:
 \Eqpic{defbubblefig}{320}{54}{ \setulen80
  \put(20,0)	{\includepic{304}{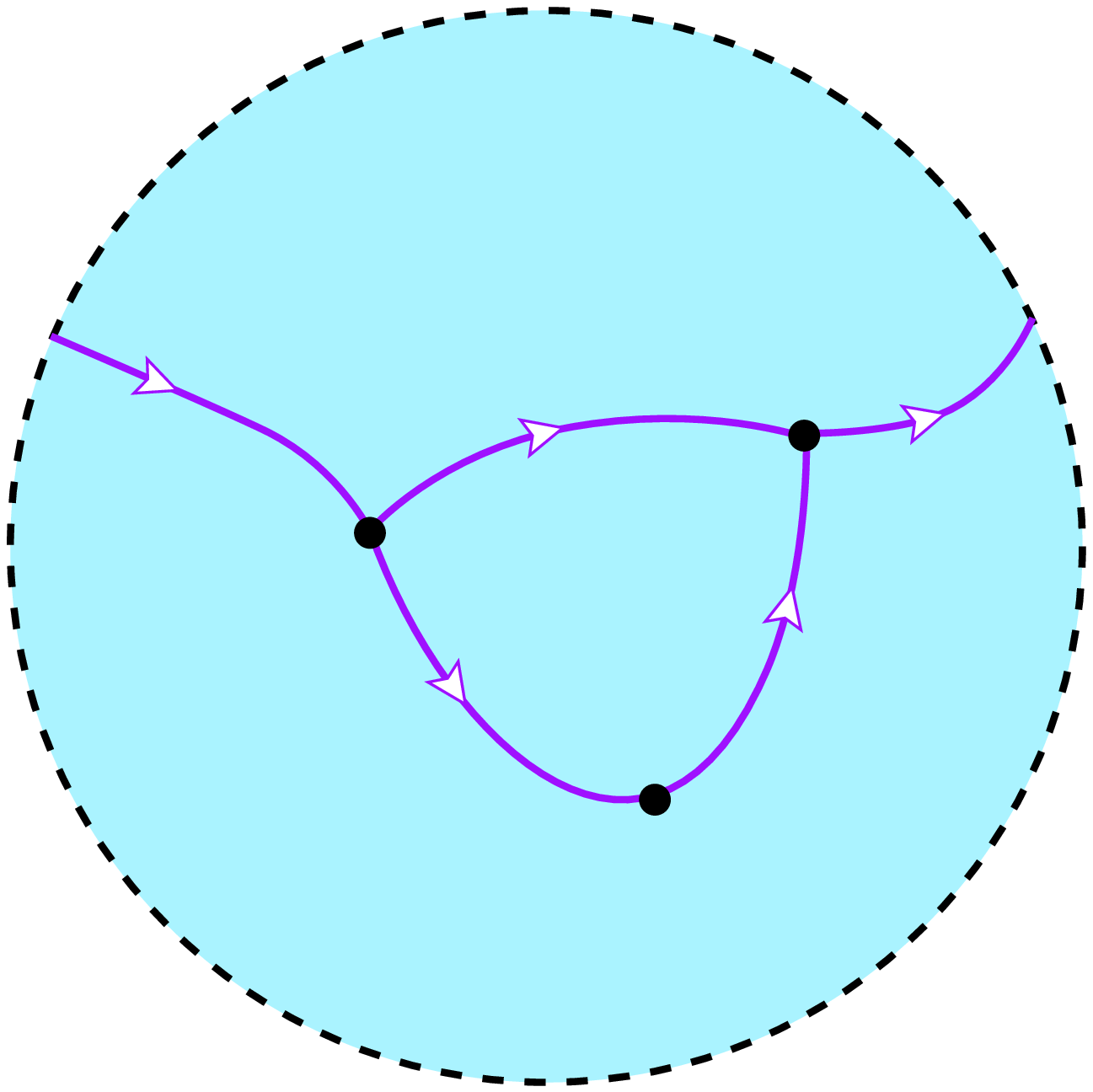}}
  \put(1,130)	{\fbox{$D$}}
  \put(-6,100)	{\footnotesize $X_{1}$}
  \put(65,93)	{\footnotesize $X_{21}$}
  \put(103,56)	{\footnotesize $X_{23}$}
  \put(39,53)	{\footnotesize $X_{31}$}
  \put(135,103)	{\footnotesize $X_{2}$}
  \put(44,80)	{\footnotesize $v_1$}
  \put(100,90)	{\footnotesize $v_2$}
  \put(81,29)	{\footnotesize $v_3$}
  \put(260,0)	{\includepic{304}{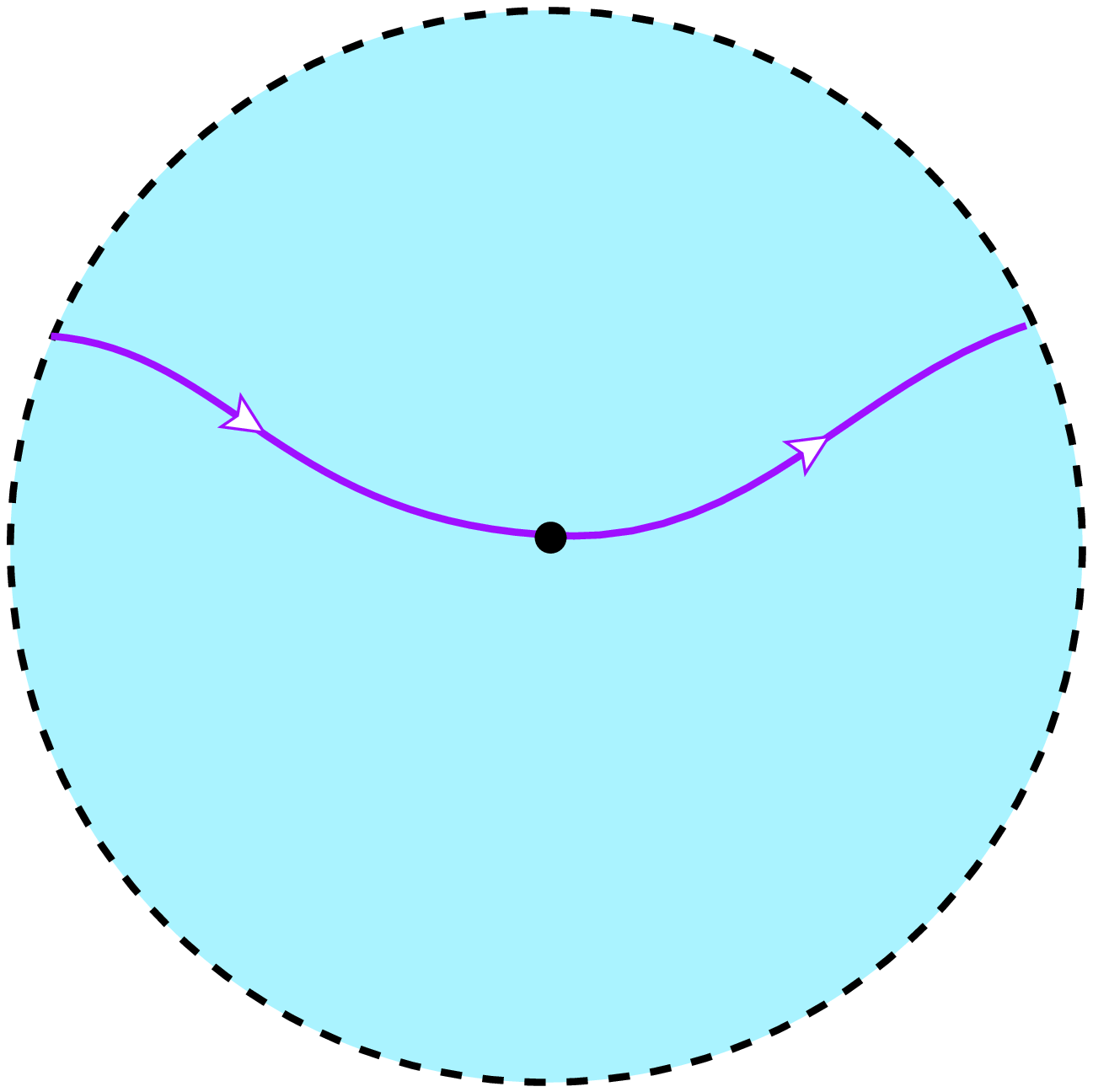}}
  \put(249,130)	{\fbox{$D'$}}
  \put(254,100)	{\footnotesize $X_{1}$}
  \put(395,103)	{\footnotesize $X_{2}$}
  \put(326.6,76){\footnotesize ${w}$}
  }
Here, $X_1$ and $X_2$ are $A$-$C$-bimodules, $X_{23}$ and $X_{31}$ are $B$-$C$ bimodules and $X_{21}$ is an $A$-$B$ bimodule.
Let the vertices in $D$ be labeled by $\kappa_{v_1}\In\HomP_{A|C}(X_{21}\linebreak{\otiB} X_{31},X_1)$, $\kappa_{v_2}\In\HomP_{A|C}(X_2,X_{21}\oti_{\!B\,} X_{23})$ and $\phi_{v_3}\In\Hombc{U_i\otip X_{23}\linebreak{\otim} U_j}{X_{31}}$, respectively. The world sheets are equivalent if $w$ in $D'$ is labeled by the morphism
\eqpic{collapsebubble}{140}{68}{
  \put(0,70)   {$\phi_w~=$}
  \put(50,10)	{\Includepic{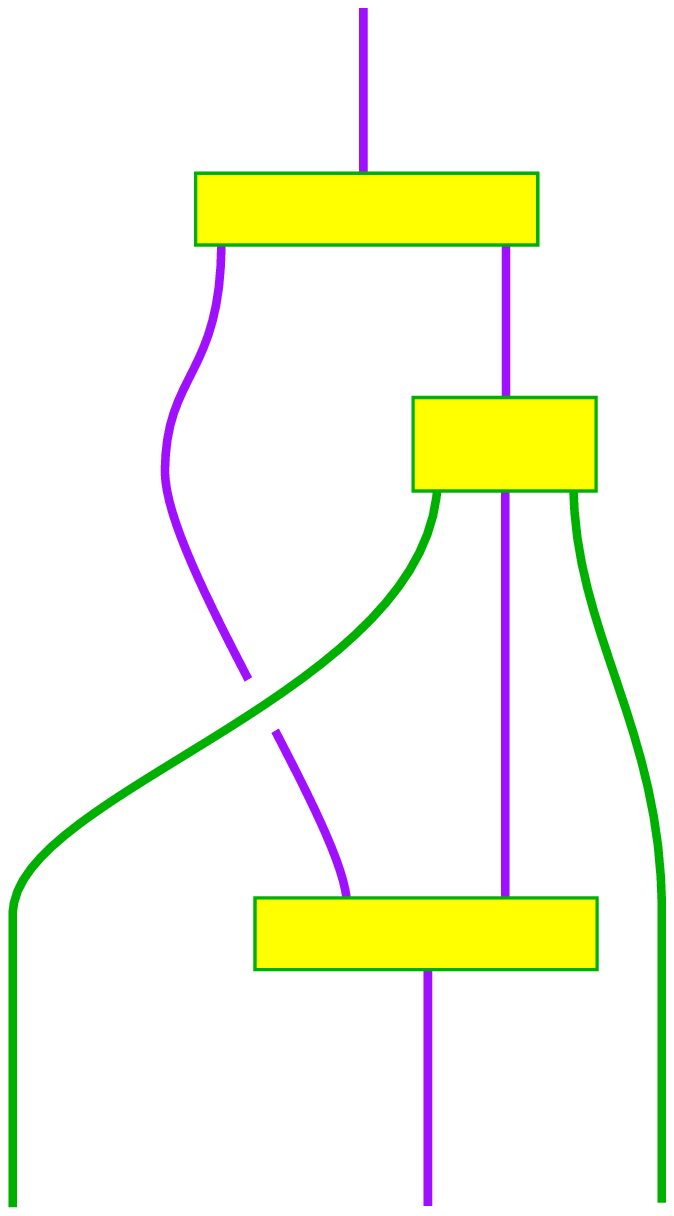}
  \put(2,-10)	{\footnotesize $U_i$}
  \put(72,-10)	{\footnotesize $U_j$}
  \put(45,-10)	{\footnotesize $X_{2}$}
  \put(51.5,45)	{\begin{turn}{90}\footnotesize $X_{23}$\end{turn}}
  \put(44,29.5)	{\footnotesize $\varkappa_{v_2}$}
  \put(55,82)	{\footnotesize $\phi_{v_3}$}
  \put(10,68)	{\footnotesize $X_{21}$}
  \put(40,109)	{\footnotesize $\varkappa_{v_1}$}
  \put(40,137)	{\footnotesize $X_{1}$}
  \put(63,96)	{\footnotesize $X_{31}$}
  } }
The morphism $\phi_w$ is the same as the one obtained from the region $D$. Note that $\phi_w\In\Homac{U_i\oti^+X_2\oti^-U_j}{X_1}$ and is accordingly an adequate morphism to label the defect field at $w$.

There is again an analogous result involving boundaries: Consider the situation that two end-points of a defect without insertions end at trivalent network vertices on the same boundary component, such that at most one boundary field is located in between the two network vertices. In this situation the bubble can be collapsed\index{collapsing boundary bubble}, giving rise to just a boundary field by arguments analogous to the ones above. This situation is treated
analogously as the "defect bubble" case involving only defects by considering the case $C=\one$.

\section{Constraints in the TFT-construction}\label{sec:TFT_constr}
In this section we outline the proofs in \cite{fjfs} that mapping class group invariance and the factorization constraints are satisfied by the correlators, obtained via the TFT-construction, of oriented world sheets containing defects. We will describe bulk factorization in the presence of defects in some detail, since bulk factorization across defect lines will be a crucial ingredient in the next chapter. Boundary factorization across defect lines on the other hand, is not new. It is covered by the results in \cite{fjfrs} by using fusion of defect lines with a boundary component. Thus we will only briefly discuss boundary factorization.
\subsection{Bulk factorization}\label{sec:bulkfact}
Bulk factorization associates to a world sheet $\ws$ a collection $\{\ws'\}$ of world sheets, with topology different from $\ws$. The factorization constraints require that the correlator of \ws\ can be written as a linear combination in which we sum over the correlators of the world sheets in $\{\ws'\}$. Bulk factorization takes place along a circle $\Scut$, the \emph{cutting circle}\index{cutting!circle}, embedded in a cylindrical region of \ws. We consider the situation that the circle $\Scut$ is crossed by a single $A$-$A$-defect $X$. By using fusion of defects this situation covers also the case that there is a finite number of defects running parallel to each other and crossing the circle $\Scut$.

\subject{Factorization of the world sheet}
In order to obtain the world sheets $\ws'$ related to \ws\ via factorization we proceed in two steps:
\begin{enumerate}
    \item Cut \ws\ along $\Scut$. This creates two new circular boundary components.
    \item Cut the world sheet $\wsSD Xpq\alpha\beta$ of the defect two-point function, displayed in \eqref{WS_twofieldsdefect}, along the equator. The so-obtained half-spheres are glued to the circular boundary components created in step 1 such that two-orientations and core orientations of defect lines match.
\end{enumerate}
The procedure is illustrated in the following picture (compare with the picture \eqref{WS_twofieldsdefect} of $\wsSD Xpq\alpha\beta$):
\Eqpic{dis_sphere}{320}{20}{\setulen 91
  \put(0,10){
  \put(0,0)     {\includepic{34}{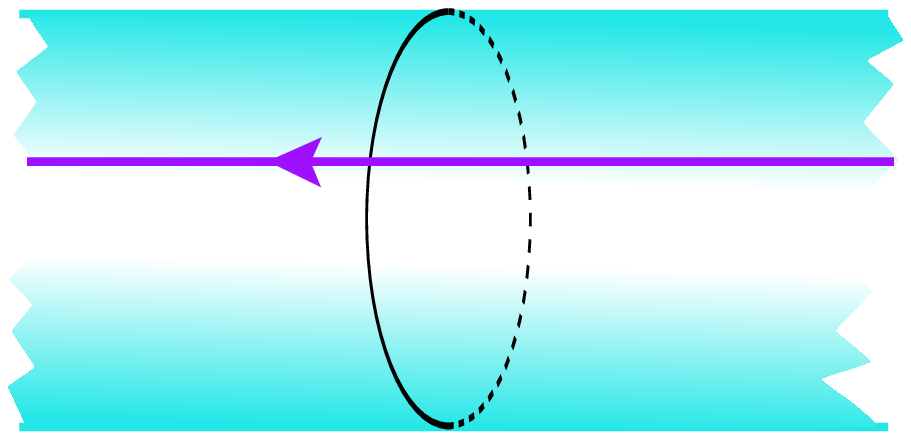}}
  \put(48,7)   {\begin{turn}{90}\scriptsize$\Scut$\end{turn}}
  \put(128,20)  {$\longmapsto$}
  \put(172,0)   {\includepic{34}{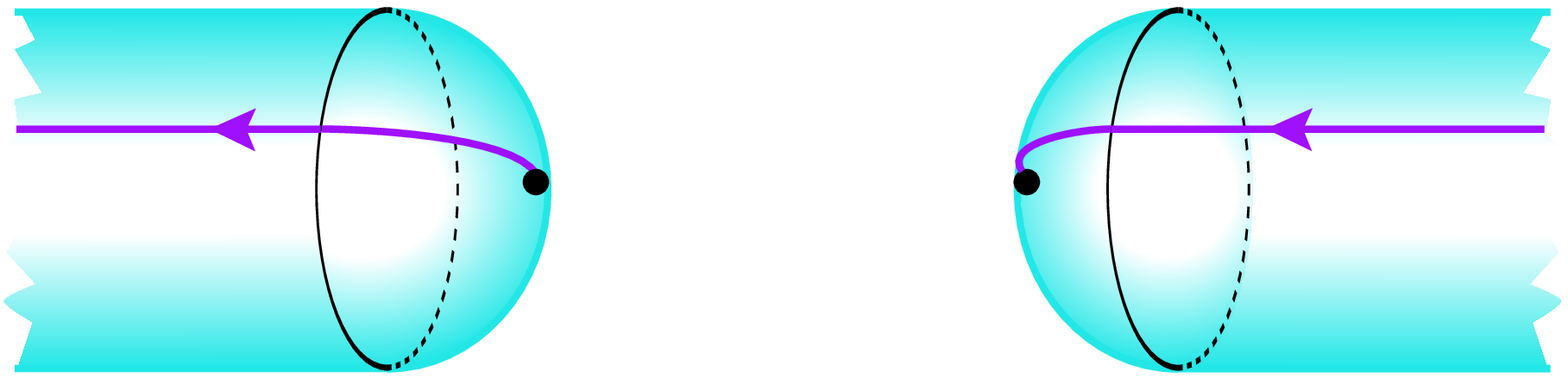}}
  \put(77,33)   {\scriptsize$X$}
  \put(187,33)  {\scriptsize$X$}
  \put(347,33)  {\scriptsize$X$}
  \put(245,22)  {\scriptsize$\Thets Xpq\alpha$}
  \put(276,22)  {\scriptsize$\Thete X\pb\qb\beta$}
  } }
This way a collection $\{\ws'\}=\{\wsf pq\alpha\beta\}$ of world sheets is obtained. There is one world sheet for each pair $(\Thets Xpq\alpha,\Thete X\pb\qb\beta)$ of disorder fields having non-zero defect two-point function on the sphere. We will refer to the world sheets in the collection $\{\wsf pq\alpha\beta\}$ as the \emph{factorized world sheets}.

\subject{The gluing homomorphism}
Note that $\wsf pq\alpha\beta$ differs from \ws\ in two respects. Factorization changes the topology of the world sheet and increases the number of marked points. Consequently, the two spaces $\tfts(\widehat\ws)$ and $\tfts(\widehat{\wsf pq\alpha\beta})$ of conformal blocks do not coincide. Thus, $\Cor(\ws)$ cannot immediately be expanded in terms of the correlators $\Cor(\wsf pq\alpha\beta)$. However, there is a linear map which we call \emph{gluing homomorphism}
\be\label{defGLL}
  \GLL pq:\quad \tfts(\Hatwsf pq) \to \tfts(\widehat\ws) \,
\ee
that relates the two spaces. We will write $\Hatwsf pq\equiv\widehat{\wsf pq\alpha\beta}$ since the double of $\wsf pq\alpha\beta$, and accordingly the associated space of conformal blocks, does not depend on the multiplicity labels $\alpha$ and $\beta$. A gluing homomorphism was obtained already in \cite{fjfrs}, where factorization across a trivial defect was proven, see \cite[eq. (2.49)]{fjfrs}. This is the relevant gluing homomorphism also in the case of non-trivial defects. The reason is that what is relevant for the gluing homomorphism is the double of the world sheet, and the double of $\wsf pq\alpha\beta$ indeed coincides with the double of the cut world sheet considered in \cite{fjfrs}. However, we will use a slightly different normalization of the gluing homomorphism than the one considered in \cite{fjfrs}. The construction of $\GLL pq$ uses a basis $\{\bar\lambda^{k\kb}\}$ of the one-dimensional space $\Hom(\one,U_k\oti U_\kb)$ for each $k\in\I$. We will use
\eqpic{lambdachoice}{240}{21}{ \put(0,-9){
  \put(15,0){ \Includepic{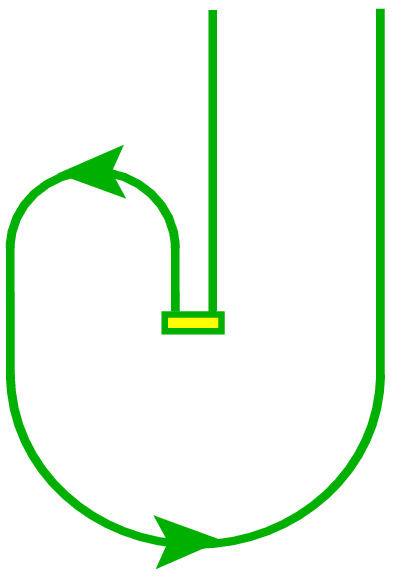}
  \put(21,66)        {\scriptsize $k $}
  \put(39,66)        {\scriptsize $\kb $}
  \put(14,17)        {\scriptsize $\bar\lambda^{\kb k}$}
  }
  \put(80,34) {$=~ \theta_k\;\RR {\kb} k0\nl\nl \,\bar\lambda^{\kb k}
               ~=~ \theta_k^{-1}\,\Rm {\kb} k0\nl\nl \,\bar\lambda^{\kb k} $}
  } }
instead of $\bar\lambda^{k\kb}$, c.f. Remark 2.4 of \cite{fjfs}.

\subject{The factorization identity}
It is easy to see that $\GLL pq(\Cor(\wsf pq\alpha\beta))$ and $\Cor(\ws)$ lie in the same space. In fact, there is
\cite[Definition 5.1.13(iv)]{BAki} an isomorphism
  \be
  \bigoplus_{p,q} \GLL pq:\quad
  \bigoplus_{p,q} \tfts(\Hatwsf pq)\stackrel{\cong}{\longrightarrow} \tfts(\widehat\ws)
  \ee
of vector spaces. However, this does not guarantee the existence of a factorization identity, since the collection $\{\Cor(\wsf pq\alpha\beta)\}$ of correlators is far from a basis of $\bigoplus_{p,q} \tfts(\Hatwsf pq)$. Nonetheless, we have indeed the following statement from \cite{fjfs} about bulk factorization\index{factorization!bulk}:
\begin{thm}\label{thm:bulkfac}
The correlator $C(\ws)$ for an oriented world sheet $\ws$ can be expressed in
terms of the correlators $C(\wsf pq\alpha\beta)$ of the factorized world sheets as
  \be\label{fact_rel}
  \Cor(\ws) = \sum_{p,q\in\I}\;\sum_{\alpha\in H^{A,X}_{p,q}}
  \sum_{\beta\in H^{X,A}_{\pb,\qb}}\, \Cfact Xpq\alpha\beta\; \GLL pq(\Cor(\wsf pq\alpha\beta)) \,,
  \ee
with $\Cfact Xpq\alpha\beta$ given by
\be
    \Cfact Xpq\alpha\beta=\dim(U_p)\,\dim(U_q)\;
  \cdefinv {X}\pb\qb\alpha\beta\,.
\ee
\end{thm}
The rest of this subsection is devoted to the proof of Theorem \ref{thm:bulkfac}. The proof follows quite closely the proof in \cite{fjfrs} of bulk factorization without defects. Here we will summarize the proof and indicate how the proof in \cite{fjfrs} is modified or generalized. However, there are some technicalities, which will also not be discussed in any detail here, see e.g.\ Lemma A.4 and Lemma A.5 of \cite{fjfs}.
\subsection{Proof of bulk factorization}\label{sec:bulk_proof}
In order to prove the bulk factorization identity \eqref{fact_rel} it is instructive to analyze the action of the gluing homomorphism \eqref{defGLL} in some more detail. The gluing homomorphism is obtained as the invariant of the \emph{gluing cobordism}
\be
  \MGL \equiv \MGL{}_{,pq}\colon\quad \Hatwsf pq \to \widehat\ws \,.
\ee
Thus,
\be
    \GLL pq(\Cor(\wsf pq\alpha\beta))=\tftm(\MGL \circ \M_{\wsf pq\alpha\beta})\,.
\ee
While both are cobordisms from $\emptyset$ to $\tfts(\ws)$, the cobordisms $\M_\ws$ and
$\MGL \circ \M_{\ws'}$ do not coincide, not even as topological manifolds. However, the discrepancy between the two manifolds, including the ribbon graphs, is entirely confined in a suitably embedded solid torus. Explicitly, $\M_\ws$ and $\MGL \circ \M_{\ws'}$ can be realized as
\be\label{ws_torus}
  \M_\ws = \MST \circ \TORS
  \ee
and
  \be\label{ws_fact_torus}
  \MGL \circ \M_{\ws'} = \MST \circ \TOR\,.
  \ee
Here
\be
    \MST:\ExtTdef X\to\wsD
\ee
is a cobordism from $\ExtTdef X$, which is the torus $T^2$ with two marked points labeled by $X$, to the double of the original world sheet. $\TOR$ and $\TORS$ are both cobordisms from $\emptyset$ to $\ExtTdef X$.
However $\TOR$ and $\TORS$ differ with respect to the embedded ribbon graph and, informally, by a modular S-transformation of their boundary, see (5.3) and (5.9) of \cite{fjfrs} and section 2.2 of \cite{fjfs}. To make the latter statement concrete, consider the mapping cylinder $\M_\mathrm{S}$ over $T^2$ of a homeomorphism in the class of the S-transformation of $T^2$. Then $\TORS$ and $\M_\mathrm{S}\circ\TOR$ differ, apart from in the embedded ribbon graphs, by an orientation preserving homeomorphism that restricts to the identity on $\partial\TORS=\partial(\M_\mathrm{S}\circ\TOR)$.
Since the derivation of the cobordisms $\TORS$ and $\TOR=\TOT$ is completely analogous to the derivation of the cobordisms appearing\footnote{In eq.\ (5.3) and (5.9) of \cite{fjfrs}, the cobordisms $\MST \circ \TORS$ and $\MST \circ \TOR$ are displayed for the case $X=A$.} in (5.3) and (5.9) in \cite{fjfrs}, we omit the derivation here and only give the result:
\eqpic{instor}{250}{128}{
  \put(0,128) {$\TOT ~=$}
  \put(80,22){ \Includepic{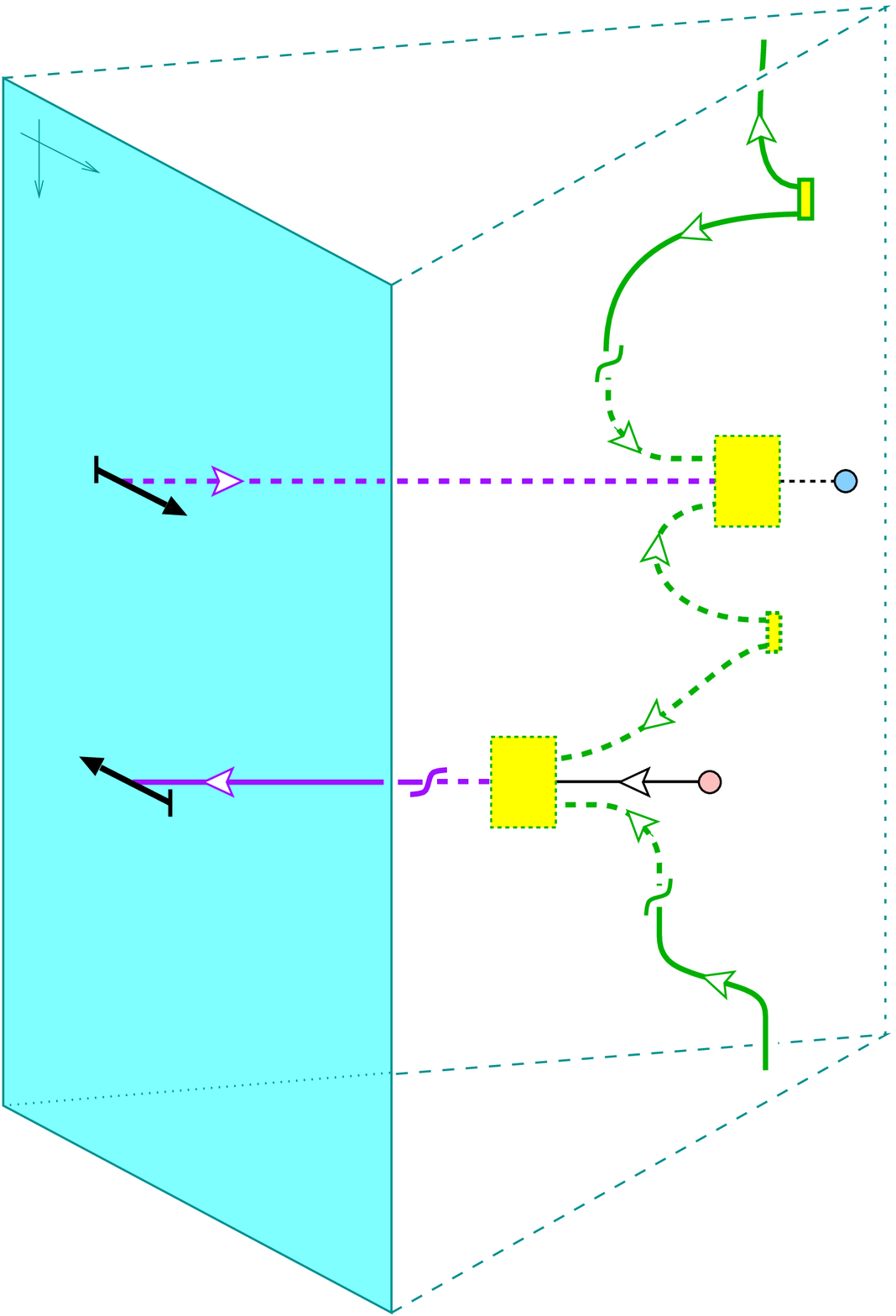}
  \put(19,199)    {\tiny $1$}
  \put(5,194)     {\tiny $2$}
  \put(140,210)   {\scriptsize $q$}
  \put(131,186)   {\scriptsize $\qb$}
  \put(120,80)    {\scriptsize $q$}
  \put(109.7,128) {\scriptsize $\pb$}
  \put(119,111)   {\scriptsize $p$}
  \put(91.2,91.6) {\small $\alpha$}
  \put(130.6,144) {\small $\beta$}
  \put(11,141)    {\begin{turn}{-27}\scriptsize $(X,+)$\end{turn}}
  \put(10,90)     {\begin{turn}{-27}\scriptsize $(X,-)$\end{turn}}
  \put(147,190)   {\tiny \begin{turn}{90}$\bar{\lambda}^{\qb q}$\end{turn}}
  \put(143,115)   {\tiny \begin{turn}{90}$\bar{\lambda}^{\pb p}$\end{turn}}
  } }
  and
  \eqpic{deftor}{240}{103}{
  \put(17,109) {$\TORS ~=$}
  \put(75,0){ \Includepic{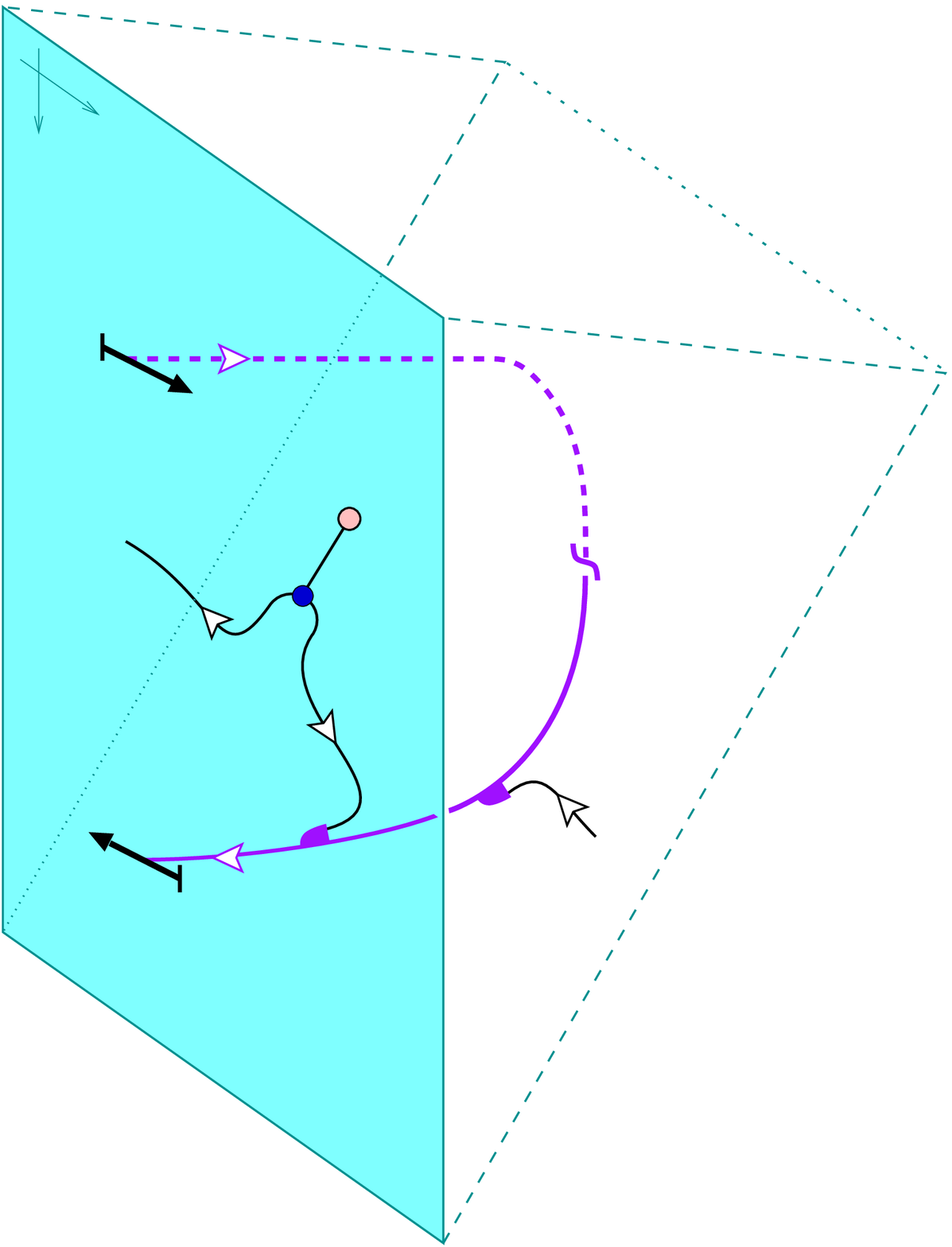}
  \put(18,197)    {\tiny $1$}
  \put(4,192)     {\tiny $2$}
  \put(11,150)    {\begin{turn}{-27}\scriptsize $(X,+)$\end{turn}}
  \put(10,64)     {\begin{turn}{-27}\scriptsize $(X,-)$\end{turn}}
  \put(58,98)     {\scriptsize $A$}
  \put(101,90)    {\scriptsize $X$}
  } }
Here, the cobordisms are displayed in "wedge presentation". This means that top and bottom, as well as horizontal white faces, are to be identified, see \cite[Section 5.1]{fjfrs} for more details. In \eqref{instor}, $\alpha$ and $\beta$ label the bases, defined in \eqref{basis_DF}, of the spaces  $\Homaa{U_p\,\otip\!A\otim U_q}X$ and $\Homaa{U_\pb\,\otip \!X\otim U_\qb}A$ respectively.

Applying the \tftfun\ to the cobordisms \eqref{ws_torus} and \eqref{ws_fact_torus} we get
\be\label{C-tor}
  \Cor(\ws) = \tftm(\MST) \circ \tftm(\TORS)
  \ee
and
  \be\label{Cfact-tor}
  \GLL pq(\Cor(\wsf pq\alpha\beta)) = \tftm(\MGL) \circ \tftm(\M_{\ws'})
  = \tftm(\MST) \circ \tftm(\TOT)\,.
  \ee
for the correlators.
Thus the factorization identity \eqref{fact_rel} amounts to a relation between the invariant of $\TORS$ and the invariants of $\TOT$ for all possible labels of $p,q,\alpha,\beta$. Indeed we have the following result from \cite{fjfs}:
\begin{proposition}\label{surg_prop}
The invariant $Z(\TORS )$ can be expanded as
  \be\label{surg_DF}
  Z(\TORS)=\sum_{p,q\in\I}\,\sum_{\alpha\in H^{A,X}_{p,q}}
  \sum_{\beta\in H^{X,A}_{\pb,\qb}}\, \Cfact Xpq\alpha\beta\; Z(\TOT) \,,
  \ee
where the coefficients $C_{pq,\alpha\beta}$ are given by
  \be\label{coeff_value}
  \Cfact Xpq\alpha\beta =
  \dim(U_p)\;\dim(U_q)\;\cdefinv {X}\pb\qb\alpha\beta \,.
  \ee
\end{proposition}
\begin{proof}
First observe that $ \tftm(\TORS)$ indeed can be expanded in terms of the invariants $\tftm(\TOT)$. Let us motivate this statement. Consider the space $\im(\proj)$, where
\be
    \proj = \tftm(\projMF)
\ee
is the projector obtained as the invariant of the cobordism
\eqpic{CCinv}{240}{93}{\setulen90
  \put(15,110)  {$\projMF~:=$}
  \put(30,0){ \includepic{342}{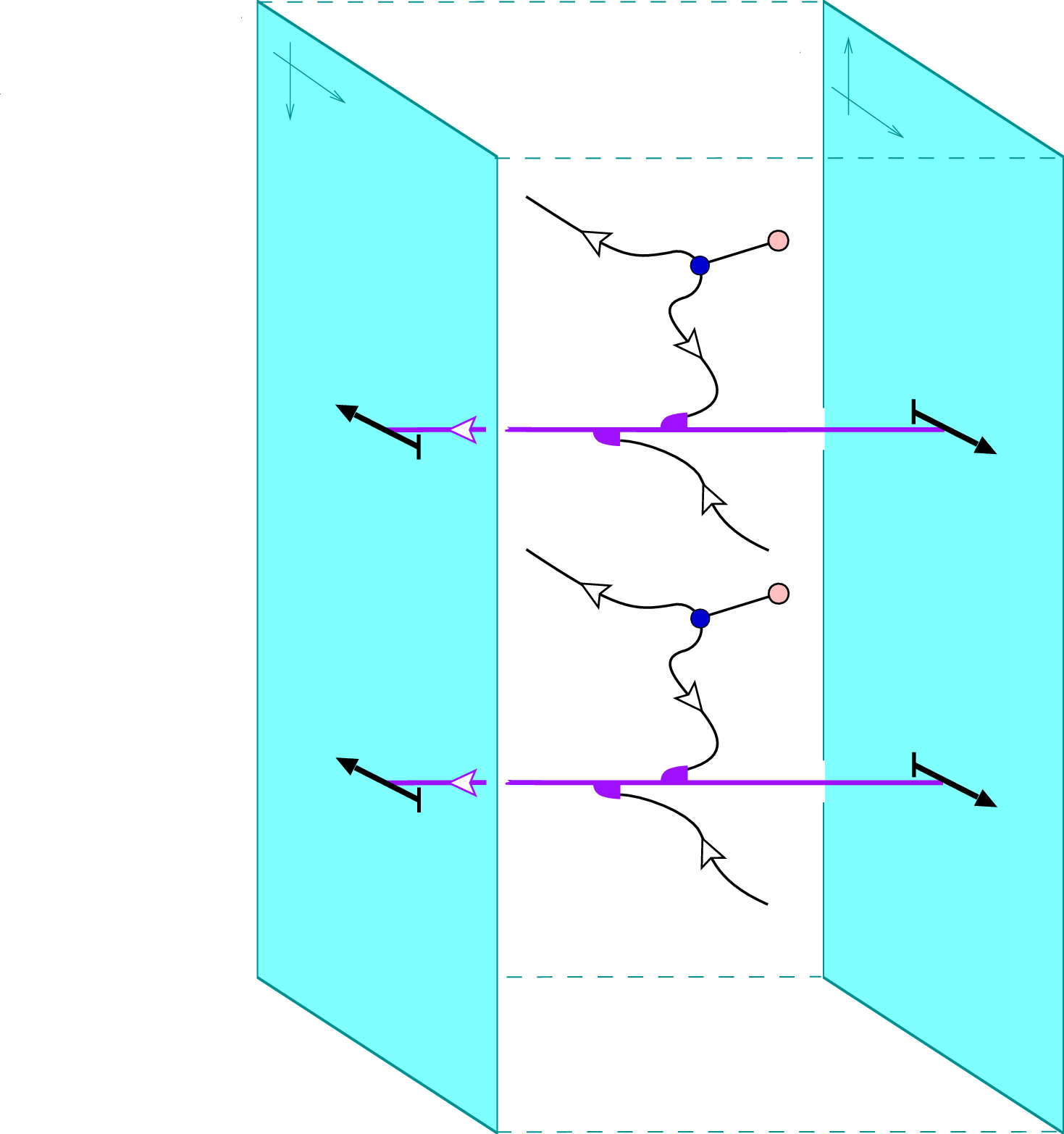}
  \put(65,195)    {\tiny $1$}
  \put(51,192)    {\tiny $2$}
  \put(59,66)     {\begin{turn}{-27}\scriptsize $(X,-)$\end{turn}}
  \put(139,74)    {\scriptsize $A$}
  \put(140,50)    {\scriptsize $A$}
  \put(175,77)    {\begin{turn}{-27}\scriptsize $(X,+)$\end{turn}}
  \put(59,134)    {\begin{turn}{-27}\scriptsize $(X,+)$\end{turn}}
  \put(139,145)   {\scriptsize $A$}
  \put(140,119)   {\scriptsize $A$}
  \put(175,145)   {\begin{turn}{-27}\scriptsize $(X,-)$\end{turn}}
  \put(174,195)   {\tiny $1$}
  \put(166,208)   {\tiny $2$}
  }}
In this picture top and bottom as well as front and back are to be identified, i.e.\ we deal with a cylinder over a torus $\ExtTdef X\eq\partial\TORS\eq\partial\TOT$. In particular the two "loose ends" of $A$-ribbons, starting and ending on an $X$-ribbon are to be identified.
The proof that $\proj$ is a projector follows from the representation properties and the fact that $A$ is special symmetric Frobenius. In the case that $X\eq A$, the proof reduces to the proof of Lemma 5.2 (i) of \cite{fjfrs}.\\
Now we make two observations. First, $ \tftm(\TORS)$ lies in the space $\im(\proj)$.
To see that this is the case, first note that the difference between the composition $\projMF\circ\TORS$ and $\TORS$ is that in $\projMF\circ\TORS$ there are two annular $A$-ribbons instead of one ending on each of the two $X$-ribbons. Next, using the representation property and then applying lemma \ref{lemma:oneA_morph}, it follows that $\tftm(\projMF\circ\TORS)$ is not affected if one of the two annular $A$-ribbons ending on each $X$-ribbon is omitted. This implies that  $\tftm(\projMF\circ\TORS)\eq\tftm(\TORS)$.
In the case that $A=X$ this proof reduces to the proof of lemma 5.2 (iii) of \cite{fjfrs}.
The second observation is that the set
\be\label{Basis_proj}
    \mathcal B:=\{\tftm(\TOT)|p,q\in\I\,, \alpha\in H^{A,X}_{p,q}\, \beta\in H^{X,A}_{\pb,\qb}\}
\ee
constitutes a basis of $\im(\proj)$. In the situation that $X\eq A$, this statement reduces to Lemma 5.2 (ii) of \cite{fjfrs}. The difference in the present setting is that the ribbons here labeled by $X$ are labeled by a \boundA\ in \cite{fjfrs},
but the proof works in the same way.
The proof of Lemma 5.2 (ii) of \cite{fjfrs} combines results of two types: On the one hand statements in which the ribbons here labeled by $X$ are labeled by an arbitrary object in \C, and on the other hand results in which they are labeled by an arbitrary $A$-$A$-bimodule. As a consequence we can combine those results in exactly the same way as in \cite{fjfrs} to prove that the vectors in \eqref{Basis_proj} forms a basis of $\im(\proj)$. Since $ \tftm(\TORS)$ lies in the space $\im(\proj)$, it follows that an expansion of the form \eqref{surg_DF} exists for suitable coefficients.\\
In order to obtain the values of the coefficients in \eqref{surg_DF} we compose both sides with a basis $\mathcal B^*$ dual to the one in \eqref{Basis_proj}. Lemma A.4 of \cite{fjfs} proves that $\mathcal B^*$ is spanned by the vectors $\tftm(\instord X pq\alpha\beta)$, with $\instord X pq\alpha\beta$ the cobordism
  \eqpic{instordual}{220}{115}{
  \put(-12,115) {$\mathcal N^{-1}\,\instord X pq\alpha\beta~:= $}
  \put(90,5){ \Includepic{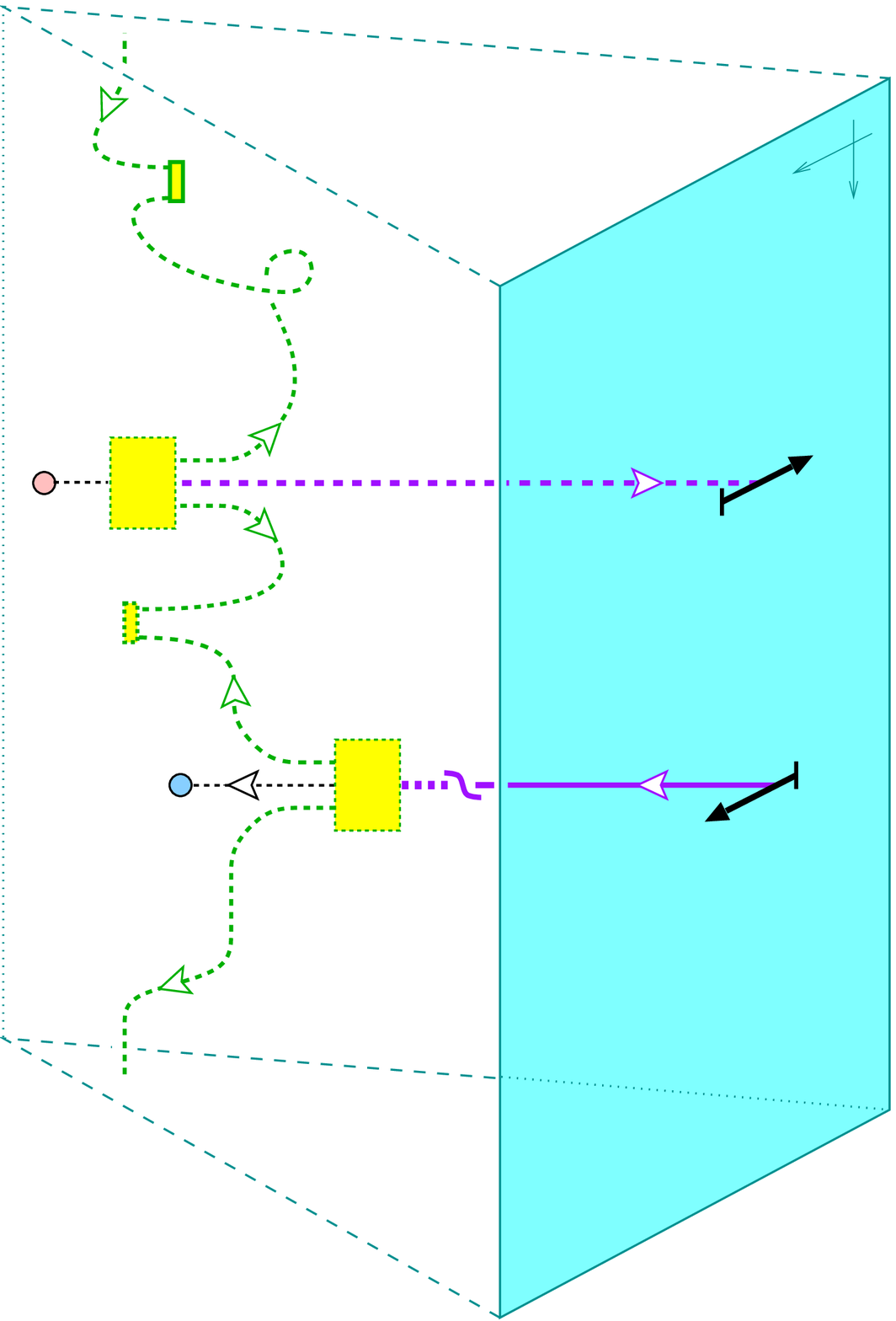}
  \put(33,194)    {\scriptsize $\lambda_{\qb q}$}
  \put(14,114.5)  {\scriptsize $\lambda_{\pb p}$}
  \put(136,199)   {\tiny $1$}
  \put(149,194)   {\tiny $2$}
  \put(24,221)    {\scriptsize $q$}
  \put(54,165)    {\scriptsize $\qb$}
  \put(42,72)     {\scriptsize $q$}
  \put(51,130)    {\scriptsize $\pb$}
  \put(45,107)    {\scriptsize $p$}
  \put(21.3,143.5){\small $\bar{\beta}$}
  \put(60.8,91)   {\small $\bar{\alpha}$}
  \put(122,78)    {\begin{turn}{27}\scriptsize $(X,+)$\end{turn}}
  \put(123,132.5) {\begin{turn}{27}\scriptsize $(X,-)$\end{turn}}
  } }
where $\bar\alpha$ and $\bar\beta$ labels bases defined in \eqref{basis_DFd} and $\mathcal N$ is the number
\be
    \mathcal N =\frac{(\dim(U_p)\, \dim(U_q))^2}  {\dim(A)\; \dim(X)}\,.
\ee
The proof of this statement is a generalization of the proof of \cite[eq. (5.20)]{fjfrs}. In particular the proof uses that the basis \eqref{basis_DF}, of the space $H^{A,X}_{p,q}$, and the basis \eqref{basis_DFd}, of the space dual to $H^{A,X}_{p,q}$, satisfy \eqref{DF_dual}.
Composing both sides of \eqref{surg_DF} with $\tftm(\instord X pq\alpha\beta)$ we obtain
\eqpic{coeff-1}{290}{64}{\setulen80
  \put(0,90) {$\mathcal N^{-1}\, \Cfact Xpq\alpha\beta~= $}
  \put(101,2){ \includepic{304}{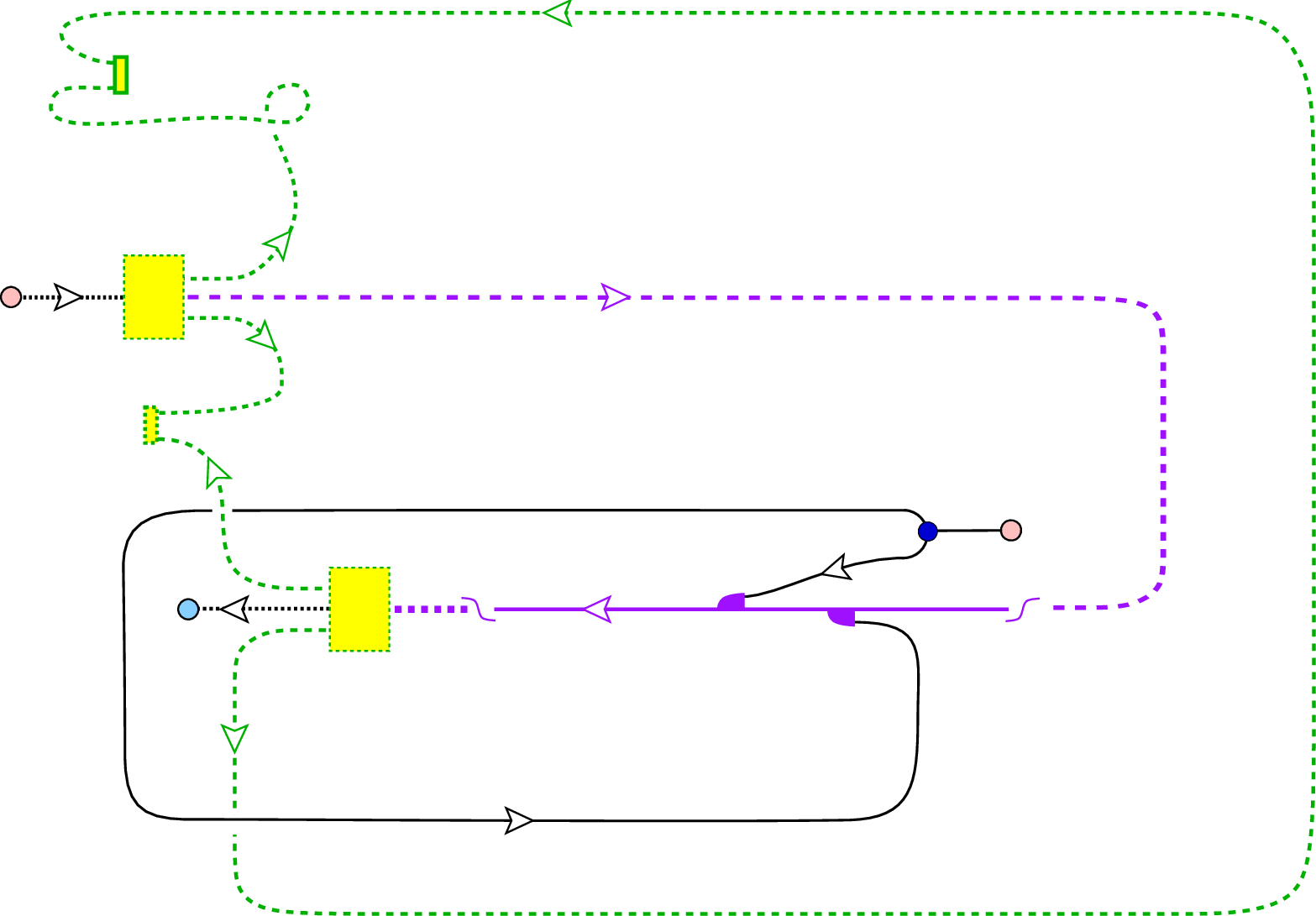}
  \put(56.5,102)   {\scriptsize $\pb $}
  \put(45.2,82)    {\scriptsize $p$}
  \put(59,134)     {\scriptsize $\qb $}
  \put(125,165)    {\scriptsize $q$}
  \put(48.2,28)    {\scriptsize $q$}
  \put(65.3,57)    {\small $\bar{\alpha}$}
  \put(26.1,115)   {\small $\bar{\beta}$}
  \put(115,80)     {\scriptsize $A$}
  \put(155,122)    {\scriptsize $X$}
  } }
Above, the right hand side is regarded as the invariant of a ribbon graph in $S^{3}$. Using that $\bar\alpha$ labels a bimodule morphism and that $A$ is symmetric special Frobenius and deforming the ribbon graph we get the following morphism in $\C$:
\eqpic{coeff-2}{176}{47}{ \put(0,-2){
  \put(0,47) {$\Cfact Xpq\alpha\beta~= S_{0,0}\,\mathcal N$}
  \put(98,0){ \Includepic{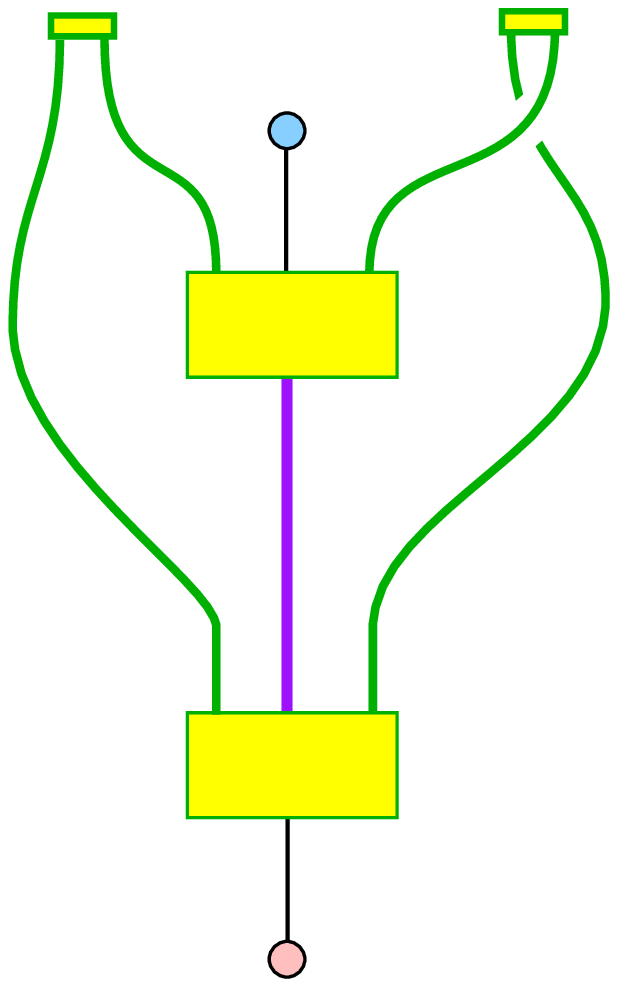}
  \put(21,82)        {\scriptsize $p $}
  \put(23,34)        {\scriptsize $\pb$}
  \put(52,82)        {\scriptsize $q $}
  \put(50.5,34)      {\scriptsize $\qb$}
  \put(35.5,19.8)    {\small $\bar{\beta}$}
  \put(36,69.3)      {\small $\bar{\alpha}$}
  \put(40.8,53)      {\scriptsize $X$}
  } } }
The factor $S_{0,0}$ arises when evaluating the ribbon graph in $S^3$.
The number on the left hand side equals $\dim(U_p)\,\dim(U_q)\,\cdefinv {X}\pb\qb\alpha\beta$, see Lemma A.5 of \cite{fjfrs}.\endofproof
\end{proof}
Combining Proposition \ref{surg_prop} with the expressions \eqref{C-tor} and \eqref{Cfact-tor} proves Theorem \ref{thm:bulkfac}.
\begin{remark}\label{rmk:norm_glue}
    Due to our choice of normalization of the gluing homomorphism our result does not directly reproduce the one in \cite{fjfrs} for $X\eq A$. If we instead normalize the gluing homomorphism as in \cite{fjfrs} the coefficients $\Cfact Xpq\alpha\beta$ gets replaced by
    \be
        \Cfactnorm Xpq\alpha\beta=\theta_p^{}\,\theta_q^{-1}\,
  \RR \pb p 0\nl\nl\, \Rm \qb q 0 \nl\nl\Cfact Xpq\alpha\beta\,,
    \ee
which for $X\eq A$ reproduces the result in \cite{fjfrs}, see \cite[Remark 2.4]{fjfs}.
\end{remark}
\subsection{Boundary factorization}
Boundary factorization is performed along a line segment, the \emph{cutting interval}\index{cutting!interval}, embedded in the world sheet in such a way that it connects two boundary components. Boundary factorization in the presence of only trivial defects is treated in \cite[Theorem 2.9 \& 2.10]{fjfrs}. In fact, including non-trivial defects does not add any new features: Consider a strip of a world sheet such that there is a finite number of defect lines running parallel to each other and crossing the cutting interval connecting the two boundary components. Using fusion of defect lines, as described in section \ref{sec:equiv_ws}, all defects can be fused to the boundary so that we obtain an equivalent world sheet with no defect lines crossing the cutting interval. This is the situation studied in \cite{fjfrs}. The procedure is illustrated in the following picture:
\eqpic{WSbnd}{245}{61}{\setulen80
  \put(0,0)	{\includepic{304}{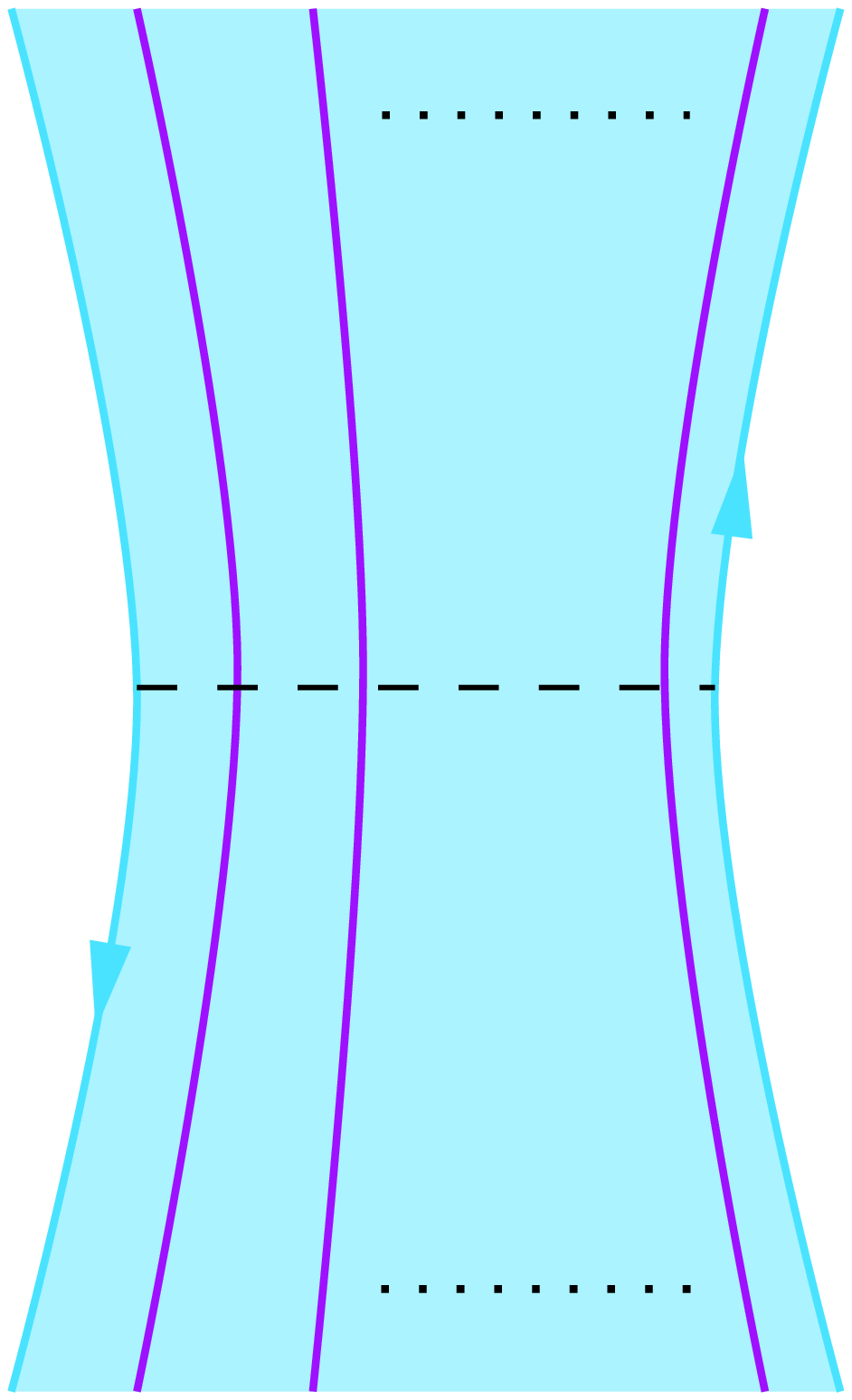}}
  \put(140,83)	{\large$\longmapsto$}
  \put(200,0)	{\includepic{304}{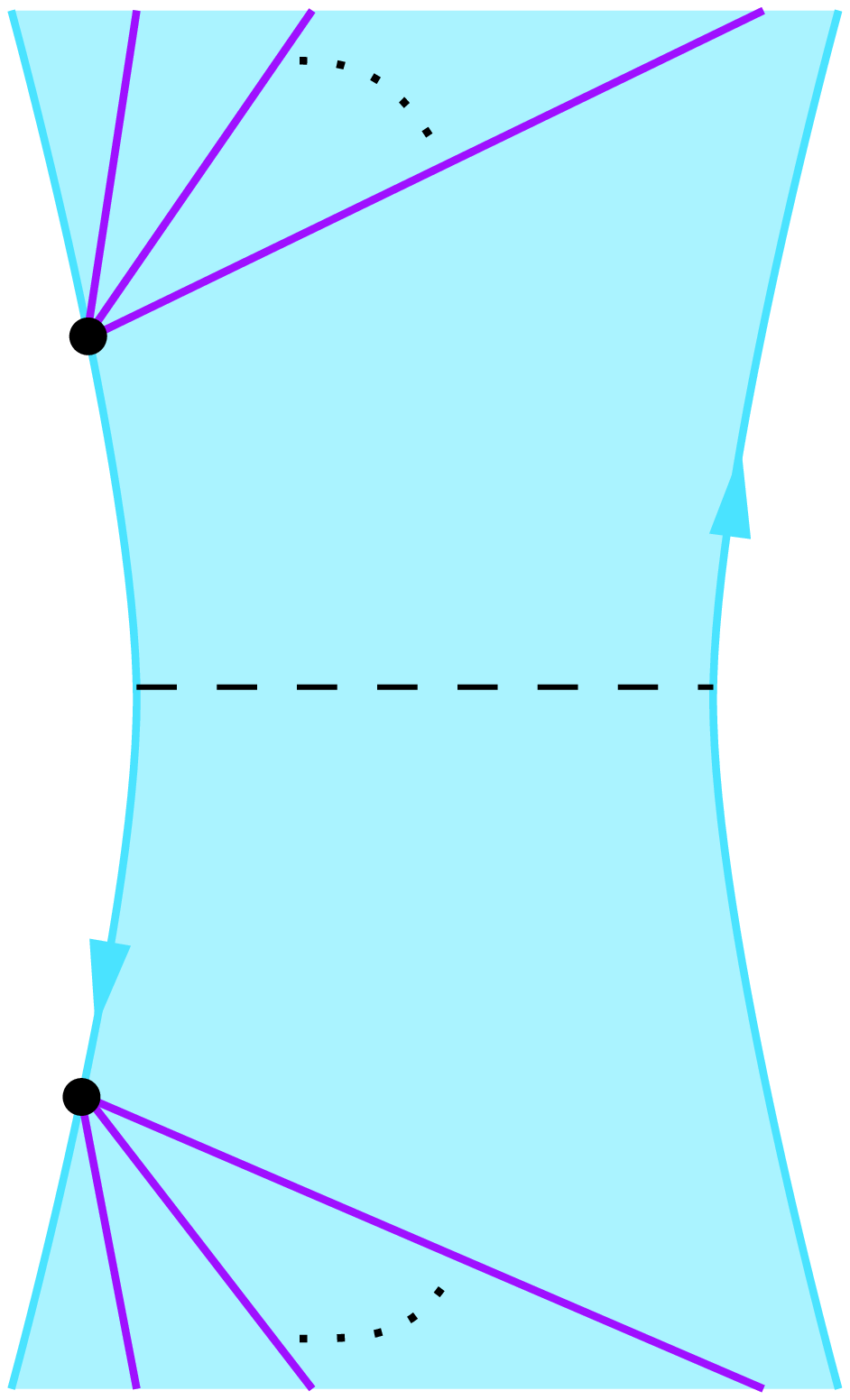}}
  }
Using this equivalence it is straightforward to establish a genuine boundary defect result.

Let $\ws$ be of the type displayed to the left in \eqref{WSbnd}, but with a single defect line. The role played by the defect two-point function in bulk factorization is in boundary factorization played by the correlator of a disc with two boundary fields and one defect line. Denote by  $D \eq D(M_l,M_r,X,i,\psi_+,\psi_-,r,e)$ the world sheet pictorially given by
 \eqpic{DD2pt}{130}{49}{ \setulen7{47}
  \put(0,3)	{\includepic{284}{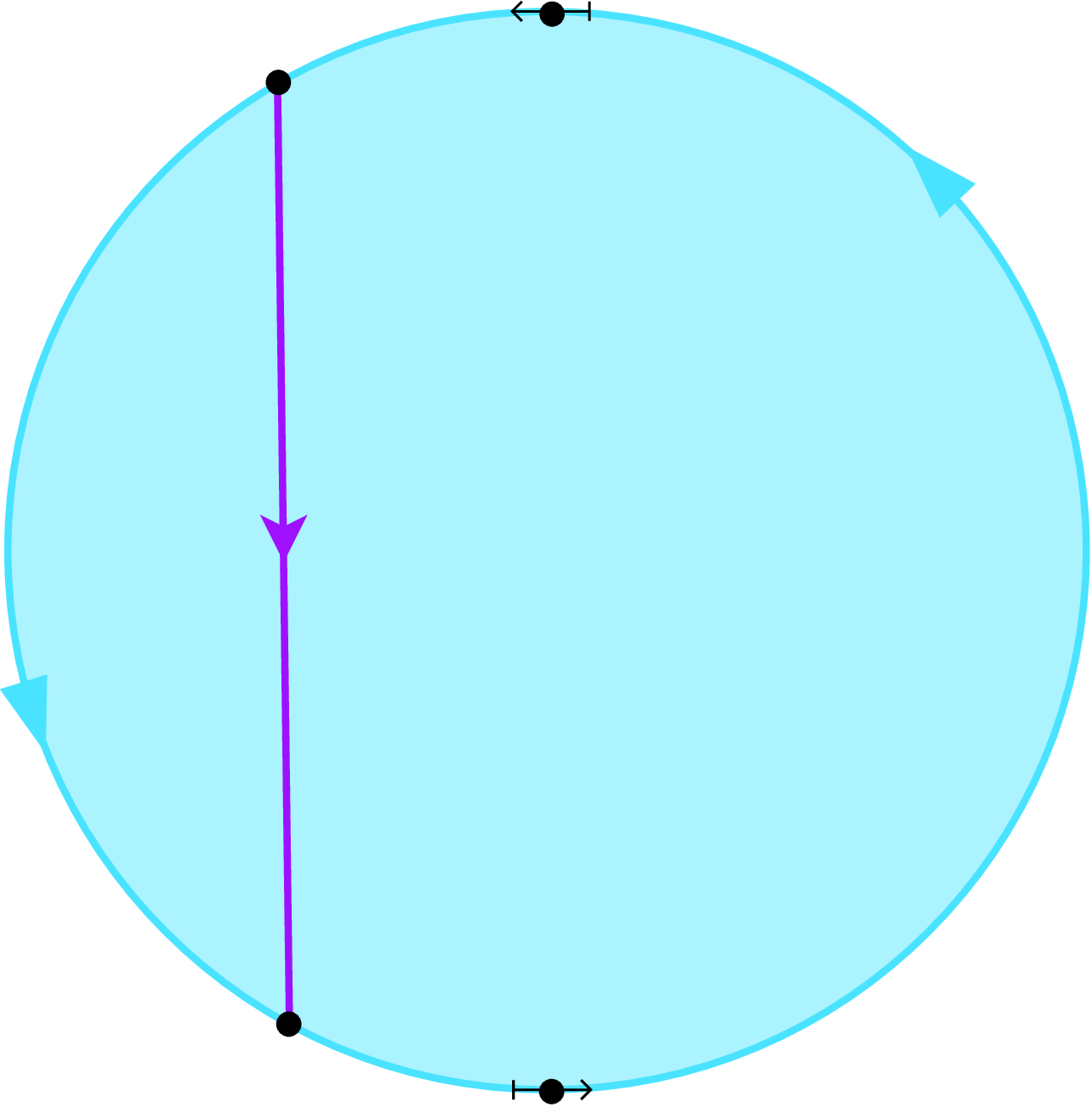}
  \put(67,147)	{\footnotesize $\psi_+$}
  \put(25,138)	{\footnotesize $r$}
  \put(57,-3)	{\footnotesize $\psi_-$}
  \put(25,5)	{\footnotesize $e$}
  \put(143,70)	{\footnotesize $M_r$}
  \put(-16,70)	{\footnotesize $M_l$}
  \put(40,65)	{\footnotesize $X$}
  } }
Here the defect line is labeled by the $A$-$B$-bimodule $X$ and the boundary edges are labeled the modules $M_l$, $X\otiB M_l$ (twice) and $M_r$, respectively.
Two of the vertices are insertion vertices, labeled
by $\psi_+\In \Hom_A((X\otiB M_l)\,\oti U_i,M_r)$ and by
$\psi_-\In \Hom_A(M_r\,\oti\, U_\ib,X\otiB M_l)$, respectively,
and the other two are labeled by the embedding and restriction morphisms
$e\In\HomP_A(X\otiB M_l,X\oti M_l)$ and
$r\In\HomP_A(X\oti M_l,X\otiB M_l)$ appearing in  \eqref{TP_dec}, respectively.
Denoting by \linebreak$\{\psi_\alpha\}$ and $\{\varphi_\beta\}$ choices
of bases of the morphism spaces for the two boundary fields, the structure constants
$c^{\mathrm{bdef}}$ of the correlator are defined through
\be
  C(D) = (c_{M_l,M_r,X,p}^{\mathrm{bdef}})_{\alpha\beta}\,\blV[p,\pb] \,,
\ee
with $\blV[p,\pb]$ defined in \eqref{npbasis_Z}.

Boundary factorization amounts to cutting the original world sheet  $\ws$ along the cutting interval and gluing to the so obtained boundary components the top and bottom halves of $D(M_l,M_r,X,i,\psi_+,\psi_-,r,e)$. This is illustrated in the following picture:
\Eqpic{BdDfact}{320}{100}{ \put(8,5){\setulen80
  \put(0,10){
  \put(0,45)	{\includepic{304}{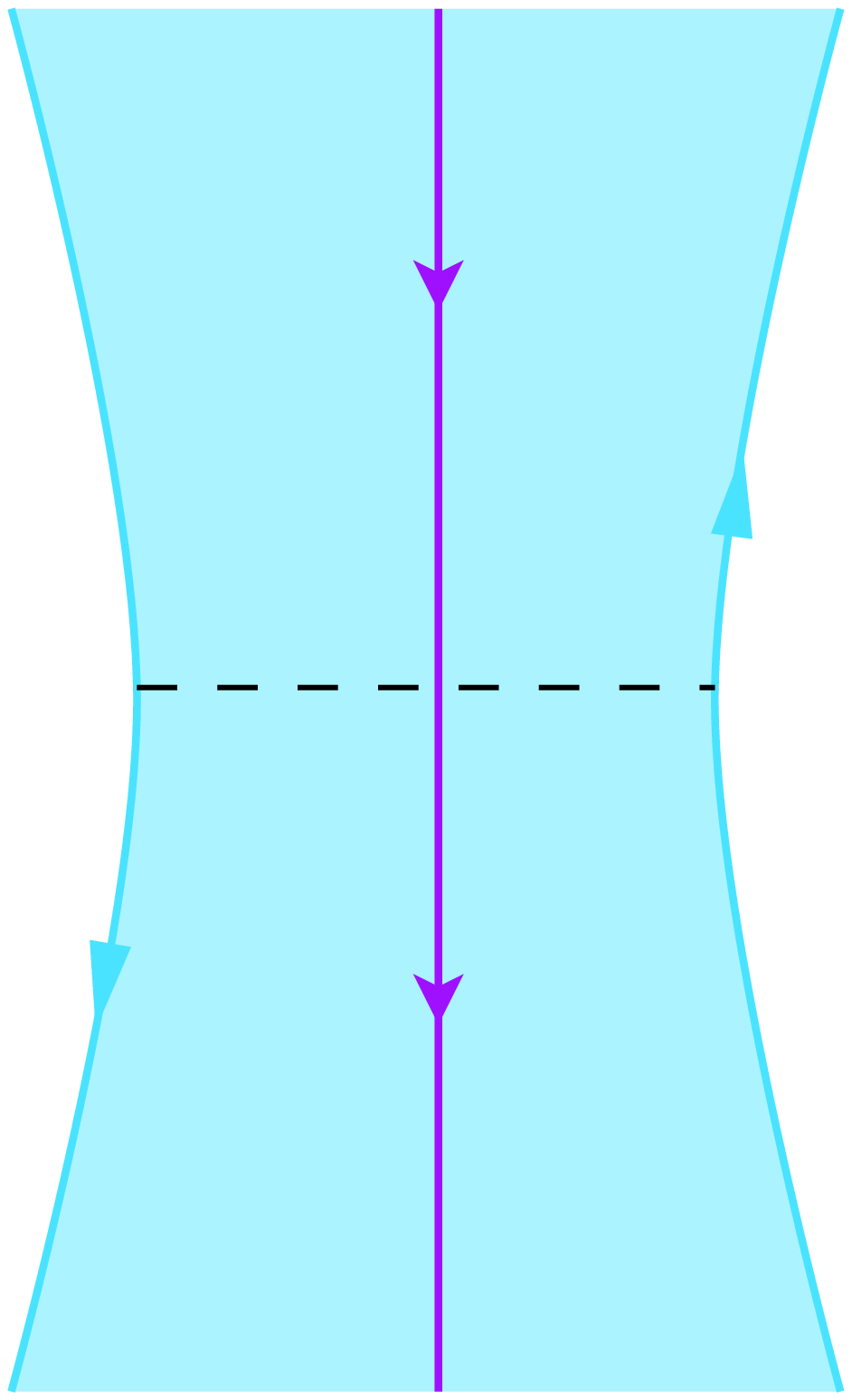}}
  \put(89,117)	{\footnotesize $M_r$}
  \put(0,117)	{\footnotesize $M_l$}
  \put(55,155)	{\footnotesize $X$}
  \put(118,127) {$ \longmapsto $}
  \put(150,30)	{\includepic{304}{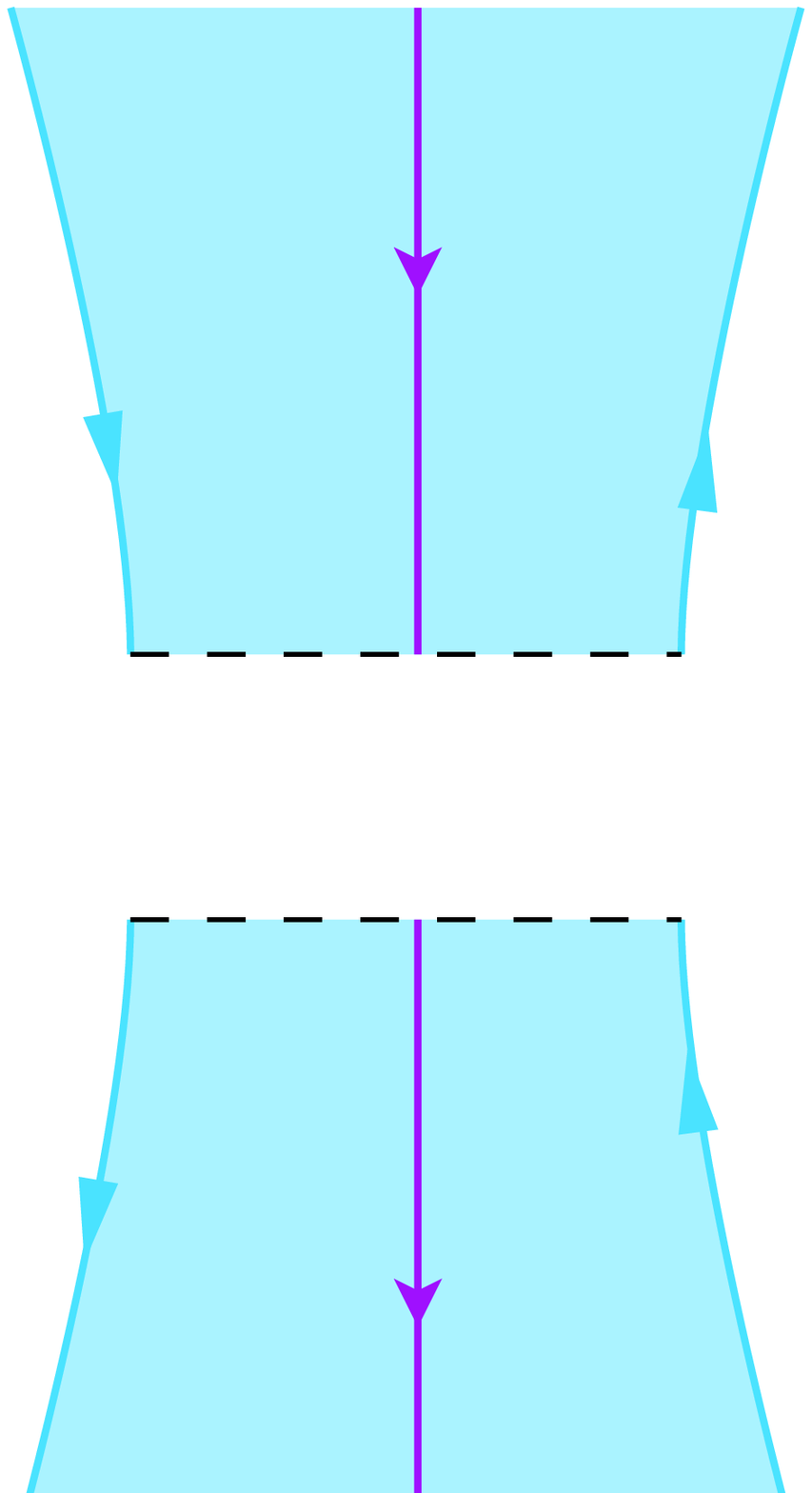}}
  \put(240,100)	{\footnotesize $M_r$}
  \put(150,100)	{\footnotesize $M_l$}
  \put(205,88)	{\footnotesize $X$}
  \put(240,150)	{\footnotesize $M_r$}
  \put(150,150)	{\footnotesize $M_l$}
  \put(205,166)	{\footnotesize $X$}
  \put(271,127) {$ \longmapsto $}
  \put(300,0)	{\includepic{304}{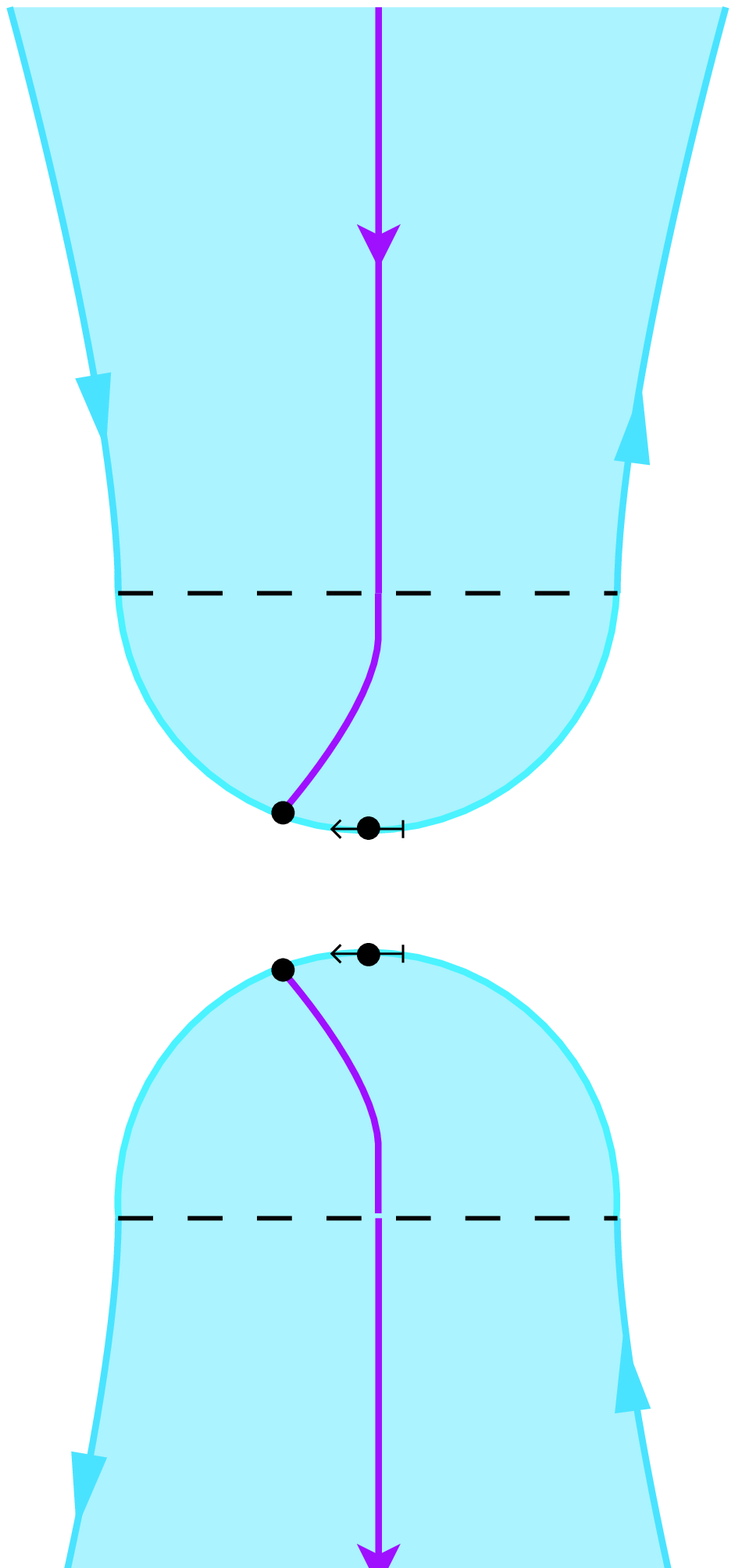}}
  \put(390,70)	{\footnotesize $M_r$}
  \put(300,70)	{\footnotesize $M_l$}
  \put(355,56)	{\footnotesize $X$}
  \put(390,185)	{\footnotesize $M_r$}
  \put(300,185)	{\footnotesize $M_l$}
  \put(355,189)	{\footnotesize $X$}
  \put(355,133)	{\footnotesize $\varphi_\beta$}
  \put(360,120)	{\footnotesize $\psi_\alpha$}
  \put(330,132)	{\footnotesize $e$}
  \put(330,120)	{\footnotesize $r$}
  } } }
Equation (2.37) of \cite{fjfrs} gives the gluing homomorphism $\GLL p{\phantom{q}}^{\mathrm{bnd}}$ relevant for\linebreak boundary factorization. Just like in the case of bulk factorization it is not relevant for the gluing homomorphism whether we deal with trivial or non-trivial defects. We can now give the following result for oriented world sheets.
\pagebreak
\begin{thm} Boundary defect factorization:\label{thm:bdfact}
The correlators of the world sheets $\ws$ and $\ws_{p,\alpha\beta}$ are related as\index{factorization!boundary}
  \be\label{bdefthm}
  C(\ws) = \sum_{p\in\I} \sum_{\alpha,\beta} \dim(U_p)\, {(c^{\mathrm{bdef}\ -1}
  _{M_l,M_r,X,p})}_{\beta\alpha}^{}\, \GLL p{\phantom q}^{\mathrm{bnd}}
  \big( C(\ws_{p,\alpha\beta})\big) \,.
  \ee
\end{thm}
By the procedure illustrated in \eqref{WSbnd} the proof of this statement follows directly from \cite[Theorem 2.9]{fjfrs}.

\subsection{Bulk factorization in practice}\label{BF_pract}
In the next section we will use Theorem \ref{thm:bulkfac} to derive the classifying algebra for defects. In order to prepare that description we describe here how to actively "perform factorization" on a world sheet. That is, we will explain how to obtain the cobordisms $\MGL \circ \M_{\wsf pq\alpha\beta}$ from the connecting manifold of a given world sheet $\ws$.

Consider bulk factorization along an embedded circle $S\subset\ws$ such that there is exactly one defect line (which may be trivial) labeled by $X$ intersecting the embedding circle $S\subset\ws$. The preimage
\be
		Y_S:=\pi_\ws\inv(\wsemb(S))\subset\Mws
\ee
of the canonical projection $\pi_\ws$ from $\Mws$ to $\ws$ will be relevant. $Y_S$ is the union of all connecting intervals intersecting the image of $S$ under the embedding $\wsemb:\ws\rightarrow\Mws$. $Y_S$ is an annulus whose two boundary components lie on $\partial\Mws$, one on each of the two copies of $\ws$ on the double with opposite orientation, c.f. \eqref{Mws_def} and \eqref{con_man}. Next we
cut $\Mws$ along the annulus $Y_S$. That is, we
remove $Y_S$ from $\Mws$ and take the closure of the so obtained manifold. The result is a manifold with corners $\Mwscut$ whose boundary contains two copies $Y_S^1$ and $Y_S^2$ of $Y_S$. With the choice of triangulation made above, there is exactly one ribbon ending on each $Y_S^i$. Both of them are labeled by $X$; one has core orientation pointing towards $Y_S^1$, and the other has core orientation pointing away from $Y_S^2$. We will refer to these parts of the boundary as the \emph{sticky parts}.

Next, consider another manifold with corners, which we refer to as the \emph{factorization torus} $ \T_{\!q_1q_2\gamma\delta}$. The factorization torus is the following full torus with embedded ribbon graph
 \eqpic{stickyft}{220}{112}{\setulen80
   \put(0,-5){
  \put(-10,146)   {$ \T_{\!pq\gamma\delta}^{A,X} ~= $}
  \put(75,0)  {  \includepic{304}{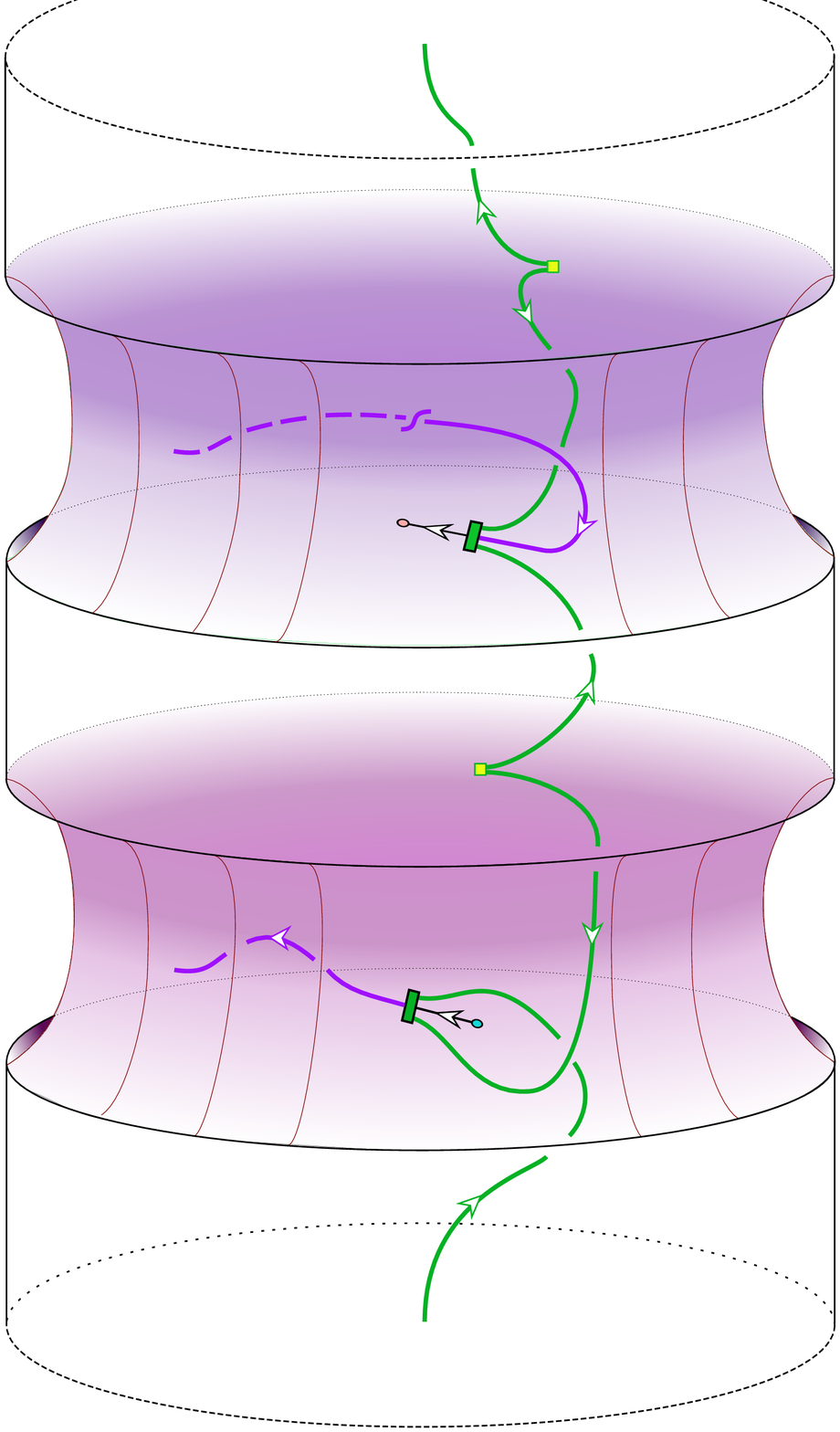} }
  \put(175,0) {
  \put(-5,26)     {\pl {q}}
  \put(14,215)    {\pl {\bar q}}
  \put(13,171)    {\pl {\bar p}}
  \put(36,129)    {\pl {p}}
  \put(-16,79)    {\pl {\phi_\gamma}}
  \put(-2,177.6)  {\pl {\phi_\delta}}
  \put(-25,97)    {\pl {X}}
  \put(-27,205)    {\pl {X}}
  \put(-15,180)    {\pl {A}}
  \put(65,111)    {\lsqarrow}
  \put(65,219)    {\lsqarrow}
  \put(92,113)    {$Y_\T^2 $}
  \put(92,221)    {$Y_\T^1 $}
  } } }
with top and bottom identified, c.f. \eqref{instor}.
Here there are ribbons labeled by the $A$-$A$-bimodule $X$ ending at the annular sticky parts $Y_\T^1$ and $Y_\T^2$, displayed with a shading in \eqref{stickyft}, of the boundary. The morphisms $\phi_\gamma$ and $\phi_\delta$ label the bases of $H^{A,X}_{p,q}$ and $H^{X,A}_{\pb,\qb}$ defined in \eqref{space_DF}, respectively. Identifying the sticky parts of $Y_S^i$ of $\Mws$ with the sticky parts $Y_\T^i$ of $\T_{\!pq\gamma\delta}^{A,X}$ for $i=1,2$ we obtain the cobordism $\MGL \circ \M_{\wsf pq\alpha\beta}$. That this is indeed the case can be verified by comparing with \eqref{ws_torus}, \eqref{ws_fact_torus}, \eqref{deftor} and \eqref{instor} and using the fact that, apart from the embedded ribbon graph, $\TORS$ and $\TOR$ differ by a modular $S$-transformation.
\subsection{Mapping class group invariance}
In \cite{fjfrs} it is proved that in the TFT-construction any correlator involving only trivial defects are invariant under the action of the mapping class group that comes with the \tftfun. In fact, a more general result has been  established: Correlators behave covariantly under \emph{isomorphisms} of world sheet. In \cite{fjfs} we generalize this result to world sheets with an arbitrary network of finitely many defects, see Theorem 5.2 and Corollary 5.3 of \cite{fjfs}. The following discussion summarizes these results.

An isomorphism $f\mapdef\ws\rightarrow\ws'$ of world sheets is an orientation preserving homeomorphism $h\mapdef\tws\rightarrow\tws'$ between world sheets such that $f(\ws)$ is intrinsically equivalent, in the sense defined on page \pageref{def:intr_eq}, to $\ws'$ and for which $h^{-1}$ shares these properties. Note that in particular any representative $f$ of a mapping class $[f]$ is an isomorphism of world sheets with $\tws=\tws'$.

As explained in \cite[Section 5.1]{fjfs} any isomorphism of word sheets $f$ has a unique lift to an isomorphism $\hat f:\widehat\ws\rightarrow\widehat{\ws}$ of the doubles. Recall that the \tftfun\ associates to $\hat f$ a linear map, see \eqref{homeomap}, $\hat f_\sharp:\tfts(\wsD\rightarrow\tfts(\wsD'))$. In fact for any isomorphism of world sheets $f$, the correlators behave covariantly \linebreak\cite[Theorem 5.2]{fjfs}:
\be\label{iso_cov}
    \Cor(\ws') = \hat f_\sharp\big(\Cor(\ws)\big)\,.
\ee
The mapping class group $\Map(\ws)$ of the world sheet with defects can be identified with the subgroup of $\Map(\wsD)$ that commutes with the orientation reversing involution of the double and that preserves the homotopy class of the network of defects.
As a consequence, for $[f]\in\Map(\ws)$ the projective action of the mapping class group $\Map(\wsD)$ of the double on $\tfts(\wsD)$, that comes with the \tftfun, furnishes an action of $\Map(\ws)$ on correlators via the linear map $\hat f_\sharp$. In fact, due to the vanishing of the gluing anomaly, $\Map(\ws)$ acts genuinely, see \cite[Remark 3.7]{fjfrs}. We are now in position to state:

\begin{cor}\label{cor:modinv}
Mapping class group invariance of correlators with defects: We have
  \be\label{modinv}
  Cor(\ws) = \hat f_\sharp\big( Cor(\ws)\big)
  \ee
for any mapping class $[f]\In\Map(\ws)$.
\end{cor}
Let us just outline the proof and refer the reader to \cite{fjfs} for further details.
The map $\hat f_\sharp$ is realized as the invariant of a cobordism $\M_f\mapdef\wsD\rightarrow\wsD$. Thus $\hat f_\sharp\big( Cor(\ws)\big)$ is the invariant of the composition $\M_f\circ\M_{\ws}$. This manifold is homeomorphic, via a homeomorphism that restricts to the identity in the boundary, to the connecting manifold of a world sheet $\ws'$ that is intrinsically equivalent to $\ws$. It follows that  $\hat f_\sharp\big( Cor(\ws)\big)\eq\Cor(\ws')\eq\Cor(\ws)$.

\section{Fundamental world sheets}
Let us conclude the discussion of the TFT-construction by discussing fundamental world sheets.
\index{world sheet!fundamental|textbf}We can define a \emph{category of world sheets} by taking as objects all world sheets and as morphisms isomorphisms of world sheets. In fact it is convenient to think of isomorphisms and factorization as two types of arrows in a larger structure, say of isomorphisms as \emph{vertical} arrows and factorization as \emph{horizontal} ones.

Consider a world sheet $\ws$ and a sequence of such arrows from $\ws$ to some world sheet $\ws'$. The factorization theorems \ref{thm:bulkfac} and \ref{thm:bdfact} and covariance under isomorphisms \eqref{iso_cov} assign to every such sequence a way of expressing $\Cor(\ws)$ in terms of $\Cor(\ws')$; or rather as a sum involving a collection of target world sheets differing only in their decoration data. In addition we may, using the equivalences of world sheets in section \ref{sec:equiv_ws}, replace the collection of world sheets $\ws'$ by an equivalent collection $\ws''$.

In fact, the procedure just described allows us to identify a collection \fundws\ of \emph{fundamental world sheets} with the property that the correlators of any world sheet can be written as a sum involving only correlators of world sheet in \fundws. For oriented CFT, this set is in the absence of non-trivial defect lines a finite set $\fundws_\circ$, which can be taken to consist of the following world sheets, see e.g.\ \cite{lewe3,fuRs10}:
\begin{itemize}\addtolength{\itemsep}{-6pt}%
\item
      three bulk fields on the sphere;
\item
      three boundary fields on the disc;
\item
      one bulk and one boundary field on the disc.
\end{itemize}

As it turns out, when including also arbitrary networks of defect lines  the set \fundws\ is not much more complicated \cite{fjfs}:
\begin{thm}\label{thm:cS}
The set $\fundws$ of fundamental world sheets with defects can be taken to consist of
      ~\\[-1.75em]\begin{itemize}\addtolength{\itemsep}{-6pt}%
 \item
      three boundary fields on the disc (without non-trivial defect lines);
 \item
      three defect fields on a circular configuration of defect lines on the sphere;
 \item
      one boundary field and one disorder field on the disc.
  \end{itemize}
\end{thm}

We refer the reader to \cite[Section 5.2]{fjfs} for the details of the proof. The basic ideas are the following: Note that inclusion of non-trivial defects does not impose any restrictions on the possible ways to factorize. Thus $\fundws$ is certainly included in the set \fundwsred\ that is obtained from $\fundws_\circ$ by replacing bulk fields by defect fields on arbitrary defects and including arbitrary networks of defect lines. However due to equivalences of world sheets the set \fundwsred\ is hugely redundant. By using the equivalences of world sheets, described in section \ref{sec:equiv_ws}, this set is reduced to the set $\fundws$ given in Theorem \ref{thm:cS}. We will explain how this works in one of the three cases.

Take $\ws$ to be the disc with one boundary field and one defect field and an arbitrary network of defect lines. Without loss of generality we may assume that there is only a single defect line ending on the boundary. If this is not the case we can remove a boundary component, as explained on page \pageref{rem_bond}, and fuse defects to achieve that situation. Thus $\ws$ is equivalent to the leftmost world sheet in the following picture:
\Eqpic{Disk21pt1}{320}{33}{ \put(0,-5){ \put(0,-5){\setulen80
  \put(10,7)	{\Includepicfj{2}{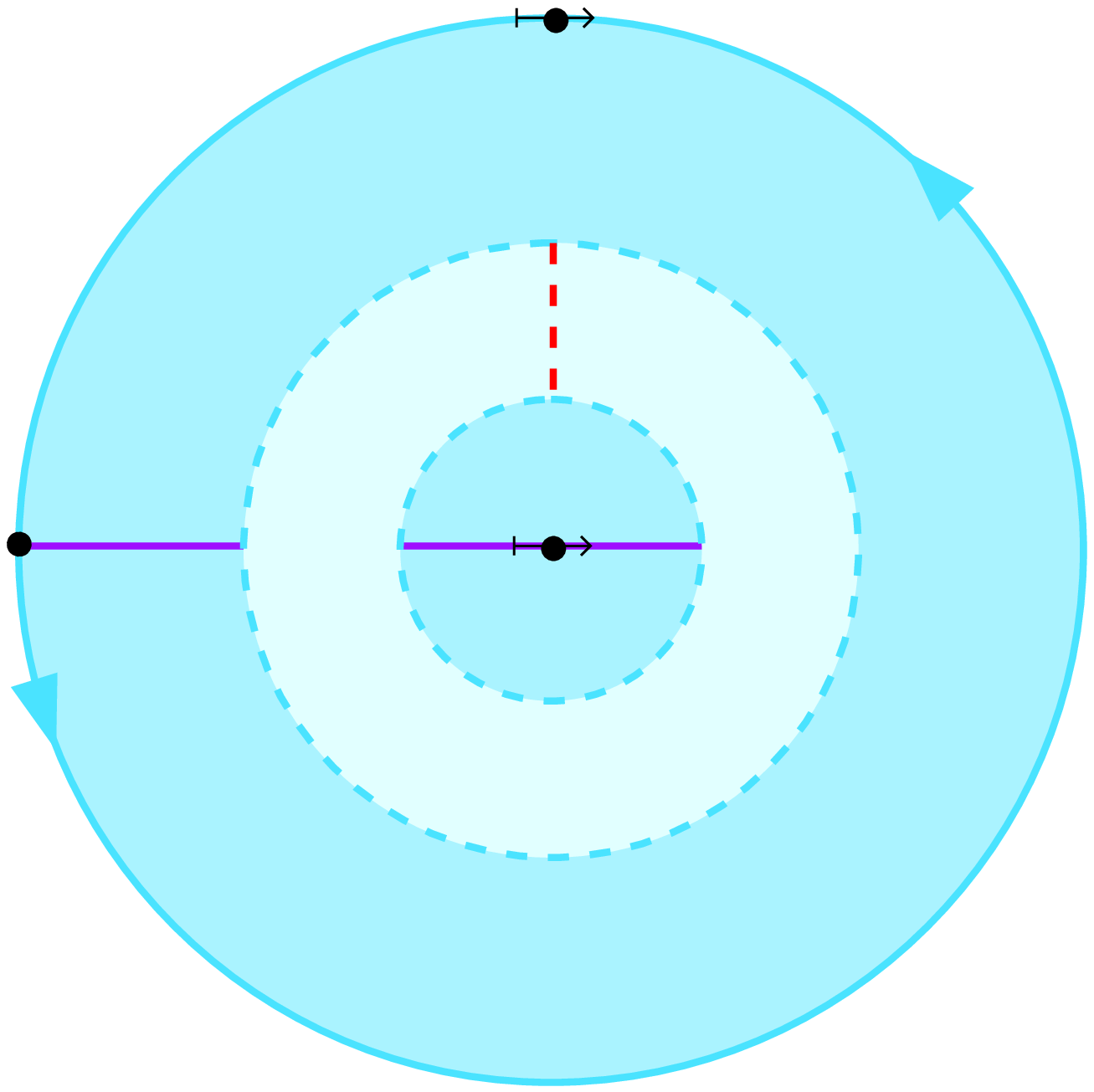}}
  \put(55,105)	{\footnotesize $\psi$}
  \put(0,53)	{\footnotesize $\chi$}
  \put(55,59)	{\footnotesize $\phi$}
  \put(123,51)  {\footnotesize $\longmapsto$}
  \put(165,7)	{\Includepicfj{2}{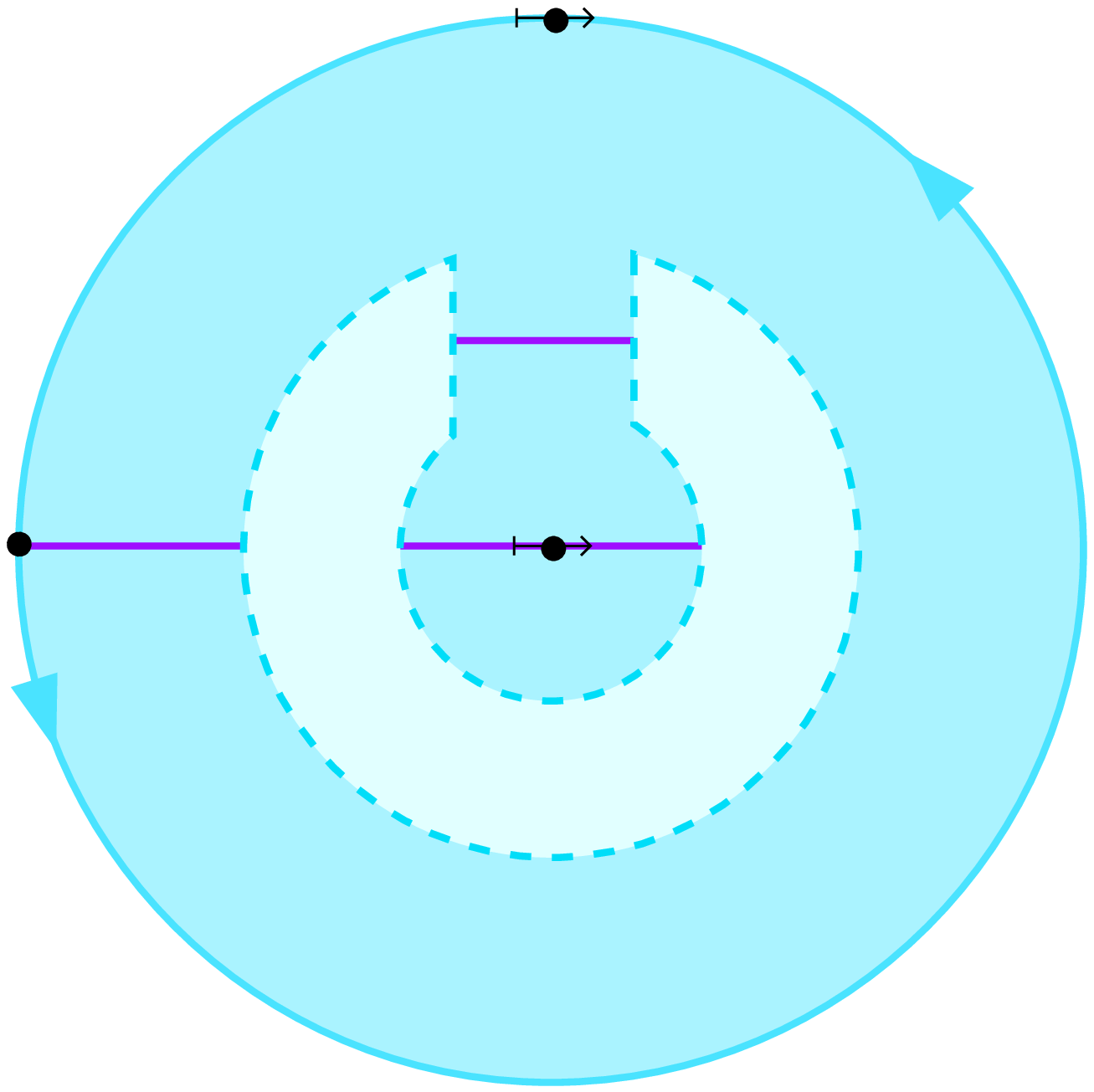}}
  \put(210,105)	{\footnotesize $\psi$}
  \put(155,53)	{\footnotesize $\chi$}
  \put(210,60)	{\footnotesize $\phi$}
  \put(279,51)  {\footnotesize $\longmapsto$}
  \put(320,7)	{\Includepicfj{2}{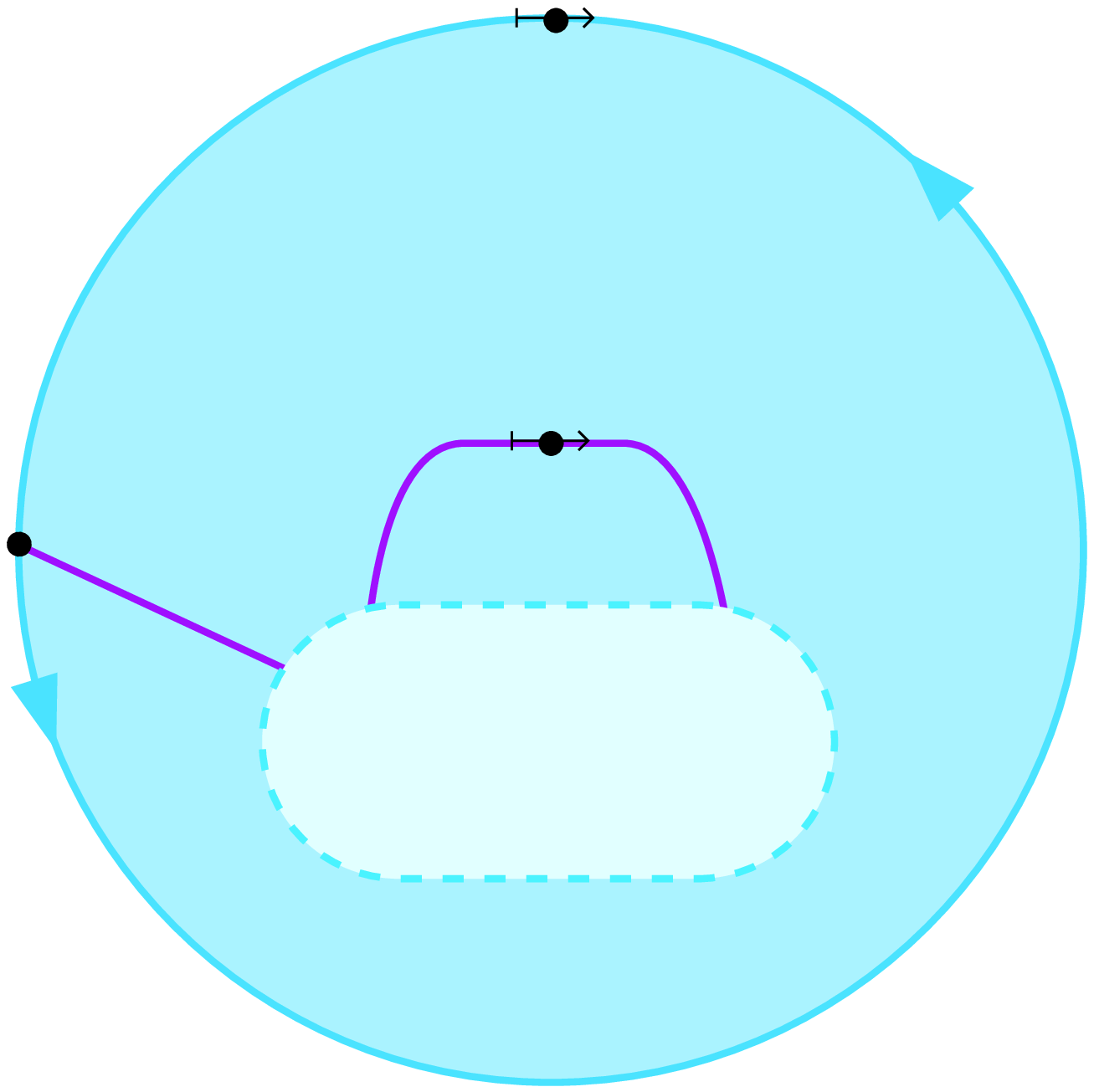}}
  \put(365,105)	{\footnotesize $\psi$}
  \put(310,53)	{\footnotesize $\chi$}
  \put(365,70)	{\footnotesize $\phi'$}
  } } }
Here the region with a lighter shading contains an arbitrary network of defect lines. We may assume that this network is connected. If it is not, then we can e.g.\ use fusion of defects to achieve a connected network. If no defect is crossing the dashed line in \eqref{Disk21pt1} we immediately obtain the third picture. If there are more than one defect line crossing the dashed line they may be fused in order to obtain the world sheet in the second picture. The world sheet in the third picture is then obtained by fusing the two parallel defects in the second picture and replacing the resulting defect bubble by a single defect field. Assuming again that the defect network in the shaded region of the last picture is connected, we use equivalences to replace it with a network with a single vertex in that region, as on the left side in the following picture
\eqpic{Disk21pt2}{227}{37}{\setulen80
  \put(10,7)	{\Includepicfj{2}{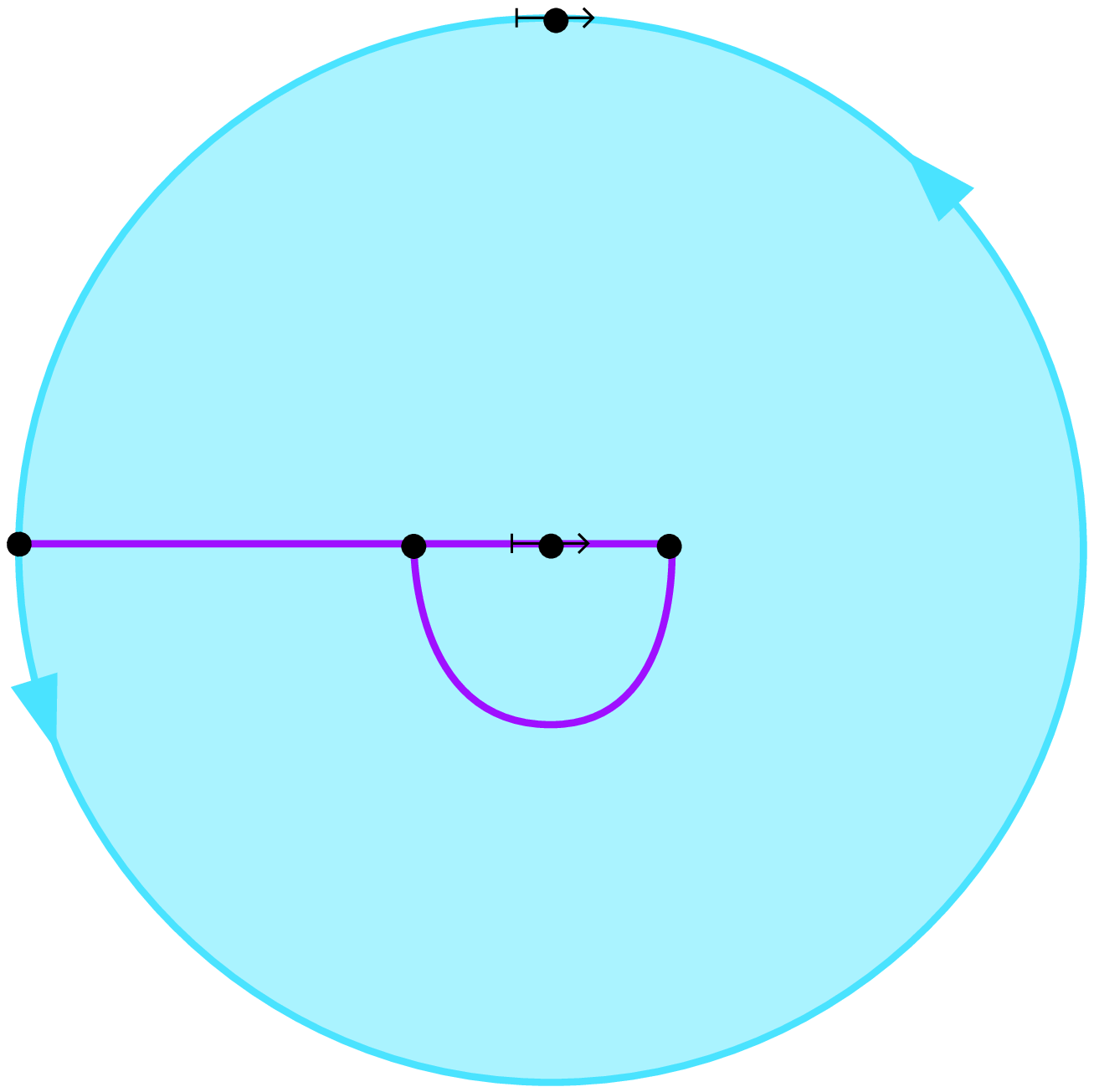}}
  \put(55,105)	{\footnotesize $\psi$}
  \put(0,53)	{\footnotesize $\chi$}
  \put(55,60)	{\footnotesize $\phi'$}
  \put(40,60)	{\footnotesize $\varkappa$}
  \put(139,51)  {\footnotesize $\longmapsto$}
  \put(190,7)	{\Includepicfj{2}{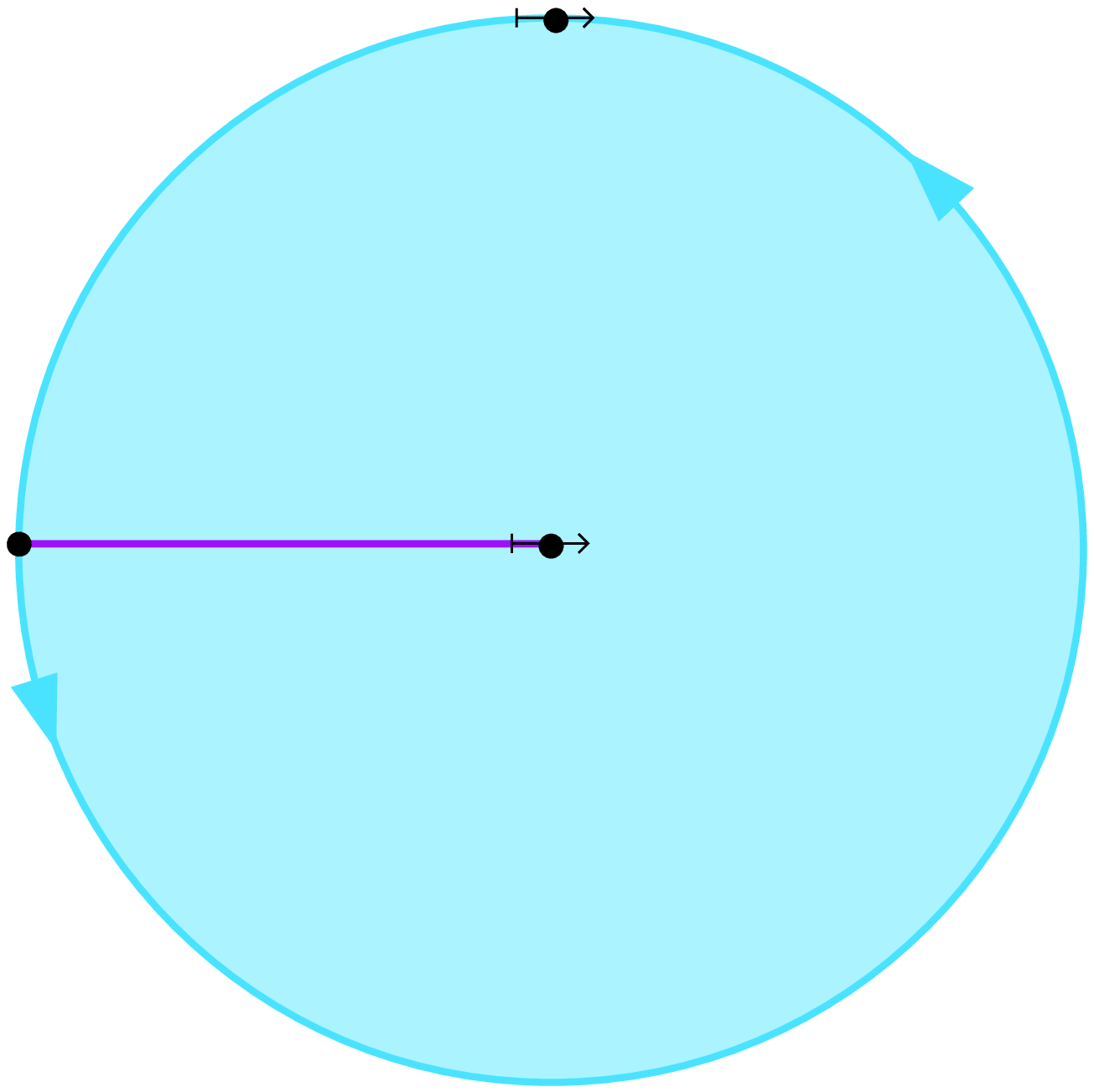}}
  \put(235,105)	{\footnotesize $\psi$}
  \put(181,53)	{\footnotesize $\chi$}
  \put(235,60)	{\footnotesize $\phi''$}
  }
The world sheet in the second picture is obtained from the first by removing the defect bubble.

By analogous arguments the two other types of world sheets in \fundwsred\ can be reduced to the ones in the set \fundws\ given in Theorem \ref{thm:cS}, see \cite{fjfs} for details. This concludes the discussion on fundamental world sheets.
\chapter{The classifying algebra for defects}\label{sec:class_alg}

In this chapter we describe the classifying algebra for defects, introduced in section \ref{CD_intro} in more detail.
We start by describing the \emph{defect transmission coefficients}\index{defect!transmission coefficient}: Let $A$ and $B$ be simple \boundA, and let $X$ be a simple $A$-$B$-defect. Consider the world sheet consisting of a sphere with two bulk fields, $\phi^\alpha_{pq}$ and $\phi^\beta_{\pb\qb}$ in phase $A$ and $B$ inserted on the northern and southern hemisphere respectively, separated by a defect running along the equator, labeled by a simple $A$-$B$-bimodule $X$. Labeling the bulk fields by the basis defined in \eqref{basis_DF}, the defect transmission coefficient $\dcoef ij\alpha\beta X$ is (up to a normalization factor) the structure constant, expressed in terms of our standard basis \eqref{npbasis_Z} of conformal blocks, of this correlator. Via the TFT-construction we obtain:
 \eqpic{def_TC}{180}{57} {\setulen90
  \put(0,63)   {$ \dim(X)\; \dcoef \ia\ja\alpha\beta X ~= $}
  \put(97,-6){ {\INcludepicclal{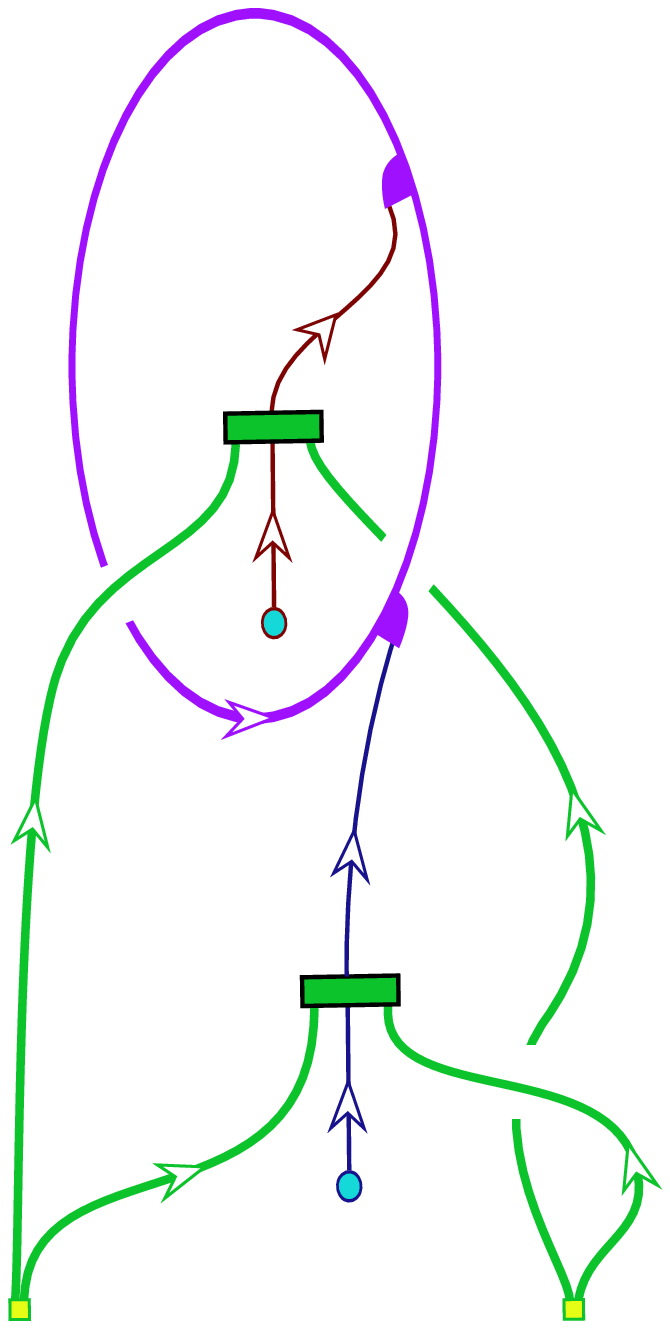}{342}}
  \put(47,138)   {\pl X}
  \put(22,103)   {\pl{\phi_{\!\alpha}^{}}}
  \put(50,39)    {\pl{\phi_{\!\beta}^{}}}
  \put(39.5,116) {\pl A}
  \put(46.8,58)  {\pl B}
  \put(2.5,15)   {\pl{\ia}}
  \put(17,7)     {\pl{\ib}}
  \put(58.8,8)   {\pl{\ja}}
  \put(76.7,8)   {\pl{\jb}}
  } }
 \section{The classifying algebra for defects}\label{sec:CD}
Recall from \eqref{CD_space} that, as a vector space, $\CD$ is the morphism space
\be
    \CD= \bigoplus_{p,q\In\I} \Homaa{U_p\otimes^+ A\otimes^- U_p} A\otimes_\CN \Hombb{U_\pb\otimes^+ A\otimes^- U_\qb} B\,.
\ee
Using the basis introduced in \eqref{basis_DF} we also choose a basis
 \be\label{basis_CD}
  \{ \phi^{pq,\alpha\beta}_{} \}
  = \{ \phi^{pq,\alpha}_A \oti \phi^{\pb\qb,\beta}_B \,|\, p,q\In\I\,,\,
  \alpha\eq1,...,Z_{pq}(A)\,,\, \beta\eq1,...,Z_{pq}(B) \}
 \ee
for $\CD$. For clarity we here display the phase of each bulk field. That is, $\phi^{pq,\alpha}_A$ equals $\phi^\alpha_{pq}$  in \eqref{basis_DF} with $X\eq Y\eq A$ and similarly for $\phi^{\pb\qb,\beta}_B$. Note that the dimension of $\CD$ is given by
\be\label{dim_CN}
    \dim_{\CN}(\CD)=\Tr(Z(A)Z(B)^{\text t})\,.
\ee

We define a product on $\CD$ in terms of the basis \eqref{basis_CD}:
\be\label{prod_CD_thm}
  \phi^{ij,\alpha\beta}_{} \cdot \phi^{kl,\gamma\delta}_{}
  := \sum_{p,q\in\I} \sum_{\mu,\nu}
  \clc{ij}\alpha\beta{kl}\gamma\delta{pq}\mu\nu\, \phi^{pq,\mu\nu}_{},
\ee
with the structure constants defined by
\be\label{clc_def}
    \bearl\dsty
  \clc{\ia\ja}\alpha\beta{kl}\gamma\delta{pq}\mu\nu
  := \frac1{S_{0,0}^2} \dim(U_p) \dim(U_q)\,\theta_\ja\,\theta_l\,\theta_q
  \nxl{2.5}\dsty \hsp{8}
  \sum_{\kappa,\lambda}
  {(\cbulki_{A;pq})}_{\kappa\mu}\, {(\cbulki_{B;\pb\qb})}_{\lambda\nu}\,
  Z(\KC \ib\kb p\jb\lb q{\alpha\gamma\kappa}{\beta\delta\lambda}) \,,
  \eear
\ee
where $Z(\KC \ib\kb p\jb\lb q{\alpha\gamma\kappa}{\beta\delta\lambda})$ is the invariant of the cobordism
\eqpic{SCIT:70}{270}{105}{ \setulen 49
  \put(0,222){$\dsty \KC{\ia_1}{\ia_2}{p}{\ja_1}{\ja_2}{q}
               {\alphz\alphv\beta_2}{\alphe\alphd\beta_4} ~:= $}
  \put(97,-123){
  \put(104,110)   {\includepicclax1{862}{70}}
  \put(177,383)   {\pl{\phi_{\!\alphz}^{}}}
  \put(251,386)   {\pl{\phi_{\!\alphv}^{}}}
  \put(304,394)   {\pl{\phi_{\!\beta_2}^{}}}
  \put(197,360)   {\pl{\ib_1^{}}}
  \put(249,360)   {\pl{\ib_2^{}}}
  \put(193,312)   {\pl{\ia_1^{}}}
  \put(251,302)   {\pl{\ia_2^{}}}
  \put(310,313)   {\pl{p}}
  \put(308,360)   {\pl{\pb}}
  \put(210,243)   {\pl{\phi_{\!\alphe}^{}}}
  \put(261,254)   {\pl{\phi_{\!\alphd}^{}}}
  \put(319,256)   {\pl{\phi_{\!\beta_4}^{}}}
  \put(205,182)   {\pl{\ja_1^{}}}
  \put(255.3,182) {\pl{\ja_2^{}}}
  \put(355,180)   {\pl{q}}
  \put(180.6,437) {\pl{\jb_1^{}}}
  \put(228,437)   {\pl{\jb_2^{}}}
  \put(336,440)   {\pl{\qb}}
  \put(294,409.5) {\pA A}
  \put(292,270.5) {\pB B}
  } }
in $S^2\times S^1$. Here, the dotted cylinder represents the $3$-manifold $S^2\times S^1$. The $S^2$-direction is horizontal, i.e.\ each horizontal disc represents a two-sphere, and the vertical direction is the $S^1$-direction, i.e.\ top and bottom are to be identified. In the sequel we will always represent $S^2\times S^1$ this way.

\pagebreak
We are now in a position to state:
\begin{thm}\label{thm:CD}
(i)\,~The complex vector space $\CD$ endowed with the product \eqref{clc_def} has the structure of a semisimple commutative unital associative algebra.
\nxl2
(ii)\,~The (one-dimensional) irreducible representations of the algebra \CD\ are in bijection
with the types of simple topological defects separating the phases
$A$ and $B$, i.e.\ with the isomorphism classes of simple $A$-$B$-bimodules.
Their representation matrices are furnished by the defect transmission coefficients.
\end{thm}
The rest of this chapter is devoted to the proof of this statement.
\section{Structure constants from factorization}
In this section we will describe the procedure in \cite{fuSs2}, that uses the bulk factorization identity \eqref{fact_rel},  to show that the defect transmission coefficients indeed are one-dimensional representation matrices of $\CD$:
\be\label{dcoef_quad}
  \dcoef ij\alpha\beta X\, \dcoef kl\gamma\delta X = \sum_{p,q}\sum_{\mu,\nu}
  \clc{ij}\alpha\beta{kl}\gamma\delta{pq}\mu\nu\, \dcoef pq\mu\nu X\,.
\ee

Our strategy is to start from a suitable world sheet and perform factorization in two different directions. By projecting the correlator to a particular basis element in the space of conformal blocks we obtain the two sides of \eqref{dcoef_quad}.
We start from the world sheet $\ws$ which is a sphere with two bulk fields in phase $A$ on the northern hemisphere and two bulk fields in phase $B$ on the southern. The two phases are separated by a simple $A$-$B$-defect $X$ running along the equator. We display the connecting manifold of $\ws$:
  \Eqpic{cob_4p:42}{320}{85}{ \setulen 80
  \put(0,110) {$\Mws ~=$}
  \put(50,-16){
  \put(5,0)   {\includepicclax3{04}{42}}
  \put(69,121.5) {\pl{\phi_{\!\alphe}^{}}}
  \put(140,129.9){\pl{\phi_{\!\alphz}^{}}}
  \put(237.8,96.1){\pl{\phi_{\!\alphd}^{}}}
  \put(253,139.7){\pl{\phi_{\!\alphv}^{}}}
  \put(168,117)  {\pA A}
  \put(275,129.6){\pA A}
  \put(80,104)   {\pB B}
  \put(265.5,108){\pB B}
  \put(83,219)   {\pl {\ia_1}}
  \put(146,219)  {\pl {\ia_3}}
  \put(202,219)  {\pl {\ia_2}}
  \put(262,219)  {\pl {\ia_4}}
  \put(81,26)    {\pl {\ja_1}}
  \put(144,26)   {\pl {\ja_3}}
  \put(201.4,26) {\pl {\ja_2}}
  \put(260,26)   {\pl {\ja_4}}
  \put(305.5,134){\pX X}
  } }
Note that, just like in \eqref{Man_twofieldsdefect}, each vertical rectangle represents a two-sphere. In particular the middle shaded rectangle is the world sheet $\ws$ embedded into $\Mws$.
The
world sheets obtained from the
two factorization operations
is illustrated schematically in the following picture\footnote{We are interested in a projection to a subspace of conformal blocks (c.f.\ \eqref{vac_comp}). For this reason, two identity defect fields have been left out in the picture of the world sheet obtained from the defect-crossing factorization.}:
\Eqpic{schematicview}{320}{95}{ \put(32,10){ \setulen60
  \put(171,245)   {\includepicclax23{501}
                   \put(23,33){\pX X} \put(57,57){\boxA} \put(57,16){\boxB} }
  \put(-20,58)  {\begin{turn}{50} \btcs \nxx\vlleftarrow{15}4
                  \nxx \budef\hspace*{5.8em}factoriz.\end{tabular}\end{turn}}
  \put(34,147)    {\includepicclax2{28}{502}}
  \put(-61,-5)    {\includepicclax2{95}{503}
                   \put(12,14){\pl X} \put(38,42){\boxA} \put(38,10){\boxB} }
  \put(11,-5)     {\includepicclax2{95}{503}
                   \put(12,14){\pl X} \put(38,42){\boxA} \put(38,10){\boxB} }
  \put(450,130)   {\includepicclax2{95}{506} \put(32,30){\boxA} }
  \put(450,65)    {\includepicclax2{95}{503}
                   \put(12,14){\pl X} \put(38,42){\boxA} \put(38,10){\boxB} }
  \put(450,0)     {\includepicclax2{95}{505} \put(32,30){\boxB} }
  \put(230,259) {\begin{turn}{-40} \btcs \nxx\vlrightarrow{15}1 \nxx double
                   bulk\hspace*{6.5em}factorization \end{tabular}\end{turn}}
  \put(317,122)   {\includepicclax2{28}{504}}
   } }
The \emph{defect crossing factorization} is a factorization along a circle intersecting the defect line twice and running between the two bulk fields in phase $A$ and phase $B$, respectively.
In the \emph{double bulk factorization} two separate factorization operations are performed along circles running parallel with the defect line such that there are no bulk fields in between. Note that the result of the defect crossing factorization yields two copies of the world sheet underlying the defect transmission coefficients whereas the double bulk factorization yields one copy of this world sheet together with two three-point functions on the sphere. This should be compared to the occurrences of factors of defect transmission coefficients in  \eqref{dcoef_quad}.

Recall that even though the world sheets we arrive at after factorization have different topology and additional field insertions compared to the original world sheet $\ws$, the factorization identity is a statement about vectors in the space of conformal blocks on the double $\wsD$ of the original world sheet. The space of conformal blocks on $\wsD$ is in analogy with \eqref{basis_TPB} given by
\be\label{CF_4psphere}
    B(U_{\ia_1}, U_{\ia_3},U_{\ia_2},U_{\ia_4})\oti B^{-}(U_{\ja_1}, U_{\ja_3},U_{\ja_2},U_{\ja_4})\,.
\ee
We will be interested in the contribution to the \emph{vacuum channel} for both factors in \eqref{CF_4psphere}. This means first that we will be interested in the subspace of conformal in which
 \be\label{labels_vacuum}
  \ia_3 = \ib_1\,, \quad \ia_4 = \ib_2\,, \quad \ja_3 = \jb_1
  \quad{\rm and}\quad \ja_4 = \jb_2 \,.
\ee
Second, using the basis \eqref{npbasis_Z} of four-point blocks on the sphere, the simple object labeling the basis elements is taken to be the tensor unit. This space is zero-dimensional unless the labels of the marked points are related as in \eqref{labels_vacuum}, in which case it is one-dimensional, c.f \eqref{npbasis}. Thus, we will be interested in coefficient of the basis element
\be\label{vac_comp}
     B[{\ia_1}, {\bar\ia_1},{\ia_2},{\bar\ia_2}]_0\oti B^{-}[{\ja_1}, {\bar\ja_1},{\ja_2},{\bar\ja_2}]_0\,
\ee
for each $\ia_1,\ia_2,\ja_1,\ja_2\In\Obj(\C)$. Here $B[{\ia_1}, {\bar\ia_1},{\ia_2},{\bar\ia_2}]_0$ is the invariant of the cobordism
 \Eqpic{4pbasis:75,4pbasisvac:76} {320} {45} { \setulen 70
  \put(0,69)  {$ \widehat\blS(\ia_1,\ib_1,\ia_2,\ib_2)_0 ~:= $}
  \put(128,0)  { \includepicclax2{66}{75}
  \put(8,118)  {\pg {\ia_1} }
  \put(29,136) {\pg {\ib_1} }
  \put(68,43)  {\pg {\ib_1} }
  \put(84,65)  {\pg 0 }
  \put(105,135){\pg {\ia_2} }
  \put(126,117){\pg {\ib_2} }
  }
  \put(295,69) {$ \equiv $}
  \put(331,0)  { \includepicclax2{66}{76}
  \put(8,118)  {\pg {\ia_1} }
  \put(29,136) {\pg {\ib_1} }
  \put(105,135){\pg {\ia_2} }
  \put(126,117){\pg {\ib_2} }
  } }
and $B^{-}[{\ja_1}, {\bar\ja_1},{\ja_2},{\bar\ja_2}]_0$ is the invariant of the corresponding cobordism with oppositely oriented boundary.
We denote by
\be
    \cXo\equiv c(\PHi{\alphz}^{\ib_1\jb_1\,\ccA},\PHi{\alphv}^{\ib_2\jb_2\,\ccA};
  X;\PHi{\alphe}^{\ia_1\ja_1\,\ccB},\PHi{\alphd}^{\ia_2\ja_2\,\ccB})_0^{}\,,
\ee
the structure constant of $\Cor(\ws)\,{\equiv}\, C(\PHi{\alphz}^{\ib_1\jb_1\,\ccA},\PHi{\alphv}^{\ib_2\jb_2\,\ccA};
  X;\PHi{\alphe}^{\ia_1\ja_1\,\ccB},\PHi{\alphd}^{\ia_2\ja_2\,\ccB})$,\linebreak corresponding to the basis element \eqref{vac_comp}.

Before going into details we illustrate the procedure schematically, also in terms of cobordisms, in the following picture (in which the connecting manifold on the top is the one of $\ws$, c.f.\ \eqref{cob_4p:42})
\be\label{schematic_MF}
   \begin{picture}(400,400)
   \put(0,10){
   \put(116,300)   { {\includepicclax12{42}} }
   \put(45,289) {\begin{turn}{35} \btcs \budef     \nxx \vlleftarrow57
                  \nxx factorization \end{tabular}\end{turn}}
   \put(236,333){\begin{turn}{-35}\btcs double bulk\nxx \vlrightarrow57
                  \nxx factorization \end{tabular}\end{turn}}
   \put(0,154)   { \includepicclax15{46t} }
   \put(239,154)   { \includepicclax0{87}{62t} }
   \put(30,140) {\begin{turn}{-50} \tovacuum \end{turn}}
   \put(280,97) {\begin{turn}{50}  \tovacuuw \end{turn}}
   \put(32,17)     { {\includepicclax20{44}}\put(30,51){\pX{\scriptstyle X}}}
   \put(77,17)    { {\includepicclax20{44}}\put(30,51){\pX{\scriptstyle X}}}
   \put(128,51)    { \vleq34 }
   \put(137,59)    { {\footnotesize \erf{dcoef_quad}} }
   \put(181,-5)    { {\includepicclax10{70t}} }
   \put(271,17)    { {\includepicclax20{44}}\put(30,51){\pX{\scriptstyle X}}}}
   \end{picture}
\ee
Here we have left out various summations and labels. The pictures of manifolds should be thought of as representing their invariants. All cylinders in \eqref{schematic_MF} represent $S^2\times S^1$.

When studying bulk factorization of a correlator of a sphere, using the strategy described in section \ref{BF_pract}, the manifold with corners $\Mwscut$ will always we a disjoint union of two three-balls. Thus the topology of the manifold obtained from the pairwise identification of sticky components of   $\Mwscut$ and \eqref{stickyft} is not obvious if we use the embedding \eqref{stickyft} of the factorization torus $ \T_{\!pq\gamma\delta}^{A,X}$ into $\R^3$. We will resolve this issue by instead embedding $ \T_{\!pq\gamma\delta}^{A,X}$ into $S^2\times S^1$. Thus we display $ \T_{\!pq\gamma\delta}^{A,X}$ as
\eqpic{glue_tor_inv:45} {250} {100} { \setulen 80 \put(0,-17){
  \put(0,144)   {$ \T_{\!pq\gamma\delta}^{A,Y} ~= $}
  \put(85,0)  {
  \put(0,15)    {\includepicclax3{04}{45}}
  \put(188,46)    {\pl {q}}
  \put(195,197)   {\pl {\qb}}
  \put(215,141)   {\pl {\pb}}
  \put(221,115)   {\pl {p}}
  \put(159,103.4) {\pl {\phi_\gamma}}
  \put(181,154)   {\pl {\phi_\delta}}
  \put(136,114)   {\pl Y}
  \put(157,187)   {\pl Y}
  \put(183,106)   {\pl A}
  \put(171,168)   {\pl A}
  } } }
Here the dotted cylinder represents $S^2\times S^1$ as described after \eqref{SCIT:70}. Note also that in this chapter we follow the conventions in \cite{fuSs2} to use the normalization in \cite{fjfrs} of the gluing homomorphism, c.f. Remark \ref{rmk:norm_glue}. Therefore, the ribbon graph embedded in \eqref{glue_tor_inv:45} differs slightly from the one in \eqref{stickyft}.
With the description \eqref{glue_tor_inv:45} of $\T_{\!pq\gamma\delta}^{A,Y}$ it is straightforward to glue in two three-balls with corners with embedded ribbon graphs.

\subsection{Defect crossing factorization}
Making use of the prescription, given in section \ref{BF_pract}, for how to perform factorization of a correlator, and the embedding \eqref{glue_tor_inv:45} of the factorization torus into $S^2\times S^1$, it is straightforward to perform the defect crossing factorization. We first replace $\ws$ by an equivalent world sheet in which we have fused the defect $X$ with it self in the region between the two bulk fields $\phi_{\!\alphz}^{}$ and $\phi_{\!\alphv}^{}$. Thus we may write the connecting manifold as
\Eqpic{cob_4p_df:43}{320}{100}{ \setulen 80
  \put(0,235) {$\Mws ~=~ \dsty \sum_Y\sum_\tau$}
  \put(50,-2){
  \put(0,0)   {\includepicclax3{04}{43}}
  \put(65,124)   {\pl{\phi_{\alphe}}}
  \put(134,131)  {\pl{\phi_{\alphz}}}
  \put(220,98)   {\pl{\phi_{\alphd}}}
  \put(248.2,144){\pl{\phi_{\alphv}}}
  \put(152,117)  {\pA A}
  \put(269,134)  {\pA A}
  \put(76,107)   {\pB B}
  \put(261,111)  {\pB B}
  \put(139,103)  {\pX X}
  \put(179,136)  {\pX Y}
  \put(176,101)  {\begin{rotate}{51}\pl {\wsemb(S)}\end{rotate}}
  \put(166,138.1){\pl \tau}
  \put(209.1,136){\pl {\bar\tau}}
  \put(79,222)   {\pl {\ia_1}}
  \put(141,222)  {\pl {\ia_3}}
  \put(197,222)  {\pl {\ia_2}}
  \put(256,222)  {\pl {\ia_4}}
  \put(76.5,27)  {\pl {\ja_1}}
  \put(140,27)   {\pl {\ja_3}}
  \put(196.4,27) {\pl {\ja_2}}
  \put(255,27)   {\pl {\ja_4}}
  \put(300,136)  {\pX X}
  } }
Here and below we slightly abuse notation. The equality \eqref{cob_4p_df:43} is to be read as an equality of invariants: The invariant of $\Mws$ equals the sum over the invariants of the manifolds displayed on the right hand side. The summation over $Y$ is taken over isomorphism classes of simple $A$-$B$-bimodules and $\tau$ labels a basis of the morphism space $\Hombb{X^\vee\otiA X}Y$. The dashed-dotted line in \eqref{cob_4p_df:43} is the embedding $\wsemb(\Scut)$ of the cutting circle.

Next we cut the connecting manifold in \eqref{cob_4p_df:43} along the preimage of $\wsemb(\Scut)$ as described in section \ref{BF_pract}. The so obtained manifold with corners consists of two three-balls,  each of them with a ribbon labeled by $Y$ ending or starting at the "sticky" annular part. The identification of these sticky parts with the ones in \eqref{glue_tor_inv:45} is now straightforward. The result is the following manifold:
 \eqpic{cob_4p_df_glue:46} {260} {125} { \put(0,-6){\setulen80
  \put(0,148)   {$ \M_{pq\gamma\delta}^{Y,\tau} ~= $}
  \put(75,0)  {  \INcludepicclal{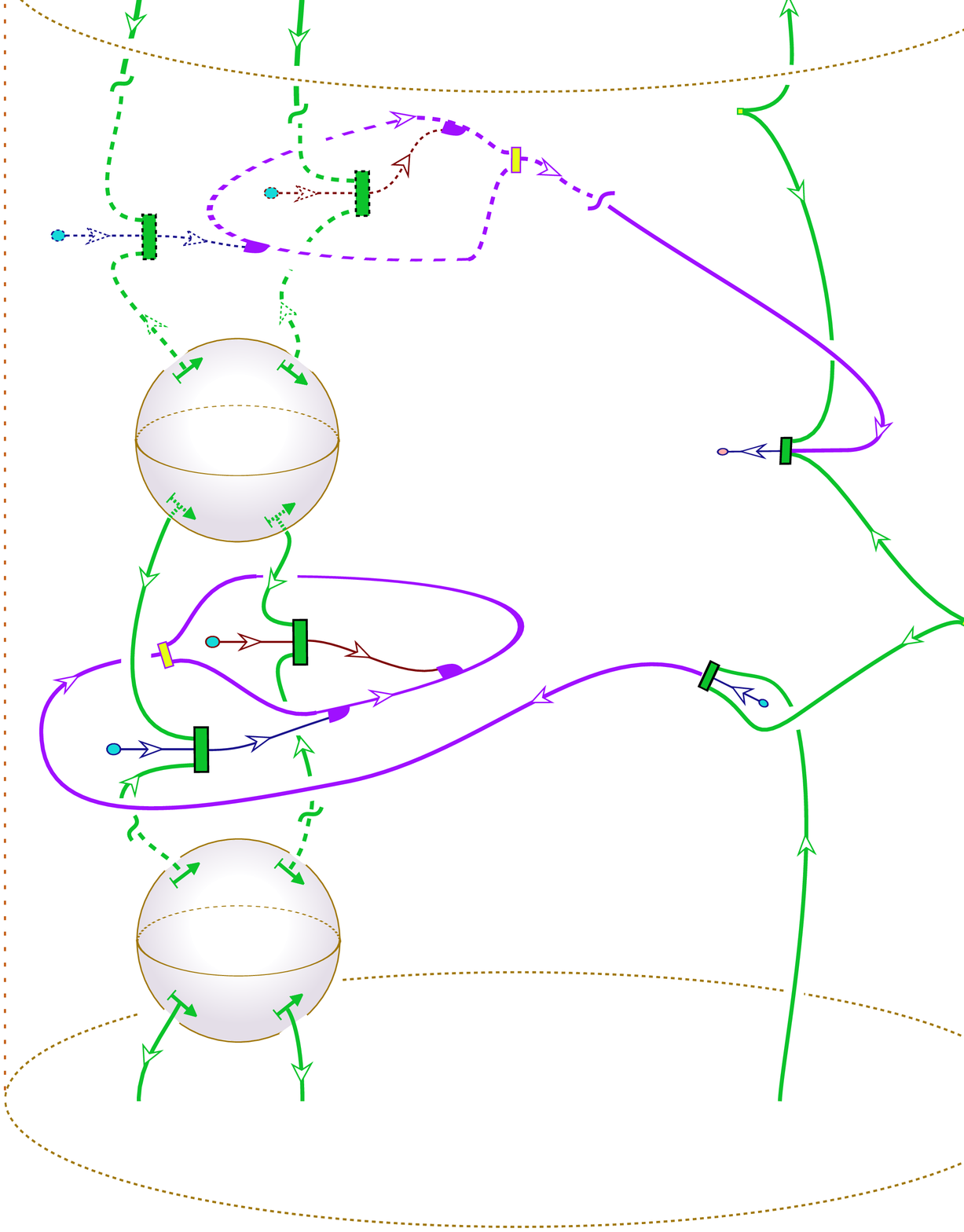}{304}
  \put(215,160)   {\pl {\pb}}
  \put(220,136)   {\pl {p}}
  \put(188,36)    {\pl {q}}
  \put(197,226)   {\pl {\qb}}
  \put(160,122.8) {\pl {\phi_{\!\gamma}^{}}}
  \put(181,174)   {\pl {\phi_{\!\delta}^{}}}
  \put(132,132)   {\pX Y}
  \put(170,228)   {\pX Y}
  \put(42,123)    {\pl{\phi_{\alphe}}}
  \put(67,147.2)  {\pl{\phi_{\alphz}}}
  \put(42,174)    {\pl{\ia_1}}
  \put(43,78.5)   {\pl{\ja_1}}
  \put(59,172)    {\pl{\ia_3}}
  \put(61,78.5)   {\pl{\ja_3}}
  \put(30,244)    {\pl{\phi_{\alphd}}}
  \put(85,235)    {\pl{\phi_{\!\alphv}}}
  \put(43,56)     {\pl{\ja_2}}
  \put(43,197)    {\pl{\ia_2}}
  \put(59,54)     {\pl{\ja_4}}
  \put(62.5,197)  {\pl{\ia_4}}
  \put(97.5,252)  {\pA A}
  \put(90,136)    {\pA A}
  \put(172,187)   {\pB B}
  \put(36,126)    {\pl {\bar\tau}}
  \put(119.5,259) {\pl \tau}
  \put(115,232)   {\pX X}
  \put(101,155)   {\pX X}
  } } }
Thus we have obtained an expression for the composition $\MGL \circ \M_{\ws'}$, whose invariant equals the gluing cobordism applied to one of the factorized world sheets, c.f. section \ref{sec:bulk_proof}.
The factorization identity \eqref{fact_rel} implies that the correlator of $\ws$ can be written as
\be\label{Corr_def_fact}
    \Cor(\ws)=\sum_{Y}\sum_{\tau}
  \sum_{q,p\in\I}\sum_{\gamma,\delta}\,\Cfactnorm Ypq\alpha\beta\, Z(\M_{pq\gamma\delta}^{Y,\tau})\,.
\ee

Next we project to the vacuum sector. That is, we compose both sides of \eqref{Corr_def_fact} with the vector in \eqref{npdualbasis_norm} dual to the one in \eqref{vac_comp}. In terms of manifolds this amounts to composing \eqref{cob_4p_df_glue:46} with cobordisms of the types displayed in  \eqref{npdualbasis}. The result is the following expression for the vacuum channel structure constants:
\be\label{def_fact_coef}
    \cXo= \frac1{S_{0,0}^2}\sum_{Y}\sum_{\tau} \sum_{p,q,\gamma,\delta}
  \Cfactnorm Xpq\alpha\beta\,
  Z(\breve\M_{pq\gamma\delta}^{\,Y,\tau})\,,
\ee
with $\breve\M_{pq\gamma\delta}^{\,Y,\tau}$ the following ribbon graph:
\eqpic{cob_4p_df_glue_proj:47} {270} {122} {\setulen80
  \put(0,-6){
  \put(0,164)   {$ \breve\M_{pq\gamma\delta}^{\,Y,\tau} ~= $}
  \put(85,0)  {  \INcludepicclal{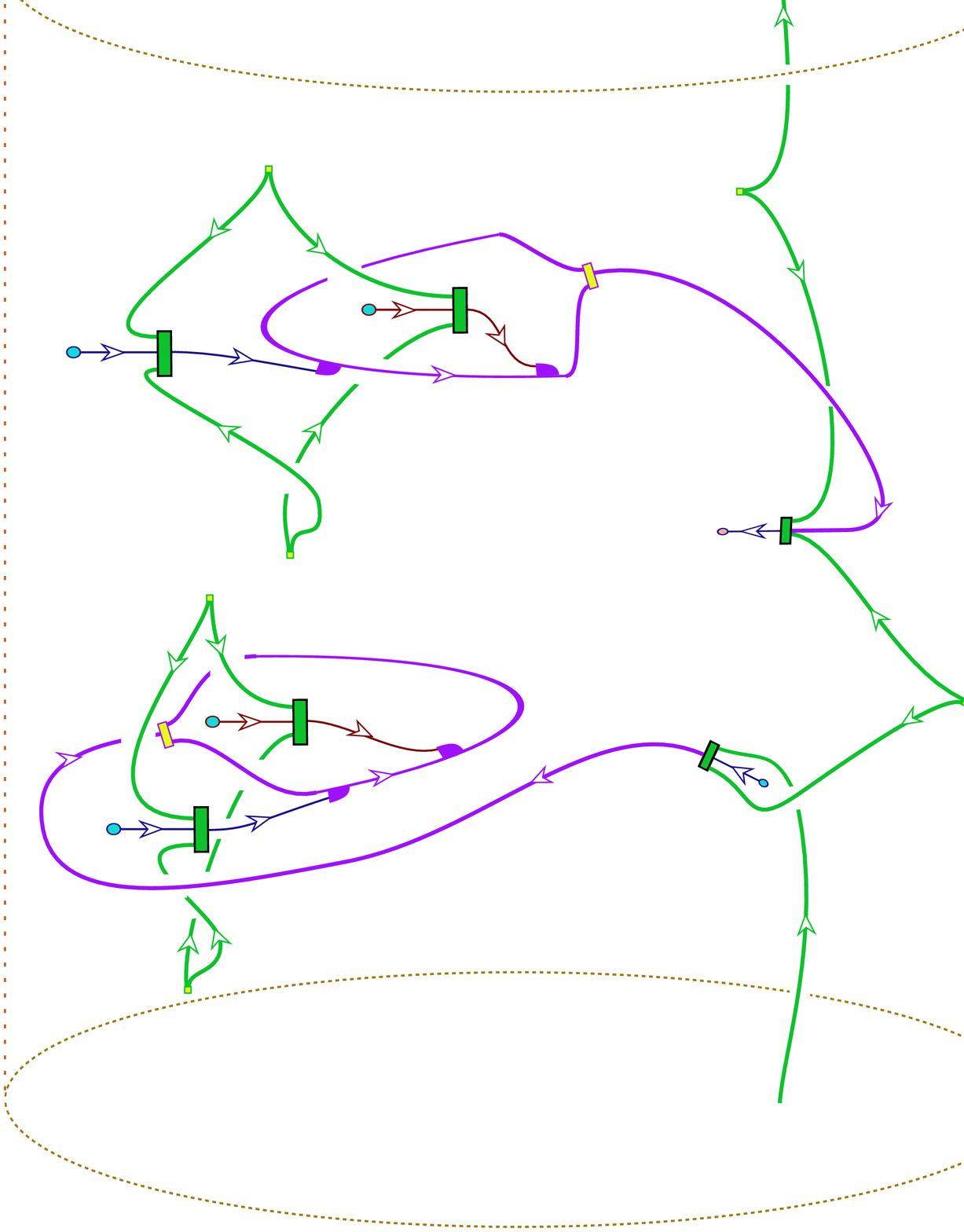}{304}
  \put(187,33)    {\pl {q}}
  \put(196,212)   {\pl {\qb}}
  \put(214,141)   {\pl {p}}
  \put(220,118)   {\pl {\pb}}
  \put(160,103.4) {\pl {\phi_\gamma}}
  \put(181,154.6) {\pl {\phi_\delta}}
  \put(134,113)   {\pX Y}
  \put(171,222)   {\pX Y}
  \put(42,103)    {\pl{\phi_{\alphe}}}
  \put(68,128)    {\pl{\phi_{\alphz}}}
  \put(31,134)    {\pl{\ia_1}}
  \put(52,62)     {\pl{\ja_1}}
  \put(53,144)    {\pl{\ib_1}}
  \put(35,60)     {\pl{\jb_1}}
  \put(35,216)    {\pl{\phi_{\alphd}}}
  \put(105,225)   {\pl{\phi_{\alphv}}}
  \put(39,235)    {\pl{\ia_2}}
  \put(77,168)    {\pl{\ja_2}}
  \put(70,240)    {\pl{\ib_2}}
  \put(58,166)    {\pl{\jb_2}}
  \put(121,209)   {\pA A}
  \put(92,116)    {\pA A}
  \put(173,168)   {\pB B}
  \put(36,107.7)  {\pl {\bar\tau}}
  \put(139,217)   {\pl {\tau}}
  \put(103,236)   {\pX X}
  \put(106,135)   {\pX X}
  } } }
in the closed three-manifold $S^2\times S^1$.
However, the summation in \eqref{def_fact_coef} is redundant. First note that $Z(\breve\M_{pq\gamma\delta}^{\,Y,\tau})$ is the trace of an endomorphism in the space of conformal one-point blocks on the sphere. This space is zero unless $p=q=0$. Second, due to semisimplicity of $\C_{B|B}$ and the fact that $B$ and $Y$ are simple, the space $\Hombb{U_0\otip Y\otim U_0}B\iso\Hombb{Y}B$ is zero unless $Y\eq B$. Thus only $Y\eq B$ gives a contribution to the sum in \eqref{def_fact_coef}. Finally, the space $\Hombb{X^\vee\otiA X}B$ is one-dimensional. By choosing a suitable basis of this space and its dual (see \cite[eq. (4.13)]{fuSs2}) in terms of (co-)evaluation morphisms we end up with the following expression for $\cXo$:
\be\label{fact_dcoeff1}
  \cXo
  = \frac1{S_{0,0}^{}}\,\dim(X)\, \dcoef {\ib_1}{\jb_1}{\alphz}{\alphe} X \,
  \dcoef {\ib_2}{\jb_2}{\alphv}{\alphd} X .
\ee
\pagebreak
\subsection{Double bulk factorization}
The double bulk factorization is more involved than the defect crossing factorization. We indicate the cutting circles in the following picture of the connecting manifold:
\Eqpic{cob_4p_BulkF:48}{320}{95}{ \setulen 80
  \put(0,118)   {$ \Mws ~= $}
  \put(50,3){
  \put(5,-2)   {\includepicclax3{04}{48}}
  \put(71,120)    {\pl{\phi_{\alphe}}}
  \put(114,130)   {\pl{\phi_{\alphd}}}
  \put(197,132)   {\pl{\phi_{\alphz}}}
  \put(234,135)   {\pl{\phi_{\alphv}}}
  \put(211,118)   {\pA A}
  \put(133,103)   {\pB B}
  \put(89,94)     {\begin{rotate}{15} \pl{\wsemb_1(S)}\end{rotate}}
  \put(265,138)   {\begin{rotate}{-36}\pl{\wsemb_2(S)}\end{rotate}}
  \put(95,219)    {\pl {\ia_1}}
  \put(133,219)   {\pl {\ia_2}}
  \put(187,219)   {\pl {\ia_3}}
  \put(242,219)   {\pl {\ia_4}}
  \put(92,26)     {\pl {\ja_1}}
  \put(130,27)    {\pl {\ja_2}}
  \put(186,27)    {\pl {\ja_3}}
  \put(241,27)    {\pl {\ja_4}}
  \put(290,140)   {\pX X}
  } }
Our aim is to obtain a manifold $\M_{n;pqrs}^{\beta_1\beta_2\beta_3\beta_4}$ whose invariant is a double application of a gluing homomorphism to the correlator of the world sheet after a double factorization.
We perform both factorizations simultaneously; accordingly the procedure involves two sticky full tori.
Thus $\Mws$ is cut into three pieces. Two of them (the ones being the union of the connecting intervals over the regions inside the two circles in \eqref{cob_4p_BulkF:48}) are three-balls with sticky part running along the equator, see \eqref{ball_n:49,ball_s:50} for pictures. The third one, depicted in
\eqref{cyl:51}, is a full torus with two sticky components running along the non-contractible cycle.

There are in total four identifications of sticky parts to be performed. Three of them are straight-forward to perform. The sticky parts of the two three-balls in \eqref{ball_n:49,ball_s:50} and the sticky part on the "exterior" of the full torus in \eqref{cyl:51} are easily identified with the sticky parts of the two full tori embedded in $S^2\times S^1$.
We omit all details (see \cite[Section 5]{fuSs2} for details)
and only display the result.
After the three pairwise identifications we obtain the two following manifolds with corners:
  \Eqpic{34glue_1}{260}{105}{ \setulen 75
  \put(0,148)   {$ \MN^{A;\,\ia_3\ja_3,\ia_4\ja_4}
                    _{\!pq\beta_1\beta_2} ~= $}
  \put(107,3){
  \put(0,0)   {\includepicclax2{85}{54}}
  \put(189,35)    {\pl {p}}
  \put(194.6,198) {\pl {\pb}}
  \put(214,142)   {\pl {q}}
  \put(221,117)   {\pl {\qb}}
  \put(157,101.8) {\pl {\phi_{\beta_1}}}
  \put(179,154)   {\pl {\phi_{\beta_2}}}
  \put(137,117)   {\pA A}
  \put(157,187)   {\pA A}
  \put(11,216)    {\pl{\ja_3}}
  \put(71,220)    {\pl{\ja_4}}
  \put(22,166)    {\pl{\ia_3}}
  \put(70,170)    {\pl{\ia_4}}
  \put(30,194)    {\pl{\phi_{\alphz}}}
  \put(80,204)    {\pl{\phi_{\alphv}}}
  \put(2,102)     {\includepicclax2{85}{rsqarrov}}
  \put(-29,98)    {\fbY \T A2}
   } }
and
  \Eqpic{34glue_2}{260}{103}{
  \setulen 75
  \put(0,148)   {$ \MN^{B;X,\ia_1\ja_1,\ia_2\ja_2}
                   _{\!rs\beta_3\beta_4} \,= $}
  \put(115,3){
  \put(0,0)   {\includepicclax2{85}{57}}
  \put(203,106)   {\pl {r}}
  \put(213,141)   {\pl {\rb}}
  \put(188,35)    {\pl {s}}
  \put(195,195)   {\pl {\sb}}
  \put(159,101.8) {\pl {\phi_{\!\beta_3}^{}}}
  \put(180,154)   {\pl {\phi_{\!\beta_4}^{}}}
  \put(145,109)   {\pB B}
  \put(158,187)   {\pB B}
  \put(12,221)    {\pl{\ja_1}}
  \put(72,221)    {\pl{\ja_2}}
  \put(23,167)    {\pl{\ia_1}}
  \put(70,171)    {\pl{\ia_2}}
  \put(30,194)    {\pl{\phi_{\alphe}}}
  \put(81,204)    {\pl{\phi_{\alphd}}}
  \put(145,125)   {\pX X}
  \put(2,102)     {\includepicclax2{85}{rsqarrov}}
  \put(-28,99)    {\fbY SA2}
  } }
The manifold $\M_{n;pqrs}^{\beta_1\beta_2\beta_3\beta_4}$ is then obtained by identifying the sticky components of the manifolds in \eqref{34glue_1} and \eqref{34glue_2}.

The procedure to perform the final identification is a bit lengthy; it is described in appendix A of \cite{fuSs2}.
Informally it can be described as follows: The manifold $\MN^{B;X,\ia_1\ja_1,\ia_2\ja_2}_{\!rs\beta_3\beta_4}$ in \eqref{34glue_2} can be turned into a ribbon graph in $S^3$ by performing surgery along a suitable tubular neighborhood. The resulting manifold is a ribbon graph in a manifold with corners that topologically is $S^3$ with a three-ball cut out and a annular component running along the boundary. This manifold, which is topologically a three-ball, can be "turned inside out" in the sense that we embedded it in $\R^3$ as the interior of a three-ball with an annular sticky component along the equator. This sticky component can easily be identified with the one of $\MN^{A;\,\ia_3\ja_3,\ia_4\ja_4}_{\!pq\beta_1\beta_2}$ in \eqref{34glue_1}. The result is
\Eqpic{fact_MF_S2S1:62}{320}{153}{ \setulen 5{658}
  \put(0,270){$ \dsty\M_{n;pqrs}^{\beta_1\beta_2\beta_3\beta_4} ~= $}
  \put(55,0){
  \put(84,-7)   {\includepicclax2{15}{62}}
  \put(213,452)   {\pl{\phi_{\!\alphz}^{}}}
  \put(264,438)   {\pl{\phi_{\!\alphv}^{}}}
  \put(361,410)   {\pl{\phi_{\!\beta_2}^{}}}
  \put(435,320)   {\pl{\phi_{\!\beta_1}^{}}}
  \put(314,436)   {\pA A}
  \put(391,275)   {\pA A}
  \put(284,270)   {\pB B}
  \put(261,195)   {\pX X}
  \put(195,397)   {\pl{\ia_3}}
  \put(248,397)   {\pl{\ia_4}}
  \put(202,353)   {\pl{\ia_1}}
  \put(249,352)   {\pl{\ia_2}}
  \put(410,400)   {\pl{\pb}}
  \put(444,235)   {\pl{q}}
  \put(205,246)   {\pl{\phi_{\!\alphe}^{}}}
  \put(257,256)   {\pl{\phi_{\!\alphd}^{}}}
  \put(303,225)   {\pl{\phi_{\!\beta_4}^{}}}
  \put(284,173)   {\pl{\phi_{\!\beta_3}^{}}}
  \put(331,224)   {\pl{\rb}}
  \put(192,72)    {\pl{\ja_3}}
  \put(235,77)    {\pl{\ja_4}}
  \put(190,114)   {\pl{\ja_1}}
  \put(235,115)   {\pl{\ja_2}}
  \put(168,200)   {\pl{n}}
  \put(254,133)   {\pl{s}}
  } }
Having obtained $\M_{n;pqrs}^{\beta_1\beta_2\beta_3\beta_4}$ we can apply the bulk factorization formula \eqref{fact_rel} twice. Thus $\Cor(\ws)$ can be written as
\be\label {doublebulk_Corr}
  \bearl\dsty
  \Cor(\ws) = \sum_{p,q,r,s\in\I}\!
  \dim(U_{p}) \dim(U_{q}) \dim(U_{r}) \dim(U_{s})
  \nxl{2}\dsty \hsp{5}
  \sum_{\beta_1,\beta_2,\beta_3,\beta_4}\! {(\cbulki_{A;p,q})}
  _{\beta_2\beta_1}\, {(\cbulki_{B;r,s})}_{\beta_4\beta_3}\,
  \sum_{n\in\I} S_{0,n}\,
  Z(\M_{n;pqrs}^{\beta_1\beta_2\beta_3\beta_4}) \,.
  \eear
\ee
The summation over $n$ is a result of the surgery turning $S^{2}\times S^1$ into $S^3$.

Next we project the result \eqref{doublebulk_Corr} to the vacuum channel. That amounts to composing $\M_{n;pqrs}^{\beta_1\beta_2\beta_3\beta_4}$ with manifolds of the type displayed in \eqref{npdualbasis}, dual to the ones in \eqref{4pbasis:75,4pbasisvac:76}. The result of this composition is
\Eqpic{fact_MF_S2S1_res:66}{320}{160}{ \setulen 5{658}
  \put(0,288){$\dsty \Mf_{0;pq\pb\qb}^{\beta_1\beta_2\beta_3\beta_4} ~= $}
  \put(33,-25){
  \put(104,27)   {\includepicclax2{15}{66}}
  \put(300,418)   {\pA A}
  \put(208,385)   {\pl{\phi_{\alphz}}}
  \put(259,395)   {\pl{\phi_{\alphv}}}
  \put(338,402)   {\pl{\phi_{\!\beta_2}^{}}}
  \put(472.5,381) {\pl{\phi_{\!\beta_1}^{}}}
  \put(193,360)   {\pl{\ib_1^{}}}
  \put(248,357)   {\pl{\ib_2^{}}}
  \put(198,319)   {\pl{\ia_1^{}}}
  \put(250,321)   {\pl{\ia_2^{}}}
  \put(403,380)   {\pl{\pb}}
  \put(428.8,91)  {\pl{q}}
  \put(332,465)   {\pl{\qb}}
  \put(291,276)   {\pB B}
  \put(211,247)   {\pl{\phi_{\alphe}}}
  \put(260,256)   {\pl{\phi_{\alphd}}}
  \put(316,258)   {\pl{\phi_{\!\beta_4}^{}}}
  \put(435,439.5) {\pl{\phi_{\!\beta_3}^{}}}
  \put(481,437)   {\pX X}
  \put(456,387)   {\pA A}
  \put(436,454.5) {\pB B}
  \put(375,262)   {\pl{p}}
  \put(326,122)   {\pl{\qb}}
  \put(189.7,93)  {\pl{\ja_1^{}}}
  \put(255.7,93)  {\pl{\ja_2^{}}}
  \put(156,437)   {\pl{\jb_1^{}}}
  \put(208,429)   {\pl{\jb_2^{}}}
  \put(498,342)   {\pl{p}}
  \put(510,365)   {\pl{\pb}}
  \put(475,338)   {\pl{q}}
  \put(457,334)   {\pl{\qb}}
  } }
To obtain \eqref{fact_MF_S2S1_res:66} we have, after having composed with the basis element dual to \eqref{vac_comp}, performed a series of straightforward manipulations of the ribbon graph, see appendix B of \cite{fuSs2} for details.
This way we obtain a second expression for the coefficient $\cXo\equiv c(\PHi{\alphz}^{\ib_1\jb_1\,\ccA},\PHi{\alphv}^{\ib_2\jb_2\,\ccA};
  X;\PHi{\alphe}^{\ia_1\ja_1\,\ccB},\PHi{\alphd}^{\ia_2\ja_2\,\ccB})_0^{}$:
   \be
   \bearl\dsty
  \cXo = \frac1{S_{0,0}^3} \sum_{p,q\in\I}\! \dim(U_{p})^2_{}
  \dim(U_{q})^2_{} \!\!
  \nxl2 \dsty \hsp{6}
   \sum_{\beta_1,\beta_2,\beta_3,\beta_4}\!
  {(\cbulki_{A;p,q})}_{\beta_2\beta_1}\, {(\cbulki_{B;\pb,\qb})}
  _{\beta_4\beta_3}\, Z(\Mf_{0;pq\pb\qb}^{\beta_1\beta_2\beta_3\beta_4})\,.
  \eear
\ee
Finally, comparing with the expression \eqref{def_TC} for the defect transmission coefficients we can perform a few more manipulations resulting in
 \be\label{fact_dcoeff2}
  \bearl\dsty
  \cXo = \frac1{S_{0,0}^3} \dim(X)\, \theta_{\ja_1}^{}\, \theta_{\ja_2}^{}
  \!\sum_{p,q\in\I}\! \dim(U_{p}) \dim(U_{q})\,\theta_{q}^{}
  \nxl{2}\dsty \hsp{6}
  \sum_{\beta_1,\beta_2,\beta_3,\beta_4}\! {(\cbulki_{A;p,q})}
  _{\beta_2\beta_1}\, {(\cbulki_{B;p,q})}_{\beta_4\beta_3}\,
  Z(\KC{\ia_1}{\ia_2}{p}{\ja_1}{\ja_2}{q}{\alphz\alphv\beta_2}
  {\alphe\alphd\beta_4})\, \dcoef{p}{q}{\beta_1}{\beta_3}X \,,
  \eear
\ee
with $Z(\KC{\ia_1}{\ia_2}{p}{\ja_1}{\ja_2}{q}{\alphz\alphv\beta_2}
  {\alphe\alphd\beta_4})$ as given in \eqref{SCIT:70}.
Comparing \eqref{fact_dcoeff1} and \eqref{fact_dcoeff2} proves \eqref{dcoef_quad} and thus that the defect transmission coefficients furnish one-dimensional representations of the algebra $\CD$ with the product \eqref{clc_def}.
\section{The algebra structure on $\CD$}
Here we outline the rest of the proof of Theorem \ref{thm:CD}, i.e.\ that $\CD$ is a semisimple commutative unital associative algebra and that the defect transmission coefficients exhaust the irreducible representations of $\CD$. We refer to \cite{fuSs2} for a more detailed proof.
\ssubject{Commutativity}
Let $C$ be a \boundA.
Using the unit property of $C$ and the bimodule property one shows that for any $\phi_{\!\alpha}^{}\In\Homcc{U_i\linebreak\otip C\otim U_l}C$
and $\phi_\beta^{}\In\Homcc{U_j\otip C\otim U_k}C$ one has:
 \Eqpic{BF_ex_1}{320}{44}{
 \put(10,0){ \setulen 77
  \put(-17,10){\includepicclax2{926}{74a}
  \put(-2,-9)  {\pl i}
  \put(18,-9)  {\pl j}
  \put(32,-9)  {\pl C}
  \put(28.2,41){\pl C}
  \put(32.8,104) {\pl C}
  \put(50,-9)  {\pl k}
  \put(70,-9)  {\pl l}
  \put(42,25.6){\pl{\phi_\beta^{}}}
  \put(42,75.6){\pl{\phi_{\!\alpha}^{}}}
  \put(78,35)  {$=$}
  }
  \put(78,10){\includepicclax2{926}{74c}
  \put(-2,-9)  {\pg i}
  \put(18,-9)  {\pg j}
  \put(32,-9)  {\pl C}
  \put(31.1,48){\pl C}
  \put(33.2,104) {\pl C}
  \put(50,-9)  {\pg k}
  \put(70,-9)  {\pg l}
  \put(78,35)  {$=$}
  }
  \put(171,10){\includepicclax2{926}{74d}
  \put(-2,-9)  {\pg i}
  \put(18,-9)  {\pg j}
  \put(32.2,-9){\pl C}
  \put(34.4,106) {\pl C}
  \put(50,-9)  {\pg k}
  \put(70.5,-9)  {\pg l}
  \put(78,35)  {$=$}
  }
  \put(266,10){\includepicclax2{926}{74e}
  \put(-2,-9)  {\pg i}
  \put(18,-9)  {\pg j}
  \put(32.5,-9){\pl C}
  \put(31.5,110) {\pl C}
  \put(50,-9)  {\pg k}
  \put(70.5,-9){\pg l}
  \put(77,35)  {$=$}
  }
  \put(360,10){\includepicclax2{926}{74f}
  \put(-2,-9)  {\pg i}
  \put(18,-9)  {\pg j}
  \put(32.4,-9){\pl C}
  \put(28.2,31){\pl C}
  \put(33.4,104) {\pl C}
  \put(50,-9)  {\pg k}
  \put(70.5,-9){\pg l}
  \put(40.1,29.1){\pl{\phi_{\!\alpha}^{}}}
  \put(42,75.8){\pl{\phi_\beta^{}}}
  } } }
and the analogous relation with braidings and inverse braidings exchanged:
  \eqpic{BF_ex_2}{180}{43}{ \setulen 90
  \put(0,10){\includepicclax3{42}{74a}
  \put(-2,-9)  {\pg i}
  \put(18,-9)  {\pg j}
  \put(32.3,-9){\pl C}
  \put(28.3,41){\pl C}
  \put(33.4,103.5) {\pl C}
  \put(50,-9)  {\pg k}
  \put(70,-9)  {\pg l}
  \put(42,25.6){\pl{\phi_\beta^{}}}
  \put(42,75.6){\pl{\phi_{\!\alpha}^{}}}
  \put(96,35)  {$=$}
  }
  \put(126,10){\includepicclax3{42}{74g}
  \put(-2,-9)  {\pg i}
  \put(18,-9)  {\pg j}
  \put(32.3,-9){\pl C}
  \put(28.2,31){\pl C}
  \put(33.1,103.5) {\pl C}
  \put(50,-9)  {\pg k}
  \put(70,-9)  {\pg l}
  \put(40.1,29.1){\pl{\phi_{\!\alpha}^{}}}
  \put(42,75.8){\pl{\phi_\beta^{}}}
  } }
Using \eqref{BF_ex_1} and \eqref{BF_ex_2} one establishes that the invariant $Z(\KC ikpjlq{\alpha\gamma\kappa}{\beta\delta\lambda})$ is \emph{totally} symmetric
in the three quadruples $(ij\alpha\beta)$, $(kl\gamma\delta)$ and
$(pq\kappa\lambda)$. By comparing with \eqref{clc_def} we see that this result extends to a symmetry of the structure constants:
\be
  \clc{ij}\alpha\beta{kl}\gamma\delta{pq}\mu\nu
  = \clc{kl}\gamma\delta{ij}\alpha\beta{pq}\mu\nu \,.
\ee
Thus $\CD$ is commutative.
\ssubject{Unitality}
Using dominance it is easily checked that if one pair of bulk fields in each phase is an identity field, the invariant of \eqref{SCIT:70} is (up to normalization) a product of two two-point functions on the sphere:
 \be
  Z(\KC \ia k0\ja l0{\beta\delta\nl}{\alpha\gamma\nl})
  = \frac{S_{0,0}^2} {\dim(U_{\ia}) \dim(U_{\ja})\,\theta_{\ia}\,\theta_{\ja}}\,
  \delta_{\ib,k}\, \,\delta_{\jb,l}\,
  {(\cbulk_{B;\ia,\ja})}_{\gamma\alpha}\, {(\cbulk_{A;\ib,\jb})}_{\beta\delta}\,,
  \ee
and analogously (using the symmetry of the invariant) if $\ia=\ja=0$ or $k=l=0$. There are two implications of this result. First:
\be
  \clc{00}\nl\nl{ij}\alpha\beta{pq}\mu\nu
  = \delta_{i,p}\, \delta_{j,q}\, \delta_{\alpha,\mu}\, \delta_{\beta,\nu}\,
  = \clc{ij}\alpha\beta{00}\nl\nl{pq}\mu\nu \,,
\ee
which means that the basis element $\phi^{00,\nl\nl}$ is a unit for $\CD$. Second, the linear map from $\CD\otiC\CD$ to $\CN$ defined by
\be\label{bil_form_CD}
  \phi^{ij,\alpha\beta}_{} \oti \phi^{kl,\gamma\delta}_{}
  \mapsto \clc{ij}\alpha\beta{kl}\gamma\delta{00}\nl\nl
\ee
is a non-degenerate bilinear form on $\CD$.
\ssubject{Associativity}
Proving associativity is more lengthy. The strategy is to define a commutative \emph{ternary} product and show that one bracketing of the twofold binary product equals the ternary product. It is then an elementary fact that this implies associativity of $\CD$.

The ternary product  on $\CD$ used for this calculation is obtained by defining structure constants similar to the ones  \eqref{clc_def} for the binary product. The  difference is that the invariant of \eqref{SCIT:70}, which is a trace over an endomorphism of the three-holed sphere, is replaced by the trace over the analogous endomorphism (with one extra pair of bulk fields) of a four-holed sphere, see \cite[eq. (6.22)]{fuSs2}

Let us also remark that associativity implies that the non-degenerate bilinear form in \eqref{bil_form_CD} is also invariant in the sense of Definition \ref{def_sym} (ii). Thus $\CD$ is in addition Frobenius, c.f. Remark \ref{rem_symm} (iii).
\ssubject{Semisimplicity and representation matrices}
We have already shown that the defect transmission coefficients furnish one-dimensional representations of \linebreak$\CD$ labeled by isomorphism classes of simple $A$-$B$-bimodules. Since the \linebreak$\dim_\CN(\CD)\times \dim_\CN(\CD)$-matrix furnished by the defect transmission coefficients is non-degenerate \cite[Theorem 4.2]{ffrs5}, non-isomorphic simple bimodules give rise to inequivalent representations. Thus, the number of inequivalent irreducible representations is at least as large as the dimension of $\CD$:
\be\label{n_LB}
    n_{\text{simp}}(\CD)\geq\dim_\CN(\CD)\,.
\ee
However, as any finite-dimensional associative algebra, $\CD$ is isomorphic as a module over itself to the direct sum over all inequivalent indecomposable projective $\CD$-modules, each one occurring with the multiplicity of the dimension of the corresponding simple module, see e.g. \cite[Satz G.10]{JAsc}. In particular this means that
\be
    n_{\text{simp}}(\CD)\leq\dim_\CN(\CD)\,.
\ee
In view of \eqref{n_LB} this is only possible if the number of inequivalent indecomposable representations equals the dimension of $\CD$, which implies that all irreducible representations are one-dimensional and projective. This in turn implies that $\CD$ is semisimple.

In addition, the dimension $\Tr(Z(A)Z(B)^{\text t})$ of $\CD$ equals the number of isomorphism classes of simple bimodules, see  \cite[Remark 5.19 (ii)]{fuRs4}. Thus, the irreducible representations are exhausted by the defect transmission coefficients. This completes the proof of Theorem \ref{thm:CD}.
\begin{remark}
    In the case $A\eq B$, an algebra $\mathscr D_{\!A|A}^{\rm PZ}$ over \CN\ with the same dimension as $\mathscr D_{\!A|A}$ has been obtained in \cite{pezu6}.
    That construction uses the numbers
    \be
    \gcoef ij\alpha\beta X
    := \sqrt{\dim(U_i)\,\dim(U_j)}\, \dim(X)\, \dcoef ij\alpha\beta X\,.
    \ee
    Under the assumption that the matrix formed by these numbers is unitary and that the mapping
    $\gcoef ij\alpha\beta X\mapsto (\gcoef ij\alpha\beta X)^*$ corresponds to an involution on the set of labels $(ij\alpha\beta)$, it follows that $\mathscr D_{\!A|A}^{\rm PZ}$ is commutative.
    Under the same assumptions one can show that $\mathscr D_{\!A|A}$ is isomorphic to $\mathscr D_{\!A|A}^{\rm PZ}$ as an algebra over \CN\ and that the irreducible representations of $D_{\!A|A}^{\rm PZ}$ are given by the defect transmission coefficients, see \cite{fuSs2} for some more details.
    Thus, $\CD$ can be seen as a generalization of the results in \cite{pezu6} to the case $B\neq A$.
\end{remark}
\section{Defect partition functions}\label{sec:def_PF}
We conclude the discussion of the classifying algebra by considering the \emph{defect partition function} $Z\TXY$. That is, the partition function of a torus with two circular defect lines, labeled by an $A$-$B$-defect $X$ and a $B$-$A$-defect $Y$, running parallel to a non-contractible cycle. Such partition functions where introduced in \cite{pezu5} were they where called "generalized twisted partition functions". Via the TFT-construction, $Z\TXY$, is obtained as the invariant of
\eqpic{manT2d_op:78}{190}{71}{\setulen80
  \put(0,86){$\dsty \M\TXY ~=$}
  \put(70,-5)   {
  \put(0,-2)   {\includepicclax3{04}{78}}
  \put(133.1,63){\pX X}
  \put(134,111) {\pX Y}
  \put(28.5,32) {\pB B}
  \put(54.5,69) {\pA A}
  } }
see e.g.\ \cite[Section 5.10]{fuSs4}. In \eqref{manT2d_op:78} top and bottom are identified, i.e.\ each horizontal cylinder represents a torus.

The \emph{defect partition function} partition function is a vector in the space of conformal blocks on the double $\torus\sqcup-\torus$ of the torus. Recall from section \ref{sec:TFT_SC} that a basis of the space of conformal  blocks on the torus is given by $\{|\chii_i;\torus\rangle|i\In\I\}$ with $|\chii_i;\torus\rangle$ the invariant of the full torus \eqref{char_MF}. Thus we can expand the defect partition function as
\be\label{DF_part_exp}
    Z\TXY=\sum_{i,j\In\I}\,Z\TXYij\,|\chii_i;\torus\rangle\oti|\chii_j;-\torus\rangle\,.
\ee
The structure constants are obtained by gluing to \eqref{manT2d_op:78}, the manifolds underlying the basis elements dual to the ones in \eqref{char_inv}.

In \cite{fuSs2} we apply a double bulk factorization along two circles running parallel to, and in between, the two defect lines. Thus the two cutting circles lie in horizontal planes in \eqref{manT2d_op:78}. The procedure is very similar to the factorizations performed in order to derive the classifying algebra for defects and introduce no novelties. We refer the reader to \cite[Section 7]{fuSs2} for details. The result is that
\be\label{def_PF_fact}
  \bearl\dsty
  Z\TXY = \sum_{p,q,r,s\in\I}
  \dim(U_{p}) \dim(U_{q}) \dim(U_{r}) \dim(U_{s})
  \nxl{-1}\dsty \hsp{6}
  \sum_{\beta_1,\beta_2,\beta_3,\beta_4}
  {(\cbulki_{A;p,q})}_{\beta_2\beta_1}\,
  {(\cbulki_{B;r,s})}_{\beta_4\beta_3}\,
  Z(\Tfact XY{pr}{r s}{\beta_1\beta_2}{\beta_3\beta_4}) \,,
  \eear
  \ee
where $\Tfact XY{pq}{rs}{\beta_1\beta_2}{\beta_3\beta_4}$ is the ribbon graph
\eqpic{M_tor:85}{270}{104}{
  \put(0,104)   {$\Tfact XY{pq}{rs}{\beta_1\beta_2}{\beta_3\beta_4}~=$}
  \put(0,0){\setulen80
  \put(104,-5)  {\includepicclax3{04}{85}}
  \put(104,-16) {
  \put(186,43)  {\pl {s}}
  \put(194,194) {\pl {\bar s}}
  \put(203,106) {\pl {r}}
  \put(213,140) {\pl {\bar r}}
  \put(165.8,117.5) {\pl {\phi_{\!\beta_3}^{}}}
  \put(179,156) {\pl {\phi_{\!\beta_4}^{}}}
  \put(95,191)  {\pA A}
  \put(148,104) {\pB B}
  \put(171,167) {\pB B}
  \put(157,191) {\pB B}
  \put(119,115) {\pA A}
  \put(207,88)  {\pX Y}
  \put(60.2,43) {\pl {q}}
  \put(50,171)  {\pl {\pb}}
  \put(42,155)  {\pg {\pb}}
  \put(62.8,115.5){\pl {\phi_{\!\beta_2}^{}}}
  \put(74,203.2){\pl {\phi_{\!\beta_1}^{}}}
  \put(89,175)  {\pX X}
  } } }
in $S^2\Times S^1$ with two full tori cut out.
Upon composing  \eqref{def_PF_fact} with dual basis elements we obtain an expansion of the structure constants in \eqref{DF_part_exp} in terms of the defect transmission coefficients (see \cite[Section 7]{fuSs2} for details):
\be\label{def_PF_dcoeff}
    \bearl\dsty
  Z\TXYij = \frac{\dim(X)\dim(Y)}{S_{0,0}^2}
  \nxl{0}\dsty \hsp{4.5}
  \sum_{p,q\in\I} S_{i,p}^{}\, S_{j,q}^{*} \!
  \sum_{\beta_1,\beta_2,\beta_3,\beta_4} \!
  {(\cbulki_{A;p,q})}_{\beta_2\beta_1}\,
  {(\cbulki_{B;p,q})}_{\beta_4\beta_3}\,
  \dcoef pq{\beta_1}{\beta_4}X\, \dcoef pq{\beta_3}{\beta_2}Y \,.
   \eear
\ee
\chapter{\BulkA s beyond rational CFT}\label{sec:beyond}
In this chapter we loosen the restrictions of rational CFT and take \C\ to be a factorizable finite ribbon category. Recall from section \ref{coends} that this means that $\C$ satisfies condition \CFN, given in page \pageref{cond:CFN}.
That is, \C\ is a finite (in the sense of Definition \ref{def_finite}) $\k$-linear ribbon category such that the Hopf pairing of the canonical Hopf algebra \L\ is non-degenerate. This is a natural generalization of the notion of a modular tensor category. The modular categories are contained in this class of theories but  we do not require \C\ to be semisimple.

The problem of characterizing the proper class of categories relevant for logarithmic CFT has been addressed in the literature from various points of view, see e.g.\ \cite{fgst3,jf30,garW,huan29,koSai,naTs2}.
The categories suggested have many similarities with factorizable finite ribbon categories. Below we consider $\C\eq H$\Mod\ for $H$ a finite-dimensional factorizable ribbon Hopf algebra over an algebrai\-cally closed field \k\ of characteristic zero, c.f. section \ref{fact_ribbon_hopf}. The Hopf algebras appearing in connection with the $(1,p)$-models \cite{fgst,fgst2,naTs2} are not quite of this form, but very close. While having a factorizable $Q$-matrix, they do not have an $R$-matrix. We will see below that in our construction the $Q$-matrix, rather than the $R$-matrix, plays a fundamental role.

In section \ref{sec_bulkA} the \bulkAC\ \BF\ was described in terms of a coend, see \eqref{BulkA_coend}. This description  does not rest on semisimplicity of $\C$ and is consequently suitable for the present setting.
In the case $\C\eq H$\Mod, \BF\ can be endowed with the structure of a commutative (as well as cocommutative) symmetric Frobenius algebra. Note that $H$\Mod\ is a factorizable finite ribbon category but need not be semisimple;  $H$\Mod\ is semisimple iff $H$ is semisimple as an algebra, in which case $H$\Mod\ is even modular.

We are now in a position to describe the morphism
\be
    \Corrgn\In\Hom_{\eC}(\HK^{\otii g},\BF^{\otii n})
\ee
that will be the center of our attention for the rest of this chapter. Let \C\ be a factorizable finite ribbon category.
Recall from section \ref{BhHopf} that this means that also \eC\ is a finite \k-linear ribbon category.
Assume that the Hopf pairing of  the \bhHopf\  in \eC, i.e.\ the coend \HK\ defined in \eqref{def_bhHopf}, is non-degenerate. This means that \eC\ satisfies condition \CFN, i.e.\ \eC\ is a factorizable finite ribbon category.
When \C\ is such that the coend \BF\ in \eC can be equipped with the structure of a symmetric commutative Frobenius algebra, i.e.\ \BF\ has an interpretation as a \bulkAC, we define the morphism
\eqpic{Sk_morph} {250} {105} {\setulen80
    \put(0,131)  {$\Sk gpq~:=~$}
      \put(50,0) {\includepichtft{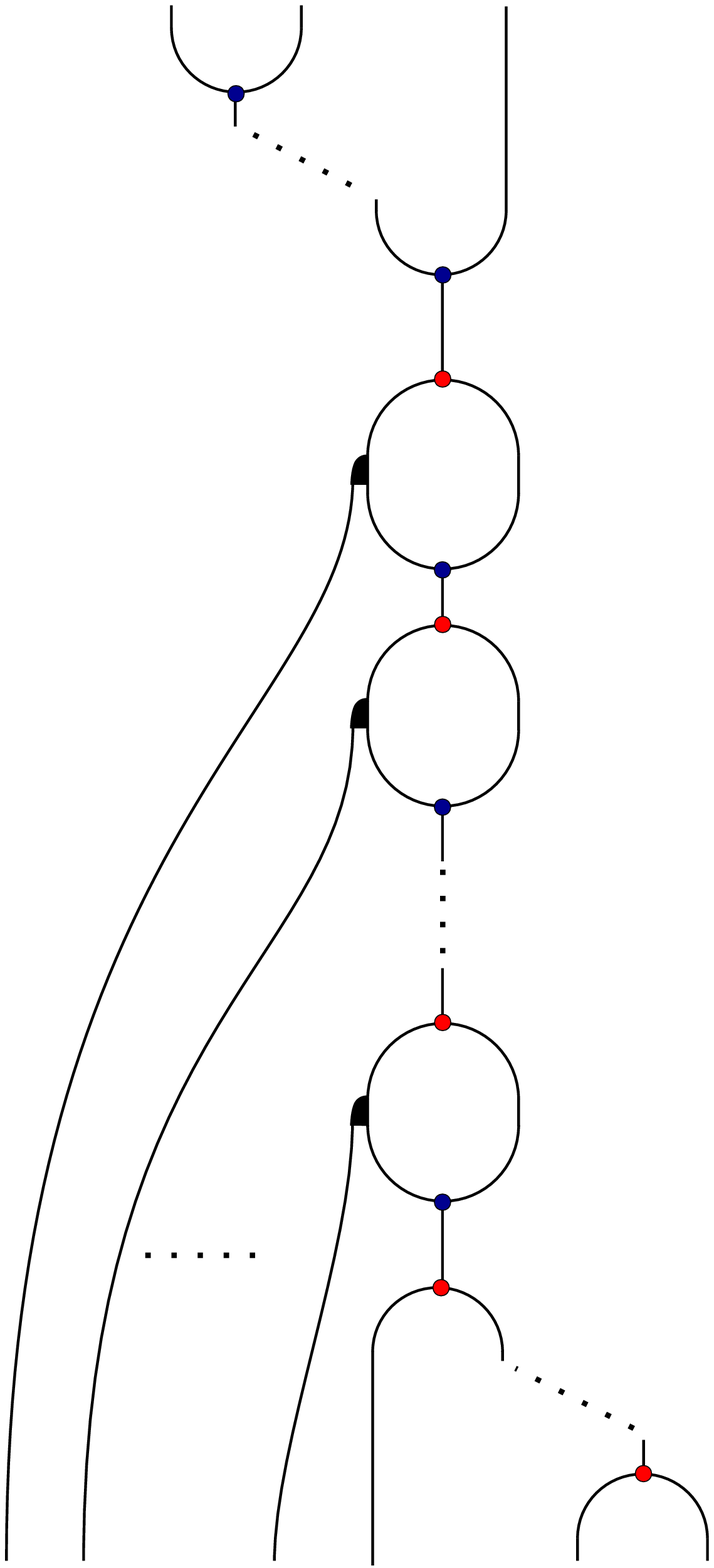}
  \put(-5,-8.5) {\sse$ \HK $}
  \put(8.2,-8.5){\sse$ \HK $}
  \put(25,-8.5){\sse$ \dots $}
  \put(43.2,-8.5){\sse$ \HK $}
  \put(60,-8.5)  {\sse$ \BF $}
  \put(75,-8.5){\sse$ \dots $}
  \put(96,-8.5)  {\sse$ \BF $}
  \put(119,-8.5)  {\sse$ \BF$}
  \put(28,278)  {\sse$ \BF $}
  \put(50,278)  {\sse$ \BF $}
  \put(66,278){\sse$ \dots $}
  \put(86,278)  {\sse$ \BF $}
  \put(47,77)     {\sse$\rho^{\HK}_{\BF}$}
  \put(47,147)     {\sse$\rho^{\HK}_{\BF}$}
  \put(47,190)     {\sse$\rho^{\HK}_{\BF}$}
  \put(130,268)   {\catpic}
  }
  \put(170,131)  {$\In\Hom_{\eC}(\HK^{\otii g}\oti\BF^{\otii p},\BF^{\otii q})$.}
  }
In this chapter, most pictures will be drawn in $\Vect$. However, some pictures, like \eqref{Sk_morph} above, will be drawn in a general factorizable finite ribbon category $\C$, or in $H$\Bimod. From here on, whenever we draw pictures in a category other than $\Vect$, this will be indicated as in  \eqref{Sk_morph}.

We define the morphism \Corrgn\ associated to the surface $\Surf gn$ to be
\be\label{CorrgnC}
    \Corrgn~:=\Sk g1n\circ(\id_{\HK^{\otii g}}\oti\eta_{\BF})\equiv \Sk g0n\,.
\ee

The morphism space $\Hom_{\eC}(\HK^{\otii g},\BF^{\otii n})$ carries a projective action, $\piKC$, of the mapping class group \Mapgn\ of $\Surf gn$, see Proposition \ref{Lyubact_prop}. We show in section \ref{sec:MGC_inv} that, in the case $\C\eq H$\Mod, the morphism $\Corrgn$ is invariant under this action for any $g,n\geq0$. To this end we work in the category $H$\Bimod, which is braided equivalent to $\HMod\rev\btimes\HMod$

More generally, we obtain a family of symmetric commutative Frobenius algebras in $H$\Bimod\ as follows:
We associate to any {ribbon Hopf algebra} automorphism $\ra$ a commutative (as well as cocommutative) symmetric Frobenius algebra $\BFw$. Thus any such algebra, $\BFw$, shares essential properties with the bulk state space of a full CFT. We show in section \ref{sec:twist_bulk} that for any surface $\Surf gn$ the morphism in $\Corrgnw\In\Hom(\HKH^{\otii g},\BFw^{\otii n})$, obtained by replacing $\BF$ by $\BFw$ in \eqref{CorrgnC}, is invariant under the projective action $\piKC$ of the mapping class group. Thus, for any $\ra$, we refer to the algebra $\BFw$ as a \bulkA\ of a full CFT whose chiral data are described by the category $H$\Mod, and the morphisms \Corrgnw\ are then candidates for correlators of bulk fields. However, we do not have a classification of the possible bulk state spaces. It is natural to explore whether a classification, similar to the one for rational theories, see \eqref{BS_FC}, also exists in the non-semisimple case. An analysis in this direction is carried out in \cite{rgaW}.

\section{Representations of \Mapgn\ from \L}\label{sec:Lyub_rep}
In order to prepare the description of the morphisms $\Corrgn$ we will discuss the coend \L\ in some more detail. Recall that when \C\ satisfies condition \CF, then the coend \L\ exists and carries the structure of a Hopf algebra in \C.
If in addition the Hopf pairing \eqref{coend_pair} is non-degenerate, i.e.\ if \C\ is factorizable ribbon,  \L\ can be equipped with a two-sided integral \cite{lyub8}. By comparing \eqref{coend_pair} and \eqref{Sij} it follows that in the semisimple case, non-degeneracy of $\omega_\L$ is equivalent to a non-degenerate $s$-matrix. For this reason, non-degeneracy of $\omega_\L$ is a natural generalization of the condition that the $s$-matrix is non-degenerate to the non-semisimple setting. For any factorizable finite ribbon category, the coend \L\ gives rise to an action of the mapping class group \Mapgn\ for any $g,n\geq0$.
\subsection{Generators of the mapping class group \Mapgn}
\index{mapping class group|textbf}
Denote, as in section \ref{BhHopf},  by $\Surf gn$ a closed oriented surface of genus $g$ with $n$ holes, labeled by objects $U_1,\dots U_n$, and by $\Mapgn$, the mapping class group of $\Surf gn$.
There are various finite presentations of \Mapgn. For checking invariance of the morphism $\Corrgn$, under the action $\piKC$, it is enough to consider the generators. One finite set of generators arises from  exact sequence
\be
  1\rightarrow B_{g,n} \rightarrow\Mapgn\rightarrow \Map_{g,0}\rightarrow 1 \,,
\ee
(compare \cite[Thm.\,9.1]{FAma}), where $B_{g,n}$ is a central extension of the
surface braid group by $\mathbb Z^n$. As a consequence of this exact sequence one can take, as a set of generators for $\Mapgn$, the union of the set of generators for a presentation of $\Map_{g,0}$ \cite{wajn}, and the one for a presentation of $B_{g,n}$ \cite{scotG3}. This set of generators is used in \cite{lyub6} and \cite{lyub11}. There are many slightly different notions of mapping class groups in the literature, see e.g.\ for a discussion \cite{FAma}. Here we include explicitly generators that interchanges two holes.
\pagebreak

\ssubject{Generators of $\Mapgn$}The mapping class group \Mapgn\ of the surface \Surf gn, depicted in \eqref{surf_gn_PIC}, is generated by (see \cite[Section 4]{lyub6} and \cite[Section 3]{lyub11}):
\begin{itemize}\addtolength{\itemsep}{-6pt}%
    \item Braidings $\wi$ interchanging the $i$'th and $i{+}1$'st hole, for $i\eq1,..., n{-}1$
    \item Dehn twists $\Ri$ around the $i$'th hole, for $i\eq1,...,n$.
    \item  Homeomorphisms $S_l$, for $l\eq1,...,g$, which act as the identity outside the $T^2\setminus D$ neighborhood $F_l$ and as an S-transformations of a one-holed torus $F_l'\subset F_l$.
    \item Inverse Dehn twists in tubular neighborhoods of the cycles $a_m$ and $e_m$ for $m\eq2,...,g$.
    \item Inverse Dehn twists in tubular neighborhoods of the cycles $b_m$ and $d_m$ for $m\eq1,...,g$.
    \item Inverse Dehn twists in tubular neighborhoods of the cycles $t_{j,k}$ for $k\eq1,\linebreak...,g$  and $j=1,...,n-1$.
\end{itemize}
The cycles and the neighborhoods $F_l$ are depicted in the following picture:
\Eqpic{surf_gn_PIC} {320} {38} {
   \put(0,0)  {\includepic{50}{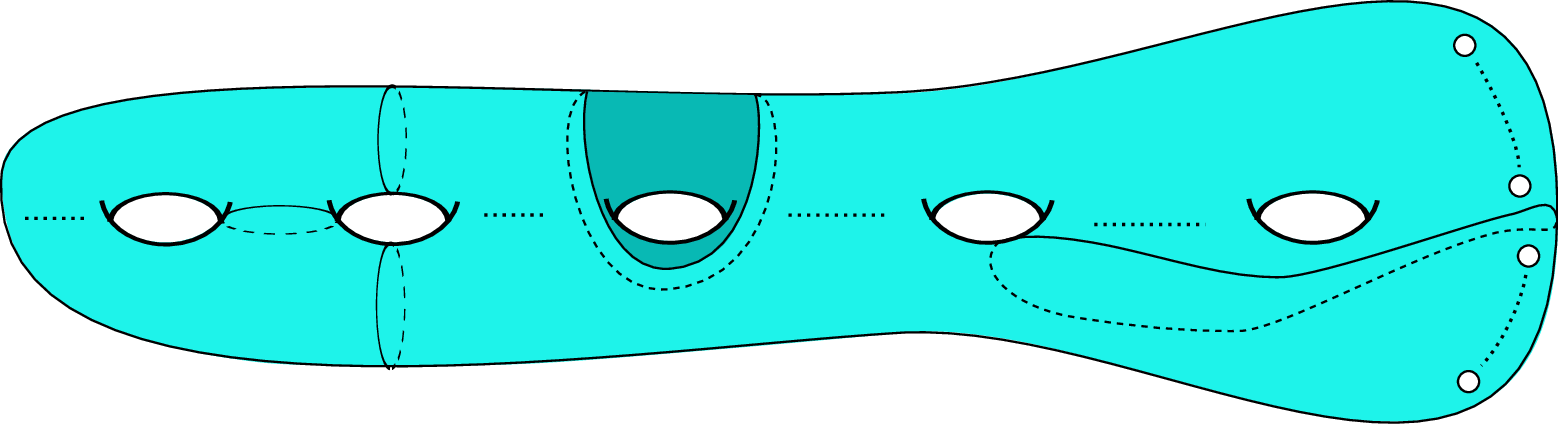}
   \put(49,46)    {\sse$a_m$}
   \put(71,38)    {\sse$b_m$}
   \put(61,54)    {\sse$d_m$}
   \put(61,22)    {\sse$e_m$}
   \put(125,54)    {\sse$S_l$}
   \put(127,38)    {\sse$b_{l}$}
   \put(188,38)    {\sse$b_{k}$}
   \put(253,39)    {\sse$b_{g}$}
   \put(223,26)    {\sse$t_{j,k}$}
   \put(291,80)    {\sse$U_1$}
   \put(307,44)    {\sse$U_j$}
   \put(307,31)    {\sse$U_{j+1}$}
   \put(293,-1)    {\sse$U_n$}
   } }
For brevity we refer to Dehn twists around any of these cycles by the same symbol as the cycle itself.
\subsection{Partial monodromies}
We can use the dinatural family $\iota^\L$, of the coend \L, to define endomorphisms of $\L\oti\L$ and $\L\oti Y$ for any $Y\In\Obj(\C)$. We define the \emph{partial monodromy} $\QQ_{\L,\L}\linebreak\In\End(\L\oti\L)$ by
\eqpic{QHH} {135} {42} {
   \put(0,-3)  {\Includepichtft{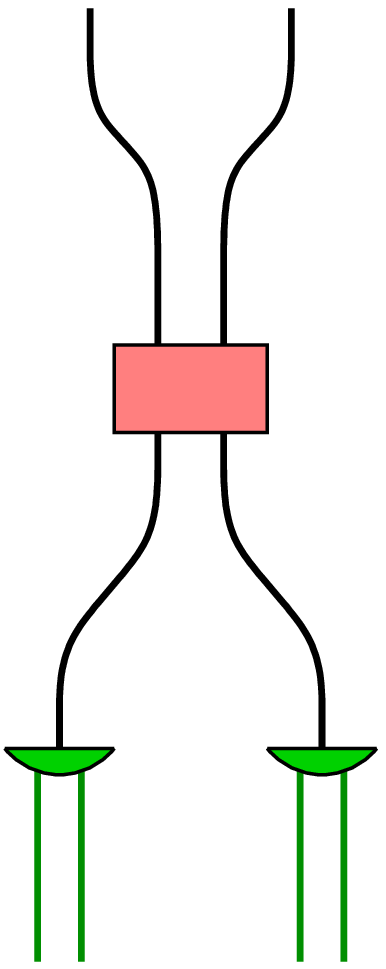}
   \put(-15.9,60.7){\small$ \QQ_{\L,\L} $}
   \put(-9.8,21)  {\sse$ \iota_X^\L $}
   \put(-4,-8.5)  {\sse$ X^{\!\vee} $}
   \put(5.3,109)  {\sse$ \L $}
   \put(7,-8.5)   {\sse$ X $}
   \put(26.5,-8.5){\sse$ Y^{\!\vee} $}
   \put(28.8,109) {\sse$ \L $}
   \put(38,-8.5)  {\sse$ Y $}
   \put(42.4,21)  {\sse$ \iota_Y^\L $}
   }
   \put(62,42)    {$ := $}
   \put(97,-3) {\Includepichtft{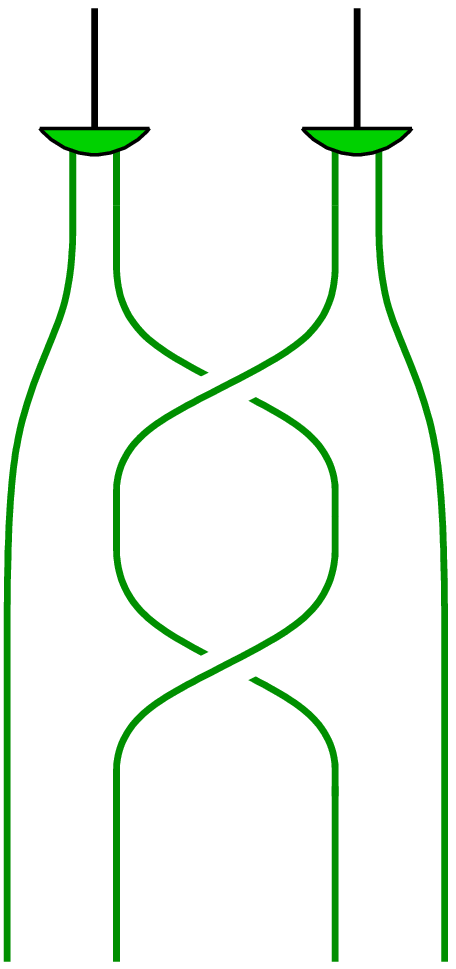}
   \put(-6,-8.5)  {\sse$ X^{\!\vee} $}
   \put(-5.8,89)  {\sse$ \iota_X^\L $}
   \put(8,-8.5)   {\sse$ X $}
   \put(6.1,109)  {\sse$ \L $}
   \put(22.8,27.3){\sse$ c $}
   \put(22.8,58.4){\sse$ c $}
   \put(32,-8.5)  {\sse$ Y^{\!\vee} $}
   \put(36,109)   {\sse$ \L $}
   \put(46.2,89)  {\sse$ \iota_Y^\L $}
   \put(46,-8.5)  {\sse$ Y $}
   \put(70,100)    {\catpic}
   }
   }
Similarly we have the \emph{left partial monodromy} $\Qq_{\L,Y}\In\End(\L\oti Y)$, defined for any $Y\In\Obj(\C)$ by
\eqpic{QHX} {110} {41} {
   \put(0,0)  {\Includepichtft{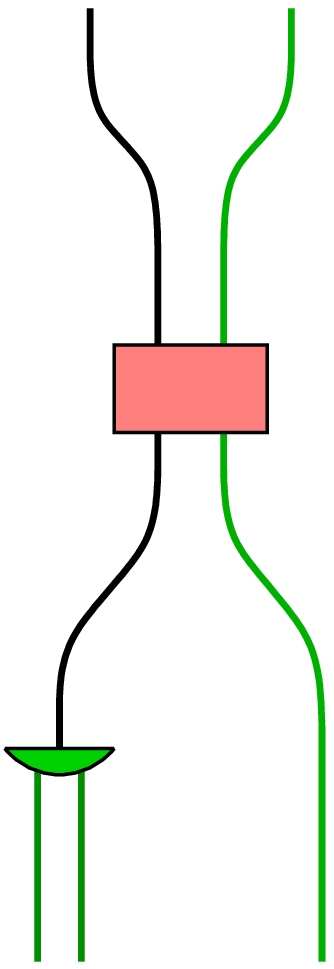}
   \put(-15,60)  {$ \Qq_{\L,Y} $}
   \put(-4,-8.5) {\sse$ X^{\!\vee} $}
   \put(5.3,109) {\sse$ \L $}
   \put(7,-8.5)  {\sse$ X $}
   \put(31.3,-8.5){\sse$ Y $}
   \put(30.4,109){\sse$ Y $}
   }
   \put(60,50)   {$ := $}
   \put(94,0) {\Includepichtft{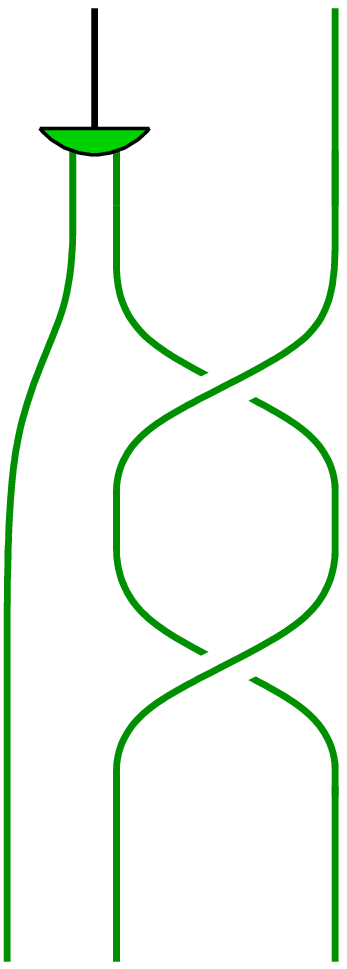}
   \put(-6,-8.5) {\sse$ X^{\!\vee} $}
   \put(6.3,109) {\sse$ \L $}
   \put(8,-8.5)  {\sse$ X $}
   \put(22.5,27.3) {\sse$ c $}
   \put(22.5,58.4) {\sse$ c $}
   \put(34,-8.5) {\sse$ Y $}
   \put(34.3,109){\sse$ Y $}
   \put(70,100)    {\catpic}
   }
   }
Since the composition of the partial monodromy $\Qq_{\L,Y}$ with a member $\iota^\L_X$ of the dinatural family is nothing but an ordinary monodromy between $X$ and $Y$, it follows that
\be\label{Qq_com}
    \Qq_{\L,Y}\circ(\id_{\L}\oti f)=(\id_{\L}\oti f)\circ\Qq_{\L,Z}\,,
\ee
for any $f\in\Hom(Z,Y)$.

Note that, with $\eps_\L$ the counit of \L\ (see\eqref{Lyub_Hopf}), the Hopf pairing \eqref{coend_pair} can be written as
\be\label{Hopf_pair_eps}
    \omega_\L=(\eps_\L\oti\eps_\L)\circ\QQ_{\L,\L}\,.
\ee
In addition, the action \eqref{Lact} of \L\ on $V\In\Obj(\C)$ can be written in terms of $\Qq_{\L,Y}$ as
\be\label{Lact_Qq}
    \Lact V=(\eps_{\L}\oti\id_Y)\circ\Qq_{\L,Y}\,.
\ee

Using the partial monodromy $\QQ_{\L,\L}$ we construct two additional morphisms $\TL$ and $\SL$ in $\End(\L)$. First, we define a morphism $\TL$ via the dinatural family
\be\label{def_TL}
    \TL\circ\iota_U^\L:=\iota_U^\L\circ(\theta_{U^\vee}\oti\id_U)\,.
\ee
Second, using the partial monodromy $\QQ_{\L,\L}$, the counit and two-sided integral of $\L$, the morphism $\SL$ is given by
\be\label{def_SK}
    \SL:=(\eps_L\oti\id_\L)\circ\QQ_{\L,\L}\circ(\id_\L\oti\Lambda_\L)\,.
\ee
It can be shown \cite{lyub8} that
\be
    (\SL\,\TL)^3=\lambda\SL^2\quand\SL^2=\apo_\L^{-1}\,,
\ee
with $\apo_\L$ the antipode of $\L$, for some scalar $\lambda$ that depends only on the category \C\ in question.
Together with the partial monodromies, the morphisms $\SK$ and $\TK$
will be central ingredients in the projective representations of mapping class, to which we now turn our attention.
\pagebreak
\subsection{Actions of mapping class groups}\label{sec:map_act}
We are now in a position to describe the projective actions of mapping class groups associated to \L.
Consider a surface $\Surf gn$ with $n$ holes labeled by objects $U_1,...,U_n$. Denote by $\mathfrak N=\mathfrak N(U_1,...,U_n)$ the subgroup of the symmetric group $\mathfrak S_n$ that
is generated by those permutations $\sigma\In\mathfrak S_n$ for which $U_i$ and $U_{\sigma(i)}$ are non-isomorphic for at least on $i=1,...,n$. We define the object
\be\label{Lyubspacedef}
    \Lyubspace:=\bigoplus_{\sigma\In \mathfrak N}U_{\sigma(1)}\oti U_{\sigma(2)}\oti\cdots\oti U_{\sigma(n)}\,.
\ee

To prepare the description of the action of $\Mapgn$ we introduce a collection of morphisms, one for each generator $\gamma$ of the mapping class group $\Mapgn$. First, for any integer $g\geq1$ and any $\gamma=\ak,\ek,\bk,\dk,S_k$, with $k=2,...,g$ and $k=1,...,g$ respectively, we define the endomorphisms
\be\label{LyubactC_pre}
\begin{split}
    \LAak&~:=~\id_{\L^{\otii g-k}}\oti[\QQ_{\L,\L}\circ(\TL\oti\TL)]\oti \id_{\L^{\otii k-2}},\\
    \LAek&~:=~\id_{\L^{\otii g-k}}\oti(\TL\oti\theta_{\L^{\otii k-1}})\circ\Qq_{\L,\L^{\otii k-1}},\\
    \LAbk&~:=~\id_{\L^{\otii g-k}}\oti(\SL^{-1}\circ\TL\circ\SL)\oti\id_{\L^{\otii k-1}},\\
    \LAdk&~:=~\id_{\L^{\otii g-k}}\oti\TL\oti\id_{\L^{\otii k-1}},\\
    \LASk&~:=~\id_{\L^{\otii g-k}}\oti\SL\oti\id_{\L^{\otii k-1}},\\
\end{split}
\ee
in $\End(\L^{g})$. Second, we define the endomorphisms $\LAwi$, with $i\eq1,...,n-1$, and $\LARi$, with $i\eq1,...,n$, of $\Lyubspace$, that act as
\be\label{LyubactC_post}
\begin{split}
    \LAwi\big|_{U_1\oti\cdots\oti U_n}&~=~\id_{U_1}\oti\cdots c_{U_i,U_{i+1}}\oti\cdots\oti\id_{U_n},\\
    \LARi\big|_{U_1\oti\cdots\oti U_n}&~=~\id_{U_1}\oti\cdots \theta_{U_i}\oti\cdots\oti\id_{U_n},\\
\end{split}
\ee
on the direct summand $U_1\oti\cdots\oti U_n$ of $\Lyubspace$. Finally, for any $j\eq1,...,n-1$ and any $k\eq1,...,g$ we define the linear map $\LAtjk$ that maps any $f\In\Hom(\L^{\otii g},U_1\linebreak\oti...\oti U_n)$ to
\be\label{tjkact}
  \bearll
  \LAtjk(f) := \!\!& \big(\, \big[\, ( \id_{U_1\otimes\cdots\otimes U_j}
  \oti \tilde d_{U_{j+1}\otimes\cdots\otimes U_n} )
  \circ
  ( f \oti \id_{\ld U_n\otimes\cdots\otimes \ld U_{j+1}} )
  \\{}\\[-.6em]& \quad
  \circ\, \{ \id_{\L^{\otimes g-k}_{}} \oti [
  \Qq_{\L^{\otimes k-1}\otimes \ld U_n\otimes\cdots\otimes \ld U_{j+1}}
  \\{}\\[-.6em]& \quad
  \circ\,
  (\TL \oti \theta_{\L^{\otimes k-1}\otimes \ld U_n\otimes\cdots\otimes
  \ld U_{j+1}}) ] \} \,\big]
  \otimes\, \id_{U_{j+1}\otimes\cdots\otimes U_n} \,\big)
  \\{}\\[-.6em]& \hspace*{15.8em}
  \circ\, \big(
  \id_{\L^{\otimes g}_{}} \oti \tilde b_{U_{j+1}\otimes\cdots\otimes U_n} \big)\,,
  \eear
\ee
in $\Hom(\L^{\otii g},U_1\oti...\oti U_n)$, and acts analogously on a any morphism in\linebreak $\Hom(\L^{\otii g},U_{\sigma(1)}\oti...\oti U_{\sigma(n)})$ for any $\sigma\In\mathfrak N$.

Pictorially, $\LAtjk$ acts as
\eqpic{t_act} {300} {140} {
    \put(0,140)  {$\LAtjk(f):=$}
    \put(85,0) {\Includepichtft{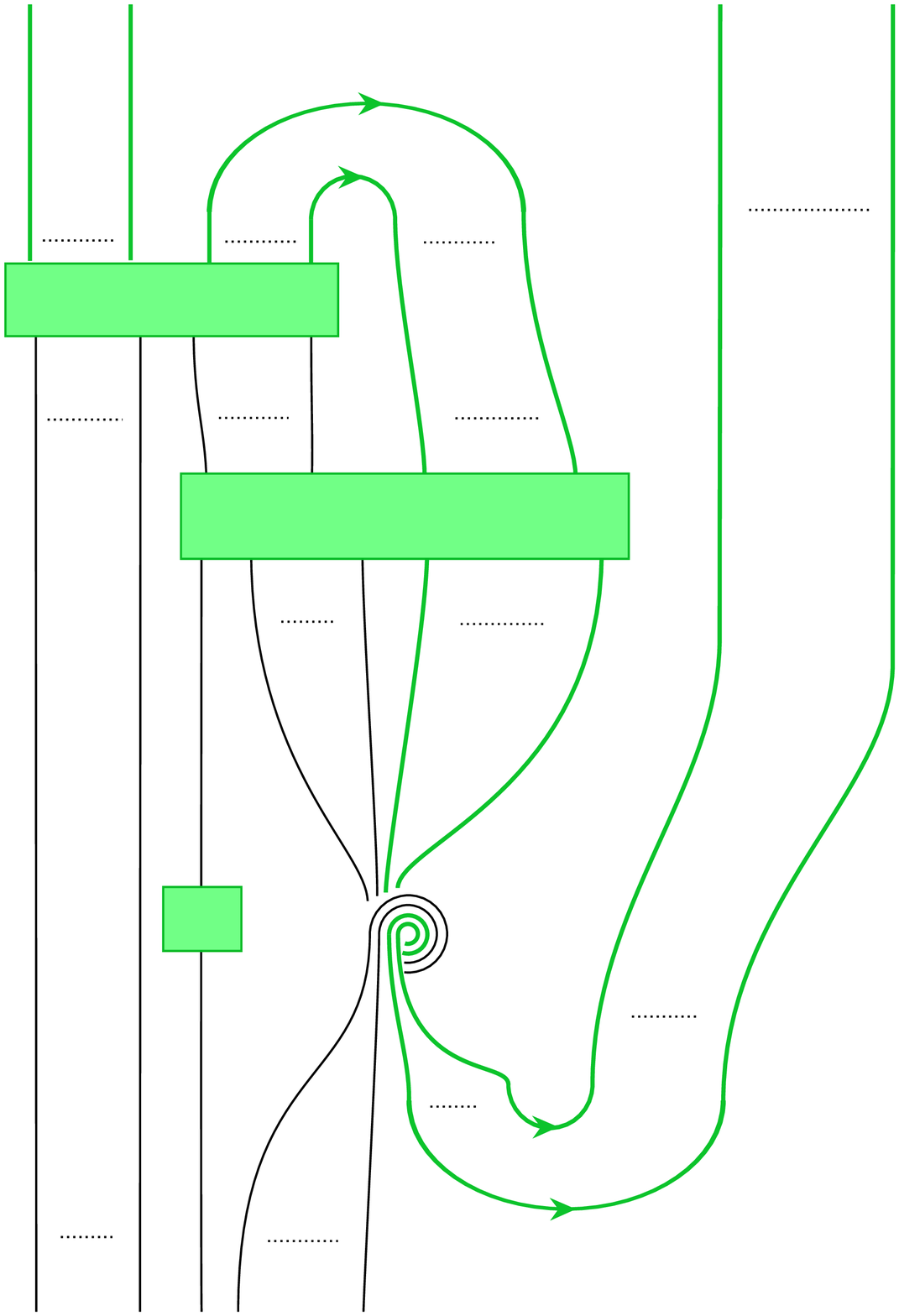}
    \put(32,222)   {\sse$f$}
    \put(40,176.5)   {\sse$\Qq_{\L,\L^{\otii k-1}\oti^\vee\!U_{n}\!\oti\!\cdots\oti^\vee\!U_{j+1}}$}
    \put(37,86)   {\sse$\TK$}
    \put(5,-8.5)   {\sse$\L$}
    \put(13,-8.5)   {\sse$\dots$}
    \put(28,-8.5)   {\sse$\L$}
    \put(42,-8.5)   {\sse$\L$}
    \put(50,-8.5)   {\sse$\L$}
    \put(60,-8.5)   {\sse$\dots$}
    \put(78,-8.5)   {\sse$\L$}
    \put(1,294)   {\sse$U_1$}
    \put(12,294)   {\sse$\dots$}
    \put(23,294)   {\sse$U_j$}
    \put(148,294)   {\sse$U_{j+1}$}
    \put(172,294)   {\sse$\dots$}
    \put(190,294)   {\sse$U_{n}$}
    \put(210,300)    {\catpic}
    }
    }

We can now rephrase the results of \cite[Section 4]{lyub6} and \cite[Section 3]{lyub11} as follows
\begin{proposition}\label{Lyubact_prop}\index{mapping class group!representation of|textbf}
Consider the morphism space $\Hom_\C(\L^{\otii g}, \Lyubspace)$, with $\Lyubspace$ as defined in \eqref{Lyubspacedef}.
The morphisms $\LAg$, defined in \eqref{LyubactC_pre}, \eqref{LyubactC_post} and \eqref{tjkact}, associated to the generators $\gamma$ of $\Mapgn$, endow the space $\Hom_\C(\L^{\otii g}, \Lyubspace)$ with a projective action $\piL$ of \Mapgn. For any $f\In\Hom_\C(\L^{\otii g}, \Lyubspace)$, and for $\gamma=\ak,\ek$ with $k=2,...,g$ and for $\bk,\dk,S_k$ with $k=1,...,g$, the action is given by
\be\label{piLpre}
    [\piL(\gamma)](f)=f\circ\LAg^{-1}\,.
\ee
For $\gamma\eq \wi,\Ri$, with $i=1,...,n-1$ and $i=1,...,n$ respectively, the action is given by
\be\label{piLpost}
    [\piL(\gamma)](f)=\LAg^{-1}\circ f\,.
\ee
Finally, for $\gamma\eq\tjk$, with $j\eq1,...,n-1$ and $k\eq1,...,n$, the action of $\tjk$ is given by
\be\label{piLtjk}
    [\piL(\tjk)](f)=\LAtjk^{-1}(f)\,.
\ee
\end{proposition}
\begin{remark}\label{rem:repio}
    Consider the subgroup $\Mapio g{p,q}$ of the mapping class group \linebreak$\Mapio g{p+q}$ that leaves two subsets, of size $p$ and $q$, of the holes invariant. Take the labels of the two sets of holes to be $U_1,...,U_p$ and $V_1,...,V_q$, respectively,  and define object $\lyubspace Up$ and $\lyubspace Vq$ analogously as in \eqref{Lyubspacedef}. The right duality of $\C$ provides a linear isomorphism
    \be
    \varphi:\quad \Hom(\L^{\otii g}\oti\,V,U) \stackrel\cong\longrightarrow
  \Hom(\L^{\otii g},U\oti V^\vee) \,.
  \ee
  Then
  \be\label{def_repio}
    \PiLio pq:=\varphi^{-1} \circ \piLio pq \circ \varphi
  \ee
  defines a projective representation of $\Mapio g{p,q}$ on $\Hom(\L^{\otii g}\oti\,V,U)$.
\end{remark}
\section{The \bulkAC\ in $H$\Bimod}\label{Frob_Hmod}
We now turn to the description of the \bulkAC\ \BF\ in the case $\C\eq H$\Mod\ for $H$ a finite-dimensional factorizable ribbon Hopf algebra over an algebraically closed field of characteristic zero.
As we will explain below, we will be able to work in $H$\Bimod\ instead of $(H\Mod)\rev\boxti H\Mod$.

Consider the functor
\be\label{def_Fbx}
    \Fbx :~~ H\Mod\op \Times H\Mod \,\to\, H\Bimod\,,
\ee
that acts on objects as
\be\label {Fbx_act}
  \bearl\dsty
  (X,\rho_X) \Times (Y,\rho_Y)
  \nxl{1}\dsty \hsp{2}
  \stackrel\Fbx\longmapsto~
   \big( X^\vee\otik Y\,,\, \rho_{X^\vee_{}}\oti\id_Y\,,\, \id_{X^*}
   \oti(\rho_Y\cir \tau_{Y,H}^{} \cir (\id_Y\oti\apo^{-1})) \big)\,,
  \eear
\ee
and on  morphisms as $f \times g\mapsto f^\vee\Otik g$.
Pictorially, the action on objects looks as
\Eqpic{Fboxmorph} {320} {42} { \put(0,10){
    \setulen80
   \put(0,0)   {\INcludepichtft{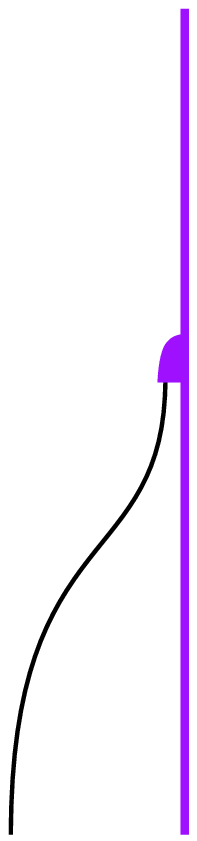}{304}
   \put(-4,-8.8) {\sse$H $}
   \put(15.1,-8.8) {\sse$ X $}
   \put(16.2,95.5) {\sse$ X $}
   \put(22,52)   {\sse$\rho_X^{}$}
   }
   \put(48,38)   {$ \times$}
   \put(73,0) { \INcludepichtft{66a}{304}
   \put(-4,-8.8) {\sse$H $}
   \put(16.1,-8.8) {\sse$ Y $}
   \put(17.2,95.5) {\sse$ Y $}
   \put(22,52)   {\sse$\rho_Y^{}$}
   }
   \put(146,38)  {$\stackrel\Fbx\longmapsto$}
   \put(200,0)   {\INcludepichtft{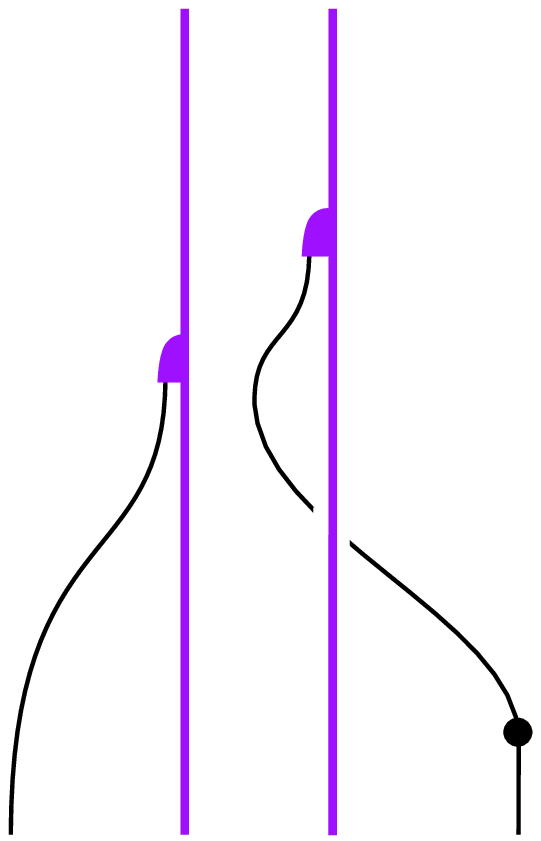}{304}
   \put(-4,-8.8) {\sse$H $}
   \put(13,-8.8) {\sse$X^\vee$}
   \put(14,95.5) {\sse$X^\vee$}
   \put(32,-8.8) {\sse$ Y $}
   \put(33,95.5) {\sse$ Y $}
   \put(51.5,-8.8) {\sse$ H $}
   \put(59,10.1) {\sse$ \apo^{-1} $}
   }
   \put(278,38)   {$ \equiv $}
   \put(308,0) { \INcludepichtft{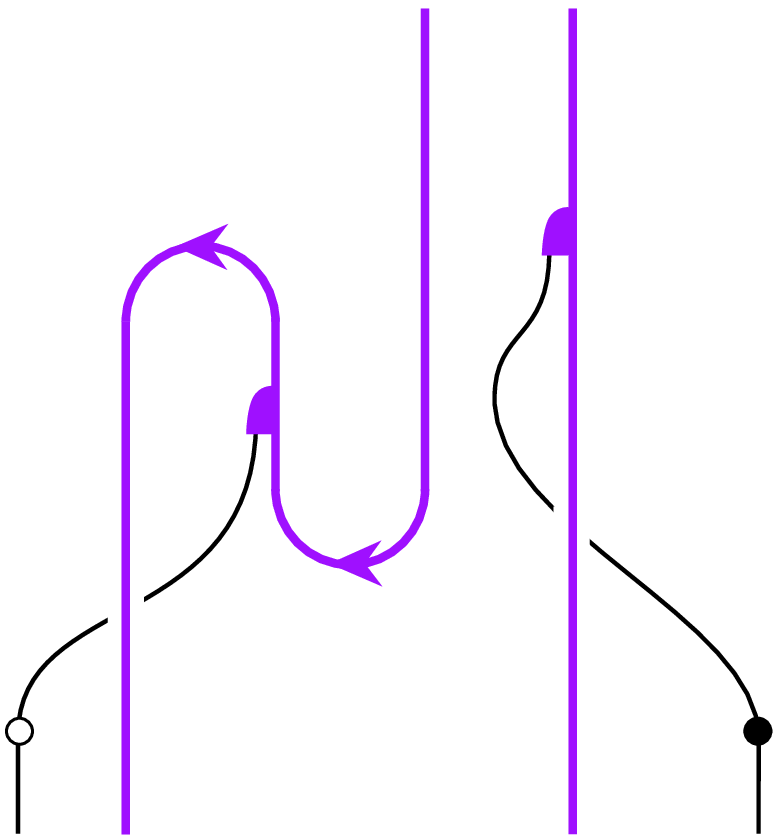}{304}
   \put(-4.7,9.8){\sse$ \apo $}
   \put(-3,-8.8) {\sse$ H $}
   \put(8.5,-8.8){\sse$ \Xs $}
   \put(31,46)   {\sse$ \rho_X^{} $}
   \put(41.6,95.5) {\sse$ \Xs $}
   \put(59,-8.8) {\sse$ Y $}
   \put(60,95.5) {\sse$ Y $}
   \put(63.7,65.6) {\sse$ \rho_Y^{} $}
   \put(79,-8.8) {\sse$ H $}
   \put(87,10.1) {\sse$ \apo^{-1} $}
   } }
   }
For $\C\eq H$\Mod, the \bulkAC\ $\BF$ is the coend
\be\label{bulkF_H}
    \BFH:=\int^X\,\Fbx(X,X)
\ee
in $H$\Bimod.
The next subsection motivates this terminology. That is, that the coend $\BFH$ in \eqref{bulkF_H} should be considered as a special case of the coend $\BF=\int^X X^\vee\boxtimes X$, introduced in \eqref{BulkA_coend}.

\subsection{Equivalences of categories}\label{sec:equiv_cat}
Denote by $H\coop$, the Hopf algebra obtained from $H$ by replacing the coproduct by the opposite product, as in \eqref{Rmat1}. Analogously $H\op$ is obtained from $H$ by replacing the product by the opposite product $m\op:=m\cir\flip HH$.
When $H$ is quasitriangular with $R$-matrix $R$ there are, according to \cite[Lemma A.4]{fuSs3}, equivalences
\be\label{equi_bim_tens}
     H\Bimod \;\simeq\, (\overline H\Otik H\coop)\Mod
   \;\simeq\, (\overline H\Otik H\op)\Mod
\ee
of braided monoidal categories. Here $H$\Bimod\ is equipped with the braiding \eqref{bibraid} and $\overline H$ is obtained from $H$ by replacing the R-matrix by $R_{21}^{-1}$ and the ribbon element $v$ by the inverse ribbon element $v^{-1}$. The equivalence is furnished by two functors which act as the identity on morphisms, and on objects as
\Eqpic{HHbimiso} {320} {39} {
    \put(0,10){\setulen90
   \put(5,0)   {\INcludepichtft{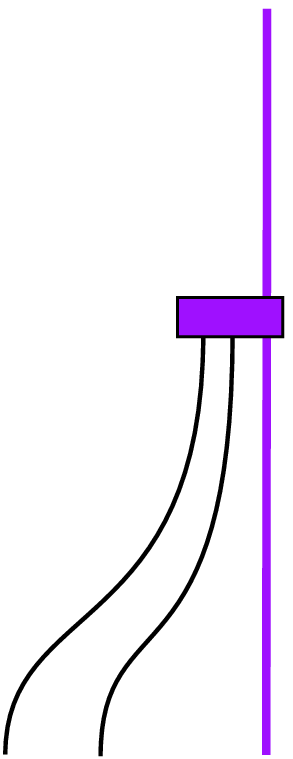}{342}
   \put(-4,-8.5) {\sse$ H $}
   \put(6,-8)    {\sse$ H$}
   \put(24.1,-8.5) {\sse$ M $}
   \put(24.8,86) {\sse$ M $}
   \put(-5.5,47) {\sse$ \rho^{H\Oti H} $}
   }
   \put(63,34)   {$ \mapsto$}
   \put(105,0) { \INcludepichtft{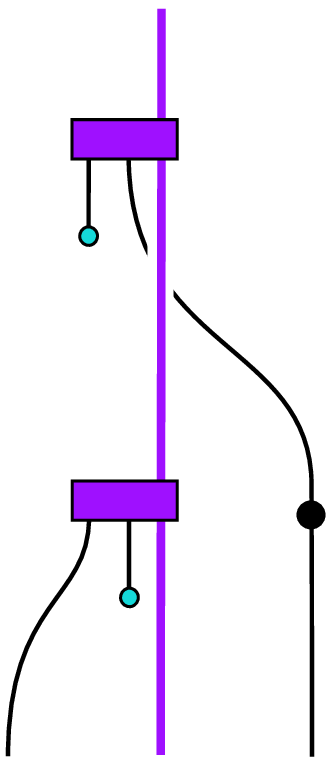}{342}
   \put(-4,-8.5) {\sse$ H $}
   \put(12.3,-8.5) {\sse$ M $}
   \put(12.9,86) {\sse$ M $}
   \put(29,-8.5) {\sse$ H$}
   \put(36,25)   {\sse$ \apo^{-1} $}
   \put(-17.2,26.4){\sse$ \rho^{H\Oti H} $}
   \put(-17.2,66.4){\sse$ \rho^{H\Oti H} $}
   }
   \put(180,34)  {and as}
   \put(230,0) {\INcludepichtft{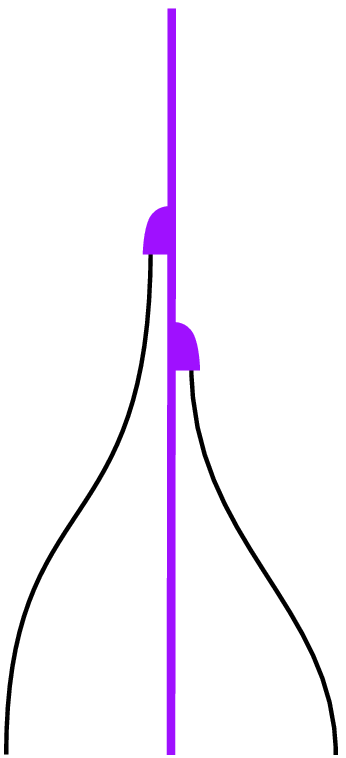}{342}
   \put(-4,-8.5) {\sse$H $}
   \put(5.2,55.4){\sse$ \rho^H $}
   \put(13,-8.5) {\sse$ M $}
   \put(13.8,86) {\sse$ M $}
   \put(22.2,43.4){\sse$ \ohr^{\!H} $}
   \put(32,-8.5) {\sse$ H$}
   }
   \put(290,34)   {$ \mapsto$}
   \put(322,0) { \INcludepichtft{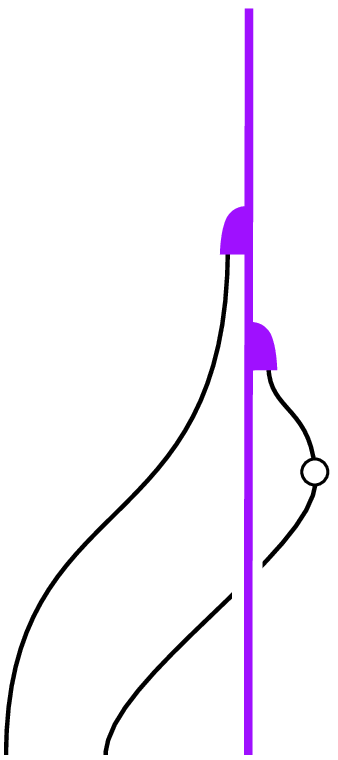}{342}
   \put(-4,-8.5) {\sse$ H $}
   \put(7,-8.5)  {\sse$ H$}
   \put(22,-8.5) {\sse$ M $}
   \put(22.8,86) {\sse$ M $}
   \put(36,29.5) {\sse$ \apo $}
   } } }
respectively.
In addition, the $R$-matrix furnishes a braided monoidal equivalence between $H$\Mod\ and $H\coop$\Mod. Thus, combined with \eqref{equi_bim_tens} we also have an equivalence
\be\label{equi_bim_tens_H}
     H\Bimod \;\simeq\, (\overline H\Otik H)\Mod
\ee
as braided monoidal categories.

Next, by the universal property \eqref{uni_deligne} of the Deligne product there is a unique functor
\be\label{Fbb_fun}
    \Fbb : \quad H\Mod \boxti H\Mod \,\longrightarrow\, H\Bimod
\ee
such that
\be\label{Fbb_split}
    \Fbx = \Fbb \circ (?^\vee \boxti \Id)\,,
\ee
where $(?^\vee \boxti \Id)$ is the functor \eqref{BA_fun}. Note that on objects of the form $X\boxti Y\linebreak\In H\Mod \boxti H\Mod$, $\Fbb$ acts as
\be
  \bearl\dsty
  (X,\rho_X) \boxti (Y,\rho_Y)
  \nxl{1}\dsty \hsp{3}
  \stackrel\Fbb\longmapsto~
   \big( X\otik Y\,,\, \rho_X\oti\id_Y\,,\, \id_X\oti(\rho_Y\cir \tau_{Y,H}
   \cir (\id_Y\Oti\apo^{-1})) \big) \,.
  \eear
\ee

By \cite[Proposition 5.3]{deli} (see also \cite[Example 7.10]{franI}) there is an equivalence $H\Mod \boxti H\Mod\,{\simeq}\, (H\Otik H)\Mod$ of abelian categories. Combining this equivalence with the equivalence \eqref{equi_bim_tens_H}, it follows that the functor \eqref{Fbb_fun} furnishes an equivalence of abelian categories. Next, remember from section \ref{sec:Deligne} that $H\Mod \boxti H\Mod$ can be equipped with the structure of a braided monoidal category, c.f. \eqref{oti_Del} and \eqref{c_Deli}. With respect to this tensor product and braiding, \eqref{Fbb_fun} extends to an equivalence of braided monoidal categories.
In view of this equivalence of categories and the functors \eqref{Fbb_split} and \eqref{BulkA_coend}, the coend $\BFH$, as defined in \eqref{bulkF_H}, is indeed the \bulkAC\ in $H$\Bimod.
\begin{remark}\label{rem:bibraid_choice}
    Recall that we want to describe $(\HMod)\rev\boxtimes\HMod$. The arguments above imply that $(\HMod)\rev\boxtimes\HMod$ is equivalent, as a braided monoidal category, to $H$\Bimod\ with the braiding defined in \eqref{bibraid}.
\end{remark}
\subsection{The coregular bimodule as a coend}
Let us now describe the coend $\BFH$, defined in \eqref{bulkF_H}, in detail.
As an object in $H$\Bimod, $\BFH$ is the coregular bimodule, i.e.\ the vectors space \Hs\ endowed with the actions \eqref{rhoHb,ohrHb}. Indeed we have, \cite[Lemma A.2]{fuSs3}:
\begin{lemma}
The family $\iHb$ of morphisms
\be\label{def_iHb}
   \iHb_X := (d_X \oti \id_{\Hs}) \cir [ \id_{X^*_{}} \oti (\rho_X^{}\cir\tau_{X,H})
   \oti \id_{\Hs}] \cir ( \id_{X^*_{}} \oti \id_X \oti b_H)\,
\ee
with $ X\In H\Mod$ is a dinatural family from $\Fbx$ to the coregular bimodule $\BFH$ in $H$\Bimod.
\end{lemma}
Pictorially, $\iHb_X$ is given by
\eqpic{def_iA_X_bi} {100} {35} {
   \put(-4,0)  { \Includepichtft{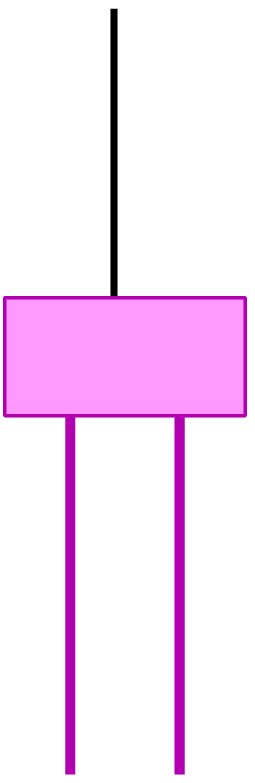}}
   \put(2,-8.5)  {\sse$ \Xs $}
   \put(8.9,44.9){\sse$ \iHb_X $}
   \put(9.2,89)  {\sse$ \Hss $}
   \put(16,-8.5) {\sse$ X $}
   \put(49,38)   {$ = $}
   \put(77,0) { \Includepichtft{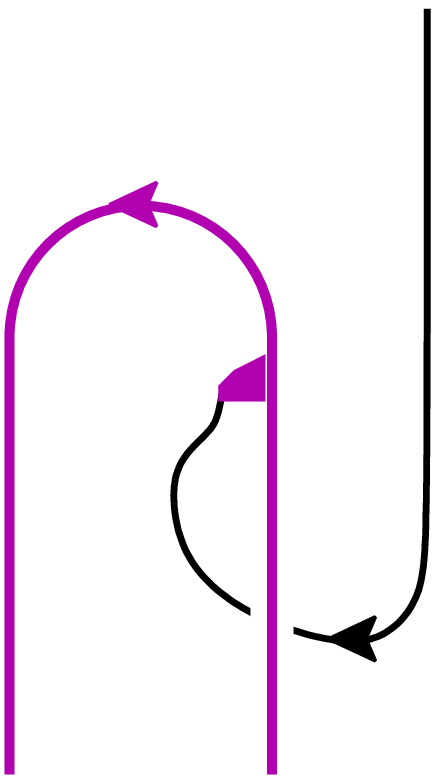}
   \put(-4,-8.5) {\sse$ \Xs $}
   \put(25,-8.5) {\sse$ X $}
   \put(31,43)   {\sse$ \rho_{\!X}^{} $}
   \put(43,89)   {\sse$ \Hss $}
   } }
That $\iHb_X$ is a morphism of bimodules follows directly by using the action \eqref{Fboxmorph} of $\Fbx$ on objects, the  representation properties and \eqref{rhoHb,ohrHb}, see \cite[eq. (A.8) and (A.9)]{fuSs3}. That $\iHb_X$ is dinatural, i.e.\ satisfies \eqref{def_dina} for any $f$ in $H\Mod$ follows directly from the action of $\Fbx$ on morphisms and that $f$ is a morphism of $H$-modules, see \cite[eq. (A.10)]{fuSs3}. In addition we have the following result:
\begin{prop}\label{prop-F=coend}
The $H$-bimodule \BFH\ together with the dinatural family $(\iHb_X)$
given by \eqref{def_iA_X_bi} is the coend of the functor \Fbx:
   \be
   (\BFH,\iHb) = \coen X \Fbx(X,X) \,.
    \ee
\end{prop}
\begin{proof}
   We have to show that the dinatural family $\iHb$ satisfies \eqref{coend_uni}. Thus, let $j^Z$ be any dinatural transformation from \Fbx\ to $Z\In H\Bimod$. Given any $X\In H\Mod$, take any element $x_\circ$ in $X$, described as $x_\circ\In \Homk(\k,X)$, and consider the morphism $f_{x_\circ} \,{:=}\, \rho_X \cir (\id_H \oti x_\circ) \In \HomH(H,X)$. Here $H$ is regarded as an $H$-module via the left regular action. Applying dinaturalness to $f_{x_\circ}$ yields
\eqpic{rZ_etc_2B} {300} {38} { \setlength\unitlength{.8pt}
   \put(0,0)   { \includepichtft{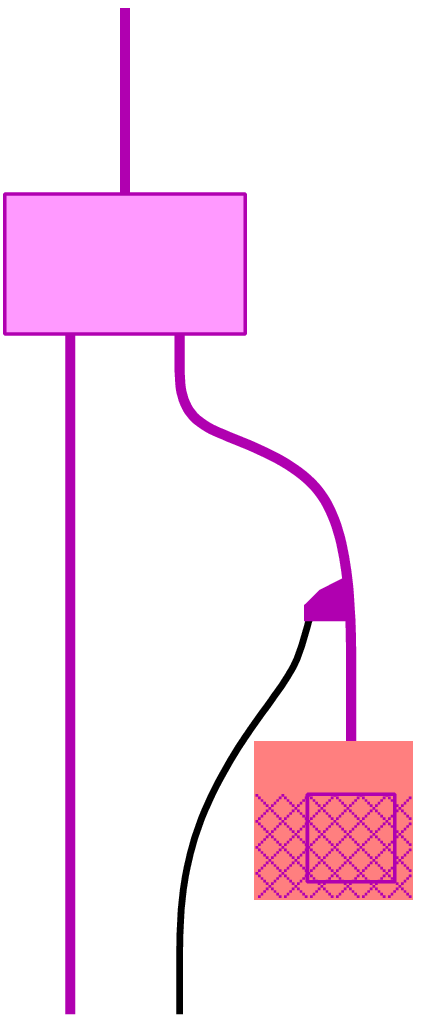}
   \put(0.2,-8.5) {\sse$ \Xs $}
   \put(8.8,81.1) {\sse$ j^Z_{\!X^{}} $}
   \put(10.6,115) {\sse$ Z $}
   \put(16,-8.5)  {\sse$ H $}
   \put(45.3,18)  {\sse$ x_\circ $}
   }
   \put(77,51)   {$ \equiv $}
   \put(108,0) { \includepichtft{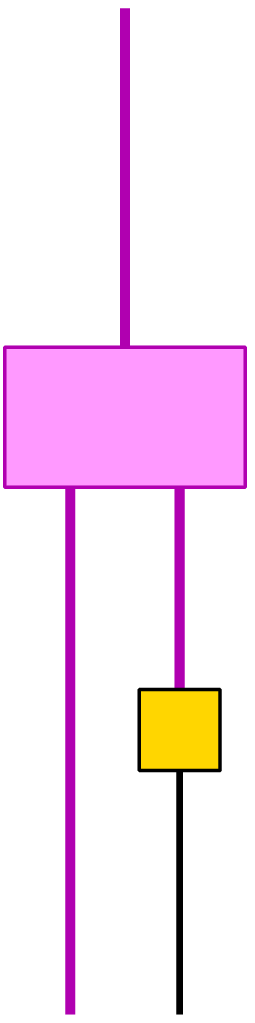}
   \put(0.2,-8.5) {\sse$ \Xs $}
   \put(9.2,64.1) {\sse$ j^Z_{\!X^{}} $}
   \put(10.6,115) {\sse$ Z $}
   \put(16,-8.5)  {\sse$ H $}
   \put(25.4,30)  {\sse$ f_{x_\circ} $}
   }
   \put(174,51)  {$ = $}
   \put(207,0) { \includepichtft{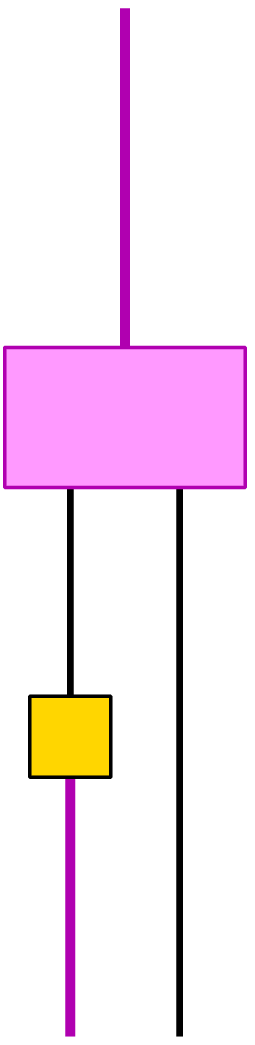}
   \put(-12.6,32.3) {\sse$ f_{x_\circ}^* $}
   \put(0.2,-8.5) {\sse$ \Xs $}
   \put(8.8,67.1) {\sse$ j^Z_{\!H^{}} $}
   \put(10.6,117) {\sse$ Z $}
   \put(16,-8.5)  {\sse$ H $}
   }
   \put(270,51)  {$ \equiv $}
   \put(303,0) { \includepichtft{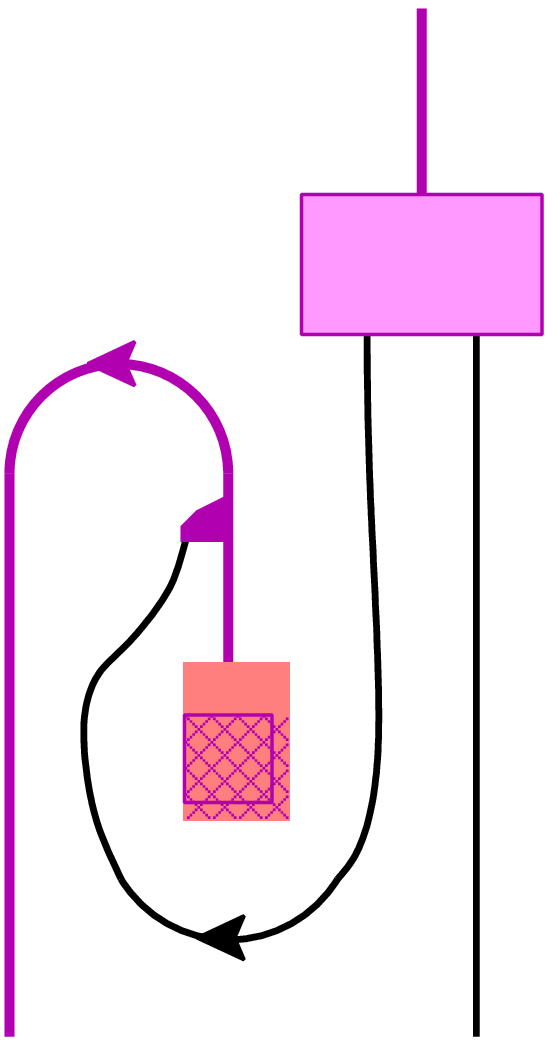}
   \put(-5,-8.5)  {\sse$ \Xs $}
   \put(20.7,20.9){\sse$ x_\circ $}
   \put(41.5,83.4){\sse$ j^Z_{\!H^{}} $}
   \put(43.2,117) {\sse$ Z $}
   \put(49,-8.5)  {\sse$ H $}
   } }
Next we compose this equality with $\id_{\Xs}\oti\eta$:
\eqpic{rZ_etc_3B} {240} {38} { \setlength\unitlength{.8pt}
   \put(-18,51)   {$ j^Z_X\circ(\id_{\Xs}\oti x_0)~ = $}
   \put(116,0) { \includepichtft{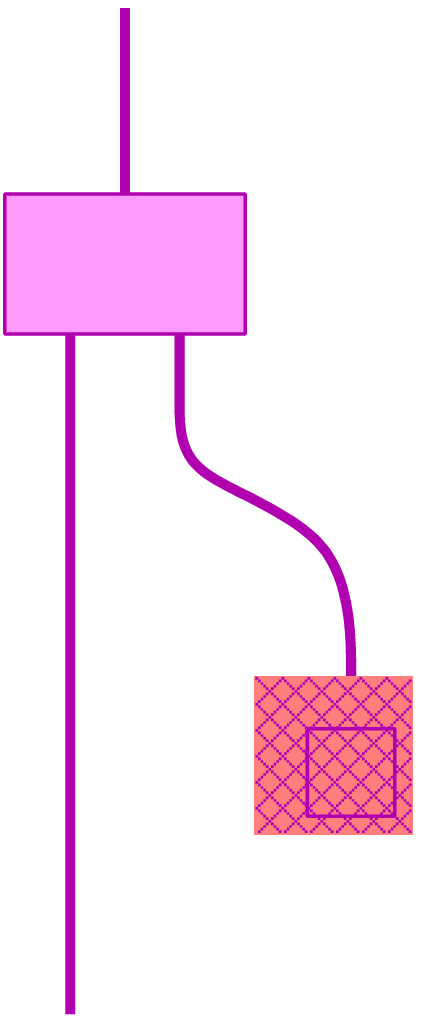}
   \put(2,-8.5)   {\sse$ \Xs $}
   \put(8.8,81.1) {\sse$ j^Z_{\!X^{}} $}
   \put(10.6,115) {\sse$ Z $}
   }
   \put(185,51)  {$ = $}
   \put(216,0) { \includepichtft{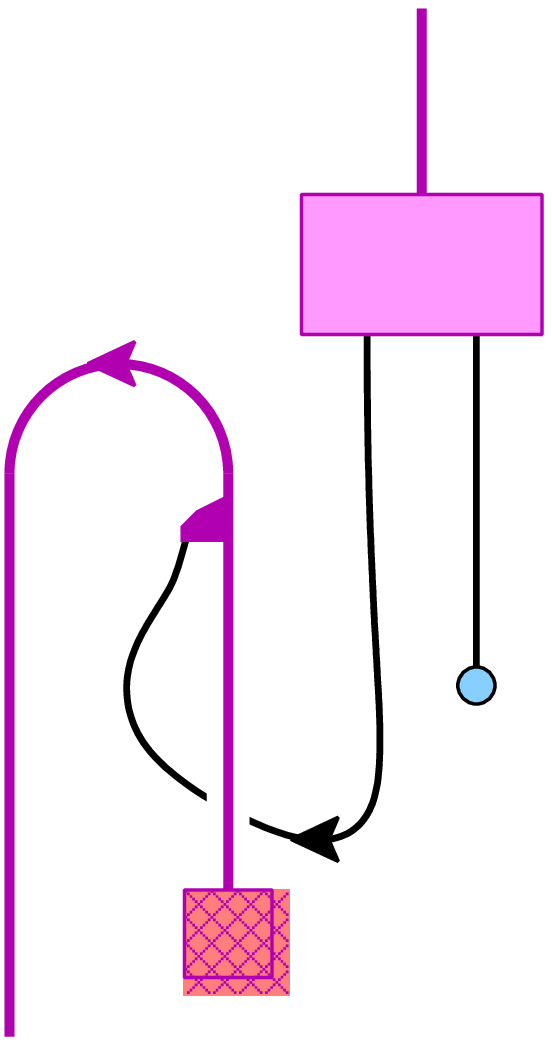}
   \put(-4,-8.5)  {\sse$ \Xs $}
   \put(41.5,83.7){\sse$ j^Z_{\!H^{}} $}
   \put(43.6,117) {\sse$ Z $}
   }
   }
Since $x_\circ$ is an arbitrary element of $X$ this implies
\be\label{j_kappa}
   j^Z_X = \kappa^Z \circ \iHb_X \,,
\ee
with $\kappa^Z$ defined as
\be
   \kappa^Z := j^Z_H \circ (\idHs \oti \eta)\,.
\ee
In fact, $\kappa^Z$ is a bimodule morphism from $\BFH$ to $Z$. Compatibility with the left $H$-action is immediate from the fact that $j^Z_H$ is a morphism of left modules. Compatibility with the right action is seen from the following chain of equalities:
 \Eqpic{kappazbim} {320} {32} {\setulen80
   \put(3,0)   {\includepichtft{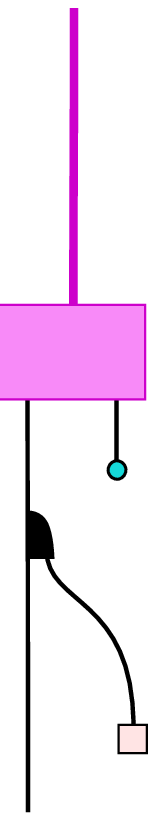}
   \put(-11.3,30) {\sse$\ohr_{\BFH}$}
   \put(-2.3,-8.5) {\sse$\Hs $}
   \put(17.4,49)  {\sse$j_H^Z$}
   \put(5,92.5)  {\sse$ Z $}
   \put(17,6)    {\sse$h$}
   }
   \put(52,38)   {$=$}
   \put(79,0) { \includepichtft{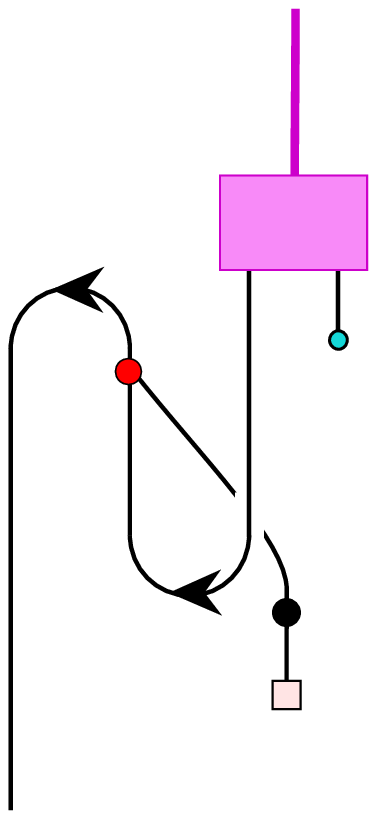}
   \put(-5,-8.5) {\sse$ \Hs $}
   \put(28.7,92.5) {\sse$ Z $}
   \put(33.5,11) {\sse$h$}
   \put(41.4,62.8) {\sse$j_H^Z$}
   }
   \put(146,38)  {$=$}
   \put(187,0) {\includepichtft{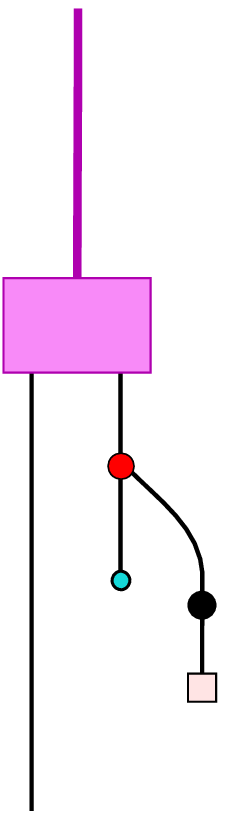}
   \put(-2,-8.5) {\sse$\Hs $}
   \put(5.4,92.5){\sse$ Z $}
   \put(24.5,12) {\sse$h$}
   \put(17.7,51.2){\sse$j_H^Z$}
   }
   \put(238,38)  {$=$}
   \put(280,0) {\includepichtft{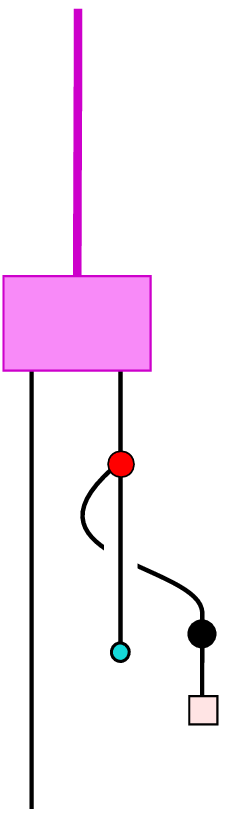}
   \put(-2,-8.5) {\sse$\Hs $}
   \put(5,92.5)  {\sse$ Z $}
   \put(24.5,9)  {\sse$h$}
   \put(17.7,51.2){\sse$j_H^Z$}
   }
   \put(334,38)  {$=$}
   \put(378,0) {\includepichtft{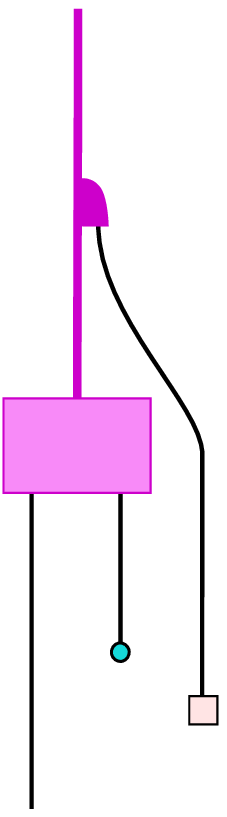}
   \put(-2,-8)   {\sse$\Hs $}
   \put(5.7,92.5){\sse$ Z $}
   \put(13,66)   {\sse$ \ohr_Z $}
   \put(24.5,9)  {\sse$h$}
   \put(-12.5,38.2){\sse$j_H^Z$}
   } }
Here the second equality invokes dinaturalness for the morphism\linebreak $m\cir(\id_H \oti (\apoi\cir h)) \In \EndH(H)$, and the last equality follows from the definition \eqref{Fboxmorph} of $\Fbx$ and the fact that $j_H^Z$ is a morphism of right $H$-modules. Thus there indeed exists a morphism $\kappa^Z\In H\Bimod$ from $\BFH$ to $Z$. \\
It remains to show that $\kappa^Z$ is unique. Assume that $\kappa'^Z$ is any morphism in $\Hom(F,Z)$ making \eqref{coend_uni} commute, i.e.
\be\label{kappa'}
    j^Z_X = \kappa'^Z \circ \iHb_X\,.
\ee
Next, note that $\iHb_H \cir (\idHs\oti\eta) \eq (d_H\oti\idHs)\cir (\idHs\oti b_H) \eq \idHs$.
Considering \eqref{kappa'} with $X=H$ and composing with $\idHs\oti\eta$ it therefore follows that
\be
    \kappa'^Z=j^Z_H\cir (\idHs\oti\eta)=\kappa^Z.
\ee
Thus $\kappa^Z$ is indeed uniquely determined.
\endofproof\end{proof}

\subsection{The Frobenius algebra structure on \BFH}
In this subsection we prove that $\BFH$ can be naturally equipped with the structure of a commutative symmetric Frobenius algebra.
\subject{The structure morphisms}
We start by introducing the structure morphisms. To this end we define the following linear maps:
\Eqpic{pic-Hb-Frobalgebra} {320} {32} {\setulen80 \put(0,-13){
   \put(0,45)       {$ m\ASF~:= $}
   \put(48,0)  {\includepichtft{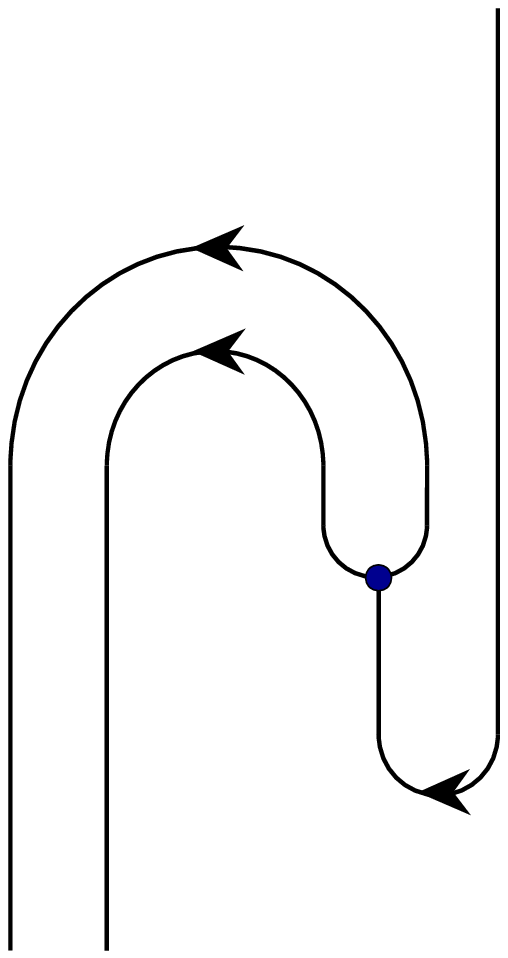}
   \put(-5.9,-8.8)  {\sse$\Hss $}
   \put(6.5,-8.8)   {\sse$\Hss $}
   \put(49.7,106.8) {\sse$\Hss $}
   }
   \put(135,45)     {$ \eta\ASF~:= $}
   \put(175,24) {\includepichtft{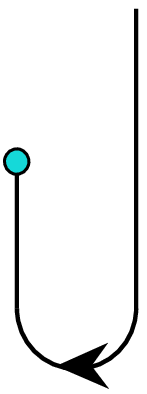}
   \put(10.7,44.1)  {\sse$\Hss $}
   }
   \put(224,45)     {$ \Delta\ASF~:= $}
   \put(272,0) {\includepichtft{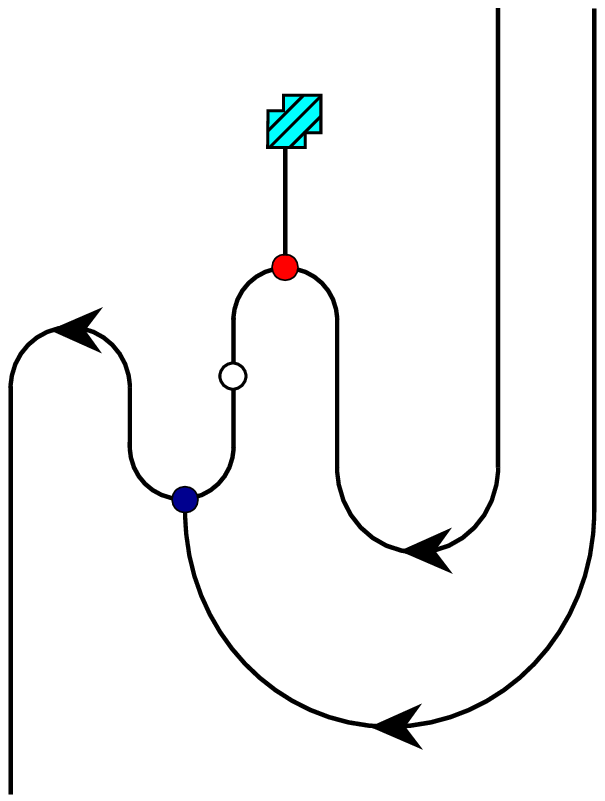}
   \put(-4.3,-8.8)  {\sse$\Hss $}
   \put(21.7,71)    {\sse$\lambda$}
   \put(48.2,89.4)  {\sse$\Hss $}
   \put(61.1,89.4)  {\sse$\Hss $}
   }
   \put(369,45)     {$ \eps\ASF~:= $}
   \put(409,24) {\includepichtft{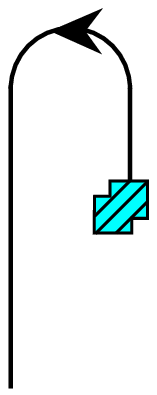}
   \put(-4.3,-8.8)  {\sse$\Hss $}
   \put(15.6,16.3)  {\sse$\Lambda$}
   } } }
When $\Hss$ is equipped with the bimodule structure \eqref{rhoHb,ohrHb}, these morphisms are not just linear maps, but bimodule morphisms, i.e.\ morphisms in $H$\Bimod\ \cite[Proposition 2.9]{fuSs3}. In order to discuss the proof of this statement we give the following useful identities from \cite[Lemma 2.7]{fuSs3}. First, for any Hopf algebra $H$ we have
\Eqpic{Hopf_Frob_trick} {320} {40} {
   \put(0,-3){\setulen80
   \put(0,10)   {\includepichtft{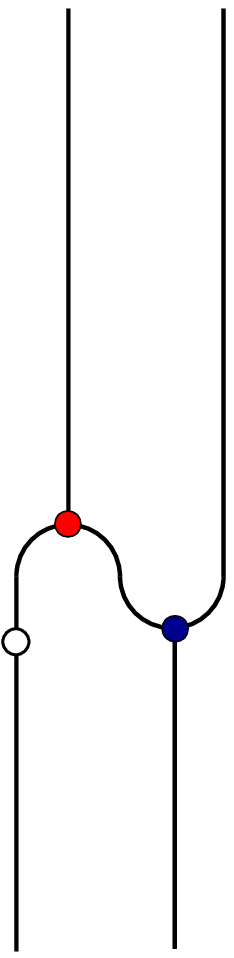}
   \put(-3,-8.5) {\sse$H$}
   \put(4,107)   {\sse$H$}
   \put(15,-8.5) {\sse$H$}
   \put(21,107)  {\sse$H$}
   }
   \put(50,57)   {$=$}
   \put(76,0) { \includepichtft{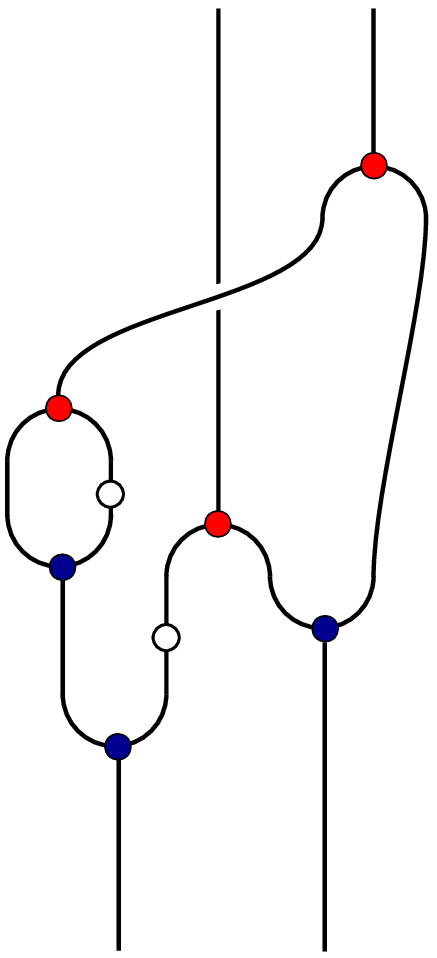}
   \put(8,-8.5)  {\sse$H$}
   \put(20,107)  {\sse$H$}
   \put(30,-8.5) {\sse$H$}
   \put(37,107)  {\sse$H$}
   }
   \put(147,57)  {$=$}
   \put(178,10)   {\includepichtft{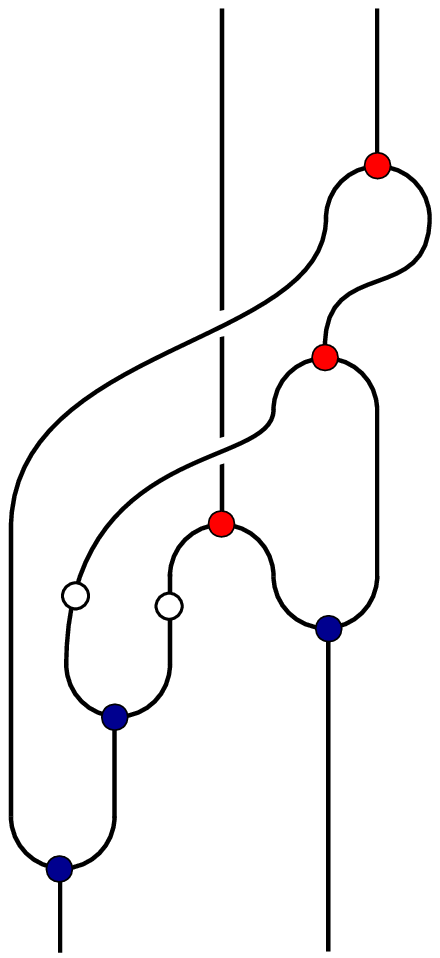}
   \put(1,-8.5)  {\sse$H$}
   \put(19.7,107){\sse$H$}
   \put(31,-8.5) {\sse$H$}
   \put(37,107)  {\sse$H$}
   }
   \put(246,57)  {$=$}
   \put(274,10) { \includepichtft{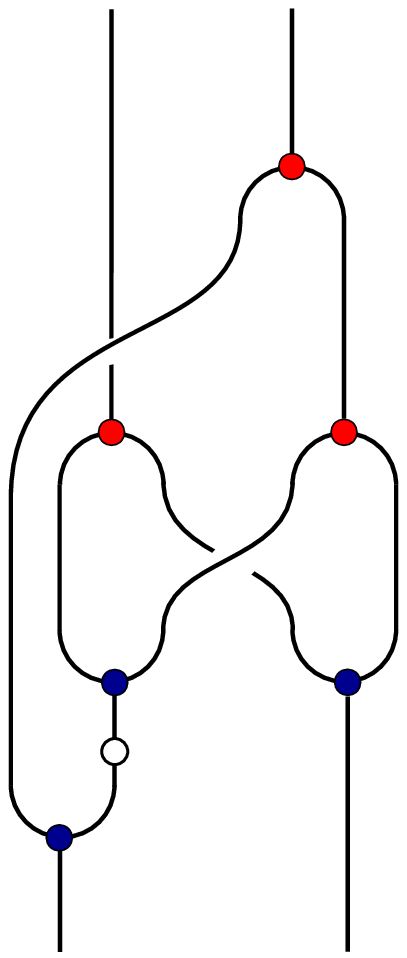}
   \put(1,-8.5)  {\sse$H$}
   \put(8,107)   {\sse$H$}
   \put(27.7,107){\sse$H$}
   \put(33,-8.5) {\sse$H$}
   }
   \put(347,57)  {$=$}
   \put(375,10) { \includepichtft{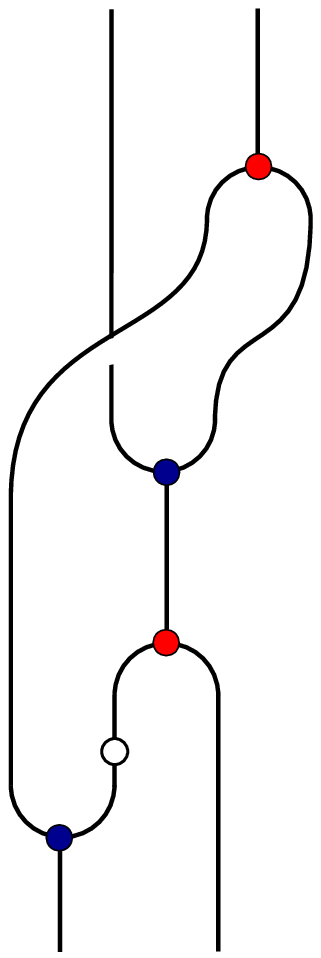}
   \put(1,-8.5)  {\sse$H$}
   \put(8,107)   {\sse$H$}
   \put(19,-8.5) {\sse$H$}
   \put(24,107)  {\sse$H$}
   } } }
Further, if $H$ is unimodular with integral $\Lambda$, we have
   \eqpic{Hopf_Frob_trick2} {300} {32} {\setulen80
   \put(0,0) {\includepichtft{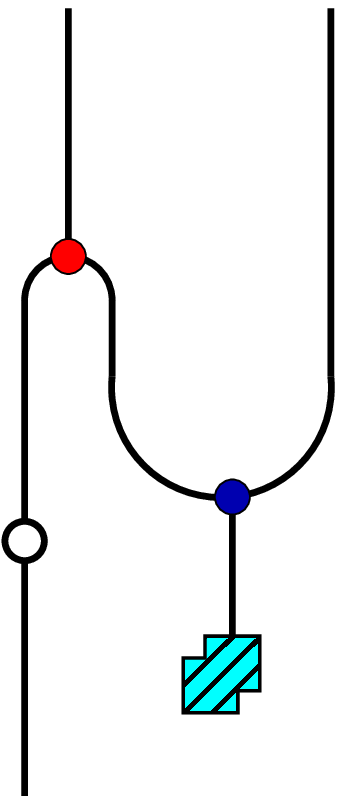}
   \put(2.2,0){
   \put(-4.2,-8.5) {\sse$ H $}
   \put(2.2,92)  {\sse$ H $}
   \put(27.7,10) {\sse$ \Lambda $}
   \put(31,92)   {\sse$ H $}
   } }
   \put(60,43)   {$ = $}
   \put(86,0) { \includepichtft{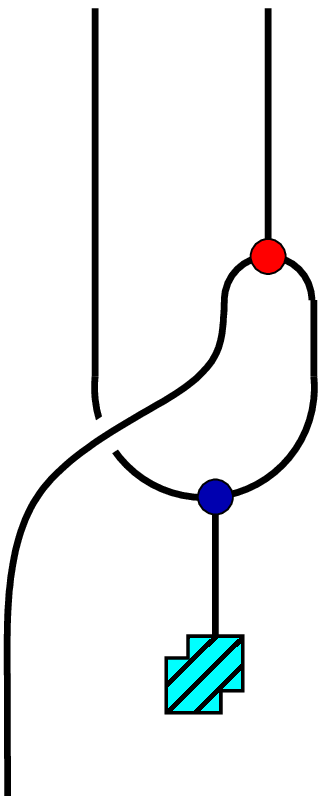}
   \put(-4.2,-8.5) {\sse$ H $}
   \put(6.6,92)  {\sse$ H $}
   \put(26.6,92) {\sse$ H $}
   \put(27.7,10) {\sse$ \Lambda $}
   }
   \put(174,43)  {and}
   \put(240,0) { \includepichtft{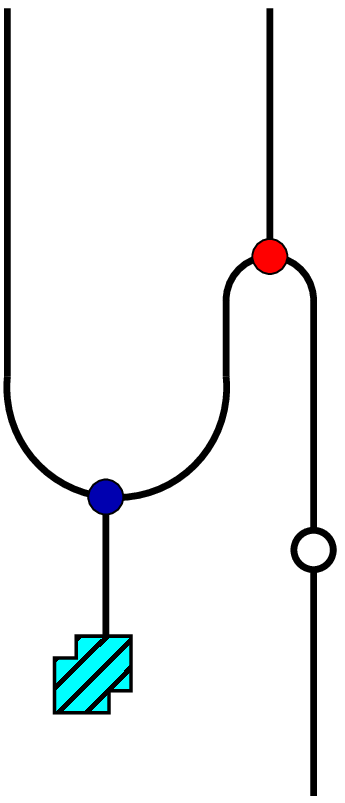}
   \put(-3.2,92) {\sse$ H $}
   \put(15.7,10) {\sse$ \Lambda $}
   \put(26.1,92) {\sse$ H $}
   \put(29.8,-8.5) {\sse$ H $}
   }
   \put(300,43)  {$ = $}
   \put(326,0) { \includepichtft{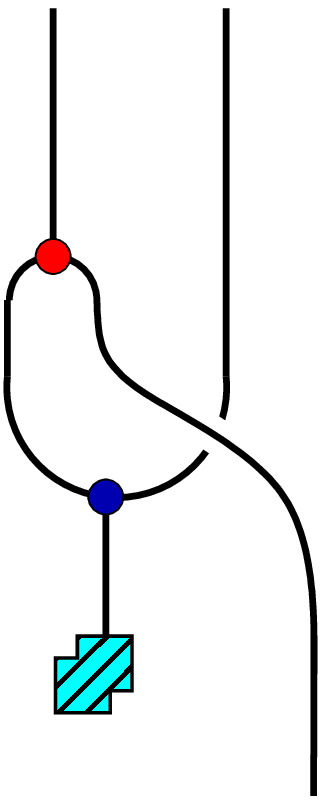}
   \put(1.2,92)  {\sse$ H $}
   \put(15.7,10) {\sse$ \Lambda $}
   \put(21.7,92) {\sse$ H $}
   \put(29.8,-8.5) {\sse$ H $}
   } }

We now outline the proof of the statement that the maps \eqref{pic-Hb-Frobalgebra} are bimodule morphisms.
\ssubject{$\boldsymbol{m\ASF}$ is a bimodule morphism} That $m\ASF$ intertwines the left action of $H$ can be seen as follows:
   \Eqpic{mb_left} {320} {32} { \put(0,3){\setulen70
   \put(5,0)   {\INcludepichtft{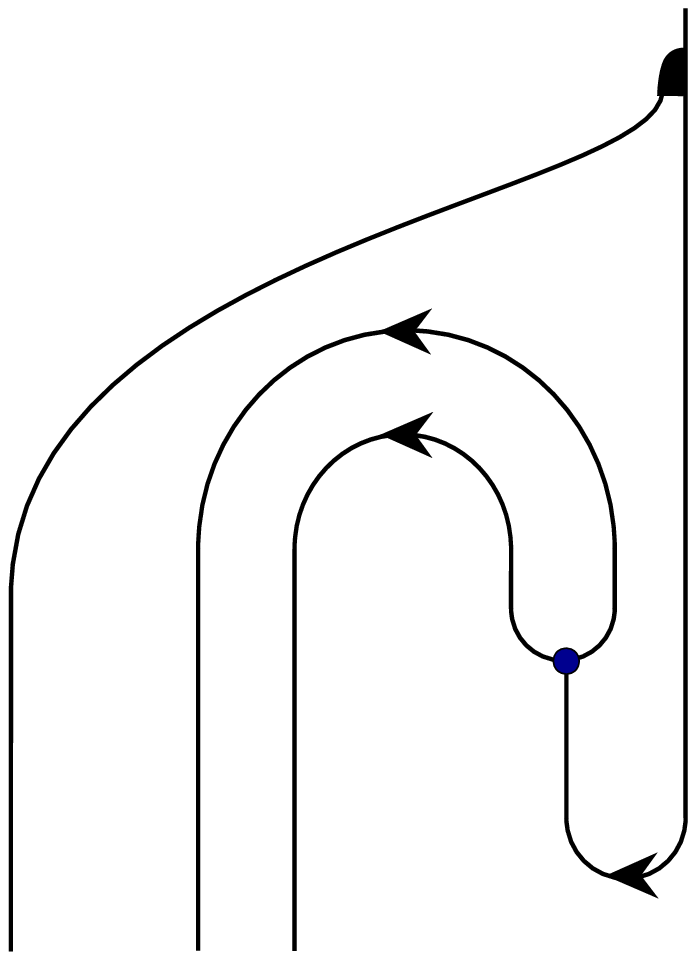}{266}
   \put(-4,-8.5) {\sse$H$}
   \put(13.8,-8.5) {\sse$\Hss $}
   \put(27.8,-8.5) {\sse$\Hss $}
   \put(70.6,107){\sse$\Hss $}
   }
   \put(88,47)   {$=$}
   \put(93,0) { \INcludepichtft{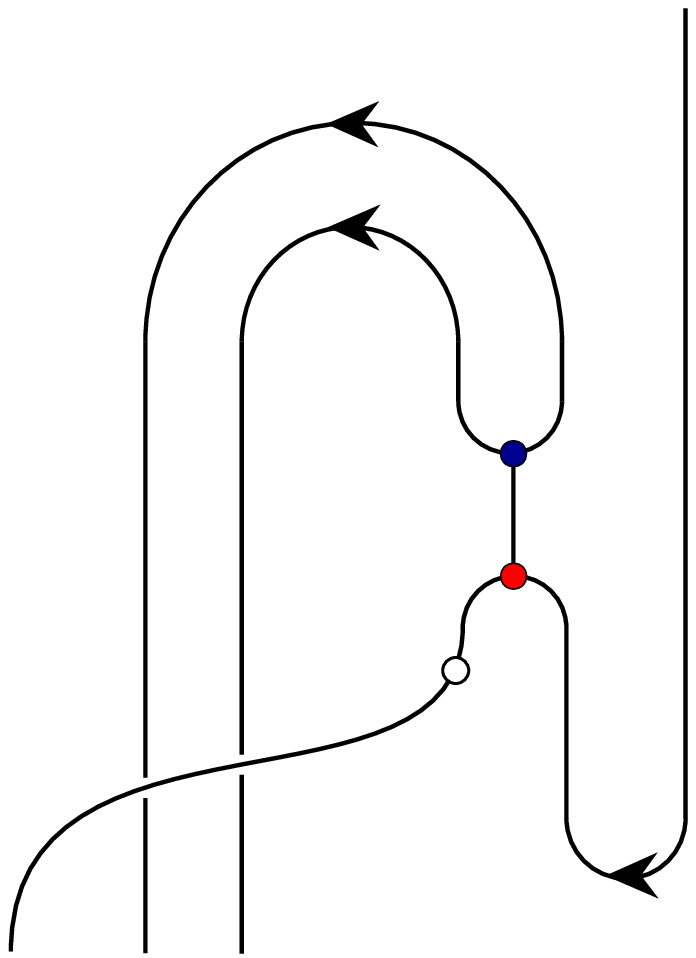}{266}
   \put(-4,-8.5) {\sse$H$}
   \put(9.5,-8.5){\sse$\Hss $}
   \put(23,-8.5) {\sse$\Hss $}
   \put(70.6,107){\sse$\Hss $}
   }
   \put(184,47)  {$=$}
   \put(193,0)   {\INcludepichtft{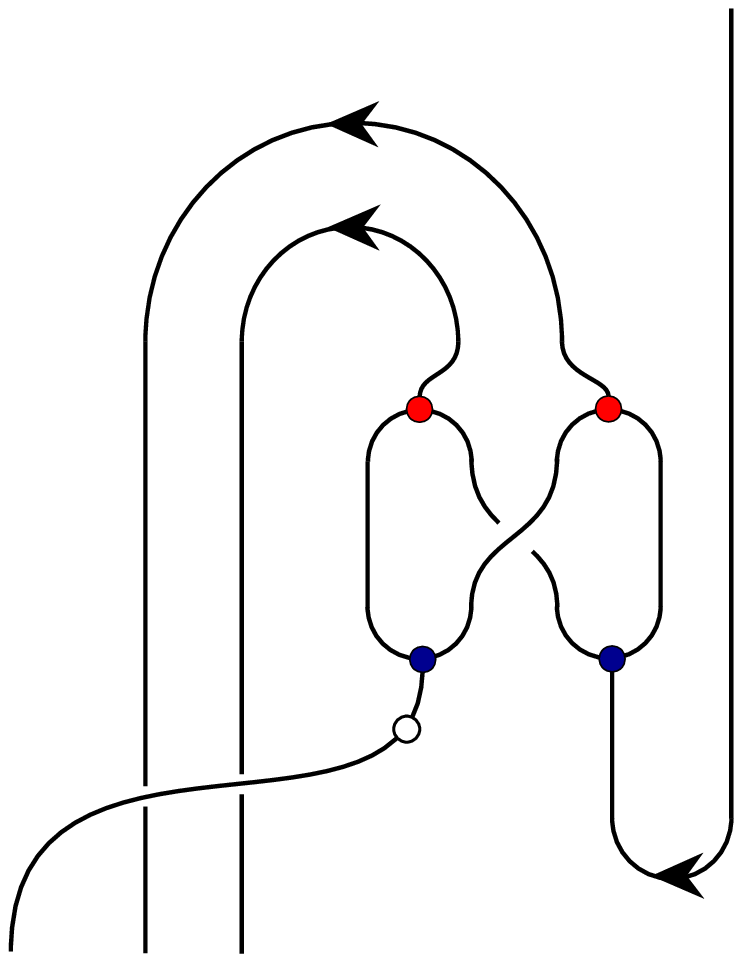}{266}
   \put(-4,-8.5) {\sse$H$}
   \put(9.5,-8.5){\sse$\Hss $}
   \put(23,-8.5) {\sse$\Hss $}
   \put(75.6,107){\sse$\Hss $}
   }
   \put(284,47)   {$=$}
   \put(288,0) { \INcludepichtft{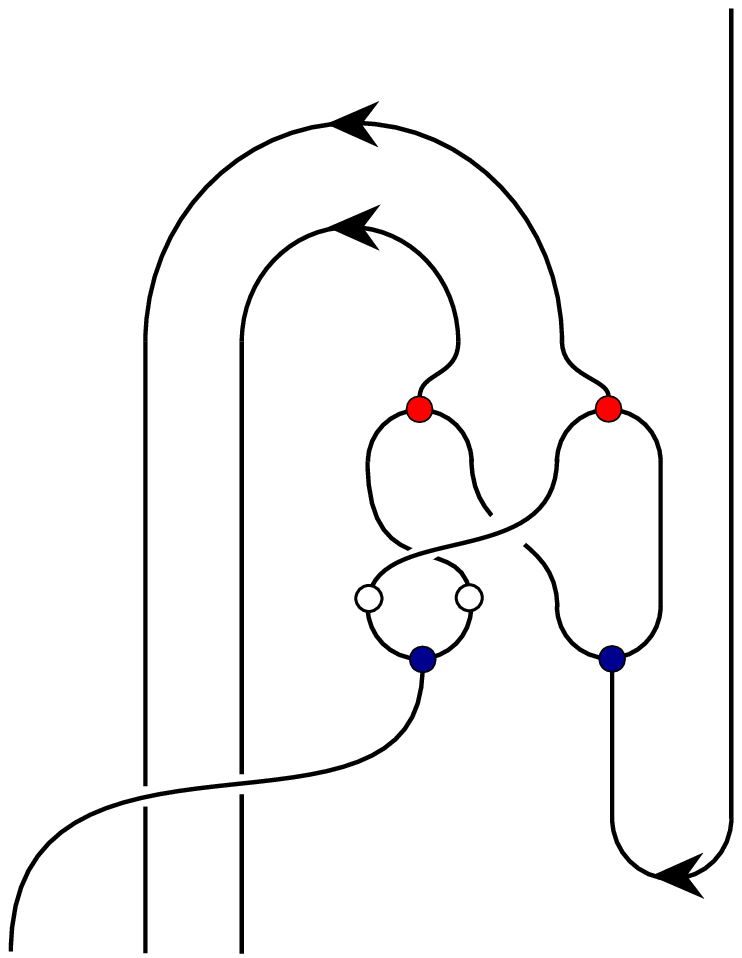}{266}
   \put(-4,-8.5) {\sse$H$}
   \put(9.5,-8.5){\sse$\Hss $}
   \put(23,-8.5) {\sse$\Hss $}
   \put(75.8,107){\sse$\Hss $}
   }
   \put(383,47)   {$=$}
   \put(394,0) { \INcludepichtft{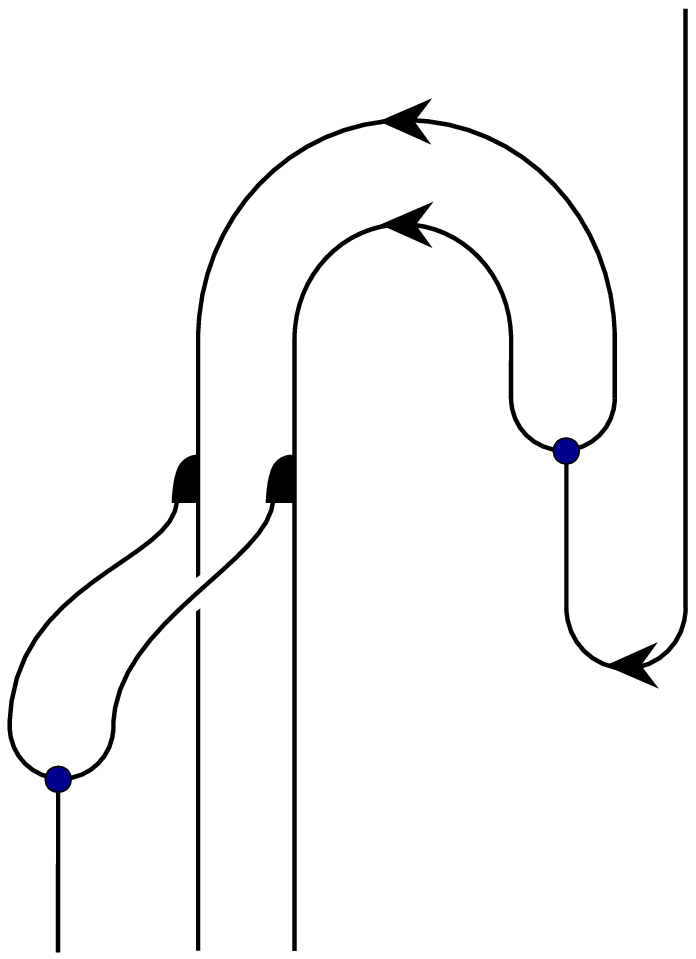}{266}
   \put(1,-8.5)  {\sse$H$}
   \put(14.5,-8.5) {\sse$\Hss $}
   \put(28,-8.5) {\sse$\Hss $}
   \put(70.2,107){\sse$\Hss $}
   } } }
Here, the first and the last equality just implement \eqref{rhoHb,ohrHb}. The second and third equality follow from the connecting axiom \eqref{con_axiom} and the anti-coalgebra morphism property \eqref{anti_alg_apo}, respectively. That $m\ASF$ is also a morphism of right $H$-modules follows in an analogous fashion, see \cite[eq. (2.22)]{fuSs3}.

\ssubject{$\boldsymbol{\eta\ASF}$ is a bimodule morphism} That $\eta\ASF$ is a bimodule morphism follows from \eqref{rhoHb,ohrHb} and direct application of \eqref{eta_m} and \eqref{eta_inv_s}, see \cite[eq. (2.23) \& (2.24)]{fuSs3}.

\ssubject{$\boldsymbol{\Delta\ASF}$ is a bimodule morphism} The chain of equalities
\Eqpic{Deltab_left} {320} {25} {
\put(5,-8){\setulen70
   \put(0,0)   {\INcludepichtft{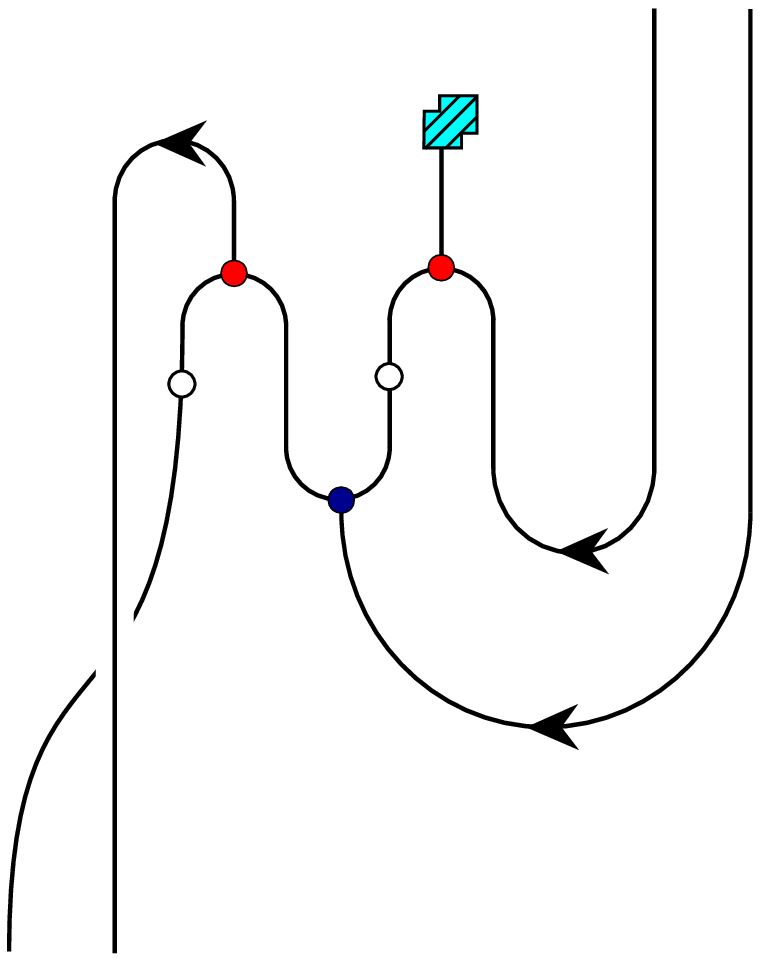}{266}
   \put(-5,-8.5) {\sse$H$}
   \put(7,-8.5)  {\sse$\Hss$}
   \put(65,107)  {\sse$\Hss$}
   \put(78.5,107){\sse$\Hss$}
   }
   \put(105,47)  {$=$}
   \put(130,0) { \INcludepichtft{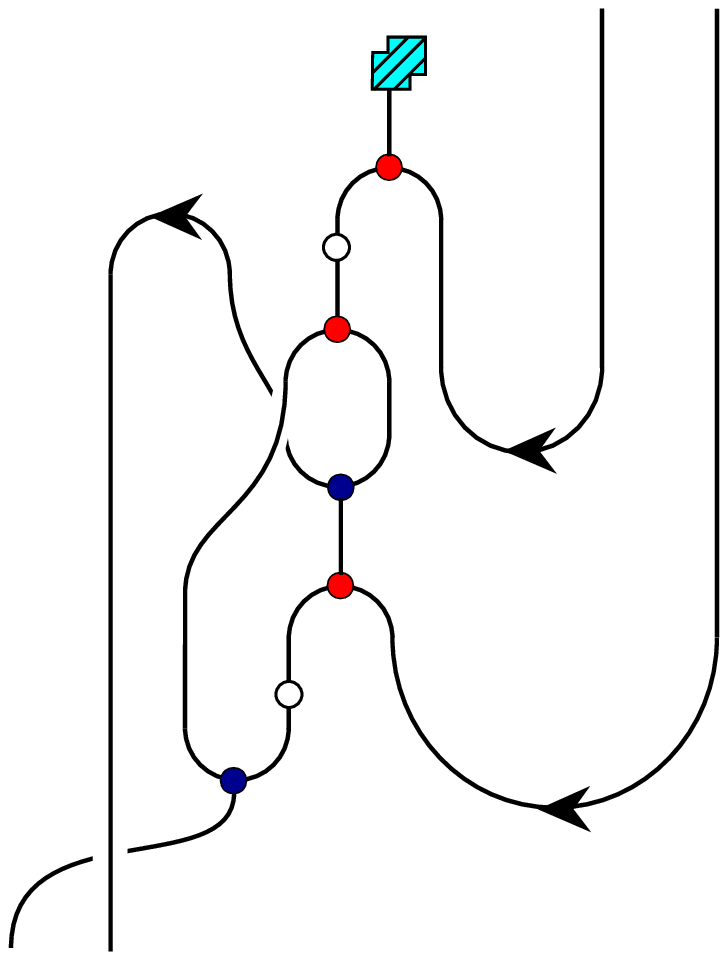}{266}
   \put(-5,-8.5) {\sse$H$}
   \put(7,-8.5)  {\sse$\Hss$}
   \put(61,107)  {\sse$\Hss$}
   \put(74,107)  {\sse$\Hss$}
   }
   \put(234,47)  {$=$}
   \put(255,0)   {\INcludepichtft{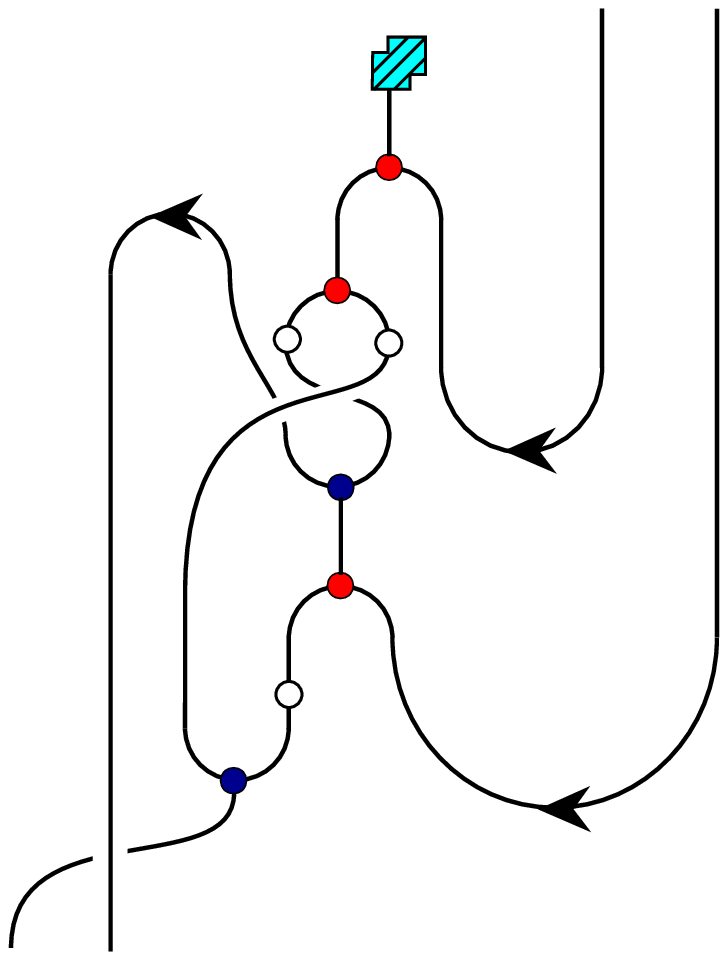}{266}
   \put(-5,-8.5) {\sse$H$}
   \put(7,-8.5)  {\sse$\Hss$}
   \put(61,107)  {\sse$\Hss$}
   \put(74.6,107){\sse$\Hss$}
   }
   \put(357,47)  {$=$}
   \put(375,0) { \INcludepichtft{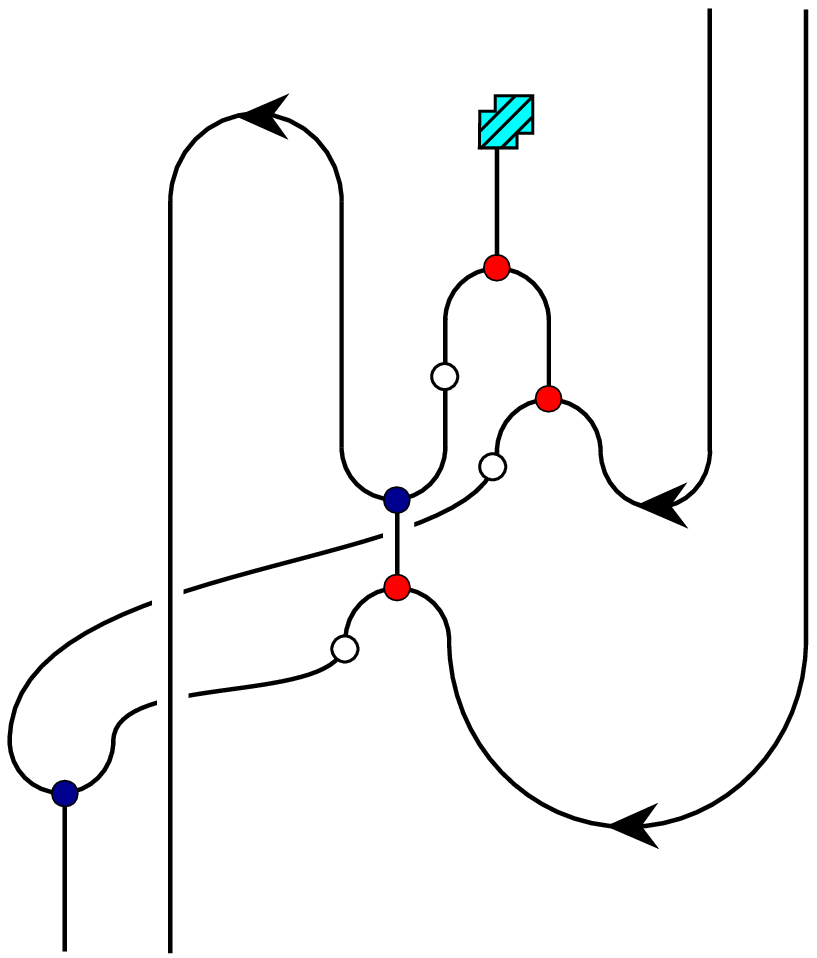}{266}
   \put(2,-8.5)  {\sse$H$}
   \put(13,-8.5) {\sse$\Hss$}
   \put(73,107)  {\sse$\Hss$}
   \put(85,107)  {\sse$\Hss$}
   } } }
establishes that $\Delta\ASF$ is a morphism of left $H$-modules. Here the first equality follows from \eqref{Hopf_Frob_trick}, while the second and third equality uses \eqref{anti_alg_apo} and associativity respectively. \\
In order to show that $\Delta\ASF$ is also a morphisms of right modules, we use that there is an alternative expression for $\Delta\ASF$, \cite[Lemma 2.8]{fuSs3}:
\eqpic{Deltaprime} {133} {25} {\setulen70
   \put(0,35)    {$ \Delta\ASF' ~:=$}
   \put(68,-2) {\INcludepichtft{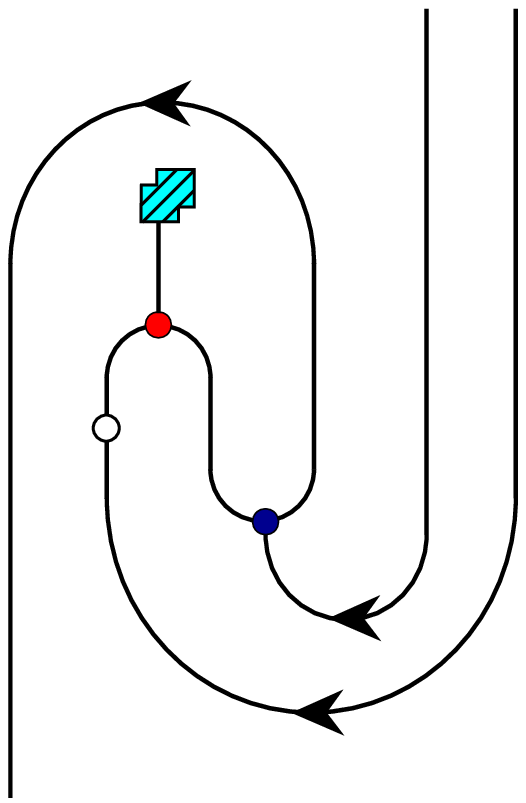}{266}
   \put(-5,-8.5) {\sse$\Hss $}
   \put(21,62)   {\sse$\lambda$}
   \put(40,91)   {\sse$\Hss $}
   \put(54,91)   {\sse$\Hss $}
   } }
That $\Delta\ASF'=\Delta\ASF$ follows by using \eqref{fmap_inv} and coassociativity of the coproduct $\Delta$ of $H$. Using the expression \eqref{Deltaprime}, the following chain of equalities establishes that $\Delta\ASF$ is a morphism of left modules:
\Eqpic{Deltab_right} {320} {80} {
\put(0,-15){
\setulen80
   \put(0,128) {\includepichtft{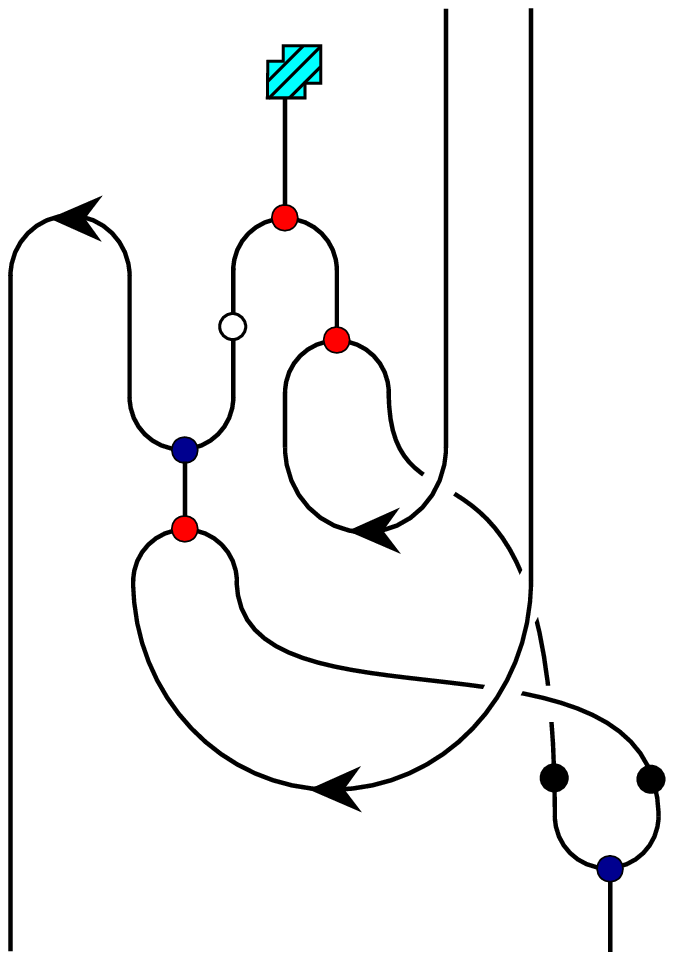}
   \put(-6,-8.5) {\sse$\Hss $}
   \put(43,107)  {\sse$\Hss $}
   \put(55,107)  {\sse$\Hss $}
   \put(62,-8.5) {\sse$H$}
   }
   \put(88,175)  {$=$}
   \put(120,128) { \includepichtft{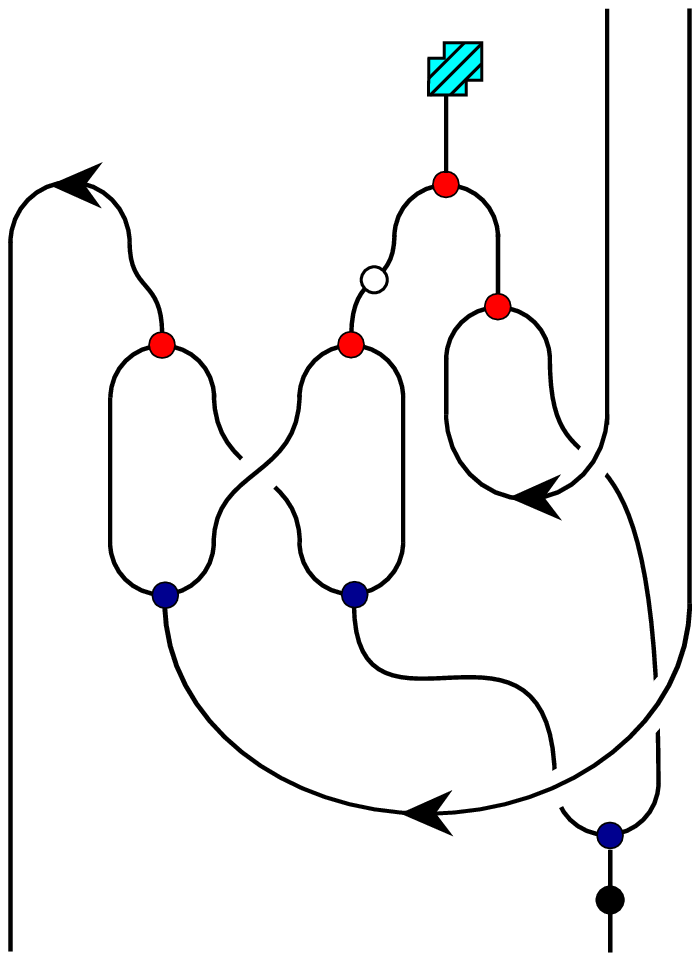}
   \put(-6,-8.5) {\sse$\Hss $}
   \put(62,-8.5) {\sse$H$}
   \put(61,107)  {\sse$\Hss $}
   \put(73,107)  {\sse$\Hss $}
   }
   \put(227,175) {$=$}
   \put(265,128) {\includepichtft{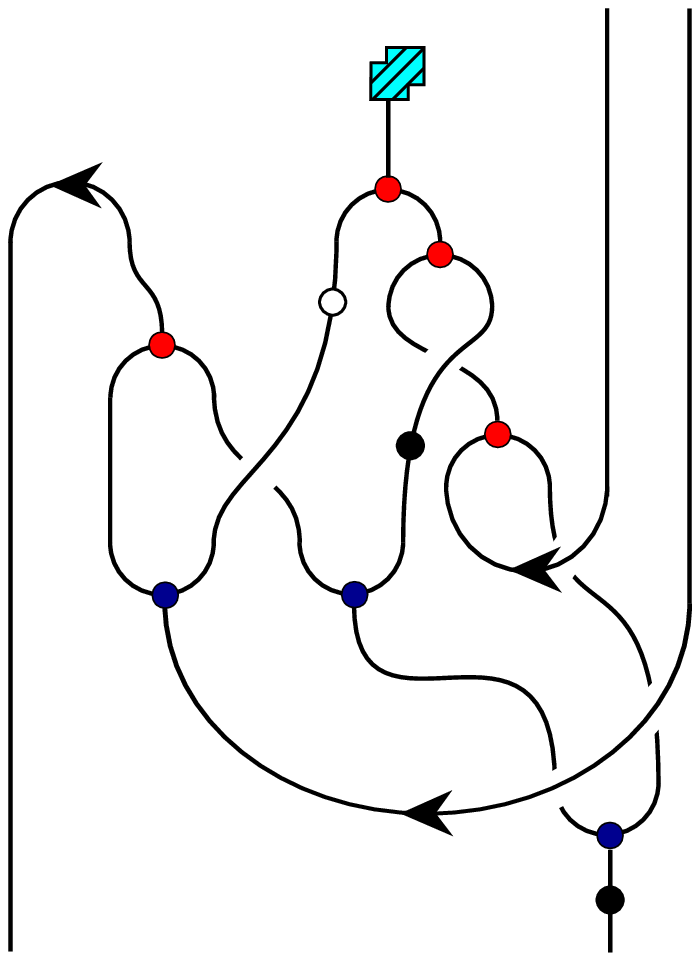}
   \put(-6,-8.5) {\sse$\Hss $}
   \put(62,-8.5) {\sse$H$}
   \put(61,107)  {\sse$\Hss $}
   \put(73,107)  {\sse$\Hss $}
   } \put(-10,-2){
   \put(50,47)   {$=$}
   \put(75,0) { \includepichtft{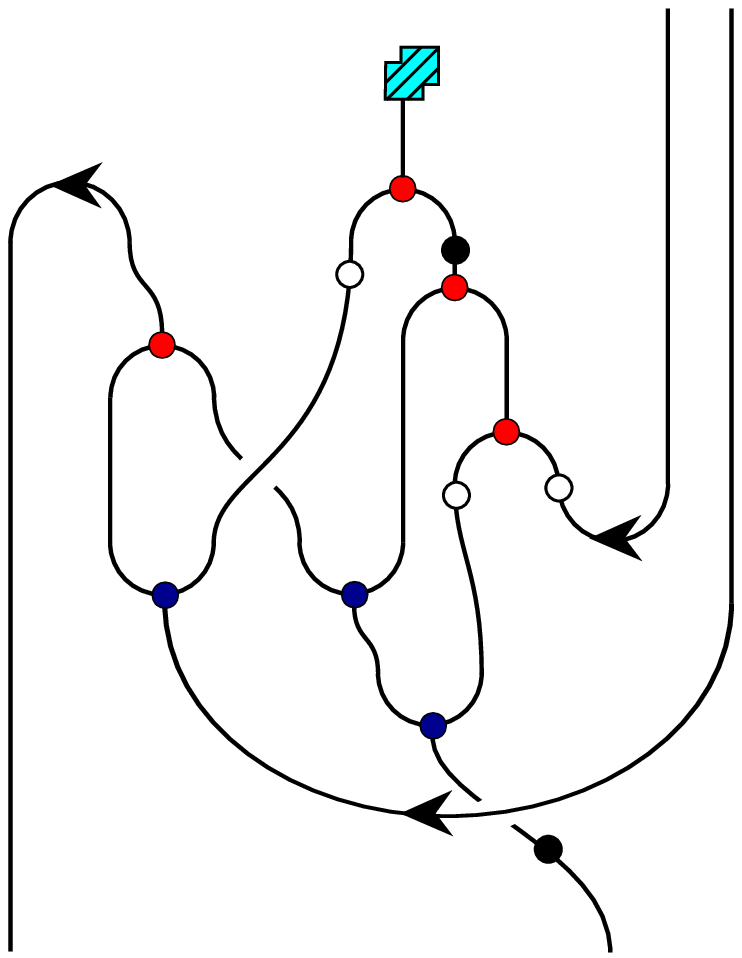}
   \put(-6,-8.5) {\sse$\Hss $}
   \put(62,-8.5) {\sse$H$}
   \put(66,107)  {\sse$\Hss $}
   \put(78,107)  {\sse$\Hss $}
   }
   \put(180,47)  {$=$}
   \put(207,0) { \includepichtft{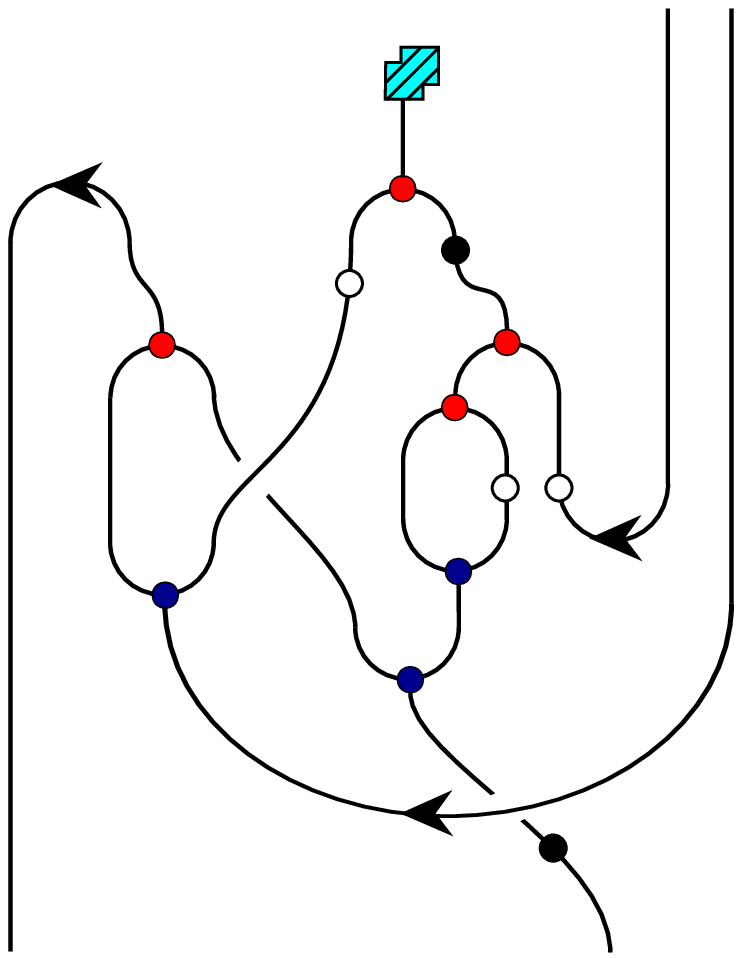}
   \put(-6,-8.5) {\sse$\Hs $}
   \put(62,-8.5) {\sse$H$}
   \put(66,107)  {\sse$\Hs$}
   \put(78,107)  {\sse$\Hs$}
   }
   \put(314,47)  {$=$}
   \put(340,0) { \includepichtft{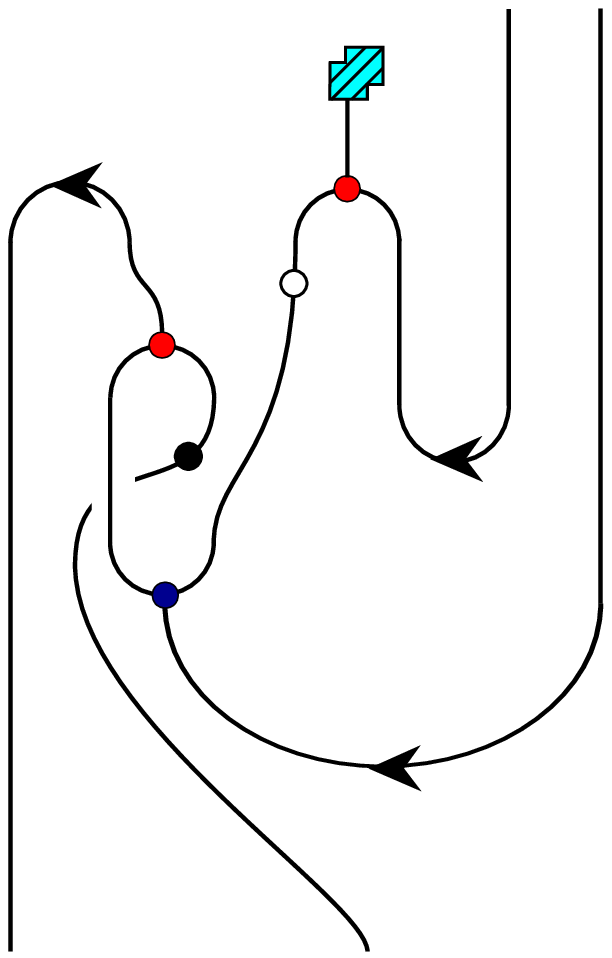}
   \put(-6,-8.5) {\sse$\Hs $}
   \put(35,-8.5) {\sse$H$}
   \put(50.5,107){\sse$\Hs$}
   \put(62.5,107){\sse$\Hs$}
   } } } }
Here, the first equality uses the anti-coalgebra morphism property of $\apo^{-1}$ and the connecting axiom. The second equality follows from \eqref{OS2_mod}. In the third equality the upper inverse antipode is pushed through the product by using that $\apo^{-1}$ is an anti-coalgebra morphism. The two last equalities use associativity and the defining property \eqref{apo_prop} of $\apo$.

\ssubject{$\boldsymbol{\eps\ASF}$ is a bimodule morphism} Finally, that $\eps\ASF$ is a bimodule morphism follows directly from the actions \eqref{rhoHb,ohrHb}, the defining property of the two-sided integral (i.e.\ that $\Lambda$ satisfies \eqref{int_prop} and the corresponding equality for a right integral) and \eqref{eta_inv_s}, see \cite[eq. (2.28)]{fuSs3}.

\subject{The Frobenius algebra structure}
Having verified that the maps \eqref{pic-Hb-Frobalgebra} are indeed morphisms in $H$\Bimod, we continue to state  \cite[Proposition 2.5]{fuSs3}
\begin{prop}
The morphisms \erf{pic-Hb-Frobalgebra} endow the object $\BFH \eq (\Hs,\rho_{\BFH},\ohr_{\BFH})$
with the structure of a Frobenius algebra in \HBImod. That is, $(\BFH,m\ASF,\eta\ASF)$ is a (unital associative)
algebra, $(\BFH,\Delta\ASF,\eps\ASF)$ is a (counital coassociative) coalgebra,
and the two structures are compatible in the sense
   \be
   (\id_{\BFH}\oti m\ASF) \circ (\Delta\ASF\oti\id_{\BFH}) = \Delta\ASF \circ m\ASF
   = (m\ASF\oti\id_{\BFH}) \circ (\id_{\BFH}\oti\Delta\ASF) \,.
   \ee
\end{prop}
\begin{proof}
We outline the proof given in \cite{fuSs3}. First, that $(\BFH,m\ASF,\eta\ASF)=(H^*,\Delta^*,\eps^*)$ is a unital algebra
follows directly by duality from the fact that $(H,\Delta,\eps)$ is a counital coalgebra. Second, coassociativity of $\Delta$ immediately implies
\be
    (\id_{\Hs}\oti\Delta\ASF')\circ\Delta\ASF=(\Delta\ASF\oti\id_{\Hs})\circ\Delta\ASF'\,,
\ee
which in turn implies that $\Delta\ASF'\eq\Delta\ASF$ is a coassociative coproduct. That $\eps\ASF$ is a counit for this product follows from \eqref{fmap_inv}, see \cite[eq. (2.32)]{fuSs3}. Finally that $\BFH$  satisfies the first equality in \eqref{frob} follows by writing out the linear maps, inserting $\Delta\ASF$ and using coassociativity of $\Delta$. Analogously the second equality in \eqref{frob} follows by inserting $\Delta\ASF'$ and using coassociativity of $\Delta$. Thus \BFH\ is indeed a Frobenius algebra in $H$\Bimod.
\endofproof\end{proof}
\subject{$\boldsymbol{\BFH}$ is commutative and has trivial twist}
We have from  \cite[Proposition 3.1]{fuSs3}
\begin{prop}
The product $m\ASF$ of the Frobenius algebra \BFH\ in $H$\Bimod\ is commutative
with respect to the braiding \eqref{bibraid}:
   \be
   m\ASF \circ c^{}_{\BFH,\BFH} = m\ASF \,.
   \ee
\end{prop}
\begin{proof}
By applying \eqref{bibraid}, \eqref{rhoHb,ohrHb}, \eqref{Rmat2} and \eqref{Rmat1} it follows that
\Eqpic{Hb_comm} {320} {90} {\setulen80 \put(-7,-6) {
   \put(0,181)   {$m\ASF\circ c_{\BFH,\BFH} ~= $}
   \put(88,134) { \includepichtft{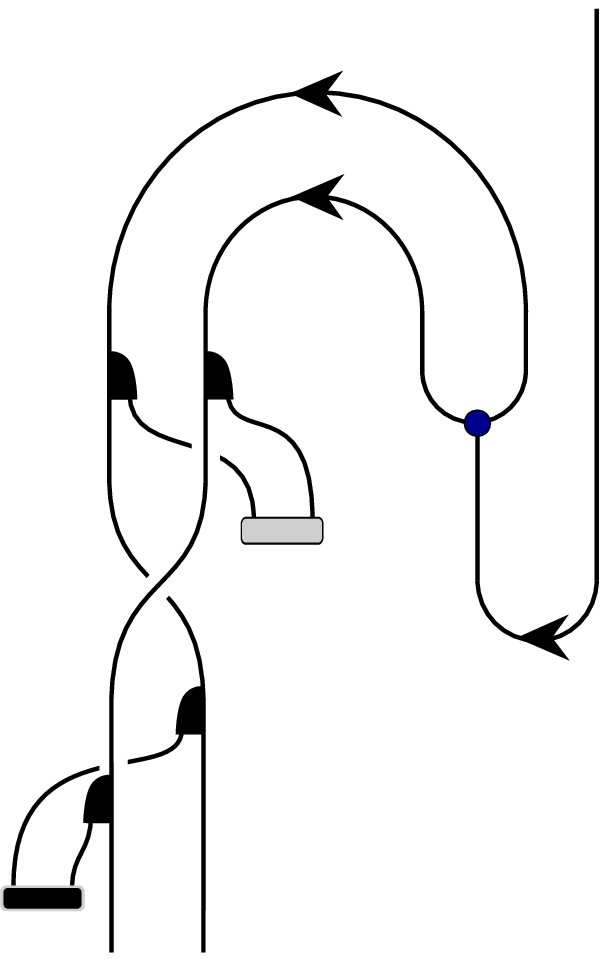}
   \put(5,-8.5)  {\sse$\Hs $}
   \put(19,-8.5) {\sse$\Hs $}
   \put(62,107)  {\sse$\Hs $}
   }
   \put(184,181)  {$=$}
   \put(221,134) {\includepichtft{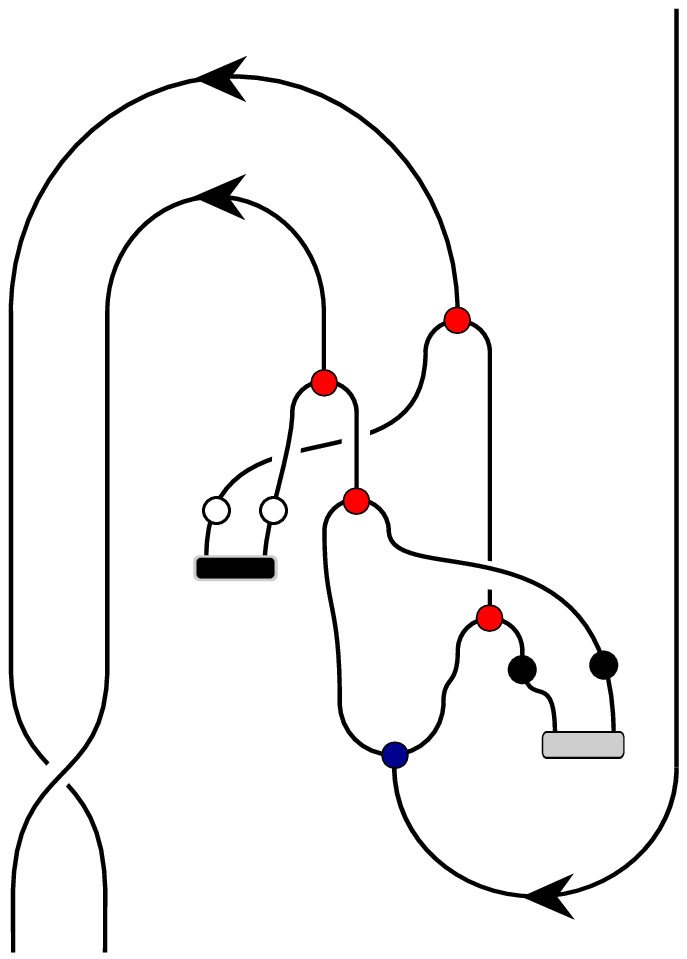}
   \put(-6,-8.5) {\sse$\Hss $}
   \put(8,-8.5)  {\sse$\Hss $}
   \put(70,107)  {\sse$\Hss $}
   }
   \put(32,0){
   \put(0,47)    {$=$}
   \put(26,0) { \includepichtft{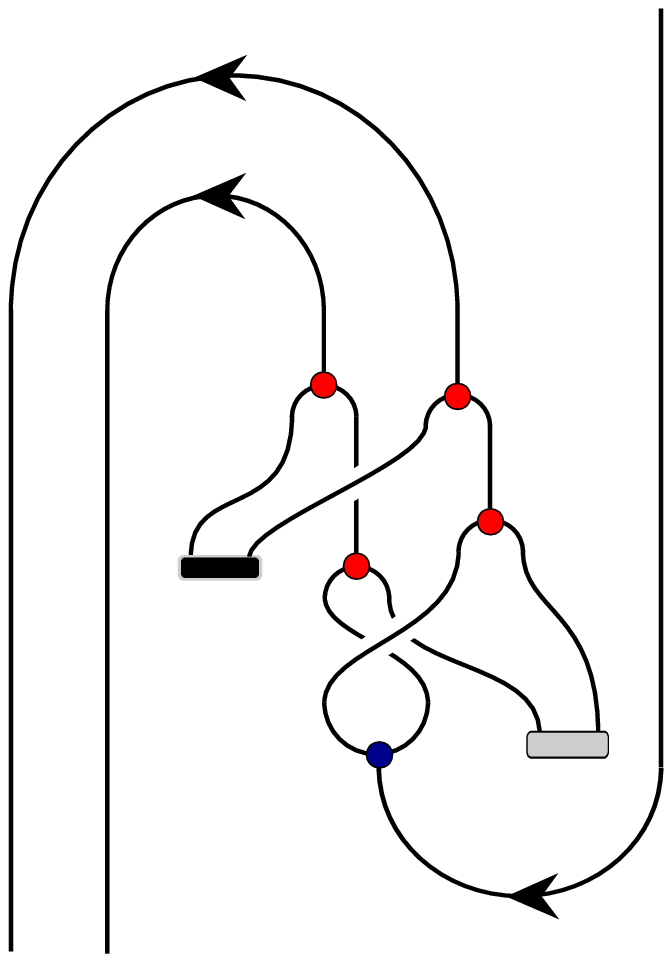}
   \put(-6,-8.5) {\sse$\Hss $}
   \put(8,-8.5)  {\sse$\Hss $}
   \put(68.7,107){\sse$\Hss $}
   }
   \put(125,47)  {$=$}
   \put(151,0) { \includepichtft{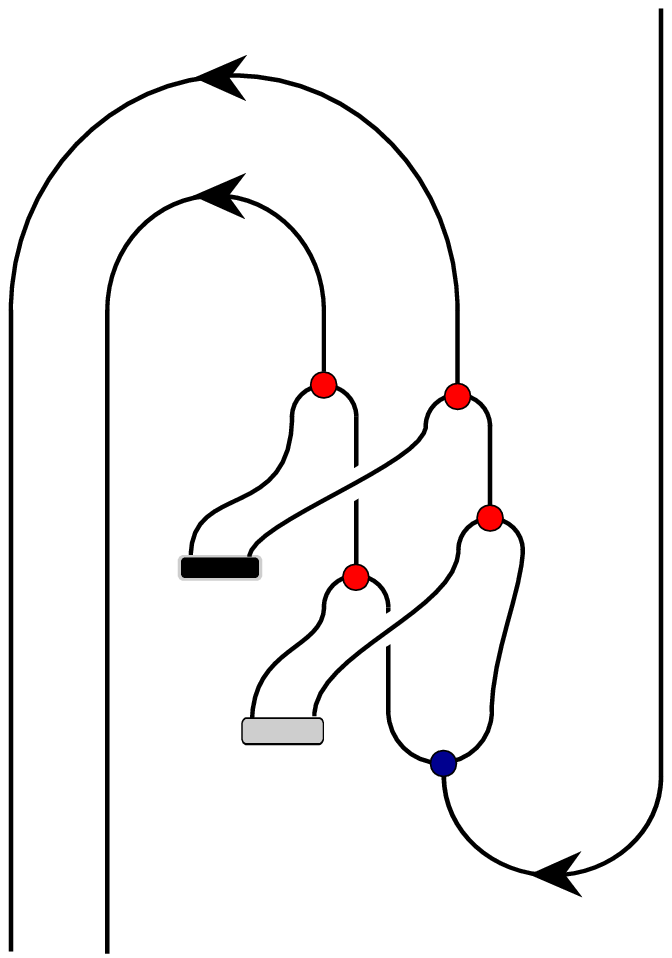}
   \put(-6,-8.5) {\sse$\Hss $}
   \put(8,-8.5)  {\sse$\Hss $}
   \put(68.7,107){\sse$\Hss $}
   }
   \put(250,47)  {$=$}
   \put(279,0) { \includepichtft{79a}
   \put(-6,-8.5) {\sse$\Hss $}
   \put(8,-8.5)  {\sse$\Hss $}
   \put(50.7,107){\sse$\Hss $}
   }
   \put(359,47)  {$ = ~m\ASF \,.$}
   } } }
Thus \BFH\ is commutative.
\endofproof
\end{proof}

Next we observe
\begin{prop}
The twist of \BFH\ is trivial:
\be\label{triv_twist}
    \theta_{\BFH}=\id_{\BFH}\,.
\ee
\end{prop}
\begin{proof}
After having inserted the left and right actions \eqref{rhoHb,ohrHb} of $H$ on $\BFH$ in the expression \eqref{bimod-twist} for the twist of $H$\Bimod, the equality \eqref{triv_twist}  is obtained by using $\apo\circ v=v$ and that $v$ is central.
\endofproof\end{proof}
\subject{$\boldsymbol{\BFH}$ is symmetric}
We have  \cite[Theorem 4.4]{fuSs3} the following result
\begin{prop}
    The Frobenius algebra \BFH\ is symmetric.
\end{prop}
\begin{proof}
According to Definition \ref{def_sym} (iii) and Remark \ref{rem_symm} (iii) \BFH\ is symmetric if \eqref{pic_csp_14} holds for $\kappa=\eps\ASF\cir m\ASF$. Indeed we have the following chain of equalities
\eqpic{sym_proof} {260} {25} {
   \put(220,10) {\Includepichtft{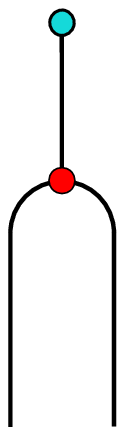}
   \put(-4,-8.5)  {\sse$\BFH$}
   \put(7.7,-8.5) {\sse$\BFH$}
   \put(12.2,26.3) {\sse$m\ASF$}
   \put(11.2,44.3) {\sse$\eps\ASF$}
   }
   \put(70,25)  {$ = $}
   \put(180,25)  {$ = $}
   \put(0,0) {\Includepichtft{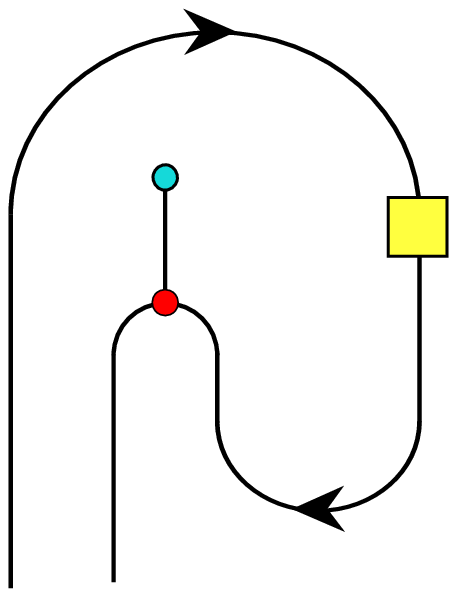}
   \put(-3.3,-8.5) {\sse$\BFH$}
   \put(22.7,44.5) {\sse$\eps\ASF$}
   \put(22.7,30.5) {\sse$m\ASF$}
   \put(8,-8.5)   {\sse$\BFH$}
   \put(50,39)     {\sse$\piv_{\!\BFH}^{}$}
   }
   \put(100,0) {\Includepichtft{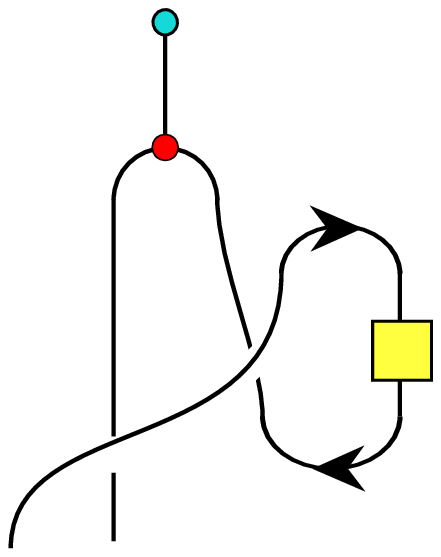}
   \put(-3.3,-8.5) {\sse$\BFH$}
   \put(20.7,56) {\sse$\eps\ASF$}
   \put(22.7,43.5) {\sse$m\ASF$}
   \put(8,-8.5)   {\sse$\BFH$}
   \put(49,17)     {\sse$\piv_{\!\BFH}^{}$}
   }
   \put(260,52)    {\Hbimodpic}
   }
in $H$\Bimod.
Note that this is a picture drawn in $H$\Bimod, not in \Vect. In particular the braiding is the braiding of $H$\Bimod. Here the first equality is a deformation. The second equality is obtained by first using \eqref{twist_NS} to recognize the twist $\theta_{\BFH}$, which is trivial, and then using that \BFH\ is commutative. \eqref{sym_proof} proves the rightmost equality in \eqref{pic_csp_14} and thus that \BFH\ is symmetric.
\endofproof\end{proof}
\begin{remark}
    Any symmetric commutative Frobenius algebra in a sove\-reign braided category is also cocommutative and has trivial twist. For a strictly sovereign category \C, this has been shown in \cite[Proposition 2.25]{ffrs}. The proof easily extends to the non-strict case.
\end{remark}

\subject{Specialness}
We first note that \cite[Lemma 4.5]{fuSs3}
\be
   \eps\ASF \cir \eta\ASF = \eps\cir\Lambda \quand
   m\ASF \cir \Delta\ASF = (\lambda\cir\eps)\, \id_\Hs \,.
\ee
The first equality is immediate from \eqref{pic-Hb-Frobalgebra}, while the second uses coassociativity of $H$ and the defining property of the antipode, see \cite[eq. (4.18)]{fuSs3}. A finite-dimensional Hopf algebra is semisimple iff the \emph{Maschke number} $\eps\cir\Lambda\In\k$ is non-zero, and it is cosemisimple iff $\lambda\cir\eta\In\k$ is non-zero \cite{laSw}. In addition, when $\k$ is of characteristic zero, cosemisimplicity is implied by semisimplicity \cite{laRad2}. Thus we have
\begin{cor}
The Frobenius algebra \BFH\ in $H$\Bimod\ is special iff $H$ is semisimple.
\end{cor}

\section{Mapping class group invariants in $H$\Bimod}\label{sec:MGC_inv}
We are now a in position to give the main result of this chapter. Recall the projective action, $\piL$, of mapping class groups, given in Proposition \ref{Lyubact_prop}. Note that for $U_1\eq\dots\eq U_n\eq V$ we have $\Lyubspace\eq V^{\otii n}$. Accordingly, $\piL$ is a projective action on the space $\Hom(\L^{\otii g}, V^{\otii n})$.
\begin{theorem}\label{thm: MPG_Hopf}
Let $H$ be a finite-dimensional factorizable ribbon Hopf algebra over an algebraically closed field \k\ of characteristic zero and \Mapgn, the mapping class group of a closed oriented surface $\Surf gn$. Then, for $\C\eq H$\Mod, the morphism $\Corrgn\linebreak\In\Hom(\HK^{\otii g},\BF^{\otii n})$, defined in \eqref{CorrgnC}, is invariant under the projective action $\piKC$ of \Mapgn,  given in Proposition \ref{Lyubact_prop}, for any $g,n\geq0$.
\end{theorem}
The rest of this section is devoted to the proof of this statement. It will be convenient to work in the category $\HBimod\simeq\HMod\rev\boxtimes\HMod$.
\subsection{The \bhHopf in $H$\Bimod}
In the case $\C\eq H$\Mod, i.e.\ $\eC\simeq H$\Bimod, the \bhHopf\, defined in \eqref{def_bhHopf}, is the coend $\HKH=\L_{\HBImod}$. As an object in $H$\Bimod, $\HKH$ is the bimodule \cite[Proposition A.6]{fuSs3}
\be
    \HKH = (\Hs \Otik \Hs,\rho\coa\oti\idHs,\idHs\oti\rho\coar)\,.
\ee
That is, $\HKH$ is the vector space $\Hs\otik\Hs$, endowed with the left and right coadjoint actions, defined in \eqref{def_lads}, on the first and second factor, respectively.
The dinatural family for the coend is given by
 \eqpic{def_iHaa_X} {120} {45} {
   \put(0,0)   {\Includepichtft{11a}}
   \put(-1.6,89)  {\sse$ \Hss\Oti\Hss $}
   \put(2,-8.5)   {\sse$ X^{\!\vee} $}
   \put(9.2,45.2) {\sse$ \iHKH_X $}
   \put(16,-8.5)  {\sse$ X $}
   \put(48,45)    {$ := $}
   \put(80,0) { \Includepichtft{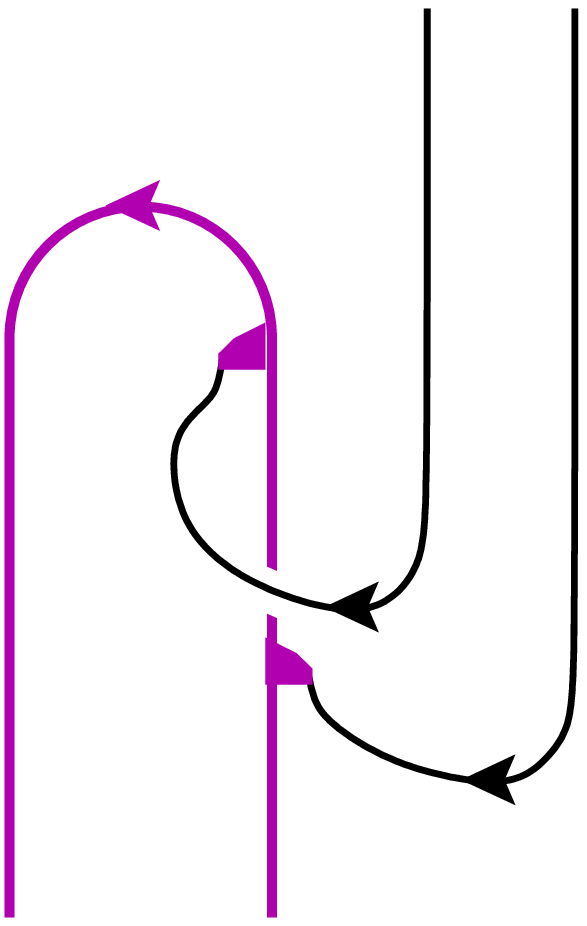}
   \put(-4.1,-8.5){\sse$ \Xs $}
   \put(19.7,28)  {\sse$ \ohr_{\!X}^{} $}
   \put(25.3,-8.5){\sse$ X $}
   \put(31.2,62)  {\sse$ \rho_{\!X}^{} $}
   \put(42.6,104) {\sse$ \Hss $}
   \put(59.4,104) {\sse$ \Hss $}
   } }
We omit the proof of this statement and refer the reader to \cite{fuSs} for details.

Since $\HKH=\L_{\HBImod}$ is the coend of the tensor product functor $\Lfun$, defined in \eqref{Lfundef}, considered in $H$\Bimod, $\HKH$ carries a natural structure of a Hopf algebra in $H$\Bimod. The structure morphisms are obtained by inserting \eqref{def_iHaa_X} in \eqref{Lyub_Hopf}. The unit, counit and coproduct are
\be\label{eem_HKH}
   \bearl
   \eta\ASK = \eps^\vee \oti \eps^\vee \,, \qquad
   \eps\ASK = \eta^\vee \oti \eta^\vee \,,
   \Nxl4
   \Delta\ASK = (\idHs \oti \tau_{\Hss,\Hss}\oti\idHs)
                  \circ \big( {(m\op)}^\vee \oti m^\vee \big) \,,
   \eear
\ee
and the antipode and product are given by
   \eqpic{apo_Haa} {200} {49} {
   \put(0,49)    {$ \apo\ASK ~= $}
   \put(52,0)  {\Includepichtft{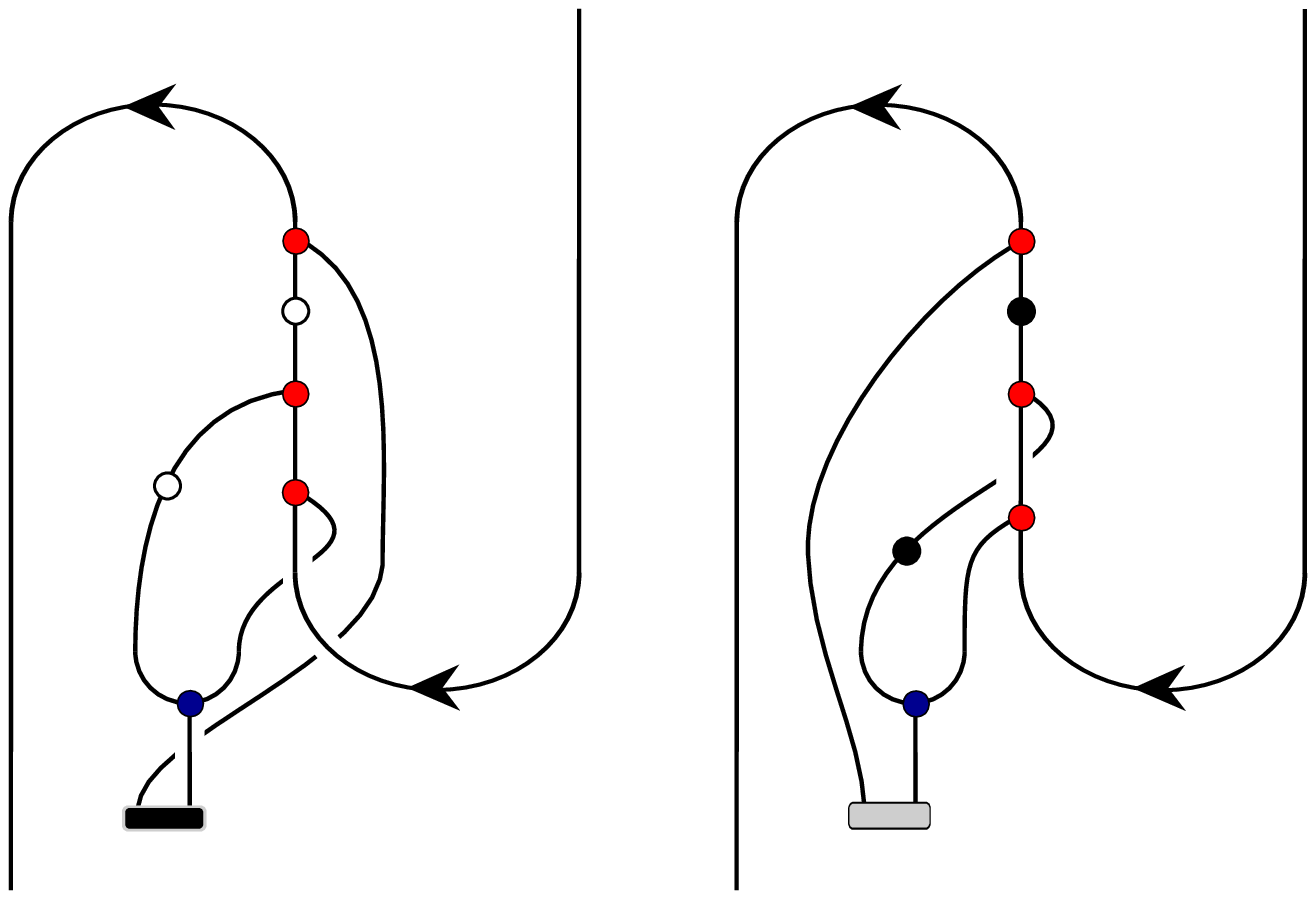}
   \put(-5,-8.5) {\sse$ \Hss $}
   \put(59,100.4){\sse$ \Hss $}
   \put(10,-.5)  {\sse$ R^{-1} $}
   \put(75,-8.5) {\sse$ \Hss $}
   \put(138.8,100.5) {\sse$ \Hss $}
   \put(93,-.5)  {\sse$ R $}
   } }
and
\Eqpic{mHaa1,2} {320} {65} {\setulen7{37}
   \put(0,87)   {$ m\ASK ~= $}
   \put(61,0)  {\INcludepichtft{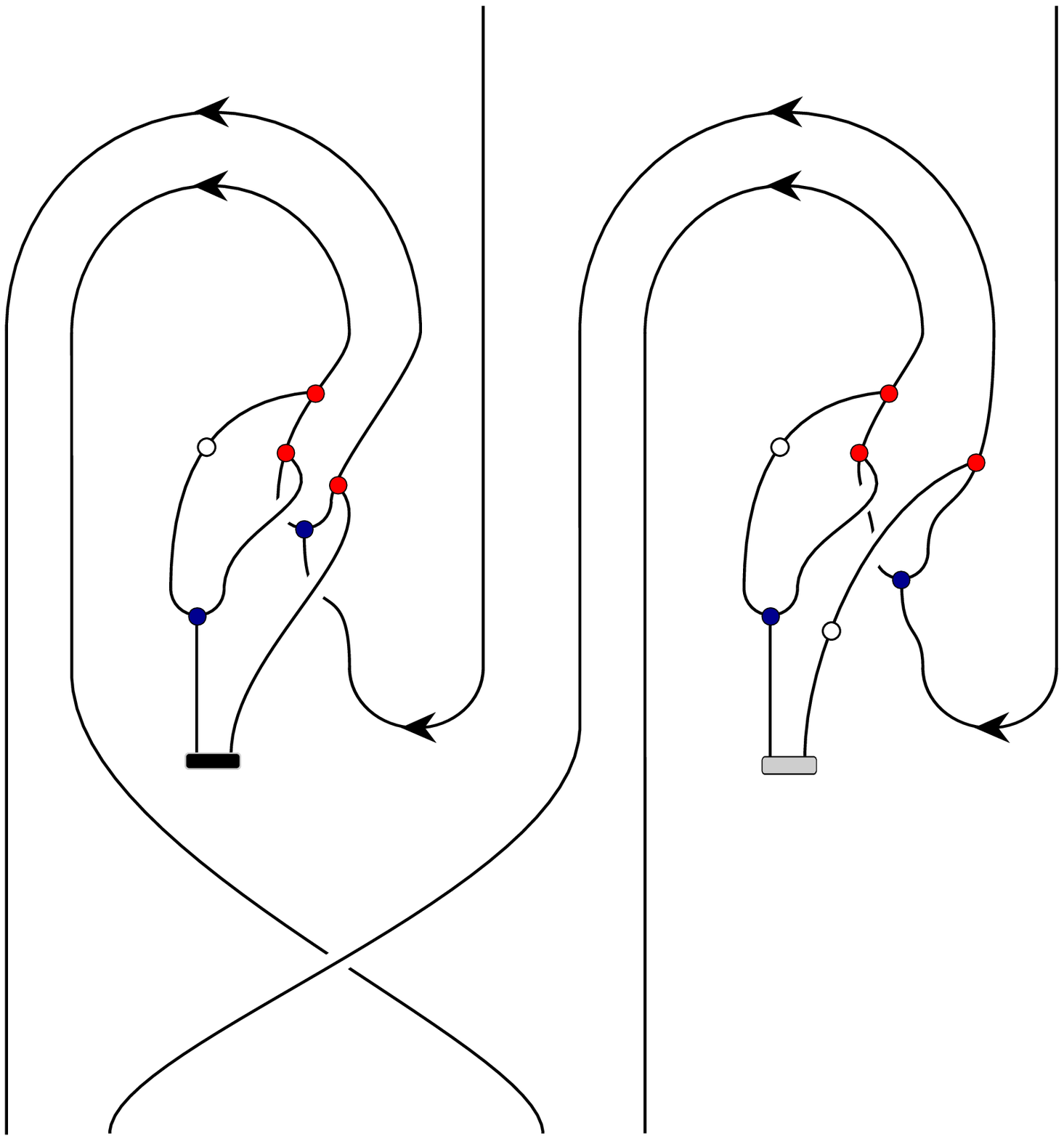}{28}
   \put(-5,-8.5) {\sse$ \Hss $}
   \put(11.7,-8.5) {\sse$ \Hss $}
   \put(27.4,52.3) {\sse$ R^{-1} $}
   \put(75,189.6){\sse$ \Hss $}
   \put(83,-8.5) {\sse$ \Hss $}
   \put(100,-8.5){\sse$ \Hss $}
   \put(124.9,51.1){\sse$ R $}
   \put(169,189.6) {\sse$ \Hss $}
   }
   \put(254,87)    {$ = $}
   \put(285,0)  {\INcludepichtft{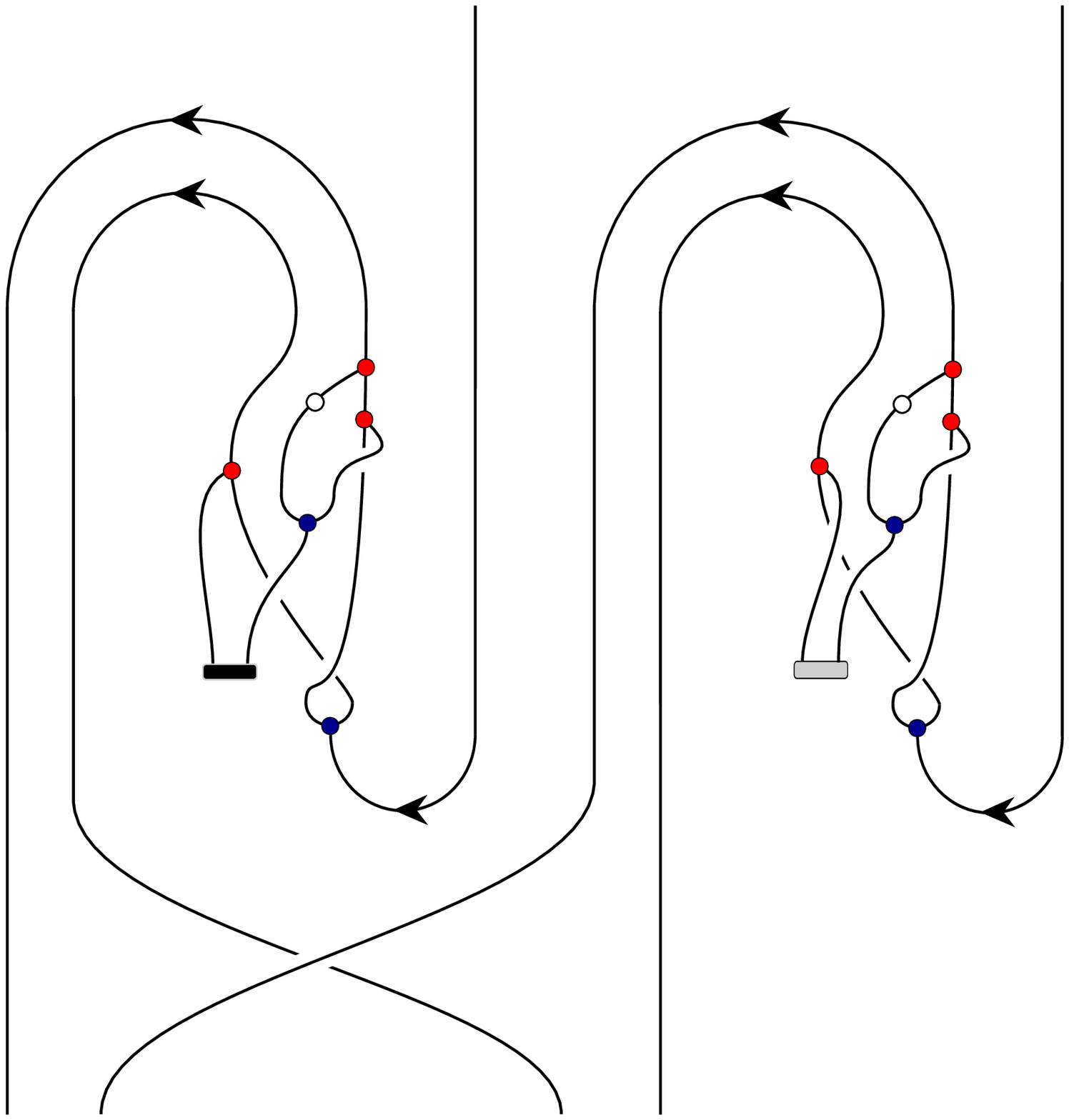}{28}
   \put(-5.4,-8.5) {\sse$ \Hss $}
   \put(11.2,-8.5) {\sse$ \Hss $}
   \put(30,65.4) {\sse$ R^{-1} $}
   \put(75,189.6){\sse$ \Hss $}
   \put(87.7,-8.5) {\sse$ \Hss $}
   \put(104,-8.5){\sse$ \Hss $}
   \put(132.1,65.4){\sse$ R $}
   \put(173,189.6) {\sse$ \Hss $}
   } }
The second expression for $m\ASK$ is obtained by using that the $R$-matrix intertwines the coproduct and the opposite coproduct.

In addition, the right cointegral $\lambda$ and the two-sided integral $\Lambda$ of $H$ give rise to an integral and a cointegral of $\HKH$. Explicitly
\be
   \lambda\ASK := \Lambda^\vee \oti \Lambda^\vee
\ee
is a two-sided cointegral of $(\HKH,m\ASK,\eta\ASK,\Delta\ASK,\eps\ASK,\apo\ASK)$ and
\be\label{Lambda_HKH}
   \Lambda\ASK := \lambda^\vee \oti \lambda^\vee
\ee
is a two-sided integral of $(\HKH,m\ASK,\eta\ASK,\Delta\ASK,\eps\ASK,\apo\ASK)$ \cite[Proposition A.8]{fuSs3}. That $\lambda\ASK$ is a two-sided cointegral is an immediate consequence of the fact that $\Lambda$ is a two-sided integral. That $\Lambda\ASK$ is a left integral follows from the left-most expression in \eqref{mHaa1,2} together with the fact that $\lambda$ is a right cointegral and satisfies \eqref{OS2}. It follows in the same manner, by using the rightmost expression in \eqref{mHaa1,2}, that $\Lambda\ASK$ is a right integral.

\subsection{Morphisms in $H$\Bimod\ as linear maps}
In order to work with the underlying linear maps for performing calculations in $H$\Bimod\ we need to express various morphisms, introduced earlier in this chapter, as linear maps.
We start by writing down the relevant morphisms, implementing the representation in Proposition \ref{Lyubact_prop}, as linear maps.
\ssubject{Partial monodromies}
Inserting the dinatural family \eqref{def_iHaa_X} and the linear map \eqref{bibraid}, implementing the braiding in $H$\Bimod, in the expressions \eqref{QHH} and \eqref{QHX} for the two partial monodromies, we obtain the partial monodromies of $\HKH$:
   \eqpic{QQ_Haa} {230}{63} {
   \put(0,57)     {$ \QQ_{\HKH,\HKH} ~= $}
   \put(81,0)  { \Includepichtft{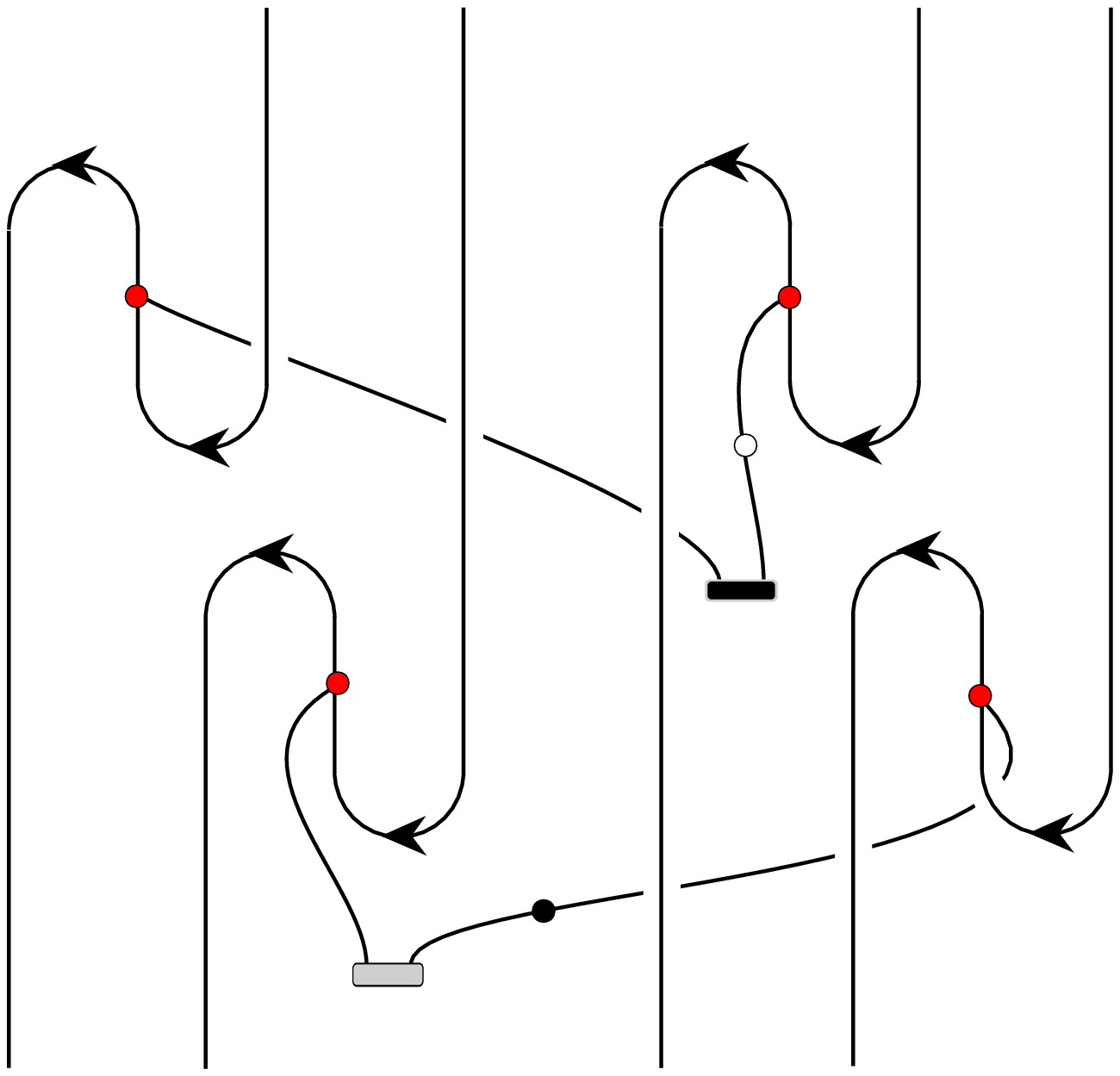}
   \put(-4.3,-8.5) {\sse$ \Hss $}
   \put(20.7,-8.5) {\sse$ \Hss $}
   \put(29.4,139)  {\sse$ \Hss $}
   \put(44.5,4.3)  {\sse$ Q $}
   \put(54.8,139)  {\sse$ \Hss $}
   \put(62.8,13.1) {\sse$ \apoi $}
   \put(78.7,-8.5) {\sse$ \Hss $}
   \put(86.6,53.5) {\sse$ Q^{-1} $}
   \put(87.6,77.9) {\sse$ \apo $}
   \put(102.9,-8.5){\sse$ \Hss $}
   \put(112.2,139) {\sse$ \Hss $}
   \put(137.7,139) {\sse$ \Hss $}
   } }
and
\eqpic{Qq_X} {160} {49} {\setulen80
    \put(0,49)  {$\Qq_{\HKH,X}=~$}
      \put(80,0) {\includepichtft{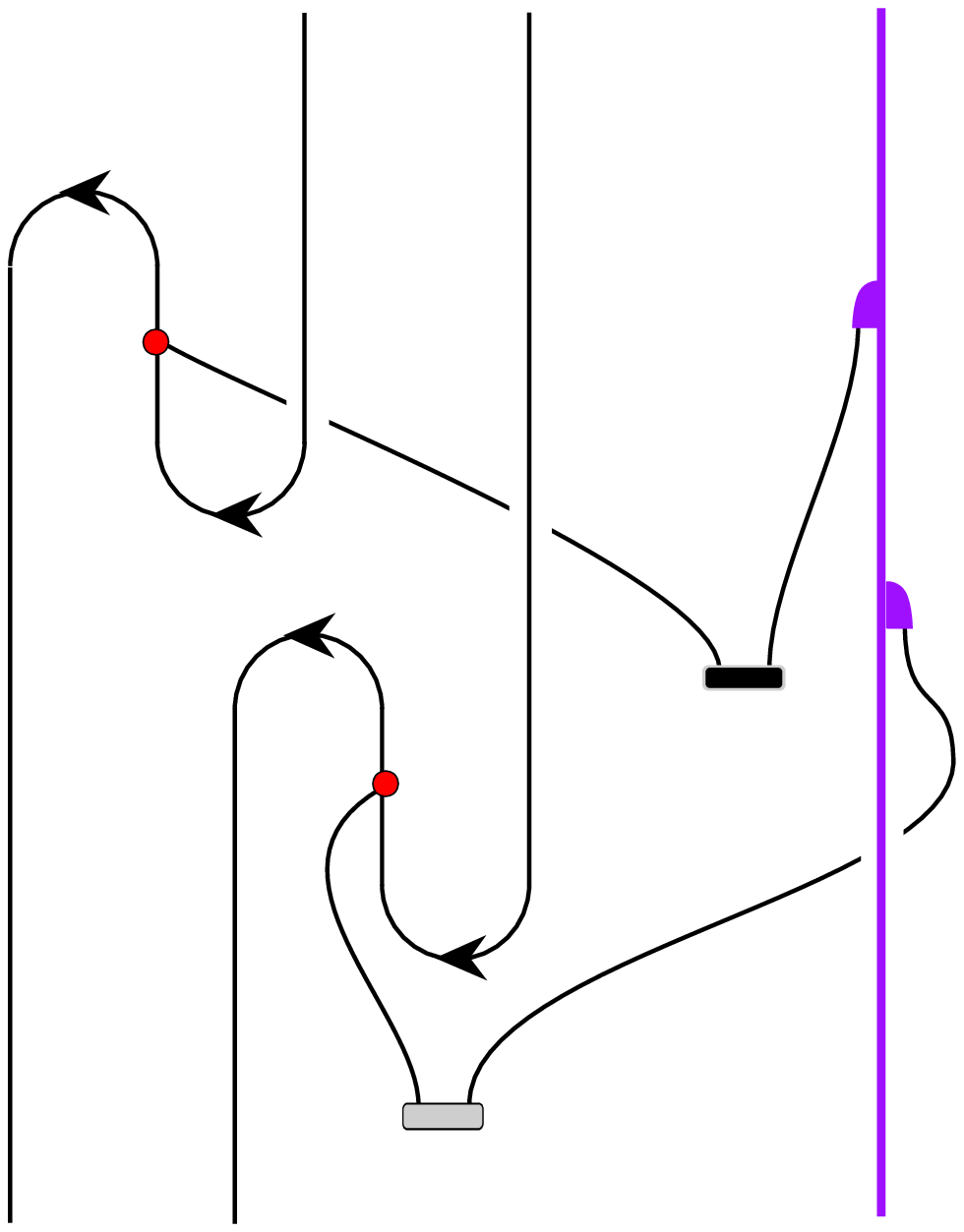}
  \put(-4,-8.5) {\sse$ \Hs $}
  \put(21,-8.5){\sse$ \Hs $}
  \put(93.2,-8.5){\sse$ X $}
  \put(53,4.5)  {\sse$ Q $}
  \put(73,50.5)  {\sse$ Q^{-1} $}
  \put(103,67.5)  {\sse$ \ohr_X$}
  \put(83,101.5)  {\sse$ \rho_X$}
  \put(30,140) {\sse$ \Hs $}
  \put(55,140){\sse$ \Hs $}
  \put(93.2,140){\sse$ X $}
  }
  }
Recall from \eqref{Hopf_pair_eps} that the Hopf pairing of $\HKH$ is given by $\omega_\HKH\eq(\eps_\HKH\oti\eps_\HKH)\cir \QQ_{\HKH,\HKH}$. It follows that $\omega_\HKH$
is non-degenerate iff the Drinfeld map of $H$ is invertible. Thus $H$\Bimod\ satisfies condition \CFN, i.e. $H$\Bimod\ is factorizable, iff $H$ is factorizable.

Inserting this result for $\Qq_{\HKH,X}$, together with $\eps\ASK \eq \eta^\vee \oti \eta^\vee$, into the expression \eqref{Lact_Qq} for the action of \L\ on any object in \C, it follows that the action of $\HKH$ on any $H$-bimodule $X$ is given by
\eqpic{Lyubact_HKH} {120} {47} {
    \put(0,47)  {$\rhoHKH X=~$}
      \put(50,4) {\Includepichtft{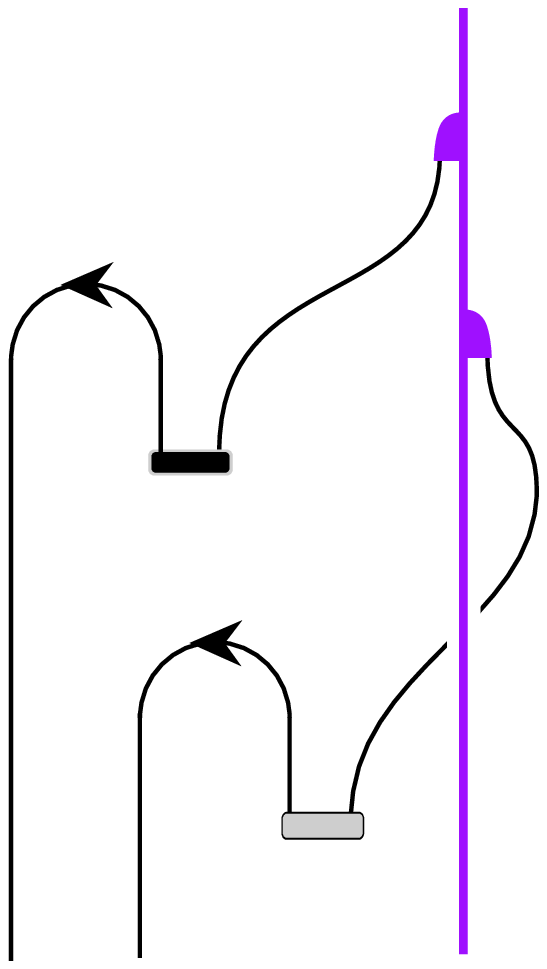}
  \put(-4,-6.5) {\sse$ \Hs $}
  \put(10,-6.5){\sse$ \Hs $}
  \put(47.2,-6.5){\sse$ X $}
  \put(30,7.5)  {\sse$ Q $}
  \put(17,47.5)  {\sse$ Q^{-1} $}
  \put(55,67.5)  {\sse$ \ohr_X$}
  \put(35,90.5)  {\sse$ \rho_X$}
  \put(47.2,109){\sse$ X $}
  }
  }
\ssubject{The maps $\TL$ and $\SL$}
In order to study the mapping class group representation in Proposition \ref{Lyubact_prop} we also need the maps $\TL$ and $\SL$, which in the $H$\Bimod\ case we denote by $\TKH$ and $\SKH$. First, inserting the dinatural family $\iHKH$, given  in \eqref{def_iHaa_X}, and the expression \eqref{bimod-twist} for the twist of \HBimod\ in the definition \eqref{def_TL} of $\TL$ we obtain
\eqpic{TKH} {170}{39} {
   \put(4,39)    {$ \TKH ~= $}
   \put(55,0)  { \Includepichtft{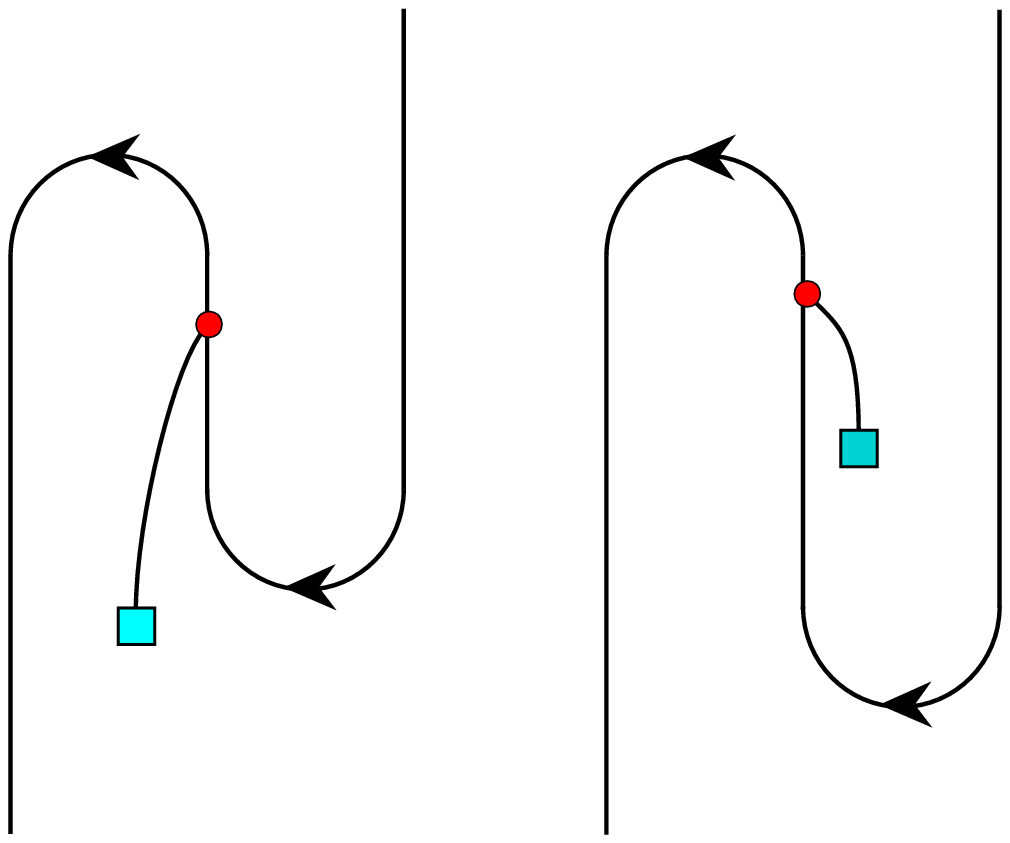}
   \put(-4.4,-8.5) {\sse$ \Hss $}
   \put(12,15)     {\sse$ v $}
   \put(39.6,94)   {\sse$ \Hss $}
   \put(61.2,-8.5) {\sse$ \Hss $}
   \put(89.7,34.5) {\sse$ v^{-1} $}
   \put(105.3,94)  {\sse$ \Hss $}
   }
   }
Second, inserting the expression \eqref{QQ_Haa} for $\QQ_{\HKH,\HKH}$, and the unit $\eps\ASK \eq \eta^\vee \oti \eta^\vee$ and integral $\Lambda\ASK \eq \lambda^\vee \oti \lambda^\vee$ of \HKH, in the definition \eqref{def_SK} of $\SL$ we obtain
\eqpic{S_KH} {140}{47} {
   \put(0,47)    {$ \SKH ~= $}
   \put(50,-2) { \Includepichtft{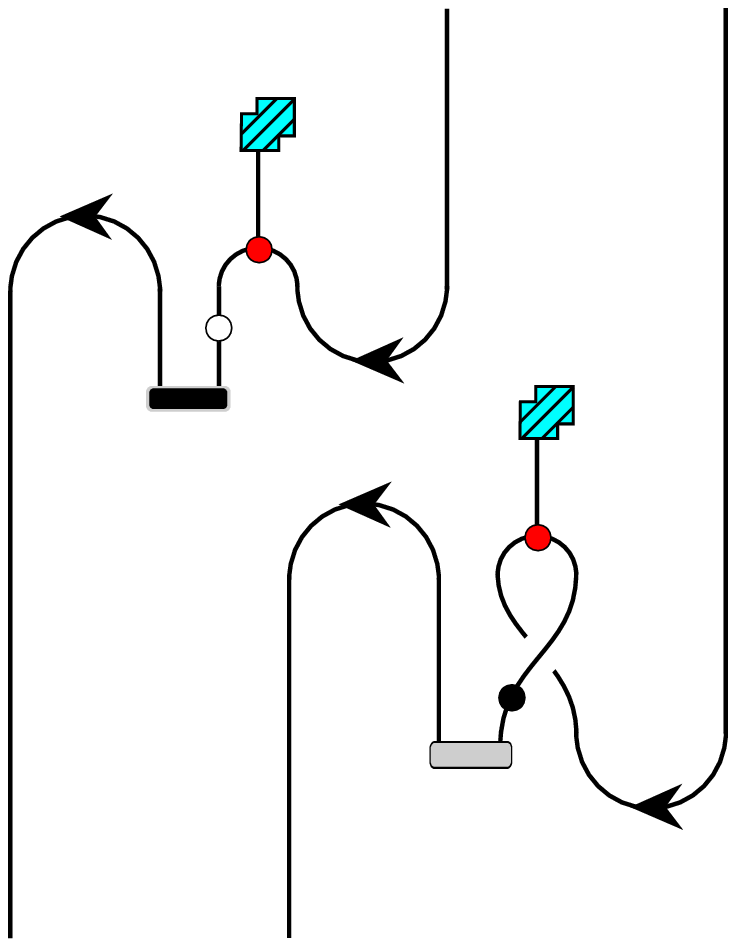}
   \put(-4.5,-8.5) {\sse$ \Hss $}
   \put(13,51.8)   {\sse$ Q^{-1} $}
   \put(26.5,-8.5) {\sse$ \Hss $}
   \put(32.5,87)   {\sse$ \lambda $}
   \put(44.6,106)  {\sse$ \Hss $}
   \put(47.4,12.1) {\sse$ Q $}
   \put(63.5,56)   {\sse$ \lambda $}
   \put(75.6,106)  {\sse$ \Hss $}
   } }
\begin{remark}
Note that, using the definitions \eqref{def_fmap} and \eqref{Drin_def} of the Frobenius map and the Drinfeld map, combined with the property \eqref{OS2} of $\lambda$, the map $\SKH$ can be written as
\be
    \SKH=(\Psi\circ f_{Q^{-1}})\oti(\Psi\circ f_{Q})\,.
\ee
Thus, $\SKH$ is invertible iff the Drinfeld map is invertible, i.e.\ iff $H$ is factorizable.
It is easily checked, using the relations \eqref{fmap_inv} and \eqref{fQS_Psi}, that the inverse is given by
\be\label{SKHinv}
    \SKH^{-1}=(\Psi\circ f_{Q})\oti(\Psi\circ f_{Q^{-1}})\,.
\ee
\end{remark}
\ssubject{The morphism $\Sk g11$}
We will also need the morphism $\Sk g11$, defined in \eqref{Sk_morph}, expressed as a linear map. In the $H$\Bimod\ case we denote the morphism $\Sk g11$ by $\SkH g11$. We start with $\SkH 101=\HCorr 11$. Inserting the expressions, given in \eqref{pic-Hb-Frobalgebra}, for the product and coproduct of $\BFH$, and the action  \eqref{Lyubact_HKH} of $\HKH$, combined with the left and right action \eqref{rhoHb,ohrHb} of $H$ on $\BFH$, in the definition \eqref{Sk_morph} of $\SkH 101$, we obtain
\Eqpic{PF_1} {320} {62} {\setulen95
   \put(-5,59)   { $\SkH 101~=$ }
   \put(53,0)  {\INcludepichtft{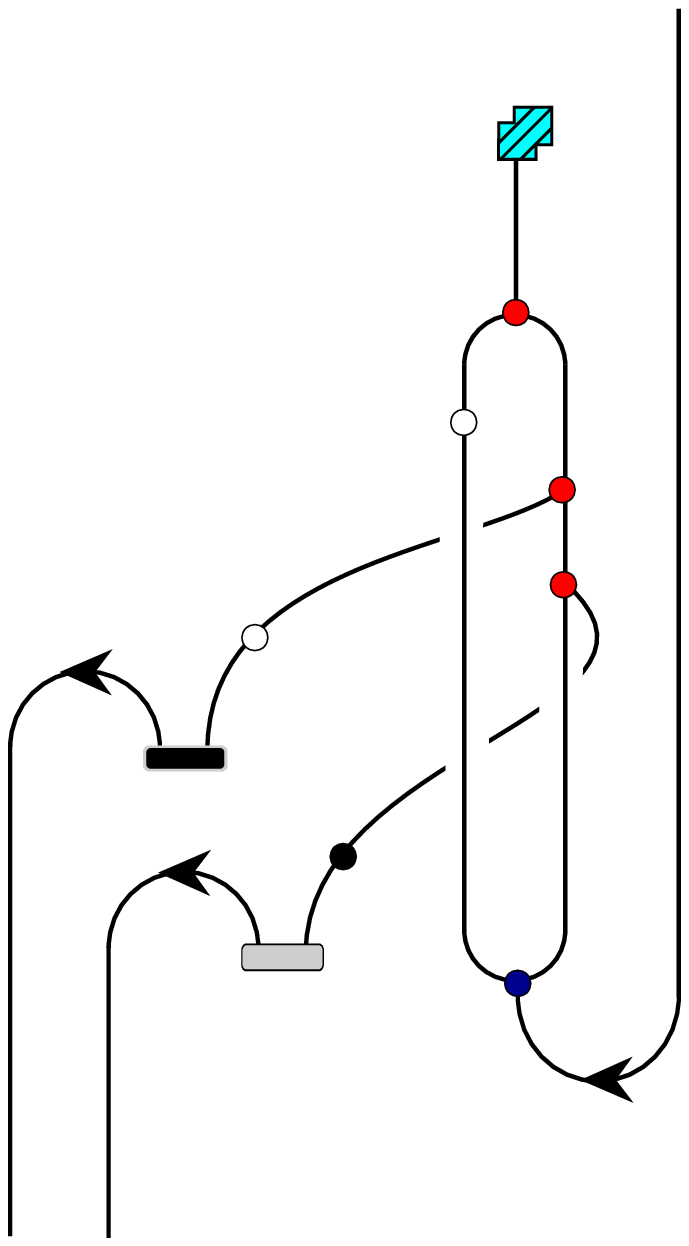}{361}
   \put(-5.5,-8.5) {\sse$ \Hss $}
   \put(7.4,-8.5){\sse$ \Hss $}
   \put(25.2,49) {\sse$ Q^{-1} $}
   \put(26.5,23) {\sse$ Q $}
   \put(61,118.7){\sse$ \lambda $}
   \put(71,139)  {\sse$ \Hss $}
   }
   \put(161,59)  {$=$}
   \put(194,0) {\INcludepichtft{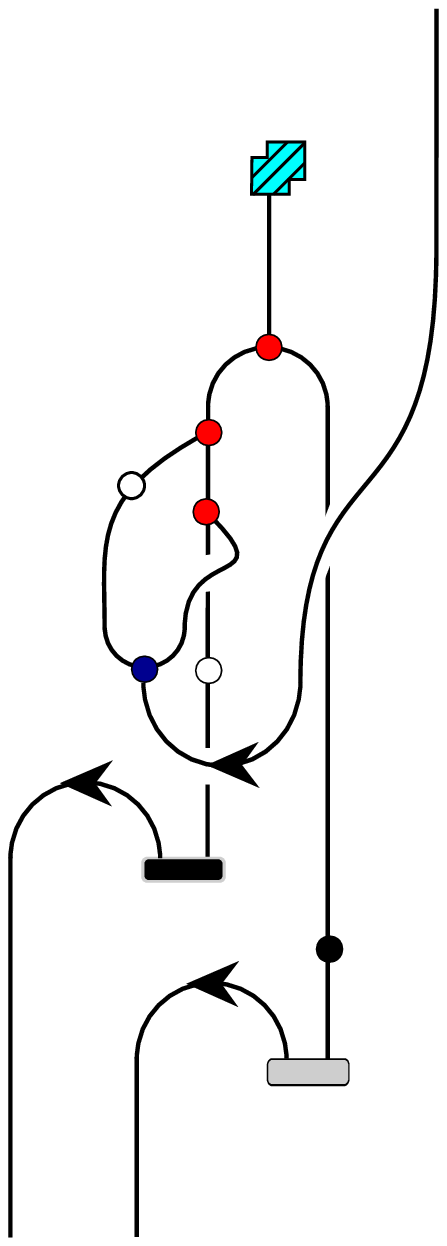}{361}
   \put(-4.5,-8.5) {\sse$ \Hss $}
   \put(10,-8.5) {\sse$ \Hss $}
   \put(11.5,33) {\sse$ Q^{-1} $}
   \put(29.5,10.8) {\sse$ Q $}
   \put(34.1,115){\sse$ \lambda $}
   \put(44,139)  {\sse$ \Hss $}
   }
   \put(264,59)  {$=$}
   \put(297,0) {\INcludepichtft{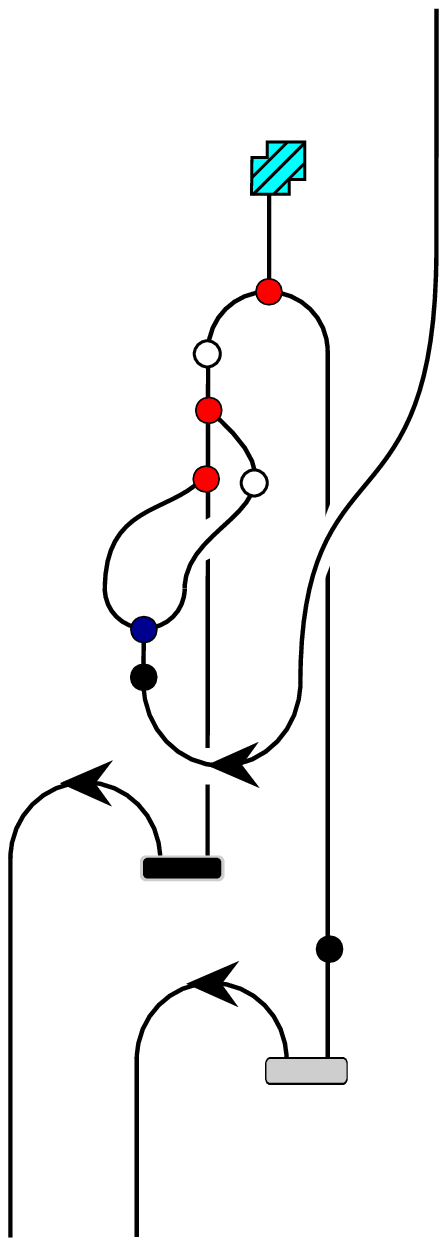}{361}
   \put(-4.5,-8.5) {\sse$ \Hss $}
   \put(10,-8.5) {\sse$ \Hss $}
   \put(11.5,33) {\sse$ Q^{-1} $}
   \put(29.5,10.8) {\sse$ Q $}
   \put(34.1,115){\sse$ \lambda $}
   \put(44,139)  {\sse$ \Hss $}
   } }
Here the second equality follows from associativity of $m$ and the third follows by using several times that $\apo$ is a anti-(co)algebra morphism. Next, using that the morphism $f_{Q^{-1}}$ intertwines the coadjoint and the adjoint left action, see \eqref{norm_int_inv}, we obtain the first of the following equalities:
\Eqpic{PF_2} {320} {45} {
   \put(10,0){
   \put(-3,49)    {$ \SkH 101~= $}\setulen95
   \put(63,0)  {\INcludepichtft{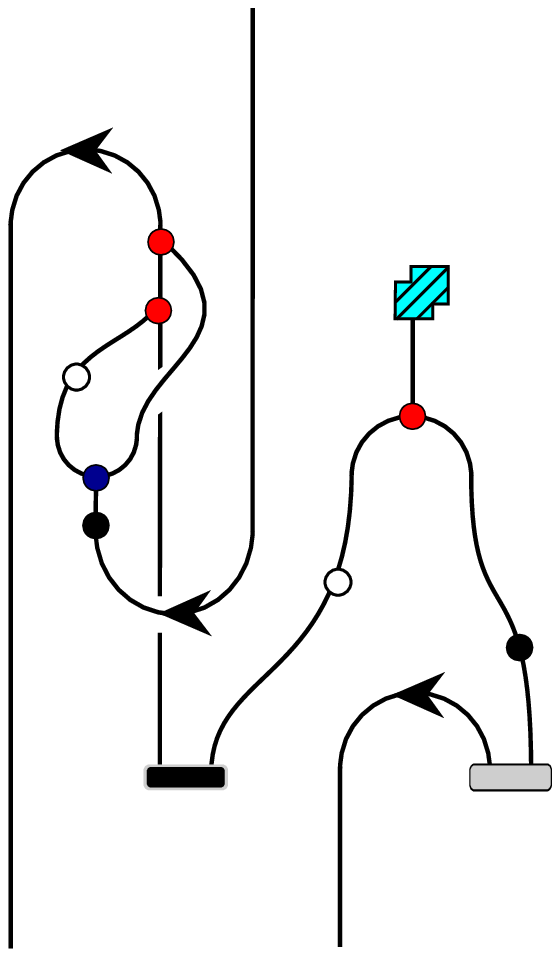}{361}
   \put(-5,-8.5) {\sse$ \Hss $}
   \put(32,-8.5) {\sse$ \Hss $}
   \put(12,11)   {\sse$ Q^{-1} $}
   \put(51.6,11) {\sse$ Q $}
   \put(49.7,69.5) {\sse$ \lambda $}
   \put(23,107)  {\sse$ \Hss $}
   }
   \put(143,49)  {$=$}
   \put(176,0) {\INcludepichtft{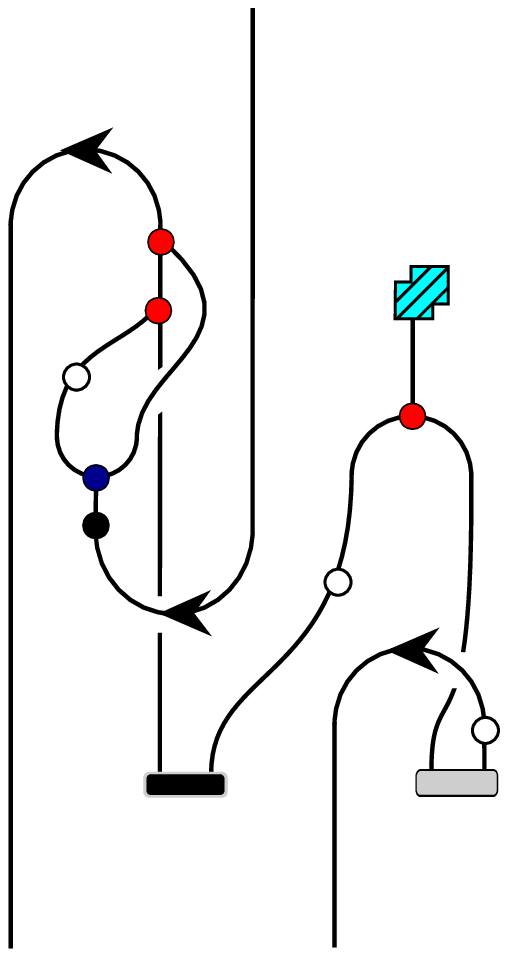}{361}
   \put(-5,-8.5) {\sse$ \Hss $}
   \put(31,-8.5) {\sse$ \Hss $}
   \put(12,11)   {\sse$ Q^{-1}$}
   \put(45.6,11) {\sse$ Q $}
   \put(48.9,69.2) {\sse$ \lambda $}
   \put(23,107)  {\sse$ \Hss $}
   }
   \put(252,49)  {$=$}
       \put(285,0) {\INcludepichtft{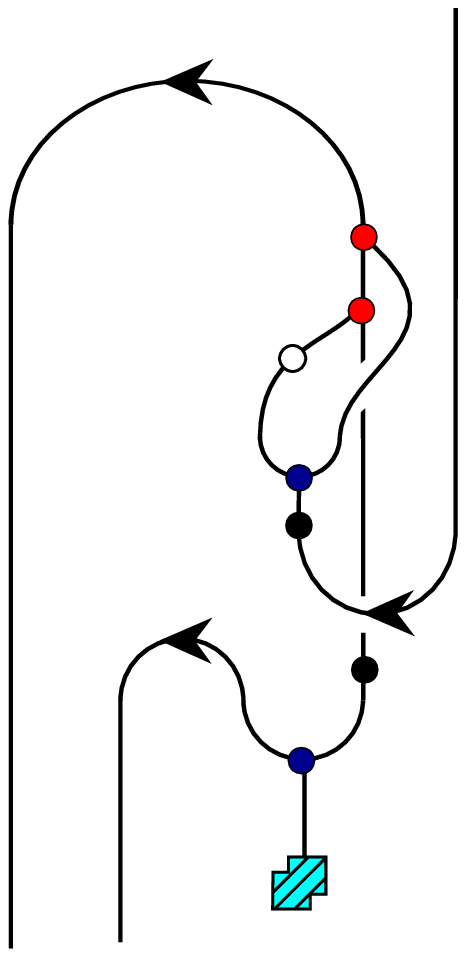}{361}
   \put(-5,-8.5) {\sse$ \Hss $}
   \put(8.2,-8.5){\sse$ \Hss $}
   \put(36,3.5)  {\sse$ \Lambda $}
   \put(46,107)  {\sse$ \Hss $}
   } } }
Here the second equality follows from ${\apo\oti\apo}\cir Q\eq \flip HH\circ Q$. The third equality is obtained by applying the second equality in \eqref{fQS_Psi} and then using that $\apoi$ is an anti-coalgebra  morphism and $\apo\circ\Lambda=\Lambda$.

Having obtained  the morphism $\SkH 101=\HCorr 11$ as a linear map, we can easily write down the linear map underlying $\SkH 111$. By the unit property, the Frobenius property and associativity of $\BFH$ we have
\be\label{SkHfrob}
    \SkH 111=m_{\BFH}\circ[(\SkH 111\circ(\id_{\HKH}\oti\eta_{\BFH}))\oti\id_{\BFH}]=m_{\BFH}\circ(\HCorr 11\oti\id_{\BFH})\,.
\ee
Inserting $m_\BFH\eq\Delta_H^\vee$ and the right-most expression in \eqref{PF_2} for $\HCorr 11=\SkH101$ in \eqref{SkHfrob} we obtain
\Eqpic{Corr12H} {320} {55} {\setulen80
    \put(0,55)  {$\SkH 111~=$}
    \put(60,23) {\includepichtft{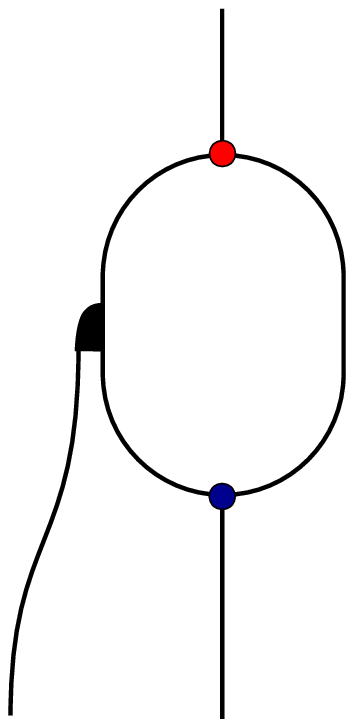}
  \put(-5,-8.5) {\sse$ \HKH $}
  \put(19.2,-8.5){\sse$ \BFH $}
  \put(21.2,82){\sse$ \BFH $}
  \put(0,109) {\Hbimodpic}
  }
  \put(120,55)  {$=$}
      \put(160,0) {\includepichtft{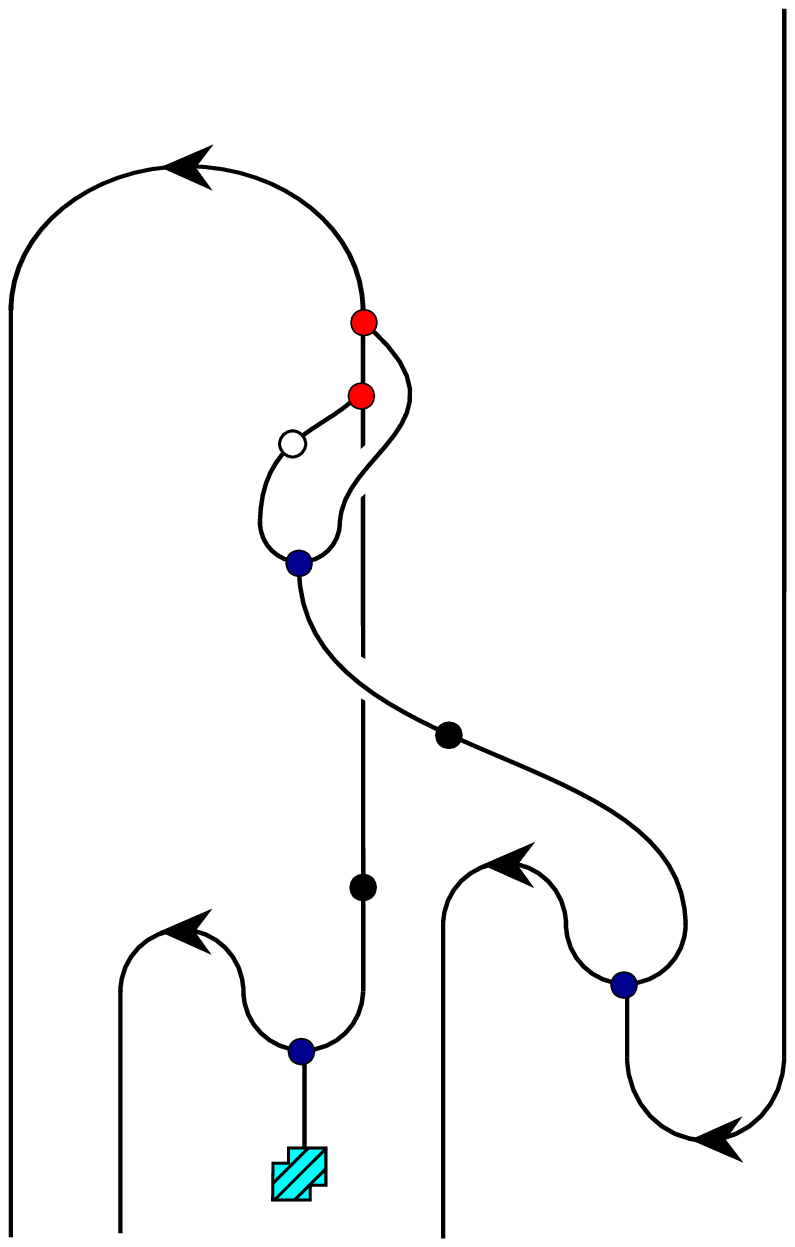}
  \put(-5,-8.5) {\sse$ \Hs $}
  \put(8.2,-8.5){\sse$ \Hs $}
  \put(43.2,-8.5){\sse$ \Hs $}
  \put(36,3.5)  {\sse$ \Lambda $}
  \put(81,139)  {\sse$ \Hs $}
  \put(0,132) {\Vectpic}
  }
  \put(270,55)  {$=$}
      \put(310,0) {\includepichtft{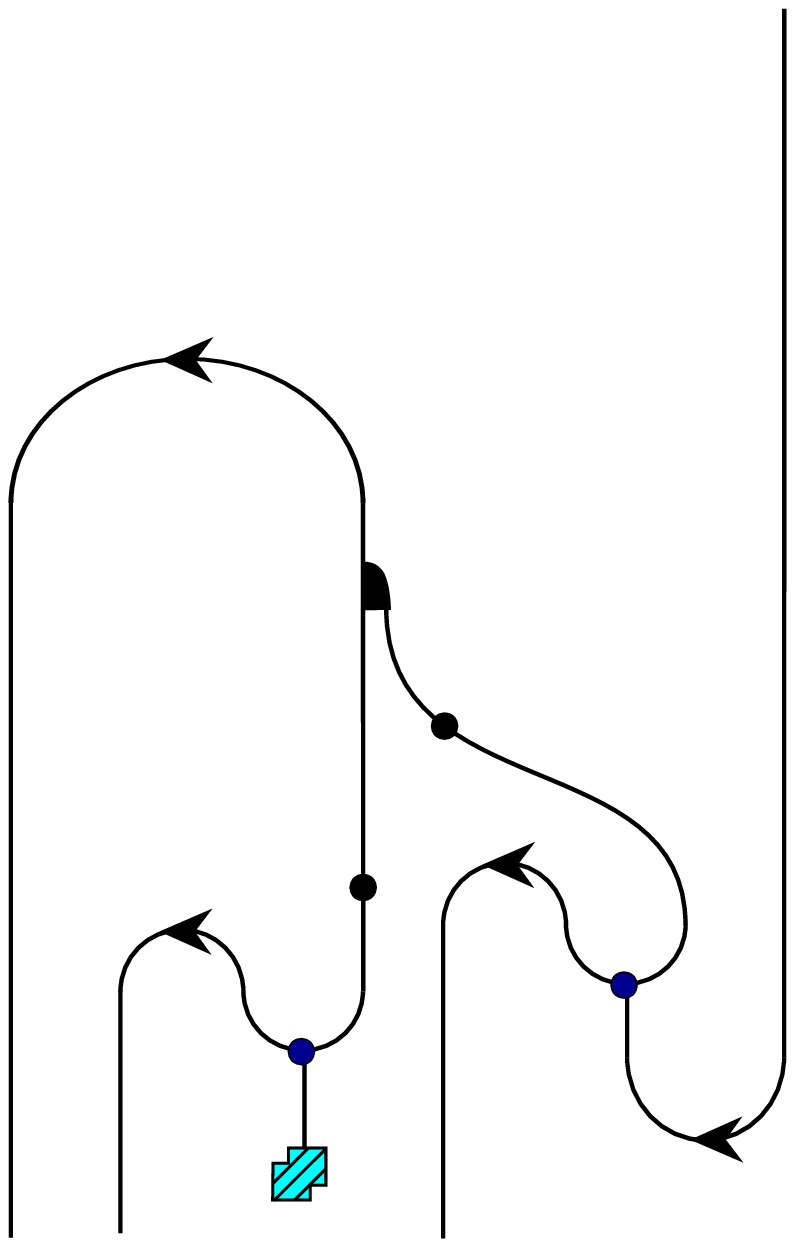}
  \put(-5,-8.5) {\sse$ \Hs $}
  \put(8.2,-8.5){\sse$ \Hs $}
  \put(43.2,-8.5){\sse$ \Hs $}
  \put(36,3.5)  {\sse$ \Lambda $}
  \put(43,69.5)  {\sse$ \rad$}
  \put(81,139)  {\sse$ \Hs $}
  \put(0,132) {\Vectpic}
  } }
Here, the third equality follows by recognizing the right adjoint action \eqref{adj_act}. Note that in \eqref{Corr12H}, the first picture is a morphism in $\HBimod$, whereas the two last pictures are linear maps.

Below we will need the linear map underlying $\SkH gpq$ only for $p\eq q\eq1$. From the definition \eqref{Sk_morph} of $\Sk gpq$, and the result \eqref{Corr12H} for $g=1$, it follows immediately that, as a linear map, $\SkH g11$ is given by
\eqpic{CorrgnH} {300} {117} { \put(0,18){ \setulen90
  \put(0,120)      {$\SkH g11=~$}
    \put(81,0) {\INcludepichtft{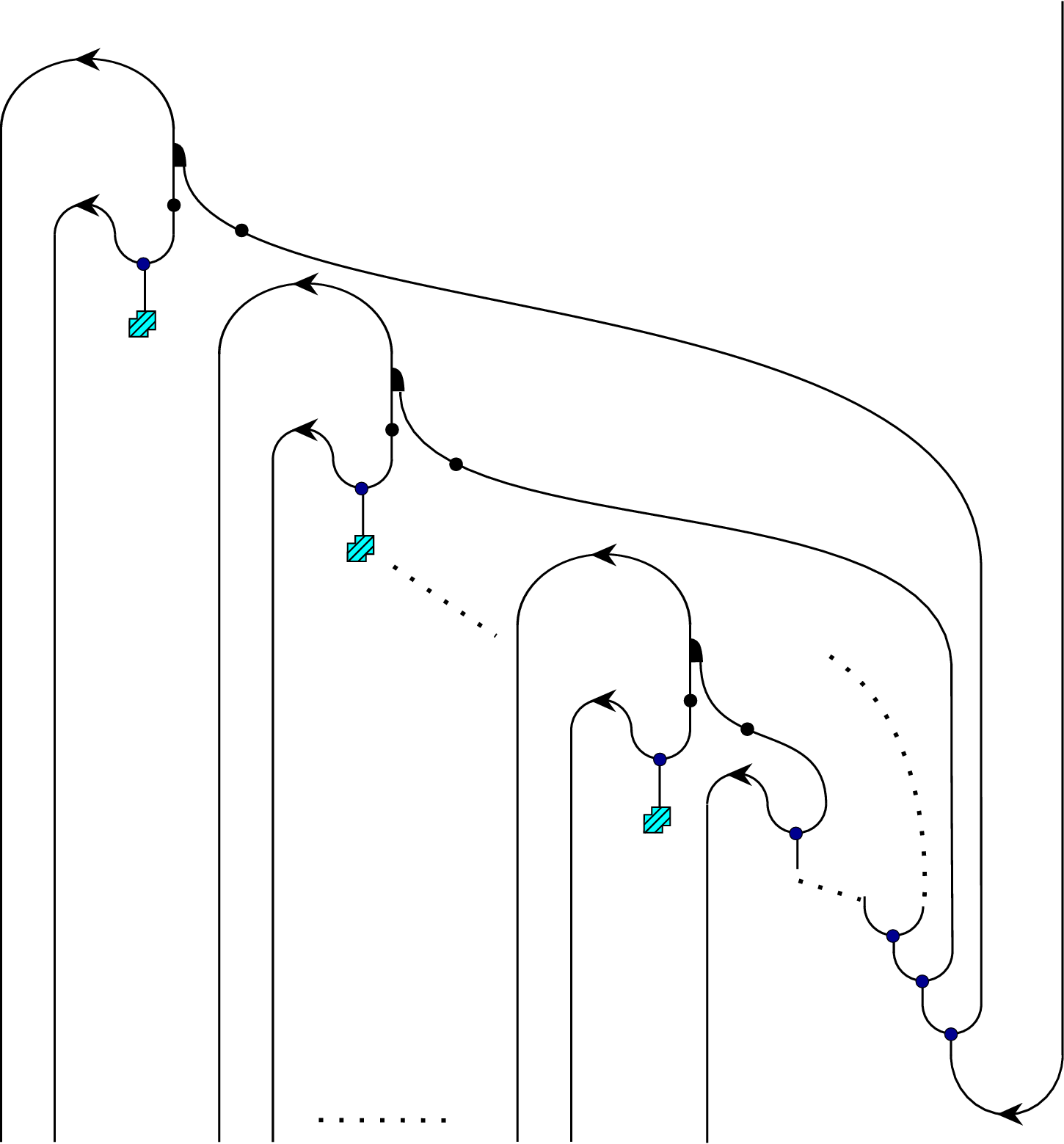}{342}
  \put(-5.9,-13)   {\sse$ \underbrace{\hspace*{17.5em}}
                    _{g~ {\rm factors~of}~\Hs{\otimes}\Hs} $}
  \put(-5,-9.2)    {\sse$ \Hss $}
  \put(8.2,-9.2)   {\sse$ \Hss $}
  \put(44.5,-9.2)  {\sse$ \Hss $}
  \put(56.5,-9.2)  {\sse$ \Hss $}
  \put(76.5,-9.2)  {\sse$ \dots\dots$}
  \put(111.5,-9.2) {\sse$ \Hss $}
  \put(123.5,-9.2) {\sse$ \Hss $}
  \put(154.5,-9.2) {\sse$ \Hss $}
  \put(234,260)    {\sse$ \Hss $}
  \put(36,180.5)   {\sse$ \Lambda$}
  \put(70,130.5)   {\sse$ \Lambda$}
  \put(137,69.5)   {\sse$ \Lambda$}
  \put(43,220)     {\sse$ \rad $}
  \put(91.4,169.4) {\sse$ \rad $}
  \put(158.2,108.5){\sse$ \rad $}
  } } }
Here, the dotted lines in the pictures means that a piece of the morphism is repeated $g$ times: There are $g\,{-}\,3$ additional occurrences of the "loops" containing the integral and there are $g\,{-}\,3$ additional coproducts such that the rightmost line coming out of this coproduct is ending on the adjoint right action on one of those "loops". Similar notation will be used in various pictures below.
\subsection{Proof of invariance}\label{sec:inv_proof}
This subsection is devoted to the proof of Theorem \ref{thm: MPG_Hopf}.
We will establish a series of lemmas, which together prove Theorem \ref{thm: MPG_Hopf}. First of all, we recall Lemma 5.8 and 5.9 of \cite{fuSs3}:
\begin{lemma}\label{Torus_inv}
    The one-point function on the torus $\HCorr 11=\SkH 111\circ(\id_{\HK}\oti\eta_{\BF})$ satisfies
    \be\label{Sinv_1p}
        \HCorr 11\circ\SKH=\HCorr 11\,
    \ee
    and
    \be\label{T_inv_11}
        \HCorr 11\circ\TKH=\HCorr 11\,.
    \ee
\end{lemma}
\begin{proof}
    (i)
We will show invariance under  $\SKH^{-1}$. Recall from \eqref{SKHinv} that $\SKH^{-1}$ is given by
\eqpic{S_KHinv} {140}{42} {
   \put(0,42)    {$ \SKH^{-1} ~= $}
   \put(50,-6) { \Includepichtft{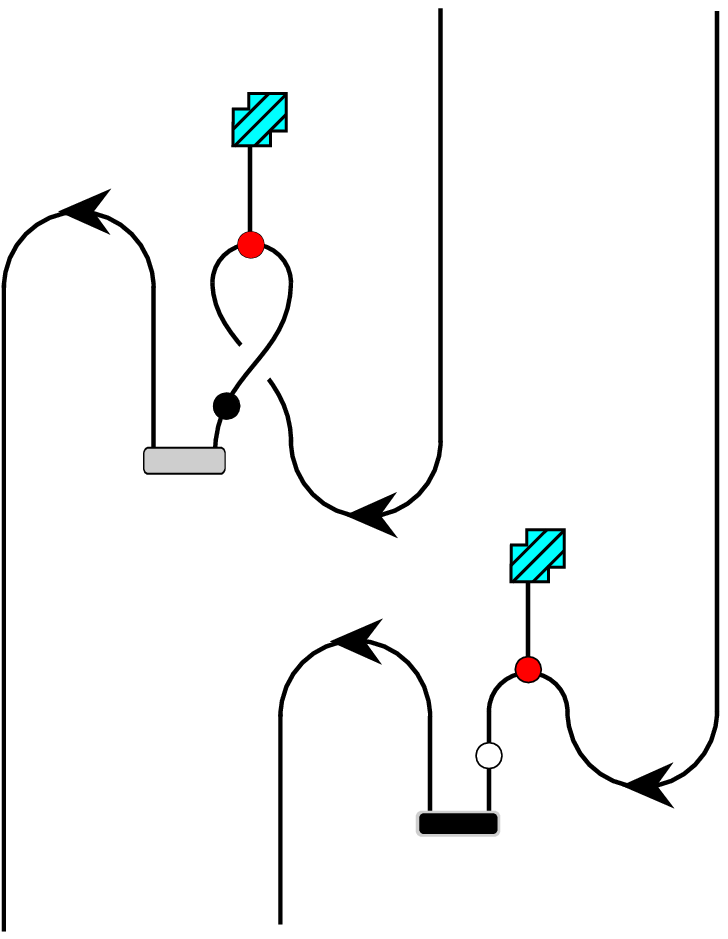}
   \put(-4.5,-6.5) {\sse$ \Hss $}
   \put(16,44.8)   {\sse$ Q $}
   \put(26.5,-6.5) {\sse$ \Hss $}
   \put(32.5,87)   {\sse$ \lambda $}
   \put(44.6,106)  {\sse$ \Hss $}
   \put(47.4,5.1) {\sse$ Q^{-1} $}
   \put(63.5,39)   {\sse$ \lambda $}
   \put(75.6,106)  {\sse$ \Hss $}
   } }
Thus, composing $\SKH^{-1}$ with the first expression in \eqref{PF_1} we obtain
\Eqpic{Sinv-6} {320} {67} {
   \put(7,67)  {$\HCorr 11\circ\SKH^{-1} ~=$}
       \put(117,0) {\Includepichtft{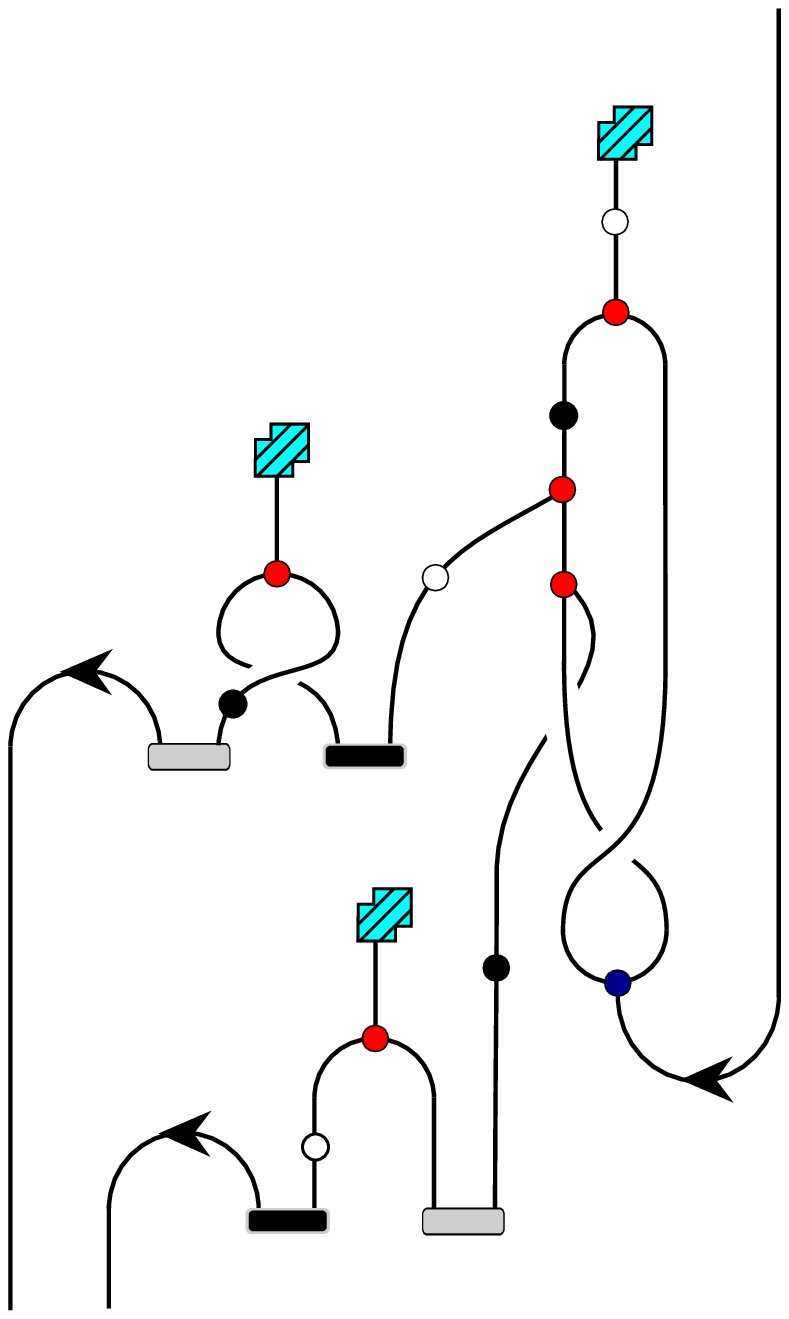}
   \put(-5.4,-8.5) {\sse$ \Hss $}
   \put(7,-8.5)  {\sse$ \Hss $}
   \put(16.3,52.8) {\sse$ Q $}
   \put(31.9,53.2) {\sse$ Q^{-1} $}
   \put(23.7,2.2){\sse$ Q^{-1} $}
   \put(47.1,2.3){\sse$ Q $}
   \put(32.3,39) {\sse$ \lambda $}
   \put(34.5,91) {\sse$ \lambda $}
   \put(72,126)  {\sse$ \lambda $}
   \put(81,147)  {\sse$ \Hss $}
   }
   \put(225,63)  {$=$}
       \put(257,0) {\Includepichtft{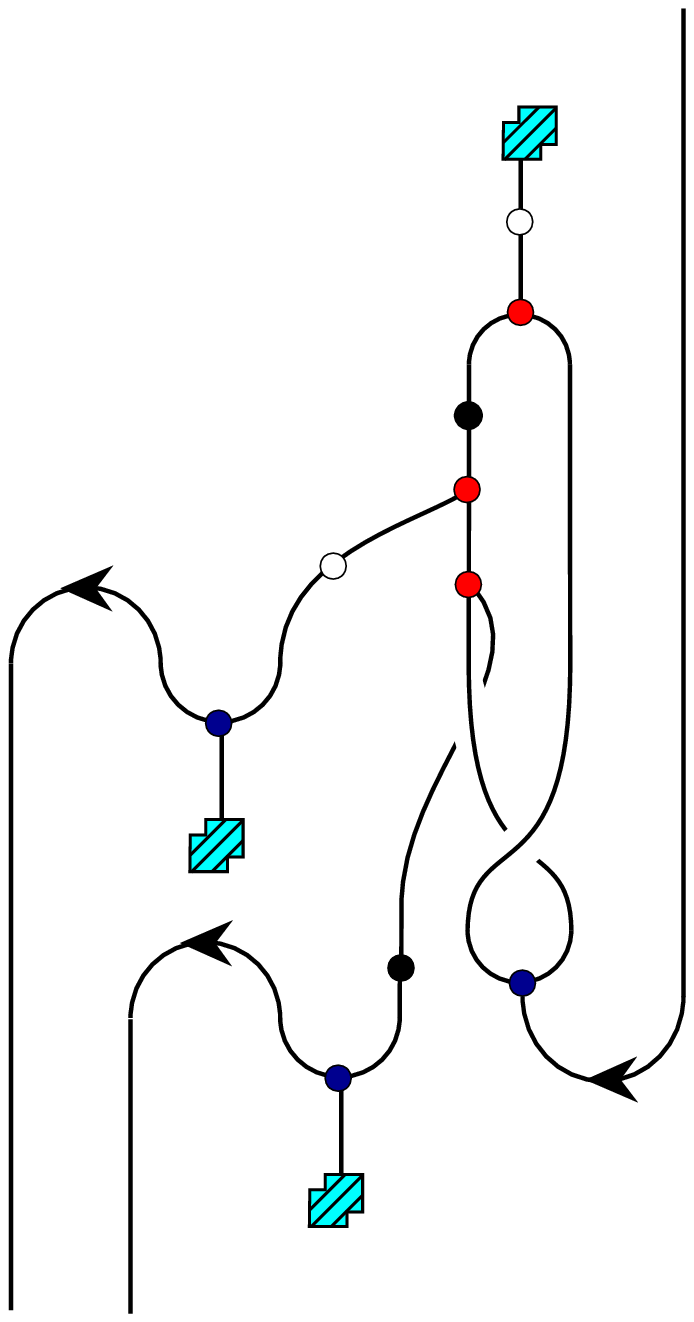}
   \put(-5,-8.5) {\sse$ \Hss $}
   \put(9,-8.5)  {\sse$ \Hss $}
   \put(26.5,48) {\sse$ \Lambda $}
   \put(40,9)    {\sse$ \Lambda $}
   \put(60.4,126){\sse$ \lambda $}
   \put(70.5,147){\sse$ \Hss $}
   } }
Here we have used the anti-algebra morphism property of the antipode $\apo$ in the first equality. The second equality is obtained by applying \eqref{fQS_Psi} together with the property \eqref{OS2} of $\lambda$.
According to Lemma \ref{lemma:Sinv-7}, the last expression in \eqref{Sinv-6} equals
\eqpic{Sinv-8} {340} {62} {
   \put(4,67)  {$\HCorr 11\circ\SKH^{-1} ~= $}
   \put(103,0) {\Includepichtft{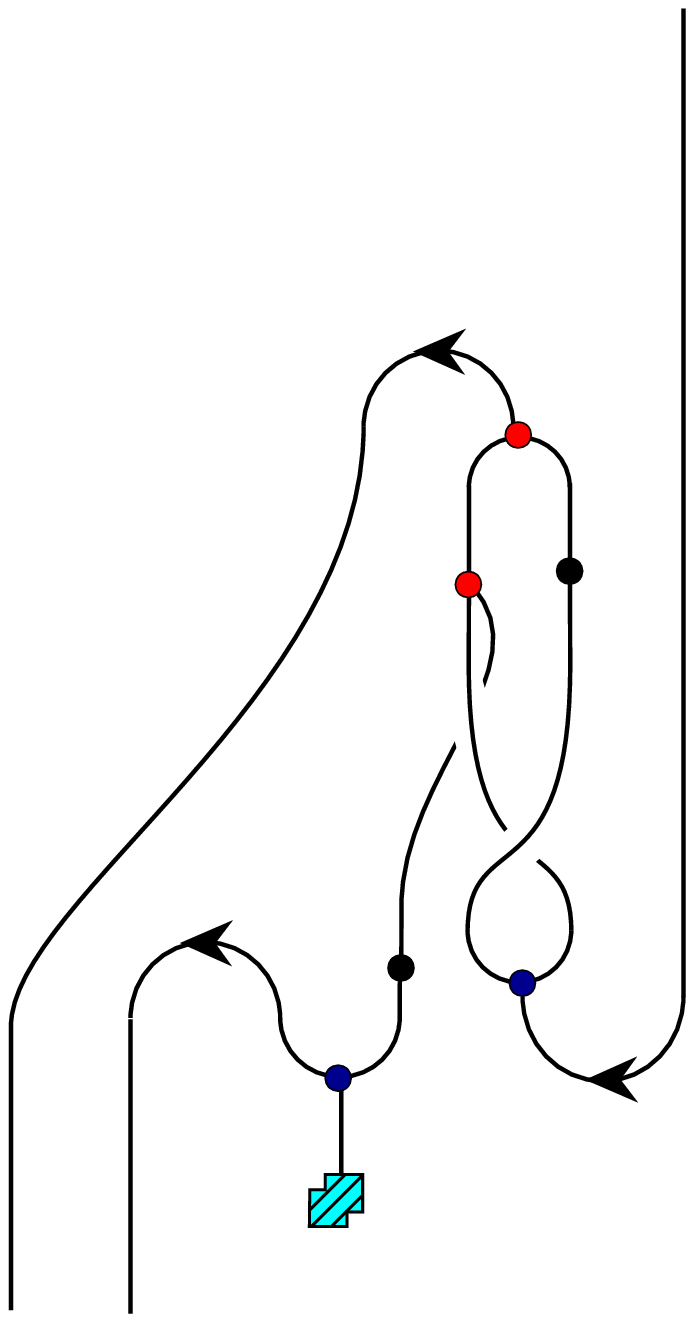}
   \put(-5,-8.5) {\sse$ \Hss $}
   \put(9,-8.5)  {\sse$ \Hss $}
   \put(40,9)    {\sse$ \Lambda $}
   \put(70.5,147){\sse$ \Hss $}
   }
   \put(210,67)  {$ = $}
   \put(241,0)   {\Includepichtft{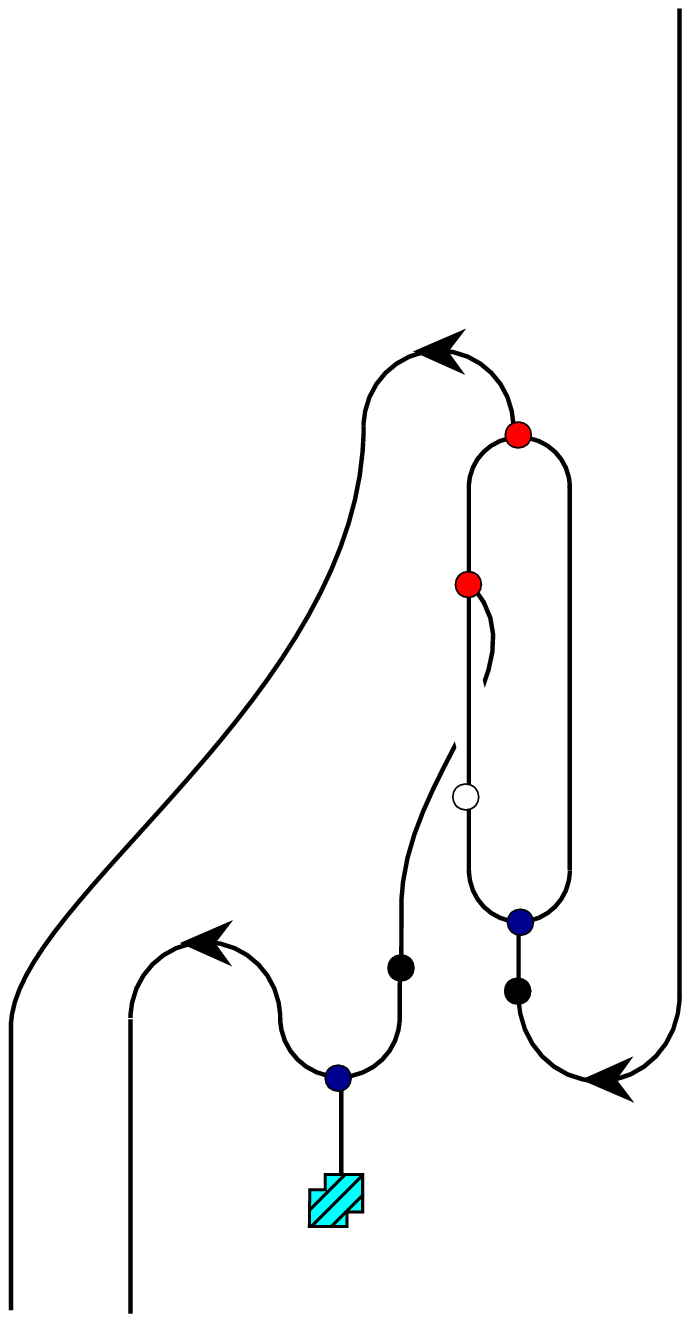}
   \put(-5,-8.5) {\sse$ \Hss $}
   \put(9,-8.5)  {\sse$ \Hss $}
   \put(40,9)    {\sse$ \Lambda $}
   \put(70.5,147){\sse$ \Hss $}
   } }
The second expression in \eqref{Sinv-8} coincides with the right-most expression in \eqref{PF_2} for $\HCorr 11$. This proves $\HCorr 11\circ\SKH=\HCorr 11$.\nxl1
    (ii) From the expression \eqref{TKH} for $\TKH$ and the last expression in \eqref{PF_2} it follows that
\eqpic{T-inv} {190} {45} {
   \put(10,52)   {$ \HCorr 11\circ\TKH ~= $}
   \put(118,0) {\Includepichtft{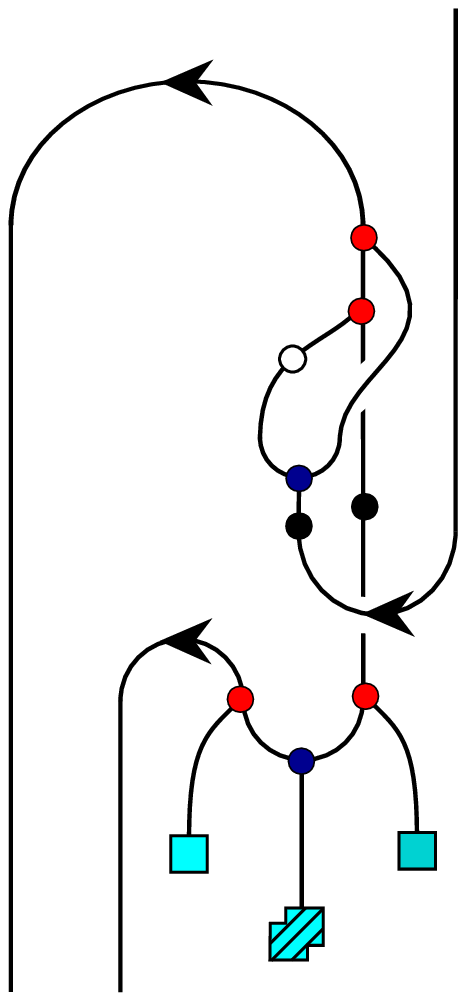}
   \put(-5,-8.5) {\sse$ \Hss $}
   \put(8.5,-8.5){\sse$ \Hss $}
   \put(23.5,13.5) {\sse$ v $}
   \put(49,13.5) {\sse$ v^{-1} $}
   \put(36,2.5)  {\sse$ \Lambda $}
   \put(45.2,112){\sse$ \Hss $}
   } }
After applying \eqref{Hopf_Frob_trick2}, the central elements $v$ and $v^{-1}$ cancel each other. This proves $\HCorr 11\circ\TKH=\HCorr 11$.
\endofproof\end{proof}
A direct consequence of this result is
\begin{lemma}\label{inv_STetci} For any integer $k\geq0$, the morphism $\SkH k11$ satisfies
\be\label{inv_S}
    \SkH k11\circ(\id_{\HKH}^{\otii r}\oti\SKH\oti \id_{\HKH}^{\otii s}\oti\id_{\BFH})=\SkH k11\,
\ee
and
\be\label{inv_T}
    \SkH k11\circ(\id_{\HKH}^{\otii r}\oti\TKH\oti \id_{\HKH}^{\otii s}\oti\id_{\BFH})=\SkH k11\,,
\ee
for all $r+s\eq k-1$.
\end{lemma}
\begin{proof}
From \eqref{SkHfrob} and Lemma \ref{Torus_inv} it follows that the morphism $\SkH 111$ satisfies
\be
    \SkH 111\circ(\TKH\oti\id_{\BFH})=\SkH 111 \quand \SkH 111\circ(\SKH\oti\id_{\BFH})=\SkH 111.
\ee
These equalities directly imply \eqref{inv_S} and \eqref{inv_T}.
\endofproof\end{proof}
\begin{lemma}\label{inv_twist}
$\SkH k11$ satisfies
       \be\label{lemma_twist}
            \SkH k11\circ(\id_{\HKH}^{\oti r}\oti\theta_{\HKH^{\otii s}}\oti \id_{\HKH}^{\oti t}\oti\id_{\BFH})=\SkH k11\,,
       \ee
for any integers $r,s,t$ such that $r+s+t= k$.
\end{lemma}
\begin{proof}
First we apply the expression \eqref{bimod-twist} for the twist and use Lemma \ref{lemma_helpid} (ii) to replace the left and right actions of $v$ and $v^{-1}$ on $\HKH^{\otii s}$ by left and right multiplication. Then, the left hand side of \eqref{lemma_twist} differs from the right hand side by a left multiplication with $v\circ\apo$ and a right multiplication with $v^{-1}\circ\apoi$, both of them multiplying the same $H$-line. The equality \eqref{lemma_twist} then follows by using $v\circ\apo=v$ and that the ribbon element $v$ is central.
\endofproof\end{proof}

\begin{lemma}\label{lemma:inv_QQ}
    For any $k=2,...,g$, we have
    \be\label{inv_QQ}
    \HCorrgn\circ(\id_{\HKH}^{\otii g-k}\oti\QQ_{\HKH,\HKH}\oti \id_{\HKH}^{\otii k-2})=\HCorrgn\,.
\ee
\end{lemma}
\begin{proof}
From the expression \eqref{QQ_Haa} for $\QQ_{\HKH,\HKH}$ and the expression \eqref{CorrgnH} for $\SkH 211$ it follows that
\Eqpic{Pic_QQ1} {320} {76} {\setulen80
  \put(0,92)   {$\SkH 211\circ(\QQ_{\HKH,\HKH}\oti\id_{\BFH})=$}
  \put(130,0) {\includepichtft{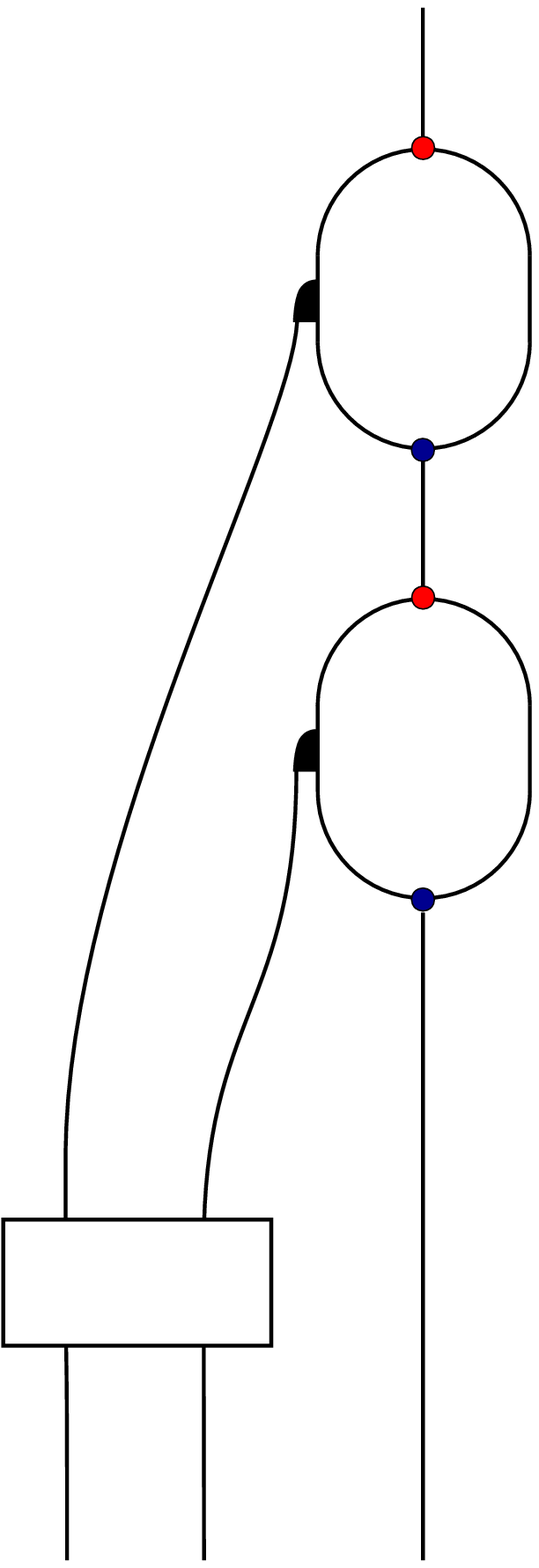}
  \put(20,-8.5) {\sse$ \HKH $}
  \put(2,-8.5) {\sse$ \HKH $}
  \put(47,-8.5) {\sse$ \BFH $}
  \put(2,33) {\sse$ \QQ_{\HKH,\HKH} $}
  \put(47,200) {\sse$ \BFH $}
  \put(-15,185)   {\Hbimodpic}
  }
  \put(215,92)   {$=$}
  \put(250,0) {\includepichtft{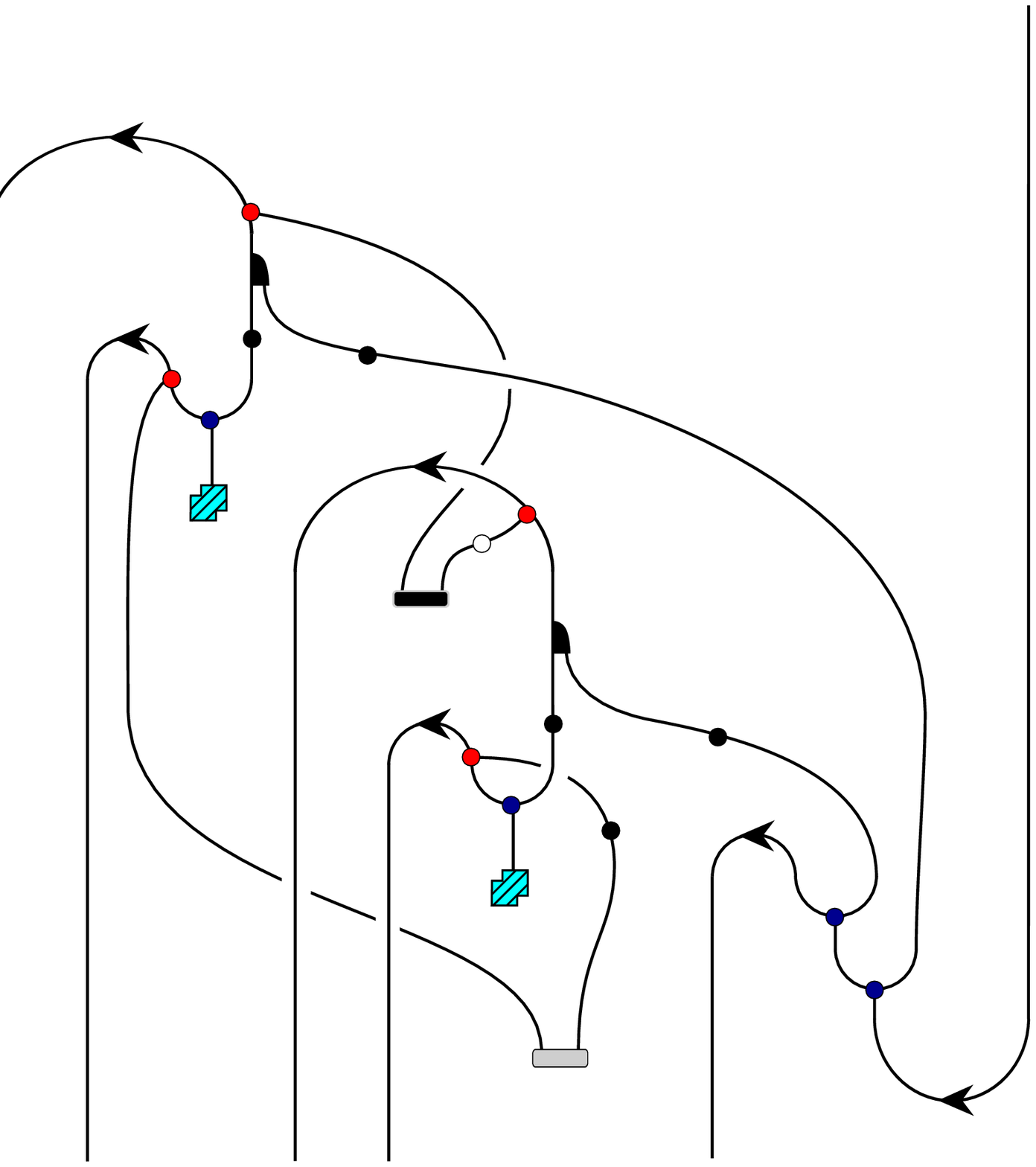}
  \put(-6,-11.2)   {\sse$ \Hss $}
  \put(11,-11.2)   {\sse$ \Hss $}
  \put(43.7,-11.2) {\sse$ \Hss $}
  \put(60,-11.2)   {\sse$ \Hss $}
  \put(112.2,-11.2){\sse$ \Hss $}
  \put(39.5,105.5) {\sse$ \Lambda $}
  \put(88.5,42.5)  {\sse$ \Lambda $}
  \put(88,7.5)     {\sse$ Q$}
  \put(62,82)      {\sse$ Q^{-1}$}
  \put(96,81.5)    {\sse$\rad$}
  \put(46,142.5)   {\sse$\rad$}
  \put(164,192)    {\sse$ \Hss $}
  \put(-15,185)   {\Vectpic}
  }
  }
Using first \eqref{Hopf_Frob_trick2} and then that $\apoi$ is an anti-algebra morphism we obtain the left hand side of the following equality:
\Eqpic{Pic_QQ2} {320} {87} {\setulen90
  \put(5,3) {\INcludepichtft{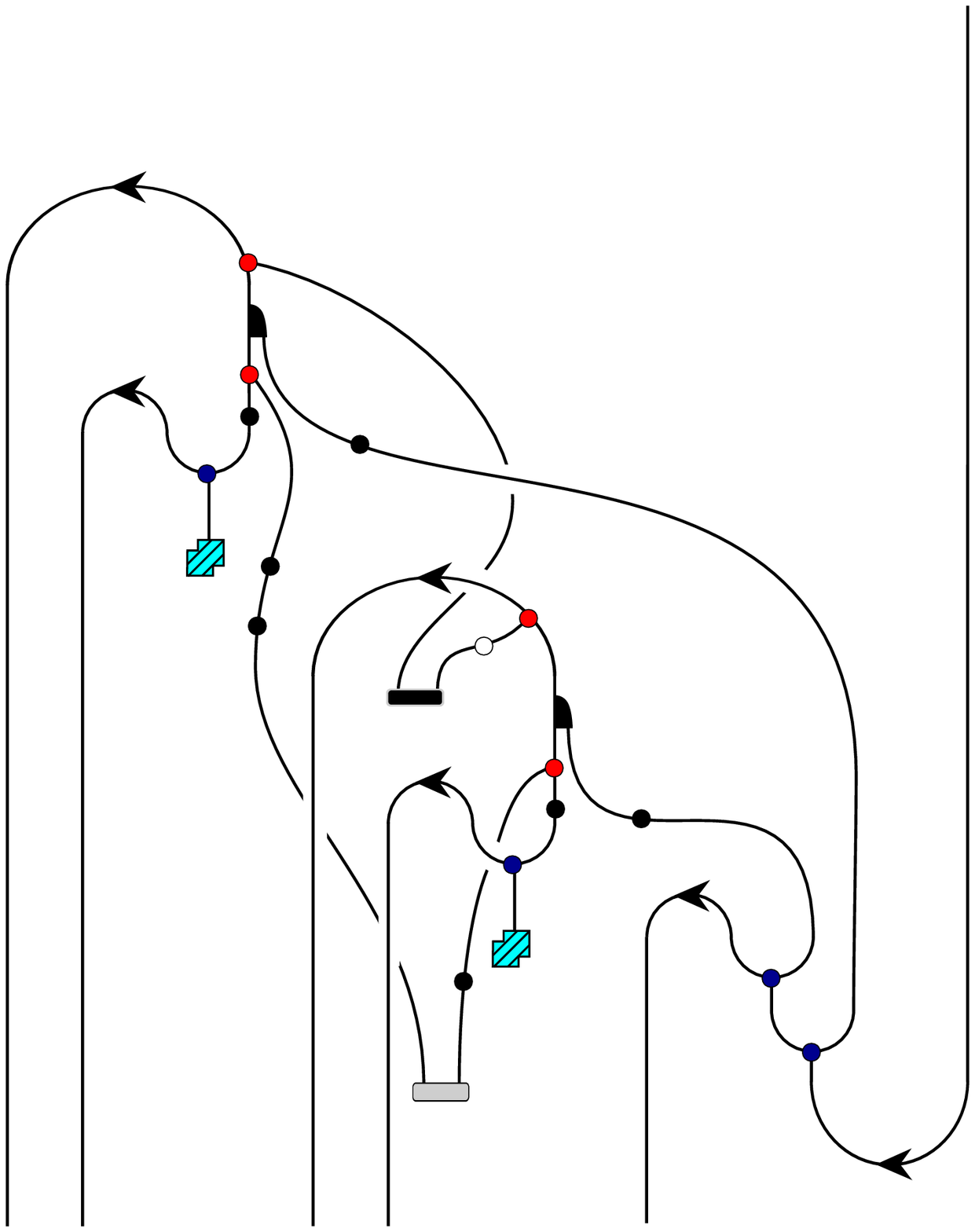}{342}
  \put(-4,-8.5) {\sse$ \Hs $}
  \put(8,-8.5) {\sse$ \Hs $}
  \put(44.5,-8.5){\sse$ \Hs $}
  \put(56,-8.5){\sse$ \Hs $}
  \put(98.2,-8.5){\sse$ \Hs $}
  \put(31,97.5)  {\sse$ \Lambda $}
  \put(85,42.5)  {\sse$ \Lambda $}
  \put(74,14.5)  {\sse$ Q$}
  \put(69,78.5)  {\sse$ Q^{-1}$}
  \put(92,80.5)  {\sse$\rad$}
  \put(42,142.5)  {\sse$\rad$}
  \put(151,199)  {\sse$ \Hs $}
  }
  \put(180,92)   {$=$}
  \put(210,0) {\INcludepichtft{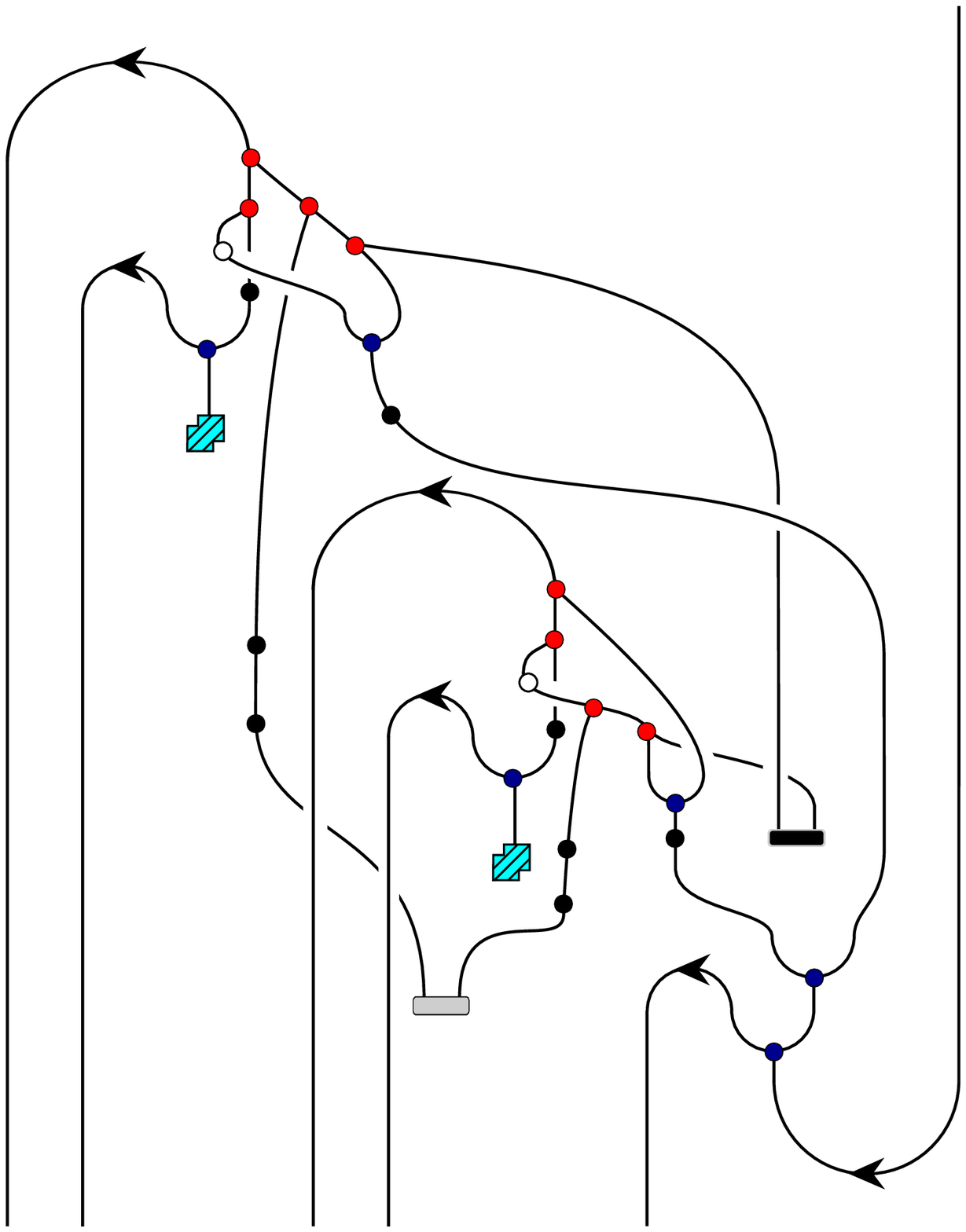}{342}
  \put(-4,-8.5) {\sse$ \Hs $}
  \put(8,-8.5) {\sse$ \Hs $}
  \put(44.5,-8.5){\sse$ \Hs $}
  \put(56,-8.5){\sse$ \Hs $}
  \put(98.2,-8.5){\sse$ \Hs $}
  \put(31,117.5)  {\sse$ \Lambda $}
  \put(70,54.5)  {\sse$ \Lambda $}
  \put(65,26.5)  {\sse$ Q$}
  \put(119,53.5)  {\sse$ Q^{-1}$}
  \put(147,199)  {\sse$ \Hs $}
  }
  }
The right hand side of \eqref{Pic_QQ2} is obtained by inserting the definition \eqref{adj_act} of the right adjoint action followed by associativity of $H$ and again that the antipode is the anti-algebra morphism.\\
By applying $(\apo^2\oti\apo^2)\circ Q=Q$ and then Lemma \ref{removeQ_2} (i)
we can omit the $Q$-matrices on the left hand side of \eqref{Pic_QQ2}.
Comparing the result with \eqref{CorrgnH} it follows that
\be\label{QQ_inv}
    \SkH 211\circ(\QQ_{\HKH,\HKH}\oti\id_{\BFH})=\SkH 211\,.
\ee
Finally, the left hand side of \eqref{inv_QQ} differs from the right hand side exactly by the action the application of $\QQ_{\HKH,\HKH}$ on two adjacent $\HKH$-factors. Thus \eqref{inv_QQ} follows directly from \eqref{QQ_inv}.
\endofproof\end{proof}

\begin{lemma}\label{inv_Qq_S}
For any integer $k\geq2$, $\SkH k02$ satisfies
\be\label{inv_Qq_S_state}
    (\id_\BF\oti\tilde b_\BF)\circ(\SkH k02\oti\id_{^\vee\!\BF})\circ\Qq_{\HK,\HK^{\otii k-1}\oti^\vee\!\BF}=(\id_\BF\oti\tilde b_\BF)\circ(\SkH k02\oti\id_{^\vee\!\BF})\,
\ee
in $\Hom_{H|H}(\HK^{\otii k}\oti^\vee\!\BF,\BF)$.\\
That is, graphically
\eqpic{removeQq_t} {250} {80} {
    \put(0,0) {\includepichtft{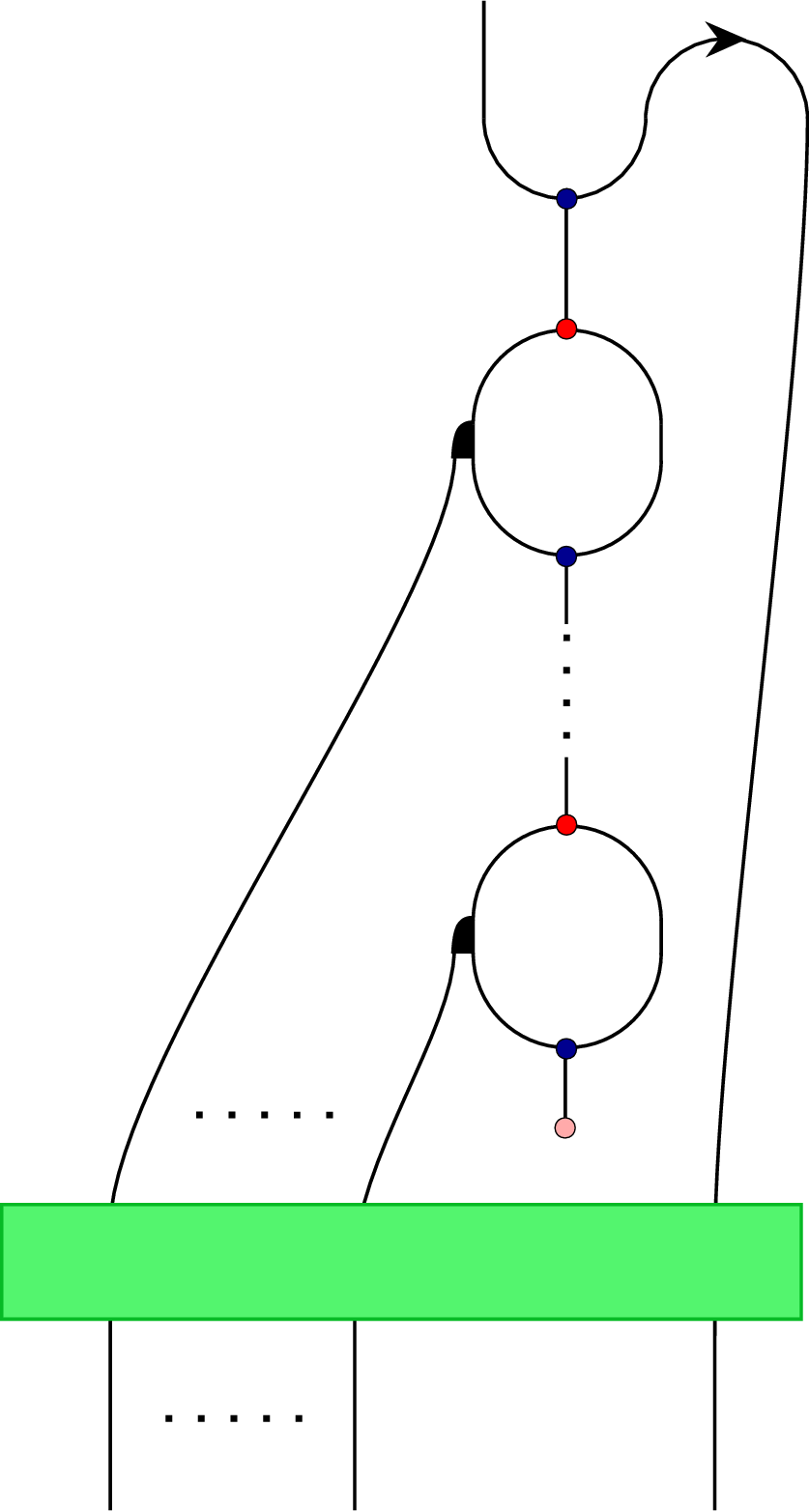}
    \put(17,26.7)   {\sse$\Qq_{\HKH,\HKH^{\otii k-1}\oti^\vee\!\BFH}$}
    \put(8,-8.5)   {\sse$\HKH$}
    \put(25,-8.5)   {\sse$\dots$}
    \put(36,-8.5)   {\sse$\HKH$}
    \put(77,-8.5)   {\sse$^\vee\!\BFH$}
    \put(53,174)   {\sse$\BFH$}
    }
    \put(130,80)    {$=$}
    \put(150,0) {
    \put(13,0)   {\includepichtft{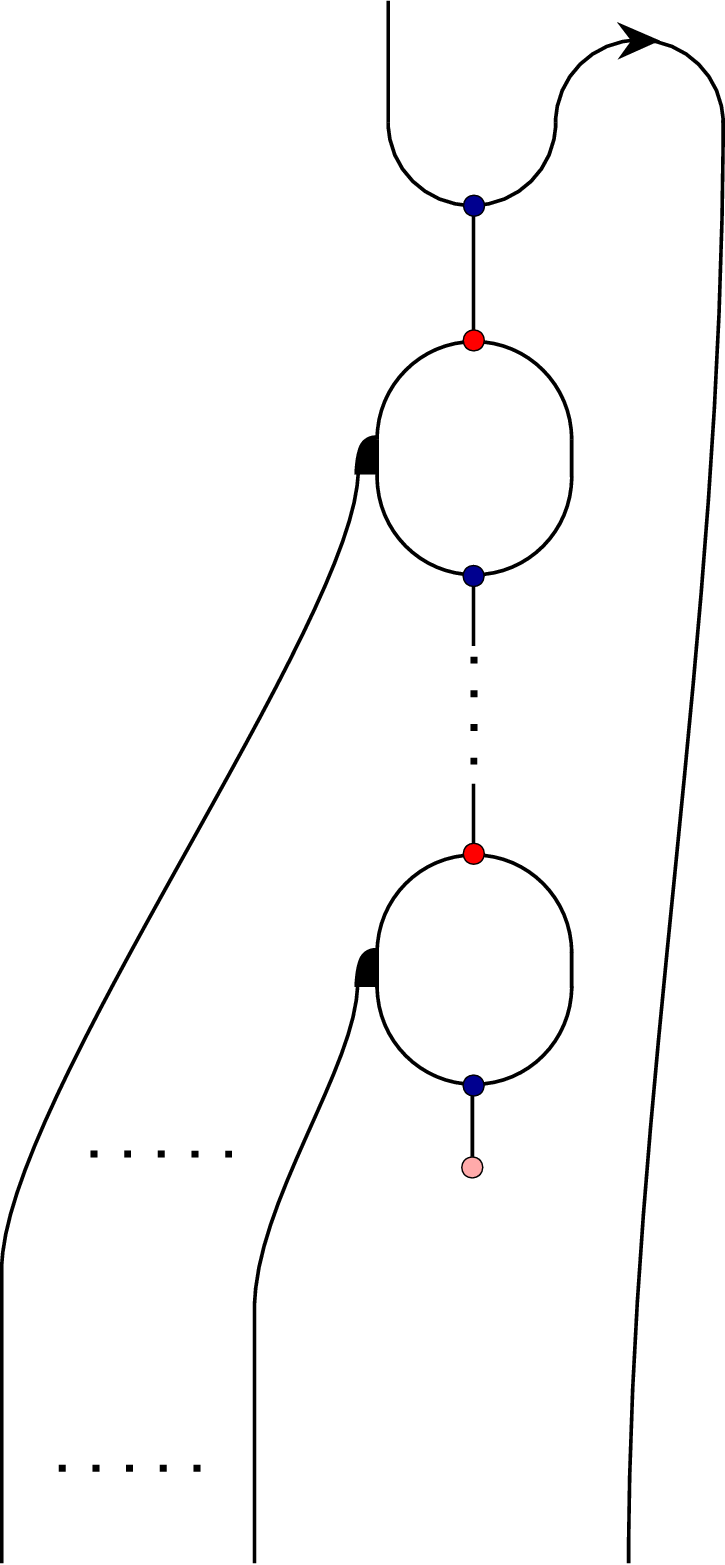}}
    \put(8,-8.5)   {\sse$\HKH$}
    \put(25,-8.5)   {\sse$\dots$}
    \put(36,-8.5)   {\sse$\HKH$}
    \put(77,-8.5)   {\sse$^\vee\!\BFH$}
    \put(53,174)   {\sse$\BFH$}
    \put(100,169) {\Hbimodpic}
    }
    }
\end{lemma}
\proof
The left hand side of \eqref{removeQq_t} is expressed as a linear map by using the expression \eqref{CorrgnH} for $\SkH k01$, the unit $\eta\ASF\eq\eta_H^\vee$ and coproduct of $\BFH$, given in \eqref{pic-Hb-Frobalgebra}, and the expression \eqref{Qq_X} for $\Qq_{\HKH,\HKH^{\otii k-1}\oti^\vee\!\BFH}$. It then follows, by applying  Lemma \ref{lemma_helpid} (ii), that
\Eqpic{inv_tjk_1} {310} {115} {
    \put(0,240)  {$(\id_\BF\oti\tilde b_\BF)\circ(\SkH k02\oti\id_{^\vee\!\BF})\circ\Qq_{\HK,\HK^{\otii k-1}\oti^\vee\!\BF}=~$}
      \put(120,10) {\includepichtft{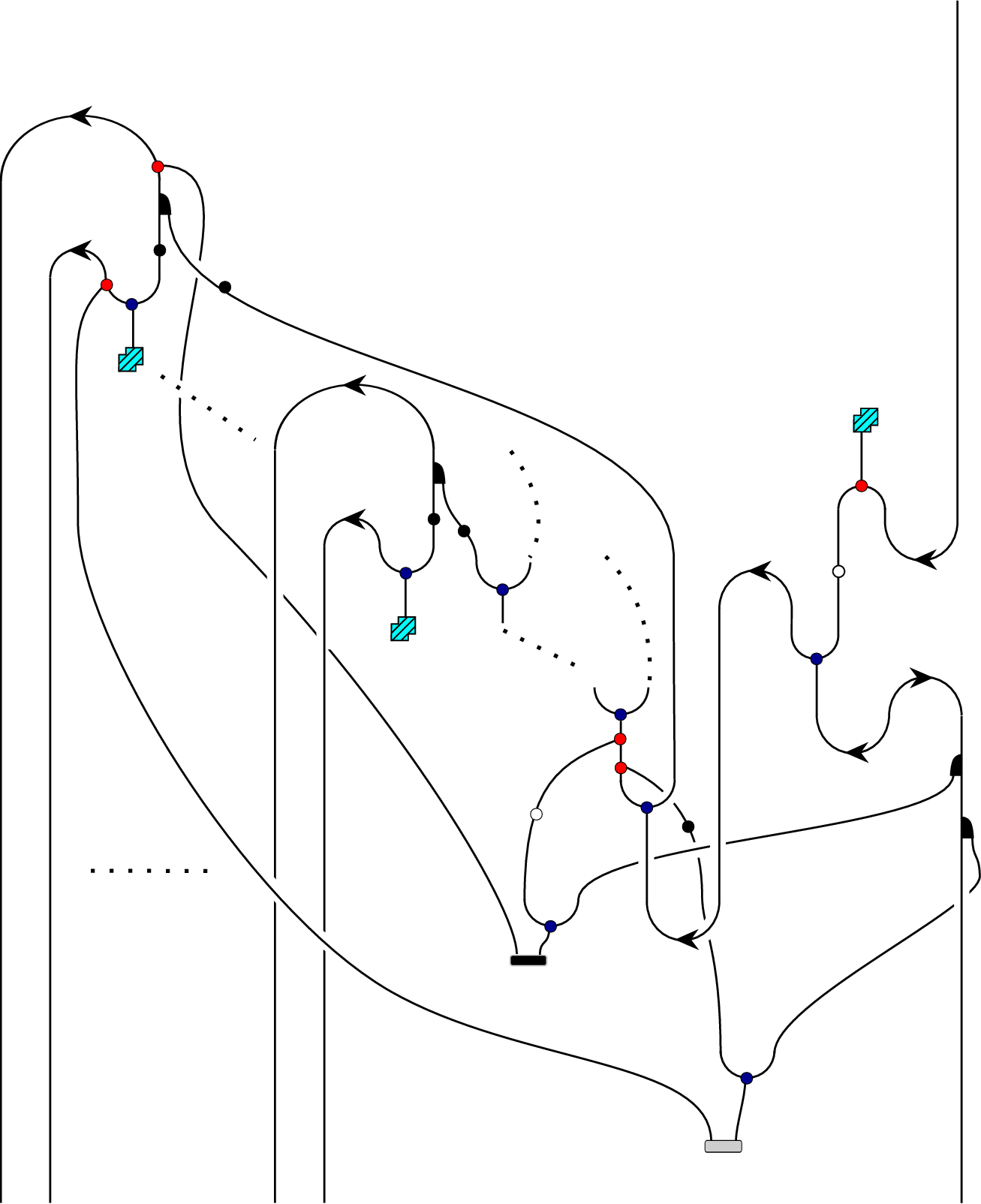}\setulen80
 \put(-7,-11.2)   {\sse$ \Hss $}
  \put(8.2,-11.2)  {\sse$ \Hss $}
  \put(60.5,-11.2) {\sse$ \Hss $}
  \put(74.5,-11.2) {\sse$ \Hss $}
  \put(28.2,-11.2) {\sse$ \cdots\cdots$}
  \put(229.5,-11.2){\sse$ H $}
  \put(229,297)    {\sse$ \Hss $}
  \put(30,195.5)   {\sse$ \Lambda$}
  \put(100,132)    {\sse$ \Lambda$}
  \put(216,188)    {\sse$ \lambda$}
  \put(173,3.5)    {\sse$ Q$}
  \put(120,49,2)   {\sse$ Q^{-1}$}
  }
  }
Next, inserting the definition \eqref{adj_act} of $\rad$, using \eqref{Hopf_Frob_trick2}, then associativity of $m_H$ together with the properties of $\apo$, and finally \eqref{Fv_act} we obtain
\eqpic{inv_tjk_2} {340} {130} {
    \put(0,255)  {$(\id_\BF\oti\tilde b_\BF)\circ(\SkH k02\oti\id_{^\vee\!\BF})\circ\Qq_{\HK,\HK^{\otii k-1}\oti^\vee\!\BF}=~$}
      \put(120,0) {\includepichtft{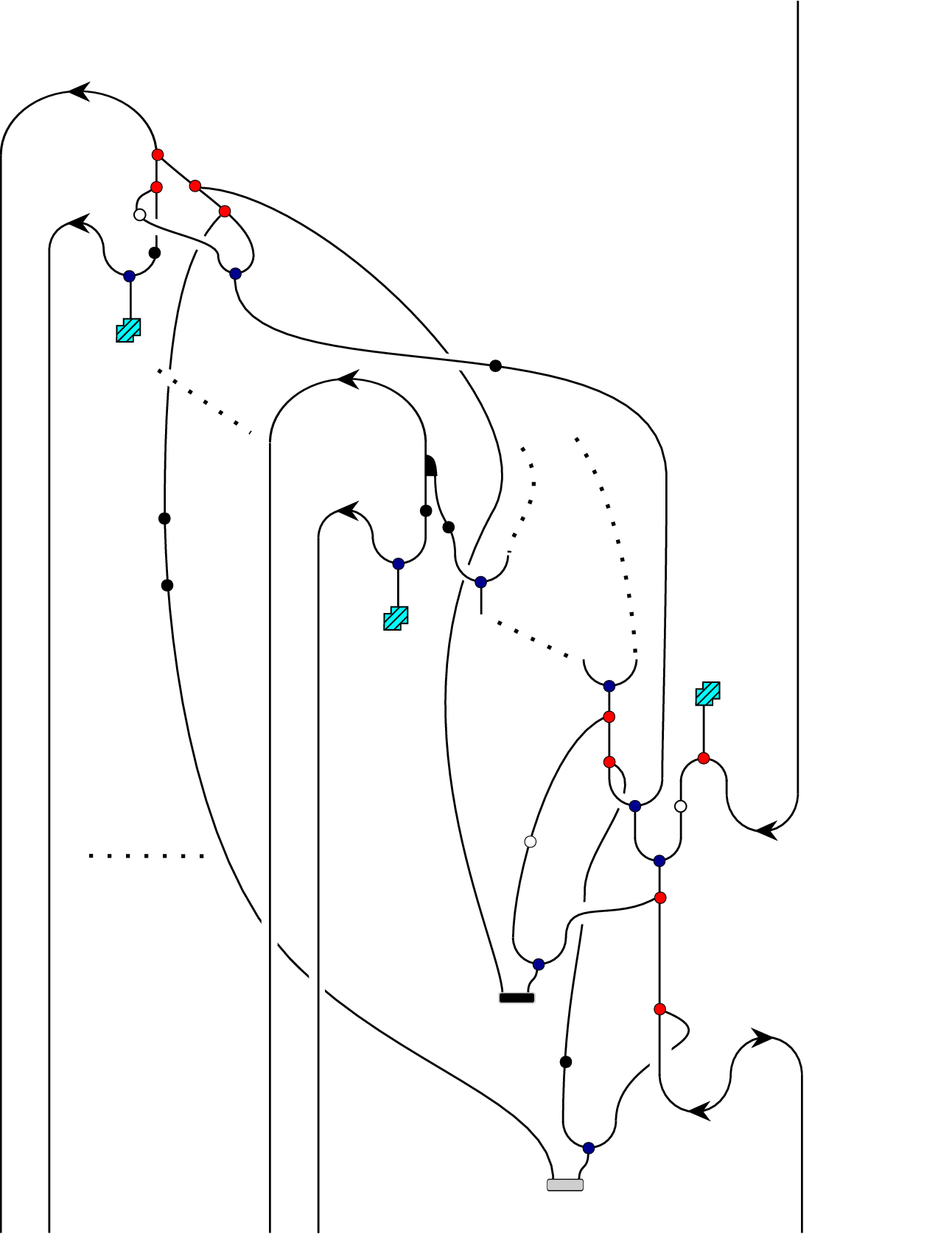}\setulen80
  \put(-5,-8.5) {\sse$ \Hs $}
  \put(8.2,-8.5){\sse$ \Hs $}
  \put(62.5,-8.5){\sse$ \Hs $}
  \put(74.5,-8.5){\sse$ \Hs $}
  \put(94.5,-8.5){\sse$ \dots\dots$}
  \put(195,-8.5){\sse$ H$}
  \put(196,310){\sse$ \Hs $}
  \put(30,213)  {\sse$ \Lambda$}
  \put(96,142.5)  {\sse$ \Lambda$}
  \put(181,131)  {\sse$ \Lambda$}
  \put(137,1)  {\sse$ Q$}
  \put(118,47)  {\sse$ Q^{-1}$}
  }
  }
Now, by Lemma \ref{lemma:cop_Delta} (ii), the $Q$-matrices can be omitted. The so-obtained linear map is $(\id_\BF\oti\tilde b_\BF)\circ(\SkH k02\oti\id_{^\vee\!\BF})$.
\endofproof

By composing \eqref{removeQq_t} with $\id_{\HKH^{\otii k}}\oti^\vee\!\eps\ASF$ and using \eqref{Qq_com} we obtain
\begin{cor}\label{lemma:inv_Qq}
For any integer $k\geq2$ we have
\be\label{inv_Qq_Sk}
    \SkH k01\circ\Qq_{\HKH,\HKH^{\otii k-1}}=\SkH k01\,.
\ee
\end{cor}

\begin{lemma}\label{lemma:twist_S_remove}
For any integer $k>0$, we have the following equality
\Eqpic{twist_S_remove} {320} {110} {\setulen80
    \put(0,0) {\includepichtft{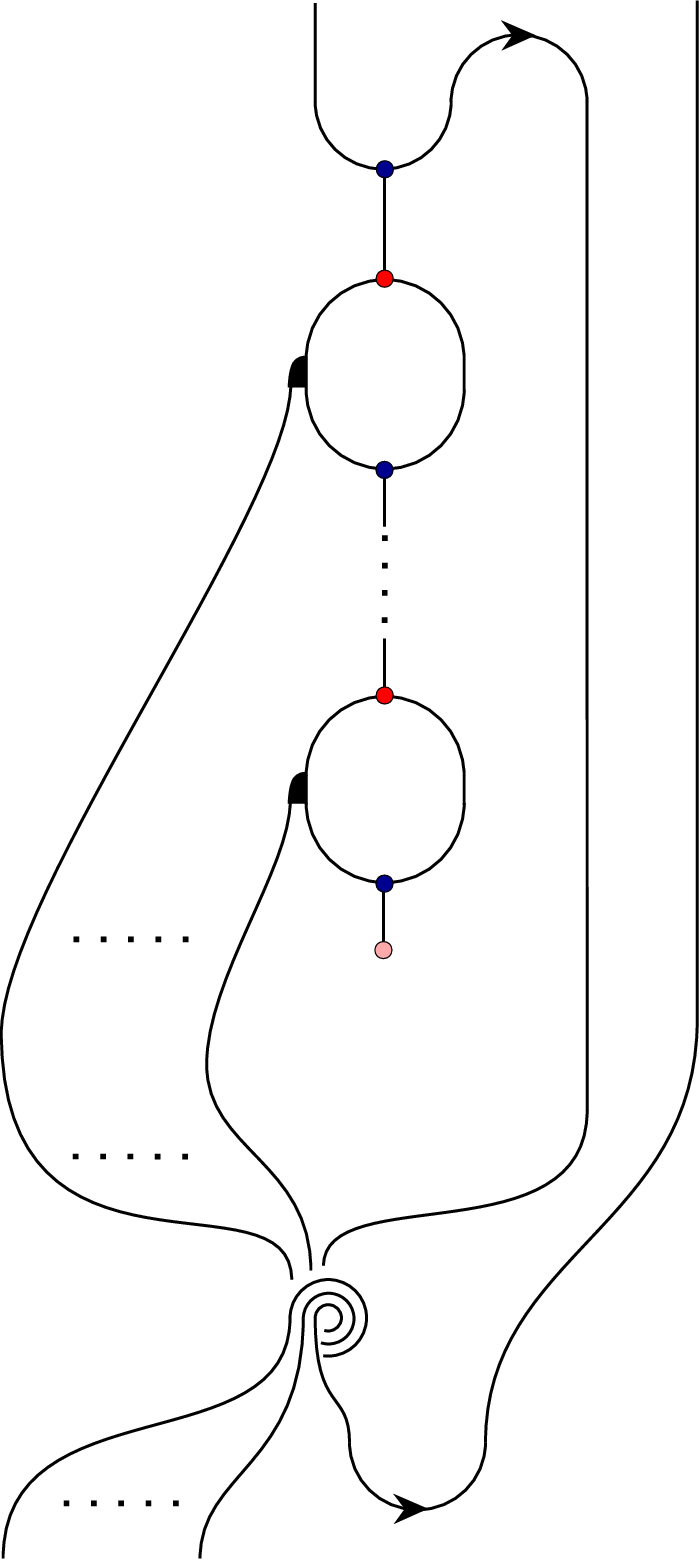}
    \put(-4,-8.5)   {\sse$\HKH$}
    \put(13,-8.5)   {\sse$\dots$}
    \put(31,-8.5)   {\sse$\HKH$}
    \put(50,265)   {\sse$\BFH$}
    \put(115,265)   {\sse$\BFH$}
    \put(1.5,39)     {\sse$ \theta_{\!\HKH^{\otimes k}_{}\otimes{}^{\vee\!}\!\BFH} $}
    }
    \put(150,130)   {$=$}
    \put(155,0) {\includepichtft{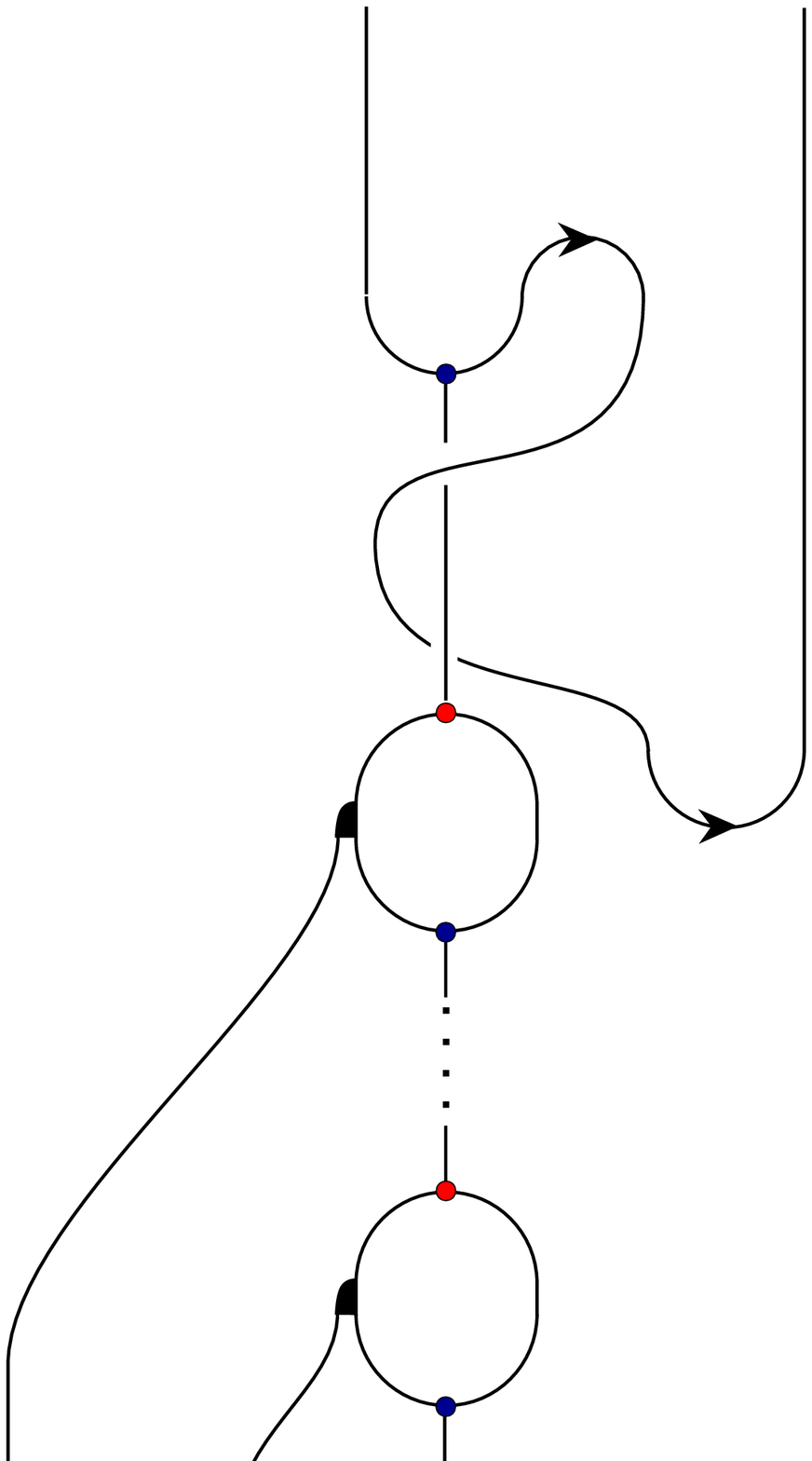}
    \put(-4,-8.5)   {\sse$\HKH$}
    \put(13,-8.5)   {\sse$\dots$}
    \put(31,-8.5)   {\sse$\HKH$}
    \put(50,265)   {\sse$\BFH$}
    \put(115,265)   {\sse$\BFH$}
    }
    \put(300,130)   {$=$}
    \put(300,0) {\includepichtft{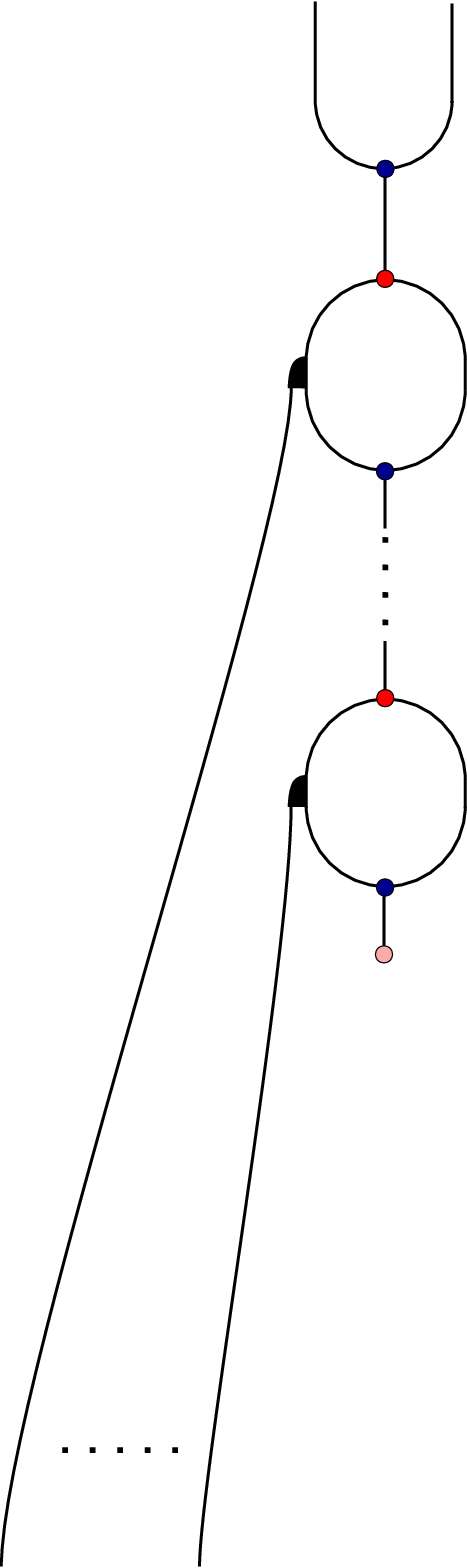}
    \put(-4,-8.5)   {\sse$\HKH$}
    \put(13,-8.5)   {\sse$\dots$}
    \put(31,-8.5)   {\sse$\HKH$}
    \put(50,265)   {\sse$\BFH$}
    \put(73,265)   {\sse$\BFH$}
    \put(115,230) {\begin{rotate}{90}\Hbimodpic\end{rotate}}
    }
    }
of morphisms in $\Hom_{H|H}(\HKH^{\otii p},\BFH\oti\BFH)$.
\end{lemma}
\proof
Using that $\BFH$ has trivial twist and Lemma \ref{inv_twist} , it follows from the compatibility between braiding and twist (see \eqref{tens_braid}) that
the twist morphism on the left hand side can be replaced by a monodromy $c_{^\vee\!\BFH,\HKH^{\otii k}}\cir c_{\HKH^{\otii k},^\vee\!\BFH}$ between $\HKH^{\otii k}$ and $^\vee\!\BFH$. The so-obtained monodromy can be commuted through all the $\BFH$-loops. This results in the first equality. The second equality is obtained by using that $\BFH$ is commutative Frobenius, by a calculation analogous tot the one in the proof of Lemma \ref{lemma:comFA_mon}.
\endofproof
\begin{lemma}\label{lemma:tjk_help}
For any integers $g,n\geq0$ we have
    \eqpic{t_act_Cor} {270} {175} {
  \put(210,177)     {$ = ~\HCorrgn $}
  \put(0,15) {\includepichtft{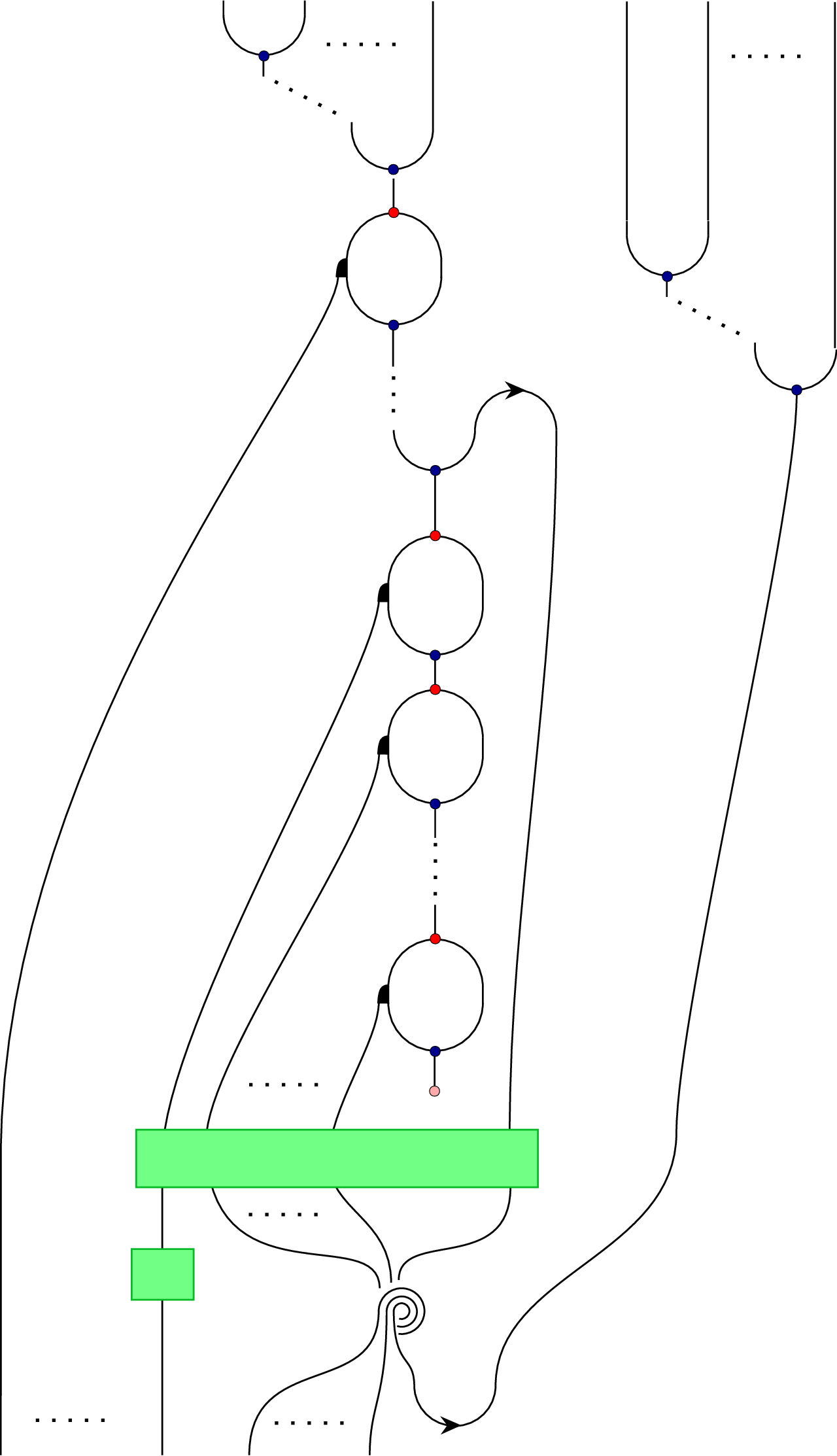}
  \put(49.3,339)   {\sse$ \overbrace{\hspace*{6em}}^{j~ \rm factors} $}
  \put(139.3,339)  {\sse$ \overbrace{\hspace*{6em}}^{n-j~ \rm factors} $}
  \put(46,64.7)    {\sse$\Qq_{\HKH,\HKH^{\otii k-1}\oti^\vee\!\BFH}$}
  \put(95,30)      {\sse\begin{turn}{40} $ \theta_{\!\HKH^{\otimes m-1}_{}
                   \otimes{}^{\vee\!}\!\BFH} $\end{turn}}
  \put(31,38)      {\sse$\TK$}
  \put(-4,-8.5)    {\sse$ \HKH $}
  \put(13,-8.5)    {\sse$\cdots$}
  \put(31,-8.5)    {\sse$ \HKH $}
  \put(51,-8.5)    {\sse$ \HKH $}
  \put(62,-8.5)    {\sse$\cdots$}
  \put(77,-8.5)    {\sse$ \HKH $}
  \put(48,329)     {\sse$ \BFH $}
  \put(66,329)     {\sse$ \BFH $}
  \put(78,329)     {\sse$\cdots$}
  \put(95,329)     {\sse$ \BFH $}
  \put(138,329)    {\sse$ \BFH $}
  \put(156,329)    {\sse$ \BFH $}
  \put(168,329)    {\sse$\cdots$}
  \put(185,329)    {\sse$ \BFH $}
  \put(215,307)    {\Hbimodpic}
  \put(-9.7,-12)   {\sse$ \underbrace{\hspace*{5em}}_{g-m+1~\rm factors~} $}
  \put(45,-12)     {\sse$ \underbrace{\hspace*{3.7em}}_{~~m-1~ \rm factors} $}
  } }
for all $m=1,2,...,g$.
\end{lemma}
\begin{proof}
We first note that by Lemma \ref{inv_Qq_S}, the endomorphism $\Qq_{\HKH,\HKH^{\otii k-1}\oti^\vee\!\BFH}$ can be omitted.
Having done that, it follows from Lemma \ref{inv_STetci} that $\TKH$ acts trivially. Finally, by Lemma \ref{lemma:twist_S_remove}, combined with the Frobenius property, the twist can be omitted as well. By duality, the Frobenius property and coassociativity of $\Delta_H$, the so-obtained morphism equals $\HCorrgn$.
\endofproof\end{proof}
\subject{Proof of Theorem \ref{thm: MPG_Hopf}}
We are now in a position to prove Theorem \ref{thm: MPG_Hopf}.

\begin{proof}
We first verify that the morphism $\HCorrgn$ in $H$\Bimod\ is invariant under the action $\piK$ for all generators of the presentation, given in section \ref{sec:Lyub_rep}, of \Mapgn\ for any $g,n\geq0$.
That is, we prove that the actions defined in \eqref{piLpre}, \eqref{piLpost} and \eqref{piLtjk} acts trivially on $\HCorrgn$:
\nxl1
(i) $\bk$, $\dk$ and $\sk$ for $k=1,...,g$: Invariance follows from the expressions \eqref{LyubactC_pre} for $\LAbk$, $\LAdk$  and $\LASk$ and straightforward application of Lemma \ref{inv_STetci}.\nxl1
(ii) $\wi$ for $i=1,...,n-1$: Invariance follows since $F$ is cocommutative.\nxl1
(iii) $\Ri$ for $i=1,,,n$: Invariance follows since $F$ has trivial twist.\nxl1
(iv) $\ak$, for $k=2,...,g$: It follows from Lemma \ref{inv_STetci} that
\be
    \HCorrgn\circ[\id_{\HKH^{\otii g-k}}\oti(\TKH\oti\TKH)\oti \id_{\HKH^{\otii k-2}}]=\HCorrgn\,.
\ee
Together with Lemma \ref{lemma:inv_QQ} this proves invariance.\nxl1
(v): $\ek$ for $k=2,...,g$: Invariance amounts to the equality
\be\label{inv_ek_Sk}
    \SkH k01\circ[(\TKH\oti\theta_{\HKH^{\otii k-1}})\circ\Qq_{\HKH,\HKH^{\otii k-1}}]=\SkH k01\,,
\ee
for $k=2,...,g$.
It follows from Lemma \ref{inv_STetci} and \ref{inv_twist} that
$\SkH k01\circ(\TKH\oti\theta_{\HK^{\otii k-1}})=\SkH k01$. Together with Corollary \ref{lemma:inv_Qq} this establishes \eqref{inv_ek_Sk}. \nxl1
(vi): $\tjk$, for $j\eq1,...,n-1$ and $k\eq1,...,g$:
First, recall from \eqref{Qq_com} that any $f\in\Hom(X,Y)$ can be commuted through the partial monodromy $\Qq_{K,Y}$.
Second, inserting the definition of \HCorrgn\ in the action  \eqref{t_act} of $\tjk$ and using first coassociativity and the Frobenius property of $\BFH$ and then \eqref{Qq_com} together with functionality of the twist it follows that the $n{-}j{-}1$-fold coproduct of $\BFH$ can be pushed through the partial monodromy $\Qq_{\HKH,\HKH^{\otii k-1}\oti(^\vee\!\BFH)^{\otii(n-j)}}$ and the twist. The so-obtained morphism coincides with the left hand side of \eqref{t_act_Cor}. Thus, invariance under $\tjk$ follows from Lemma \ref{lemma:tjk_help}.
\nxl1
This proves that the morphism $\HCorrgn$ in $H$\Bimod\ is invariant under the action $\piK$ of \Mapgn\ for any $g,n\geq0$. Finally, combining this result with the equivalences in section \ref{sec:equiv_cat} completes the proof of Theorem \ref{thm: MPG_Hopf}.
\endofproof\end{proof}
\begin{remark} (i)
When discussing factorization in CFT, one has to distinguish between \emph{ingoing} and \emph{outgoing} insertions.
This means that there is a partition of the holes of the into $p$ outgoing and $q$ ingoing holes for any integers $p$ and $q$ such that $p+q\eq n$.
The correlator has to be invariant under an action of the subgroup $\Mapio g{p,q}$ of the mapping class group $\Mapio g{p+q}$ that leaves each of these two subsets of ingoing and outgoing insertions separately invariant, c.f.\ Remark \ref{rem:repio}.
In that setting, a natural candidate for a correlator is the morphism $\Sk gqp\In\Hom_{H|H}(\HKH^{\otii g}\oti\BFH^{\otii q},\BFH^{\otii p})$, defined in \eqref{Sk_morph}. Combining the projective action $\PiLio pq$ of $\Mapio g{p,q}$, constructed in Remark \ref{rem:repio}, with the isomorphism $\Phi\In\Hom(\BFH^\vee,\BFH)$, that comes with the Frobenius structure of $\BFH$, we obtain a projective action $\PIKio pq$ of $\Mapio g{p,q}$ on $\Hom_{H|H}(\HKH^{\otii g}\oti\BFH^{\otii q},\BFH^{\otii p})$. It is straightforward to verify, using that $\BFH$ is a symmetric Frobenius algebra, that invariance under this action follows from Theorem \ref{thm: MPG_Hopf}.
\nxl2
(ii) As explained above, the morphism $\Sk 110$ is invariant under $\PIKione 10(\sk)$. In the semisimple case, this is equivalent to $\BFH$ being modular invariant in the sense of \cite[Definition 3.1 (ii)]{koRu2}.
\end{remark}
\subsection{Holomorphic factorization}
Recall that holomorphic factorization implies that the partition function of a rational CFT is a bilinear combination, with integer coefficients $Z_{i,j}$,  of characters of the vertex algebra \V\ of the theory, see c.f. \eqref{PF_mat}.
A similar result is established for $\HCorr 10$ in \cite{fuSs4}.

Consider $\Surf 10$, i.e.\ the torus without holes. The morphism $\HCorr 10$ serves as a candidate for a partition function.
Using that the \bulkAC\ \BFH\ is symmetric Frobenius it follows that $\HCorr 10$ is nothing but the $\HKH$-character, as defined in Definition \ref{def:char}, of \BFH.
Via the equivalences described in section \ref{sec:equiv_cat}, the \bulkAC\ \BFH\ can be considered as a module over $\HK$ in $H\otimes H\op$\Mod, with the action of $H\oti H\op$ obtained by translating the bimodule structure of \BFH\ via the equivalence in \eqref{HHbimiso}.

Denote by $\LH$ the categorical Hopf algebra \L, described in section \ref{coends},   in \HMod. Recall that even when  $H$\Mod\ is not semisimple, it still has a finite set number of isomorphism classes of simple objects. We choose a set $\{M_j|j\In\IJ\}$ of representatives of the isomorphism classes of simple objects.
Working in $H\oti H\op$\Mod, we establish the following Theorem \cite[Theorem 3]{fuSs4}
\begin{theorem}\label{thm_HF_Hopf}
The partition function $Z=\HCorr 10$ can be chirally decomposed as
  \be\label{HF_Hmod}
  Z = \sum_{i,j\in\IJ} c_{\overline i,j}^{}\, \chi^\LH_i \otimes \chii^\LH_j \,,
  \ee
where for each $i\in\IJ$,
$\chii^\LH_i$ is the character of the \LH-module given by the simple $H$-module $M_i$, with the \LH-action defined in \eqref{Lact},
$\overline i \,{\in}\, \IJ$ is the label of the isomorphism class of $M_i^\vee$,
and $\,C \,{=}\, \big(c_{i,j}\big){}_{i,j\in\IJ}^{}$ is the Cartan matrix of \HMod.
\end{theorem}
In particular (recall the definition \eqref{Cartan_def} of the Cartan matrix) the coefficients in \eqref{HF_Hmod} are integers
and in the semisimple case $c_{i,j}=\delta_{i,j}$.
Furthermore, as shown in \cite[Lemma 6]{fuSs4}, the character $\chi^\LH_i$ is, as a linear map, given by
\be
    \chi^\LH_i=\chi^H_{M_i}\circ m\circ(f_Q\oti t),
\ee
where $t$ is the special group-like element, defined in \eqref{t_def}, and $f_Q$ the Drinfeld map. Recall that $f_Q$ and $t$ are invertible. Thus $\chi^\LH_i$, and a consequently also the partition function $Z$, is non-zero.
\section{\BulkA s from ribbon automorphisms}\label{sec:twist_bulk}

In this section we generalize the construction described in the previous section by \emph{twisting} $\BF$ by an automorphism of $H$.
\begin{definition}
(i) A \emph{Hopf algebra automorphism} of a Hopf algebra $H$ is a linear map $\ra$ from $H$ to $H$ that is both an algebra and a coalgebra automorphism and that commutes with the antipode. \\
(ii) A \emph{ribbon\index{ribbon!automorphism}} automorphism of a Hopf algebra $H$ is a Hopf algebra automorphism that preserves the ribbon element and the $R$-matrix of $H$:
\be\label{ribbon_Hopf}
    \ra\circ v=v \quand (\ra\oti\ra)\circ R=R\,.
\ee
\end{definition}

For any two Hopf algebra automorphisms $\ra$ and $\ra'$, the $\ra$-{\ra'}-twisted bimodule $\twistmod X{\ra}{\ra'}$ is obtained from $X\equiv(X,\rho_X,\ohr_X)$ by twisting the actions according to
\be
    \twistmod X{\ra}{\ra'}=(X,\rho_X\circ(\ra\oti\id_X),\ohr_X\circ(\id_X\oti\ra'))\,.
\ee
The twisting is compatible with the monoidal structure of $H$\Bimod, and if $\ra$ and $\ra'$ are in addition ribbon automorphisms the twisting is also compatible with the ribbon structure.
It also follows immediately that for any $H$-bimodule morphism $f\mapdef X\rightarrow Y$, $f$  also intertwines the $\ra$-$\ra'$-twisted actions.

\subsection{Twisting the \bulkAC}
Consider now the twisted \bulkAC\
\be\label{BFtwist}
    \BFw:=\twistmod \BFH{\,\id_H\!}{\,\ra}\,.
\ee
The $\ra$-twisted actions of \BFw\ differs from the untwisted ones only by pre-compo\-sing the right action by $\ra$. Since any $H$-bimodule morphism also intertwines the $\ra$-twisted right action, \BFw\ is a symmetric commutative Frobenius algebra in $H$\Bimod\ with the structure morphisms given by the same linear maps \eqref{pic-Hb-Frobalgebra} as $\BFH$, see \cite[Proposition 6.1]{fuSs3}. In addition, just like in the untwisted case, \BFw\ is special iff $H$ is semisimple. It is also worth mentioning that for any two ribbon automorphisms $\ra$ and $\ra'$, there is an isomorphism
\be
    \twistmod \BFH{\,\ra'}{\,\ra}~\cong~\twistmod \BFH{\,\id_H\,}{\,\ra'^{-1}\cir\ra}~=~\BFww\,,
\ee
furnished by $\ra'^*$.
Thus \eqref{BFtwist} is the most general twisting of $\BFw$.

\ssubject{\BFw\ as a coend} Consider the functor
\be
\Fbxo \mapdef \HMod\op \Times \HMod \To \HBimod
\ee
acting on morphisms as $f\Times g\mapsto f^\vee\otik g$, and on objects as
\be
   \big( X^\vee\otik Y\,,\, [\rho_{X^\vee_{\phantom:}}^{}
   \cir (\omega^{-1}\oti\id_{X^*})] \oti\id_Y\,,\,
   \id_{X^*} \oti(\rho_Y\cir \tau_{Y,H}^{} \cir (\id_Y\oti\apo^{-1})) \big)\,.
\ee
That is, $\Fbxo$ acts on objects by composing the action of the functor $\Fbx$, defined in \eqref{Fbx_act}, with a twisting by $\ra^{-1}$ of the left action.
It is straightforward to check \cite[Proposition 6.2]{fuSs3} that:
\begin{prop}
The $H$-bimodule \BFw\ together with the dinatural family of morphisms
   \be
   \iHbo_X := (\omega^{-1})^* \circ \iHb_X \,,
   \ee
with $\iHb_X$ as defined in \erf{def_iHb}, is the coend of the functor $\Fbxo$.
\end{prop}

In order to proceed we need to know how a ribbon automorphism acts on the cointegral. To this end, we first note:
\begin{lemma}
    For any factorizable ribbon Hopf algebra $H$, the following equality holds:
    \be\label{v-v-inv}
     f_{Q^{-1}} \big( \lambda \cir m \cir (v \oti \id_H) \big)
  = (\lambda \cir v)\, v^{-1}\,.
    \ee
\end{lemma}
\begin{proof}
    The equality is obtained by using the compatibility \eqref{prop_v} between the ribbon element and the monodromy matrix, and then the defining property \eqref{prop_lambda} of the cointegral.
\endofproof\end{proof}
As a consequence we have:
\begin{lemma}\label{lambda_pres}
    Any ribbon automorphism $\ra$ of a finite-dimensional factorizable ribbon Hopf algebra $H$ preserves the integral and cointegral:
    \be\label{ra_pres_lambda}
    \lambda \cir \ra = \lambda  \quad\qquand   \ra \cir \Lambda = \Lambda \,.
    \ee
\end{lemma}
\begin{proof}
    Compose the equality \eqref{v-v-inv} with $\ra^{-1}$. Using that $\ra$ preserves $v$, the right hand side remains unchanged. On the left hand side we can use that $\ra$ also preserves $Q^{-1}$ to obtain an expression that differs from the left hand side of \eqref{v-v-inv} only in that $\lambda$ is replaced by $\lambda\cir\ra$. Thus:
    \be
        f_{Q^{-1}} \big( (\lambda\cir\ra) \cir m \cir (v \oti \id_H) \big)=f_{Q^{-1}} \big( \lambda \cir m \cir (v \oti \id_H) \big)\,.
    \ee
    Since $f_{Q^{-1}}$ and $v$ are invertible, this implies the first equality in \eqref{ra_pres_lambda}.\\
    The second equality in \eqref{ra_pres_lambda} is obtained from the first by using that $\omega\cir\Lambda$ is a non-zero integral (and thus proportional to $\Lambda$) and that $\lambda\cir\Lambda\In\k$ is non-zero.
\endofproof\end{proof}
\subsection{Mapping class group invariants from \BFw}
By replacing \BFH\ by \BFw\ in the prescription \eqref{Sk_morph} we define a morphism
\be
    \Skw gpq\linebreak\In\Hom(\HKH^{\otii g}\oti\BFw^{\otii p},\BFw^{\otii q})\,,
\ee
for any ribbon automorphism $\ra$. By composing with the unit of \BFw\ we obtain a morphism
\be
    \Corrgnw:=\Skw g1n\circ(\id_{\HKH^{\otii g}}\oti\eta_{\BFw})\In\Hom(\HKH^{\otii g},\BFw^{\otii n})
\ee
 for any $g,n\geq0$.

In order to proceed we need to express $\Skw g11$ as a linear map. To this end, first note that, since $\ra$ commutes with the antipode and $(\ra\oti\ra)\cir Q=Q$, the automorphism $\ra$ can be pushed through the $Q$-matrix:
\be\label{w_com_Q}
    (\id_H\oti(\ra\circ\apoi))\circ Q=(\ra^{-1}\oti\apoi)\circ Q\,.
\ee
Inserting the $\ra$-twisted right action in the first expression in \eqref{PF_1} and using \eqref{w_com_Q}, it follows that $\Skw 101$ is related to the untwisted version by
\be
    \Skw 101=\SkH 101\circ(\id_{\Hs}\oti{(\ra^{-1})}^*)\,.
\ee
As a consequence we have the following twisted version of \eqref{CorrgnH}
\Eqpic{CorrgnHw} {270} {110} {\setulen90
    \put(0,120)  {$\Skw g11=~$}
      \put(60,0) {\INcludepichtft{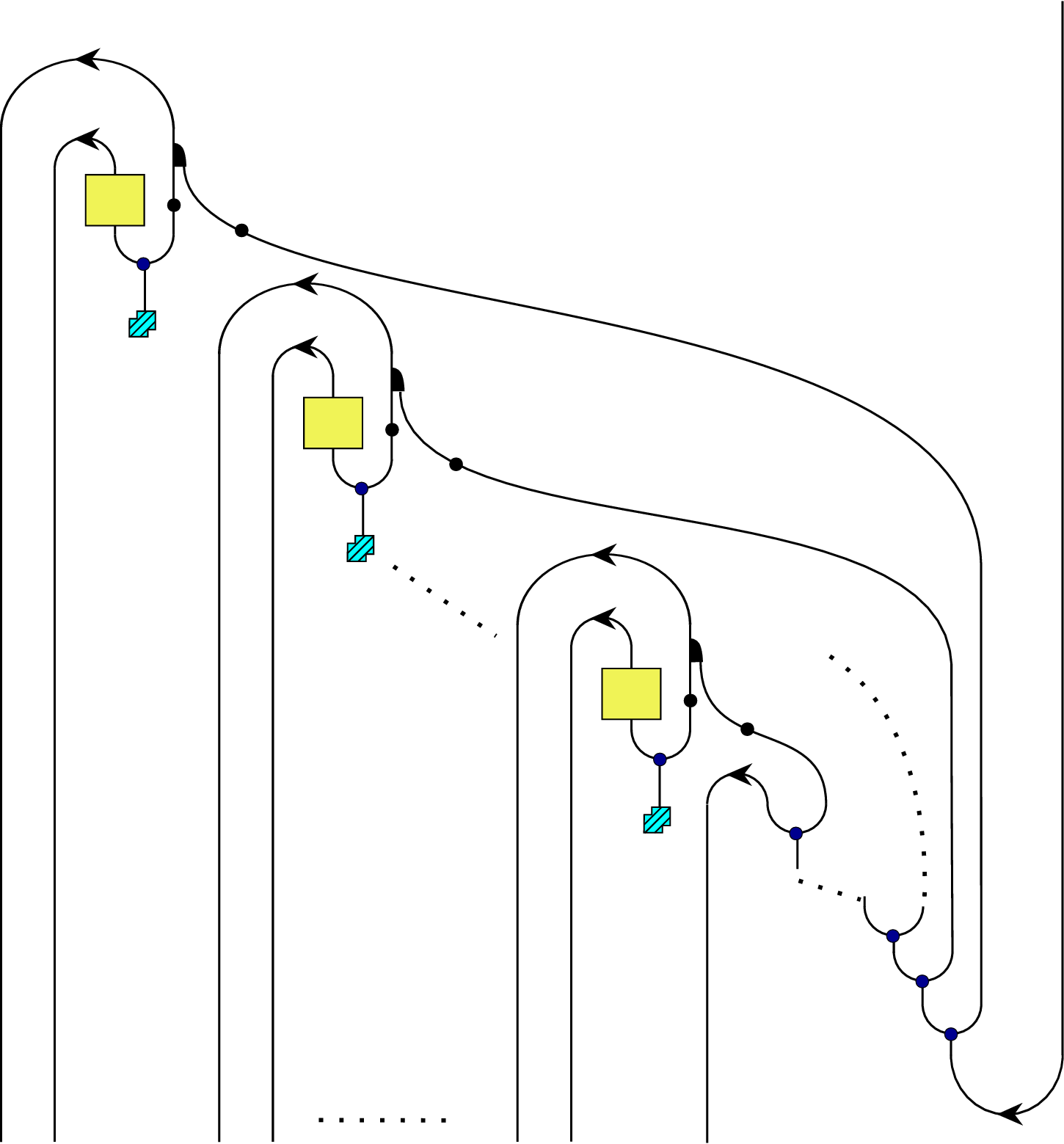}{342}
  \put(-5,-8.5) {\sse$ \Hs $}
  \put(8.2,-8.5){\sse$ \Hs $}
  \put(44.5,-8.5){\sse$ \Hs $}
  \put(56.5,-8.5){\sse$ \Hs $}
  \put(76.5,-8.5){\sse$ \dots\dots$}
  \put(111.5,-8.5){\sse$ \Hs $}
  \put(123.5,-8.5){\sse$ \Hs $}
  \put(154.5,-8.5){\sse$ \Hs $}
  \put(234,260){\sse$ \Hs $}
  \put(36,180.5)  {\sse$ \Lambda$}
  \put(70,130.5)  {\sse$ \Lambda$}
  \put(137,69.5)  {\sse$ \Lambda$}
  \put(19.5,208.5)  {\sse$ \ra^{\!\!-1}$}
  \put(68.7,158.5)  {\sse$ \ra^{\!\!-1}$}
  \put(135.3,97.5)  {\sse$ \ra^{\!\!-1}$}
  \put(43,220)     {\sse$ \rad $}
  \put(91.4,169.4) {\sse$ \rad $}
  \put(158.2,108.5){\sse$ \rad $}
  }
  }
In particular it follows that
\be\label{Corr_Corrw}
    \Corrgnw=\Corrgn\circ(\id_{\Hs}\oti{(\ra^{-1})}^*)^{\otii g}\,.
\ee

We have the following generalization of Theorem \ref{thm: MPG_Hopf}:
\begin{thm}
    For any ribbon automorphism $\ra$, the morphism \Corrgnw\ is invariant under the action, described in Proposition \ref{Lyubact_prop}, of the mapping class group \Mapgn\ on $\Hom(\HKH^{\otii g},\BFw^{\otii n})$.
\end{thm}
\begin{proof}
Using the expression \eqref{CorrgnHw} for $\Skw g11$ it is straightforward to establish the invariance under the action of \Mapgn\ by using the results from the untwisted case:\nxl1
(i) That \Corrgnw\ is invariant under the actions of \wi\ and \Ri\ follows, just like in the untwisted case, directly from the fact that \BFw\ is symmetric with trivial twist.\nxl1
(ii) Next consider any endomorphism $Z_\circ\In\End(\HKH^{\otii g})$ appearing in the morphisms $\LAak$, $\LAbk$, $\LAdk$, $\LAek$ and $\LASk$ defined in \eqref{LyubactC_pre}.
We can use that \ra\ is a ribbon automorphism that preserves the cointegral to show that $Z_\circ$ commutes with $(\id_\Hs\oti{(\ra^{-1})}^*)^{\otii g}$.
Consider e.g. the endomorphism $\SKH$. We have the following chain of equalities
\Eqpic{S_KH_w} {320}{42} {
   \put(0,-6) { \Includepichtft{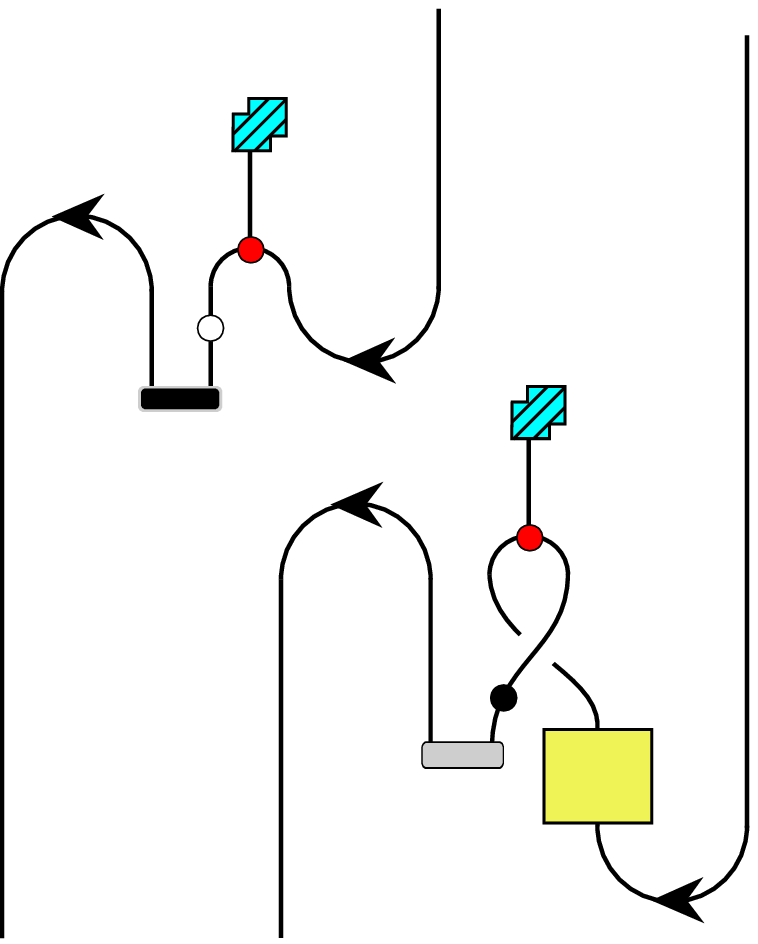}
   \put(-4.5,-8.5) {\sse$ \Hss $}
   \put(13,51.8)   {\sse$ Q^{-1} $}
   \put(26.5,-8.5) {\sse$ \Hss $}
   \put(32.5,87)   {\sse$ \lambda $}
   \put(44.6,106)  {\sse$ \Hss $}
   \put(47.4,12.1) {\sse$ Q $}
   \put(63.5,56)   {\sse$ \lambda $}
   \put(75.6,106)  {\sse$ \Hss $}
   \put(60,15.1) {\sse$ \ra^{\!\!-1}$}
   }
   \put(95,42)  {$=$}
   \put(120,-6) { \Includepichtft{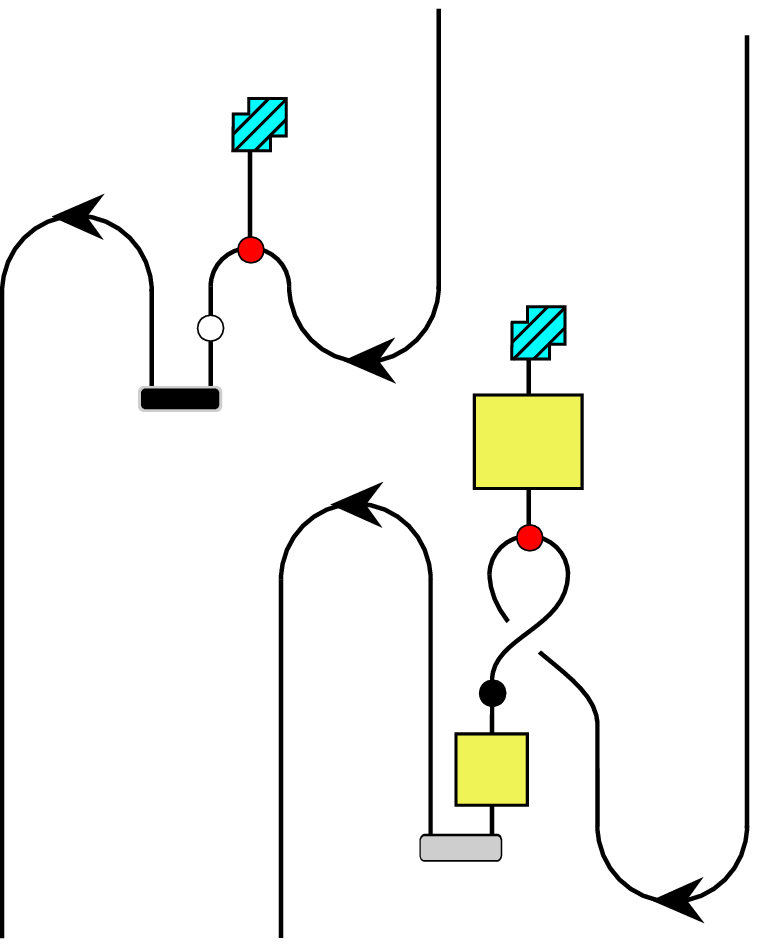}
   \put(-4.5,-8.5) {\sse$ \Hss $}
   \put(13,51.8)   {\sse$ Q^{-1} $}
   \put(26.5,-8.5) {\sse$ \Hss $}
   \put(32.5,87)   {\sse$ \lambda $}
   \put(44.6,106)  {\sse$ \Hss $}
   \put(47.4,2.1) {\sse$ Q $}
   \put(63.5,64)   {\sse$ \lambda $}
   \put(75.6,106)  {\sse$ \Hss $}
   \put(52.3,52.1) {\sse$ \ra^{\!\!-1}$}
   \put(52,18.1) {\sse$ \ra$}
   }
   \put(215,42)  {$=$}
   \put(240,-6) { \Includepichtft{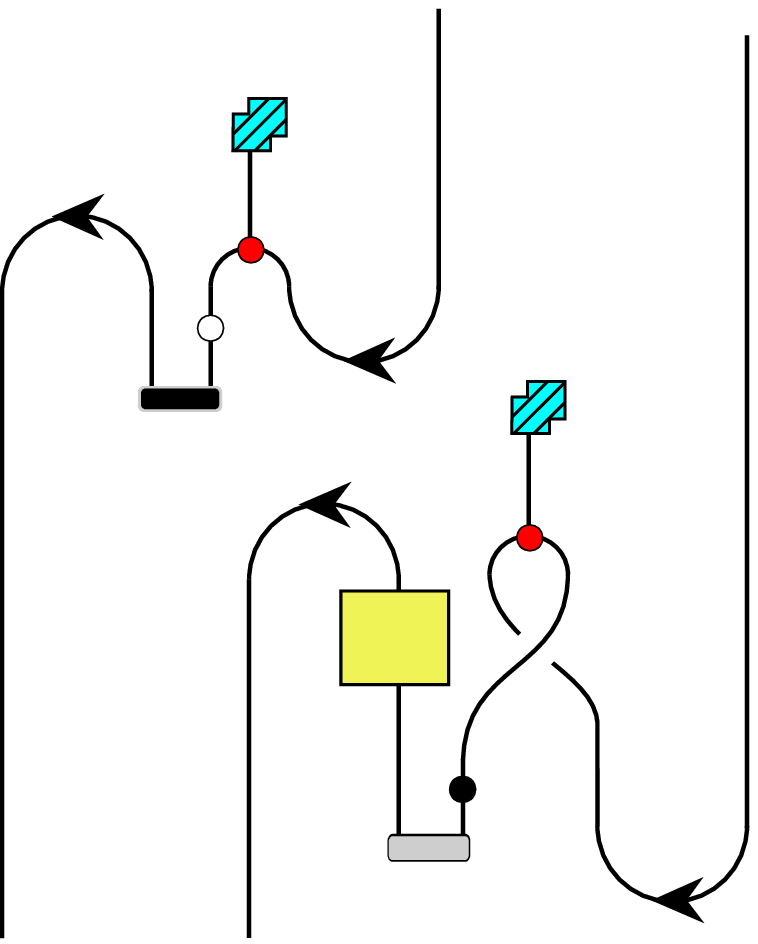}
   \put(-4.5,-8.5) {\sse$ \Hss $}
   \put(13,51.8)   {\sse$ Q^{-1} $}
   \put(26.5,-8.5) {\sse$ \Hss $}
   \put(32.5,87)   {\sse$ \lambda $}
   \put(40.6,106)  {\sse$ \Hss $}
   \put(43.4,2.1) {\sse$ Q $}
   \put(63.5,64)   {\sse$ \lambda $}
   \put(75.6,106)  {\sse$ \Hss $}
   \put(38,30.1) {\sse$ \ra^{\!\!-1}$}
   }
   }
Here the first equality follows by pushing $\ra$ through the product and using $\ra\cir\apoi=\apoi\cir\ra$. The second equality is a consequence of $\lambda\cir\ra=\lambda$  (see Lemma \ref{lambda_pres})  and $(\ra\oti\ra)\cir Q=Q$. From \eqref{S_KH_w} it follows that
\be
    (\id_\Hs\oti{(\ra^{-1})}^*)\cir\SKH=\SKH\cir{(\id_\Hs\oti(\ra^{-1})}^*)\,.
\ee
It is shown in a completely analogous manner that any of the endomorphisms of $\HKH^{\otii g}$, appearing in \eqref{LyubactC_pre}, commutes with $(\id_\Hs\oti{(\ra^{-1})}^*)^{\otii g}$. Invariance of $\Corrgnw$ under the action of the generators $\ak$, $\bk$, $\dk$, $\ek$ and $\sk$ then follows from \eqref{Corr_Corrw} and invariance in the untwisted case.\nxl1

(iii) Finally, invariance of \Corrgnw\ under the action of $\tjk$ is a bit more lengthy but works in manner similar to the other generators.
Due to \eqref{CorrgnHw}, there is a version of Lemma \ref{inv_STetci} with $\BFH$ replaced by $\BFw$. In addition Lemma \ref{lemma:Hpastloopsw} implies that there are $\ra$-twisted versions of Lemma \ref{inv_Qq_S} and Lemma \ref{lemma:twist_S_remove} in which $\BFH$ has been replaced by $\BFw$. The $\ra$-twisted versions of  Lemma \ref{inv_STetci}, \ref{inv_Qq_S} and \ref{lemma:twist_S_remove} imply that there is a version of Lemma \ref{lemma:tjk_help} in which $\BFH$ has been replaced by $\BFw$. Next we observe that in the case $\ra=\id_H$, invariance of $\HCorrgn$ under the action of $\tjk$ follows by using that $\BFH$ is Frobenius and Lemma \ref{lemma:tjk_help}. Thus, using the $\ra$-twisted version of  Lemma \ref{lemma:tjk_help}, invariance of \Corrgnw\ under the action of $\tjk$ follows in precisely the same manner as for $\ra=\id_H$.
\endofproof\end{proof}
\chapter*{Conclusion and outlook}
In this thesis we have discussed on one hand results for rational CFT, and on the other hand results that are motivated by the quest for a model independent description of logarithmic CFT. The extension of the proof that the correlators of the TFT-construction satisfy all consistency conditions of the theory and the derivation of the classifying algebra for defects give additional insight into the already quite well understood structure of rational conformal field theory. The classifying algebras show that boundary conditions and defect conditions can be largely understood without the need to resort to the full TFT-construction. Furthermore, as briefly described, the classifying algebra for boundary conditions can be obtained from a coend, without explicitly referring to semisimplicity. Thus, one might expect that the existence of classifying algebras is not a feature unique for rational CFT, but that similar structures exist also beyond the rational case.
\bigskip

The construction of the bulk Frobenius algebra $\BFH$ and the mapping class group invariant morphisms $\HCorrgn$ in $H$\Bimod\ aims at gaining insight into the so far much less developed class of logarithmic conformal field theories. A natural goal is to generalize the construction to an arbitrary factorizable finite ribbon category $\C$. For a general category $\C$ in this class, the coend $\BF$ is naturally equipped with an algebra structure. In the case $\C\eq H$\Mod, the coalgebra structure of the coend $\BF$, i.e.\ the bulk Frobenius algebra, uses the integral of $H$. What we are missing so far is the corresponding property of the category $H$\Bimod\ that is related to the integral and equips the coend with a coalgebra structure.
\bigskip

A different direction of generalization of the coend $\BF$ is to, inspired by the structure of the bulk state space of a rational conformal field theory, consider the object $C_l((A\btimes\one)\oti\BF)\In\Ce$, for $A$ a special symmetric Frobenius algebra in $\C$. Checking whether this object still gives rise to a symmetric commutative Frobenius algebra and mapping class group invariants would be interesting already in the particular case $\C\eq H$\Mod.

\titleformat{\chapter}
  {\normalfont\LARGE\bfseries}{ }{0pt}{\LARGE}
\titlespacing*{\chapter}
  {0pt}{30pt}{25pt}
\appendix
\chapter{Appendix}
In this appendix we collect some formulas, results and calculations, that are relevant for the main text.
\section{Ribbon categories}\label{app:ribbon}

\subsection{Compatibility conditions for ribbon categories}\label{app:comp_ribbon}
Here we give the full list of compatibility conditions for the structures of a (strict) ribbon category.
The axioms for right the duality and the compatibility between the right duality and the twist looks graphically as follows:
\eqpic{right_d_t} {320} {42} {\setulen80
\put(0,29)   {
    \put(0,0){\includepic{304}{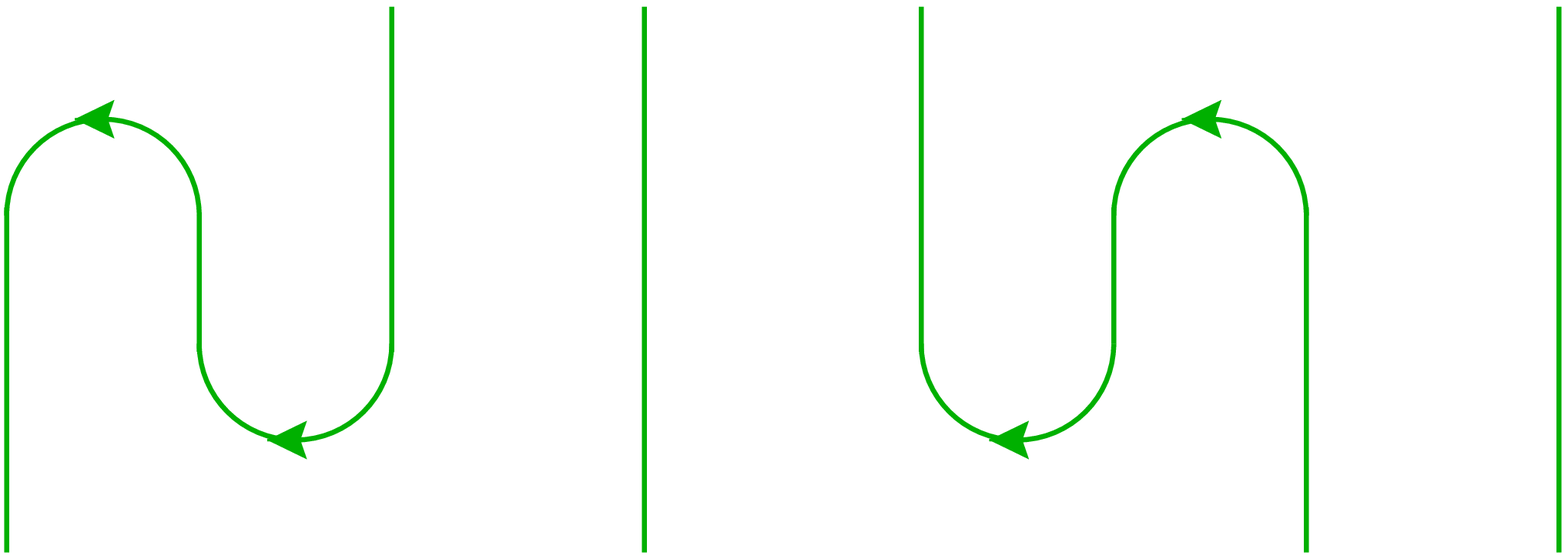}}
 \put(-3.2,-7)  {\scriptsize$U\Vee$}
  \put(51,82.8)  {\scriptsize$U\Vee$}
  \put(69,39)    {$=$}
  \put(88,-7)    {\scriptsize$U\Vee$}
  \put(88,82.8)  {\scriptsize$U\Vee$}
  \put(129,82.8) {\scriptsize$U$}
  \put(182.4,-7) {\scriptsize$U$}
  \put(200,39)   {$=$}
  \put(218.5,-7) {\scriptsize$U$}
  \put(218.5,82.8){\scriptsize$U$}
  \put(58,-34)   {\footnotesize\Fbox{axioms for (right-) duality}}
}
\put(290,35) {
\put(0,10)  {\includepic{304}{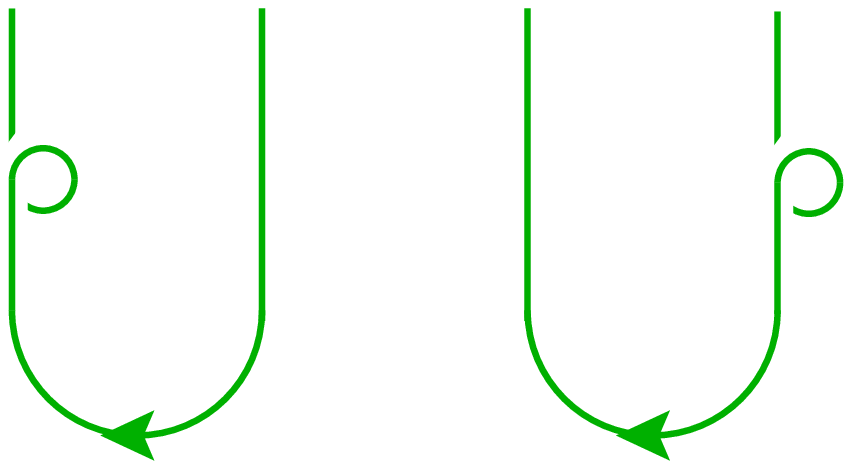}}
\put(.9,62.8)  {\scriptsize$U$}
  \put(26,62.8)   {\scriptsize$U\Vee$}
  \put(41,36)     {$=$}
  \put(57.4,62.8) {\scriptsize$U$}
  \put(82.5,62.8) {\scriptsize$U\Vee$}
  \put(3,-34)    {\footnotesize\Fbox{\begin{tabular}c{$\!\!$duality and twist:$\!\!$}
                                \\$\theta_{U\Vee}=(\theta_U)\Vee$\end{tabular}}}
}
}
The braiding and twist is functorial:
\eqpic{func_twist_braid} {260} {45} {
\put(0,30){
    \put(0,0)   {\Includepic{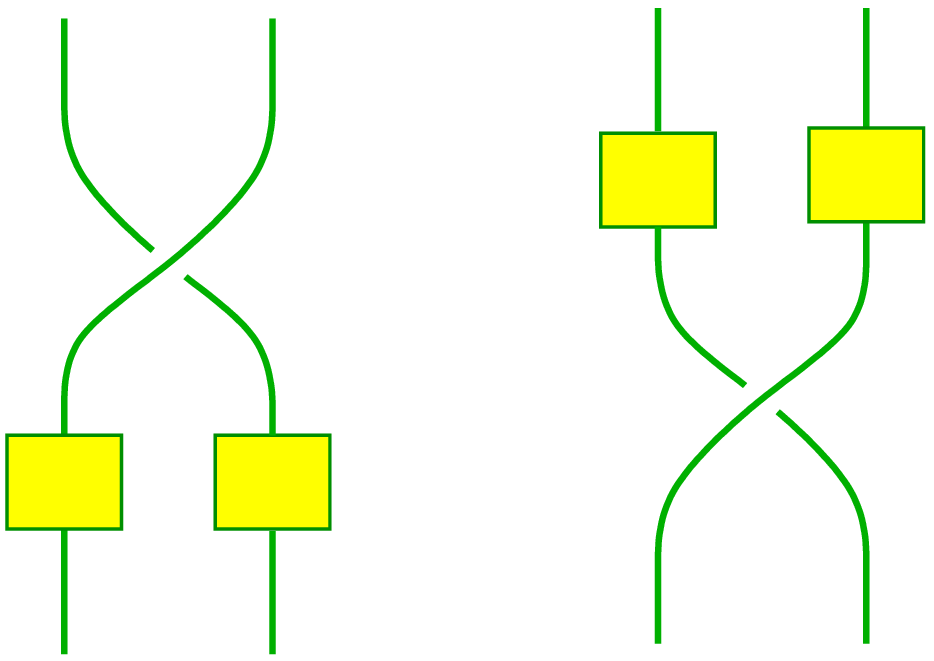}}
     \put(2.9,-7)   {\scriptsize$U$}
  \put(3.9,17.8) {\scriptsize$f$}
  \put(3.3,73)   {\scriptsize$X$}
  \put(-3.1,44)  {\scriptsize$c_{Y,X}^{}$}
  \put(26.8,-7)  {\scriptsize$V$}
  \put(27.4,18.7){\scriptsize$g$}
  \put(27.3,73)  {\scriptsize$Y$}
  \put(49,31)    {$=$}
  \put(69.4,-7)  {\scriptsize$U$}
  \put(70.2,51.7){\scriptsize$g$}
  \put(69.4,74)  {\scriptsize$X$}
  \put(89,27.3)  {\scriptsize$c_{U,V}^{}$}
  \put(91.8,-7)  {\scriptsize$V$}
  \put(92.6,51.4){\scriptsize$f$}
  \put(93.3,74)  {\scriptsize$Y$}
  \put(-8,-34)   {\footnotesize\Fbox{functoriality of braiding}}
  \put(200,0){
  {\Includepic{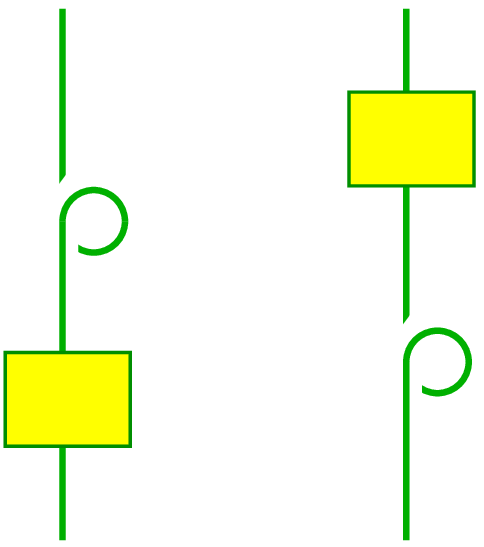}}
  \put(4.4,-7)   {\scriptsize$U$}
  \put(4.9,14.2) {\scriptsize$f$}
  \put(3.8,63)   {\scriptsize$V$}
  \put(22,26)    {$=$}
  \put(42.3,-7)  {\scriptsize$U$}
  \put(43.0,42.7){\scriptsize$f$}
  \put(41.8,63)  {\scriptsize$V$}
  \put(-16,-34)  {\footnotesize\Fbox{functoriality of twist}}}
} }
The braid relations and compatibility between braiding and twist looks as follows (note that this is only one of the braid relations, c.f.\ \eqref{braid_rel}):
\eqpic{tens_braid}{320}{62}{
 \put(10,30)   {
    \put(0,5){\Includepic{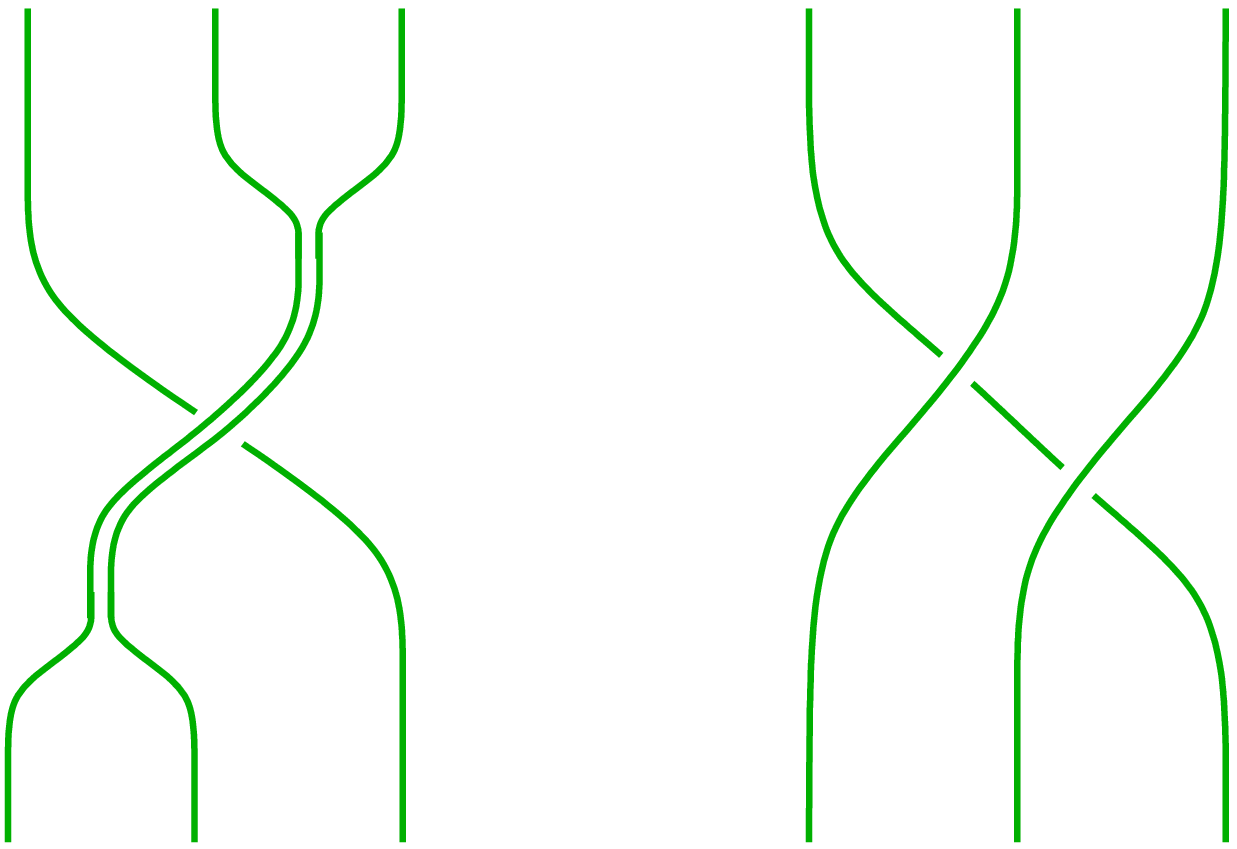}}
  \put(-8.3,50.5) {\scriptsize$c_{U\otimes V,W}^{\phantom x}$}
  \put(-2.3,-2.5) {\scriptsize$U$}
  \put(1.2,100.4) {\scriptsize$W$}
  \put(18.1,-2.5) {\scriptsize$V$}
  \put(21,100.4)  {\scriptsize$U$}
  \put(41.1,-2.5) {\scriptsize$W$}
  \put(42,100.4)  {\scriptsize$V$}
  \put(60,48)     {$=$}
  \put(83.0,56)   {\scriptsize$c_{U,Y}^{\phantom x}$}
  \put(85.5,-2.5) {\scriptsize$U$}
  \put(86.5,100.4){\scriptsize$W$}
  \put(108.8,-2.5){\scriptsize$V$}
  \put(108.8,100.4){\scriptsize$U$}
  \put(132.2,-2.5){\scriptsize$W$}
  \put(132.2,100.4){\scriptsize$V$}
  \put(123.1,44)  {\scriptsize$c_{V,W}^{\phantom x}$}
  \put(11,-34)    {\footnotesize\Fbox{tensoriality of braiding}}
            }
              \put(220,30){
              \Includepic{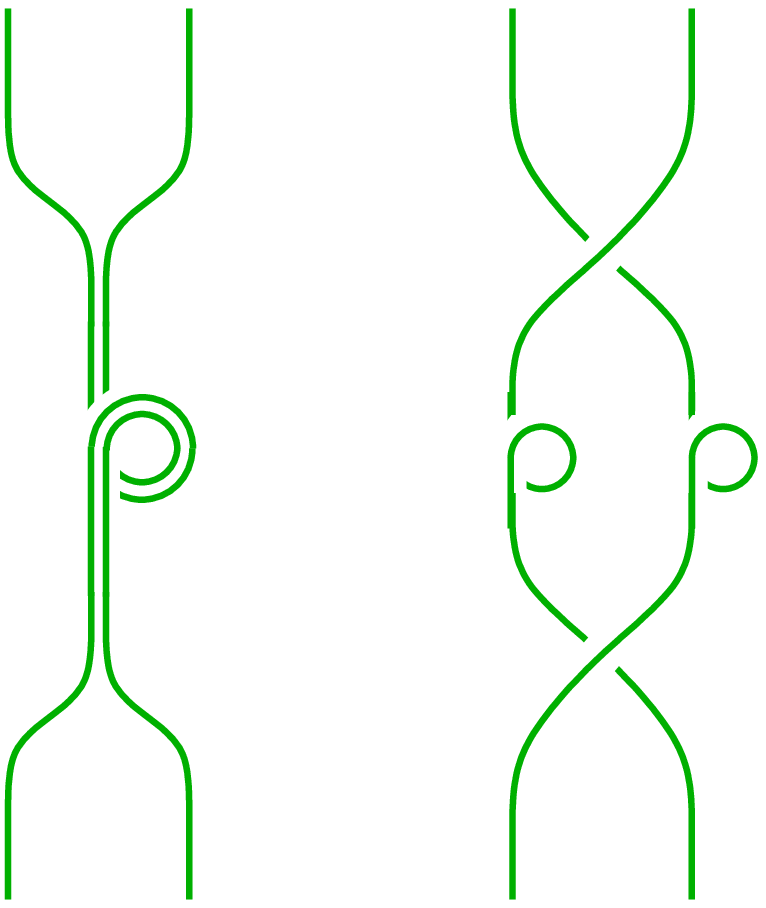}
  \put(-13.8,48) {\scriptsize$\theta_{U\otimes V}$}
  \put(-2.5,-7)  {\scriptsize$U$}
  \put(-2.2,101.3){\scriptsize$U$}
  \put(17.4,-7)  {\scriptsize$V$}
  \put(18.0,101.3){\scriptsize$V$}
  \put(30,46)    {$=$}
  \put(46.8,45)  {\scriptsize$\theta_{\!V}^{\phantom x}$}
  \put(53.1,-7)  {\scriptsize$U$}
  \put(53.8,101.3){\scriptsize$U$}
  \put(72.2,-7)  {\scriptsize$V$}
  \put(72.8,101.3){\scriptsize$V$}
  \put(83.3,45)  {\scriptsize$\theta_{\!U}^{\phantom x}$}
  \put(-5.4,-34) {\footnotesize\Fbox{braiding and twist}}
}   }

\subsection{More on dualities}
The (co)-evaluation morphisms of the left duality in a strictly sovereign ribbon category can be obtained from the right duality, via the braiding and twist, as in the following picture. For completeness we also display the expression for the left duality applied to a morphism $f\In\Hom(U,V)$.
\eqpic{left_dual}{320}{30}{\setulen80
\put(0,0){
\includepic{304}{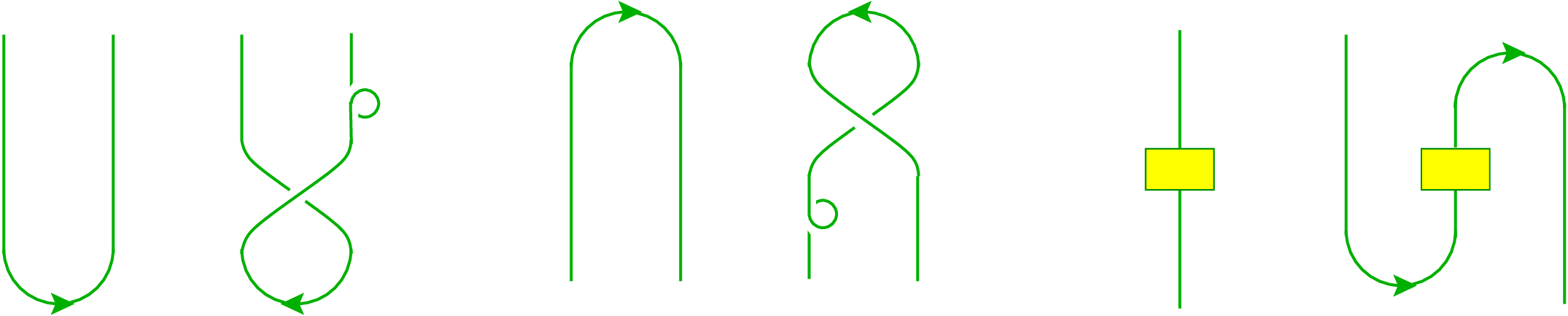}
  \put(-3.5,72)   {\scriptsize$\Eev U$}
  \put(25.5,72)   {\scriptsize$U$}
  \put(41.3,31)   {$=$}
  \put(55.5,72)   {\scriptsize$\Eev U$}
  \put(84.9,72)   {\scriptsize$U$}
  \put(139.9,-1)  {\scriptsize$U$}
  \put(165.5,-1)  {\scriptsize$\Eev U$}
  \put(182.1,31)  {$=$}
  \put(199.8,-1)  {\scriptsize$U$}
  \put(225.5,-1)  {\scriptsize$\Eev U$}
  \put(291.2,-8)  {\scriptsize$\Eev V$}
  \put(291.2,73.2){\scriptsize$\Eev U$}
  \put(291.9,33.5){\scriptsize$\Eev f$}
  \put(316,33)    {$=$}
  \put(335.2,72.9){\scriptsize$\Eev U$}
  \put(362.9,34.2){\scriptsize$f$}
  \put(388.6,-8)  {\scriptsize$\Eev V$}
} }
\ssubject{Uniqueness of duality}\label{app:dual_un}
If a duality exists it is unique up to natural isomorphism. Assume we have two dualities, $?^{\vee}$ and $?^{\isod}$, with evaluation and coevaluation morphisms $d_U$, $b_U$ and $d_U'$, $b_U'$, respectively. Then the
two dualities are naturally isomorphic via the isomorphism
\be
    \Phisod^U:=(d_U\oti\id_{U^{\isod}})\circ(\id_{U^\vee}\oti b'_U)\,\in\Hom(U^\vee,U^{\isod})\,.
\ee
It is straightforward to check that for any $f\in\Hom(U,V)$ we have:
\be
    \Phisod^U \circ f^\vee=f^{\isod}\circ\Phisod^V\,.
\ee
\pagebreak

\subsection{The antipode of $\L$}
Here we provide the calculation, showing that the definition of $\apo_\L$ in \eqref{Lyub_Hopf} indeed defines an antipode.
 \Eqpic{coend_apopr} {320} {94} {
 \put(10,0){\setulen70
   \put(0,160){
  \put(0,0)    {\INcludepichtft{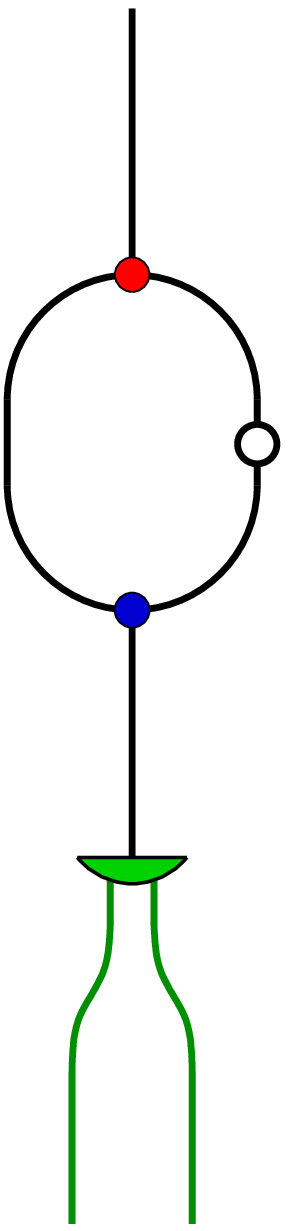}{266}
  \put(1,-9)    {\sse$ U^{\!\vee} $}
  \put(18,-9)   {\sse$ U $}
  }
  \put(50,64)  {$ = $}
  \put(80,0)   {\INcludepichtft{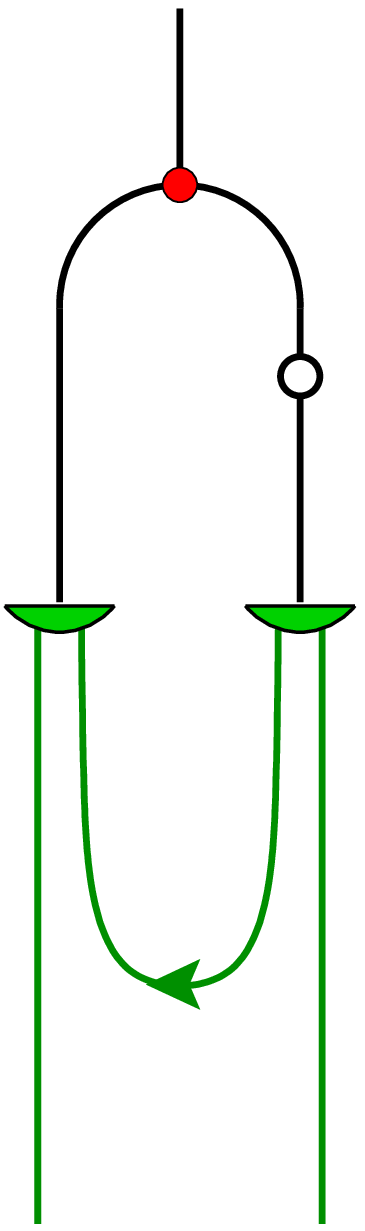}{266}}
  \put(140,64) {$ = $}
  \put(170,0)  {\INcludepichtft{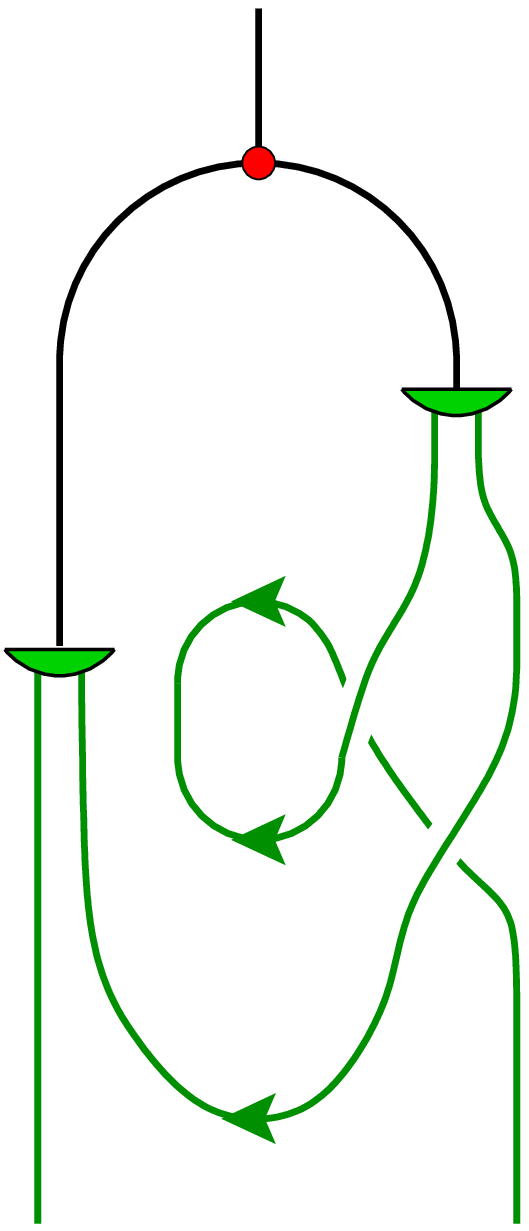}{266}
  \put(25,81)   {\sse$ U^{\!\vee\vee} $}
  \put(54,83)   {\sse$ U^{\!\vee} $}
  }
  \put(250,64) {$ = $}
  \put(280,0)  {\INcludepichtft{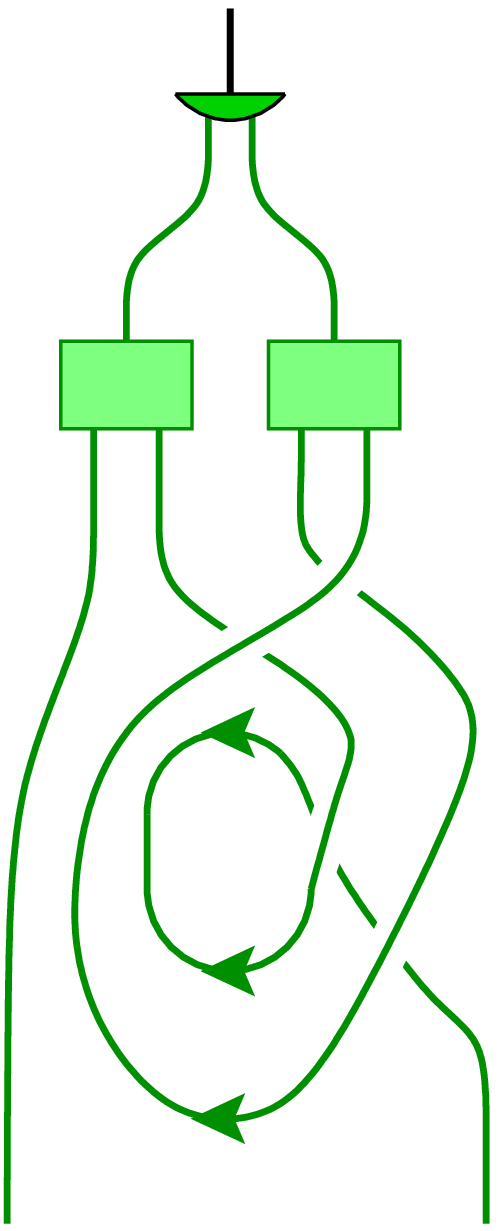}{266}
  \put(-4,-9)   {\sse$ U^{\!\vee} $}
  \put(41,78)   {\sse$ U $}
  \put(50,63)   {\sse$ U^{\!\vee} $}
  }
  \put(350,64) {$ = $}
  \put(380,0)  {\INcludepichtft{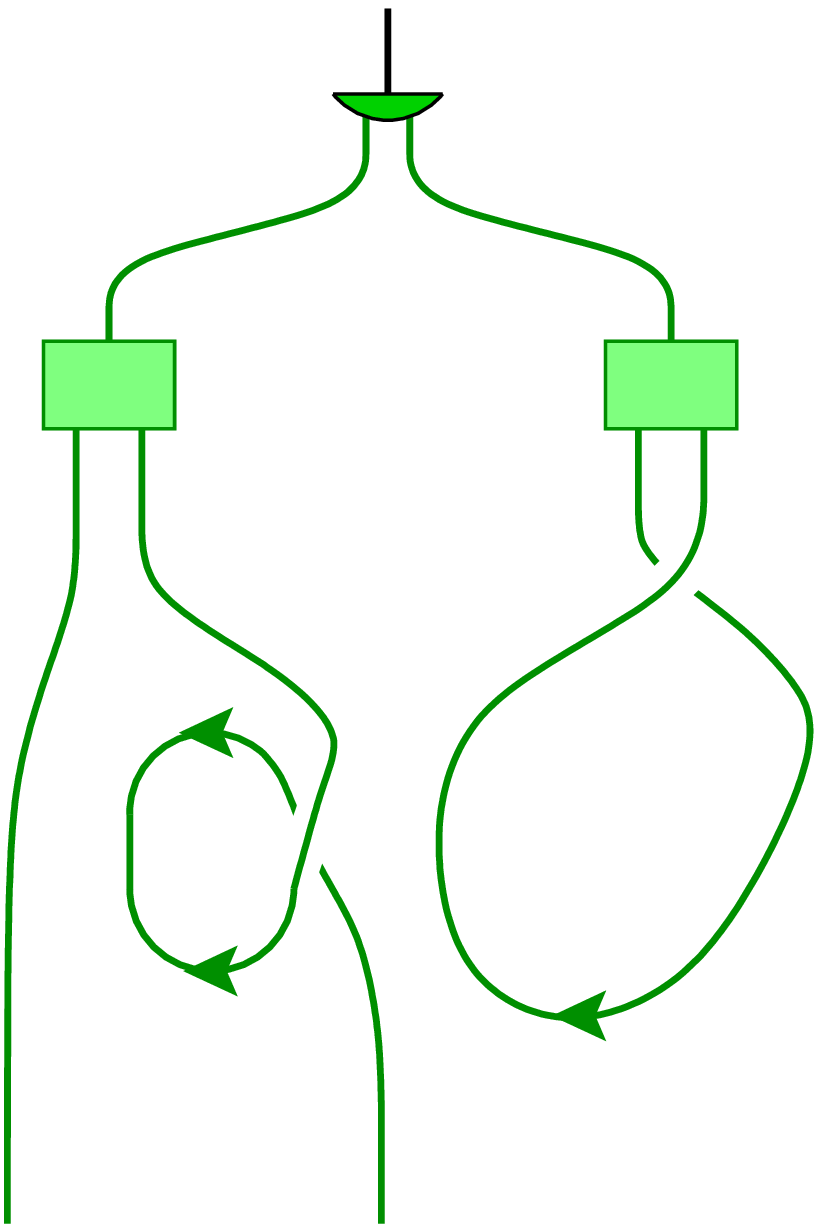}{266}
  \put(-10,112.5) {\sse$ {(U^{\!\vee}\Oti X^{})}^{\!\vee} $}
  \put(60,112.5)  {\sse$ U^{\!\vee}\Oti X^{} $}
  } }
   \put(56,0){
  \put(-28,64) {$ \stackrel*= $}
  \put(0,0)    {\INcludepichtft{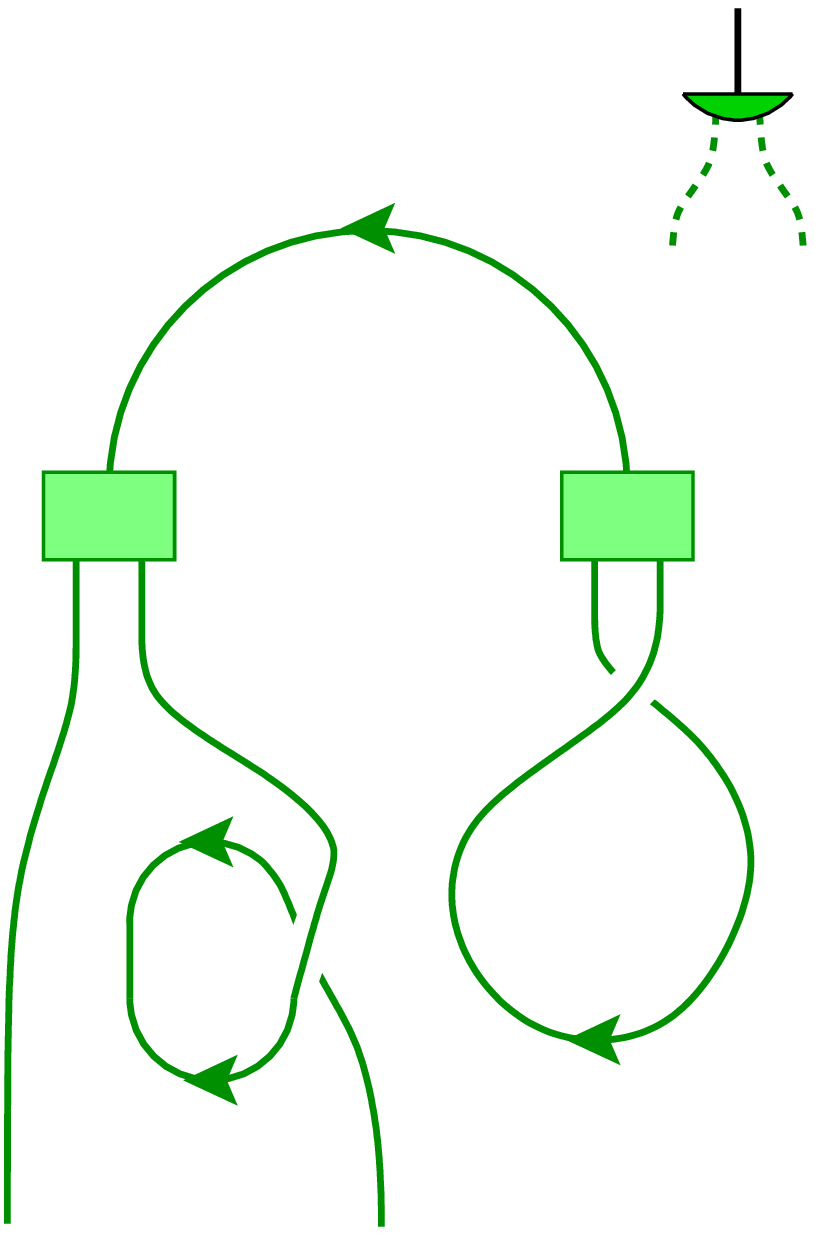}{266}
  \put(-4,-9)   {\sse$ U^{\!\vee} $}
  \put(-42,91)  {\sse$ {(U^{\!\vee}\Oti U^{})}^{\!\vee} $}
  \put(17,64)   {\sse$ U^{\!\vee\vee} $}
  \put(38,-9)   {\sse$ U $}
  \put(69,89)   {\sse$ U^{\!\vee}\Oti U^{} $}
  }
  \put(103,64) {$ = $}
  \put(130,0)  {\INcludepichtft{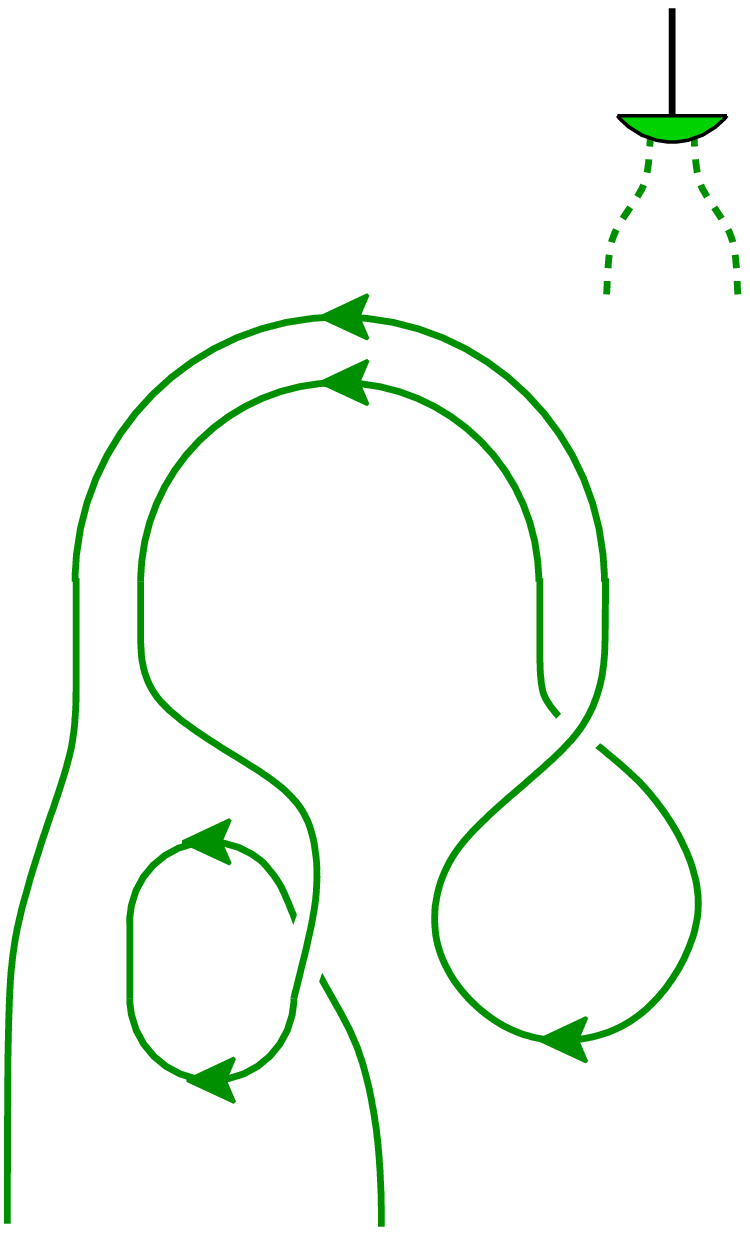}{266}
  \put(-4,-9)   {\sse$ U^{\!\vee} $}
  \put(17,67)   {\sse$ U^{\!\vee\vee} $}
  \put(38,-9)   {\sse$ U $}
  }
  \put(227,64) {$ = $}
  \put(255,0)  {\INcludepichtft{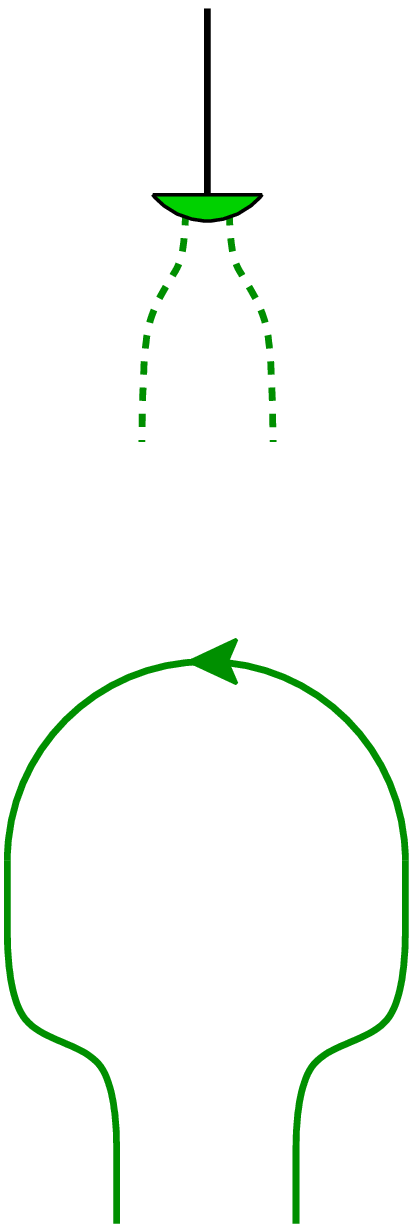}{266}}
  \put(319,64) {$ \equiv $}
  \put(349,0)  {\includepichtft{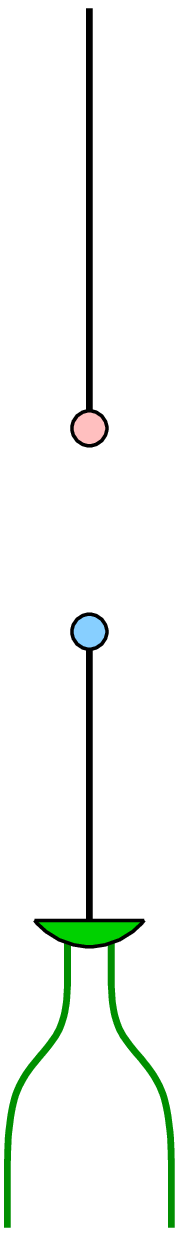}}
  } } }
To enlighten notation we have omitted many labels in this equations. The equality marked by $*$ implements the dinaturalness of the dinatural family applied to the morphism $\id_{U^\vee\oti U}\cir c_{U,U^\vee}\cir b_U$. The rest of the equalities are either straightforward implementations of the definitions in \eqref{Lyub_Hopf} or just deformations. The equality $m_\L\cir(\apo_\L\oti\id_\L)\cir\Delta_\L=\eta_\L\cir\eps_\L$ follows in an analogous manner.
\section{Algebra structure on morphism spaces}\label{app:alg_Hom}
Given an algebra $A$, in a $\k$-linear monoidal category, the space $V_A:=\Hom(\one,A)$ can be equipped with the structure of an associative unital algebra over $\k$. We define the product of two elements $f,g\In\Hom(\one,A)$ by
\be
    f\cdot g=m\circ(f\oti g)\,.
\ee
It follows immediately from associativity and unitality of $A$ that this is an associative product on $V_A$, and that it is unital  with unit map
\be
    \eta_{_{V_{\!A}}}=\eta_{_A}\,.
\ee
\section{Some formulas used in the derivation of $\CD$}
Below we display, the three manifolds that are obtained when cutting the connecting manifold $\Mws$ in \eqref{cob_4p_BulkF:48} along the annuli over the two cutting circles. The two three-balls with corners looks as follows, \cite[eq. (5.2)]{fuSs2}:
 \Eqpic{ball_n:49,ball_s:50}{320}{177}{ \setulen 80
  \put(0,344)  {$ \M^A_{\ia_3\ja_3,\ia_4\ja_4} ~= $}
  \put(96,240){
  \put(0,0)   {\includepicclax3{04}{49}}
  \put(105,103)  {\pl{\phi_{\alphz}}}
  \put(131,101)  {\pl{\phi_{\alphv}}}
  \put(147,80)   {\pA A}
  \put(90,193)   {\pl {\ia_3}}
  \put(132,185)  {\pl {\ia_4}}
  \put(95,8)     {\pl {\ja_3}}
  \put(137,18)   {\pl {\ja_4}}
  \put(188,80)   {\includepicclax3{04}{lsqarrov}}
  \put(215,83)   {\fbY SA1}
  }
  \put(-160,0){
  \put(239,96)  {$ \M^B_{\ia_1\ja_1,\ia_2\ja_2} ~= $}
  \put(335,-8){
  \put(0,0)   {\includepicclax3{04}{50}}
  \put(101,102)  {\pl{\phi_{\alphe}}}
  \put(126,101)  {\pl{\phi_{\alphd}}}
  \put(142,81)   {\pB B}
  \put(87,193)   {\pl {\ia_1}}
  \put(129,185)  {\pl {\ia_2}}
  \put(92,8)     {\pl {\ja_1}}
  \put(134,18)   {\pl {\ja_2}}
  \put(188,80)   {\includepicclax3{04}{lsqarrov}}
  \put(215,83)   {\fbY SB1}
  } } }
 and the full torus with corners looks as follows \cite[eq. (5.3)]{fuSs2}:
    \eqpic{cyl:51}{225}{101}{ \setulen 70 \put(0,-2){
  \put(0,144) {$ \M^{AB}_X ~= $}
  \put(83,0){ {\includepicclax2{66}{51}}
  \put(151,137)  {\pA A}
  \put(38,94)    {\pB B}
  \put(148,96)   {\pX X}
  \put(188,80)   {\includepicclax2{66}{lsqarrov}}
  \put(216,82)   {\fbY SB2}
  \put(122,240)  {\includepicclax2{66}{lsqarrov}}
  \put(150,242)  {\fbY SA2}
  } } }

\pagebreak
\section{Identities used in the proof of Theorem \ref{thm: MPG_Hopf}}\label{app:calculations}
In this appendix we collect a number of identities that are used in the proofs in section \ref{sec:inv_proof}.
As a consequence of the invertibility of the Frobenius map we have
\begin{lemma}\label{lemma:Sinv-7}
We have the following equality
 \eqpic{Sinv-7} {290} {38} {\setulen80
   \put(0,0) {\includepichtft{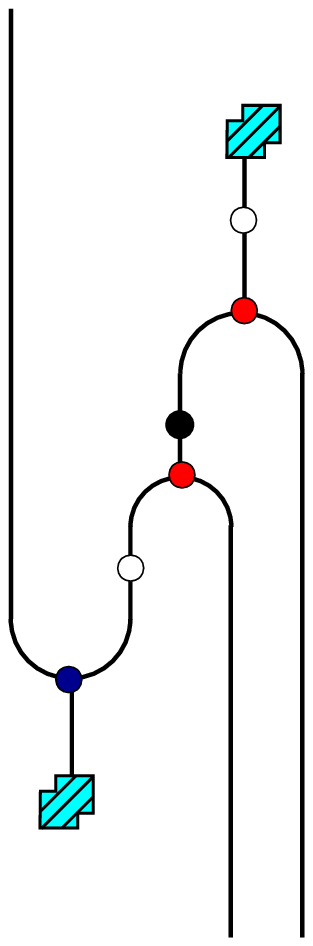}
   \put(18,-8)   {\sse$ \Hss $}
   \put(30,-8)   {\sse$ \Hss $}
   \put(11,12)   {\sse$ \Lambda $}
   \put(31.5,86) {\sse$ \lambda $}
   \put(-3,106)  {\sse$ \Hss $}
   }
   \put(53,38)   {$ = $}
   \put(85,0)  {\includepichtft{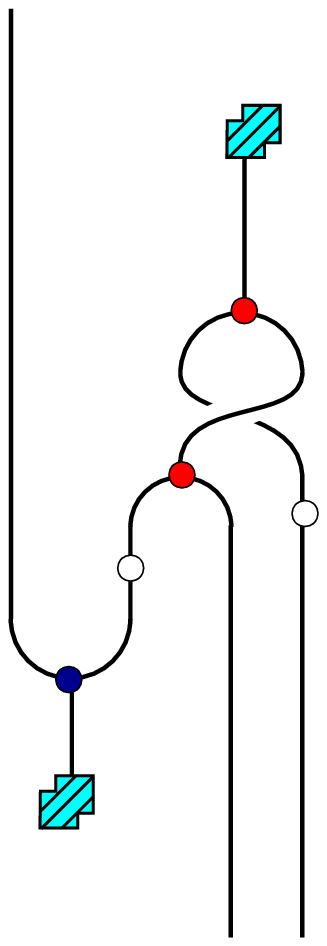}
   \put(18,-8.5) {\sse$ \Hss $}
   \put(30,-8.5) {\sse$ \Hss $}
   \put(11,12)   {\sse$ \Lambda $}
   \put(31.5,86) {\sse$ \lambda $}
   \put(-3,106)  {\sse$ \Hss $}
   }
   \put(138,38)  {$=$}
   \put(170,0) {\includepichtft{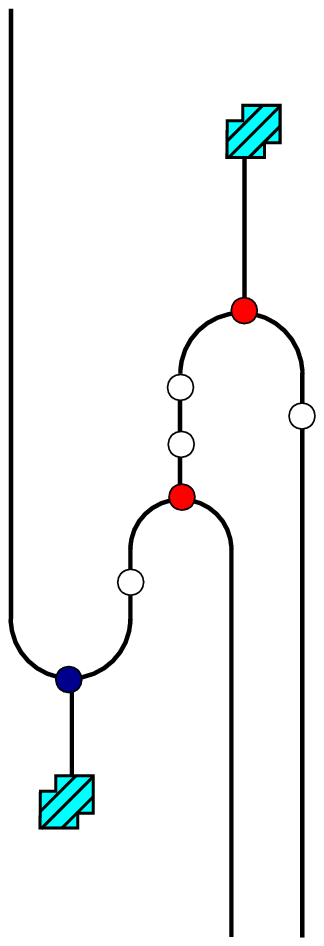}
   \put(18,-8.5) {\sse$ \Hss $}
   \put(30,-8.5) {\sse$ \Hss $}
   \put(11,12)   {\sse$ \Lambda $}
   \put(31.5,86) {\sse$ \lambda $}
   \put(-3,106)  {\sse$ \Hss $}
   }
   \put(223,38)  {$=$}
   \put(255,0) {\includepichtft{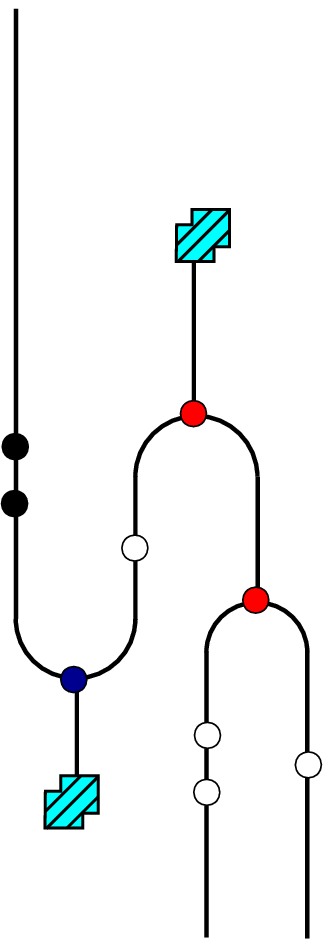}
   \put(17.5,-8.5) {\sse$ \Hss $}
   \put(30,-8.5) {\sse$ \Hss $}
   \put(11.4,11) {\sse$ \Lambda $}
   \put(26.5,74) {\sse$ \lambda $}
   \put(-1.4,106){\sse$ \Hss $}
   }
   \put(307,38)  {$ = $}
   \put(338,16) {\includepichtft{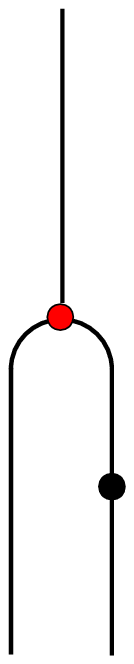}
   \put(-5,-8.5) {\sse$ \Hss $}
   \put(8,-8)    {\sse$ \Hss $}
   \put(3,75)    {\sse$ \Hss $}
   } }
\end{lemma}
\begin{proof}
The second equality from the property \eqref{OS2} of the cointegral, the third from $\apo\circ\Lambda\eq\Lambda$ and associativity, and the last from the invertibility of the Frobenius map, see \eqref{fmap_inv}.
\endofproof\end{proof}

We continue with a number of results involving coproducts and monodromy matrices:
\begin{lemma}\label{removeQ_1}
Conjugating by the monodromy matrix preserves the coproduct:
\eqpic{removeQs} {215} {31} {
  \put(0,0) {\Includepichtft{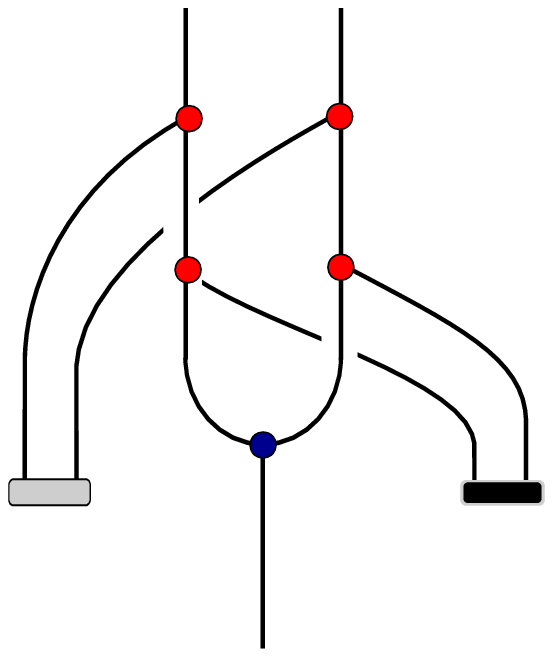}
  \put(23,-8.5) {\sse$ H $}
  \put(1,7.5) {\sse$ Q$}
  \put(47,7.5) {\sse$ Q^{-1}$}
  \put(16.5,74)    {\sse$ H $}
  \put(33,74)      {\sse$ H $}
  }
  \put(85,27)   {$=~\Delta_H~=$}
  \put(155,0) {\Includepichtft{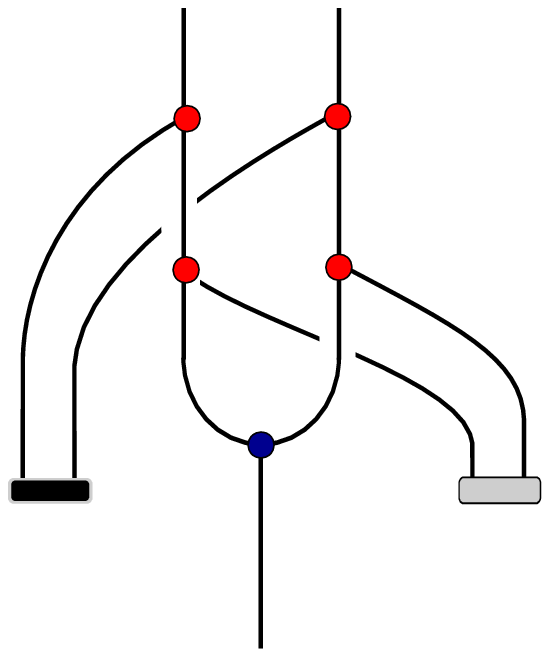}
  \put(23,-8.5) {\sse$ H $}
  \put(1,7.5) {\sse$ Q^{-1}$}
  \put(51,8.5) {\sse$ Q$}
  \put(16.5,74)    {\sse$ H $}
  \put(33,74)      {\sse$ H $}
  }
  }
\end{lemma}
\begin{proof}
    The first equality follows by using twice the fact that the $R$-matrix intertwines the coproduct and the opposite coproduct of $H$. The second follows from the first by multiplying with $Q^{-1}$ and $Q$.
\endofproof\end{proof}
\begin{remark}
The equality \eqref{removeQs} also follows directly by writing down the equality $m_{\BFH}=m_{\BFH}\cir c_{\BFH,\BFH}\cir c_{\BFH,\BFH}$ of morphisms in $\Hom_{H|H}(\BFH\oti\BFH,\BFH)$ (which holds since $\BFH$ is commutative), where $c_{F,F}$ is the braiding of $H$\Bimod.
\end{remark}
A consequence of Lemma \ref{removeQ_1} is the following results, which are used in the proof of mapping class group invariance.
\begin{lemma}\label{removeQ_2}
We have the following identities involving the coproduct and monodromy matrices:\\
  (i)
 \Eqpic{Pic_QQ3} {320} {47} {
 \put(0,10){\setulen80
  \put(0,0) { \includepichtft{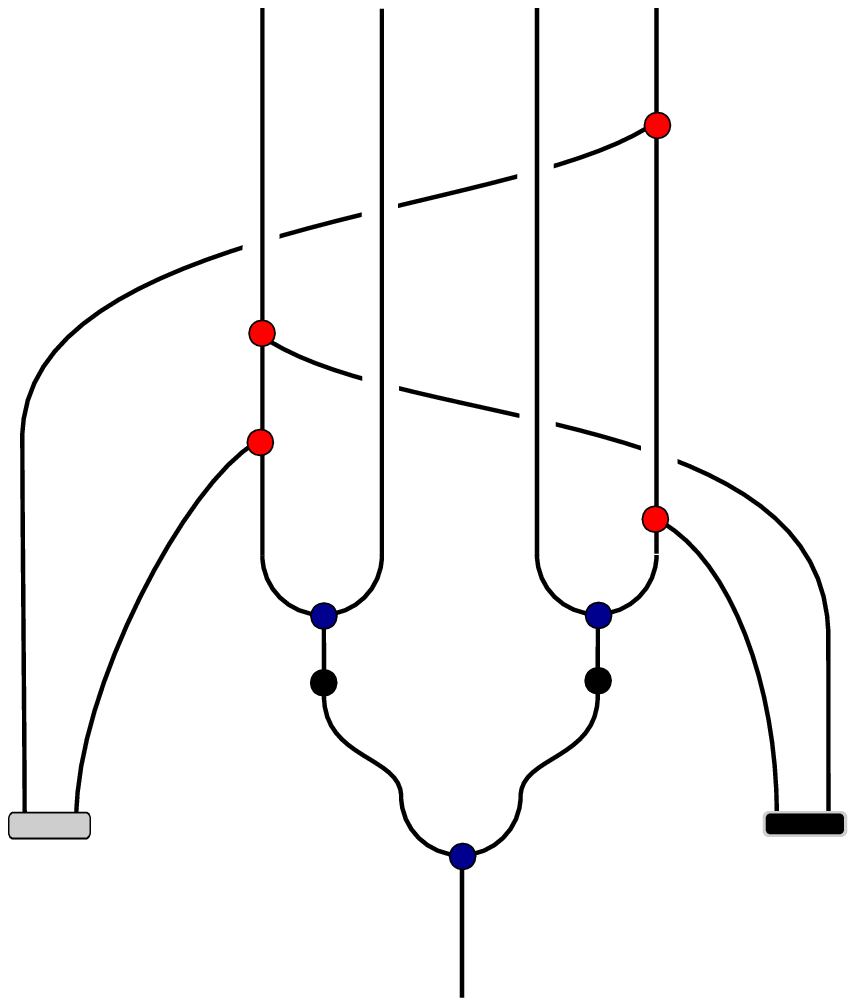}
  \put(24,114.1)   {\sse$ H $}
  \put(37.1,114.1) {\sse$ H $}
  \put(45,-11.2)   {\sse$ H $}
  \put(54,114.1)   {\sse$ H $}
  \put(67.5,114.1) {\sse$ H $}
  \put(0,8)        {\sse$Q$}
  \put(78.6,8.5)   {\sse$Q^{-1}$}
  }
  \put(111,52)     {$=$}
  \put(133,0) { \includepichtft{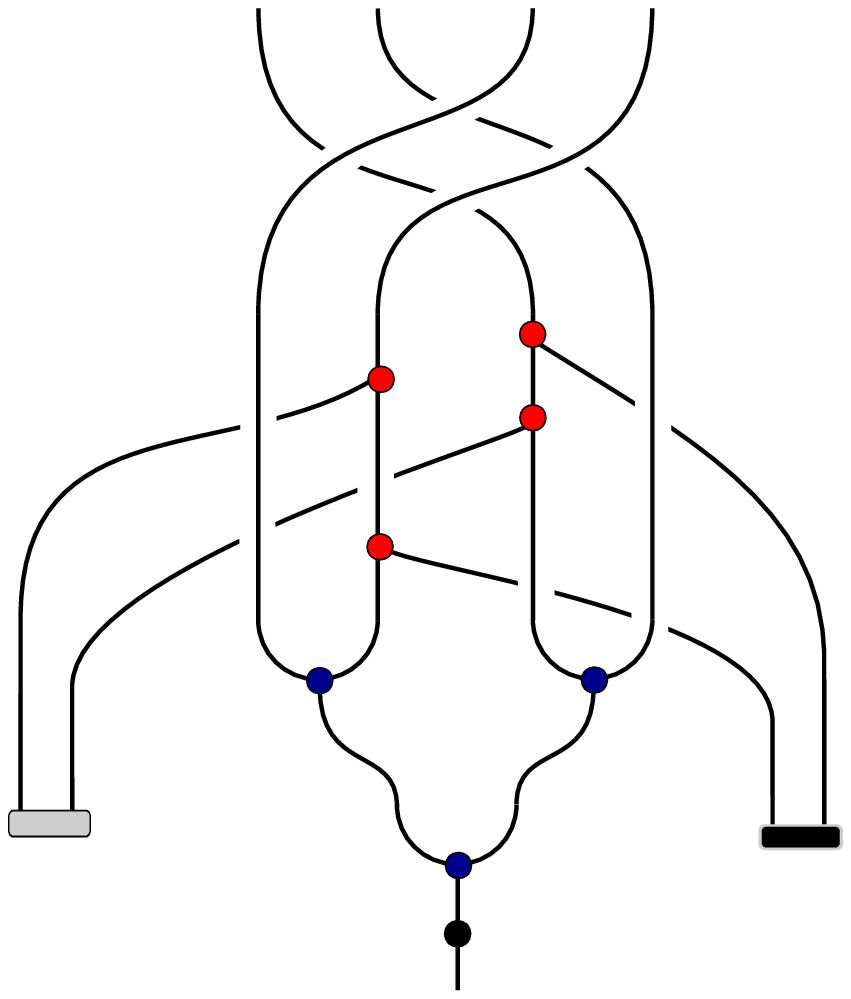}
  \put(23.3,114.1) {\sse$ H $}
  \put(36.7,114.1) {\sse$ H $}
  \put(44.2,-11.2) {\sse$ H $}
  \put(54,114.1)   {\sse$ H $}
  \put(67.5,114.1) {\sse$ H $}
  \put(0,8)        {\sse$Q$}
  \put(78.6,6.9)   {\sse$Q^{-1}$}
  }
  \put(249,52)     {$=$}
  \put(280,0) { \includepichtft{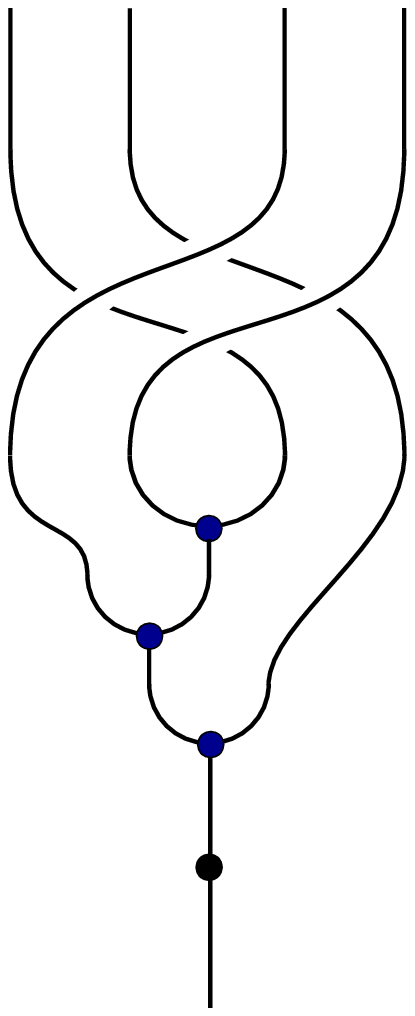}
  \put(-2.5,114.1) {\sse$ H $}
  \put(11.5,114.1) {\sse$ H $}
  \put(19,-11.2)   {\sse$ H $}
  \put(28,114.1)   {\sse$ H $}
  \put(40.7,114.1) {\sse$ H $}
  }
  \put(353,52)     {$=$}
  \put(379,0) { \includepichtft{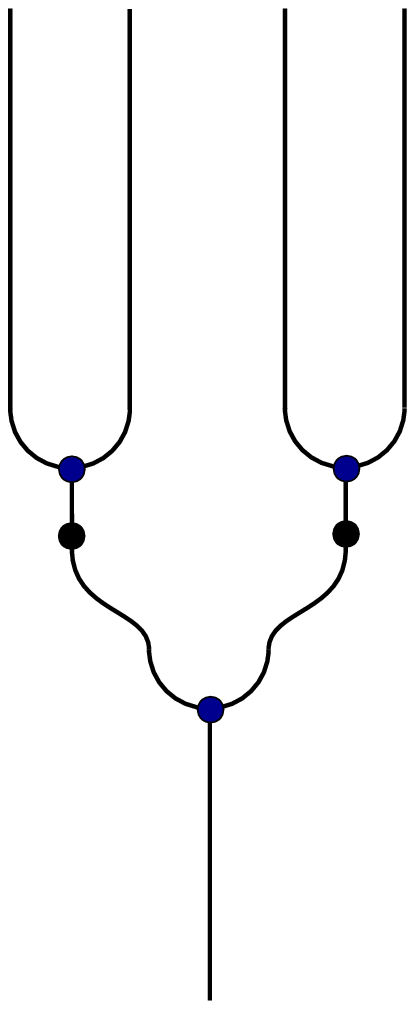}
  \put(-3.7,114.1) {\sse$ H $}
  \put(9.3,114.1)  {\sse$ H $}
  \put(17.5,-11.2) {\sse$ H $}
  \put(25.6,114.1) {\sse$ H $}
  \put(39.8,114.1) {\sse$ H $}
  } } }
  (ii)
   \Eqpic{inv_Qq3} {320} {52} { \put(0,2) { \setulen80
    \put(-2,0) { \includepichtft{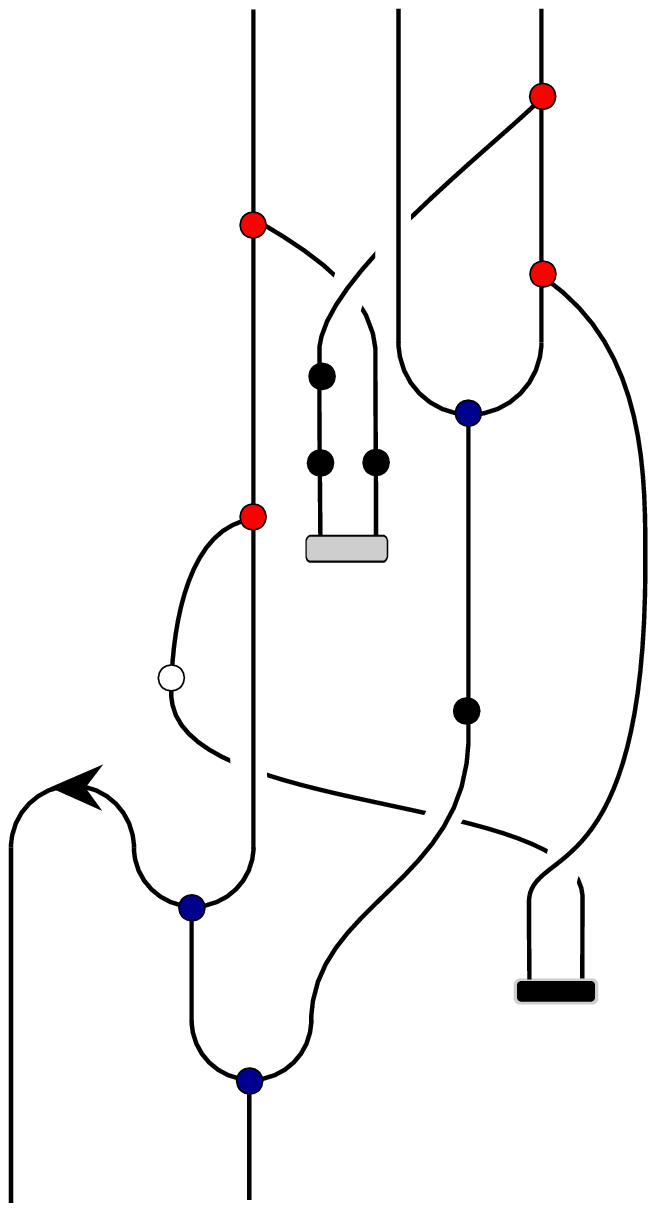}
  \put(-6.2,-11.2) {\sse$ \Hss $}
  \put(21.2,-11.2) {\sse$ H $}
  \put(22.8,135.2) {\sse$ H $}
  \put(32.6,60.9)  {\sse$Q$}
  \put(37.9,135.2) {\sse$ H $}
  \put(51.2,13.2)  {\sse$Q^{-1}$}
  \put(54.2,135.2) {\sse$ H $}
  }
  \put(90,57)     {$=$}
    \put(120,0) { \includepichtft{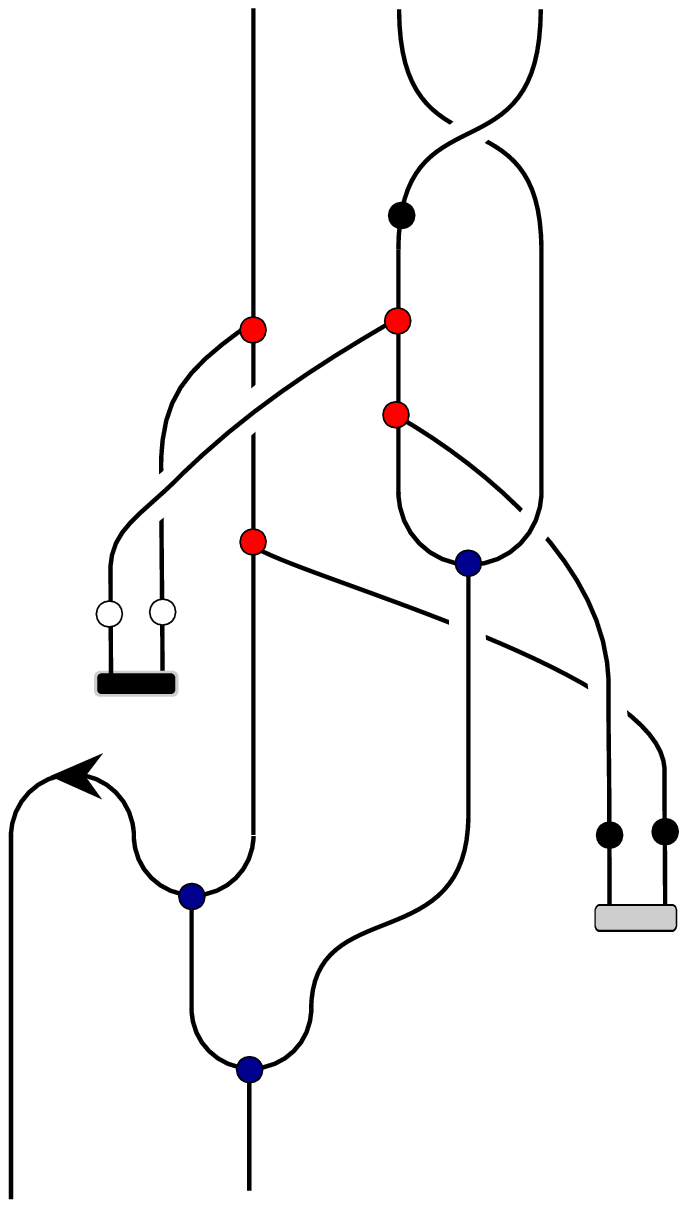}
  \put(-11.5,53)   {\sse$Q^{-1}$}
  \put(-6.5,-11.2) {\sse$ \Hss $}
  \put(21.2,-11.2) {\sse$ H $}
  \put(22.7,135.2) {\sse$ H $}
  \put(39.2,135.2) {\sse$ H $}
  \put(55.2,135.2) {\sse$ H $}
  \put(64.6,21.3)  {\sse$Q$}
  }
  \put(220,57)     {$=$}
    \put(250,0) { \includepichtft{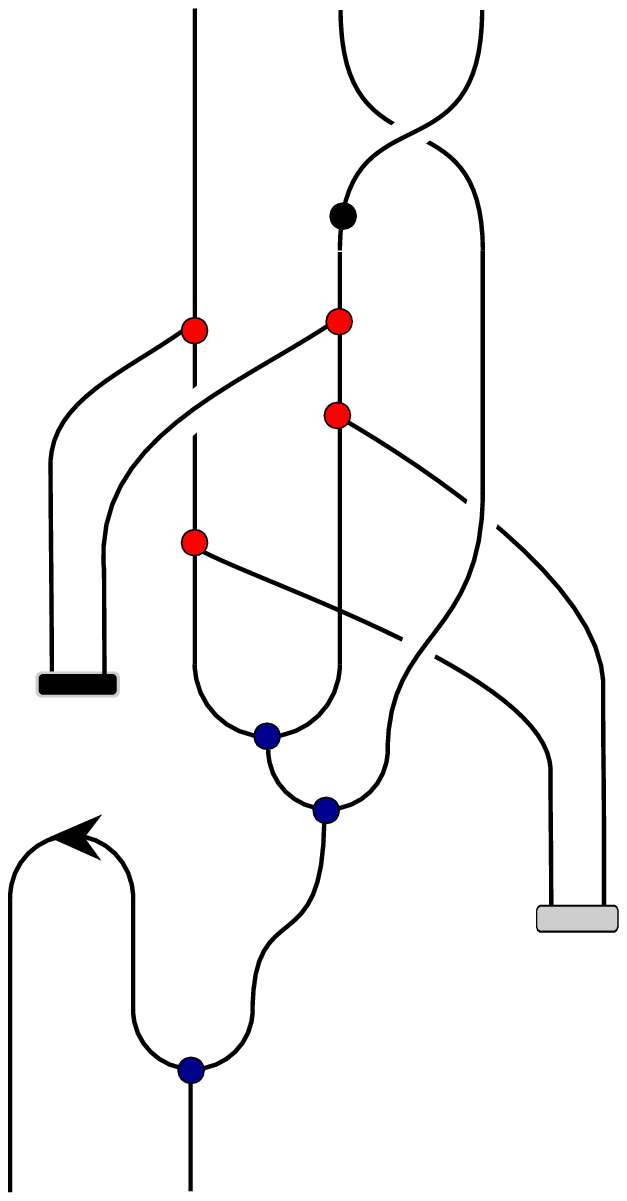}
  \put(-3,45.4)    {\sse$Q^{-1}$}
  \put(-6.2,-11.2) {\sse$ \Hss $}
  \put(15.2,-11.2) {\sse$ H $}
  \put(15.7,135.2) {\sse$ H $}
  \put(32.1,135.2) {\sse$ H $}
  \put(48.3,135.2) {\sse$ H $}
  \put(58.5,20.7)  {\sse$Q$}
  }
  \put(330,57)     {$=$}
    \put(360,0) { \includepichtft{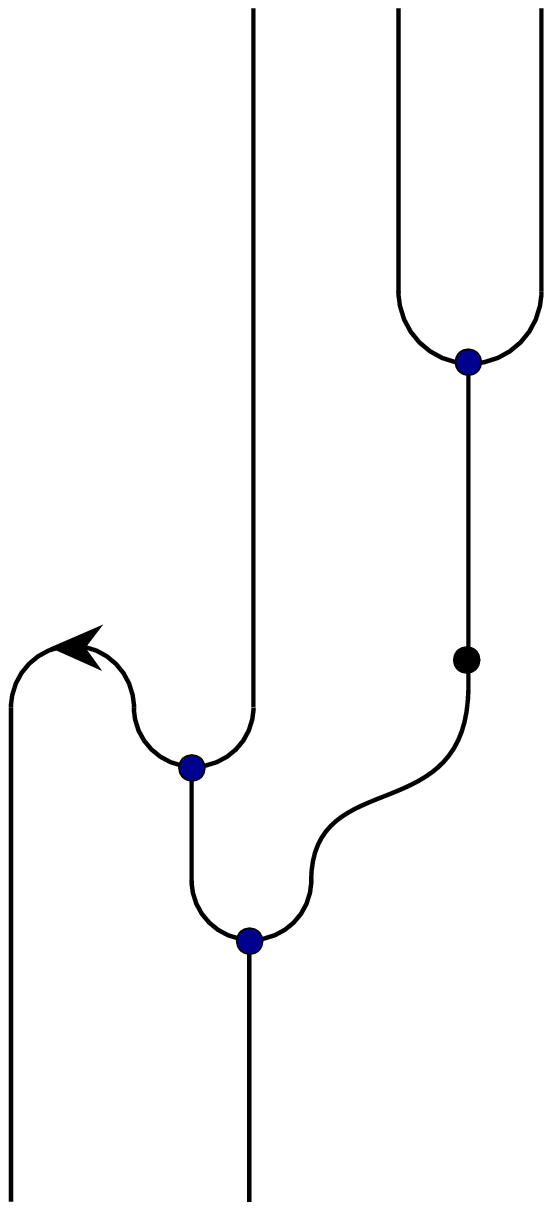}
  \put(-6.2,-11.2) {\sse$ \Hss $}
  \put(21.4,-11.2) {\sse$ H $}
  \put(21.6,135.2) {\sse$ H $}
  \put(37.8,135.2) {\sse$ H $}
  \put(54.3,135.2) {\sse$ H $}
  } } }

\end{lemma}
\begin{proof}
(i) The first equality follows by pushing the inverse antipodes through the coproduct, followed by a deformation compatible with the braid relations. The second equality follows from coassociativity combined with Lemma \ref{removeQ_1}. The third equality follows from coassociativity and the anti-coalgebra morphism property of $\apoi$.

(ii) The first equality is obtained by pushing $\apoi$ through the upper coproduct and products. The second equality follows from coassociativity and $(\apo\oti\apo)\circ Q=\flip HH\cir Q$. The third equality is obtained by applying \eqref{removeQs}, followed by coassociativity and that  $\apoi$ is an anti-algebra morphism, i.e. by undoing the corresponding rewritings in the previous equalities.
\endofproof\end{proof}
Next we observe
\pagebreak
\begin{lemma}\label{lemma:comFA_mon}
For any commutative Frobenius algebra $C$ in a ribbon category \C\ we have
\eqpic{2prod_mon} {220} {32} {\setulen80
  \put(0,0) {
  \put(0,0)    {\includepichtft{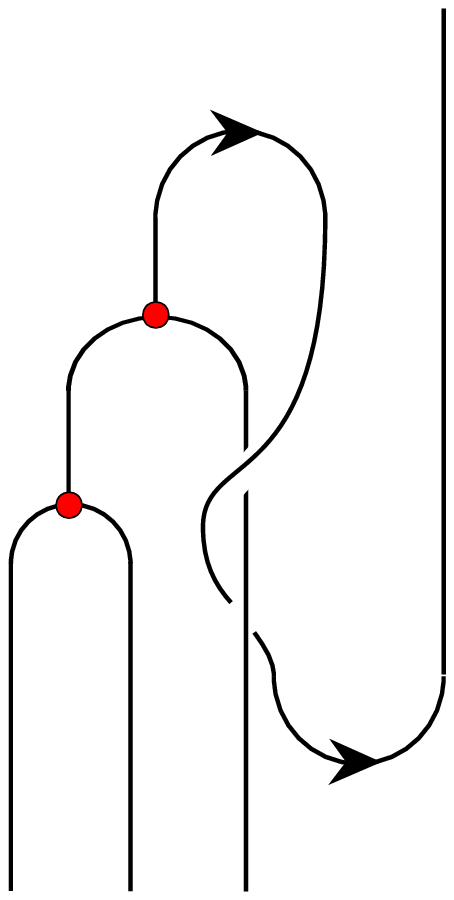}}
  \put(-3,-8.5){\sse$ C $}
  \put(10,-8.5){\sse$ C $}
  \put(22.2,-8.5){\sse$ C $}
  \put(45,101){\sse$ C $}
  }
  \put(80,35) {$=$}
  \put(120,0) {
  \put(0,0)    {\includepichtft{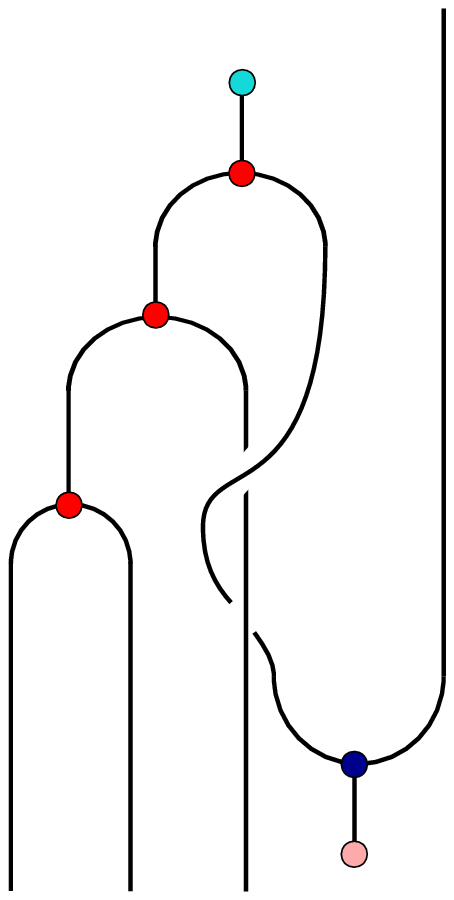}}
  \put(-3,-8.5){\sse$ C $}
  \put(10,-8.5){\sse$ C $}
  \put(22.2,-8.5){\sse$ C $}
  \put(45,101){\sse$ C $}
  }
  \put(200,35) {$=$}
  \put(240,0) {
  \put(0,0)    {\includepichtft{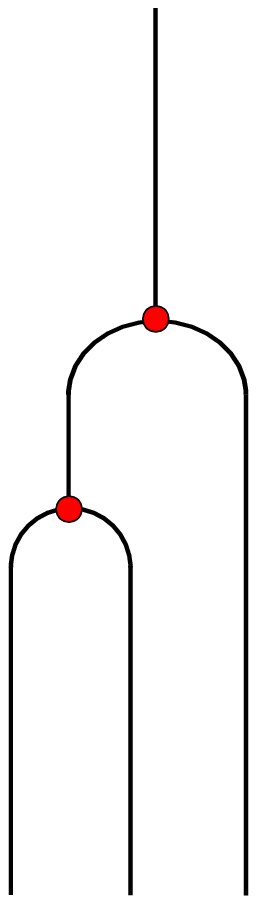}}
  \put(-3,-8.5){\sse$ C $}
  \put(10,-8.5){\sse$ C $}
  \put(22.2,-8.5){\sse$ C $}
  \put(14,101){\sse$ C $}
  }
  \put(280,100) {\catpic}
  }
\end{lemma}
\begin{proof}
The first equality follows by using that $C$ is Frobenius and inserting $\Id_{\rd C}\eq\Phi_{\text l}\cir\Phi_{\text l}^{-1}$, with $\Phi_{\text l}$ the isomorphism defined in \eqref{def_Phi}. The second equality is obtained by first using associativity, commutativity twice and then again associativity followed by the Frobenius property \eqref{frob}.
\endofproof\end{proof}
A consequence of Lemma \ref{lemma:comFA_mon} is
\begin{lemma}\label{lemma:cop_Delta}
We have the following identities involving the coproduct and monodromy matrices:\\
  (i)
\eqpic{Q2Delta} {120} {35} {\setulen80
  \put(0,0) {
  \put(0,0)    {\includepichtft{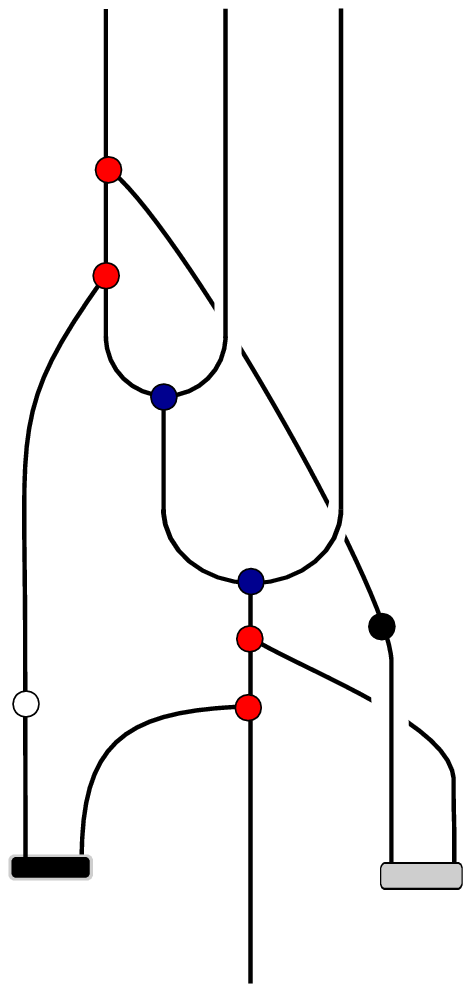}}
  \put(22.2,-8.5){\sse$ H $}
  \put(6.6,112)    {\sse$ H $}
  \put(19.7,112)   {\sse$ H $}
  \put(32.2,112)   {\sse$ H $}
  \put(0,4){\sse$Q^{-1}$}
  \put(42.2,3){\sse$Q$}
  }
  \put(75,35) {$=$}
  \put(120,0) {
  \put(0,0)    {\includepichtft{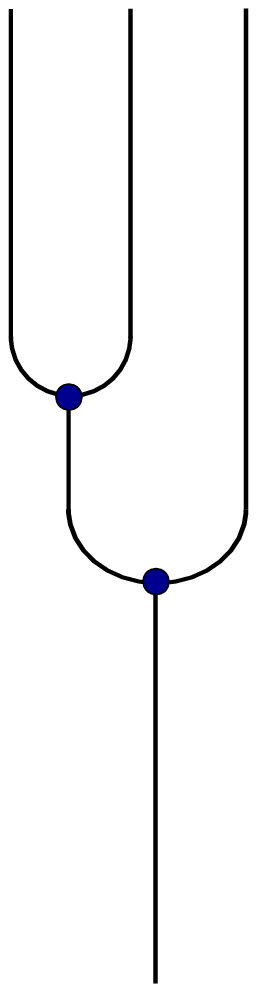}}
  \put(-3.6,112)   {\sse$ H $}
  \put( 9.5,112)   {\sse$ H $}
  \put(22.2,112)   {\sse$ H $}
  \put(12.2,-8.5){\sse$ H $}
  }
  }
  (ii)
\eqpic{Q3Delta} {160} {55} {\setulen80
  \put(0,0) {
  \put(0,0)    {\includepichtft{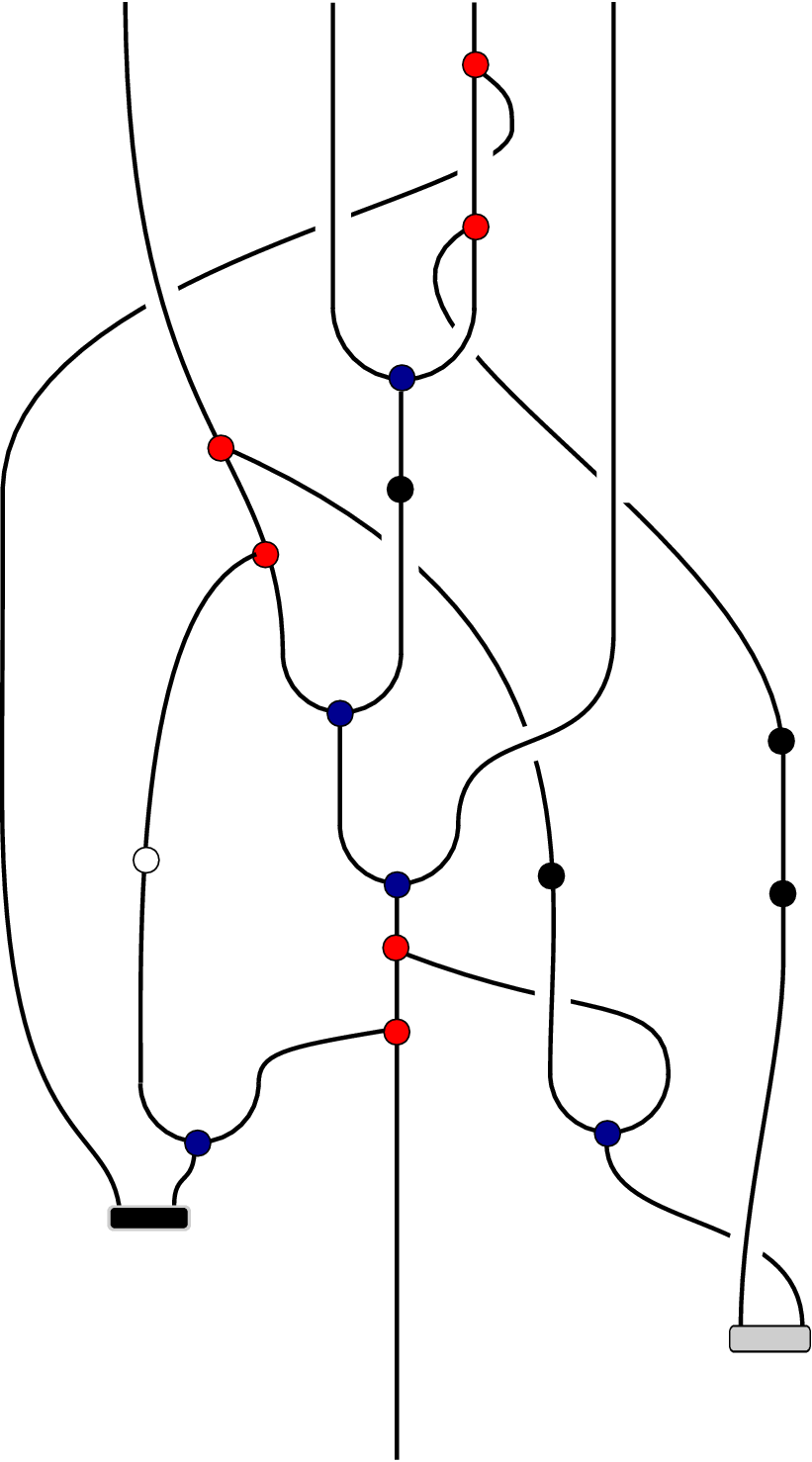}}
    \put(9.5,165.7)  {\sse$ H $}
  \put(33.1,165,7) {\sse$ H $}
  \put(39.2,-9.8)  {\sse$ H $}
  \put(47.2,165.7) {\sse$ H $}
  \put(63.5,165.7) {\sse$ H $}
  \put(9.2,16.4)   {\sse$Q^{-1}$}
  \put(80.3,2.5)   {\sse$Q$}
  }
  \put(110,55) {$=$}
  \put(140,0) {
  \put(0,0)    {\includepichtft{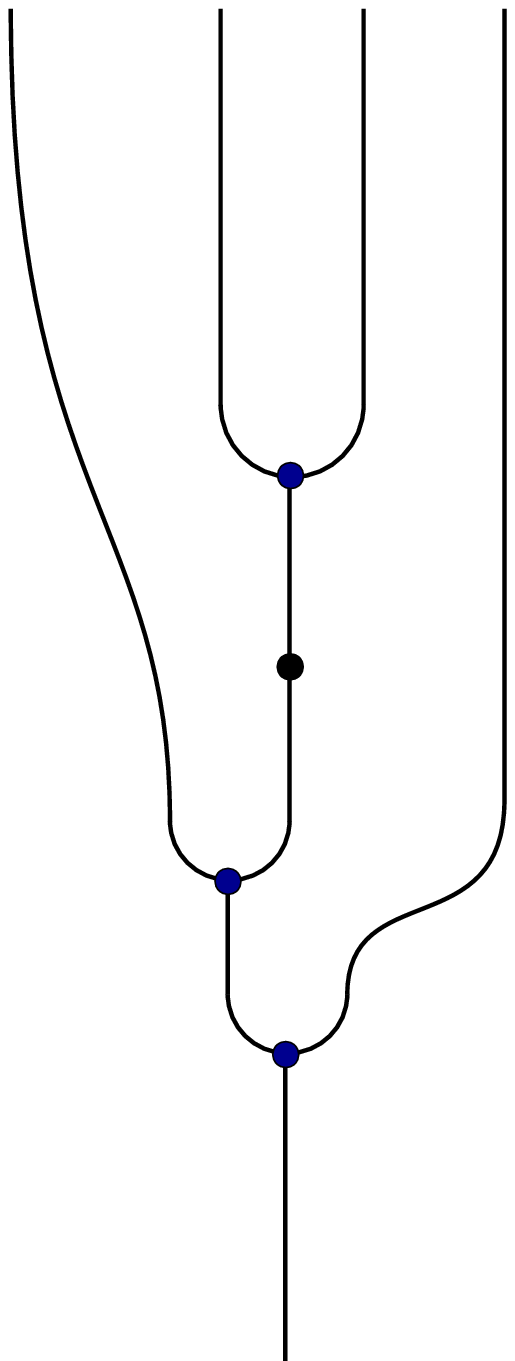}}
  \put(-4.1,165.7) {\sse$ H $}
  \put(18.4,165.7) {\sse$ H $}
  \put(25.4,-9.8)  {\sse$ H $}
  \put(35.1,165.7) {\sse$ H $}
  \put(50.5,165.7) {\sse$ H $}
  }
  }
\end{lemma}
\begin{proof}
(i)
Writing out the right-most expression in \eqref{2prod_mon} for $C\eq\BFH$, which is a commutative Frobenius algebra in $H$\Bimod, as a linear map and composing with dualities we obtain the right hand side of \eqref{Q2Delta}.
The left hand side of \eqref{Q2Delta} is obtained from the left-most expression in \eqref{2prod_mon} by applying Lemma \ref{lemma_helpid} (iii) (note that \eqref{Q2Delta} involves a monodromy between $\BFH$ and $\rd\BFH$) and composing with dualities.

(ii) Starting from the left hand side of \eqref{Q3Delta}, pushing the upper inverse antipode to the top and then using the bialgebra axiom several times to push the two lower products, we obtain the left hand side of
\Eqpic{Q3Delta_proof} {320} {60} {\setulen80
  \put(0,0) {
  \put(0,0)    {\includepichtft{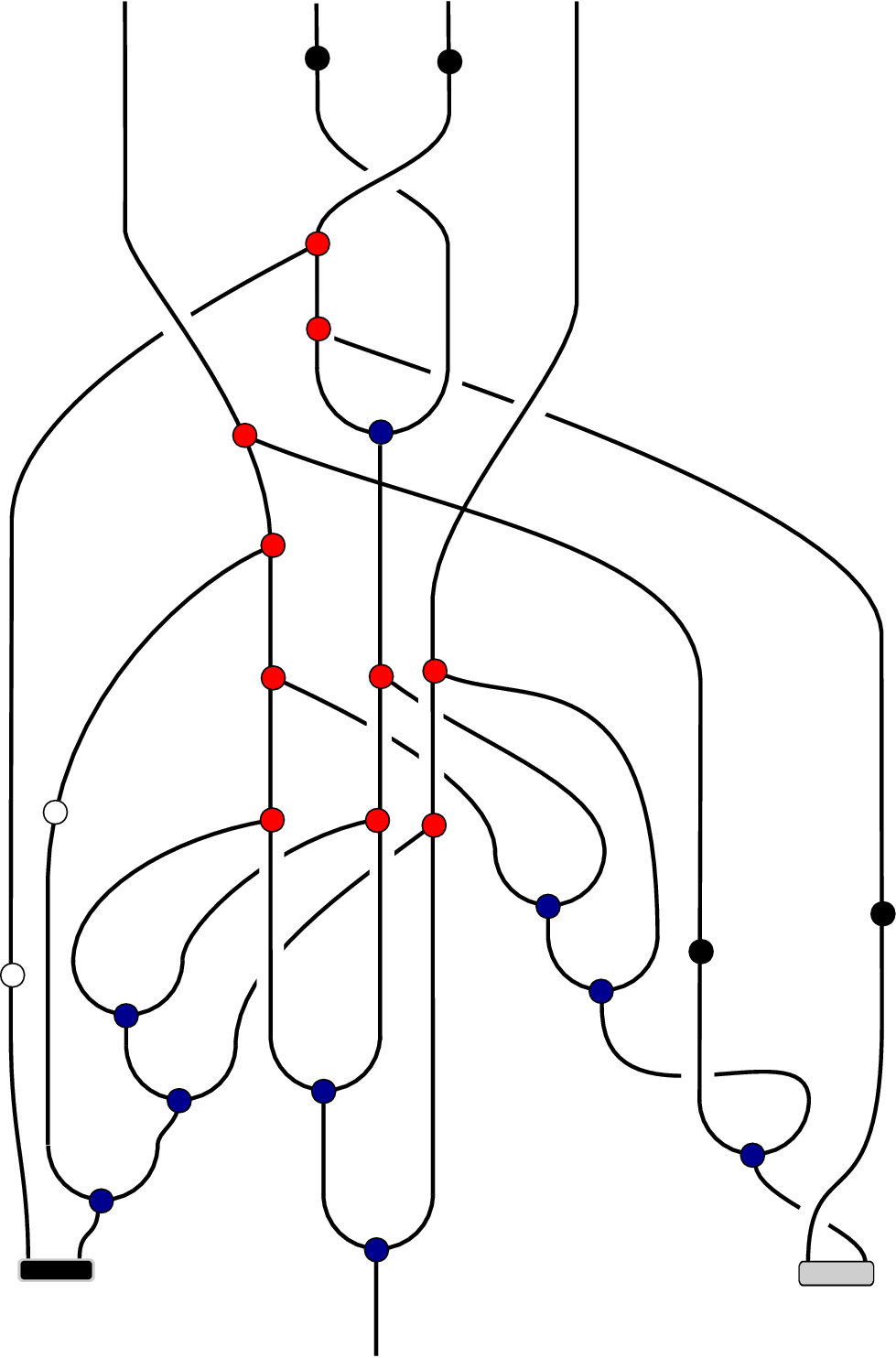}}
  \put(40.2,-8.5){\sse$ H $}
  \put(2.2,1){\sse$Q^{-1}$}
  \put(96.2,1){\sse$Q$}
  \put(10.7,166.2) {\sse$ H $}
  \put(33.7,166.2) {\sse$ H $}
  \put(49.7,166.2) {\sse$ H $}
  \put(64.7,166.2) {\sse$ H $}
  }
  \put(127,60) {$=$}
  \put(154,0) {
  \put(0,0)    {\includepichtft{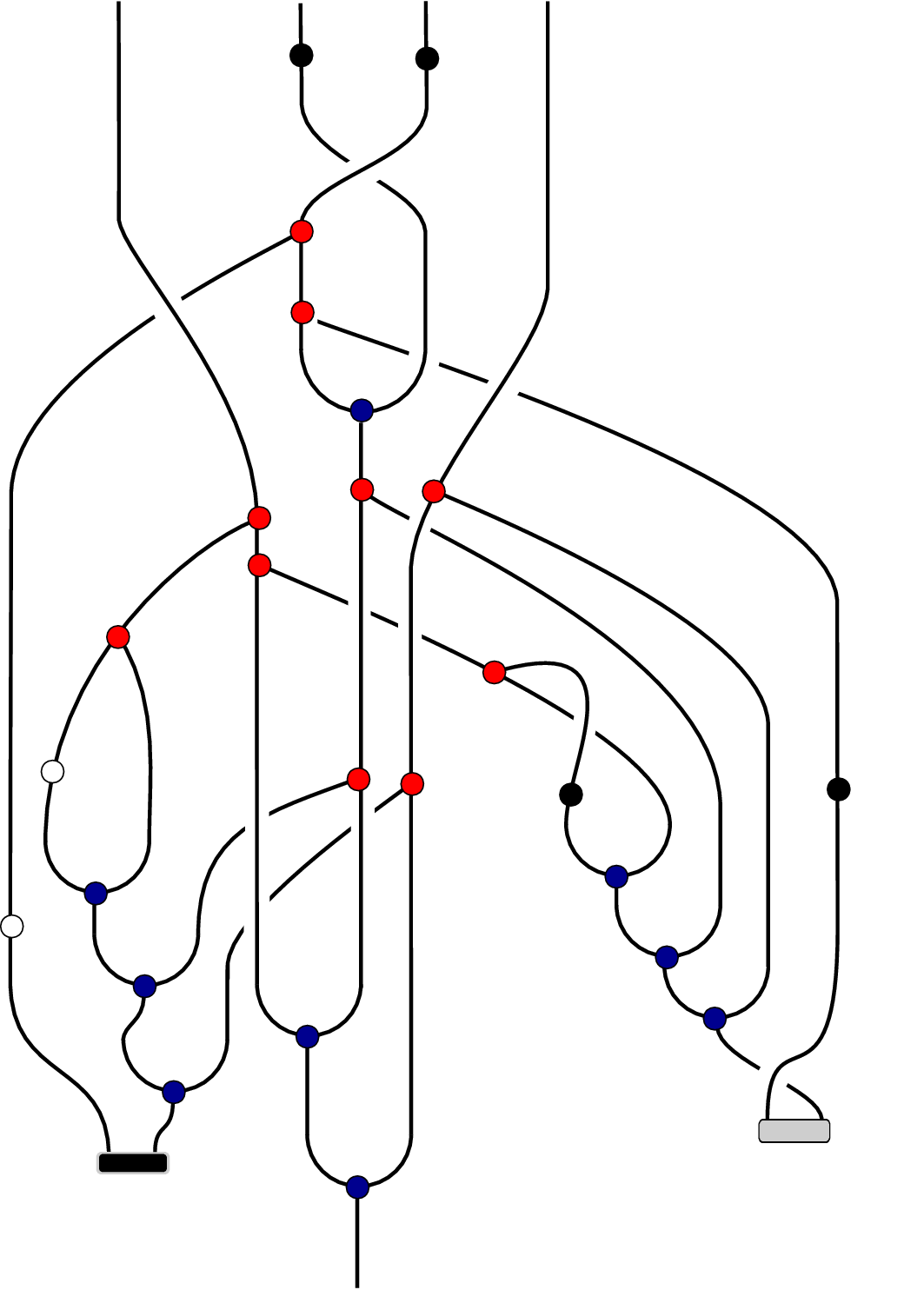}}
  \put(40.2,-8.5){\sse$ H $}
  \put(12.2,8){\sse$Q^{-1}$}
  \put(96.2,10){\sse$Q$}
  \put(10.7,166.2) {\sse$ H $}
  \put(33.7,166.2) {\sse$ H $}
  \put(49.7,166.2) {\sse$ H $}
  \put(64.7,166.2) {\sse$ H $}
  }
  \put(281,60) {$=$}
  \put(308,0) {
  \put(0,0)    {\includepichtft{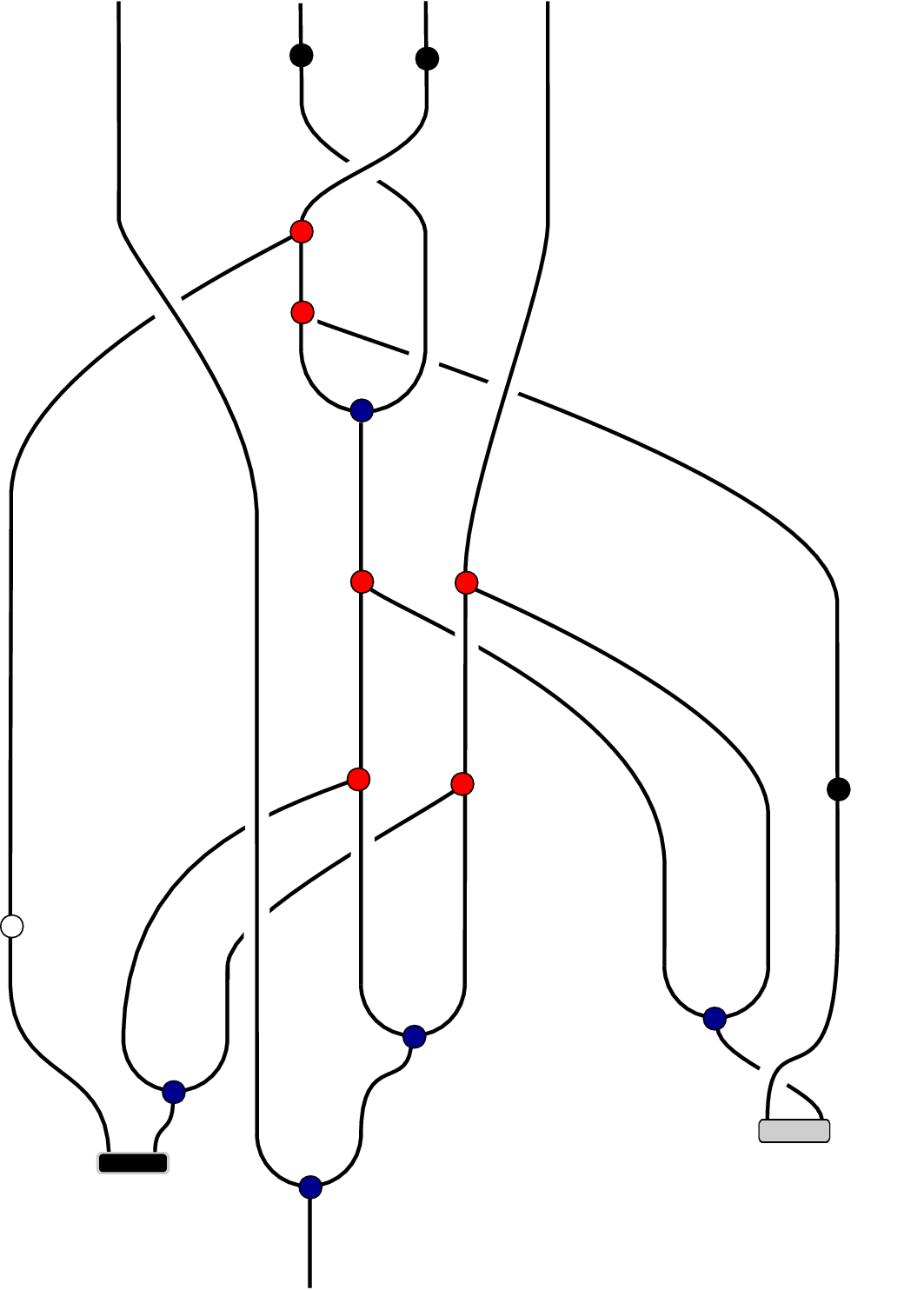}}
  \put(32.2,-8.5){\sse$ H $}
  \put(12.2,8){\sse$Q^{-1}$}
  \put(96.2,10){\sse$Q$}
  \put(10.7,166.2) {\sse$ H $}
  \put(33.7,166.2) {\sse$ H $}
  \put(49.7,166.2) {\sse$ H $}
  \put(64.7,166.2) {\sse$ H $}
  }
  }
Here the first equality follows from associativity and coassociativity applied to the coproducts taken after the $Q$-matrices. The second equality follows from the defining property of the antipode and coassociativity. That the right hand side of \eqref{Q3Delta_proof} equals the right hand side of \eqref{Q3Delta} follows by applying first the connecting axiom to the coproducts taken after the monodromy matrices and then \eqref{Q2Delta}.
\endofproof
\end{proof}
In addition, we have the following identities for linear maps in $H$\Bimod:
\begin{lemma}\label{lemma_helpid}
We have the following identities:\\

(i)
\eqpic{reppastint} {120} {29} {
  \put(0,0) {\Includepichtft{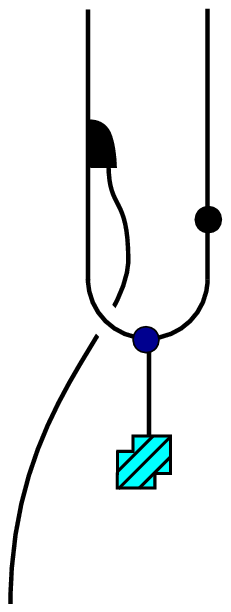}
  \put(-4,-8.5) {\sse$ H $}
  \put(20,14.5)  {\sse$ \Lambda $}
  \put(-1.8,49.3)  {\sse$ \rad$}
  \put(6,68.5) {\sse$ H $}
  \put(19,68.5) {\sse$ H $}
  }
  \put(50,29)   {$=$}
  \put(80,0) {\Includepichtft{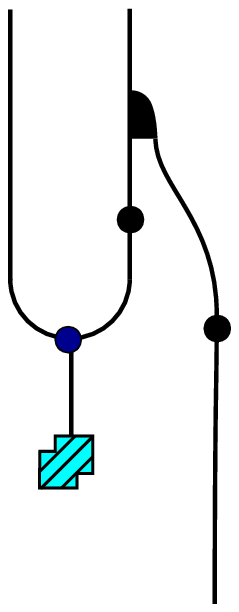}
  \put(18.5,-8.5){\sse$ H $}
  \put(-4.5,12.8){\sse$ \Lambda $}
  \put(17.6,52.8){\sse$ \rad $}
  \put(-3,68.5)  {\sse$ H $}
  \put(10,68.5)  {\sse$ H $}
  }
  }

  (ii) Denoting the left and right $H$-actions on
$\HKH^{\otimes p}_{}$ by $\rho^{}_{\HKH^{\otimes p}_{}}$
and $\ohr^{}_{\HKH^{\otimes p}_{}}$, respectively, we have
\Eqpic{Hpastloops} {320} {105} {\setulen80
  \put(0,115) {$\twoact:=$}
  \put(30,0) {\includepichtft{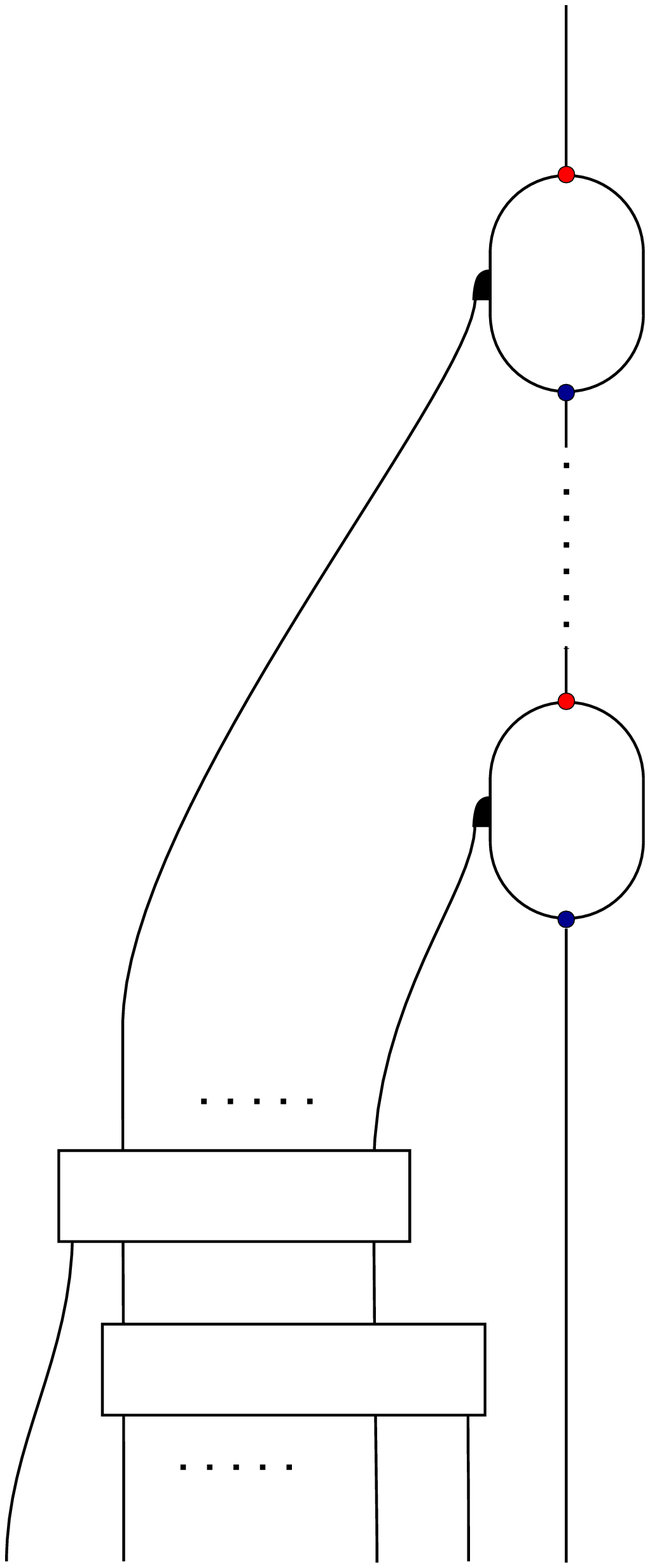}
   \put(-5,-11.2)   {\sse$ H $}
  \put(17.2,-11.2) {\sse$ K $}
  \put(60.2,-11.2) {\sse$ K $}
  \put(66.8,128.2) {\sse$ \rho^{\HKH}_{\BFH} $}
  \put(66.8,218.9) {\sse$ \rho^{\HKH}_{\BFH} $}
  \put(76.2,-11.2) {\sse$ H $}
  \put(92.2,-11.2) {\sse$ F $}
  \put(92.9,273.3) {\sse$ F $}
  \put(36,31.5)  {\sse$ \ohr_{\HKH^{\!\otii p}}^H $}
  \put(36,61.5)  {\sse$ \rho_{\HKH^{\!\otii p}}^H $}
  }
  \put(175,105) {$=$}
  \put(220,0) {\includepichtft{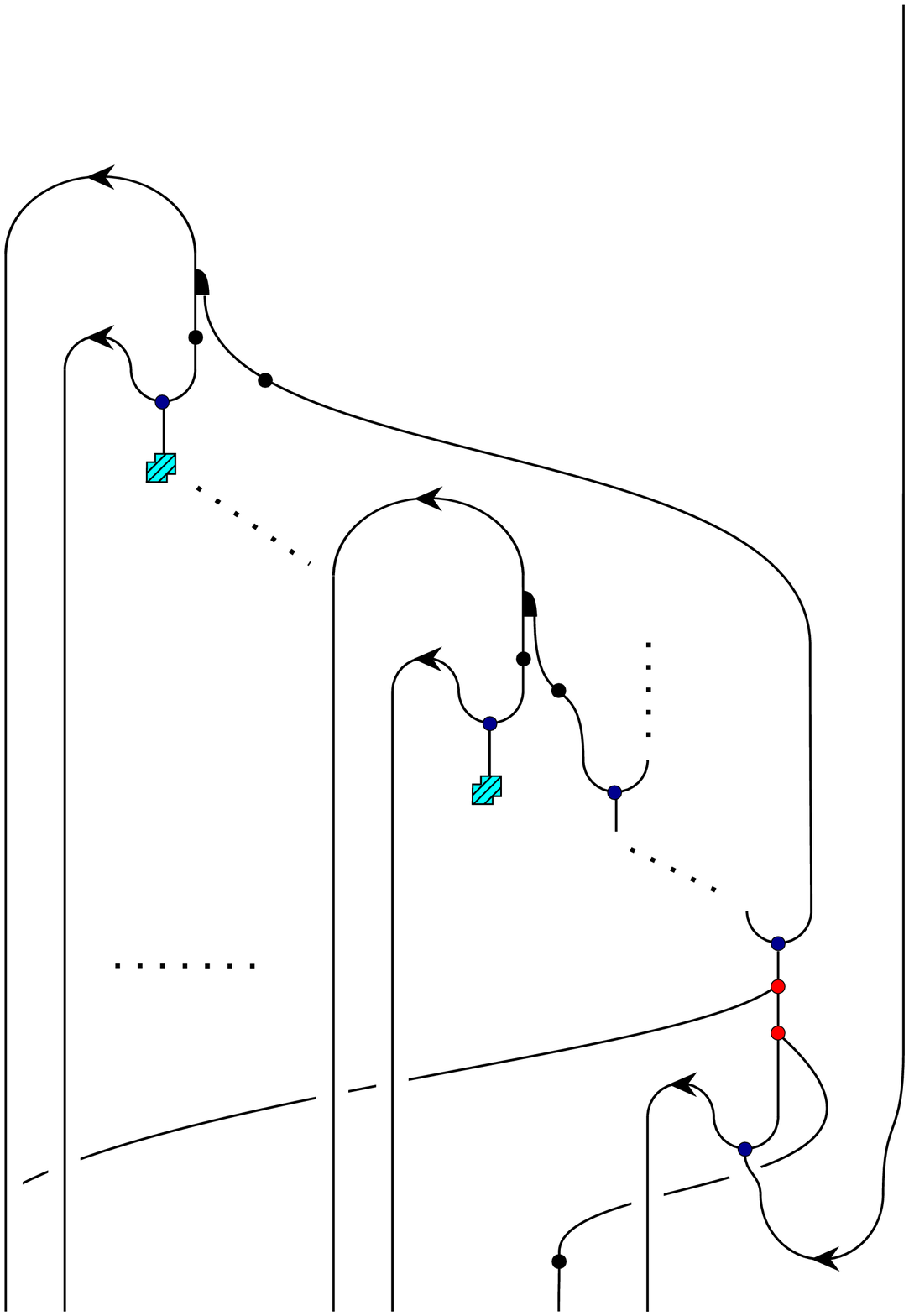}
  \put(-5,-8.5) {\sse$ H $}
  \put(8.2,-8.5){\sse$ \Hs $}
  \put(20.2,-8.5){\sse$ \Hs $}
  \put(76.2,-8.5){\sse$ \Hs $}
  \put(88.2,-8.5){\sse$ \Hs $}
  \put(121.2,-8.5){\sse$H$}
  \put(140.2,-8.5){\sse$\Hs$}
  \put(191.2,271.5){\sse$\Hs$}
  \put(34,171.5)  {\sse$ \Lambda$}
  \put(100,103.5)  {\sse$ \Lambda$}
  \put(57,208)  {\sse$ \rad$}
  \put(124,142)  {\sse$ \rad$}
  }
  }
where $\rho_{\HKH^{\!\otii p}}^H$ ($\ohr_{\HKH^{\!\otii p}}^H$) is the left (right) action of $H$ on $\HKH^{\!\otii p}$.\\
(iii)
\eqpic{Fv_act} {180} {39} {
  \put(0,0) {\Includepichtft{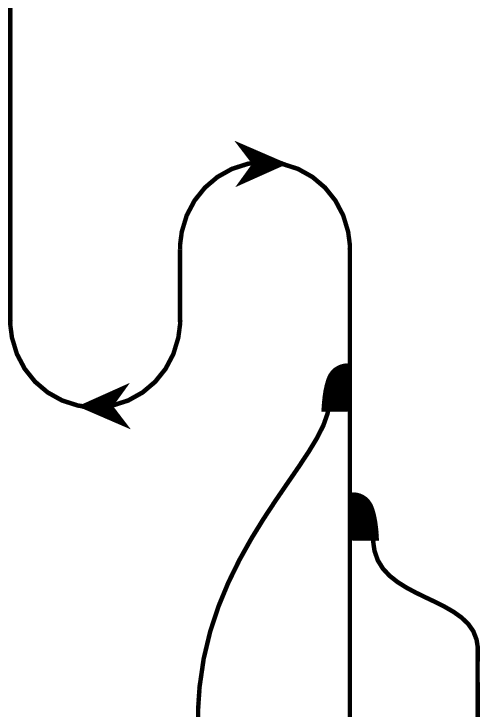}
  \put(17,-8.5) {\sse$ H $}
  \put(31,-8.5) {\sse$ ^\vee\!\BFH $}
  \put(48,-8.5) {\sse$ H$}
  \put(41,20.5)  {\sse$ \ohrV^H $}
  \put(39,35.5)  {\sse$ \rhoV^H $}
  \put(-3,82) {\sse$ H $}
  }
  \put(70,29)   {$=$}
  \put(100,0) {\Includepichtft{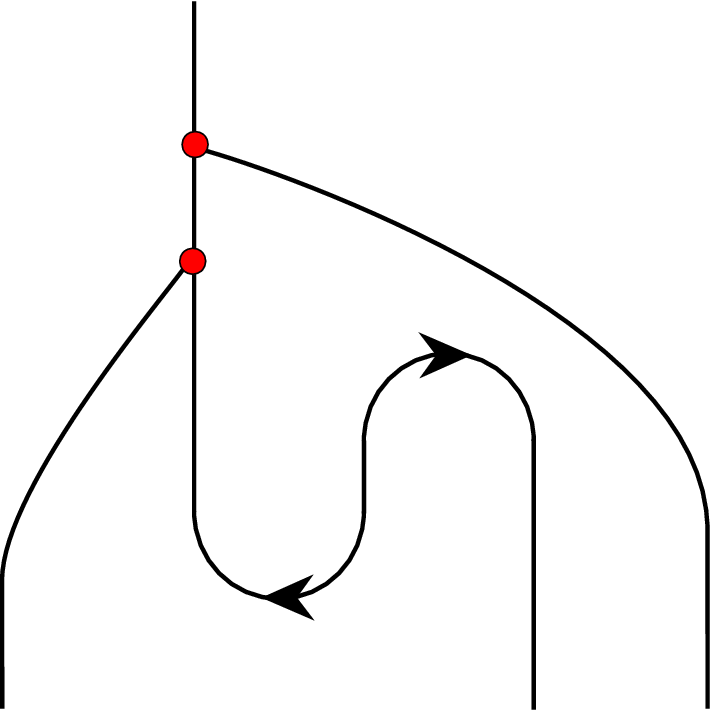}
  \put(-4,-8.5) {\sse$ H $}
  \put(51,-8.5) {\sse$ ^\vee\!\BFH $}
  \put(73,-8.5) {\sse$ H$}
  \put(18,82) {\sse$ H $}
  }
  }
where here $\ohrV^H $ and $ \rhoV^H $ are the left dual actions of $H$ on $^\vee\!\BFH$.
\end{lemma}
\begin{proof}

(i) The equality follows by inserting the definition of the right adjoint action, applying \eqref{Hopf_Frob_trick2} and the algebra anti-automorphism property of $\apoi$.

(ii) Starting from the left hand side of \eqref{Hpastloops}, inserting the expression \eqref{Corr12H} for $\SkH 111$, and the definitions of  $\rho_{\HKH^{\!\otii p}}^H$ and $\ohr_{\HKH^{\!\otii p}}^H$, we obtain, after applying \eqref{reppastint}:
\Eqpic{Hpastloops_1} {320} {103} {
  \put(10,-10){\setulen80
  \put(0,275) {$\twoact=$}
  \put(0,0) {\includepichtft{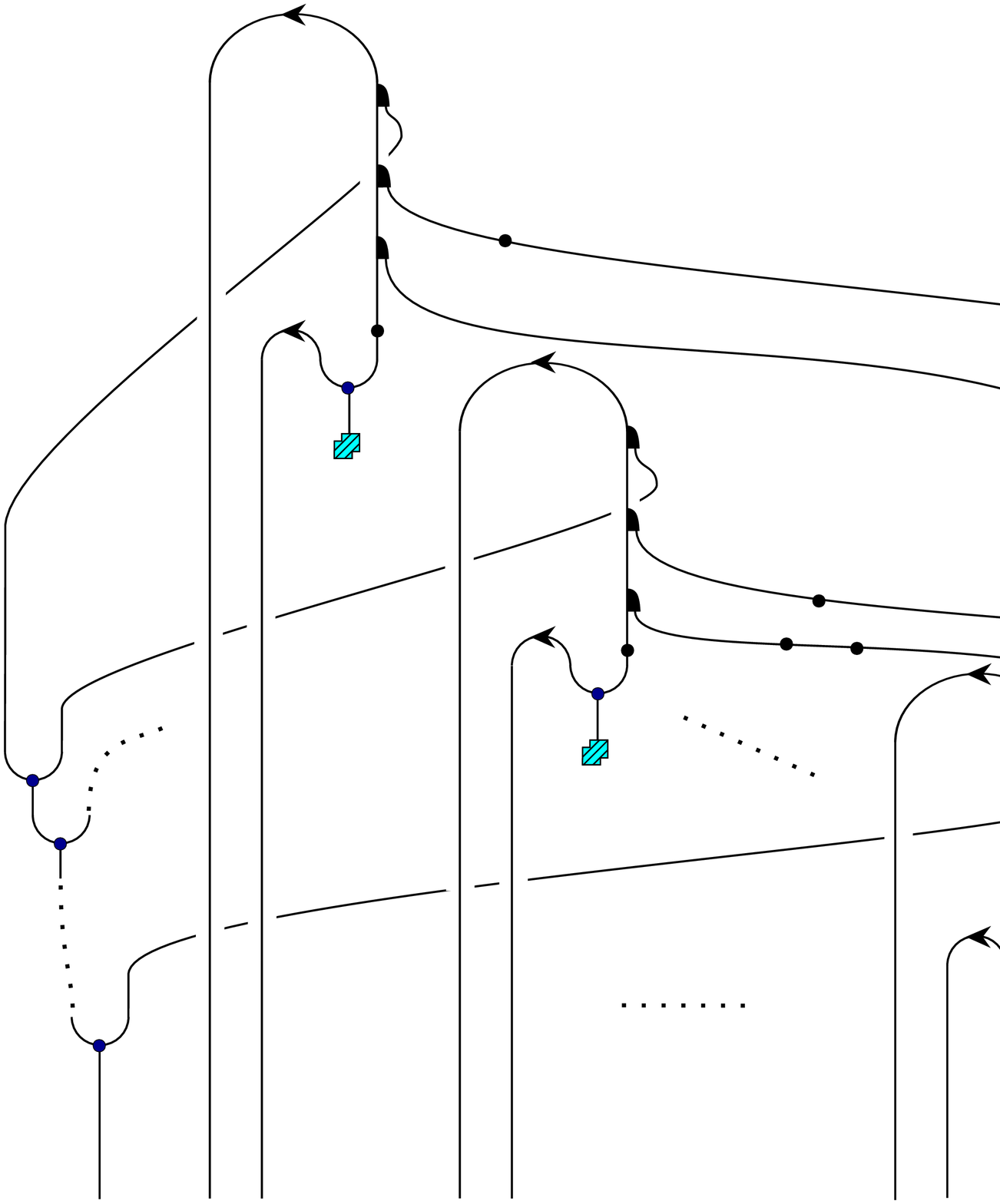}
  \put(18,-8.5) {\sse$ H $}
  \put(44.2,-8.5){\sse$ \Hs $}
  \put(55.2,-8.5){\sse$ \Hs $}
  \put(145,-8.5){\sse$ \cdots\cdots $}
  \put(101.2,-8.5){\sse$ \Hs $}
  \put(112.2,-8.5){\sse$ \Hs $}
  \put(202.2,-8.5){\sse$ \Hs$}
  \put(214.2,-8.5){\sse$ \Hs$}
  \put(313.2,-8.5){\sse$ H$}
  \put(329.2,-8.5){\sse$ \Hs$}
  \put(398,289)  {\sse$ \Hs $}
  \put(84,171.5)  {\sse$ \Lambda$}
  \put(141,101.5)  {\sse$ \Lambda$}
  \put(241,31.5)  {\sse$ \Lambda$}
    \put(92.2,217)   {\sse$\rad$}
  \put(92.2,234)   {\sse$\rad$}
  \put(92.2,254)   {\sse$\rad$}
  \put(149.7,137.5){\sse$\rad$}
  \put(149.7,155.8){\sse$\rad$}
  \put(149.7,175.3){\sse$\rad$}
  \put(249.7,63)   {\sse$\rad$}
  \put(249.7,82)   {\sse$\rad$}
  \put(249.7,102.8){\sse$\rad$}
  } }
  }
Next, invoking the representation property of $\rad$ and the anti-(co)algebra morphism property of the antipode several times we obtain
\Eqpic{Hpastloops_2} {320} {100} {
  \put(10,-10){\setulen80
  \put(0,275) {$\twoact=$}
  \put(0,0) {\includepichtft{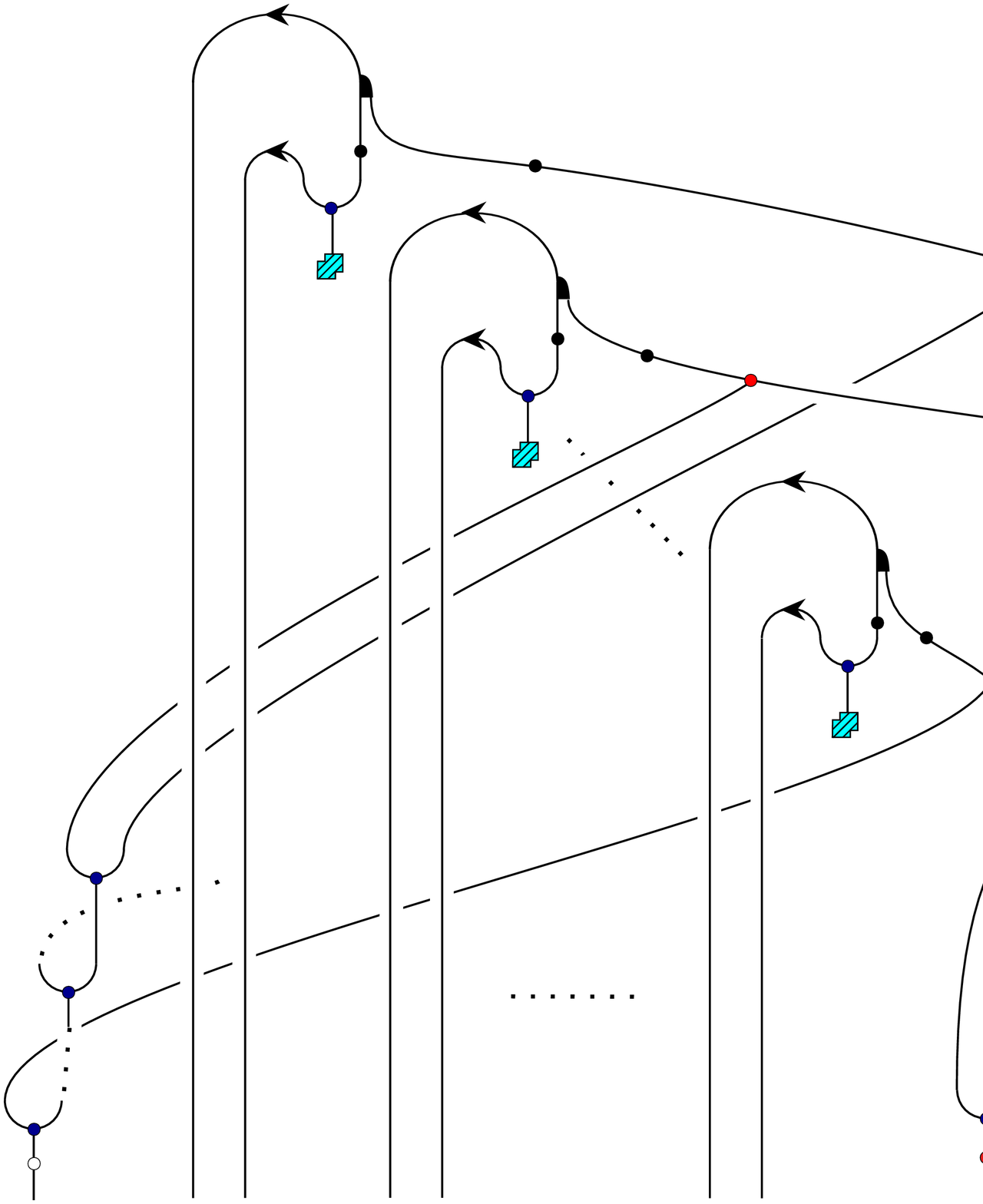}
  \put(3,-8.5) {\sse$ H $}
  \put(40.2,-8.5){\sse$ \Hs $}
  \put(52.2,-8.5){\sse$ \Hs $}
  \put(69,-8.5){\sse$\cdots $}
  \put(85.2,-8.5){\sse$ \Hs $}
  \put(96.2,-8.5){\sse$ \Hs $}
  \put(159.2,-8.5){\sse$ \Hs$}
  \put(171.2,-8.5){\sse$ \Hs$}
  \put(223.2,-8.5){\sse$ H$}
  \put(325.2,-8.5){\sse$ \Hs$}
  \put(400,295)  {\sse$ \Hs $}
  \put(80,213.5)  {\sse$ \Lambda$}
  \put(111,169.5)  {\sse$ \Lambda$}
  \put(185,107.5)  {\sse$ \Lambda$}
  } }
  }
That this expression is equal to the right hand side of \eqref{Hpastloops} is seen by using the fact that any concatenation of ${p-1}$ of $H$ is a bimodule morphism from $H$ to $H^{\otii p}$.

(iii) The equality \eqref{Fv_act} follows directly from duality, by inserting the left and right action  \eqref{rhoHb,ohrHb} of $H$ on $\BFH$, and the definition \eqref{Hbim_rightdualactions} of the left dual actions.
\endofproof\end{proof}
We also have the following generalization of Lemma \ref{lemma_helpid} (ii):
\begin{lemma}\label{lemma:Hpastloopsw}
Denoting, as in picture \eqref{Hpastloops}, by $\rho^H_{\HKH^{\otimes p}}$
and $\ohr^H_{\HKH^{\otimes p}}$ the left and right $H$-actions on
$\HKH^{\otimes p}_{}$, we have
  \Eqpic{Hpastloopsw} {300} {95} {
  \put(10,0){\setulen70
      \put(0,0) {\INcludepichtft{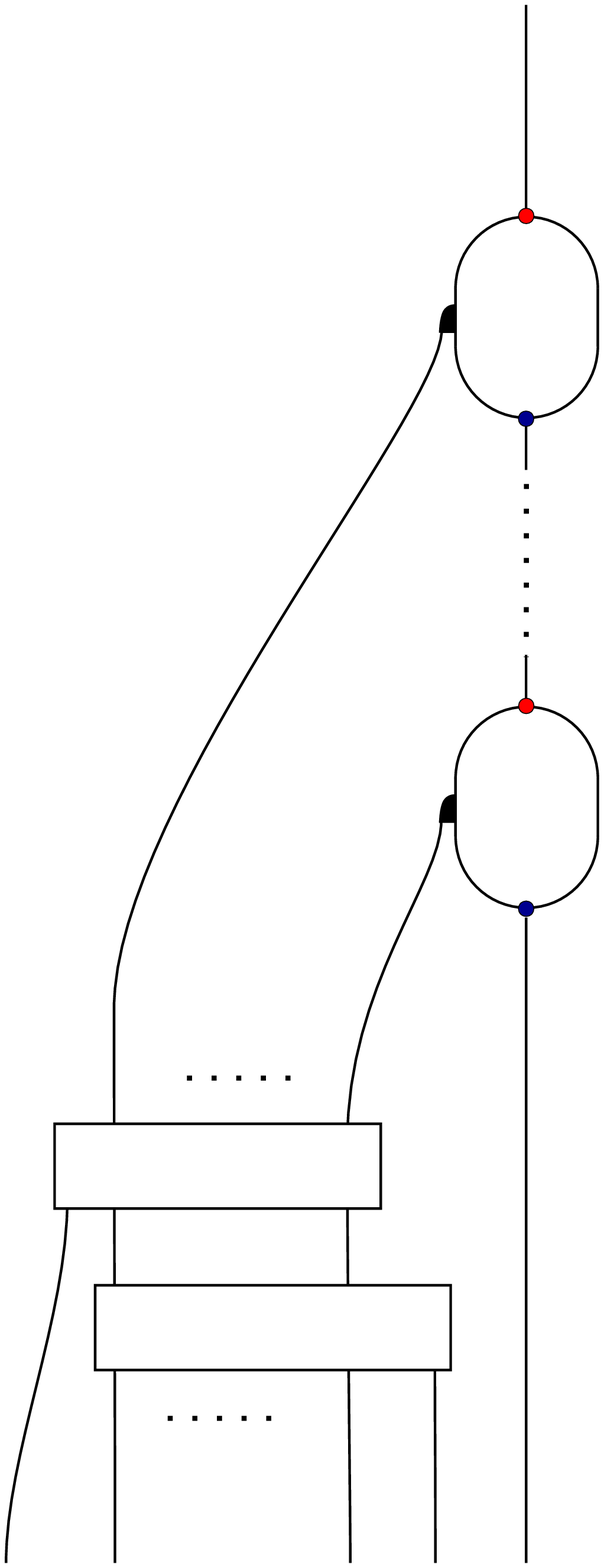}{266}
  \put(-5,-10.2)   {\sse$ H $}
  \put(17.2,-10.2) {\sse$ \HKH $}
  \put(60.2,-10.2) {\sse$ \HKH $}
  \put(76.2,-10.2) {\sse$ H $}
  \put(92.2,-10.2) {\sse$ \BFw $}
  \put(92.9,293.3) {\sse$ \BFw $}
  \put(28,71.9)    {\sse$ \rho_{\HKH^{\!\otimes p}}^H $}
  \put(39,41.9)    {\sse$ \ohr_{\HKH^{\!\otimes p}}^H $}
  }
  \put(150,115)    {$=$}
  \put(230,0) {
  \put(-37.5,0)    {\INcludepichtft{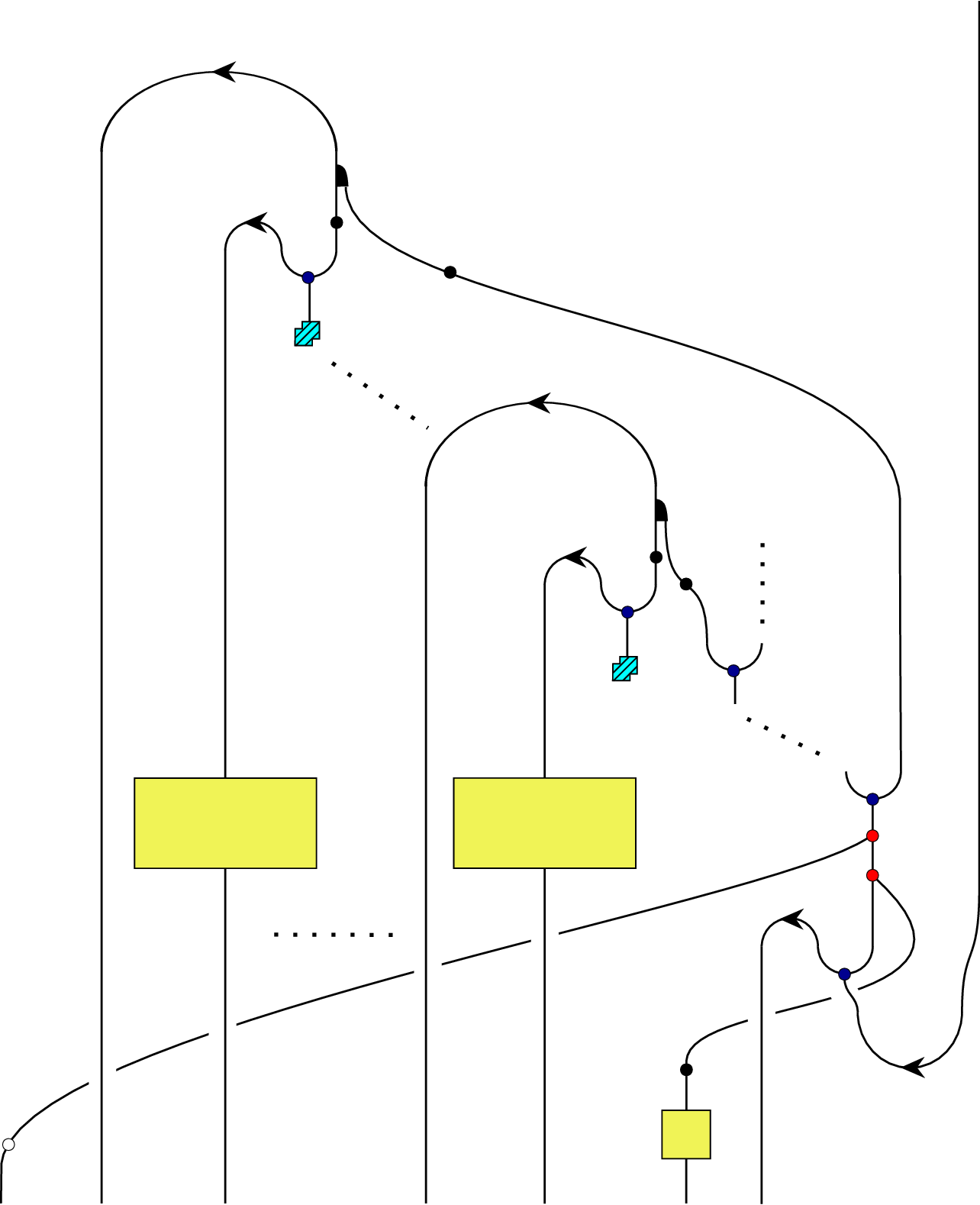}{266}}
  \put(-42.5,-10.2)   {\sse$ H $}
  \put(-18,-10.2)  {\sse$ \Hss $}
  \put(11.2,-10.2) {\sse$ \Hss $}
  \put(59.2,-10.2) {\sse$ \Hss $}
  \put(88.2,-10.2) {\sse$ \Hss $}
  \put(121.2,-10.2){\sse$ H $}
  \put(139.2,-10.2){\sse$ \Hss $}
  \put(192.2,292.5){\sse$ \Hss $}
  \put(23.6,205.5) {\sse$ \Lambda $}
  \put(99.6,125.1) {\sse$ \Lambda $}
  \put(47.9,243)   {\sse$ \rad $}
  \put(124.3,165)  {\sse$ \rad $}
  \put(122.2,14.1) {$\sse\omega$}
  \put(71.5,87.5)    {$\sse{(\ra^{-1})}^*$}
  \put(-5,87.5)    {$\sse{(\ra^{-1})}^*$}
  } } }
\end{lemma}

\begin{proof}
This follows by the same line of arguments as in Lemma \ref{lemma_helpid} (ii),
combined with the identity $\rad \cir (\ra^{-1}\oti\id_{H})
\eq \ra^{-1}\cir\rad \cir (\id_H\oti\omega)$.
\endofproof\end{proof}

\newcommand\LN[4]           {{\em #4}, \text{Lecture notes }{#1} {(#2)}{\tt[#3]} }
\newcommand\LNo[3]           {{\em #3}, \text{Lecture notes }{#1} {(#2)} }
\newcommand\PR[5]     {{\em #5}, {#1} {(#2)} {#3} {\tt[#4]} }

 \newcommand\wb{\,\linebreak[0]}
 \def\wB {$\,$\wb}
 \newcommand\Bi[1]    {\bibitem{#1}}
 \newcommand\Erra[3]  {\,[{\em ibid.}\ {#1} ({#2}) {#3}, {\em Erratum}]}
 \newcommand\J[5]     {{\em #5}, {#1} {#2} ({#3}) {#4} }
 \newcommand\K[6]   {{\em #6}, {#1} {#2} ({#3}) {#4} {\tt[#5]} }
 \newcommand\nqma[2] {\inBO{Non-perturbative QFT Methods and Their Applications}
              {Z.\ Horv\'ath and L.\ Palla, eds.} {World Scientific}{Singapore}{2001} {{#1}}{{#2}} }
 \newcommand\Prep[2]  {{\em #2}, preprint {\tt #1} }
 \newcommand\PhD[2]   {{\em #2}, Ph.D.\ thesis (#1)}
 \newcommand\BOOK[4]  {{\em #1\/} ({#2}, {#3} {#4})}
 \newcommand\inBO[7]  {{\em #7}, in:\ {\em #1}, {#2}\ ({#3}, {#4} {#5}), p.\ {#6}}
 \newcommand\inBOp[8] {{\em #8}, in:\ {\em #1}, {#2}\ ({#3}, {#4} {#5}), p.\ {#6} {[\tt #7]}}
 \def\adma  {Adv.\wb in Math.}
 \def\ajse  {Arabian Journal for Science and Engineering}
 \def\appb  {Acta\wB Phys.\wb Pol.\ B}
 \def\atmp  {Adv.\wb Theor.\wb Math.\wb Phys.}
 \def\bacp  {Ba\-nach\wB Cen\-ter\wB Publ.}
 \def\bams  {Bull.\wb Amer.\wb Math.\wb Soc.}
 \def\coma  {Con\-temp.\wb Math.}
 \def\cocm  {Com\-mun.\wb Con\-temp.\wb Math.}
 \def\comp  {Com\-mun.\wb Math.\wb Phys.}
 \def\cpma  {Com\-pos.\wb Math.}
 \newcommand\Epub[2]  {{\em #2}, published electronically at {#1}}
 \def\fiic  {Fields\wB Institute\wB Commun.}
 \newcommand\grfii[2] {\inBO{The Grothendieck Festschrift{\rm, vol.\ II}}
              {P.\ Cartier et al., eds.} {Birk\-h\"au\-ser}{Boston}{1990} {{#1}}{{#2}}}
 \def\foph  {Fortschr.\wb Phys.}
 \def\inma  {Invent.\wb math.}
 \def\isjm  {Israel\wB J.\wb Math.}
 \def\jams  {J.\wb Amer.\wb Math.\wb Soc.}
 \def\jhep  {J.\wb High\wB Energy\wB Phys.}
 \def\jktr  {J.\wB Knot\wB Theory\wB and\wB its\wB Ramif.}
 \def\joal  {J.\wB Al\-ge\-bra}
 \def\jopa  {J.\wb Phys.\ A}
 \def\josp  {J.\wb Stat.\wb Phys.}
 \def\jpaa  {J.\wB Pure\wB Appl.\wb Alg.}
 \def\jste  {J.\wb Stat.\wb Mech.}
 \def\momj  {Mos\-cow\wB Math.\wb J.}
 \def\nuci  {Nuovo\wB Cim.}
 \def\nupb  {Nucl.\wb Phys.\ B}
 \def\pcps  {Proc.\wB Cam\-bridge\wB Philos.\wb Soc.}
 \def\pnas  {Proc.\wb Natl.\wb Acad.\wb Sci.\wb USA}
 \def\phrl  {Phys.\wb Rev.\wb Lett.}
 \def\phlb  {Phys.\wb Lett.\ B}
 \def\plms  {Proc.\wB Lon\-don\wB Math.\wb Soc.}
 \def\remp  {Rev.\wb Mod.\wb Phys.}
\def\rpip  {Rep.\wb Prog.\wB in\wB Phys.}
\def\ruma  {Revista de la Uni\'on Matem\'atica Argentina}
 \def\taac  {Theo\-ry\wB and\wB Appl.\wb Cat.}
 \def\tams  {Trans.\wb Amer.\wb Math.\wb Soc.}
 \def\thmp  {Theor.\wb Math.\wb Phys.}

\printindex
\end{document}